\documentclass[12pt,a4paper,PSfrag]{article}
\usepackage[T1]{fontenc}
\usepackage[latin9]{inputenc}
\usepackage{color}
\definecolor{shadecolor}{rgb}{1, 0, 0}
\usepackage{calc}
\usepackage{framed}
\usepackage{mathrsfs}
\usepackage{amscd}
\usepackage{amsthm}
\usepackage{amsmath}
\usepackage{amssymb}
\usepackage{graphicx}
\usepackage[all]{xy}
\usepackage[unicode=true,pdfusetitle,
 bookmarks=true,bookmarksnumbered=false,bookmarksopen=false,
 breaklinks=false,pdfborder={0 0 1},backref=false,colorlinks=false]
 {hyperref}
\usepackage{breakurl}

\makeatletter

\newcommand{\noun}[1]{\textsc{#1}}
\numberwithin{equation}{section}
\numberwithin{figure}{section}
\numberwithin{table}{section}
\theoremstyle{plain}
\newtheorem{thm}{\protect\theoremname}[section]
  \theoremstyle{definition}
  \newtheorem{defn}[thm]{\protect\definitionname}
  \theoremstyle{remark}
  \newtheorem{rem}[thm]{\protect\remarkname}
  \theoremstyle{definition}
  \newtheorem{problem}[thm]{\protect\problemname}
  \theoremstyle{plain}
  \newtheorem{prop}[thm]{\protect\propositionname}
  \theoremstyle{plain}
  \newtheorem{lem}[thm]{\protect\lemmaname}
  \theoremstyle{plain}
  \newtheorem{cor}[thm]{\protect\corollaryname}
\newenvironment{lyxcode}
{\par\begin{list}{}{
\setlength{\rightmargin}{\leftmargin}
\setlength{\listparindent}{0pt}% needed for AMS classes
\raggedright
\setlength{\itemsep}{0pt}
\setlength{\parsep}{0pt}
\normalfont\ttfamily}%
 \item[]}
{\end{list}}
  \theoremstyle{remark}
  \newtheorem*{rem*}{\protect\remarkname}

\usepackage{a4wide}
\usepackage{color}
\usepackage{graphics}

\usepackage{amsmath, amsthm, amssymb}

\usepackage[left,modulo]{lineno}% pour rajouter des numeros de ligne 
\linenumbers 

\makeatother

  \providecommand{\corollaryname}{Corollary}
  \providecommand{\definitionname}{Definition}
  \providecommand{\lemmaname}{Lemma}
  \providecommand{\problemname}{Problem}
  \providecommand{\propositionname}{Proposition}
  \providecommand{\remarkname}{Remark}
\providecommand{\theoremname}{Theorem}

\begin{document}

\title{Prequantum transfer operator for symplectic Anosov diffeomorphism}

\author{Frédéric Faure\\
 {\small Institut Fourier, UMR 5582}\\
 {\small 100 rue des Maths, BP74 }\\
{\small FR-38402 St Martin d'Hères, FRANCE}\\
{\small frederic.faure@ujf-grenoble.fr}\\
{\small http://www-fourier.ujf-grenoble.fr/\textasciitilde{}faure}
\and\\
Masato Tsujii{\small }\\
{\small Department of Mathematics, Kyushu University,}\\
{\small Moto-oka 744, Nishi-ku, Fukuoka, 819-0395, JAPAN }\\
{\small tsujii@math.kyushu-u.ac.jp}\\
{\small http://www2.math.kyushu-u.ac.jp/\textasciitilde{}tsujii }}

\date{2013 May 8}
\maketitle
\begin{abstract}
We define the prequantization of a symplectic Anosov diffeomorphism
$f:M\rightarrow M$, which is a $\mathbf{U}(1)$ extension of the
diffeomorphism $f$ preserving an associated specific connection,
and study the spectral properties of the associated transfer operator,
called \emph{prequantum transfer operator}. This is a model for the
transfer operators associated to geodesic flows on negatively curved
manifolds (or contact Anosov flows). 

We restrict the prequantum transfer operator to the $N$-th Fourier
mode with respect to the $\mathbf{U}(1)$ action and investigate the
spectral property in the limit $N\to\infty$, regarding the transfer
operator as a Fourier integral operator and using semi-classical analysis.
In the main result, we show a ``\,\emph{band structure}\,'' of
the spectrum, that is, the spectrum is contained in a few separated
annuli and a disk concentric at the origin.

We show that, with the special (Hölder continuous) potential $V_{0}=\frac{1}{2}\log\left|\det Df_{x}|_{E_{u}}\right|$,
the outermost annulus is the unit circle and separated from the other
parts. For this, we use an extension of the transfer operator to the
Grassmanian bundle. Using Atiyah-Bott trace formula, we establish
the Gutzwiller trace formula with exponentially small reminder for
large time. We show also that, for a potential $V$ such that the
outermost annulus is separated from the other parts, most of the eigenvalues
in the outermost annulus concentrate on a circle of radius $\exp\left(\left\langle V-V_{0}\right\rangle \right)$
where $\left\langle .\right\rangle $ denotes the spatial average
on $M$. The number of these eigenvalues is given by the ``Weyl law'',
that is, $N^{d}\mathrm{Vol}M$ with $d=\frac{1}{2}\mathrm{dim}M$
in the leading order.

We develop a semiclassical calculus associated to the prequantum operator
by defining quantization of observables $\mathrm{Op}_{\hbar}\left(\psi\right)$
in an intrinsic way. We obtain that the semiclassical Egorov formula
of quantum transport is exact. We interpret all these results from
a physical point of view as the emergence of quantum dynamics in the
classical correlation functions for large time. We compare these results
with standard quantization (geometric quantization) in quantum chaos.
\end{abstract}
\newpage

\tableofcontents{}

%Real and complex numbers
\global\long\def\real{\mathbb{R}}
 \global\long\def\complex{\mathbb{C}}
 %Bargmann transform and projector
\global\long\def\Bargmann{\mathcal{B}}
 \global\long\def\BargmannP{\mathcal{P}}
 %The space of homogeneous polynomials
 \global\long\def\Polynomial{\mathrm{Polynom}}
 %Some projection operators related to Taylor expansions
\global\long\def\Taylor{T}
 \global\long\def\calT{{\mathcal{T}}}
 \global\long\def\checktau{\check{\tau}}
 \global\long\def\hattau{\hat{\tau}}
 \global\long\def\boldT{\mathbf{T}}
 \global\long\def\boldt{\boldsymbol{t}}

\global\long\def\checklambda{\check{\lambda}}
 \global\long\def\hatlambda{\hat{\lambda}}

%Some symbols related to the transfer operators

\global\long\def\Llift{L^{\mathrm{lift}}}
 \global\long\def\prequantumL{\mathcal{L}}
 \global\long\def\prequantumLlift{\mathcal{L}^{\mathrm{lift}}}

%some bold symbols

\global\long\def\boldI{\mathbf{I}}
 \global\long\def\boldu{\mathbf{u}}
 \global\long\def\boldF{\mathbf{F}}
 \global\long\def\boldDelta{\boldsymbol{\Delta}}
 \global\long\def\boldq{\boldsymbol{q}}

% some useful symbols that appears several times

\global\long\def\supp{\mathrm{supp}\,}
 \global\long\def\Im{\mathrm{Im}\,}

\global\long\def\scrX{\mathscr{X}}
 \global\long\def\scrG{\mathscr{G}}

%not quite necessary.....

\global\long\def\romet{t}
 \global\long\def\romeq{q}
 \global\long\def\domain{\mathcal{D}}

%%%%Section 4

\global\long\def\BargmannModified{\widetilde{\Bargmann}}

\global\long\def\BargmannPModified{\widetilde{\BargmannP}}

\global\long\def\multiplication{\mathscr{M}}

\newpage

\newpage

\paragraph{}

\section{\label{sec:Introduction-and-results}Introduction and results}

\subsection{Introduction}

We consider a smooth symplectic Anosov diffeomorphism $f:M\rightarrow M$
on a $2d$-dimensional closed symplectic manifold $\left(M,\omega\right)$
as a standard model of \textquotedbl{}chaotic\textquotedbl{} dynamical
system. Following the geometric quantization procedure introduced
by Kostant, Souriau and Kirillov in 1970s', we consider the prequantum
bundle $P\rightarrow M$. This is the $\mathbf{U}(1)$-principal bundle
over $M$ equipped with a connection whose curvature is $\left(-2\pi i\right)\cdot\omega$.
Then we introduce the prequantum map $\tilde{f}:P\rightarrow P$ as
the $\mathbf{U}(1)$-equivariant extension of the map $f$ preserving
the connection. The prequantum map $\tilde{f}$ thus defined is known
to be exponentially mixing%
\footnote{Exponential mixing of the map $\tilde{f}$ is already known \cite{dolgopyat_02}
but is also a direct consequence of results presented in this paper. %
}, that is, any smooth probability density which evolves under the
iteration of $\tilde{f}$ converges weakly towards the uniform equilibrium
distribution on $P$ and the speed of convergence is exponentially
fast if it is measured by a smooth observable. We study the fluctuations
in this convergence to the equilibrium by investigating spectral properties
of the transfer operator $\hat{F}$ associated to the prequantum map
$\tilde{f}$. Following the approach taken by David Ruelle in his
study of expanding dynamical systems\cite{Ruelle1991}, we first show
that the transfer operator displays discrete spectrum, which is sometimes
called \emph{Ruelle-Pollicott resonances}. Precisely we consider the
restriction $\hat{F}_{N}$ of the transfer operator $\hat{F}$ to
the $N$-th Fourier mode with respect to the $\mathbf{U}(1)$ action
on $P$ and show that its natural extension to appropriate generalized
Sobolev spaces of distributions has discrete spectrum. This result
concerning discrete spectrum is already known in the preceding works
\cite{rugh_92},\cite{liverani_02,liverani_04},\cite[theorem 1.1]{baladi_05},\cite[theorem 1]{fred-roy-sjostrand-07}
and will be recalled in Theorem \ref{thm:Discrete-spectrum}. In this
paper we are mainly concerned with the limit $N\rightarrow\infty$
of small wavelength (high Fourier modes). We will use the standard
notation of semiclassical analysis and put throughout this paper:
\[
\hbar:=\frac{1}{2\pi N}
\]
 The new result of this paper is in Theorem \ref{thm:band_structure},
where we show that\textcolor{green}{{} }the spectrum of $\hat{F}_{N}$
has a particular ``band'' structure: for every $N$ large enough,
there is an annulus that contains finitely many (but increasing to
infinity as $N\to\infty$) eigenvalues; they are separated from the
rest of the internal spectrum by a gap under some pinching conditions.
This means that the convergence to the equilibrium mentioned above,
restricted to the $N$-th Fourier mode, is described by a finite rank
operator denoted $\mathcal{F}_{\hbar}:\mathcal{H}_{\hbar}\rightarrow\mathcal{H}_{\hbar}$,
up to relatively small exponentially decaying errors. The finite rank
operator $\mathcal{F}_{\hbar}$ is the spectral restriction of the
prequantum transfer operator $\hat{F}_{N}$ on the external annulus.
We show, in Theorem \ref{thm:Weyl-formula}, that the dimension of
$\mathcal{H}_{\hbar}$ is proportional to $N^{d}$ asymptotically
as%
\footnote{The precise value of $\mathrm{dim}\mathcal{H}_{\hbar}$ is given by
an index formula of Atiyah-Singer in Th. \ref{thm:Weyl-formula}.%
} $N\to\infty$. These results are generalizations of the results in
\cite{fred-PreQ-06} for the linear Arnold cat map to the case of
general non-linear symplectic Anosov diffeomorphisms.

\paragraph{Motivations of the study}

From the construction above, the prequantum map $\tilde{f}:P\rightarrow P$
is partially hyperbolic, that is, hyperbolic in the directions transverse
to the fibers but is neutral (because of equivariance) in the direction
of the fibers. Also note that $\tilde{f}$ preserves the connection
one form on the prequantum bundle $P$ which is a contact form on
$P$ (See Remark \ref{Rem:contact_form} page \pageref{Rem:contact_form}).
These properties of the prequantum map are very similar to those of
the time-$t$-map of the geodesic flow $\phi_{1}:T_{1}^{*}\mathcal{M}\rightarrow T_{1}^{*}\mathcal{M}$
on a closed Riemannian manifold $\mathcal{M}$ with negative curvature,
acting on the unit cotangent bundle $T_{1}^{*}\mathcal{M}$: In the
latter case the time-$t$-map of the geodesic flow is partially hyperbolic
and preserves the canonical Liouville contact one form $\xi dx$ on
$T_{1}^{*}\mathcal{M}$. (See \cite{liverani_contact_04,tsujii_08,tsujii_FBI_10,fred_flow_09}).
With this point of view, the prequantum transfer operator can be considered
as a model of the transfer operators for the geodesic flows on negatively
curved manifolds. One of our objective behind the present work is
to show some band structure of the spectrum for the case of geodesic
flow and extend other results presented in this paper to that case
\cite{faure-tsujii_anosov_flows_13,faure_tsujii_band_CRAS_2013}.
In the special case of manifolds with constant curvature, such a band
structure is readily observed from the classical theorem of Selberg
on zeta functions \cite{selberg_1956}. 

Another motivation already discussed in \cite{fred-PreQ-06} in a
special case is the following observation: The finite rank operator
$\mathcal{F}_{\hbar}$ which describes the long time classical correlation
functions of the map $\tilde{f}$ has the properties of a \textquotedbl{}quantum
map\textquotedbl{} i.e. a \textquotedbl{}quantization of $f$\textquotedbl{}
but with additional interesting properties. It satisfies the Gutzwiller
trace formula with an error term which decreases exponentially fast
in large time, an exact Egorov theorem, etc. Surprisingly this ``quantization''
or quantum behavior, appears here dynamically (after long time) in
the classical correlation functions of the ``classical'' map $\tilde{f}$:
the finite dimensional ``quantum space'' $\mathcal{H}_{\hbar}$
in which $\mathcal{F}_{\hbar}$ acts is defined from the dynamics.
There are many open questions in ``quantum chaos'' for example related
to ``unique quantum ergodicity'' or ``random matrix theory''\cite{nonnenmacher_08}.
These questions can be posed for the family quantum operators $\left(\mathcal{F}_{\hbar}\right)_{\hbar}$considered
here and may be their special properties with respect to the dynamics
may help.

\paragraph{Semiclassical approach}

The general method that we use to obtain the main results is \emph{semiclassical
analysis}. We regard the prequantum transfer operator as a Fourier
Integral Operator (FIO), which means that we consider its action on
wave packets in the high frequency limit $N\to\infty$. From the general
idea in semiclassical analysis, this action is effectively described
by the associated canonical map $(Df^{*})^{-1}$ on the cotangent
space $T^{*}M$ equipped with the symplectic structure $\Omega=dx\wedge d\zeta+\pi^{*}\omega$
(where $dx\wedge d\zeta$ stands for the canonical symplectic structure
on $T^{*}M$ and $\pi^{*}\omega$ is the pull-back of $\omega$ on
$T^{*}M$). For the action of the canonical map $(Df^{*})^{-1}$,
the non-wandering set is the zero section $K\subset T^{*}M$ and is
called the \emph{trapped set}. The additional term $\pi^{*}\omega$
in $\Omega$ makes $K$ a symplectic submanifold. The trapped set
is therefore symplectic and normally hyperbolic. We will see that
these facts are the core of our argument and give the band structure
of the spectrum in the main theorem.

\paragraph{Structure of the paper}

In Section \ref{sub:Definitions} we define precisely the prequantum
map $\tilde{f}$ and the prequantum transfer operator $\hat{F}$ that
are uniquely associated to the Anosov map $f$. In Section \ref{sub:1.3}
we present the main results concerning the discrete spectrum of $\hat{F}_{N}$
(after Fourier decomposition of $\hat{F}$ on Fourier component $N=1/\left(2\pi\hbar\right)\in\mathbb{Z}$)
acting on a Hilbert space $\mathcal{H}_{\hbar}^{r}$ called the anisotropic
Sobolev space. These results are summarized on Figure \ref{fig:annuli}.
The spectral restriction of the operator \textbf{$\hat{F}_{N}$} on
this external band will be called the \textbf{quantum operator }and
denoted $\hat{\mathcal{F}}_{\hbar}$ in Definition \ref{def:quantum_operator}.
The associated spectral projector is denoted by$\Pi_{\hbar}$. In
Section \ref{sub:Spectral-results-with-Grassman} we show how to extend
the results so that the potential function $V$ that enters in the
definition \ref{def:Prequant_op} of the transfer operator $\hat{F}$
can be chosen such that the external band of resonances concentrates
on the unit circle in the limit $\hbar\rightarrow0$, although this
requires that $V$ is only Hölder continuous. In Section \ref{sub:1.5Gutzwiller-trace-formula}
we show that the quantum operator $\mathcal{\hat{F}}_{\hbar}^{n}$
satisfies the Gutzwiller trace formula with an error that decays exponentially
fast as $n\rightarrow\infty$. We discuss the fact that this property
determines the spectrum of $\hat{\mathcal{F}}_{\hbar}$ and shows
somehow that the family of operators $\left(\hat{\mathcal{F}}_{\hbar}\right)_{\hbar}$
is a kind of ``natural quantization'' of the Anosov symplectic map
$f$. In Section \ref{sub:Emergence-of-quantum} we pursue the exploration
of the properties of this quantum operator $\hat{\mathcal{F}}_{\hbar}$
obtained by the spectral restriction of $\hat{F}_{N}$ by the spectral
projector $\Pi_{\hbar}$. In Theorem \ref{Th1.7} it is shown that
$\hat{\mathcal{F}}_{\hbar}$ describes the exponential decay of correlations
of the Anosov map $f$. In Section \ref{sub:1.7Semiclassical-calculus-on}
we show in which sense the quantum operator $\hat{\mathcal{F}}_{\hbar}$
is a kind of ``quantum map'': it satisfies an exact Egorov formula
with respect to an algebra of quantum observables $\mathrm{Op}_{\hbar}\left(\psi\right)$.
For this, we define a new kind of quantization procedure $\mathrm{Op}_{\hbar}:\psi\in C^{\infty}\left(M\right)\rightarrow\mathrm{Op}_{\hbar}\left(\psi\right)\in\mathrm{End}\left(\mathcal{H}_{\hbar}\right)$
which satisfies most of the usual ``axioms of quantization''. In
particular the spectral projector on the external band is $\Pi_{\hbar}=\mathrm{Op}_{\hbar}\left(1\right)$.
In Theorem \ref{prop:express_Op_Psi}, $\mathrm{Op}_{\hbar}\left(\psi\right)$
is expressed as an integral over $x\in M$ of $\psi\left(x\right)\cdot\pi_{x}$
where $\pi_{x}$ is a rank one projector over a ``localized wave
packet''%
\footnote{This localization property is with respect to the anisotropic Sobolev
space $\mathcal{H}_{\hbar}^{r}$.%
} at position $x\in M$. In Section \ref{sub:1.8Geometric-quantization-of}
we consider the usual ``geometric quantization'' of the map $f$
in the sense of Toeplitz quantization and compare it with the quantum
operator $\hat{\mathcal{F}}_{\hbar}$ or ``natural quantization''.
We show that both quantization coincide up a to small error in the
limit $\hbar\rightarrow0$. In Section 2 we present the main ideas
of semiclassical analysis used in the proofs. In Section \ref{sec:Covariant-derivative-}
we present additional results concerning the spectrum of the rough
Laplacian operator. This is used to discuss geometric (Toeplitz) quantization.
Subsequent sections contain the proofs.

\paragraph{Acknowledgments:}

F. Faure would like to thank Yves Colin de Verdière\noun{, }Louis
Funar\noun{, }Sébastien Gouëzel, Colin Guillarmou\noun{, }Malik Mezzadri,
Johannes Sjöstrand for interesting discussions related to this work.
During the period of this research project, M. Tsujii was partially
supported by Grant-in-Aid for Scientific Research (B) (No.22340035)
from Japan Society for the Promotion of Science. F. Faure has been
supported by ``Agence Nationale de la Recherche'' under the grants
JC05\_52556 and ANR-08-BLAN-0228-01.

\subsection{Definitions\label{sub:Definitions}}

\subsubsection{Symplectic Anosov map}

Let $M$ be a $C^{\infty}$ closed connected symplectic manifold of
dimension $2d$ with symplectic two form $\omega$. Let $f:M\rightarrow M$
be a $C^{\infty}$ symplectic Anosov diffeomorphism, i.e. a $C^{\infty}$
Anosov diffeomorphism such that $f^{*}\omega=\omega$. We recall the
definition of an Anosov diffeomorphism:

%red
\begin{center}{\color{red}\fbox{\color{black}\parbox{16cm}{
\begin{defn}
\cite[p.263]{katok_hasselblatt} \label{def:Anosov_diffeo}A diffeomorphism
$f:M\rightarrow M$ is said to be \emph{Anosov} if there exists a
$C^{\infty}$ Riemannian metric $g$ on $M$, an $f$-\emph{invariant}
continuous decomposition of $TM$,
\begin{equation}
T_{x}M=E_{u}\left(x\right)\oplus E_{s}\left(x\right),\quad\forall x\in M\label{eq:foliation}
\end{equation}
and a constant $\lambda>1$, such that, for any $x\in M$, hold
\begin{eqnarray}
\left|D_{x}f\left(v_{s}\right)\right|_{g} & \leq & \frac{1}{\lambda}\left|v_{s}\right|_{g}\quad\forall v_{s}\in E_{s}\left(x\right),\quad\mbox{and}\label{eq:def_dynamics}\\
\left|D_{x}f^{-1}\left(v_{u}\right)\right|_{g} & \leq & \frac{1}{\lambda}\left|v_{u}\right|_{g}\quad\forall v_{u}\in E_{u}\left(x\right).\nonumber 
\end{eqnarray}
The subbundle $E_{s}$ (resp.$E_{u}$) in which $f$ is uniformly
contracting (resp. expanding) is called the stable (resp. unstable)
sub-bundle. See figure \ref{fig:Anosov map}.
\end{defn}
}}}\end{center}

\begin{figure}[h]
\centering{}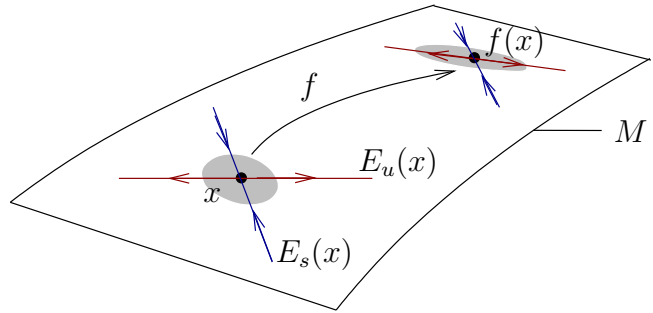\caption{\label{fig:Anosov map}A symplectic Anosov map $f$}
\end{figure}

\begin{rem}
\label{Remark:Holder_exp_beta}
\begin{enumerate}
\item The subspaces $E_{u}(x)$ and $E_{s}(x)$ do not depend smoothly on
the point $x$ in general. However it is known that they are Hölder
continuous in $x$ with some Hölder exponent \cite{pesin_04}. In
what follows, we assume that the Hölder exponent is 
\begin{equation}
0<\beta<1.\label{eq:def_beta}
\end{equation}
The subspaces $E_{u}\left(x\right)$ and $E_{s}\left(x\right)$ are
Lagrangian%
\footnote{To prove that $E_{s}\left(x\right)$ is Lagrangian, let $u,v\in E_{s}\left(x\right)$;
we have $\omega\left(u,v\right)=\omega\left(D_{x}f^{n}\left(u\right),D_{x}f^{n}\left(v\right)\right)\underset{n\rightarrow+\infty}{\rightarrow}0$.
Similarly for $E_{u}$. So $E_{s}$ and $E_{u}$ are isotropic subspaces,
$E_{u}\oplus E_{s}=TM$, hence they are Lagrangian.%
} linear subspace of $T_{x}M$ and both have dimension $d$.
\item The\textbf{ }\emph{Arnold cat map}\cite{arnold-avez} is a simple
example of a symplectic Anosov diffeomorphism on the torus $\mathbb{T}^{2}=\mathbb{R}^{2}/\mathbb{Z}^{2}$,
\begin{equation}
f_{0}\left(\begin{array}{c}
q\\
p
\end{array}\right)=\left(\begin{array}{cc}
2 & 1\\
1 & 1
\end{array}\right)\left(\begin{array}{c}
q\\
p
\end{array}\right)\quad\mod\;\mathbb{Z}^{2}.\label{eq:f0_cat_map}
\end{equation}
It preserves the symplectic form $\omega=dq\wedge dp$. If $h:M\to M$
is a diffeomorphism close enough to identity in the $C^{1}$ norm
and preserves the symplectic form $\omega,$ the \emph{perturbed cat
map} 
\begin{equation}
f\left(x\right):=h\left(f_{0}\left(x\right)\right)\label{eq:perturbated_cat_map}
\end{equation}
is also a (probably non-linear) symplectic Anosov diffeomorphism\cite[p.266]{katok_hasselblatt}.
Similarly, we get examples of symplectic Anosov diffeomorphisms on
$\mathbb{T}^{2d}$ from any symplectic linear map $f_{0}\in\mathrm{Sp}{}_{2d}\left(\mathbb{Z}\right)$
with no eigenvalues on the unit circle. 
\end{enumerate}
\end{rem}

\subsubsection{\label{sub:Hypothesis}The prequantum bundle and the lift map $\tilde{f}$}

A prequantum bundle is a $\mathbf{U}(1)$-principal bundle $P$ equipped
with a specific connection. In a few paragraphs below, we recall the
definition of a $\mathbf{U}(1)$-principal bundle and that of a connection
on it. (For the detailed account, we refer \cite{woodhouse2}.) The
one-dimensional unitary group $\mathbf{U}(1)$ is the multiplicative
group of complex numbers of the form $e^{i\theta},\theta\in\mathbb{R}$.
A \textbf{$\mathbf{U}(1)$}-principal bundle $P$ over $M$ is a manifold
with a free action of $\mathbf{U}(1)$, written 
\begin{equation}
p\in P\rightarrow\left(e^{i\theta}p\right)\in P,\label{eq:action_U1}
\end{equation}
such that the quotient space is $M=P/\mathbf{U}(1)$. We write $\pi:P\rightarrow M$
for the projection map. From the definition, the \textbf{$\mathbf{U}(1)$}-principal
bundle $P$ has a local product structure over $M$: There exist
a finite cover of $M$ by simply connected open subsets $U_{\alpha}\subset M$,
$\alpha\in I$, and smooth sections $\tau_{\alpha}:U_{\alpha}\rightarrow P$
on each of $U_{\alpha}$, called a \emph{local smooth section}; A
\emph{local trivialization} of $P$ over $U_{\alpha}$ is defined
as the diffeomorphism

\begin{equation}
T_{\alpha}:\begin{cases}
U_{\alpha}\times\mathbf{U}\left(1\right) & \rightarrow\pi^{-1}\left(U_{\alpha}\right)\\
\left(x,e^{i\theta}\right) & \rightarrow e^{i\theta}\tau_{\alpha}\left(x\right)
\end{cases}\label{eq:local_trivialization-1}
\end{equation}
A \emph{connection} on $P$ is a differential one form $A\in C^{\infty}\left(P,\Lambda^{1}\otimes\left(i\mathbb{R}\right)\right)$
on $P$ with values in the Lie algebra $\mathfrak{u}\left(1\right)=i\mathbb{R}$
which is invariant by the action of $\mathbf{U}(1)$ and normalized
so that 
\begin{equation}
A\left(\frac{\partial}{\partial\theta}\right)=i\label{eq:normalization_A}
\end{equation}
 where $\frac{\partial}{\partial\theta}$ denotes the vector field
on $P$ generating the action of $\mathbf{U}(1)$. Consequently the
pull-back of the connection $A$ on $P$ by the trivialization map
(\ref{eq:local_trivialization-1}) is written as
\begin{equation}
T_{\alpha}^{*}A=id\theta-i2\pi\eta_{\alpha}\label{eq:def_eta-1}
\end{equation}
where $\eta_{\alpha}\in C^{\infty}\left(U_{\alpha},\Lambda^{1}\right)$
is a one-form on $U_{\alpha}$ which depends on the choice of the
local section $\tau_{\alpha}$. A different local section $\tau_{\beta}:U_{\beta}\to P$
with $U_{\alpha}\bigcap U_{\beta}\neq\emptyset$ is written as $\tau_{\beta}=e^{i\chi}\tau_{\alpha}$
with using a function $\chi:U_{\alpha}\bigcap U_{\beta}\rightarrow\mathbb{R}$
and hence the connection $A$ pulled-back by the corresponding trivialization
$T_{\beta}$ is written as (\ref{eq:def_eta-1}) with  
\begin{equation}
\eta_{\beta}=\eta_{\alpha}-\frac{1}{2\pi}d\chi\quad\mbox{on \ensuremath{U_{\alpha}\bigcap U_{\beta}}}.\label{eq:gauge_transform}
\end{equation}
The \emph{curvature} of the connection $A$ is the two form $\Theta=dA$
on $P$. In the local trivialization (\ref{eq:local_trivialization-1}),
we have $T_{\alpha}^{*}\Theta=-i2\pi\left(d\eta_{\alpha}\right)$
and (\ref{eq:gauge_transform}) implies that $d\eta_{\alpha}=d\eta_{\beta}$.
Therefore the curvature two form is written as 
\[
\Theta=-i\left(2\pi\right)\left(\pi^{*}\tilde{\omega}\right)
\]
 where $\tilde{\omega}=d\eta_{\alpha}$ is a closed two form on $M$
independent of the trivialization.

Since there is a given symplectic two form $\omega$ on $M$ in our
setting, we naturally require below in (\ref{eq:Curvature_omega})
that the two form $\tilde{\omega}$ coincides with $\omega$ and then
\begin{equation}
\omega=d\eta_{\alpha}.\label{eq:omega_eta}
\end{equation}

For the construction of the prequantum bundle and prequantum transfer
operator, we will need the following two assumptions:

%red
\begin{center}{\color{red}\fbox{\color{black}\parbox{16cm}{
\begin{description}
\item [{Assumption}] 1: The cohomology class $\left[\omega\right]\in H^{2}\left(M,\mathbb{R}\right)$
represented by the symplectic form $\omega$ is integral, that is,
\begin{equation}
\left[\omega\right]\in H^{2}\left(M,\mathbb{Z}\right)\label{eq:integral_assumption-1}
\end{equation}

\item [{Assumption}] 2: The integral homology group $H_{1}\left(M,\mathbb{Z}\right)$
has no torsion part and $1$ is not an eigenvalue of the linear map
$f_{*}:H_{1}\left(M,\mathbb{R}\right)\rightarrow H_{1}\left(M,\mathbb{R}\right)$
induced by $f:M\rightarrow M$.
\end{description}
}}}\end{center}
\begin{rem}
 The second assumption above is not restrictive and may not be necessary.
In fact this hypothesis is conjectured to be true in general. For
the case $M=\mathbb{T}^{2d}$, this is always satisfied.
\end{rem}
%blue
\begin{center}{\color{blue}\fbox{\color{black}\parbox{16cm}{
\begin{thm}
\label{thm:Bundle-P_map_f_tilde} Under Assumption 1 above, there
exists a $\mathbf{U}(1)$-principal bundle $\pi:P\rightarrow M$ and
a connection $A\in C^{\infty}\left(P,\Lambda^{1}\otimes\left(i\mathbb{R}\right)\right)$
on $P$ such that the curvature two form $\Theta=dA$ satisfies
\begin{equation}
\Theta=-i\left(2\pi\right)\left(\pi^{*}\omega\right).\label{eq:Curvature_omega}
\end{equation}
If we put Assumption 2 in addition, we can choose the connection $A$
as above so that there exists an equivariant lift $\tilde{f}:P\rightarrow P$
of the map $f:M\rightarrow M$ preserving the connection $A$, that
is:
\begin{equation}
\left(\pi\circ\tilde{f}\right)\left(p\right)=\left(f\circ\pi\right)\left(p\right),\quad\forall p\in P\qquad:\tilde{f}\mbox{ is a lift of \ensuremath{f}.}\label{eq:lift_of_f}
\end{equation}
\begin{equation}
\tilde{f}\left(e^{i\theta}p\right)=e^{i\theta}\tilde{f}\left(p\right),\quad\forall p\in P,\forall\theta\in\mathbb{R}\qquad:\mbox{ \ensuremath{\ensuremath{\tilde{f}}}\:\ is equivariant w.r.t. the \ensuremath{\mathbf{U}(1)}action. }\label{eq:equivariance_f_tilde}
\end{equation}
\begin{equation}
\tilde{f}^{*}A=A\qquad:\mbox{\ensuremath{\tilde{f}}\:\ preserves the connection \ensuremath{A}.}\label{eq:preserve_connection}
\end{equation}

\end{thm}
(See Figure \ref{fig:map f_tilde}.)

}}}\end{center}

The proof of Theorem \ref{thm:Bundle-P_map_f_tilde} is given in Section
\ref{sub:Proof-of-Theorem_bundle}.

%red
\begin{center}{\color{red}\fbox{\color{black}\parbox{16cm}{
\begin{defn}
The $\mathbf{U}(1)$-principal bundle $\pi:P\rightarrow M$ equipped
with the connection $A\in C^{\infty}\left(P,\Lambda^{1}\otimes\left(i\mathbb{R}\right)\right)$
satisfying (\ref{eq:Curvature_omega}) is called \emph{prequantum
bundle} over the symplectic manifold $(M,\omega)$. The map $\tilde{f}:P\rightarrow P$
satisfying the conditions (\ref{eq:lift_of_f}),(\ref{eq:equivariance_f_tilde})
and (\ref{eq:preserve_connection}) is called \emph{prequantum map}. 
\end{defn}
}}}\end{center}

\begin{figure}[h]
\centering{}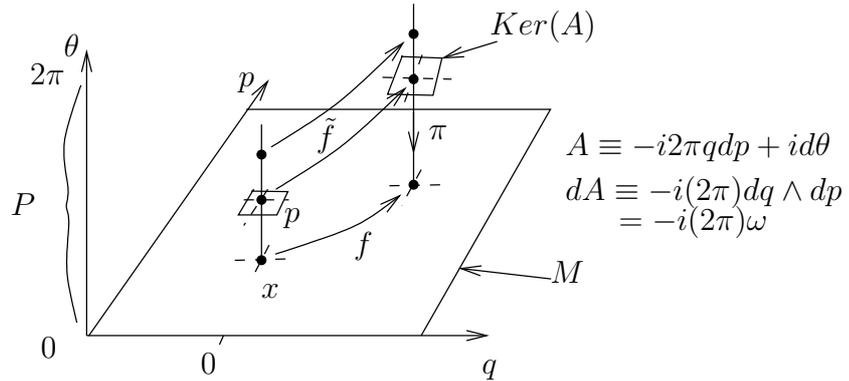\caption{\label{fig:map f_tilde} A picture of the prequantum bundle $P\rightarrow M$
in the case of $M=\mathbb{T}^{2}$, with connection one form $A$
and the prequantum map $\tilde{f}:P\rightarrow P$ which is a lift
of $f:M\rightarrow M$. A fiber $P_{x}\equiv U\left(1\right)$ over
$x\in M$ is represented here as a segment $\theta\in[0,2\pi[$. The
plane at a point $p$ represents the horizontal space $H_{p}P=\mbox{Ker}\left(A_{p}\right)$
which is preserved by $\tilde{f}$. These plane form a non integrable
distribution with curvature given by the symplectic form $\omega$.}
\end{figure}

\begin{rem}
\textbf{(}Uniqueness of the prequantum bundle and the prequantum map\textbf{)}
The prequantum bundle $P$ is unique (as a smooth manifold) if it
exists, because it is determined by its first Chern class $c_{1}\left(P\right)=\left[\omega\right]\in H^{2}\left(M,\mathbb{Z}^{2}\right)$.
However the connection $A$ on the prequantum bundle $P$ is not unique.
In the proof of the theorem above, we will explicitly show that there
may be finitely many connections $A$ which satisfy the condition
(\ref{eq:Curvature_omega}) and they differ from each other by a flat
connection. Once the prequantum bundle $P$ and the connection $A$
on it is given, the lifted map $\tilde{f}$ is unique up to a global
phase $e^{i\theta_{0}}\in\mathbf{U}(1)$, i.e. another map $\tilde{g}$
satisfies the conditions in (2) of Theorem \ref{thm:Bundle-P_map_f_tilde}
if and only if $\tilde{g}=e^{i\theta_{0}}\tilde{f}$ for some $e^{i\theta_{0}}\in\mathbf{U}(1)$.
\begin{rem}
\label{Rem:contact_form}Let $\alpha:=\frac{i}{2\pi}A$. Then the
differential $(2d+1)$-form 
\begin{equation}
\mu_{P}:=\frac{1}{d!}\alpha\wedge\left(d\alpha\right)^{d}\label{eq:volume_form_on_P}
\end{equation}
is a non-degenerate volume form on $P$. This means that $\alpha$
is a \emph{contact one form} on $P$ preserved by $\tilde{f}$.
\end{rem}
\end{rem}

\begin{rem}
\label{Rem:Action_of_PerOrb}Suppose that $x\in M$ is a periodic
point of the map $f$ with period $n\in\mathbb{N}$, $n\geq1$, i.e
$x=f^{n}\left(x\right)$. Then if $p\in\pi^{-1}\left(x\right)$ is
in the fiber, the condition (\ref{eq:lift_of_f}) implies that $\left(\tilde{f}\right)^{n}\left(p\right)\in\pi^{-1}\left(x\right)$
lies in the same fiber and therefore differ by a phase:
\begin{equation}
\left(\tilde{f}\right)^{n}\left(p\right)=e^{i2\pi S_{n,x}}p\label{eq:def_action_Snx}
\end{equation}
with $S_{n,x}\in\mathbb{R}/\mathbb{Z}$ called the \textbf{action
of the periodic point $x$}, see Figure \ref{fig:Action_periodic_orbit}.
These actions are important quantities in semiclassical analysis and
will appear in the Gutzwiller trace formula in (\pageref{eq:Gutz_formula_large_time}).
\end{rem}
\begin{figure}[h]
\centering{}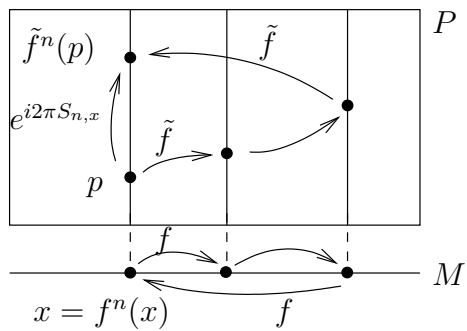\caption{\label{fig:Action_periodic_orbit}Action of a periodic point $x=f^{n}\left(x\right)$.}
\end{figure}

\subsubsection{The prequantum transfer operator $\hat{F}$ and the reduced operator
$\hat{F}_{N}$}

As usual in dynamical system theory, we consider the transfer operator
associated to the prequantum map $\tilde{f}$:

%red
\begin{center}{\color{red}\fbox{\color{black}\parbox{16cm}{
\begin{defn}
\label{def:Prequant_op}Let $V\in C^{\infty}\left(M\right)$ be a
real-valued smooth function, called \emph{potential}. The\textbf{
}\emph{prequantum transfer operator} is defined as
\begin{equation}
\hat{F}:\begin{cases}
C^{\infty}\left(P\right) & \rightarrow C^{\infty}\left(P\right)\\
u & \rightarrow\hat{F}(u)=e^{V\circ\pi}\left(u\circ\tilde{f}^{-1}\right)
\end{cases}\label{eq:def_prequantum_operator_F}
\end{equation}
where $V\circ\pi\in C^{\infty}\left(P\right)$ is the function $V$
lifted on $P$.
\end{defn}
}}}\end{center}
\begin{rem}
The fact that $\tilde{f}^{-1}$ appears instead of $\tilde{f}$ in
(\ref{eq:def_prequantum_operator_F}) is a matter of choice. In our
choice, $\tilde{f}$ maps the support of $u$ to that of $\hat{F}u$.
\end{rem}
From the equivariance property (\ref{eq:equivariance_f_tilde}), the
prequantum transfer operator commutes with the action of $U(1)$ on
functions on $P$ and therefore is naturally decomposed into each
Fourier mode with respect to the $U(1)$ action:

%red
\begin{center}{\color{red}\fbox{\color{black}\parbox{16cm}{
\begin{defn}
For a given $N\in\mathbb{Z},$ we consider the space of functions
in the $N$-th Fourier mode 
\begin{equation}
C_{N}^{\infty}\left(P\right):=\left\{ u\in C^{\infty}\left(P\right)\,\mid\,\forall p\in P,\forall\theta\in\mathbb{R},\quad u\left(e^{i\theta}p\right)=e^{iN\theta}u\left(p\right)\right\} .\label{eq:def_C_N_P}
\end{equation}
The prequantum transfer operator $\hat{F}$ restricted to $C_{N}^{\infty}\left(P\right)$
is denoted by:
\begin{equation}
\hat{F}_{N}:=\hat{F}_{/C_{N}^{\infty}\left(P\right)}:\quad C_{N}^{\infty}\left(P\right)\rightarrow C_{N}^{\infty}\left(P\right).\label{eq:def_F_N_on_P}
\end{equation}

\end{defn}
}}}\end{center}
\begin{rem}
The complex conjugation maps $C_{N}^{\infty}\left(P\right)$ to $C_{-N}^{\infty}\left(P\right)$
and commutes with $\hat{F}$. It is therefore enough to study $\hat{F}_{N}$
with $N\geq0$. 
\begin{rem}
\label{Rem:line_bundle_terminology}The space of equivariant functions
$C_{N}^{\infty}\left(P\right)$ defined in (\ref{eq:def_C_N_P}) can
be identified with the space of smooth sections of an associated Hermitian
complex line bundle $L^{\otimes N}$ over $M$ (i.e. the $N$ tensor
power of a line bundle $L\rightarrow M$) with covariant derivative
$D$, called the \emph{prequantum line bundle}\textbf{ }i.e. we have
\[
C_{N}^{\infty}\left(P\right)\cong C^{\infty}\left(M,L^{\otimes N}\right).
\]
See \cite[p.502, eq.(6.1)]{taylor_tome2}. In order to simplify the
presentation we will not use this identification in this paper although
it will be present implicitly. Notice however that most of references
about geometric quantization use the ``line bundle terminology''.
\end{rem}
\end{rem}
In this paper the main object of study is the resonance spectrum of
the operator $\hat{F}_{N}$, (\ref{eq:def_F_N_on_P}), in the limit
$N\rightarrow\infty$. For $N>0$, we set
\begin{equation}
\hbar=\frac{1}{2\pi N}.\label{eq:def_hbar}
\end{equation}
This new variable $\hbar$ is in one-to-one correspondence to $N$,
and $\hbar\to+0$ as $N\to\infty$. We introduce it for convenience
in referring some argument in semi-classical analysis where $\hbar$
is regarded as the Plank constant and the limit $\hbar\to+0$ is considered. 
\begin{rem}
In the following, we will confuse the parameters $N$ and $\hbar$
in the notation. For instance, the operator $\hat{F}_{N}$ will be
written $\hat{F}_{\hbar}$ sometimes.\end{rem}

\subsection{\label{sub:1.3}Results on the spectrum of the prequantum operator
$\hat{F}_{N}$}

The following theorem has been obtained essentially in the works of
Rugh \cite{rugh_92}, Liverani et al.\cite{liverani_02,liverani_04},
Baladi et al.\cite[Theorem 1.1]{baladi_05}, Faure et al. \cite[theorem 1]{fred-roy-sjostrand-07}.
The method employed in the present paper is close to the semiclassical
approach given in \cite[Theorem 1]{fred-roy-sjostrand-07}. Before
giving the theorem, let us mention that the transfer operator $\hat{F}_{N}$
has been defined on the space of smooth functions $C_{N}^{\infty}\left(P\right)$
and can be extended by duality to the distributions space $\mathcal{D}_{N}'\left(P\right)$.

%blue
\begin{center}{\color{blue}\fbox{\color{black}\parbox{16cm}{
\begin{thm}
\textbf{\label{thm:Discrete-spectrum}``Discrete spectrum of prequantum
transfer operators''.} For any $N\in\mathbb{Z}$, there exists a
family of Hilbert spaces $\mathcal{H}_{N}^{r}\left(P\right)$ for
arbitrarily large $r>0$, called anisotropic Sobolev space such that
\[
C_{N}^{\infty}\left(P\right)\subset\mathcal{H}_{N}^{r}\left(P\right)\subset\mathcal{D}'_{N}\left(P\right)
\]
 and such that the operator $\hat{F}_{N}$ extends to a bounded operator
\[
\hat{F}_{N}:\mathcal{H}_{N}^{r}\left(P\right)\rightarrow\mathcal{H}_{N}^{r}\left(P\right),
\]
and its essential spectral radius $r_{ess}\left(\hat{F}_{N}\right)$
is bounded by $\varepsilon_{r}:=\frac{1}{\lambda^{r}}\max e^{V}$,
which shrinks to zero if $r\rightarrow+\infty$. The discrete eigenvalues
of $\hat{F}_{N}$ on the domain $\left|z\right|\geq\varepsilon_{r}$
(and their associated eigenspaces) are independent on the choice of
$r$ and are therefore intrinsic to the Anosov diffeomorphism $f$.
These discrete eigenvalues $\mbox{Res}\left(\hat{F}_{N}\right):=\left\{ \lambda_{i}\right\} _{i}\subset\mathbb{C}^{*}$
are called \textbf{Ruelle-Pollicott resonances}. The definition of
the space $\mathcal{H}_{N}^{r}\left(P\right)$ depends on the diffeomorphism
Anosov $f$ but do not depend on the potential function $V$.
\end{thm}
}}}\end{center}
\begin{rem}
See \cite{Baladi-Tsujii08},\cite[Cor. 1.3]{fred-roy-sjostrand-07}
for a general argument about this independence of $\mbox{Res}\left(\hat{F}_{N}\right)$
on the choice of $r$.
\end{rem}
\begin{figure}[h]
\begin{centering}
\scalebox{0.9}[0.9]{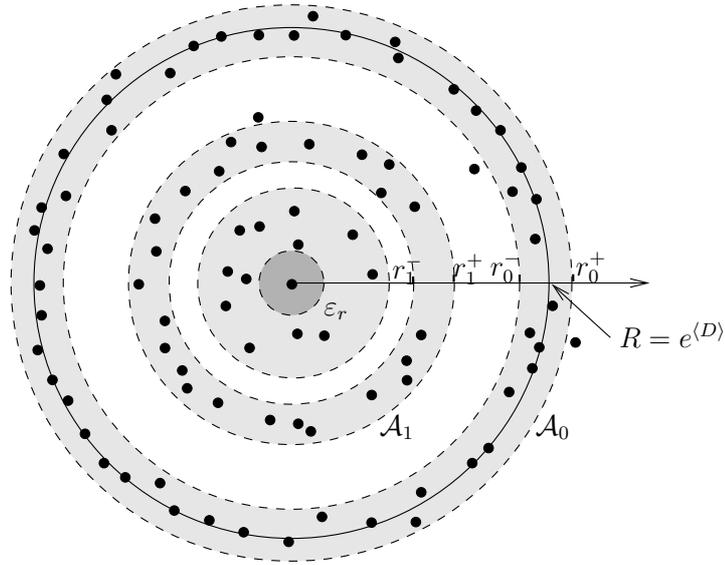}
\par\end{centering}

\caption{\label{fig:annuli}Theorem \ref{thm:Discrete-spectrum} shows that
the spectrum of \textbf{$\hat{F}_{N}$} has discrete eigenvalues called
Ruelle-Pollicott resonances. Theorem \ref{thm:band_structure} shows
that for $N$ large enough it is structured in bands and that the
resolvent is bounded between the bands, uniformly with respect to
$\hbar=1/\left(2\pi N\right)$. In the external band $\mathcal{A}_{0}$,
the number of resonances is given by the Weyl formula $\frac{1}{\left(2\pi\hbar\right)^{d}}\mathrm{Vol}\left(M\right)$
at leading order. The precise number is given by the Atiyah-Singer
formula in Theorem \ref{thm:Weyl-formula}. Theorem \ref{Thm:Distribution-of-resonances.}
shows that in this external band $\mathcal{A}_{0}$ almost all the
resonances are distributed uniformly on the circle of radius $R=e^{\left\langle D\right\rangle }$
in the limit $N\rightarrow\infty$. The spectral restriction of the
operator \textbf{$\hat{F}_{N}$} on this external band will be called
the \textbf{quantum operator }and denoted $\hat{\mathcal{F}}_{\hbar}$
in Definition \ref{def:quantum_operator}. The spectral projector
is $\Pi_{\hbar}$.}
\end{figure}

The main new result of this paper is the following Theorem . It is
illustrated in Figure \ref{fig:annuli}. Let us prepare some notations.
For a linear invertible map $L$ we will use the notation
\begin{equation}
\left\Vert L\right\Vert _{\mathrm{max}}:=\left\Vert L\right\Vert ,\qquad\left\Vert L\right\Vert _{\mathrm{min}}:=\left\Vert L^{-1}\right\Vert ^{-1}.\label{eq:def_of_max_min_norm}
\end{equation}

We define the special ``potential of reference'' 
\begin{equation}
V_{0}\left(x\right):=\frac{1}{2}\log\left|\det\, Df_{x}|_{E_{u}\left(x\right)}\right|\label{eq:def_Potential_V0}
\end{equation}
Notice that the unstable foliation $E_{u}\left(x\right)$ is not smooth
in $x$ in general (see Remark \ref{Remark:Holder_exp_beta}(1)) which
implies that this function $V_{0}$ is Hölder continuous but not smooth
in $x$. We then consider the difference 
\begin{equation}
D:=V-V_{0}\quad\in C^{\beta}\left(M\right)\label{eq:damping_function}
\end{equation}
which is also a Hölder continuous function on $M$ and that will be
called the ``\textbf{effective damping function}''. It will appear
in many results below. Finally we denote by 
\begin{equation}
D_{n}\left(x\right):=\sum_{j=1}^{n}D\left(f^{j}\left(x\right)\right)\label{eq:Birkhoff_Sum}
\end{equation}
 the Birkhoff sum of the damping function.

%blue
\begin{center}{\color{blue}\fbox{\color{black}\parbox{16cm}{
\begin{thm}
\textbf{\label{thm:band_structure}``Band structure of the spectrum
of $\hat{F}_{N}$''}. For any $\varepsilon>0$, there exists $C_{\varepsilon}>0$,
$N_{\varepsilon}\geq1$ such that for any $N\geq N_{\varepsilon}$\end{thm}
\begin{enumerate}
\item the Ruelle-Pollicott resonances of $\hat{F}_{N}$ is contained in
a small neighborhood of the union of annuli $\left(\mathcal{A}_{k}:=\{r_{k}^{-}\leq\left|z\right|\leq r_{k}^{+}\}\right)_{k\geq0}$:
\begin{equation}
\mbox{Res}\left(\hat{F}_{N}\right)\subset\bigcup_{k\geq0}\underbrace{\left\{ r_{k}^{-}-\varepsilon\leq\left|z\right|\leq r_{k}^{+}+\varepsilon\right\} }_{\mbox{\ensuremath{\varepsilon}-neighborhood of }\mathcal{A}_{k}}\label{eq:spectrum_in_annuli}
\end{equation}
with 
\begin{eqnarray}
r_{k}^{-} & := & \liminf_{n\rightarrow\infty}\inf_{x\in M}\left(e^{\frac{1}{n}D_{n}\left(x\right)}\left\Vert Df_{x}^{n}|_{E_{u}}\right\Vert _{\mathrm{max}}^{-k/n}\right),\label{eq:def_r+_r-}\\
r_{k}^{+} & := & \limsup_{n\rightarrow\infty}\sup_{x\in M}\left(e^{\frac{1}{n}D_{n}\left(x\right)}\left\Vert Df_{x}^{n}|_{E_{u}}\right\Vert _{\mathrm{min}}^{-k/n}\right)\nonumber 
\end{eqnarray}

\item Suppose that $r_{k}^{+}<r_{k-1}^{-}$ for some $k\geq1$. For any
$z\in\mathbb{C}$ such that $r_{k}^{+}+\varepsilon<\left|z\right|<r_{k-1}^{-}-\varepsilon$,
i.e. such that $z$ is in a ``gap'', the resolvent of $\hat{F}_{N}$
on $\mathcal{H}_{N}^{r}\left(P\right)$ is controlled uniformly with
respect to $N$:
\begin{equation}
\left\Vert \left(z-\hat{F}_{N}\right)^{-1}\right\Vert \leq C_{\varepsilon}\label{eq:bound_resolvent}
\end{equation}
This is true also for $\left|z\right|>r_{0}^{+}+\varepsilon$.
\item If $r_{1}^{+}<r_{0}^{-}$, i.e. if the outmost annulus $\mathcal{A}_{0}$
is\textbf{ }isolated from other annuli, then the number of resonances
in its neighborhood satisfies the estimate called ``\textbf{Weyl
formula}'' 
\begin{equation}
\sharp\left\{ \mbox{Res}\left(\hat{F}_{N}\right)\bigcap\left\{ r_{0}^{-}-\varepsilon\leq\left|z\right|\leq r_{0}^{+}+\varepsilon\right\} \right\} =N^{d}\mathrm{Vol}_{\omega}\left(M\right)\left(1+O\left(N^{-\delta}\right)\right).\label{eq:weyl_law}
\end{equation}
with $\mbox{Vol}_{\omega}\left(M\right):=\int_{M}\frac{1}{d!}\omega^{\wedge d}$
being the symplectic volume of $M$ and $\delta>0$.
\end{enumerate}
}}}\end{center}

The proof of Theorem \ref{thm:band_structure} is given in Section
\ref{sec:Proofs-of-the_mains}. There in Theorem \ref{thm:More-detailled-description_of_spectrum},
we will provide a more detailed but more technical result about the
operators $\hat{F}_{N}$, without any assumption on existence of gaps
$r_{k}^{+}<r_{k-1}^{-}$. The Weyl law (\ref{eq:weyl_law}) is obtained
in Corollary \ref{cor:Weyl_law}. We will give a more precise value
in Theorem \ref{thm:Weyl-formula} below.
\begin{rem}
\label{remarks_1.17}~
\begin{enumerate}
\item Since $\left\Vert Df_{x}^{n}|_{E_{u}}\right\Vert _{\mathrm{max}}^{1/n}\geq\left\Vert Df_{x}^{n}|_{E_{u}}\right\Vert _{\mathrm{min}}^{1/n}>\lambda>1$,
from (\ref{eq:def_dynamics}), we have obviously $r_{k+1}^{-}<r_{k}^{-}$
and $r_{k+1}^{+}<r_{k}^{+}$ for every $k\geq0$. However we don't
always have $r_{k+1}^{+}<r_{k}^{-}$ therefore the annuli $\mathcal{A}_{k}$
may intersect each other.
\item From Theorem \ref{thm:band_structure} it is tempting to take the
potential $V=V_{0}$ defined in (\ref{eq:def_Potential_V0}) which
would indeed give $D=0$ hence $r_{0}^{+}=r_{0}^{-}=1$ in (\ref{eq:def_r+_r-}).
In that case the external band $\mathcal{A}_{0}$ would be the unit
circle, separated from the internal band $\mathcal{A}_{1}$ by a spectral
gap $r_{1}^{+}$ given by
\[
r_{1}^{+}=\limsup_{n\rightarrow\infty}\sup_{x\in M}\left(\left\Vert Df_{x}^{n}|_{E_{u}}\right\Vert _{\mathrm{min}}^{-1/n}\right)<\frac{1}{\lambda}<1
\]
However Theorem \ref{thm:band_structure} does not apply in this case
because the function $V_{0}$ is not smooth in $x$ as required. This
is the purpose of the next Section \ref{sub:Spectral-results-with-Grassman}
to show how to handle this non smooth potential $V_{0}$ and still
get spectral results similar to Theorem \ref{thm:band_structure}.
For the moment let us remark that if $\widetilde{E}_{u}\subset TM$
is a smooth approximation of the unstable sub-bundle $E_{u}\subset TM$
in $C^{0}$ norm and if one chooses the potential: 
\begin{equation}
V_{0}\left(x\right)=\frac{1}{2}\log\left|\det\, Df_{x}|_{\tilde{E}_{u}\left(x\right)}\right|\label{eq:choice_V_smooth}
\end{equation}
then on can have $r_{0}^{-},r_{0}^{+}$ (arbitrarily) closed to one
and the annulus $\mathcal{A}_{0}$ of the external band get isolated
from the other ones ($\mathcal{A}_{0}\bigcap\mathcal{A}_{k}=\emptyset$
for $k\neq0$).
\item In the simple case of the linear hyperbolic map on the torus $\mathbb{T}^{2}$
in (\ref{eq:f0_cat_map}) with $V\left(x\right)=0$, we have $r_{k}^{+}=r_{k}^{-}=\lambda^{-k-\frac{1}{2}}$
with $\lambda=Df_{0/E_{u}}=\frac{3+\sqrt{5}}{2}\simeq2.6$ (constant),
and each annulus $\mathcal{A}_{k}$ is a circle. In this case Theorem
\ref{thm:band_structure} has been obtained in \cite{fred-PreQ-06}
and is depicted on Figure (1-b) in \cite{fred-PreQ-06}. If one chooses
$V\left(x\right)=\frac{1}{2}\log\left|\mbox{det}Df_{x}|_{E_{u}}\right|=\frac{1}{2}\log\lambda$
the external band $\mathcal{A}_{0}$ is the unit circle and it is
shown in \cite{fred-PreQ-06} that the Ruelle-Pollicott resonances
on the external band coincide with the spectrum of the quantized map
called the ``quantum cat map''.
\item The estimate (\ref{eq:bound_resolvent}) will be useful in Section
\ref{sub:Emergence-of-quantum} to express dynamical correlation functions.
\item For a given $\varepsilon>0$, in (\ref{eq:spectrum_in_annuli}), only
a finite number of annuli $\mathcal{A}_{k}$ can be distinguished. 
\end{enumerate}
\end{rem}
From now on, we suppose that $r_{1}^{+}<r_{0}^{-}$, i.e. that the
external annulus $\mathcal{A}_{0}$ defined in (\ref{eq:spectrum_in_annuli})
is isolated from other annuli $\bigcup_{k\geq1}\mathcal{A}_{k}$.
We have seen in (\ref{eq:choice_V_smooth}), how to achieve this situation
by a suitable choice of the potential $V\left(x\right)$.

%red
\begin{center}{\color{red}\fbox{\color{black}\parbox{16cm}{
\begin{defn}
\label{def:quantum_operator}Assume $r_{1}^{+}<r_{0}^{-}$. For $N=1/\left(2\pi\hbar\right)$
large enough let
\begin{equation}
\Pi_{\hbar}:\mathcal{H}_{N}^{r}\left(P\right)\rightarrow\mathcal{H}_{N}^{r}\left(P\right)\label{eq:def_Pi_hbar}
\end{equation}
be the spectral projector of the operator $\hat{F}_{N}$ on its external
band, i.e. on the spectral domain $\left\{ r_{0}^{-}-\varepsilon\leq\left|z\right|\leq r_{0}^{+}+\varepsilon\right\} $.
Let
\[
\mathcal{H}_{\hbar}:=\mathrm{Im}\left(\Pi_{\hbar}\right)
\]
that we call the \textbf{``quantum space''} which is finite dimensional
from (\ref{eq:weyl_law}) and let
\begin{equation}
\hat{\mathcal{F}}_{\hbar}:\mathcal{H}_{\hbar}\rightarrow\mathcal{H}_{\hbar}\label{eq:def_quantum_op}
\end{equation}
be the finite dimensional spectral restriction of $\hat{F}_{N}$ on
the exterior annulus $\mathcal{A}_{0}$. We call $\hat{\mathcal{F}}_{\hbar}$
the \textbf{``quantum operator''.}
\end{defn}
}}}\end{center}
\begin{rem}
We will justify in Section \ref{sub:Emergence-of-quantum} this name
of ``quantum operator''.
\end{rem}
In (\ref{eq:weyl_law}), Weyl law gives the leading order for the
value of $\dim\mathcal{H}_{\hbar}$ . The next Theorem gives its exact
value.

%blue
\begin{center}{\color{blue}\fbox{\color{black}\parbox{16cm}{
\begin{thm}
\textbf{\label{thm:Weyl-formula}``Index formula for the number of
resonances''}. If the external annulus $\mathcal{A}_{0}$ is isolated,
i.e. $r_{1}^{+}<r_{0}^{-}$, then the number of resonances in the
external annulus $\mathcal{A}_{0}$ is given by the \emph{Atiyah-Singer
index formula:} for $N$ large enough, 
\begin{equation}
\dim\mathcal{H}_{\hbar}=\int_{M}\left[e^{N\omega}\mbox{Todd}\left(TM\right)\right]_{2d}\label{eq:Atiyah-Singer}
\end{equation}
where 
\[
e^{N\omega}=1+N\omega+\ldots+\frac{N^{d}\omega^{d}}{d!}
\]
 is the Chern character and 
\[
\mbox{Todd}\left(TM\right)=\mbox{det}\left(\frac{\Omega\left(TM\right)}{1-\exp\left(-\Omega\left(TM\right)\right)}\right)=1+\frac{\Omega\left(TM\right)}{2}+\ldots\in H_{DR}^{\bullet}\left(M\right)
\]
 is the Todd class of the tangent bundle defined from the Riemannian
curvature $\Omega\left(TM\right)$ and $\left[.\right]_{2d}$ denotes
the restriction to volume $2d$-forms. Consequently we recover ``Weyl
formula'' of (\ref{eq:weyl_law}) and the remainder is also better:
\begin{equation}
\dim\mathcal{H}_{\hbar}=N^{d}\mbox{Vol}_{\omega}\left(M\right)+\mathcal{O}\left(N^{d-1}\right)\label{eq:number_of_resonances_Weyl_formula}
\end{equation}

\end{thm}
}}}\end{center}

Theorem \ref{thm:Weyl-formula} above follows from Theorem \ref{thm:band_structure-of_Laplacian}
where we will introduce a differential operator $\Delta=D^{*}D$ acting
in $C_{N}^{\infty}\left(P\right)$ called the \emph{rough Laplacian}.
In Theorem \ref{thm:band_structure-of_Laplacian}, we will show that
its low energy spectrum has band spectrum and that the cardinality
of the eigenvalues in the first (i.e. the lowest) band equals the
quantity on the right hand side of the formula (\ref{eq:Atiyah-Singer}).
The latter is actually a consequence of a theorem in geometry. We
will also show that the rank of the projector $\Pi_{0}$ coincides
with the rank of the spectral projector for eigenvalues in the first
band. We thus obtain the formula (\ref{eq:Atiyah-Singer}). Then the
Weyl formula (\ref{eq:number_of_resonances_Weyl_formula}) is a direct
consequence. Indeed we have $\left[e^{N\omega}\mbox{Todd}\left(TM\right)\right]_{2d}=\frac{N^{d}\omega^{d}}{d!}+O\left(N^{d-1}\right)$
and hence 
\[
\int_{M}\left[e^{N\omega}\mbox{Todd}\left(TM\right)\right]_{2d}=N^{d}\mbox{Vol}_{\omega}\left(M\right)+\mathcal{O}\left(N^{d-1}\right).
\]

\begin{rem}
In the case of $M=\mathbb{T}^{2}$ which correspond to example (\ref{eq:f0_cat_map})
and treated in \cite{fred-PreQ-06}, the projector $\Pi_{0}$ has
exactly $\mbox{rank}\left(\Pi_{0}\right)=N$. Indeed, for Riemann
surfaces $M$ of genus $g$, we have $\mbox{Todd}\left(TM\right)=1+\frac{c_{1}\left(TM\right)}{2}$
with first Chern number $\int_{M}c_{1}\left(TM\right)=2-2g$ (the
Gauss-Bonnet integral formula). Hence $\mbox{rank}\left(\Pi_{0}\right)=\int_{M}\left(N\omega\right)+\int_{M}c_{1}\left(TM\right)=N$
for $M=\mathbb{T}^{2}$ with genus $g=1$.
\end{rem}
For the next Theorem, recall the definition of the damping function
$D\left(x\right)=V\left(x\right)-V_{0}\left(x\right)$ in (\ref{eq:damping_function}).

%blue
\begin{center}{\color{blue}\fbox{\color{black}\parbox{16cm}{
\begin{thm}
\textbf{\label{Thm:Distribution-of-resonances.}``Distribution of
resonances''}. Assume $r_{1}^{+}<r_{0}^{-}$. In the limit $\hbar\rightarrow0$,
most of eigenvalues of $\hat{\mathcal{F}}_{\hbar}$ \textbf{concentrate
and equidistribute} on the circle of radius
\begin{equation}
R:=e^{\left\langle D\right\rangle },\quad\mbox{with }\left\langle D\right\rangle :=\frac{1}{\mbox{Vol}_{\omega}\left(M\right)}\int_{M}D\left(x\right)dx\label{eq:<R>}
\end{equation}
(See Figure \ref{fig:annuli}). More precisely, for any $\varepsilon>0$,
we have
\[
\lim_{N\to\infty}\frac{\sharp\left\{ \mathrm{Res}(\hat{F}_{N})\cap\{||z|-R|<\epsilon\}\right\} }{\sharp\{\mathrm{Res}(\hat{F}_{N})\}}=1
\]
and for any $0\le\theta_{1}<\theta_{2}\le2\pi$,
\[
\lim_{N\to\infty}\frac{\sharp\left\{ \mathrm{Res}(\hat{F}_{N})\cap\{\theta_{1}<\arg(z)<\theta_{2}\}\right\} }{\sharp\{\mathrm{Res}(\hat{F}_{N})\}}=\frac{\theta_{2}-\theta_{1}}{2\pi}.
\]
 
\end{thm}
}}}\end{center}

The proof of Theorem \ref{Thm:Distribution-of-resonances.} will be
given in Section \ref{sec:9} .
\begin{rem}
With some pinching conditions, it is possible to show that the resonances
in the internals bands $\mathcal{A}_{k}$ also concentrate uniformly
on circles.
\end{rem}
~
\begin{rem}
The proof of Theorem \ref{Thm:Distribution-of-resonances.} uses ergodicity
of the map $f:M\rightarrow M$ and follows a techniques presented
by J. Sjöstrand in \cite{sjostrand_2000} for the damped wave equation.
Using mixing and large deviations properties of the map $f$ it may
be possible to improve the results as in \cite{nalini_anantharam_2010}.\end{rem}

\subsection{\label{sub:Spectral-results-with-Grassman}Spectral results with
extended models on the Grassmanian bundle}

In this Section we extend the previous results for a family of prequantum
transfer operators more general than that considered in Theorem \ref{thm:band_structure}
in the sense that we will admit some functions $V$ for the potential,
defined in (\ref{eq:def_V}) below that may be only Hölder continuous.
This is the case of $V_{0}$ given in (\ref{eq:def_Potential_V0}).
The trick is to consider transfer operators associated to the dynamics
of $f:M\rightarrow M$ lifted on the $d$-dimensional Grassmanian
bundle $p:G_{d}\left(TM\right)\rightarrow M$ so that the potential
function $V$ on $M$ is derived from a smooth potential function
$\tilde{V}$ on $G_{d}\left(TM\right)$. We explain now this construction.

\subsubsection{The Grassmanian bundle $G\rightarrow M$ and the lifted map $f_{G}$}

At a given point $x\in M$ of the manifold $M$, recall from Remark
\ref{Remark:Holder_exp_beta}(1) that the stable and unstable linear
space $E_{s}\left(x\right)$, $E_{u}\left(x\right)\subset T_{x}M$,
$x\in M$ have dimension $d=\dim M/2$ each. For that reason we consider
all the $d$-dimensional linear subspaces of $T_{x}M$:
\begin{defn}
At a given point $x\in M$, the $d$-dimensional linear subspaces
of $T_{x}M$ form a compact manifold of dimension $d^{2}$ called
the \textbf{Grassmanian}%
\footnote{For $x\in M$, the Grassmanian is a homogeneous space
\[
G_{x}=G_{d}\left(T_{x}M\right)\equiv\frac{O\left(2d\right)}{O\left(d\right)\times O\left(d\right)},\quad\mathrm{dim}G_{x}=d^{2}.
\]
} $G_{d}\left(T_{x}M\right)$. The \textbf{Grassmanian bundle} is the
bundle
\begin{equation}
G_{d}\left(TM\right)\overset{p}{\longrightarrow}M\label{eq:def_Gd}
\end{equation}
whose base space is $M$ and the fiber over point $x\in M$ is the
Grassmanian $G_{d}\left(T_{x}M\right)$. For simplicity we will denote
$p:G\rightarrow M$ this bundle and $G_{x}:=G_{d}\left(T_{x}M\right)$
the fiber.
\end{defn}
~
\begin{defn}
The diffeomorphism $f:M\rightarrow M$ induces a natural lifted map
\begin{equation}
f_{G}=Df:\quad G\rightarrow G\label{eq:def_fG}
\end{equation}
 where (abusively) $Df_{x}$ denotes the induced action of the differential
$Df_{x}$ on linear subspaces $l\in G_{x}$. By definition we have
a commutative diagram:
\[
\begin{array}{ccc}
G & \overset{f_{G}}{\longrightarrow} & G\\
\downarrow p &  & \downarrow p\\
M & \overset{f}{\longrightarrow} & M
\end{array}
\]
\end{defn}
\begin{rem}
The stable $d$-dimensional bundle $E_{s}\rightarrow M$ defines a
Hölder continuous section $E_{s}$ of the bundle $G\rightarrow M$
and from Definition \ref{def:Anosov_diffeo}, $E_{s}$ is a repeller
for the map $f_{G}$. Similarly the unstable bundle $E_{u}$ is a
Hölder continuous section of the bundle $G\rightarrow M$. The image
of $E_{u}$ (for which we also write $E_{u}$) is an attractor for
the map $f_{G}$. See Figure \ref{fig:bundle_G}.
\end{rem}
\begin{figure}
\centering{}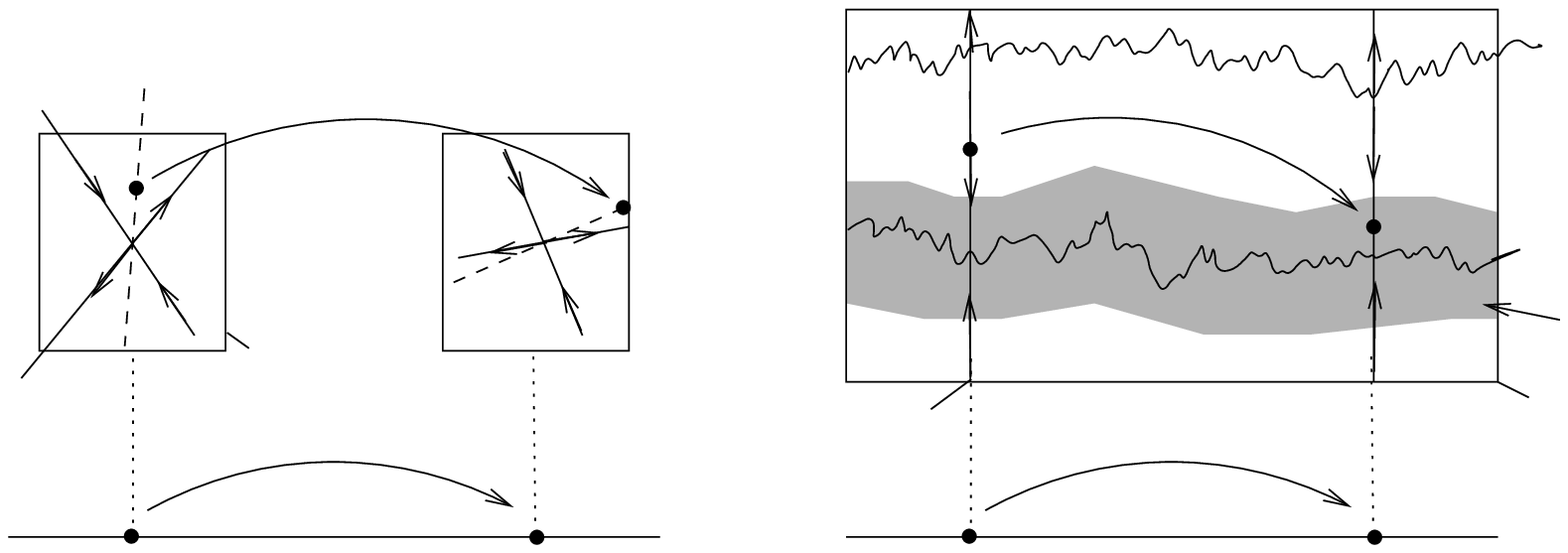\caption{\label{fig:bundle_G}}
\end{figure}

\subsubsection{Prequantum transfer operator $\tilde{F}_{N}$ on $P_{G}\rightarrow G$}

Let
\[
\tilde{V}\in C^{\infty}\left(G\right)
\]
be a smooth real valued function called \textbf{potential function}
and its restriction to the unstable bundle will be denoted by:
\begin{equation}
V:=\tilde{V}\circ E_{u}\in C^{\beta}\left(M\right)\label{eq:def_V}
\end{equation}
which is Hölder continuous with Hölder exponent $0<\beta\leq1$.
\begin{rem}
~
\begin{enumerate}
\item In Definition \ref{def:Prequant_op} we have considered the particular
case of a smooth function $V\in C^{\infty}\left(M\right)$ which can
be derived here from the function $\tilde{V}:=V\circ p\in C^{\infty}\left(G\right)$
(i.e. constant on the fibers). So results presented in this Section
will also apply to the prequantum operator (\ref{eq:def_prequantum_operator_F}).
\item The special potential function $V_{0}:=\tilde{V}_{0}\circ E_{u}$
in (\ref{eq:def_Potential_V0}) is derived as in (\ref{eq:def_V})
from the smooth function $\tilde{V}_{0}\in C^{\infty}\left(G\right)$
given by 
\begin{equation}
\tilde{V}_{0}\left(l\right)=\frac{1}{2}\log\left(\mathrm{det}Df_{x}\mid_{l}\right),\qquad l\in G,\quad x=p\left(l\right)\label{eq:V_tilde_0}
\end{equation}

\end{enumerate}
\end{rem}
Recall the principal $U\left(1\right)$ bundle $\pi:P\rightarrow M$
and the prequantum map $\tilde{f}:P\rightarrow P$ defined in Theorem
\ref{thm:Bundle-P_map_f_tilde}.
\begin{defn}
The principal bundle $\pi:P\rightarrow M$ with connection $A$ can
be pulled back by $p:G\rightarrow M$ on $G$ and gives a principal
bundle 
\[
\pi_{G}:P_{G}\rightarrow G
\]
 with connection. Let 
\begin{equation}
\tilde{f}_{G}:P_{G}\rightarrow P_{G}\label{eq:def_ftilde_G}
\end{equation}
 be the lift of the map $\tilde{f}:P\rightarrow P$ on $P_{G}\rightarrow P$. \end{defn}
\begin{rem}
Of course $\tilde{f}_{G}$ is also a lift of the map $f_{G}:G\rightarrow G$,
i.e $\pi_{G}\circ\tilde{f}_{G}=f_{G}\circ\pi_{G}$. By construction,
for any $x\in M$, the restricted bundle $P_{G}\rightarrow G_{x}$
is trivial. In fact we have a 3 dimensional commutative diagram:
\[
\xymatrix{P_{G}\ar[rr]^{\tilde{f}_{G}}\ar[dd]^{\tilde{\pi}}\ar[dr]^{\pi_{G}} &  & P_{G}\ar[dr]^{\pi_{G}}\ar[dd]|(0.3){\tilde{\pi}}|!{[dl];[dr]}\hole\\
 & G\ar[rr]|(0.3){f_{G}}\ar[dd]|(0.3){p} &  & G\ar[dd]^{p}\\
P\ar[rr]|(0.3){\tilde{f}}|!{[ur];[dr]}\hole\ar[dr]^{\pi} &  & P\ar[rd]^{\pi}\\
 & M\ar[rr]^{f} &  & M
}
\]

\end{rem}
%red 
\begin{center}{\color{red}\fbox{\color{black}\parbox{16cm}{
\begin{defn}
Let $\tilde{V}\in C^{\infty}\left(G\right)$ . The \textbf{prequantum
transfer operator} is defined by
\begin{equation}
\widetilde{F}:\begin{cases}
C^{\infty}\left(P_{G}\right) & \rightarrow C^{\infty}\left(P_{G}\right)\\
u & \mapsto e^{\tilde{V}}\cdot\left|\mbox{det}\left(Df_{G}\mid_{\mathrm{ker}p}\right)\right|^{-1}\cdot u\circ\tilde{f}_{G}^{-1}
\end{cases}\label{eq:def_transfert_gutzwiller}
\end{equation}
It preserves the space of the $N$-th Fourier modes for every $N\in\mathbb{Z}$:
\begin{equation}
C_{N}^{\infty}\left(P_{G}\right):=\left\{ u\in C^{\infty}\left(P_{G}\right)\,\mid\,\forall p\in P_{G},\forall e^{i\theta}\in U\left(1\right),\quad u\left(e^{i\theta}p\right)=e^{iN\theta}u\left(p\right)\right\} \label{eq:def_C_N_PG}
\end{equation}
and its restriction is denoted
\begin{equation}
\widetilde{F}_{N}:=\widetilde{F}_{/C_{N}^{\infty}\left(P_{G}\right)}:\quad C_{N}^{\infty}\left(P_{G}\right)\rightarrow C_{N}^{\infty}\left(P_{G}\right).\label{eq:def_F_tilde_N}
\end{equation}

\end{defn}
}}}\end{center}
\begin{rem}
In the definition (\ref{eq:def_transfert_gutzwiller}), the additional
``potential function'' $\left|\mbox{det}\left(Df_{G}\mid_{\mathrm{ker}p}\right)\right|^{-1}>1$
has been put in order to compensate the attraction effect of the attractor
$E_{u}$ on the extended space $G$. (See the next subsection.)
\end{rem}

\subsubsection{\label{sub:Truncation-near_Eu}Truncation in a neighborhood of $E_{u}$}

In order to define the discrete spectrum of resonances we first have
to consider a specific truncation of the operator. Let $K_{0}\subset G$
be an open absorbing neighborhood of the attractor $E_{u}$ and 
\begin{equation}
K_{1}:=f_{G}\left(K_{0}\right)\Subset K_{0}\label{eq:def_K1}
\end{equation}
 i.e. $\overline{K_{1}}$ is a proper subset of $K_{0}$. See Figure
\ref{fig:bundle_G}. For any $n\geq1$, let 
\begin{equation}
K_{n}:=f_{G}^{n}\left(K_{0}\right).\label{eq:def_Kn}
\end{equation}
Then, by definition (of the absorbing neighborhood), we have that
\begin{equation}
E_{u}=\bigcap_{n=1}^{\infty}K_{n}.\label{eq:E_u_K_n}
\end{equation}
 Let $\chi\in C^{\infty}\left(G\right)$ be a function such that $\chi\left(l\right)=0$
for $l\notin K_{0}$, $\chi\left(l\right)=1$ for $l\in K_{1}$. We
denote 
\begin{equation}
\hat{\chi}:C^{\infty}\left(P_{G}\right)\rightarrow C^{\infty}\left(P_{G}\right)\label{eq:def_Chi_hat}
\end{equation}
 the multiplication operator by the function $\chi\circ p$, where
$p:G\to M$ is the projection. For any $n\geq1$ we have from (\ref{eq:def_K1})
that

\begin{equation}
\left(\widetilde{F}\hat{\chi}\right)^{n}=\widetilde{F}^{n}\hat{\chi}\label{eq:F_Chi_n}
\end{equation}
Also $\hat{\chi}$ preserves the space of equivariant functions $C_{N}^{\infty}\left(P_{G}\right)$
defined in (\ref{eq:def_C_N_PG}). By duality the operator $\widetilde{F}_{N}\hat{\chi}$
extends to equivariant distributions $\mathcal{D}'_{N}\left(P_{G}\right)$.
From definition (\ref{eq:def_Kn}) and (\ref{eq:F_Chi_n}) we have
that, for any $n\geq1$ and $u\in\mathcal{D}'\left(P_{G}\right)$,
\begin{equation}
\mbox{supp}\left(\left(\widetilde{F}_{N}\hat{\chi}\right)^{n}u\right)\subset\pi_{G}^{-1}\left(K_{n}\right).\label{eq:supp_F_phi-1}
\end{equation}

\subsubsection{Results on the spectrum of the prequantum operator $\widetilde{F}_{N}$}

The following theorem (and its proof) is similar to Theorem \ref{thm:Discrete-spectrum}
but concerns the transfer operator $\widetilde{F}_{N}$ defined in
(\ref{eq:def_F_tilde_N}).

%blue
\begin{center}{\color{blue}\fbox{\color{black}\parbox{16cm}{
\begin{thm}
\textbf{\label{thm:Discrete-spectrum-F_G}``Discrete spectrum''.}
For every $N\in\mathbb{Z}$, there exists a family of Hilbert spaces
$\mathcal{H}_{N}^{r}\left(P_{G}\right)$ for arbitrarily large $r>0$,
such that $C_{N}^{\infty}\left(P_{G}\right)\subset\mathcal{H}_{N}^{r}\left(P_{G}\right)\subset\mathcal{D}'_{N}\left(P_{G}\right)$
and such that the operator $\widetilde{F}_{N}\hat{\chi}$ extends
to a bounded operator 
\[
\widetilde{F}_{N}\hat{\chi}:\mathcal{H}_{N}^{r}\left(P_{G}\right)\rightarrow\mathcal{H}_{N}^{r}\left(P_{G}\right),
\]
and its essential spectral radius $r_{ess}\left(\hat{F}_{N}\hat{\chi}\right)$
is bounded by $\varepsilon_{r}:=\frac{1}{\lambda^{r}}\max e^{\tilde{V}}$,
which shrinks to zero if $r\rightarrow+\infty$. The discrete eigenvalues
of $\hat{F}_{N}\hat{\chi}$ on the domain $\left|z\right|\geq\varepsilon_{r}$
(and their associated eigenspaces) are independent on the choice of
$\chi$ and $r$. The support of an eigendistribution is contained
in the attractor $\pi_{G}^{-1}\left(E_{u}\right)$. These discrete
eigenvalues are called \emph{Ruelle-Pollicott resonances} and are
denoted $\mbox{Res}\left(\widetilde{F}_{N}\right):=\left\{ \lambda_{i}\right\} _{i}\subset\mathbb{C}^{*}$.
\end{thm}
}}}\end{center}

The fact that the support of an eigendistribution is contained in
the attractor $\pi_{G}^{-1}\left(E_{u}\right)$ is a direct consequence
of (\ref{eq:supp_F_phi-1}) and (\ref{eq:E_u_K_n}).

The next Theorem is similar to Theorem \ref{thm:band_structure} but
here we restrict ourselves to the description of the external band
\[
\mathcal{A}_{0}:=\left\{ z\in\mathbb{C},\left|z\right|\in\left[r_{0}^{-},r_{0}^{+}\right]\right\} 
\]
although description of internal bands may be possible also. Recall
that $V\left(x\right)$ is defined in (\ref{eq:def_V}), $D\left(x\right):=V\left(x\right)-V_{0}\left(x\right)$
is the damping function and $D_{n}\left(x\right)=\sum_{j=1}^{n}D\left(f_{G}^{j}\left(x\right)\right)$is
the Birkhoff sum of $D\left(x\right)$.

%blue
\begin{center}{\color{blue}\fbox{\color{black}\parbox{16cm}{
\begin{thm}
\textbf{\label{thm:band_structure-1}``External band''}. For any
$\varepsilon>0$, there exists $C_{\varepsilon}>0$, $N_{\varepsilon}\geq1$
such that, for any $N\geq N_{\varepsilon}$, we have
\[
\mbox{Res}\left(\widetilde{F}_{N}\right)\subset\left\{ z\in\mathbb{C},\,\left|z\right|\in\left[0,r_{1}^{+}+\varepsilon\right]\bigcup\left[r_{0}^{-}-\varepsilon,r_{0}^{+}+\varepsilon\right]\right\} 
\]
with
\begin{eqnarray}
r_{0}^{-} & := & \liminf_{n\rightarrow\infty}\inf_{x\in M}\left(e^{\frac{1}{n}D_{n}\left(x\right)}\right),\quad r_{0}^{+}:=\limsup_{n\rightarrow\infty}\sup_{x\in M}\left(e^{\frac{1}{n}D_{n}\left(x\right)}\right),\label{eq:def_r+_r--1}
\end{eqnarray}
\[
r_{1}^{+}:=\limsup_{n\rightarrow\infty}\sup_{x\in M}\left(e^{\frac{1}{n}D_{n}\left(x\right)}\left\Vert Df_{x}^{n}|_{E_{u}}\right\Vert _{\mathrm{min}}^{-1/n}\right)
\]
For any $z\in\mathbb{C}$ such that $r_{1}^{+}+\varepsilon<\left|z\right|<r_{0}^{-}-\varepsilon$
or $\left|z\right|>r_{0}^{+}+\varepsilon$ we have:
\begin{equation}
\left\Vert \left(z-\widetilde{F}_{N}\right)^{-1}\right\Vert \leq C_{\varepsilon}.\label{eq:bound_resolvent-1}
\end{equation}

\end{thm}
}}}\end{center}

In the rest of this section we will assume that the potential $\tilde{V}$
is such that the external annulus $\mathcal{A}_{0}$ is\textbf{ }isolated
i.e. $r_{1}^{+}<r_{0}^{-}$, giving a ``spectral gap''. We take
the same definition of $\mathcal{H}_{\hbar}$ and $\hat{\mathcal{F}}_{\hbar}$
as in Definition \ref{def:quantum_operator}, for the spectral restriction
of $\widetilde{F}_{N}$ to the external band $\mathcal{A}_{0}$.

%blue
\begin{center}{\color{blue}\fbox{\color{black}\parbox{16cm}{
\begin{thm}
\textbf{\label{thm:weyl_law_extended}Index formula and Weyl law:}
\begin{equation}
\dim\mathcal{H}_{\hbar}=\int_{M}\left[e^{N\omega}\mbox{Todd}\left(TM\right)\right]_{2d}=N^{d}\mbox{Vol}_{\omega}\left(M\right)+\mathcal{O}\left(N^{d-1}\right)\label{eq:index_formula_quantum_op}
\end{equation}

\end{thm}
}}}\end{center}

The next Theorem is a particular case of Theorem \ref{thm:band_structure-1},
but there, compared to \ref{thm:band_structure}, we emphasize again
that the main interest is the case of the particular smooth potential
$\tilde{V}_{0}$ in (\eqref{eq:V_tilde_0}) giving $V_{0}\left(x\right)=\frac{1}{2}\log\left|\det\, Df_{x}|_{E_{u}\left(x\right)}\right|$,
$D=0$ hence $r_{0}^{+}=r_{0}^{-}=1$ in (\ref{eq:def_r+_r--1}),
so that the external annulus $\mathcal{A}_{0}$ coincides with the
unit circle.

%blue
\begin{center}{\color{blue}\fbox{\color{black}\parbox{16cm}{
\begin{thm}
\label{cor:The-special-choice_V0} Let $\widetilde{F}_{N}$ be the
transfer operator defined in (\ref{eq:def_transfert_gutzwiller})
with the special choice of the smooth potential $\tilde{V}_{0}\left(l\right)=\frac{1}{2}\log\left(\mathrm{det}Df_{x}\mid_{l}\right)$
on $G$. Then the operator $\widetilde{F}_{N}\hat{\chi}$ extends
to a bounded operator on $\mathcal{H}_{N}^{r}(P_{G})$. The Ruelle
spectrum of $\widetilde{F}_{N}\hat{\chi}$ \textbf{concentrates on
the unit circle} for $N=1/\left(2\pi\hbar\right)\rightarrow\infty$
and is separated from the internal resonances by a non vanishing asymptotic
spectral gap ($r_{1}^{+}<r_{0}^{-}=r_{0}^{+}=1$). That is, for any
given $\varepsilon>0$, its (discrete) spectrum is contained in 
\[
\left\{ ||z|-1|<\varepsilon\right\} \cup\left\{ |z|<r_{1}^{+}+\varepsilon\right\} 
\]
for sufficiently large $N$. The spectrum in $\left\{ ||z|-1|<\varepsilon\right\} $
obeys the Weyl law and the angular equidistribution law stated in
Theorem \ref{Thm:Distribution-of-resonances.}. (See Figure \ref{fig:semic_spectrum}). 
\end{thm}
}}}\end{center}

We will call the potential function $\tilde{V}_{0}$ above the \textbf{``potential
of reference''}. 

\begin{figure}

\begin{centering}
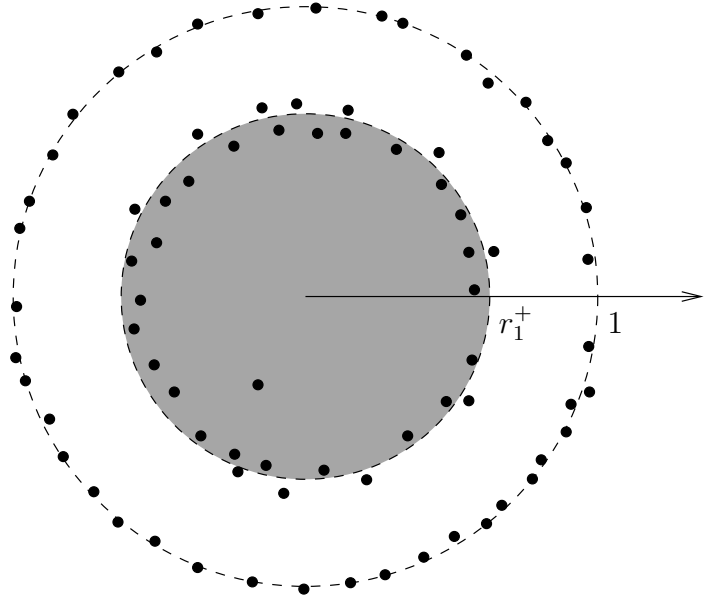\caption{\label{fig:semic_spectrum}With the particular potential $V_{0}=\frac{1}{2}\log\left|\det\, Df_{x}|_{E_{u}\left(x\right)}\right|$
Corollary \ref{cor:The-special-choice_V0} shows that the external
spectrum of the transfer operator concentrates uniformly on the unit
circle as $N=1/\left(2\pi\hbar\right)\rightarrow\infty$. (We have
not represented here the structure of the internal bands inside the
disc of radius $r_{1}^{+}$).}

\par\end{centering}

\end{figure}

\subsection{\label{sub:1.5Gutzwiller-trace-formula}Gutzwiller trace formula}

In this Section we continue to consider the prequantum transfer operators
$\widetilde{F}_{N}\hat{\chi}$ on the Grassmanian bundle $P_{G}$,
defined in (\ref{eq:def_transfert_gutzwiller}). We assume the condition
$r_{1}^{+}<r_{0}^{-}$. (This condition holds if we consider the potential
of reference $\widetilde{V}_{0}$.) As in Definition \ref{def:quantum_operator},
let $\Pi_{\hbar}:\mathcal{H}_{N}^{r}(P_{G})\to\mathcal{H}_{N}^{r}(P_{G})$
be the spectral projector for the external band and let $\mathcal{H}_{\hbar}$
be its image. Let $\hat{\mathcal{F}}_{\hbar}:\mathcal{H}_{\hbar}\rightarrow\mathcal{H}_{\hbar}$
be the restriction of $\widetilde{F}_{N}\hat{\chi}$ to $\mathcal{H}_{\hbar}$. 

%blue 
\begin{center}{\color{blue}\fbox{\color{black}\parbox{16cm}{
\begin{thm}
\textbf{\label{Th1.5}``Gutzwiller trace formula for large time''.}
Let $\varepsilon>0$. For any $\hbar=1/\left(2\pi N\right)$ small
enough, in the limit $n\rightarrow\infty$, we have\textbf{ }
\begin{equation}
\left|\mbox{Tr}\left(\hat{\mathcal{F}}_{\hbar}^{n}\right)-\sum_{x=f^{n}\left(x\right)}\frac{e^{D_{n}\left(x\right)}e^{iS_{n,x}/\hbar}}{\sqrt{\left|\mbox{Det}\left(1-Df_{x}^{n}\right)\right|}}\right|<CN^{d}(r_{1}^{+}+\varepsilon)^{n}\label{eq:Gutz_formula_large_time}
\end{equation}
where $e^{i2\pi S_{n,x}}$ is the \textbf{action} of a periodic point
defined in (\ref{eq:def_action_Snx}) and $D_{n}$ is the Birkhoff
sum (\ref{eq:Birkhoff_Sum}) of the effective damping function $D\left(x\right)=V\left(x\right)-V_{0}\left(x\right)$.
\end{thm}
}}}\end{center}

The proof of Theorem \ref{Th1.5} will be given in Section \ref{sec:9}.
It is based on the general and remarkable flat trace formula of Atiyah-Bott
that we recall in Lemma \ref{lem:Atiyah-Bott_flat_trace}.
\begin{rem}
We will see in Proposition \ref{prop:unicity-of-Gutz} below, the
simple but remarkable fact that the formula (\ref{eq:Gutz_formula_large_time})
determines the spectrum of $\hat{\mathcal{F}}_{\hbar}$ with multiplicities.
\end{rem}
~
\begin{rem}
For large time $n$ we have the equivalence 
\[
\left|\mbox{Det}\left(1-Df_{x}^{n}\right)\right|^{-1/2}\sim\left|\mbox{Det}\left(Df_{x}^{n}\mid_{E_{u}\left(x\right)}\right)\right|^{-1/2}=e^{-\left(V_{0}\right)_{n}\left(x\right)}
\]
Therefore one can expect that
\begin{align}
\mbox{Tr}\left(\hat{\mathcal{F}}_{\hbar}^{n}\right) & =\sum_{x=f^{n}\left(x\right)}e^{D_{n}\left(x\right)}e^{iS_{n,x}/\hbar}\left|\mbox{Det}\left(Df_{x}^{n}\mid_{E_{u}\left(x\right)}\right)\right|^{-1/2}+\mathcal{O}\left(1\right)N^{d}e^{\sup_{x}D_{n}}\lambda^{-n}\label{gutz2}\\
 & =\sum_{x=f^{n}\left(x\right)}e^{V_{n}\left(x\right)}e^{iS_{n,x}/\hbar}\left|\mbox{Det}\left(Df_{x}^{n}\mid_{E_{u}\left(x\right)}\right)\right|^{-1}+\mathcal{O}\left(1\right)N^{d}e^{\sup_{x}D_{n}}\lambda^{-n}
\end{align}
This is indeed true and this will be obtained in the proof. Let us
remark however that (\ref{gutz2}) is not an immediate consequence
of (\ref{eq:Gutz_formula_large_time}) and vise-versa due to the control
of the remainder. Indeed, suppose for example that (\ref{eq:Gutz_formula_large_time})
holds, with $D=V-V_{0}=0$ (for simplicity). We have for individual
periodic points $x$ and for $n\rightarrow\infty$, 
\[
\frac{1}{\sqrt{\left|\mbox{Det}\left(1-Df_{x}^{n}\right)\right|}}=\left|\mbox{Det}\left(Df_{x}^{n}\mid_{E_{u}\left(x\right)}\right)\right|^{-1/2}+\mathcal{O}\left(1\right)\lambda^{-\frac{dn}{2}-n}
\]
 but the number of periodic points grows like $\#\left\{ x=f^{n}\left(x\right)\right\} \geq\lambda^{dn}$
in general with $d=\frac{1}{2}\dim M$. Hence we can only deduce that
\begin{align*}
\left|\sum_{x=f^{n}\left(x\right)}\frac{e^{iS_{n,x}/\hbar}}{\sqrt{\left|\mbox{Det}\left(1-Df_{x}^{n}\right)\right|}}-\sum_{x=f^{n}\left(x\right)}e^{iS_{n,x}/\hbar}\left|\mbox{det}\left(Df_{x}^{n}\mid_{E_{u}\left(x\right)}\right)\right|^{-1/2}\right| & =\mathcal{O}\left(\lambda^{dn}\lambda^{-\frac{dn}{2}-n}\right)\\
 & =\mathcal{O}\left(\left(\lambda^{\frac{d}{2}-1}\right)^{n}\right)
\end{align*}
and if $d\geq2$ the term in the remainder is $\left(\lambda^{\frac{d}{2}-1}\right)\geq1$,
this is not enough to deduce (\ref{gutz2}).
\end{rem}

\subsubsection{\label{sub:The-question-of-natural-quantization}The question of
existence of a ``natural quantization''}

The following problem is a recurrent question in mathematics and physics
in the field of quantum chaos, since the discovery of the Gutzwiller
trace formula. For simplicity of the discussion we consider $V=V_{0}$
i.e. no effective damping, as in Figure \ref{fig:semic_spectrum}.

%cadre  bleu
\vspace{0.cm}\begin{center}{\color{blue}\fbox{\color{black}\parbox{16cm}{
\begin{problem}
\label{problem:Gutzwiller}Does there exists a sequence $\hbar_{j}>0$,
$\hbar_{j}\rightarrow0$ with $j\rightarrow\infty$, such that for
every $\hbar=\hbar_{j}$,
\begin{enumerate}
\item there exists a space $\mathcal{H}_{\hbar}$ of \textbf{finite dimension},
an operator $\hat{\mathcal{F}}_{\hbar}:\mathcal{H}_{\hbar}\rightarrow\mathcal{H}_{\hbar}$
which is \textbf{quasi unitary} in the sense that there exists $\varepsilon_{\hbar}\geq0$
with $\varepsilon_{\hbar_{j}}\rightarrow0$, with $j\rightarrow\infty$
and 
\begin{equation}
\forall u\in\mathcal{H}_{\hbar},\left(1-\varepsilon_{\hbar}\right)\left\Vert u\right\Vert \leq\left\Vert \hat{\mathcal{F}}_{\hbar}u\right\Vert \leq\left(1+\varepsilon_{\hbar}\right)\left\Vert u\right\Vert \label{eq:quasi-unitary}
\end{equation}

\item The operator $\hat{\mathcal{F}}_{\hbar}$ satisfies the \textbf{asymptotic
Gutzwiller Trace formula} for large time; i.e. there exists $0<\theta<1$
independent on $\hbar$ and some $C_{\hbar}>0$ which may depend on
$\hbar$, such that for $\hbar$ small enough (such that $\theta<1-\varepsilon_{\hbar}$):
\begin{equation}
\forall n\in\mathbb{N},\quad\left|\mbox{Tr}\left(\hat{\mathcal{F}}_{\hbar}^{n}\right)-\sum_{x=f^{n}\left(x\right)}\frac{e^{iS_{x,n}/\hbar}}{\sqrt{\left|\mbox{det}\left(1-Df_{x}^{n}\right)\right|}}\right|\leq C_{\hbar}\theta^{n}\label{eq:Gutz_Trace_Formula}
\end{equation}

\end{enumerate}
\end{problem}
}}}\end{center}\vspace{0.cm}

Let us notice first that Theorem \ref{Th1.5} (for the case $V=V_{0}$)
 provides a solution to Problem \ref{problem:Gutzwiller}: this is
the quantum operator $\hat{\mathcal{F}}_{h}:\mathcal{H}_{h}\rightarrow\mathcal{H}_{h}$
defined in (\ref{eq:def_quantum_op}) obtained with the choice of
potential $\tilde{V}=\tilde{V}_{0}$, Eq.(\eqref{eq:V_tilde_0}),
giving $V=V_{0}$. Indeed (\ref{eq:quasi-unitary}) holds true from
Corollary \ref{cor:The-special-choice_V0} and (\ref{eq:Gutz_Trace_Formula})
holds true from (\ref{eq:Gutz_formula_large_time}) and because $\theta:=r_{1}^{+}+\varepsilon<1$.

Some importance of the Gutzwiller trace formula (\ref{eq:Gutz_Trace_Formula})
comes from the following property which shows uniqueness of the solution
to the problem:

%cadre  bleu
\vspace{0.cm}\begin{center}{\color{blue}\fbox{\color{black}\parbox{16cm}{
\begin{prop}
\label{prop:unicity-of-Gutz}If $\hat{\mathcal{F}}_{\hbar}:\mathcal{H}_{\hbar}\rightarrow\mathcal{H}_{\hbar}$
is a solution of Problem \ref{problem:Gutzwiller} then the spectrum
of $\hat{\mathcal{F}}_{\hbar}$ is \textbf{uniquely defined} (with
multiplicities). In particular $\mbox{dim}\left(\mathcal{H}_{\hbar}\right)$
is uniquely defined.
\end{prop}
}}}\end{center}\vspace{0.cm}
\begin{proof}
This is consequence of the following lemma.
\begin{lem}
\label{lem:on Tra^n}If $A,B$ are matrices and for any $n\in\mathbb{N}$,
$\left|\mbox{Tr}\left(A^{n}\right)-\mbox{Tr}\left(B^{n}\right)\right|<C\theta^{n}$
with some $C>0$, $\theta\geq0$ then $A$ and $B$ have the same
spectrum with same multiplicities on the spectral domain $\left|z\right|>\theta$.\end{lem}
\begin{proof}[Proof of Lemma \ref{lem:on Tra^n}.]
 From the formula%
\footnote{This formula is easily proved by using eigenvalues $\lambda_{j}$
of $A$ and the Taylor series of $\log\left(1-x\right)=-\sum_{n\geq1}\frac{x^{n}}{n}$
which converges for $\left|x\right|<1$:
\begin{eqnarray*}
\mbox{det}\left(1-\mu A\right) & = & \prod_{j}\left(1-\mu\lambda_{j}\right)=\exp\left(\sum_{j}\log\left(1-\mu\lambda_{j}\right)\right)\\
 & = & \exp\left(-\sum_{j}\sum_{n\geq1}\frac{\left(\mu\lambda_{j}\right)^{n}}{n}\right)=\exp\left(-\sum_{n\geq1}\frac{\mu^{n}}{n}\mbox{Tr}\left(A^{n}\right)\right)
\end{eqnarray*}
}:
\[
\mbox{det}\left(1-\mu A\right)=\exp\left(-\sum_{n\geq1}\frac{\mu^{n}}{n}\mbox{Tr}\left(A^{n}\right)\right)
\]
The sum on the right is convergent if $1/\left|\mu\right|>\left\Vert A\right\Vert $.
Notice that we have (with multiplicities): $\mu$ is a zero of $d_{A}\left(\mu\right)=\mbox{det}\left(1-\mu A\right)$
if and only if $z=\frac{1}{\mu}$ is a (generalized) eigenvalue of
$A$. Using the formula we get that if $1/\left|\mu\right|>\theta$
then
\begin{align*}
\left|\frac{\det\left(1-\mu A\right)}{\det\left(1-\mu B\right)}\right| & \leq\exp\left(\sum_{n\geq1}\frac{\left|\mu\right|^{n}}{n}\left|\mbox{Tr}\left(A^{n}\right)-\mbox{Tr}\left(B^{n}\right)\right|\right)\\
 & <\exp\left(C\sum_{n\geq1}\frac{\left(\left|\mu\right|\theta\right)^{n}}{n}\right)=\left(1-\theta\left|\mu\right|\right)^{-C}=:\mathcal{B}
\end{align*}
Similarly $\left|\frac{\det\left(1-\mu A\right)}{\det\left(1-\mu B\right)}\right|>\frac{1}{\mathcal{B}}$,
hence $d_{A}\left(\mu\right)$ and $d_{B}\left(\mu\right)$ have the
same zeroes on $1/\left|\mu\right|>\theta$. Equivalently $A$ and
$B$ have the same spectrum on $\left|z\right|>\theta$.
\end{proof}
If $\hat{G}_{\hbar}$ is another solution of the problem \ref{problem:Gutzwiller}
then (\ref{eq:Gutz_Trace_Formula}) implies that $\left|\mbox{Tr}\left(\hat{\mathcal{F}}_{\hbar}^{n}\right)-\mbox{Tr}\left(\hat{G}_{\hbar}^{n}\right)\right|\leq2C\theta^{n}$
and Lemma \ref{lem:on Tra^n} tells us that $\hat{G}_{\hbar}$ and
$\hat{\mathcal{F}}_{\hbar}$ have the same spectrum on $\left|z\right|>\theta$.
But by hypothesis (\ref{eq:quasi-unitary}) their spectrum is in $\left|z\right|>1-\varepsilon_{\hbar}>\theta$.
Therefore all their spectrum coincides. This finishes the proof of
Proposition \ref{prop:unicity-of-Gutz}.\end{proof}
\begin{rem}
Previous results in the literature concerning the ``semiclassical
Gutzwiller formula'' for ``quantum maps'' do not provide an answer
to the problem \ref{problem:Gutzwiller} above. We explain why. For
any reasonable quantization of the Anosov map $f:M\rightarrow M$,
e.g. the Weyl quantization or geometric quantization, one obtains
a family of unitary operators $\hat{\mathcal{F}}_{\hbar}:\mathcal{H}_{\hbar}\rightarrow\mathcal{H}_{\hbar}$
acting in some finite dimensional (family of) Hilbert spaces. So this
answer to (\ref{eq:quasi-unitary}). Using semiclassical analysis
it is possible to show a Gutzwiller formula like (\ref{eq:Gutz_Trace_Formula})
but with an error term on the right hand side of the form $\mathcal{O}\left(\hbar\theta^{n}\right)$
with $\theta=e^{h_{0}/2}>1$ where $h_{0}>0$ is the topological entropy
which represents the exponential growing number of periodic orbits
(\cite{fred-trace-06} and references therein). Using more refined
semiclassical analysis at higher orders, the error can be made 
\begin{equation}
\mathcal{O}\left(\hbar^{M}\theta^{n}\right)\label{eq:error_term}
\end{equation}
 with any $M>0$ \cite{fred-trace-06}, but nevertheless one has a
total error which gets large after the so-called Ehrenfest time: $n\gg M\frac{\log\left(1/\hbar\right)}{\lambda_{0}}$.
So all these results obtained from any quantization scheme do not
provide an answer to the problem \ref{problem:Gutzwiller}. We may
regard the operator in (\ref{eq:def_quantum_op}) as the only ``quantization
procedure'' for which (\ref{eq:Gutz_Trace_Formula}) holds true.
For that reason we may call it a \textbf{natural quantization} of
the Anosov map $f$.\end{rem}

\subsection{\label{sub:Emergence-of-quantum}Dynamical correlation functions
and emergence of quantum dynamics}

As explained in \cite[cor. 1.3]{fred-roy-sjostrand-07} for example,
the Ruelle-Pollicott spectrum of the transfer operator $\hat{F}$
has an important meaning in terms of time evolution of correlation
functions. If $u,v\in C^{\infty}\left(P\right)$ the time correlation
function is defined by 
\[
C_{v,u}\left(n\right):=\int_{P}\overline{v\left(p\right)}\left(\hat{F}^{n}u\right)\left(p\right)d\mu_{P}=\left(v,\hat{F}^{n}u\right){}_{L^{2}\left(P\right)}.
\]
In this Section we show that $C_{v,u}\left(n\right)$ can be expressed
as an asymptotic over the Ruelle resonances, up to exponentially small
error term. In particular we emphasize the role of the external band
in the spectrum as a manifestation of ``quantum behavior'' in the
fluctuations of the correlation functions. In this subsection and
the next, we consider the transfer operators on $P$ (not the Grassmanian
extension $P_{G}$).

\subsubsection{Use of the resolvent estimate for dynamical control}

Before giving results about dynamical correlation functions, we give
a proposition which expresses differently (but equivalently) the estimate
(\ref{eq:bound_resolvent-1}) about the resolvent. Recall from (\ref{eq:def_quantum_op})
that if $r_{1}^{+}<r_{0}^{+}$, then for $N$ large enough, the transfer
operator has a spectral decomposition $\widetilde{F}_{N}=\hat{\mathcal{F}}_{\hbar}+\left(\hat{F}_{N}-\hat{\mathcal{F}}_{\hbar}\right)$
into a finite rank operator $\hat{\mathcal{F}}_{\hbar}$ for the external
band and $\left(\hat{F}_{N}-\hat{\mathcal{F}}_{\hbar}\right)$ is
for the internal structure.

%blue
\begin{center}{\color{blue}\fbox{\color{black}\parbox{16cm}{
\begin{prop}
\label{prop:bound_on_F^n}For any $\varepsilon>0$, there exists $C_{\varepsilon}>0$
and $N_{\varepsilon}\geq1$ such that for any $N\geq N_{\varepsilon}$
and for any $n\geq0$,
\begin{equation}
\left\Vert \hat{F}_{N}^{n}\right\Vert \leq C_{\varepsilon}\left(r_{0}^{+}+\varepsilon\right)^{n}\label{eq:bound_F_N^n}
\end{equation}
Moreover if $r_{1}^{+}<r_{0}^{-}$ then
\begin{equation}
\left\Vert \hat{\mathcal{F}}_{\hbar}^{n}\right\Vert _{\max}\leq C_{\varepsilon}\left(r_{0}^{+}+\varepsilon\right)^{n},\quad\left\Vert \hat{\mathcal{F}}_{\hbar}^{n}\right\Vert _{\mathrm{min}}\geq\frac{1}{C_{\varepsilon}}\left(r_{0}^{-}-\varepsilon\right)^{n}\label{eq:bound_F_h^n}
\end{equation}
and
\begin{equation}
\left\Vert \left(\hat{F}_{N}-\hat{\mathcal{F}}_{\hbar}\right)^{n}\right\Vert \leq C_{\varepsilon}\left(r_{1}^{+}+\varepsilon\right)^{n}.\label{eq:bound_F1^n}
\end{equation}

\end{prop}
}}}\end{center}
\begin{rem}
More generally, if there is some internal isolated band $k$, i.e.
if $r_{k+1}^{+}<r_{k}^{-}$,$r_{k}^{+}<r_{k-1}^{-}$ for some $k\geq1$,
then we can consider the corresponding spectral decomposition $\hat{F}_{N}=\ldots+\hat{\mathcal{F}}_{k,\hbar}+\ldots$
isolating some finite rank operator $\hat{\mathcal{F}}_{k,\hbar}$
and similarly we can show that (\ref{eq:bound_resolvent}) is equivalent
to 
\[
\frac{1}{C_{\varepsilon}}\left(r_{k}^{-}-\varepsilon\right)^{n}\le\left\Vert \hat{\mathcal{F}}_{k,\hbar}^{n}\right\Vert _{\mathrm{min}}\le\left\Vert \hat{\mathcal{F}}_{k,\hbar}^{n}\right\Vert _{\max}\leq C_{\varepsilon}\left(r_{k}^{+}+\varepsilon\right)^{n}.
\]
\end{rem}
\begin{proof}
Let $\gamma$ be the closed path in $\mathbb{C}$ made by the union
of the circle of radius $r_{0}^{+}+\varepsilon$ in the direct sense
and the circle of radius $r_{0}^{-}-\varepsilon$ in the indirect
sense. From Cauchy formula one has
\[
\hat{\mathcal{F}}_{\hbar}^{n}=\frac{1}{2\pi i}\oint_{\gamma}z^{n}\left(z-\hat{F}_{N}\right)^{-1}dz
\]
which implies
\[
\left\Vert \hat{\mathcal{F}}_{\hbar}^{n}\right\Vert _{\max}\leq\left(r_{0}^{+}+\varepsilon\right)^{n+1}\max_{z\in\gamma}\left\Vert \left(z-\hat{F}_{N}\right)^{-1}\right\Vert 
\]
Then the uniform bound on the resolvent (\ref{eq:bound_resolvent-1})
implies 
\[
\left\Vert \hat{\mathcal{F}}_{\hbar}^{n}\right\Vert \leq C_{\varepsilon}\left(r_{0}^{+}+\varepsilon\right)^{n}
\]
which is the first equation of (\ref{eq:bound_F_h^n}). Reversing
the sign of $n$ one gets 
\[
\left\Vert \hat{\mathcal{F}}_{\hbar}^{-n}\right\Vert \leq C_{\varepsilon}\left(r_{0}^{-}-\varepsilon\right)^{-n}.
\]
Hence
\[
\left\Vert \hat{\mathcal{F}}_{\hbar}^{n}\right\Vert _{\mathrm{min}}=\left\Vert \hat{\mathcal{F}}_{\hbar}^{-n}\right\Vert ^{-1}\geq\frac{1}{C_{\varepsilon}}\left(r_{0}^{-}-\varepsilon\right)^{n}
\]
which is the second equation of (\ref{eq:bound_F_h^n}). The bound
(\ref{eq:bound_F_N^n}) is obtained similarly but more simply with
the closed path $\gamma$ being only the circle of radius $r_{0}^{+}+\varepsilon$
in the direct sense. The estimate (\ref{eq:bound_F1^n}) is obtained
with the closed path $\gamma$ being the circle of radius $r_{1}^{+}+\varepsilon$
in the direct sense.
\end{proof}

\subsubsection{Decay of correlations expressed with the quantum operator}

We first introduce a notation: for a given $N\in\mathbb{Z}$, we have
seen that the prequantum transfer operator $\hat{F}_{N}$ has a discrete
spectrum of resonances. For $\rho>0$ such that there is no eigenvalue
on the circle $\left|z\right|=\rho$ for any $N$, we denote by $\Pi_{\rho,N}$
the projector on the Fourier space of mode $N$ composed with the
spectral projector of the operator $\hat{F}_{N}$ on the domain $\left\{ z\in\mathbb{C},\left|z\right|>\rho\right\} $.
This is a finite rank operator (which obviously commutes with $\hat{F}_{N}$).
Recall that $\hbar=\frac{1}{2\pi N}$.

%blue 
\begin{center}{\color{blue}\fbox{\color{black}\parbox{16cm}{
\begin{thm}
\textbf{\label{Th1.7}}Suppose that $r_{1}^{+}<r_{0}^{-}$ (i.e. the
external band is isolated). For any $\varepsilon>0$, there exists
$N_{\varepsilon}\geq1$ such that for any $N\geq N_{\varepsilon}$,
for any $u,v\in C^{\infty}\left(P\right)$ such that $\mathrm{supp}\left(u\right)\subset K_{0}$,
and for $n\rightarrow\infty$, one has
\begin{equation}
\underbrace{\left(v,\hat{F}^{n}u\right)_{L^{2}}}_{\mathrm{"classical"}}=\sum_{\left|N\right|\leq N_{\varepsilon}}\left(v_{N},\left(\hat{F}_{N}\Pi_{\rho,N}\right)^{n}u_{N}\right)+\sum_{\left|N\right|>N_{\varepsilon}}\underbrace{\left(v_{N},\hat{\mathcal{F}}_{\hbar}^{n}u_{N}\right)}_{"\mathrm{quantum"}}+O\left(\left(r_{1}^{+}+\varepsilon\right)^{n}\right)\label{eq:correlations_classical-quantum}
\end{equation}
with $\rho=r_{1}^{+}+\varepsilon$ and where $u_{N},v_{N}\in C_{N}^{\infty}\left(P\right)$
are the Fourier components of the functions $u$ and $v$. In the
right hand side of (\ref{eq:correlations_classical-quantum}), the
first sum is a finite sum and involves finite rank operators. The
second sum is an infinite but convergent sum.
\end{thm}
}}}\end{center}
\begin{rem}
The asymptotic formula (\ref{eq:correlations_classical-quantum})
has a nice interpretation: the classical correlation functions $\left(v,\hat{F}^{n}u\right)$
are governed by the quantum correlation functions $\left(v_{N},\hat{\mathcal{F}}_{\hbar}^{n}u_{N}\right)$
for large time (up to the first finite rank expression), or equivalently
the \textbf{``quantum dynamics emerge dynamically from the classical
dynamics''}.
\end{rem}
~
\begin{rem}
From (\ref{eq:bound_F_h^n}) one has that 
\[
\left|\left(v_{N},\hat{\mathcal{F}}_{\hbar}^{n}u_{N}\right)\right|\leq\left\Vert u_{N}\right\Vert _{\left(\mathcal{H}_{N}^{r}\right)'}\left\Vert v_{N}\right\Vert _{\mathcal{H}_{N}^{r}}\left\Vert \hat{\mathcal{F}}_{\hbar}^{n}\right\Vert _{\mathcal{H}_{N}^{r}}\leq O\left(N^{-\infty}\right)O\left(\left(r_{0}^{+}+\varepsilon\right)^{n}\right)
\]
uniformly in $n$ and $N$. We have used the fact that for smooth
functions one has fast decay in $N$: $\left\Vert u_{N}\right\Vert ,\left\Vert v_{N}\right\Vert =O\left(N^{-\infty}\right)$.
(This is also because the weight in the space $\mathcal{H}_{N}^{r}$
is polynomially bounded in frequencies $\xi$).
\end{rem}
~
\begin{rem}
It is known that for $n\rightarrow\infty$,
\[
\left(v,\hat{F}^{n}u\right)=\lambda_{0}^{n}\left(v,\Pi_{\lambda_{0}}u\right)+O\left(\left|\lambda_{1}\right|^{n}\right)
\]
where $\lambda_{0}>0$ is the leading and simple eigenvalue of $\widetilde{F}$
(in the space $\mathcal{H}_{N=0}^{r})$ and $\lambda_{1}$ is the
second eigenvalue with $\left|\lambda_{1}\right|<\lambda_{0}$. The
case $V=0$ for which $\lambda_{0}=1$ gives that the map $\tilde{f}:P\rightarrow P$
is mixing with exponential decay of correlations.\end{rem}
\begin{proof}
Let $N_{\varepsilon}$ given by Theorem \ref{thm:Discrete-spectrum}
or Proposition \ref{prop:bound_on_F^n} and write

\begin{equation}
\left(v,\hat{F}^{n}u\right)=\sum_{\left|N\right|\leq N_{\varepsilon}}\left(v_{N},\left(\hat{F}_{N}\Pi_{\rho,N}\right)^{n}u_{N}\right)+\sum_{\left|N\right|\leq N_{\varepsilon}}\left(v_{N},\left(\hat{F}_{N}\left(1-\Pi_{\rho,N}\right)\right)^{n}u_{N}\right)+\sum_{\left|N\right|>N_{\varepsilon}}\left(v_{N},\hat{F}_{N}^{n}u_{N}\right)\label{eq:sum-1}
\end{equation}
Let us consider the second term on the right hand side. From definition
of $\Pi_{\rho,N}$ one has $r_{s}\left(\hat{F}_{N}\left(1-\Pi_{\rho,N}\right)\right)\leq\left(r_{0}^{+}+\varepsilon\right)$
so for the finite number of terms $\left|N\right|\leq N_{\varepsilon}$
one has$\left\Vert \left(\hat{F}_{N}\left(1-\Pi_{\rho,N}\right)\right)^{n}\right\Vert \leq C_{\varepsilon}\left(r_{0}^{+}+\varepsilon\right)^{n}$
and we deduce the estimate 
\[
\left|\sum_{\left|N\right|\leq N_{\varepsilon}}\left(v_{N},\left(\hat{F}_{N}\left(1-\Pi_{\rho,N}\right)\right)^{n}u_{N}\right)\right|=O\left(\left(r_{0}^{+}+\varepsilon\right)^{n}\right).
\]

As in Proposition \ref{prop:bound_on_F^n} we consider the spectral
decomposition $\hat{F}_{N}=\hat{\mathcal{F}}_{\hbar}+\left(\hat{F}_{N}-\hat{\mathcal{F}}_{\hbar}\right)$
and decompose accordingly the last term of (\ref{eq:sum-1}) as
\[
\sum_{\left|N\right|>N_{\varepsilon}}\left(v_{N},\hat{F}_{N}^{n}u_{N}\right)=\sum_{\left|N\right|>N_{\varepsilon}}\left(v_{N},\hat{\mathcal{F}}_{\hbar}^{n}u_{N}\right)+\sum_{\left|N\right|>N_{\varepsilon}}\left(v_{N},\left(\hat{F}_{N}-\hat{\mathcal{F}}_{\hbar}\right)^{n}u_{N}\right)
\]
From (\ref{eq:bound_F1^n}), one has then 
\begin{align*}
\left|\left(v_{N},\left(\hat{F}_{N}-\hat{\mathcal{F}}_{\hbar}\right)^{n}u_{N}\right)\right| & \leq\left\Vert u_{N}\right\Vert _{\left(\mathcal{H}_{N}^{r}\right)'}\left\Vert v_{N}\right\Vert _{\mathcal{H}_{N}^{r}}\left\Vert \left(\hat{F}_{N}-\hat{\mathcal{F}}_{\hbar}\right)^{n}\right\Vert _{\mathcal{H}_{N}^{r}}\\
 & \leq O\left(N^{-\infty}\right)O\left(\left(r_{1}^{+}+\varepsilon\right)^{n}\right)
\end{align*}
This implies that $\left|\sum_{\left|N\right|>N_{\varepsilon}}\left(v_{N},\left(\hat{F}_{N}-\hat{\mathcal{F}}_{\hbar}\right)^{n}u_{N}\right)\right|=O\left(\left(r_{1}^{+}+\varepsilon\right)^{n}\right)$
and we get (\ref{eq:correlations_classical-quantum}).
\end{proof}

\subsection{\label{sub:1.7Semiclassical-calculus-on}Semiclassical calculus on
the quantum space}

In Definition \ref{def:quantum_operator} we have defined the quantum
space $\mathcal{H}_{\hbar}$ for every $\hbar=\frac{1}{2\pi N}$ (small
enough), the quantum operator $\hat{\mathcal{F}}_{\hbar}:\mathcal{H}_{\hbar}\rightarrow\mathcal{H}_{\hbar}$
and the finite rank spectral projector $\Pi_{\hbar}:\mathcal{H}_{N}^{r}\left(P\right)\rightarrow\mathcal{H}_{\hbar}$.
In this section we introduce the definition of ``quantization of
symbols'' on this quantum space and give some properties of them.
In semiclassical analysis these properties are considered as ``standard''
or ``basic properties'' for defining ``a good semiclassical calculus''.
We will comment on them at the end of the Section. Beware that the
quantum space $\mathcal{H}_{\hbar}$ defined here depends on the given
Anosov diffeomorphism $f:M\rightarrow M$. 
\begin{rem}
The results below extend readily to the case of Grassmanian extension
considered in Subsection \ref{sub:Spectral-results-with-Grassman}
and \ref{sub:1.5Gutzwiller-trace-formula}.
\end{rem}
%red 
\begin{center}{\color{red}\fbox{\color{black}\parbox{16cm}{
\begin{defn}
\label{def:class_symbols}For $0\leq\delta<1/2$ and for some family
of constant $C_{\alpha}>0,\alpha\in\mathbb{N}^{2d}$, we define the
\textbf{class of symbols}
\[
S_{\delta}:=\left\{ \psi\in C^{\infty}\left(M\right)\mbox{ s.t. }\forall\alpha\in\mathbb{N}^{2d},\left|\partial_{x}^{\alpha}\psi\right|<C_{\alpha}\hbar^{-\delta\left|\alpha\right|}\right\} 
\]
A symbol $\psi\in S_{\delta}$ is therefore a family of smooth functions
$\left(\psi_{\hbar}\right)_{\hbar}$.
\end{defn}
}}}\end{center}

%red 
\begin{center}{\color{red}\fbox{\color{black}\parbox{16cm}{
\begin{defn}
\label{def:quantization}For any symbol $\psi\in S_{\delta}$ we define
its quantization as the operator
\begin{equation}
\mathrm{Op}_{\hbar}\left(\psi\right):=\Pi_{\hbar}\circ\mathcal{M}(\psi)\circ\Pi_{\hbar}\quad:\mathcal{H}_{\hbar}\rightarrow\mathcal{H}_{\hbar}\label{eq:def_Op_h_Psi}
\end{equation}
where $\mathcal{M}(\psi)$ is the multiplication operator by the function
$\psi$ in $\mathcal{H}_{N}^{r}\left(P\right)$.
\end{defn}
}}}\end{center}

As it is usual in quantum mechanics, we call the operator $\mathrm{Op}_{\hbar}\left(\psi\right)$
a ``\textbf{quantum observable''}.
\begin{rem}
Definition \ref{def:quantization} is very similar to the definition
of Toeplitz (or anti-Wick) quantization of a symbol. The difference
is that the quantum space $\mathcal{H}_{\hbar}$ considered here is
attached to a given Anosov diffeomorphism $f:M\rightarrow M$. 
\end{rem}
In this Section beware that the operator norms $\left\Vert .\right\Vert $
and trace norms $\left\Vert .\right\Vert _{\mathrm{Tr}}$ are define
with respect to the norm on the Hilbert space $\mathcal{H}_{N}^{r}\left(P\right)$.

%blue 
\begin{center}{\color{blue}\fbox{\color{black}\parbox{16cm}{
\begin{thm}
\label{prop:express_Op_Psi}For any class of symbols $S_{\delta}$,
there exist constants $C>0$ and $\varepsilon>0$ such that the following
holds: For every $\hbar$, there exists a smooth family of rank one
projectors $\pi_{x}:\mathcal{H}_{N}^{r}(P)\rightarrow\mathcal{H}_{\hbar}\subset\mathcal{H}_{N}^{r}(P)$
with $\left\Vert \pi_{x}\right\Vert \leq C$, parametrized by $x\in M$,
such that for any symbol $\psi\in S_{\delta}$, we have
\begin{equation}
\left\Vert \mathrm{Op}_{\hbar}\left(\psi\right)-\frac{1}{\left(2\pi\hbar\right)^{d}}\int_{M}\psi\left(x\right)\pi_{x}dx\right\Vert \leq C\hbar^{\varepsilon}\label{eq:Op_and_intgral_wavepckets}
\end{equation}
and
\begin{equation}
\left\Vert \left[\Pi_{\hbar},\mathcal{M}(\psi)\right]\right\Vert \leq C\hbar^{\varepsilon}.\label{eq:commut}
\end{equation}

\end{thm}

}}}\end{center}

The proof of Theorem \ref{prop:express_Op_Psi} is given in Section
\ref{sub:Proof-of-Theorem1.40}.
\begin{rem}
In particular for the choice $\psi=1$ we get an approximate expression
of $\Pi_{\hbar}\equiv\mathrm{Id}\mid_{\mathcal{H}_{\hbar}}$ as $\left\Vert \Pi_{\hbar}-\frac{1}{\left(2\pi\hbar\right)^{d}}\int_{M}\pi_{x}dx\right\Vert \leq C\hbar^{\varepsilon}$
sometimes called ``resolution of identity''.
\end{rem}
~
\begin{rem}
Since $\mathrm{dim}\mathcal{H}_{\hbar}\leq C\cdot\hbar$ we have obviously
that an estimate in norm operator implies an estimate in trace class
norm by $\left\Vert \cdot\right\Vert _{\mathrm{Tr}}\leq C.\hbar^{-d}\left\Vert \cdot\right\Vert .$
\end{rem}
%blue 
\begin{center}{\color{blue}\fbox{\color{black}\parbox{16cm}{
\begin{cor}
\textbf{\label{cor:Composition-formula}``Composition formula''}.
There exist $C>0$ and $\varepsilon>0$ such that for any $\hbar$
and any $\psi_{1},\psi_{2}\in S_{\delta}$
\begin{equation}
\left\Vert \mathrm{Op}_{\hbar}\left(\psi_{1}\right)\circ\mathrm{Op}_{\hbar}\left(\psi_{2}\right)-\mathrm{Op}_{\hbar}\left(\psi_{1}\psi_{2}\right)\right\Vert \leq C\hbar^{\varepsilon}\label{eq:composition_formula}
\end{equation}

\end{cor}
}}}\end{center}
\begin{proof}
Below the notation $O\left(\hbar^{\varepsilon}\right)$ means that
this is a term with norm less than $C\hbar^{\varepsilon}$. We have
\begin{align*}
\mathrm{Op}_{\hbar}\left(\psi_{1}\right)\circ\mathrm{Op}_{\hbar}\left(\psi_{1}\right)= & \Pi_{\hbar}\mathcal{M}(\psi_{1})\Pi_{\hbar}\mathcal{M}(\psi_{2})\Pi_{\hbar}\\
\underset{(\ref{eq:commut})}{=} & \Pi_{\hbar}\mathcal{M}(\psi_{1})\mathcal{M}(\psi_{2})\Pi_{\hbar}+O\left(\hbar^{\varepsilon}\right)\\
= & \Pi_{\hbar}\mathcal{M}(\psi_{1}\psi_{2})\Pi_{\hbar}+O\left(\hbar^{\varepsilon}\right)=\mathrm{Op}_{\hbar}\left(\psi_{1}\psi_{2}\right)+O\left(\hbar^{\varepsilon}\right)
\end{align*}
\end{proof}
\begin{rem}
Composition formula (\ref{eq:composition_formula}) expressed the
property that the operator $\mathrm{Op}_{\hbar}\left(\psi\right)$
has the so-called ``microlocal property''. This is seen by taking
$\psi_{1}$ and $\psi_{2}$ with disjoint supports giving that $\left\Vert \mathrm{Op}_{\hbar}\left(\psi_{1}\right)\circ\mathrm{Op}_{\hbar}\left(\psi_{2}\right)\right\Vert \leq C\hbar^{\varepsilon}$.

\end{rem}
~
\begin{rem}
In Proposition \ref{prop:express_Op_Psi} the operator $\mathrm{Op}_{\hbar}\left(\psi\right)$
is decomposed into rank one operators $\pi_{x}$. Each operator $\pi_{x}$
is a projector and can be written $\pi_{x}(\cdot)=\left(\varphi_{x},.\right)_{\mathcal{H}_{\hbar}^{r}}\cdot\psi_{x}$
with $\psi_{x},\varphi_{x}\in\mathcal{H}_{\hbar}^{r}$. From the ``microlocal
property'' given in the previous remark, we can think of $\psi_{x},\varphi_{x}$
as ``microlocal wave packets'' and $\pi_{x}$ as a projection over
these wave packets. In the proof of Proposition \ref{eq:Op_and_intgral_wavepckets},
(\ref{eq:def_projector_pi_alpha}), we can find a very explicit expression
of $\pi_{x}$ on local coordinate charts where $\psi_{x},\varphi_{x}$
are respectively along the unstable and stable directions in the orthogonal
space $\left(K\right)^{\perp}$ (in variables $\zeta$) and are localized
wave packets within the trapped set $K$ (in the variables $\nu$). 
\end{rem}
%blue 
\begin{center}{\color{blue}\fbox{\color{black}\parbox{16cm}{
\begin{cor}
\textbf{\label{cor:Adjoint_of_observable}``Adjoint of observables''}.
There exist $C>0$, $\varepsilon>0$, such that for any $\hbar$ and
any $\psi\in S_{\delta}$
\begin{equation}
\left\Vert \left(\mathrm{Op}_{\hbar}\left(\psi\right)\right)_{\mathcal{H}_{\hbar}}^{\dagger}-\mathrm{Op}_{\hbar}\left(\overline{\psi}\right)\right\Vert \leq C\hbar^{\varepsilon}\label{eq:adjoint_observable}
\end{equation}
The adjoint operator is defined here in the space $\mathcal{H}_{\hbar}$
by the relation $\left(u,\left(\mathrm{Op}_{\hbar}\left(\psi\right)\right)^{\dagger}v\right)_{\mathcal{H}_{\hbar}}=\left(\mathrm{Op}_{\hbar}\left(\psi\right)u,v\right)_{\mathcal{H}_{\hbar}}$for
any $u,v\in\mathcal{H}_{\hbar}$.
\end{cor}
}}}\end{center}
\begin{proof}
For any $u,v\in\mathcal{H}_{\hbar}$ we have
\begin{align*}
\left(u,\left(\mathrm{Op}_{\hbar}\left(\psi\right)\right)^{\dagger}v\right)_{\mathcal{H}_{\hbar}} & =\left(\mathrm{Op}_{\hbar}\left(\psi\right)u,v\right)_{\mathcal{H}_{\hbar}}=\left(\Pi_{\hbar}\mathcal{M}(\psi)u,v\right)_{\mathcal{H}_{N}^{r}\left(P\right)}\\
 & \underset{(\ref{eq:commut})}{=}\left(\mathcal{M}(\psi)\Pi_{\hbar}u,v\right)_{\mathcal{H}_{N}^{r}\left(P\right)}+O\left(\hbar^{\varepsilon}\|u\|\cdot\|v\|\right)\\
 & \underset{}{=}\left(u,\mathcal{M}(\overline{\psi})v\right)_{\mathcal{H}_{N}^{r}\left(P\right)}+O\left(\hbar^{\varepsilon}\|u\|\cdot\|v\|\right)\\
 & \underset{(\ref{eq:commut})}{=}\left(u,\Pi_{\hbar}\mathcal{M}(\overline{\psi})v\right)_{\mathcal{H}_{N}^{r}\left(P\right)}+O\left(\hbar^{\varepsilon}\|u\|\cdot\|v\|\right)\\
 & =\left(u,\mathrm{Op}_{\hbar}\left(\overline{\psi}\right)v\right)_{\mathcal{H}_{\hbar}}+O\left(\hbar^{\varepsilon}\|u\|\cdot\|v\|\right)
\end{align*}
(The equality in the middle will be given in Corollary \ref{cor:transpose}.
)
\end{proof}
%blue 
\begin{center}{\color{blue}\fbox{\color{black}\parbox{16cm}{
\begin{prop}
\textbf{\label{prop:Exact-Egorov-formula.}``Exact Egorov formula''.}
For any $\hbar$ and any $\psi\in S_{\delta}$, we have
\begin{equation}
\hat{\mathcal{F}}_{\hbar}\circ\mathrm{Op}_{\hbar}\left(\psi\right)=\mathrm{Op}_{\hbar}\left(\psi\circ f^{-1}\right)\circ\hat{\mathcal{F}}_{\hbar}\label{eq:exact_Egorov_formula}
\end{equation}

\end{prop}
}}}\end{center}
\begin{proof}
We use that $\Pi_{\hbar}$ is a spectral projector of $\hat{\mathcal{F}}_{\hbar}=\Pi_{\hbar}\hat{F}_{N}\Pi_{\hbar}$
and that from its definition (\ref{eq:def_prequantum_operator_F}),
$\hat{F}_{N}$ is a transfer operator hence $\hat{F}_{N}\circ\mathcal{M}(\psi)=\mathcal{M}(\psi\circ f^{-1})\circ\hat{F}_{N}$
:
\begin{align*}
\hat{\mathcal{F}}_{\hbar}\circ\mathrm{Op}_{\hbar}\left(\psi\right) & =\Pi_{\hbar}\hat{F}_{N}\Pi_{\hbar}\mathcal{M}(\psi)\Pi_{\hbar}=\Pi_{\hbar}\hat{F}_{N}\mathcal{M}(\psi)\Pi_{\hbar}=\Pi_{\hbar}\mathcal{M}(\psi\circ f^{-1})\hat{F}_{N}\Pi_{\hbar}\\
 & =\Pi_{\hbar}\mathcal{M}(\psi\circ f^{-1})\Pi_{\hbar}\hat{F}_{N}\Pi_{\hbar}=\mathrm{Op}_{\hbar}\left(\psi\circ f^{-1}\right)\circ\hat{\mathcal{F}}_{\hbar}
\end{align*}

\end{proof}
%blue 
\begin{center}{\color{blue}\fbox{\color{black}\parbox{16cm}{
\begin{prop}
\textbf{\label{prop:Trace-of-observables.}``Trace of observables''.}
There exists $C>0$, $\varepsilon>0$, such that for any $\hbar$
and any $\psi\in S_{\delta}$ we have
\begin{equation}
\left|\left(2\pi\hbar\right)^{d}\mathrm{Tr}\left(\mathrm{Op}_{\hbar}\left(\psi\right)\right)-\int_{M}\psi dx\right|\leq C\hbar^{\varepsilon}\label{eq:trace_of_observables}
\end{equation}

\end{prop}
}}}\end{center}
\begin{proof}
Since $\mathrm{Tr}\left(\pi_{x}\right)=1$, we have
\[
\mathrm{Tr}\left(\int_{M}\psi\left(x\right)\pi_{x}dx\right)=\int_{M}\psi dx.
\]
Then
\begin{align*}
\left|\left(2\pi\hbar\right)^{d}\mathrm{Tr}\left(\mathrm{Op}_{\hbar}\left(\psi\right)\right)-\int_{M}\psi dx\right| & \leq\left(2\pi\hbar\right)^{d}\left\Vert \mathrm{Op}_{\hbar}\left(\psi\right)-\frac{1}{\left(2\pi\hbar\right)^{d}}\int_{M}\psi\left(x\right)\pi_{x}dx\right\Vert _{\mathrm{Tr}}\\
 & \leq C\left\Vert \mathrm{Op}_{\hbar}\left(\psi\right)-\frac{1}{\left(2\pi\hbar\right)^{d}}\int_{M}\psi\left(x\right)\pi_{x}dx\right\Vert \\
 & \underset{(\ref{eq:Op_and_intgral_wavepckets})}{\leq}C'\hbar^{\varepsilon}
\end{align*}
In the second line we have used (\ref{eq:weyl_law}) to get that $\left\Vert \cdot\right\Vert _{\mathrm{Tr}}\leq C.\hbar^{-d}\left\Vert \cdot\right\Vert $.
\end{proof}
Taking $\psi=1$ in (\ref{eq:trace_of_observables}) and because $\mathrm{Op}_{\hbar}\left(1\right)=\Pi_{\hbar}$
hence $\mathrm{dim}\mathcal{H}_{\hbar}=\mathrm{Tr}\left(\Pi_{\hbar}\right)=\mathrm{Tr}\left(\mathrm{Op}_{\hbar}\left(1\right)\right)$
we obtain the Weyl law (\ref{eq:weyl_law}):

%blue 
\begin{center}{\color{blue}\fbox{\color{black}\parbox{16cm}{
\begin{cor}
\label{cor:Weyl_law}We have
\[
\mathrm{dim}\mathcal{H}_{\hbar}=\frac{1}{\left(2\pi\hbar\right)^{d}}\mathrm{Vol}_{\omega}\left(M\right)\left(1+O\left(\hbar^{\varepsilon}\right)\right)
\]

\end{cor}
}}}\end{center}

\begin{rem}
Eq.(\ref{eq:exact_Egorov_formula}) expresses transport properties
of the operator $\hat{\mathcal{F}}_{\hbar}$. For usual quantization
scheme of non-linear map $f:M\rightarrow M$ the Egorov formula has
some error $O\left(\hbar\right)$ in operator norm. It is therefore
remarkable that the formula is exact here. We can iterate it for any
time $n\geq1$ and obtain $\hat{\mathcal{F}}_{\hbar}^{n}\circ\mathrm{Op}_{\hbar}\left(\psi\right)=\mathrm{Op}_{\hbar}\left(\psi\circ f^{-n}\right)\circ\hat{\mathcal{F}}_{\hbar}^{n}$.

For Schrodinger equation, it is a difficult task to express long time
dynamics of initial states using the classical underlying dynamics
(relevant in the limit of small wavelength). This is due to interferences
effects, semiclassical corrections that grows rapidly and that are
difficult to control. In our situation, the prequantum operator $\tilde{F}$
is a transfer operator for which transport properties are exact. Since
the quantum operator $\mathcal{F}_{\hbar}$ is a spectral restriction
of $\tilde{F}$ we have immediately that for any initial state $u\in\mathcal{H}_{\hbar}$,
any time $n\geq0$,
\[
\mathcal{F}_{\hbar}^{n}u=\left(\tilde{F}^{n}u\right)_{\mid\mathcal{H}_{\hbar}}
\]
so that quantum evolution and classical transport coincide for any
time.
\end{rem}
~
\begin{rem}
In their paper \cite{marklof-05}, J. Marklof and S. O'Keefe propose
some axioms for quantum observables associated to quantum maps. Their
Axioms \cite[Axiom2.1,(a),(b),(c) page282]{marklof-05} correspond
respectively to Proposition \ref{cor:Adjoint_of_observable}, \ref{cor:Composition-formula}
and \ref{prop:Trace-of-observables.} above. Their Axiom \cite[Axiom2.2, page282]{marklof-05}
corresponds to Proposition \ref{prop:Exact-Egorov-formula.} (with
no remainder here).
\end{rem}

\subsection{\label{sub:1.8Geometric-quantization-of}Geometric quantization of
the symplectic map}

We have discussed above the ``natural quantization'' of the map
$f$ as the operator $\mathcal{F}_{\hbar}:\mathcal{H}_{\hbar}\rightarrow\mathcal{H}_{\hbar}$
and showed its nice properties. There is however a ``standard quantization''
of the map $f$ defined in the literature in the framework of geometric
quantization. In this subsection we recall first the definition of
``geometric quantization'' or ``Toeplitz quantization'' of the
symplectic map $f:M\rightarrow M$. Then we compare it with the ``natural
quantization'' $\mathcal{F}_{\hbar}:\mathcal{H}_{\hbar}\rightarrow\mathcal{H}_{\hbar}$
introduced above.

In Section \ref{sub:The-rough-Laplacian} we introduce the rough Laplacian
operator $\Delta$. In Theorem \ref{thm:band_structure-of_Laplacian},
for each $N$ large enough, we define $\mathfrak{P}_{0}$ the spectral
projector of the rough Laplacian on its first band $k=0$. See also
Figure \ref{fig:Spectrum-of-the_Laplacian}. The projector $\mathfrak{P}_{0}$
has finite rank, the same rank as the spectral projector $\Pi_{\hbar}$
on the external band of $\hat{F}_{N}$. There is some major difference
between these projectors even if they have same rank: $\Pi_{\hbar}$
depends on the map $f$ and its image $\mathrm{Im}\left(\Pi_{\hbar}\right)\subset\mathcal{D}_{N}^{'}\left(P\right)$
consists of distributions whereas the projector $\mathfrak{P}_{0}$
does not depend on the map $f$ and its image $\mathrm{Im}\left(\mathfrak{P}_{0}\right)\subset C_{N}^{\infty}\left(P\right)$
contains smooth sections. The projector $\mathfrak{P}_{0}$ depends
only on the symplectic space $\left(M,\omega\right)$ together with
an additional compatible metric $g$. The following definition is
standard in ``geometric quantization'' \cite{ma_08_toeplitz}.

%red 
\begin{center}{\color{red}\fbox{\color{black}\parbox{16cm}{
\begin{defn}
\label{def_toeplitz}For every $N$ large enough, the \textbf{Toeplitz
quantum space} $\mathcal{H}_{T}$ is the finite dimensional space
\begin{equation}
\mathcal{H}_{T}:=\mbox{Im}\left(\mathfrak{P}_{0}\right)\label{eq:def_Toeplitz_quantum_space}
\end{equation}
The \textbf{Toeplitz quantum operator} is 
\begin{equation}
\hat{\mathcal{F}}_{T,V}:=\mathfrak{P}_{0}\hat{F}_{N}\mathfrak{P}_{0}\quad:\mathcal{H}_{T}\rightarrow\mathcal{H}_{T}\label{eq:Toeplitz_quantization}
\end{equation}
In the notation we have emphasized the dependence on the potential
$V$ which enters in the definition \ref{eq:def_quantum_op} of $\hat{F}_{N}$.
\end{defn}
}}}\end{center}

In (\ref{eq:def_isomorphism-1}) we obtain the following Lemma:
\begin{lem}
For $N=\frac{1}{2\pi\hbar}$ large enough we have that 
\begin{equation}
\Phi:=\mathfrak{P}_{0}:\mathcal{H}_{\hbar}\rightarrow\mathcal{H}_{T}\label{eq:def_isomorphism}
\end{equation}
is a finite rank isomorphism.
\end{lem}
The previous Lemma allows to consider $\Phi\hat{\mathcal{F}}_{\hbar}\Phi^{-1}:\mathcal{H}_{T}\rightarrow\mathcal{H}_{T}$
and thus to ``compare'' the quantum operator $\hat{\mathcal{F}}_{\hbar}:\mathcal{H}_{\hbar}\rightarrow\mathcal{H}_{\hbar}$
with the Toeplitz quantum operator $\hat{\mathcal{F}}_{T}:\mathcal{H}_{T}\rightarrow\mathcal{H}_{T}$.
The next Theorem shows that they are close to each other under the
condition that one adds a correction in the potential function. For
that reason we will emphasize the dependence on the potential $V$
which enters in the definition of $\hat{\mathcal{F}}_{\hbar}$ by
noting it $\hat{\mathcal{F}}_{\hbar,V}$.

%blue 
\begin{center}{\color{blue}\fbox{\color{black}\parbox{16cm}{
\begin{thm}
\textbf{\label{thm:The-quantum-operator_closed_to_Toeplitz}``The
quantum operator $\hat{\mathcal{F}}_{\hbar}$ is close to a Toeplitz
quantum operator $\hat{\mathcal{F}}_{T}$''}. There exists $\delta>0$,
there exists a function $\mathcal{M}\in C\left(M\right)$ such that
for $\hbar\rightarrow0$, 
\begin{equation}
\left\Vert \Phi\hat{\mathcal{F}}_{\hbar,V}\Phi^{-1}-\hat{\mathcal{F}}_{T,V'}\right\Vert \leq\mathcal{O}\left(\hbar^{\delta}\right),\label{eq:quantum_operator_closed_to_Toeplitz}
\end{equation}
where $\hat{\mathcal{F}}_{T,V'}:\mathcal{H}_{T}\rightarrow\mathcal{H}_{T}$
is the Toeplitz operator (\ref{eq:Toeplitz_quantization}) but constructed
with the potential $V'=V+\mathcal{M}$.
\end{thm}
}}}\end{center}

The proof of Theorem \ref{thm:The-quantum-operator_closed_to_Toeplitz}
is given in Section \ref{sec:6}.
\begin{rem}
We have already said but we repeat here, that Toeplitz quantization
(or geometric quantization%
\footnote{or Weyl quantization when available e.g. if $M=\mathbb{T}^{2d}$ is
a torus.%
}) gives a family of ``quantum operators'' $\hat{\mathcal{F}}_{T,V'}$
that have nice semiclassical properties as any quantization gives.
For example they satisfy Egorov formula (\ref{eq:exact_Egorov_formula})
but with an error of order $O\left(h\right)$. They satisfy the Gutzwiller
trace formula (\ref{eq:Gutz_Trace_Formula}) but with a larger error
$O\left(\theta^{n}\right)$ with $\theta>1$. The ``natural quantization''
of the Anosov map $f$ presented in this paper has the unique additional
property that the errors in these formula vanish (or are improved). \end{rem}

\section{\label{sec:Preliminary-results-and}Preliminary results and sketch
of the proofs}

In this Section we begin with establishing some preliminary results
which will be used in the rest of the paper. Then we sketch the proofs
of the main theorems presented in Section \ref{sec:Introduction-and-results}.

\subsection{\label{sub:Semiclassical-description-of-prequantum_op}Semiclassical
description of the prequantum operator $\hat{F}_{N}$}

\subsubsection{The associated canonical map $F:T^{*}M\rightarrow T^{*}M$}

We first give a local expression of the transfer operator $\hat{F}_{N}$
defined in (\ref{eq:def_F_N_on_P}) with respect to local charts and
local trivialization of the bundle $P$. These local expressions will
be useful in the sequel of the paper.

As in Section \ref{sub:Hypothesis}, let $\left(U_{\alpha}\right)_{\alpha\in I}$
be a finite collection of simply connected open subsets which cover
$M$ and, for every open set $U_{\alpha}\subset M$, let
\[
\tau_{\alpha}:U_{\alpha}\rightarrow P
\]
be a local section of the bundle. Recall the local trivialization
$T_{\alpha}$ defined in (\ref{eq:local_trivialization-1}) and the
one forms $\eta_{\alpha}$ in the local expression (\ref{eq:def_eta-1})
of the connection $A$.

%blue
\begin{center}{\color{blue}\fbox{\color{black}\parbox{16cm}{
\begin{lem}
\textbf{``Local expression of the prequantum map $\tilde{f}$''.
}Suppose that $V\subset U_{\alpha}\cap f^{-1}\left(U_{\beta}\right)$
is a simply connected open set. Then for $x\in V$
\begin{equation}
\tilde{f}\left(\tau_{\alpha}\left(x\right)\right)=e^{i2\pi\mathcal{A}_{\beta,\alpha}\left(x\right)}\tau_{\beta}\left(f\left(x\right)\right)\label{eq:f_tilde_local}
\end{equation}
with the ``action function'' given by
\begin{equation}
\mathcal{A}_{\beta,\alpha}\left(x\right)=\int_{f\left(\gamma\right)}\eta_{\beta}-\int_{\gamma}\eta_{\alpha}+c\left(x_{0}\right)=\int_{\gamma}\left(f^{*}\left(\eta_{\beta}\right)-\eta_{\alpha}\right)+c\left(x_{0}\right).\label{eq:def_A_action}
\end{equation}
In the last integral, $x_{0}\in V$ is any point of reference, $\gamma\subset V$
is a path from $x_{0}$ to $x$ and $c\left(x_{0}\right)$ does not
depend on $x$. See figure \ref{fig:map f_tilde-1}.
\end{lem}
}}}\end{center}

\begin{figure}[h]
\centering{}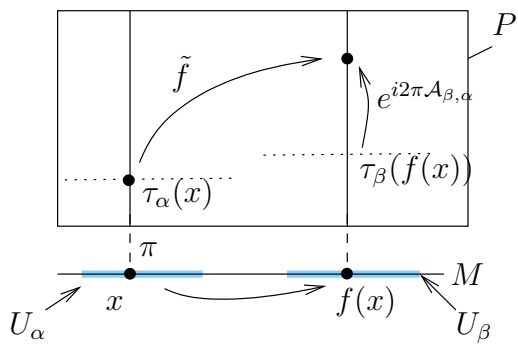\caption{\label{fig:map f_tilde-1}Illustrates the expression (\ref{eq:f_tilde_local})
of the prequantum map $\tilde{f}$ with respect to local trivialization.
It is characterized by the action function $\mathcal{A}_{\beta,\alpha}\left(x\right)$.}
\end{figure}

\begin{rem}
In other terms, if $p\in P$ is such that $\pi\left(p\right)=x$,
let $\theta_{\alpha},\theta_{\beta}$ such that
\[
p=e^{i\theta_{\alpha}}\tau_{\alpha}\left(x\right),\quad\tilde{f}\left(p\right)=e^{i\theta_{\beta}}\tau_{\beta}\left(f\left(x\right)\right)
\]
then
\[
\theta_{\beta}=\theta_{\alpha}+2\pi\mathcal{A}_{\beta,\alpha}\left(x\right)
\]

\end{rem}
~
\begin{rem}
Notice that the integral (\ref{eq:def_A_action}) does not depend
on the path $\gamma$ from $x_{0}$ to $x$ because the one form $f^{*}\left(\eta_{\beta}\right)-\eta_{\alpha}$
is closed. Indeed $d\left(f^{*}\left(\eta_{\beta}\right)-\eta_{\alpha}\right)=f^{*}\omega-\omega=0$
since $f$ is symplectic.\end{rem}
\begin{proof}
Let $\gamma\subset V$ be a path from $x_{0}$ to $x$. Let $\tilde{\gamma}:t\rightarrow\tilde{\gamma}\left(t\right)$
be the lifted path parallel transported above $\gamma$ starting from
$\tau_{\alpha}\left(x_{0}\right)$ and ending at point $p$. (See
Figure \ref{fig:Local expression of FN}.) Since the connection one
form vanishes along the path $\tilde{\gamma}$, we have
\[
0=\left(T_{\alpha}^{*}A\right)\left(\frac{d\tilde{\gamma}_{\alpha}}{dt}\right)=\left(id\theta-i2\pi\eta_{\alpha}\right)\left(\frac{d\tilde{\gamma}_{\alpha}}{dt}\right)
\]
with $\tilde{\gamma}_{\alpha}=T_{\alpha}^{-1}\left(\tilde{\gamma}\right)$.
From the construction of the lifted map $\tilde{f}$ in the proof
of Lemma \ref{pro:map_f_tilde-1}, we have
\begin{equation}
p=e^{i\theta_{\alpha}\left(x\right)}\tau_{\alpha}\left(x\right)\label{eq:e1}
\end{equation}
with 
\[
\theta_{\alpha}\left(x\right)=\int_{\tilde{\gamma}}d\theta=2\pi\int_{\gamma}\eta_{\alpha}.
\]

\begin{figure}[h]
\centering{}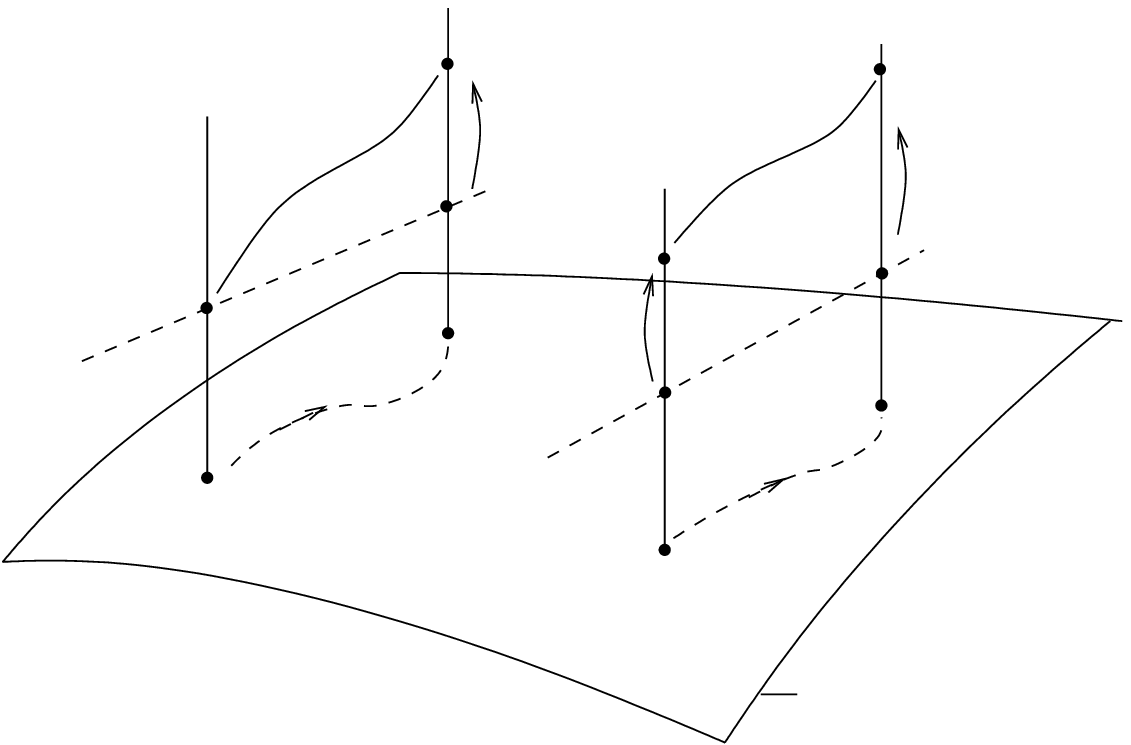\caption{\label{fig:Local expression of FN}}
\end{figure}
Let $\theta_{0}$ given by $\tilde{f}\left(\tau_{\alpha}\left(x_{0}\right)\right)=e^{i\theta_{0}}\tau_{\beta}\left(f\left(x_{0}\right)\right)$.
Similarly we have
\[
\tilde{f}\left(p\right)=e^{i\theta_{0}}e^{ib\left(x\right)}\tau_{\beta}\left(f\left(x\right)\right)
\]
with
\[
b\left(x\right)=2\pi\int_{f\left(\gamma\right)}\eta_{\beta}.
\]
From equivariance of $\tilde{f}$ and (\ref{eq:e1}), we have $\tilde{f}\left(p\right)=e^{i\theta_{\alpha}\left(x\right)}\tilde{f}\left(\tau_{\alpha}\left(x\right)\right)$.
Therefore
\[
\tilde{f}\left(\tau_{\alpha}\left(x\right)\right)=e^{-i\theta_{\alpha}\left(x\right)}\tilde{f}\left(p\right)=e^{-i\theta_{\alpha}\left(x\right)}e^{i\theta_{0}}e^{ib\left(x\right)}\tau_{\beta}\left(f\left(x\right)\right)=e^{i2\pi\mathcal{A}_{\beta,\alpha}\left(x\right)}\tau_{\beta}\left(f\left(x\right)\right)
\]
with 
\[
\mathcal{A}_{\beta,\alpha}\left(x\right)=\int_{f\left(\gamma\right)}\eta_{\beta}-\int_{\gamma}\eta_{\alpha}+\frac{\theta_{0}}{2\pi}.
\]

\end{proof}
For a given $N\in\mathbb{Z}$, an equivariant function $u\in C_{N}^{\infty}\left(P\right)$
defines the set of \emph{associated functions}%
\footnote{In the language of associate line bundle, these functions $u_{\alpha}$
are sections of a line bundle $L^{\otimes N}$ expressed with respect
to the local trivializations. %
} $u_{\alpha}:U_{\alpha}\rightarrow\mathbb{C}$, $\alpha\in I$, defined
by 
\begin{equation}
u_{\alpha}\left(x\right):=u\left(\tau_{\alpha}\left(x\right)\right)\label{eq:def_u_alpha}
\end{equation}
for $x\in U_{\alpha}$. Conversely one reconstructs $u$ from $(u_{\alpha})_{\alpha\in I}$
by the relation 
\[
u\left(p\right)=u\left(e^{i\theta}\tau_{\alpha}\left(x\right)\right)=e^{iN\theta}u\left(\tau_{\alpha}\left(x\right)\right)=e^{iN\theta}u_{\alpha}\left(x\right)\quad\mbox{for }p=e^{i\theta}\tau_{\alpha}\left(x\right)\mbox{ and }x\in U_{\alpha}.
\]

%blue
\begin{center}{\color{blue}\fbox{\color{black}\parbox{16cm}{
\begin{prop}
\textbf{\label{prop:Local-expression-of_FN}``Local expression of
$\hat{F}_{N}$''.} Let $u\in C_{N}^{\infty}\left(P\right)$ and $u':=\hat{F}_{N}u\quad\in C_{N}^{\infty}\left(P\right)$.
Let the respective associated functions be $u_{\alpha}=u\circ\tau_{\alpha}$
and $u'_{\alpha}=u'\circ\tau_{\alpha}$ for any $\alpha\in I$. Then
\begin{equation}
u'_{\beta}=e^{V}\cdot e^{-i2\pi N\mathcal{A}_{\beta,\alpha}\circ f^{-1}}\left(u_{\alpha}\circ f^{-1}\right)\label{eq:expression_op_locale}
\end{equation}

\end{prop}
}}}\end{center}
\begin{proof}
From definition (\ref{eq:def_prequantum_operator_F}) of the transfer
operator
\[
u'_{\beta}\left(x\right)=u'\left(\tau_{\beta}\left(x\right)\right)=\left(\hat{F}u\right)\left(\tau_{\beta}\left(x\right)\right)=e^{V\left(x\right)}u\left(\tilde{f}^{-1}\left(\tau_{\beta}\left(x\right)\right)\right)
\]
From (\ref{eq:f_tilde_local}) we have
\[
\tilde{f}^{-1}\left(\tau_{\beta}\left(x\right)\right)=e^{-i2\pi\mathcal{A}_{\beta,\alpha}\left(f^{-1}\left(x\right)\right)}\tau_{\alpha}\left(f^{-1}\left(x\right)\right)
\]
hence
\begin{eqnarray*}
u'_{\beta}\left(x\right) & = & e^{V\left(x\right)}u\left(e^{-i2\pi\mathcal{A}_{\beta,\alpha}\left(f^{-1}\left(x\right)\right)}\tau_{\alpha}\left(f^{-1}\left(x\right)\right)\right)\\
 & = & e^{V\left(x\right)}e^{-i2\pi N\mathcal{A}_{\beta,\alpha}\left(f^{-1}\left(x\right)\right)}u\left(\tau_{\alpha}\left(f^{-1}\left(x\right)\right)\right)\\
 & = & e^{V\left(x\right)}e^{-i2\pi N\mathcal{A}_{\beta,\alpha}\left(f^{-1}\left(x\right)\right)}u_{\alpha}\left(f^{-1}\left(x\right)\right)
\end{eqnarray*}

\end{proof}

The next few results will concern some semiclassical aspects of the
transfer operator. The first easy result but of fundamental importance
is the following Proposition. (It is not used in the proofs in this
paper and the explanation below may be a bit sloppy. But it gives
the background of the following argument.)
\begin{rem}
For the general definition of a Fourier integral operator, we refer
to Martinez \cite{martinez-01}, Zworski \cite{zworski-03} or Duistermaat
\cite{duistermaat_FIO_96}. If the reader is not familiar with Fourier
integral operator, it is enough for the reading of this paper to understand
the rough idea of a Fourier integral operator $\hat{F}$, which is
quite simple as explained in \cite{cordoba_78}: if a function $\psi$
is localized at point $x\in M$ and its Fourier transform is localized
at point $\xi\in T_{x}^{*}M$ (which means that these functions decay
fast outside these points and we say that $\psi$ is\textbf{ }\emph{micro-localized}
at $\left(x,\xi\right)\in T^{*}M$ ) then the operator $\hat{F}$
transforms this function $\psi$ to a function $\psi'$ micro-localized
in another point $\left(x',\xi'\right)\in T^{*}M$. The map $F:T^{*}M\rightarrow T^{*}M$,
giving $\left(x',\xi'\right)=F\left(x,\xi\right)$ is called the associated
\textbf{canonical map}. Note that we are interested in the situation
of high frequencies $|\xi|\gtrsim N=(2\pi\hbar)^{-1}\gg1$ and, in
the limit $N\to\infty$ (or $\hbar\to+0$), we will normalize $\xi$
by multiplying $\hbar$.
\end{rem}
%blue
\begin{center}{\color{blue}\fbox{\color{black}\parbox{16cm}{
\begin{prop}
\label{prop:The-transfer-operator_is_FIO}The prequantum transfer
operator $\hat{F}_{N}$ is a \emph{Fourier Integral Operator} if we
view it in the local trivializations as in (\ref{eq:expression_op_locale}).
The \emph{associated canonical map} on the cotangent space is given
by
\begin{eqnarray}
F_{\alpha,\beta}:\begin{cases}
T^{*}U_{\alpha} & \rightarrow T^{*}U_{\beta}\\
\left(x,\xi\right) & \rightarrow\left(x',\xi'\right)=\left(f\left(x\right),{}^{t}\left(Df_{x'}^{-1}\right)\left(\xi+\eta_{\alpha}\left(x\right)\right)-\eta_{\beta}\left(x'\right)\right)
\end{cases}\label{eq:expresion_canonical_map_Fab}
\end{eqnarray}

where $x\in U_{\alpha}$, $f\left(x\right)\in U_{\beta}$ and $\xi\in T_{x}^{*}U_{\alpha}$.
The map $F_{\alpha,\beta}$ preserves the canonical symplectic structure
\begin{equation}
\Omega:=\sum_{j=1}^{2d}dx_{j}\wedge d\xi_{j}\label{eq:big_Omega}
\end{equation}

\end{prop}
}}}\end{center}
\begin{proof}
The transfer operator is given in local chart by (\ref{eq:expression_op_locale}).
This expression shows that it is the composition of two operators
$\hat{F}_{2}\circ\hat{F}_{1}$ where $\hat{F}_{1}$ is the pull-back
by the diffeomorphism $f^{-1}$ and $\hat{F}_{2}$ is the multiplication
by a phase function. Both operators are basic examples of F.I.O as
explained in~\cite[chap.5]{martinez-01} and the Proposition \ref{prop:The-transfer-operator_is_FIO}
follows.\end{proof}
\begin{rem}
We give here a more detailed explanation of the associated canonical
map (\ref{eq:expresion_canonical_map_Fab}). From (\ref{eq:expression_op_locale}),
the transfer operator $\hat{F}_{N}$ can be decomposed as elementary
operators. Let 
\[
\hat{F}_{1}:\quad u\left(x\right)\rightarrow u\left(f^{-1}\left(x\right)\right).
\]
This is a Fourier integral operator with canonical map
\begin{equation}
F_{1}\left(\begin{array}{c}
x\\
\xi
\end{array}\right)=\left(\begin{array}{c}
x'\\
\xi'
\end{array}\right)=\left(\begin{array}{c}
f\left(x\right)\\
^{t}\left(Df_{x'}^{-1}\right)\xi
\end{array}\right)\label{eq:canonical_map_F1}
\end{equation}
Indeed, it is clear that $\mbox{supp}\left(\psi\right)$ is transported
to $f\left(\mbox{supp}\left(\psi\right)\right)$ hence $x'=f\left(x\right)$.
Also an oscillating function $u\left(x\right)=e^{\frac{i}{\hbar}\xi.x}$
is transformed to $u'\left(y\right)=\left(\hat{F}_{1}u\right)\left(y\right)=e^{\frac{i}{\hbar}\xi.f^{-1}\left(y\right)}$
and, for $y=f\left(x\right)+y'$ with $\left|y'\right|\ll1$, we have
$f^{-1}\left(y\right)=x+Df_{y}^{-1}.y'+o\left(\left|y'\right|\right)$
, hence 
\[
u'\left(y\right)\simeq e^{\frac{i}{\hbar}\xi.\left(x+Df_{x'}^{-1}.y'\right)}=Ce^{\frac{i}{\hbar}\,^{t}\left(Df_{x'}^{-1}\right)\xi.y}=Ce^{\frac{i}{\hbar}\xi'.y}
\]
with $\xi'={}^{t}\left(Df_{x'}^{-1}\right)\xi$ and $C=e^{\frac{i}{\hbar}\xi\cdot(x-Df_{x'}^{-1}\cdot f(x))}$.
We deduced (\ref{eq:canonical_map_F1}). Next we consider a multiplication
operator by a ``fast oscillating phase'' (recall that $\hbar\ll1$):
\[
\hat{F}_{2}:\quad\psi\left(x\right)\rightarrow e^{iS\left(x\right)/\hbar}\psi\left(x\right)
\]
For the same reasons, it is a F.I.O. and its canonical map is
\[
F_{2}\left(\begin{array}{c}
x\\
\xi
\end{array}\right)=\left(\begin{array}{c}
x\\
\xi+dS(x)
\end{array}\right)
\]
Indeed $u\left(x\right)=e^{\frac{i}{\hbar}\xi.x}$ is transformed
to $u'\left(y\right)=\left(\hat{F}_{2}u\right)\left(y\right)=e^{\frac{i}{\hbar}\left(\xi.y+S\left(y\right)\right)}$
and for $y=x+y'$ with $\left|y'\right|\ll1$, we have
\[
u'\left(x'\right)\simeq Ce^{\frac{i}{\hbar}\left(\xi.y+dS\cdot y\right)}=Ce^{\frac{i}{\hbar}\xi'.y}
\]
with $\xi'=\xi+dS$ and $C=e^{\frac{i}{\hbar}\left(S\left(x\right)-dS_{x}.x\right)}$.
From these two previous examples and (\ref{eq:expression_op_locale}),
we can deduce (\ref{eq:expresion_canonical_map_Fab}). Notice that
the multiplication operator by $e^{V}$ does not appear in the canonical
map because it is not a ``fast oscillating function'' (in the limit
$\hbar\rightarrow0$).
\end{rem}
Notice that the canonical maps $F_{\alpha,\beta}$ in the last proposition
are given with respect to local trivializations of $P$. The following
proposition gives a global and geometric description of the canonical
map (\ref{eq:expresion_canonical_map_Fab}).

%blue
\begin{center}{\color{blue}\fbox{\color{black}\parbox{16cm}{
\begin{prop}
\label{prop:Canonical_map F}Consider the following change of variable
on $T^{*}U_{\alpha}$ for every $\alpha\in I$:
\begin{equation}
\left(x,\xi\right)\in T^{*}U_{\alpha}\rightarrow\left(x,\zeta\right)=\left(x,\xi+\eta_{\alpha}\left(x\right)\right)\in T^{*}M\label{eq:change_variable_zeta}
\end{equation}
Then the \textbf{canonical map} (\ref{eq:expresion_canonical_map_Fab})
get the simpler and global expression (independent on the set $U_{\alpha}$)
on the phase space $T^{*}M$
\begin{equation}
F:\begin{cases}
T^{*}M & \rightarrow T^{*}M\\
\left(x,\zeta\right) & \rightarrow\left(x',\zeta'\right)=\left(f(x),{}^{t}\left(Df_{x'}^{-1}\right)\zeta\right)
\end{cases}.\label{eq:new_expression_of_Fab}
\end{equation}
The symplectic form $\Omega$in (\ref{eq:big_Omega}) preserved by
$F$ is expressed:
\begin{equation}
\Omega=\sum_{j=1}^{2d}\left(dx_{j}\wedge d\zeta_{j}\right)+\tilde{\pi}^{*}\left(\omega\right).\label{eq:new_expression_Omega}
\end{equation}
with the canonical projection map $\tilde{\pi}:T^{*}M\rightarrow M$.
\end{prop}
}}}\end{center}
\begin{rem}
We will see in Remark \ref{rem:2.22} that the variables $\zeta$
can be interpreted as the symbol of the covariant derivative operator
$D$. The change of variables (\ref{eq:change_variable_zeta}) and
the global geometric description (\ref{eq:new_expression_of_Fab})
is standard in mathematical physics for semiclassical problems involving
large magnetic fields on manifolds. \end{rem}
\begin{proof}
The relation (\ref{eq:new_expression_of_Fab}) is obvious from (\ref{eq:expresion_canonical_map_Fab})
and (\ref{eq:change_variable_zeta}). To prove (\ref{eq:new_expression_Omega}),
we write in coordinates
\[
\eta_{\alpha}=\sum_{j=1}^{2d}\eta_{j}dx_{i}
\]
and from (\ref{eq:omega_eta}),
\[
\omega=d\eta_{\alpha}=\sum_{i,j}\left(\frac{\partial\eta_{j}}{\partial x_{i}}\right)\left(dx_{i}\wedge dx_{j}\right).
\]
Then from (\ref{eq:big_Omega})
\begin{eqnarray*}
\Omega & = & \sum_{j=1}^{2d}dx_{j}\wedge d\xi_{j}=\sum_{j=1}^{2d}dx_{j}\wedge\left(d\zeta_{j}-d\eta_{j}\right)=\sum_{j=1}^{2d}dx_{j}\wedge\left(d\zeta_{j}-\sum_{i=1}^{2d}\left(\frac{\partial\eta_{j}}{\partial x_{i}}\right)dx_{i}\right)\\
 & = & \sum_{j=1}^{2d}\left(dx_{j}\wedge d\zeta_{j}\right)+\omega
\end{eqnarray*}

\end{proof}

\subsubsection{The trapped set $K$}

In relation to the global expression (\ref{eq:new_expression_of_Fab})
of the canonical map associated to the prequantum transfer operator
$\hat{F}_{N}$, it should be natural to introduce the following definition.

%red
\begin{center}{\color{red}\fbox{\color{black}\parbox{16cm}{
\begin{defn}
The\textbf{ }\emph{trapped set }$K\Subset T^{*}M$ is the set of points
$\left(x,\xi\right)\in T^{*}M$ which do not escape to infinity in
the past neither in the future $n\rightarrow\pm\infty$ under the
dynamics of the canonical map $F$ :
\[
K:=\left\{ \left(x,\xi\right),\exists C>0,\left|\xi\left(n\right)\right|\leq C,\forall n\in\mathbb{Z},\mbox{ with }\left(x\left(n\right),\xi\left(n\right)\right):=F^{n}\left(x,\xi\right)\right\} .
\]

\end{defn}
}}}\end{center}
\begin{rem}
In terms of the theory of dynamical systems, the trapped set $K$
is the \emph{non-wandering set} for the dynamical system generated
by $F$.
\end{rem}
At every point $\rho\in K$ of the trapped set we have the decomposition 

\[
T_{\rho}\left(T^{*}M\right)=E_{u}^{*}\left(\rho\right)\oplus E_{s}^{*}\left(\rho\right)
\]
 into the unstable and stable subspaces with respect to the action
of $DF$ defined by: 
\[
E_{u}^{*}(\rho):=\{v\in T_{\rho}(T^{*}M)\mid|DF_{\rho}^{-n}(v)|\to0\mbox{ as }n\to+\infty\}
\]
and 
\[
E_{s}^{*}(\rho):=\{v\in T_{\rho}(T^{*}M)\mid|DF_{\rho}^{n}(v)|\to0\mbox{ as }n\to+\infty\}.
\]

This decomposition is dual to the decomposition (\ref{eq:foliation})
in the sense that in the fiber $T_{\rho}^{*}M$ we have $\left(E_{u}^{*}\cap T_{\rho}^{*}M\right)\left(E_{u}\right)=0$,
$\left(E_{s}^{*}\cap T_{\rho}^{*}M\right)\left(E_{s}\right)=0$.

%blue
\begin{center}{\color{blue}\fbox{\color{black}\parbox{16cm}{
\begin{prop}
\textbf{``Description of the trapped set $K$''.} The trapped set
$K\subset T^{*}M$ is the zero section:
\begin{equation}
K=\left\{ \left(x,\zeta\right)\in T^{*}M\mid x\in M,\,\zeta=0\right\} .\label{eq:K=00003Dzeta=00003D0}
\end{equation}
This is a symplectic submanifold of $\left(T^{*}M,\Omega\right)$
isomorphic to $\left(M,\omega\right)$. For every point $\rho\in K$,
the tangent space is decomposed as an $\Omega$-orthogonal sum of
\emph{symplectic} linear subspaces 
\begin{equation}
T_{\rho}\left(T^{*}M\right)=T_{\rho}K\overset{\perp}{\oplus}\left(T_{\rho}K\right)^{\perp}\label{eq:TK_TKorth}
\end{equation}
Moreover each part is decomposed into \emph{isotropic} unstable/stable
linear spaces
\begin{equation}
T_{\rho}K=E_{u}^{\left(1\right)}\left(\rho\right)\oplus E_{s}^{\left(1\right)}\left(\rho\right),\qquad\left(T_{\rho}K\right)^{\perp}=E_{u}^{\left(2\right)}\left(\rho\right)\oplus E_{s}^{\left(2\right)}\left(\rho\right)\label{eq:def_E_u_s_1}
\end{equation}
where \textup{the subspaces $E_{\sigma}^{\left(i\right)}\left(\rho\right)$
for $i=1,2$ and $\sigma=s,u$ are $d$-dimensional subspaces defined
by} 
\[
E_{\sigma}^{\left(1\right)}\left(\rho\right)=T_{\rho}K\cap E_{\sigma}^{*}(\rho),\quad E_{\sigma}^{\left(2\right)}\left(\rho\right)=(T_{\rho}K)^{\perp}\cap E_{\sigma}^{*}(\rho)\quad\mbox{for \ensuremath{\sigma=s,u}.}
\]
All the decompositions above are invariant by the map $F$. See Figure
\ref{fig:Symplectic-orthogonal-decomposit-1}.
\end{prop}
}}}\end{center}

\begin{center}
\begin{figure}[h]
\begin{centering}
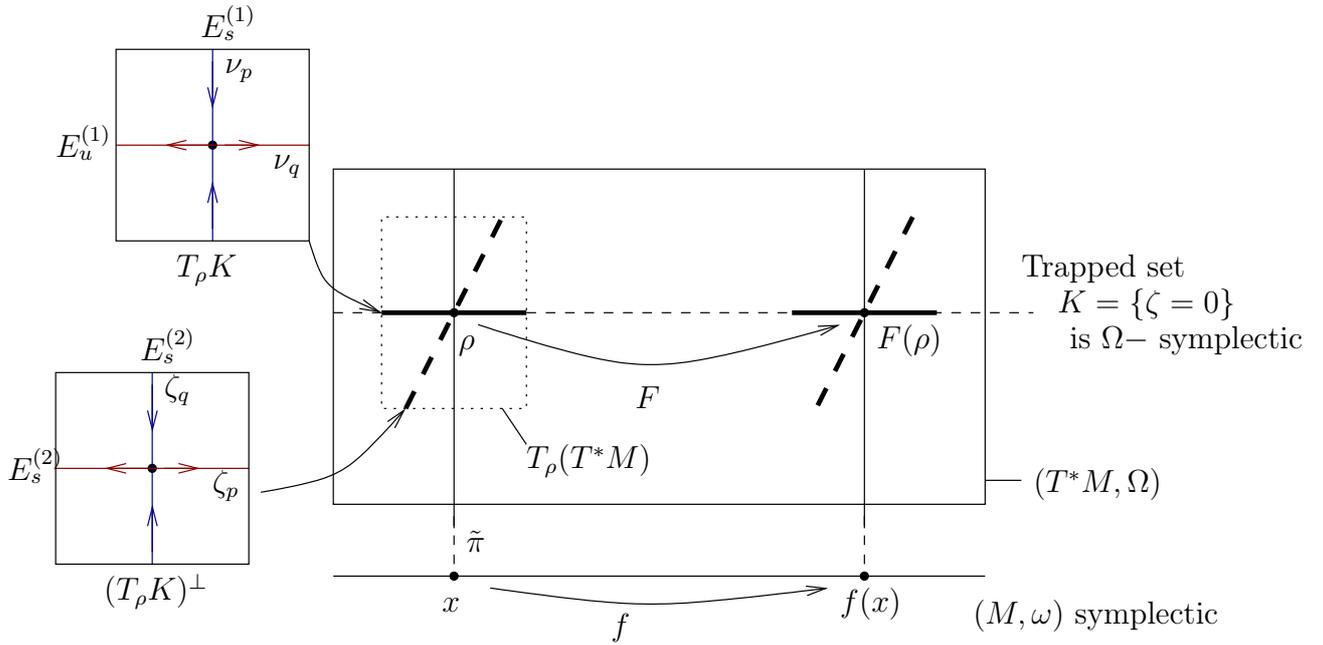
\par\end{centering}

\caption{\label{fig:Symplectic-orthogonal-decomposit-1} The decompositions
of the tangent space $T_{\rho}(T^{*}M)$.}
\end{figure}

\par\end{center}
\begin{rem}
There is another $F$-invariant decomposition:
\[
\underbrace{T_{\rho}\left(T^{*}M\right)}_{4d}=\underbrace{T_{\rho}K}_{2d}\oplus\underbrace{T_{x}^{*}M}_{2d}
\]
with $x=\tilde{\pi}\left(\rho\right)$ and $T_{x}^{*}M$ the fiber
of the cotangent space. However this sum is not $\Omega$-orthogonal
and moreover $T_{x}^{*}M$ is $\Omega$-Lagrangian.
\end{rem}
~
\begin{rem}
It is shown in \cite[appendix]{berman_sjostrand_2005_asymptotics}
that the cotangent space $\left(T^{*}M,\Omega\right)$ is an ``affine
cotangent space'' and can be geometrically defined as the space of
connections on the principal bundle $P\rightarrow M$. The trapped
set $K$ is the section defined by connection itself \cite[appendix]{berman_sjostrand_2005_asymptotics}.\end{rem}
\begin{proof}
From (\ref{eq:new_expression_of_Fab}) it is clear that the trapped
set is the zero section $\left\{ \left(x,\zeta\right),\zeta=0\right\} $
since it is invariant and, if $\zeta\neq0$, we have $\left|F^{n}\left(x,\zeta\right)\right|\rightarrow\infty$
at least either as $n\rightarrow\infty$ or $n\to-\infty$. From (\ref{eq:new_expression_Omega})
$\Omega_{/K}=\tilde{\pi}^{*}\omega$ therefore $\tilde{\pi}:\left(K,\Omega\right)\rightarrow\left(M,\omega\right)$
is a symplectomorphism. The symplectic maps $f:M\rightarrow M$ and
$F:K\rightarrow K$ are conjugated by $\tilde{\pi}$. For every point
$\rho\in K$, $T_{\rho}K$ is a linear symplectic subspace of the
symplectic linear space $T_{\rho}\left(T^{*}M\right)$ and therefore
admits a unique symplectic orthogonal $\left(T_{\rho}K\right)^{\bot}$.
The decomposition (\ref{eq:TK_TKorth}) is invariant under the map
$F$ because the trapped set $K$ is invariant and because $F$ preserves
the symplectic form $\Omega$. 
\end{proof}
In the next proposition, we introduce convenient local coordinates,
called normal coordinates or Darboux coordinates. We will use them
later in the proof.

%blue
\begin{center}{\color{blue}\fbox{\color{black}\parbox{16cm}{
\begin{prop}
\textbf{\label{prop:Normal-coordinates.}``Normal coordinates''}.
On a sufficiently small neighborhood $U$ of every point $x\in M$,
there exist coordinates

\[
x=\left(q,p\right)=\left(q^{1},\ldots q^{d},p^{1},\ldots p^{d}\right)
\]
and a trivialization of $P$ such that the connection one-form in
(\ref{eq:def_eta-1}) is given by
\begin{equation}
\eta=\sum_{i=1}^{d}\left(\frac{1}{2}q^{i}dp^{i}-\frac{1}{2}p^{i}dq^{i}\right)\label{eq:eta_euclidian_model-1}
\end{equation}
and consequently the symplectic form $\omega$ is given by
\begin{equation}
\omega=d\eta=\sum_{i=1}^{d}dq^{i}\wedge dp^{i}.\label{eq:omega_R2n}
\end{equation}
On the cotangent bundle $T^{*}U$, there is a change of coordinates
$\Phi:\left(x,\xi\right)\rightarrow\left(\nu,\zeta\right)$ where
the variables $\zeta=\left(\zeta_{p}^{j},\zeta_{q}^{j}\right)_{j=1,\ldots,d}$
are already defined in (\ref{eq:change_variable_zeta}) as 
\begin{eqnarray}
\zeta_{p}^{i}: & = & \xi_{p}^{i}+\eta_{p}^{i}=\xi_{p}^{i}+\frac{1}{2}q^{i}\label{eq:def_zeta_j-1-1}\\
\zeta_{q}^{i}: & = & \xi_{q}^{i}+\eta_{q}^{i}=\xi_{q}^{i}-\frac{1}{2}p^{i}\nonumber 
\end{eqnarray}
while $\nu=\left(\nu_{q}^{j},\nu_{p}^{j}\right)_{j=1,\ldots,d}$ are
given by
\begin{eqnarray}
\nu_{q}^{j} & = & q^{j}-\zeta_{p}^{j}=\frac{1}{2}q^{j}-\xi_{p}^{j}\label{eq:def_nu-1}\\
\nu_{p}^{j} & = & p^{j}+\zeta_{q}^{j}=\frac{1}{2}p^{j}+\xi_{q}^{j}\nonumber 
\end{eqnarray}
This change of variables transforms the symplectic form $\Omega$
in (\ref{eq:new_expression_Omega}) to the normal form:
\begin{equation}
\Omega=\sum_{j=1}^{d}\left(d\nu_{q}^{j}\wedge d\nu_{p}^{j}\right)+\left(d\zeta_{p}^{j}\wedge d\zeta_{q}^{j}\right).\label{eq:omega_AM-1}
\end{equation}

\end{prop}
}}}\end{center}
\begin{rem}
Recall from (\ref{eq:K=00003Dzeta=00003D0}) that $T_{\rho}K=\left\{ \left(\nu,\zeta\right),\zeta=0\right\} $.
Then (\ref{eq:omega_AM-1}) implies that its symplectic orthogonal
is given by $\left(T_{\rho}K\right)^{\bot}=\left\{ \left(\nu,\zeta\right),\nu=0\right\} $.
In other words, $\left(\nu,\zeta\right)$ are symplectic coordinates
related to the decomposition (\ref{eq:TK_TKorth}). These coordinates
were introduced in the paper \cite{fred-PreQ-06} treating the linear
case, under the different names $\left(Q_{1},P_{1}\right)\equiv\left(\nu_{q},\nu_{p}\right)$
and $\left(Q_{2},P_{2}\right)\equiv\left(\zeta_{p},\zeta_{q}\right)$.\end{rem}
\begin{proof}
Darboux theorem on symplectic structure (see \cite{Aebischer}) tells
that, if we take sufficiently small neighborhood $U$ of $x$, there
exist coordinates $x=\left(q,p\right)=\left(q^{1},\ldots q^{d},p^{1},\ldots p^{d}\right)$
such that the symplectic form is expressed in the normal form $\omega=dq\wedge dp=\sum_{i=1}^{d}dq^{i}\wedge dp^{i}.$
Take any local smooth section $\tau':U\to P$ and let $\eta'$ be
the local connection one form (see (\ref{eq:def_eta-1})) with respect
to the corresponding local trivialization of $P$. Let $\eta$ be
given in (\ref{eq:eta_euclidian_model-1}). Since we have $d(\eta'-\eta)=\omega-\omega=0$
there is a smooth function $\chi:U\to\mathbb{R}$ such that $\eta'-\eta=\frac{1}{2\pi}d\chi$.
Setting $\tau=e^{i\chi}\tau'$ and recalling the formula (\ref{eq:gauge_transform}),
we see that the former statement of the proposition holds for the
coordinates $x=\left(q,p\right)=\left(q^{1},\ldots q^{d},p^{1},\ldots p^{d}\right)$
and the trivialization of $P$ associated to the local smooth section
$\tau$ thus taken. 

We prove now the latter statement. $ $ (\ref{eq:change_variable_zeta})
and (\ref{eq:eta_euclidian_model-1}) imply (\ref{eq:def_zeta_j-1-1}).
Clearly (\ref{eq:def_zeta_j-1-1}) and (\ref{eq:def_nu-1}) are coordinates
on $U$ as we can give the inverse explicitly. Starting from (\ref{eq:new_expression_Omega})
we get
\begin{eqnarray*}
\Omega & = & \sum_{j=1}^{d}\left(dq^{j}\wedge d\zeta_{q}^{j}+dp^{j}\wedge d\zeta_{p}^{j}+dq^{j}\wedge dp^{j}\right)\\
 & = & \sum_{j=1}^{d}\left(\left(d\nu_{q}^{j}+d\zeta_{p}^{j}\right)\wedge d\zeta_{q}^{j}+\left(d\nu_{p}^{j}-d\zeta_{q}^{j}\right)\wedge d\zeta_{p}^{j}+d\left(\nu_{q}^{j}+\zeta_{p}^{j}\right)\wedge\left(d\nu_{p}^{j}-d\zeta_{q}^{j}\right)\right)\\
 & = & \sum_{j=1}^{d}d\nu_{q}^{j}\wedge d\nu_{p}^{j}-d\zeta_{q}^{j}\wedge d\zeta_{p}^{j}=\sum_{j=1}^{d}d\nu_{q}^{j}\wedge d\nu_{p}^{j}+d\zeta_{p}^{j}\wedge d\zeta_{q}^{j}.
\end{eqnarray*}
This completes the proof. \end{proof}
\begin{rem}
\label{Def:sharp_flat_operator}Since the symplectic 2-form $\omega=\sum_{j}dq^{j}\wedge dp^{j}$
on $T_{x}^{*}M\equiv\mathbb{R}^{2d}$ is non degenerate, it defines
an isomorphism, called \emph{flat operator}, 
\[
\flat:\mathbb{R}^{2d}\rightarrow\left(\mathbb{R}^{2d}\right)^{*},\quad v^{\flat}=\omega\left(v,\cdot\right).
\]
Its inverse is called the \emph{sharp operator}. 
\[
\sharp=\flat^{-1}:\left(\mathbb{R}^{2d}\right)^{*}\rightarrow\mathbb{R}^{2d}.
\]
For a one-form $\alpha\in\left(\mathbb{R}^{2d}\right)^{*}$, $\alpha^{\sharp}$
is defined by the relation $\alpha=\omega\left(\alpha^{\sharp},.\right)$.
In coordinates, for 
\[
v=\sum_{j=1}^{n}v_{q}^{j}\frac{\partial}{\partial q^{j}}+v_{p}^{j}\frac{\partial}{\partial p^{j}}\equiv\left(v_{q},v_{p}\right)\in\mathbb{R}^{2d}\quad\mbox{and }\quad\alpha=\sum_{j=1}^{n}\alpha_{q}^{j}dq^{j}+\alpha_{p}^{j}dp^{j}\equiv\left(\alpha_{q},\alpha_{p}\right)\in\left(\mathbb{R}^{2d}\right)^{*}
\]
we have
\begin{equation}
v^{\flat}=\sum_{j=1}^{n}-v_{p}^{j}dq^{j}+v_{q}^{j}dp^{j}\equiv\left(-v_{p},v_{q}\right)\in\left(\mathbb{R}^{2d}\right)^{*}\label{eq:flat-1}
\end{equation}
 
\begin{equation}
\alpha^{\sharp}=\sum_{j=1}^{n}\alpha_{p}^{j}\frac{\partial}{\partial q^{j}}-\alpha_{q}^{j}\frac{\partial}{\partial p^{j}}\equiv\left(\alpha_{p},-\alpha_{q}\right)\in\mathbb{R}^{2d}\label{eq:sharp-1}
\end{equation}
We also have 
\begin{equation}
\alpha\left(v\right)=-v^{\flat}\left(\alpha^{\sharp}\right)\in\mathbb{R},\quad\flat^{t}=-\flat,\quad\sharp^{t}=-\sharp.\label{eq:sym_J-1}
\end{equation}
where $\flat^{t},\sharp^{t}$ denote the transposed maps. Using these
notation, the relation (\ref{eq:def_nu-1}) can be written in a more
intrinsic manner: 
\begin{equation}
\nu:=x-\zeta^{\#}.\label{eq:def_nu-2}
\end{equation}
With these notations (\ref{eq:omega_AM-1}) get also a more geometric
expression: 
\begin{equation}
\Omega\left(\left(\nu_{1},\zeta_{1}\right),\left(\nu_{2},\zeta_{2}\right)\right)=\omega\left(\nu_{1},\nu_{2}\right)+\omega\left(\sharp\zeta_{2},\sharp\zeta_{1}\right)\label{eq:Omega_nu_zeta-1}
\end{equation}

\end{rem}
%blue
\begin{center}{\color{blue}\fbox{\color{black}\parbox{16cm}{
\begin{prop}
\label{prop:DF_ro} The normal coordinates $\left(\nu_{q},\nu_{p},\zeta_{p},\zeta_{q}\right)$
introduced in the last proposition can be chosen so that the differential
of the coordinate map $\Phi:\tilde{\pi}^{-1}(U)\to\real^{4d}$ at
any point $\rho\in K$ such that $\tilde{\pi}\left(\rho\right)=x\in U$
carries the subspaces $E_{u}^{\left(1\right)}\left(\rho\right)$,
$E_{s}^{\left(1\right)}\left(\rho\right)$, $E_{u}^{\left(2\right)}\left(\rho\right)$,
$E_{s}^{\left(2\right)}\left(\rho\right)$ in (\ref{eq:def_E_u_s_1})
to the subspaces
\begin{align*}
\real_{\nu_{q}}^{d}:= & \left\{ \left(\nu_{q},0,0,0\right)\mid\nu_{q}\in\real^{d}\right\} ,\quad\real_{\nu_{p}}^{d}:=\left\{ \left(0,\nu_{p},0,0\right)\mid\nu_{q}\in\real^{d}\right\} ,\\
\real_{\zeta_{p}}^{d}:= & \left\{ \left(0,0,\zeta_{p},0\right)\mid\zeta_{p}\in\real^{d}\right\} ,\quad\real_{\zeta_{q}}^{d}:=\left\{ \left(0,0,0,\zeta_{q}\right)\mid\zeta_{q}\in\real^{d}\right\} 
\end{align*}
respectively. That is to say,
\begin{eqnarray*}
T_{\rho}\left(T^{*}M\right) & = & \underbrace{E_{u}^{\left(1\right)}\left(\rho\right)\oplus E_{s}^{\left(1\right)}\left(\rho\right)}_{T_{\rho}K}\overset{\perp}{\oplus}\underbrace{E_{u}^{\left(2\right)}\left(\rho\right)\oplus E_{s}^{\left(2\right)}\left(\rho\right)}_{\left(T_{\rho}K\right)^{\perp}}\\
D\Phi\downarrow &  & \qquad\quad\downarrow\qquad\qquad\qquad\downarrow\\
T^{*}\mathbb{R}_{\left(q,p\right)}^{2d} & = & \underbrace{\left(\mathbb{R}_{\nu_{q}}^{d}\oplus\mathbb{R}_{\nu_{p}}^{d}\right)}_{T^{*}\mathbb{R}_{\nu_{q}}^{d}}\quad\overset{\perp}{\oplus}\quad\underbrace{\left(\mathbb{R}_{\zeta_{p}}^{d}\oplus\mathbb{R}_{\zeta_{q}}^{d}\right)}_{T^{*}\mathbb{R}_{\zeta_{p}}^{d}}
\end{eqnarray*}
With respect to these coordinates the differential of the canonical
map $DF_{\rho}:T_{\rho}\left(T^{*}M\right)\rightarrow T_{F\left(\rho\right)}\left(T^{*}M\right)$
is expressed as 
\begin{equation}
D\Phi\circ DF_{\rho}\circ D\Phi^{-1}=F^{\left(1\right)}\oplus F^{\left(2\right)},\qquad F^{\left(1\right)}\equiv\left(\begin{array}{cc}
A_{x} & 0\\
0 & ^{t}A_{x}^{-1}
\end{array}\right),\quad F^{\left(2\right)}\equiv\left(\begin{array}{cc}
A_{x} & 0\\
0 & ^{t}A_{x}^{-1}
\end{array}\right)\label{eq:expression_Ax}
\end{equation}
where $A_{x}\equiv Df\mid_{E_{u}\left(x\right)}:\mathbb{R}^{d}\rightarrow\mathbb{R}^{d}$
is an expanding linear map.
\end{prop}
}}}\end{center}
\begin{proof}
We can take the coordinates $x=\left(q,p\right)=\left(q^{1},\ldots q^{d},p^{1},\ldots p^{d}\right)$
in the beginning of the proof of the last proposition so that the
stable and unstable subspaces, $E_{s}(x)$ and $E_{u}(x)$, correspond
to the subspaces given by the equations $p=0$ and $q=0$ respectively,
because they are Lagrangian subspaces. Then the coordinates and the
trivialization constructed in the proof have the required property.
The differential of the Anosov map $Df_{x}:T_{x}M\rightarrow T_{f\left(x\right)}M$
splits according to the invariant decomposition $T_{x}M=E_{u}\left(x\right)\oplus E_{s}\left(x\right)$
as $Df_{x}=\left(A_{x},B_{x}\right)$ with
\[
A_{x}:=Df\mid_{E_{u}\left(x\right)}:E_{u}\left(x\right)\rightarrow E_{u}\left(f\left(x\right)\right)
\]
and
\[
B_{x}:=Df\mid_{E_{s}\left(x\right)}:E_{s}\left(x\right)\rightarrow E_{s}\left(f\left(x\right)\right)
\]
But since $E_{u}\left(x\right)$ and $E_{s}\left(x\right)$ are Lagrangian
subspaces, $\omega$ provides an isomorphism $\flat:E_{s}\left(x\right)\rightarrow E_{u}\left(x\right)^{*}$
defined by $\flat\left(S\right):U\in E_{u}\left(x\right)\rightarrow\omega\left(S,U\right)\in\mathbb{R}$
for $S\in E_{s}\left(x\right)$. Since $Df_{x}$ is symplectic, i.e.
preserves $\omega$, then $B_{x}=\flat\circ{}^{t}A_{x}^{-1}\circ\flat^{-1}$
i.e. $B_{x}$ is isomorphic to 
\[
^{t}A_{x}^{-1}:E_{u}^{*}\left(x\right)\rightarrow E_{u}^{*}\left(f\left(x\right)\right)
\]
So with this identification we have 
\begin{equation}
Df_{x}\equiv\left(\begin{array}{cc}
A_{x} & 0\\
0 & ^{t}A_{x}^{-1}
\end{array}\right)\label{eq:Df_f_with_A}
\end{equation}
 
\end{proof}

\subsubsection{\label{sub:2.1.3}Microlocal description of the prequantum transfer
operator near the trapped set}

In order to focus on the action of the canonical map $F:T^{*}M\rightarrow T^{*}M$
on the vicinity of the trapped set $K$ and relate it to the spectral
properties of the prequantum transfer operator $\hat{F_{N}}$, we
use an\emph{ escape function} (or a \emph{weight function}) $W_{\hbar}\left(x,\xi\right)$
on the phase space, which decreases strictly along the flow outside
a vicinity of the trapped set $K$, and use it to define some associate
norm and associated\textbf{ }\emph{anisotropic Sobolev spaces}%
\footnote{This idea of defining a generalized Sobolev space using an escape
(or weight) function on the phase space has been used several times
before and is not new in this paper. For instance, in the context
of semiclassical scattering theory in the phase space, such an idea
was developed by B. Helffer and J. Sjöstrand in \cite{sjostrand_87}.
In the context of transfer operators for hyperbolic dynamical systems,
it was developed in \cite{Baladi05,baladi_05} more recently. %
}. From the fact that the trapped set is compact in phase space and
from the property of the escape function mentioned above, we deduce
that the spectrum of the prequantum transfer operator $\hat{F}_{N}$
is discrete in these anisotropic Sobolev spaces \cite{fred-roy-sjostrand-07}.
The eigenvalues are called \textbf{``resonances''} (from the physical
meaning in scattering theory).

Here is how we will proceed in the proofs:
\begin{enumerate}
\item In Section \ref{ss:local_charts} we will consider a system of local
charts on the manifold $M$ depending on $\hbar$ (or on $N$), which
is of small size $\hbar^{1/2-\theta}\ll1$ in the semiclassical limit
$\hbar\rightarrow0$, and then consider the local trivializations
of the prequantum bundle $P\rightarrow M$ on each chart, as in Proposition
\ref{prop:Normal-coordinates.} and \ref{prop:DF_ro}. On each of
such charts, the map $f$ is approximated at first order by its linear
approximation, and correspondingly the canonical map $F$ is approximated
on the trapped set by its differential $D_{\rho}F$ given in (\ref{eq:expression_Ax}).
In Section \ref{sub:The-prequantum-transfer_decomposed}, we show
how to decompose the global prequantum operator into ``prequantum
operators on charts'' and how to recompose it, i.e. passing from
global to local and \emph{vice versa}.
\item In view of the decomposition above called ``microlocal decomposition'',
we study first the prequantum operator associated to a linear hyperbolic
map. This is done in Section \ref{sec:Resonance-of-hyperbolic_preq_4}.
The prequantum operator for a linear hyperbolic map turns out to be
the tensor product of two operators, according to the decomposition
(\ref{eq:expression_Ax}): one operator is a prequantum operator associated
to the linear map tangent to the trapped set $T_{\rho}K$ and the
other one is associated to the linear map in its (symplectic) orthogonal
$\left(T_{\rho}K\right)^{\perp}$. The first part is a unitary operator,
while the second part is treated by using the property of the escape
function $W_{\hbar}\left(x,\xi\right)$ as described above and shown
to display a discrete decomposition (a discrete spectrum) in Proposition
\ref{pp:structure_of_hLA}. For rigorous argument, we will present
some technical tools first: 

\begin{enumerate}
\item the Bargmann transform $\mathcal{B}_{\hbar}$ which represents functions
and operators in phase space, is explained in Section \ref{sub:The-Bargmann-transform}.
\item definition of the escape function in phase space and the associated
anisotropic Sobolev space, for the hyperbolic dynamics orthogonal
to the trapped set, in Section \ref{ss:wL2}. 
\end{enumerate}
\end{enumerate}
Let us now give a quick description of the tensor product decomposition
of point (2) above in order to explain how the band structure estimates
(\ref{eq:def_r+_r-}) may be obtained. Eq.(\ref{eq:expression_Ax})
gives a description of the differential $DF_{\rho}$ of the canonical
map $F$ at the point $\rho\in K$ of the trapped set. It is therefore
natural to consider similarly the transfer operator ``microlocally''
at point $\rho$ and express it as an operator $\mathcal{L}_{\rho}$
which somehow ``quantizes'' the symplectic map $DF_{\rho}$. We
will see in Section \ref{sec:Resonance-of-hyperbolic_preq_4} that
this operator has the following form (which can easily be guessed
from (\ref{eq:expression_op_locale}))
\[
\mathcal{L}_{\rho}=e^{V_{x}}e^{-i2\pi N\mathcal{A}_{x}}\hat{F}^{\left(1\right)}\otimes\hat{F}^{\left(2\right)}
\]
with constants $V_{x},\mathcal{A}_{x}$ with $x=\tilde{\pi}\left(\rho\right)\in M$
and where $\hat{F}^{\left(1\right)},\hat{F}^{\left(2\right)}$ are
the metaplectic operators associated to the linear symplectic map
$F^{\left(1\right)},F^{\left(2\right)}$ respectively, given by: 
\[
\hat{F}^{\left(1\right)}=\hat{F}^{\left(2\right)}=\hat{F}:\begin{cases}
C^{\infty}\left(\mathbb{R}^{d}\right) & \rightarrow C^{\infty}\left(\mathbb{R}^{d}\right)\\
u & \rightarrow\left|\mathrm{det}A_{x}\right|^{-1/2}\cdot u\circ A_{x}^{-1}
\end{cases}
\]
with $A_{x}=Df\mid_{E_{u}\left(x\right)}:\mathbb{R}^{d}\rightarrow\mathbb{R}^{d}$
being the linear expanding map. Recall that the linear map $F^{\left(1\right)}:T_{\rho}K\rightarrow T_{\rho}K$
acts on the tangent space of the trapped set. Observe that $\hat{F}^{\left(1\right)}$
is unitary in $L^{2}\left(\mathbb{R}^{d}\right)$ with spectrum on
the unit circle. On the other hand $F^{\left(2\right)}:\left(T_{\rho}K\right)^{\perp}\rightarrow\left(T_{\rho}K\right)^{\perp}$
acts on the orthogonal symplectic, transverse to the trapped set.
It is a hyperbolic linear map and we will see in Section \ref{sec:Resonances-of-linear_exp_3},
Theorem \ref{pp:structure_of_hLA}, that correspondingly the operator
$\hat{F}^{\left(2\right)}:H^{r}\left(\mathbb{R}^{d}\right)\rightarrow H^{r}\left(\mathbb{R}^{d}\right)$
in anisotropic Sobolev spaces has a discrete spectrum contained in
bands indexed by $k\in\mathbb{N}$:
\[
\left|\mathrm{det}A_{x}\right|^{-1/2}\left\Vert A_{x}\right\Vert _{\mathrm{max}}^{-k}\leq\left|z\right|\leq\left|\mathrm{det}A_{x}\right|^{-1/2}\left\Vert A_{x}\right\Vert _{\mathrm{min}}^{-k}.
\]
(We get in fact a more precise description of the operator $\hat{F}^{\left(2\right)}$
as a sum of invertible operators with controlled norms). This implies
that the operator $\mathcal{L}_{\rho}:L^{2}\left(\mathbb{R}^{d}\right)\otimes H^{r}\left(\mathbb{R}^{d}\right)\rightarrow L^{2}\left(\mathbb{R}^{d}\right)\otimes H^{r}\left(\mathbb{R}^{d}\right)$
has its spectrum contained in bands $r_{k}^{-}\left(x\right)\leq\left|z\right|\leq r_{k}^{+}\left(x\right)$
with 
\[
r_{k}^{\pm}\left(x\right)=\exp\left(D\left(x\right)\right)\left\Vert Df\mid_{E_{u}\left(x\right)}\right\Vert _{\mathrm{max/min}}^{-k},\qquad D\left(x\right):=V_{x}-\frac{1}{2}\log\mathrm{det}Df\mid_{E_{u}\left(x\right)}
\]
After averaging along the trajectory starting from $x$ during time
$n$ and taking the $\mathrm{min}/\mathrm{max}$ over $x\in M$ we
see that we obtain the estimates $r_{k}^{\pm}$ of (\ref{eq:def_r+_r-}).
In fact the precise justification of this last step in the purpose
of Sections 5 and 6 and proceed as follows:
\begin{enumerate}
\item In Section \ref{sec:Nonlinear-prequantum-maps} we develop some results
that will be used in Section \ref{sec:Proofs-of-the_mains} in order
to show that non-linearity of the map $f$ can be neglected in the
reconstruction of the global prequantum operator from its local parts.
\item In Section \ref{sec:Proofs-of-the_mains}, we assemble all these previous
results and obtain the proof of the main Theorem \ref{thm:More-detailled-description_of_spectrum}.
\end{enumerate}
Finally let us mention that Sections \ref{sec:Resonances-of-linear_exp_3}
and \ref{sec:Resonance-of-hyperbolic_preq_4}, where we study the
resonances of the linear model, as explained above, are the core of
our argument because they reveal the main mechanism responsible for
the band structure of the spectrum described in Theorem \ref{thm:band_structure}.
This mechanism was discovered in the paper \cite{fred-PreQ-06} originally
in the study of the prequantum linear cat map.

\section{\label{sec:Covariant-derivative-}Covariant derivative $D$ and spectrum
of the rough Laplacian $\Delta$}

In this Section we provide additional results concerning the rough
Laplacian operator $\Delta_{N}$ and its spectrum in the semiclassical
limit $N\rightarrow\infty$. The rough Laplacian is defined from the
covariant derivative $D$ as $\Delta=D^{*}D$. The main result presented
in Theorem \ref{thm:band_structure-of_Laplacian} shows that the bottom
of the spectrum of $\Delta_{N}$ in $L_{N}^{2}\left(P\right)$ has
discrete spectrum with a band structure on the positive real line
(beware that $\Delta$ is positive and self-adjoint). We show that
the spectral projector $\Pi_{0}$ of the prequantum operator $\hat{F}_{N}$
on its external band \emph{$\mathcal{A}_{0}$} has the same rank as
the spectral projector $\mathfrak{P}_{0}$ of $\Delta_{N}$ on its
first band. Since it is known from \cite[th. 2]{guillemin_1988_laplace}\cite[cor. 1.2]{ma_02_spin_C}
that this rank is given by an index formula, we deduce (\ref{eq:Atiyah-Singer})
and (\ref{eq:index_formula_quantum_op}). The spectral projector $\mathfrak{P}_{0}$
of the Laplacian $\Delta_{N}$ is also used in definition \ref{def_toeplitz}
page \pageref{def_toeplitz} to define the Toeplitz quantum space
and the Toeplitz quantization.

\subsection{The covariant derivative $D$}

The connection one form $A$ on $P$ given in Theorem \ref{thm:Bundle-P_map_f_tilde}
induces a differential operator $D$ called the covariant derivative.
We recall its general definition and give its expression in local
coordinates. We will use it for the definition of the rough Laplacian
operator $\Delta$ in the Section \ref{sub:The-rough-Laplacian} and
also to treat the affine models on \noun{$\mathbb{R}^{2d}$} in Section
\ref{sub:Prequantum-operator-on_R2d}.

Recall that the exterior derivative $du$ of a function $u\in C^{\infty}\left(P\right)$
(i.e. the differential of $u$) is defined at point $p\in P$ by 
\[
\left(du\right)_{p}\left(V\right)=V\left(u\right)\left(p\right)
\]
where $V\in T_{p}P$ is any tangent vector. The connection one form
$A\in C^{\infty}\left(P,\Lambda^{1}\otimes\left(i\mathbb{R}\right)\right)$
on the principal bundle $\pi:P\rightarrow M$ defines a splitting
of the tangent space at every point $p\in P$:
\[
T_{p}P=V_{p}P\oplus H_{p}P
\]
with $H_{p}P=\mbox{Ker}\left(A\left(p\right)\right)$ and $V_{p}P=\mbox{Ker}\left(\left(D\pi\right)\left(p\right)\right)$.
The subspaces $V_{p}P$ and $H_{p}P$ are called respectively vertical
subspace and horizontal subspace. We will denote 
\begin{equation}
H:T_{p}P\rightarrow H_{p}P\label{eq:def_projection_H}
\end{equation}
the projection onto the horizontal space with $\mbox{Ker}\left(H\right)=V_{p}P$.
Explicitly if $V\in T_{p}P$ then from (\ref{eq:normalization_A})
its horizontal component is
\begin{equation}
HV=V+iA\left(V\right)\frac{\partial}{\partial\theta}.\label{eq:expression_of_HV}
\end{equation}
This can be checked easily from the requirements that $A\left(HV\right)=0$
and $V-HV\in V_{p}P$.

\begin{framed}%
\begin{defn}
If $u\in C^{\infty}\left(P\right)$ is a smooth function, its\emph{
exterior covariant derivative}\textbf{ }$Du\in C^{\infty}\left(P;\Lambda^{1}\right)$
is a one form on $P$ defined by
\begin{equation}
\left(Du\right)_{p}\left(V\right)=\left(\left(HV\right)\left(u\right)\right)\left(p\right)=\left(du\left(HV\right)\right)\left(p\right)\quad\mbox{for }p\in P,V\in T_{p}P\label{eq:def_Du}
\end{equation}
\end{defn}
\end{framed}

The operator $D$ is equivariant with respect to (or commutes with)
the $\mathbf{U}(1)$ action (\ref{eq:action_U1}) in $P$ and therefore
restricts naturally to the space (\ref{eq:def_C_N_P}): 
\[
D:C_{N}^{\infty}\left(P\right)\rightarrow C_{N}^{\infty}\left(P,\Lambda^{1}\right)
\]
 for every $N\in\mathbb{Z}$.

\subsubsection{\label{sub:Expression-of-D}Expression of $D$ in local charts}

%blue
\begin{center}{\color{blue}\fbox{\color{black}\parbox{16cm}{
\begin{prop}
\label{prop:expression_of_D}With respect to the local trivialization
(\ref{eq:local_trivialization-1}) of the bundle $P$ over open sets
$U_{\alpha}\subset M$, if $u\in C_{N}^{\infty}\left(P\right)$ then
$Du$ is expressed as the first order differential operator $D_{\alpha}:C^{\infty}\left(U_{\alpha}\right)\rightarrow C^{\infty}\left(U_{\alpha},\Lambda^{1}\right)$
given by:

\begin{equation}
D_{\alpha}u_{\alpha}=du_{\alpha}+\frac{i}{\hbar}u_{\alpha}\eta_{\alpha}\label{eq:D_in_local_chart}
\end{equation}
with $u_{\alpha}:=\left(u\circ\tau_{\alpha}\right)\in C^{\infty}\left(U_{\alpha}\right)$
as in (\ref{eq:def_u_alpha}) and $D_{\alpha}u_{\alpha}:=\left(Du\right)\circ\tau_{\alpha}\in C^{\infty}\left(U_{\alpha},\Lambda^{1}\right)$.
More specifically, in the normal coordinates $x=\left(q,p\right)=\left(q^{1},\ldots q^{d},p^{1},\ldots p^{d}\right)$
and the local trivialization on $U_{\alpha}$ in Proposition \ref{prop:Normal-coordinates.},
it is expressed as 
\begin{equation}
D_{\alpha}u_{\alpha}=\frac{i}{\hbar}\left(\sum_{j=1}^{d}\left(\widehat{\zeta_{q}^{j}}u_{\alpha}\right)dq^{j}+\left(\widehat{\zeta_{p}^{j}}u_{\alpha}\right)dp^{j}\right)\label{eq:expression_of_D-1}
\end{equation}
with the basis $\left(dq,dp\right)$ of $\Lambda_{x}^{1}$, where
$\widehat{\zeta_{q}^{j}}$ and $\widehat{\zeta_{p}^{j}}$ are the
differential operators on $C^{\infty}\left(\mathbb{R}^{2d}\right)$
defined respectively by
\begin{eqnarray}
\widehat{\zeta_{q}^{j}} & := & \widehat{\xi_{q}^{j}}-\frac{1}{2}p^{j}\qquad\mbox{with }\widehat{\xi_{q}^{j}}:=-i\hbar\frac{\partial}{\partial q^{j}}\label{eq:def_zeta_hat-1}\\
\widehat{\zeta_{p}^{j}} & := & \widehat{\xi_{p}^{j}}+\frac{1}{2}q^{j}\qquad\mbox{with }\widehat{\xi_{p}^{j}}:=-i\hbar\frac{\partial}{\partial p^{j}}.\nonumber 
\end{eqnarray}

\end{prop}
}}}\end{center}

\begin{rem}
\label{rem:2.22}Notice that the canonical variables $\left(\zeta_{q},\zeta_{p}\right)$
defined in (\ref{eq:def_zeta_j-1-1}) are the symbol of the operators
(\ref{eq:def_zeta_hat-1}) as a pseudodifferential operator. In more
geometrical terms, the symbol of the covariant derivative $-i\hbar D$
is the one form $\sigma\left(-i\hbar D\right)=\zeta dx=\sum_{j}\zeta_{q}^{j}dq^{j}+\zeta_{p}^{j}dp^{j}$
on $T^{*}U_{\alpha}$. This can be understood as a generalization
of the simpler case of the exterior derivative $d:C^{\infty}\left(M,\Lambda^{p}\right)\rightarrow C^{\infty}\left(M,\Lambda^{p+1}\right)$
(by taking $\eta=0$, i.e. a connection with zero curvature), for
which the principal symbol is known to be $\sigma\left(d\right)=\frac{i}{\hbar}\left(\xi dx\right)\wedge.$,
\cite[(10.12) on p.162]{taylor_tome1}.\end{rem}
\begin{proof}
Consider local coordinates $x=\left(x^{1},\ldots x^{2d}\right)\in U_{\alpha}\subset M$
and a local trivialization of $P$ giving local coordinates $\left(x,\theta\right)$
on $P$. Let $V=V^{x}\frac{\partial}{\partial x}+V^{\theta}\frac{\partial}{\partial\theta}$
be a vector field on $P$ expressed in these local coordinates. From
(\ref{eq:expression_of_HV}) and (\ref{eq:def_eta-1}) we have
\begin{eqnarray*}
HV & = & V+iA\left(V\right)\frac{\partial}{\partial\theta}=V+\left(-d\theta\left(V\right)+2\pi\eta_{\alpha}\left(V\right)\right)\frac{\partial}{\partial\theta}\\
 & = & V^{x}\frac{\partial}{\partial x}+2\pi\eta_{\alpha}\left(V\right)\cdot\frac{\partial}{\partial\theta}.
\end{eqnarray*}
Then, from the definition (\ref{eq:def_Du}),
\[
\left(Du\right)\left(V\right)=\left(HV\right)\left(u\right)=V^{x}\frac{\partial u}{\partial x}+2\pi\eta_{\alpha}\left(V\right)\frac{\partial u}{\partial\theta}.
\]
Suppose now that $u\in C_{N}^{\infty}\left(P\right)$ and write $p=e^{i\theta}\tau_{\alpha}\left(x\right)\in P$.
Then 
\[
u\left(p\right)=u\left(e^{i\theta}\tau_{\alpha}\left(x\right)\right)=e^{iN\theta}u\left(\tau_{\alpha}\left(x\right)\right)=e^{iN\theta}u_{\alpha}\left(x\right)
\]
and
\begin{eqnarray*}
\left(Du\right)_{p}\left(V\right) & = & e^{iN\theta}\left(V^{x}\frac{\partial u_{\alpha}}{\partial x}+iN2\pi\eta_{\alpha}\left(V\right)u_{\alpha}\right)\\
 & = & e^{iN\theta}\left(du_{\alpha}\left(V\right)+iN2\pi\eta_{\alpha}\left(V\right)u_{\alpha}\right)\\
 & = & e^{iN\theta}\left(du_{\alpha}+\frac{i}{\hbar}u_{\alpha}\eta_{\alpha}\right)(V).
\end{eqnarray*}
Hence
\[
D_{\alpha}u_{\alpha}=\left(Du\right)_{p}\circ\tau_{\alpha}=du_{\alpha}+\frac{i}{\hbar}u_{\alpha}\eta_{\alpha}
\]
We obtain the rest of the claims by simple calculation.
\end{proof}

\subsection{\label{sub:The-rough-Laplacian}The rough Laplacian $\Delta$}

In order to define the adjoint operator $D^{*}$ and the Laplacian
$\Delta=D^{*}D$ which is used in \emph{geometric quantization} we
need an additional structure on the manifold $M$, namely a metric
$g$ compatible with $\omega$. References are \cite[p.400]{tuynman_book_supermanifolds_04},\cite[p.168]{berline_book_92},\cite[p.504]{taylor_tome2}.

\subsubsection{\label{sub:Compatible-metrics-and}Compatible metrics and Laplacian}

We recall \cite[p.72]{da_silva_01} that on a symplectic manifold
$\left(M,\omega\right)$, there exists a Riemannian metric $g$ \emph{compatible}
with $\omega$ in the sense that there exists an \emph{almost complex
structure}%
\footnote{\emph{An almost complex structure $J$ on $M$ is a section of the
bundle $\mathrm{Hom}\left(TM,TM\right)$ such that $J\circ J=-\mathrm{Id}$.}%
} $J$ on $M$ such that
\begin{equation}
\omega\left(Ju,Jv\right)=\omega\left(u,v\right)\mbox{ and }g\left(u,v\right)=\omega\left(u,Jv\right)\mbox{ for all }x\in M\mbox{ and }u,v\in T_{x}M.\label{eq:def_g_compatible}
\end{equation}
In general $J$ is not integrable, i.e. it is not a complex structure.
In the rest of this Section we suppose given such a metric $g$ and
an almost complex structure $J$ on $M$.

The metric $g$ on $M$ induces an equivariant metric $g_{P}$ on
$P$ by declaring that \cite[ex.1,ex.2 p.508]{taylor_tome2}:
\begin{enumerate}
\item for every point $p\in P$, $V_{p}P\perp H_{p}P$ are orthogonal,
\item on the horizontal space $H_{p}P$, $g_{P}$ is the pull back of $g$
by $\pi$: $\left(g_{P}\right)_{/H_{p}P}=\pi^{*}\left(g\right)$ 
\item on the vertical space $V_{p}P$, $g_{P}$ is the canonical (Killing)
metric on $\mathfrak{u}\left(1\right)$ i.e. $\left\Vert \frac{\partial}{\partial\theta}\right\Vert _{g_{P}}=1$. 
\end{enumerate}
This metric $g_{P}$ induces a $L^{2}$ scalar product $\langle\alpha|\beta\rangle_{\Lambda^{1}\left(p\right)}$
in the space of one forms $\Lambda^{1}\left(p\right)$ at point $p\in P$.
Using the volume form $\mu_{P}$ on $P$ in (\ref{eq:volume_form_on_P}),
we define a $L^{2}$ scalar product in the space of differential one
forms $C^{\infty}\left(P,\Lambda^{1}\right)$ by 
\[
\langle\alpha|\beta\rangle_{L^{2}\left(P,\Lambda^{1}\right)}:=\int\langle\alpha\left(p\right)|\beta\left(p\right)\rangle_{\Lambda^{1}\left(p\right)}d\mu_{P}(p)\text{\quad\mbox{for }}\alpha,\beta\in C^{\infty}\left(P,\Lambda^{1}\right)
\]
The $L^{2}$ product of functions is of course defined by

\[
\langle u|v\rangle_{L^{2}\left(P\right)}:=\int\overline{u\left(p\right)}\cdot v\left(p\right)d\mu_{P}(p)\text{\quad\mbox{for }}u,v\in C^{\infty}\left(P\right).
\]
\[
\]
Then the operators $D^{*}$ and $\Delta$ are defined as follows.

\begin{framed}%
\begin{defn}
The \textbf{adjoint covariant derivative} $D^{*}:C^{\infty}\left(P,\Lambda^{1}\right)\rightarrow C^{\infty}\left(P\right)$
is defined by the relation
\[
\langle u|D^{*}\alpha\rangle_{L^{2}\left(P,\Lambda^{1}\right)}=\langle Du|\alpha\rangle_{L^{2}\left(P\right)}\quad\mbox{for all }u\in C^{\infty}\left(P\right)\quad\mbox{and }\alpha\in C^{\infty}\left(P,\Lambda^{1}\right).
\]
The \emph{rough Laplacian} $\Delta:C^{\infty}\left(P\right)\rightarrow C^{\infty}\left(P\right)$
is defined as the composition 
\[
\Delta=D^{*}D.
\]
\end{defn}
\end{framed}

The operators introduced above are equivariant, i.e. $D^{*}R_{\theta}=R_{\theta}D^{*}$
and $\Delta R_{\theta}=R_{\theta}\Delta$ with $R_{\theta}:p\rightarrow e^{i\theta}p$,
because so is the metric $g_{P}$. Hence $D^{*}$ and $\Delta$ restrict
naturally to

\[
D^{*}:C_{N}^{\infty}\left(P,\Lambda^{1}\right)\rightarrow C_{N}^{\infty}\left(P\right)\quad\mbox{{and}}\quad\Delta_{N}:C_{N}^{\infty}\left(P\right)\rightarrow C_{N}^{\infty}\left(P\right)
\]
for each $N\in\mathbb{Z}$. We have denoted $\Delta_{N}$ for the
restriction of $\Delta$ to $C_{N}^{\infty}\left(P\right)$. 

It is known that, for every $N\in\mathbb{Z}$, the operator $\Delta_{N}$
is an essentially self-adjoint positive operator with compact resolvent.
Hence its spectrum is discrete and consists of real positive eigenvalues.
The next theorem shows that these eigenvalues form some ``clusters''
(also called ``bands'') in the lower part of this spectrum. Precisely
the eigenvalues of $\frac{1}{2\pi N}\Delta_{N}$ concentrate around
the specific integer values $d+2k$ with $d=\frac{1}{2}\dim M$, $k\in\mathbb{N}$.
(These half-integer values correspond essentially to the eigenvalues
of a harmonic oscillator model as we will see later.) See Figure \ref{fig:Spectrum-of-the_Laplacian}.
These clusters of eigenvalues are called \emph{Landau levels} or \emph{Landau
bands} in physics. The existence of the first band is given in various
papers, see \cite[Cor 1.2]{ma_02_spin_C} and reference therein.%
\footnote{ The result for all the bands seems to be known to specialists although
it does not appear explicitly in the literature to the best of our
knowledge. %
}

\begin{framed}%
\begin{thm}
\textbf{\label{thm:band_structure-of_Laplacian}``The bottom spectrum
of $\Delta_{N}$ has band structure''}. For any $\alpha>0$, the
spectral set of the rough Laplacian $\frac{1}{2\pi N}\Delta_{N}$
in the interval $[0,\alpha]$ is contained in the $N^{-\epsilon}$-neighborhood
of the subset $\left\{ d+2k,\; k\in\mathbb{N}\right\} $ for sufficiently
large $N$ and any $0<\varepsilon<1/2$. The number of eigenvalues
in the $N^{-\epsilon}$-neighborhood of $d+2k$ (or in the $k$-th
band) is proportional to $N^{d}$, that is, if we write $\mathfrak{P}_{k}$
for the spectral projector for the eigenvalues on the $k$-th band,
we have 
\[
C^{-1}N^{d}<\mathrm{rank\,}\mathfrak{P}_{k}<CN^{d}
\]
for some constant $C$ independent of $N$. In particular, for the
spectral projector $\mathfrak{P}_{0}$ for the first band, we have
\begin{equation}
\mathrm{rank\,}\left(\mathfrak{P}_{0}\right)=\int_{M}\left[e^{N\omega}\mbox{Todd}\left(TM\right)\right]_{2d}.\label{eq:Atiyah-Singer-1}
\end{equation}
Further, for the relation to the prequantum transfer operator $\hat{F}_{N}$,
we have 
\[
\mathrm{rank\,}\mathfrak{P}_{k}=\mathrm{rank\,}\tau^{(k)}=\dim\,\mathcal{H}_{k}\quad\mbox{for }0\le k\le n
\]
for sufficiently large $N$, where $n$, $\tau^{(k)}$ and $\mathcal{H}_{k}$
are those in Theorem \ref{thm:More-detailled-description_of_spectrum}.
In particular
\begin{equation}
\mathfrak{P}_{0}:\mathrm{Im}\left(\tau^{(0)}\right)\rightarrow\mathrm{Im}\left(\mathfrak{P}_{0}\right)\label{eq:def_isomorphism-1}
\end{equation}
is a finite rank isomorphism.\end{thm}
\end{framed}

\begin{figure}[h]
\begin{centering}
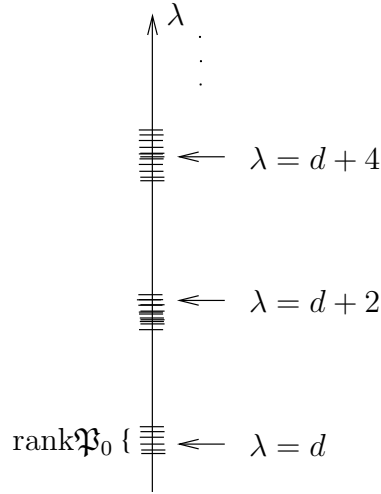\caption{\label{fig:Spectrum-of-the_Laplacian}The \emph{Landau levels }of
the spectrum of the rough Laplacian $\frac{1}{2\pi N}\Delta:L_{N}^{2}\left(P\right)\rightarrow L_{N}^{2}\left(P\right)$
for $N\gg1$.}

\par\end{centering}

\end{figure}

In the separate Section \ref{sec:Proofs-for-Laplacian} we give the
proof for Theorem \ref{thm:band_structure-of_Laplacian} concerning
the spectrum of the rough Laplacian $\Delta=D^{*}D$ . This proof
follows the same strategy as for the proof of Theorem \ref{thm:More-detailled-description_of_spectrum}
namely the Laplacian is decomposed into local charts and approximated
by an Euclidean Laplacian. We use the Harmonic oscillator described
in Section \ref{sub:The-harmonic-oscillator3.5}. In Section \ref{sub:The-rough-Laplacian_4.6},
we obtain the spectrum of the rough Euclidean Laplacian. In Section
\ref{sub:The-multiplication-operators_Laplacian} we establish some
lemmas in order to show that non-linearity can be neglected. The proof
of the index formula (\ref{eq:Atiyah-Singer-1}) is given in \cite[Th. 2]{guillemin_1988_laplace},
see also \cite[Cor. 1.2]{ma_02_spin_C} and references therein.

\subsubsection{Expression of $D^{*}$ and $\Delta$ in local charts}

Let us see the expression of differential operators $D^{*}$ and $\Delta$
introduced above in the local trivialization. Consider a local trivialization
(\ref{eq:local_trivialization-1}) of the bundle $P$ over an open
set $U_{\alpha}\subset M$. As explained in Proposition \ref{prop:expression_of_D},
the operator $D:C_{N}^{\infty}\left(P\right)\rightarrow C_{N}^{\infty}\left(P,\Lambda^{1}\right)$
in such local trivialization is represented by the differential operator
\[
D_{\alpha}:C^{\infty}\left(U_{\alpha}\right)\rightarrow C^{\infty}\left(U_{\alpha},\Lambda^{1}\right),\quad D_{\alpha}u_{\alpha}=\left(Du\right)\circ\tau_{\alpha}
\]
where $u\in C_{N}^{\infty}\left(P\right)$ and $u_{\alpha}:=\left(u\circ\tau_{\alpha}\right)\in C^{\infty}\left(U_{\alpha}\right)$.
Similarly $D^{*}:C_{N}^{\infty}\left(P,\Lambda^{1}\right)\rightarrow C_{N}^{\infty}\left(P\right)$
and $\Delta_{N}=D^{*}D$ are represented by operators 
\[
D_{\alpha}^{*}:C^{\infty}\left(U_{\alpha},\Lambda^{1}\right)\rightarrow C^{\infty}\left(U_{\alpha}\right)\quad\mbox{and }\quad\Delta_{\alpha}:C^{\infty}\left(U_{\alpha}\right)\rightarrow C^{\infty}\left(U_{\alpha}\right).
\]
The next proposition gives the explicit expression of the differential
operators $D_{\alpha}$, $D_{\alpha}^{*}$ and $\Delta_{\alpha}$
using local coordinates $\left(x^{1},\ldots,x^{2d}\right)$ on $U_{\alpha}$
(we have already obtain such an expression for $D_{\alpha}$ in Proposition
\ref{prop:expression_of_D}). Note that the operators $D_{\alpha}^{*}$
and $\Delta_{\alpha}$ depend on the Riemann metric $g$ on $M$ and
also on $N$ (or $\hbar$) though it is not explicit in the notation.
We write $g=\sum_{j,k}g_{j,k}dx^{j}\otimes dx^{k}$ for the metric
tensor and $g^{j,k}=\left(g_{j,k}\right)_{j,k}^{-1}$ for the entries
of the inverse matrix. 

\begin{framed}%
\begin{prop}
\label{prop:expression_of_D*_Delta}With respect to the local trivialization
and coordinate system described above, we have the following expressions
for $D,D^{*}$ and $\Delta=D^{*}D$:
\begin{equation}
D_{\alpha}u_{\alpha}=\frac{i}{\hbar}\sum_{j=1}^{2d}\left(\widehat{\zeta^{j}}u_{\alpha}\right)dx^{j},\qquad\mbox{with }\widehat{\zeta^{j}}=-i\hbar\frac{\partial}{\partial x^{j}}+\eta_{j},\label{eq:expression_D_a}
\end{equation}
\[
D_{\alpha}^{*}\left(\sum_{j=1}^{2d}v_{j}dx^{j}\right)=-\frac{i}{\hbar}\sum_{j,k=1}^{2d}\left(g^{j,k}\widehat{\zeta^{j}}-i\hbar\left(\partial_{j}g^{j,k}\right)\right)v_{k},
\]

\begin{equation}
\Delta_{\alpha}u_{\alpha}=\frac{1}{\hbar^{2}}\sum_{j,k=1}^{2d}\left(g^{jk}\widehat{\zeta^{j}}\widehat{\zeta^{k}}-i\hbar\left(\partial_{j}g^{jk}\right)\widehat{\zeta^{k}}\right)u_{\alpha}.\label{eq:Delta_local_coordinates}
\end{equation}
\end{prop}
\end{framed}
\begin{proof}
The expression (\ref{eq:expression_of_D-1}) gives (\ref{eq:expression_D_a}).
Let $v=\sum_{j=1}^{2d}v_{j}dx^{j}\in C^{\infty}\left(U_{\alpha},\Lambda^{1}\right)$
and $u\in C^{\infty}\left(U_{\alpha}\right)$. Using integration by
parts, we have 
\begin{eqnarray*}
\left\langle u,D_{\alpha}^{*}v\right\rangle _{L^{2}\left(U_{\alpha}\right)} & = & \left\langle D_{\alpha}u,v\right\rangle _{L^{2}\left(U_{\alpha},\Lambda^{1}\right)}=\int\langle du+\frac{i}{\hbar}\eta u|v\rangle dx\\
 & = & \sum_{j,k}\int\overline{\left(\partial_{j}u+\frac{i}{\hbar}\eta_{j}u\right)}g^{jk}v_{k}dx=-\sum_{j,k}\int\left(\overline{u}\partial_{j}\left(g^{jk}v_{k}\right)+\frac{i}{\hbar}\eta_{j}\overline{u}g^{jk}v_{k}\right)dx\\
 & = & -\int\overline{u}\sum_{j,k}\left(\left(\partial_{j}g^{jk}\right)v_{k}+g^{jk}\left(\partial_{j}v_{k}+\frac{i}{\hbar}\eta_{j}v_{k}\right)\right)dx.
\end{eqnarray*}
Hence
\[
D_{\alpha}^{*}v=-\frac{i}{\hbar}\sum_{j,k}\left(g^{jk}\widehat{\zeta^{j}}-i\hbar\left(\partial_{j}g^{jk}\right)\right)v_{k}.
\]
We deduce (\ref{eq:Delta_local_coordinates}) from the computation
\begin{eqnarray*}
\Delta_{\alpha}u_{\alpha} & = & D_{\alpha}^{*}D_{\alpha}u_{\alpha}=-\frac{i}{\hbar}\sum_{j,k}\left(g^{jk}\widehat{\zeta^{j}}-i\hbar\left(\partial_{j}g^{jk}\right)\right)\frac{i}{\hbar}\left(\widehat{\zeta^{k}}u_{\alpha}\right)\\
 & = & \frac{1}{\hbar^{2}}\sum_{j,k}\left(g^{jk}\widehat{\zeta^{j}}\widehat{\zeta^{k}}-i\hbar\left(\partial_{j}g^{jk}\right)\widehat{\zeta^{k}}\right)u_{\alpha}
\end{eqnarray*}

\end{proof}
\begin{framed}%
\begin{cor}
In local Darboux coordinates $x=\left(q,p\right)=\left(q^{1},\ldots q^{d},p^{1},\ldots p^{d}\right)$
on $U_{\alpha}$ and in the special case of the Euclidean metric $g=\sum_{j=1}^{d}dq^{j}\otimes dq^{j}+dp^{j}\otimes dp^{j}$,
we have
\begin{equation}
\Delta_{\alpha}=\frac{1}{\hbar^{2}}\left(\sum_{j=1}^{d}\widehat{\zeta_{p}^{j}}^{2}+\widehat{\zeta_{q}^{j}}^{2}\right)\label{eq:Lap_a}
\end{equation}
with $\widehat{\zeta_{q}^{j}},\widehat{\zeta_{p}^{j}}$ given in (\ref{eq:def_zeta_hat-1}).\end{cor}
\end{framed}

The operator (\ref{eq:Lap_a}) is called the \emph{Euclidean rough
Laplacian}. In Section \ref{sub:The-rough-Laplacian_4.6}, we will
deduce the cluster structure (or the Landau levels) of the spectrum
of the Euclidean rough Laplacian by identifying it with the \emph{harmonic
oscillator.}

\section{\label{sec:Resonances-of-linear_exp_3}Resonances of linear expanding
maps}

This Section is self-contained. The main result of this Section is
Proposition \ref{pp:structure_of_hLA}.

\subsection{The Bargmann transform\label{sub:The-Bargmann-transform}}

\subsubsection{Definitions}

In this section, we recall some basic facts related to the so-called
Bargmann transform. For more detailed account about the Bargmann transform,
we refer to the books \cite[Ch.3]{martinez-01}, \cite[p.39]{folland-88},
\cite{hall99}, \cite[p.19]{perelomov1}.

Let $D$ be a positive integer. Let $\hbar>0$. We consider the Euclidean
space $\mathbb{R}^{D}$ with its canonical Euclidean norm written
$\left|.\right|$. For each point $(x,\xi)\in T^{*}\real^{D}=\real^{D}\oplus\real^{D}$,
we assign the complex-valued smooth function $\phi_{x,\xi}\in\mathcal{S}\left(\real^{D}\right)$
defined by 
\begin{equation}
\phi_{x,\xi}(y)=a_{D}\exp\left(\frac{i}{\hbar}\xi\cdot(y-\frac{x}{2})-\frac{1}{2\hbar}|y-x|^{2}\right),\quad y\in\mathbb{R}^{D}\label{eq:def_of_phi}
\end{equation}
with 
\begin{equation}
a_{D}=(\pi\hbar)^{-D/4}.\label{eq:def_of_ah}
\end{equation}

We will henceforth consider the measure $dx=dx^{1}\ldots dx^{D}$
on $\mathbb{R}^{D}$ defining the Hilbert spaces $L^{2}\left(\mathbb{R}^{D}\right)$
and the measure $\left(2\pi\hbar\right)^{-D}dxd\xi$ on $T^{*}\real^{D}=\mathbb{R}^{2D}$
defining $L^{2}\left(\mathbb{R}^{2D}\right)$. Notice that the constant
$a_{D}$ is taken so that $\left\Vert \phi_{x,\xi}\right\Vert _{L^{2}\left(\mathbb{R}^{D}\right)}=1$.
\begin{defn}
\label{def:Bargmann-transform}The \emph{Bargmann transform} is the
continuous operator 
\begin{equation}
\Bargmann_{\hbar}:\mathcal{S}(\real^{D})\to\mathcal{S}(\real^{2D}),\qquad\left(\Bargmann_{\hbar}u\right)(x,\xi)=\int\overline{\phi_{x,\xi}(y)}\cdot u(y)dy\label{eq:def_Bargman_transform}
\end{equation}
on the Schwartz space $\mathcal{S}(\real^{D}).$ Then the (formal)
adjoint of $\Bargmann_{\hbar}$ defined by $\left\langle u,\Bargmann_{\hbar}^{*}v\right\rangle =\left\langle \Bargmann_{\hbar}u,v\right\rangle $
is given by 
\begin{equation}
\Bargmann_{\hbar}^{*}:\mathcal{S}(\real^{2D})\to\mathcal{S}(\real^{D}),\qquad\left(\Bargmann_{\hbar}^{*}v\right)(y)=\int\phi_{x,\xi}(y)\cdot v(x,\xi)\frac{dxd\xi}{\left(2\pi\hbar\right)^{D}}.\label{eq:B*}
\end{equation}
\end{defn}
\begin{lem}
\label{lm:Bargmann_is_isometry}\cite[p70, Proposition 3.1.1]{martinez-01}
We have that
\begin{enumerate}
\item $\Bargmann_{\hbar}$ extends uniquely to an isometric embedding $\Bargmann_{\hbar}:L^{2}(\real^{D})\to L^{2}(\real^{2D})$.
\item $\Bargmann_{\hbar}^{*}$ extends uniquely to a bounded operator $\Bargmann_{\hbar}^{*}:L^{2}(\real^{2D})\to L^{2}(\real^{D})$. 
\item $\Bargmann_{\hbar}^{*}\circ\Bargmann_{\hbar}=\mathrm{{Id}}$ on $L^{2}(\real^{D})$.
\end{enumerate}
\end{lem}
\begin{proof}
For any $u\in\mathcal{S}(\real^{D})$, we have 
\begin{align*}
\|\mathcal{B}_{\hbar}u\|_{L^{2}\left(\mathbb{R}^{2D}\right)}^{2} & =\frac{a_{D}^{2}}{\left(2\pi\hbar\right)^{D}}\int\phi_{x,\xi}(y')\overline{u(y')}\cdot\overline{\phi_{x,\xi}(y)}u(y)dxd\xi dydy'\\
 & =\frac{a_{D}^{2}}{\left(2\pi\hbar\right)^{D}}\int u(y)\overline{u(y')}\exp\left(\frac{i}{\hbar}\xi(y'-y)-\frac{1}{2\hbar}\left(|x-y|^{2}+|x-y'|^{2}\right)\right)dxd\xi dydy'\\
 & =(\pi\hbar)^{-D/2}\int|u(y)|^{2}\exp(-|x-y|^{2}/\hbar)dxdy\\
 & =\int|u(y)|^{2}dy=\|u\|_{L^{2}\left(\mathbb{R}^{D}\right)}^{2}.
\end{align*}
This gives the claims of the lemma by the usual continuity argument.
\end{proof}

\subsubsection{Bargmann projector}
\begin{prop}
\label{prop:The-Bargmann-projector}The space $L^{2}(\real^{2D})$
is orthogonally decomposed as 
\begin{equation}
L^{2}(\real^{2D})=\mbox{Im}\Bargmann_{\hbar}\oplus\ker\Bargmann_{\hbar}^{*}.\label{eq:decomp_L2R2D}
\end{equation}
The \emph{Bargmann projector} $\BargmannP_{\hbar}$ is the orthogonal
projection onto $\mbox{\ensuremath{\mathrm{Im}}\,}\Bargmann_{\hbar}\subset L^{2}(\real^{2D})$
along $\ker\Bargmann_{\hbar}^{*}$. It is given by 
\begin{equation}
\BargmannP_{\hbar}:=\Bargmann_{\hbar}\circ\Bargmann_{\hbar}^{*}\quad:L^{2}(\real^{2D})\to L^{2}(\real^{2D}).\label{eq:def_Bargmann_projector}
\end{equation}
It can be expressed as an integral operator $\left(\BargmannP_{\hbar}v\right)(z)=\int K_{\mathcal{P},\hbar}\left(z,z'\right)v(z')dz'$
with the Schwartz kernel:
\begin{equation}
K_{\mathcal{P},\hbar}\left(z,z'\right)=\exp\left(\frac{i}{2\hbar}\omega(z,z')-\frac{1}{4\hbar}|z-z'|^{2}\right)\label{eq:Bargman_Kernel}
\end{equation}
with $z=\left(x,\xi\right),\; z'=\left(x',\xi'\right)\in\mathbb{R}^{2D}$,
the measure $dz'=dx'd\xi'/\left(2\pi\hbar\right)^{D}$, the Euclidean
norm $|z|^{2}:=|x|^{2}+|\xi|^{2}$ and the canonical symplectic form
on $T^{*}\mathbb{R}^{D}$, $\omega\left(z,z'\right)=x\cdot\xi'-\xi\cdot x'$\end{prop}
\begin{proof}
$\BargmannP_{\hbar}$ is an orthogonal projection because $\mathcal{P}_{h}^{*}=\left(\Bargmann_{\hbar}\circ\Bargmann_{\hbar}^{*}\right)^{*}=\mathcal{P}_{\hbar}$
and $\mathcal{P}_{\hbar}^{2}=\Bargmann_{\hbar}\circ\Bargmann_{\hbar}^{*}\circ\Bargmann_{\hbar}\circ\Bargmann_{\hbar}^{*}=\mathcal{P}_{\hbar}$
from Lemma \ref{lm:Bargmann_is_isometry}. From Definition \ref{def:Bargmann-transform},
the Schwartz kernel of $\mathcal{P}_{\hbar}=\Bargmann_{\hbar}\circ\Bargmann_{\hbar}^{*}$
is
\begin{eqnarray*}
K_{\mathcal{P},\hbar}\left(z,z'\right) & = & \int dy\overline{\phi_{z}(y)}\phi_{z'}(y)\\
 & = & a_{D}^{2}\int_{\mathbb{R}^{D}}dy\exp\left(-i\frac{\xi}{\hbar}(y-\frac{x}{2})+i\frac{\xi'}{\hbar}(y-\frac{x'}{2})-\frac{1}{2\hbar}\left(|y-x|^{2}+|y-x'|^{2}\right)\right)\\
 & = & a_{D}^{2}\exp\left(i\frac{1}{2\hbar}(\xi x-\xi'x')-\frac{1}{2\hbar}\left(|x|^{2}+|x'|^{2}\right)\right)\\
 &  & \times\int dy\exp\left(\frac{1}{\hbar}\left(i\left(\xi'-\xi\right)+\left(x+x'\right)\right)y-\frac{1}{\hbar}|y|^{2}\right).
\end{eqnarray*}
We use the following formula for Gaussian integral in $\mathbb{R}^{D}$:
\begin{equation}
\int_{\mathbb{R}^{D}}e^{-\frac{1}{2}y\cdot\left(Ay\right)+b.y}dy=\sqrt{\frac{\left(2\pi\right)^{D}}{\mbox{det}A}}\exp\left(\frac{1}{2}b\cdot\left(A^{-1}b\right)\right),\qquad b\in\mathbb{C}^{D},A\in\mathcal{L}\left(\mathbb{R}^{D}\right)\label{eq:Gaussian_integrale}
\end{equation}
with $A=(2/\hbar)\cdot\mathrm{Id}$, $b=\frac{i}{\hbar}\left(\xi'-\xi\right)+\frac{1}{\hbar}\left(x+x'\right)$
and get
\begin{eqnarray*}
K_{\mathcal{P},\hbar}\left(z,z'\right) & = & \frac{1}{\left(\pi\hbar\right)^{D/2}}\frac{\left(2\pi\right)^{D/2}}{\left(2/\hbar\right)^{D/2}}\exp\left(\frac{1}{4\hbar}\left(i\left(\xi'-\xi\right)+\left(x+x'\right)\right)^{2}\right)\\
 &  & \quad\times\exp\left(i\frac{1}{2\hbar}(\xi x-\xi'x')-\frac{1}{2\hbar}\left(|x|^{2}+|x'|^{2}\right)\right)\\
 & = & \exp\left(\frac{1}{4\hbar}\left(-|\xi'-\xi|^{2}-|x-x'|^{2}\right)+\frac{i}{2\hbar}(\xi x-\xi'x'+\left(\xi'-\xi\right)\left(x+x'\right))\right)\\
 & = & \exp\left(\frac{i}{2\hbar}\omega(z,z')-\frac{1}{4\hbar}|z-z'|^{2}\right).
\end{eqnarray*}

\end{proof}

\subsubsection{The Bargmann transform in more general setting\label{sub:Bargmann_transform_generalized}}

We have seen that the Bargmann transform gives a phase-space representation,
i.e. a unitary isomorphism: $\Bargmann_{\hbar}:L^{2}(\real^{D})\to\mbox{Im}\left(\Bargmann_{\hbar}\right)\subset L^{2}(T^{*}\real^{D})$.
The next proposition gives the Bargmann transform in a slightly more
general setting. We start from a symplectic linear space $(E,\omega)$
of dimension $2D$ and a Lagrangian subspace $L\subset E$. Let $g(\cdot,\cdot)$
be a scalar product on $E$ that is compatible with the symplectic
form $\omega$ on $E$ in the sense that there is a linear map $J:E\to E$
such that $J\circ J=-\mathrm{Id}$ and holds 
\[
g\left(z,z'\right)=\omega\left(z,Jz'\right)\quad\mbox{for all }z,z'\in E.
\]
This is nothing but the point-wise version of the condition (\ref{eq:def_g_compatible}).
Let $L^{\perp_{g}}$ be the orthogonal complement of $L$ with respect
to the inner product $g$, so that a point $z\in E$ can be decomposed
uniquely as $z=x+\xi$, $x\in L$, $\xi\in L^{\perp_{g}}$. 

For each point $z=x+\xi$, $x\in L$, $\xi\in L^{\perp_{g}}$, we
define the wave packet $\phi_{z}\left(y\right)\in\mathcal{S}\left(L\right)$
by
\[
\phi_{z}(y)=a_{D}\exp\left(\frac{i}{\hbar}\omega(y-\frac{x}{2},\xi)-\frac{1}{2\hbar}|y-x|_{g}^{2}\right)\quad\mbox{for }y\in L.
\]
We define the Bargmann transform $\mathcal{B}_{\hbar}:\mathcal{S}\left(L\right)\rightarrow\mathcal{S}\left(E\right)$
and its adjoint $\mathcal{B}_{\hbar}^{*}:\mathcal{S}\left(E\right)\rightarrow\mathcal{S}\left(L\right)$
as in Definition \ref{def:Bargmann-transform}. Then the statement
corresponding to Lemma \ref{lm:Bargmann_is_isometry} and Proposition
\ref{prop:The-Bargmann-projector} holds. Namely 
\begin{prop}
For the Bargmann transform $\Bargmann_{\hbar}$ and its adjoint defined
as above, 
\begin{enumerate}
\item $\Bargmann_{\hbar}$ extends uniquely to an isometric embedding $\Bargmann_{\hbar}:L^{2}(L)\to L^{2}(E)$.
\item $\Bargmann_{\hbar}^{*}$ extends uniquely to a bounded operator $\Bargmann_{\hbar}^{*}:L^{2}(E)\to L^{2}(L)$. 
\item $\Bargmann_{\hbar}^{*}\circ\Bargmann_{\hbar}=\mathrm{{Id}}$ on $L^{2}(L)$.
\item The Bargmann projector \textup{$\mathcal{P}_{\hbar}:L^{2}\left(E\right)\rightarrow L^{2}\left(E\right)$,
defined by $\mathcal{P}_{\hbar}:=\Bargmann_{\hbar}\circ\Bargmann_{\hbar}^{*}$,
is the orthogonal projection onto $\mathrm{Im\,}\mathcal{B}_{\hbar}\subset L^{2}(E)$.
It is expressed as an integral operator with kernel (\ref{eq:Bargman_Kernel})
with $|\cdot|$ replaced by $|\cdot|_{g}$. }
\end{enumerate}
\end{prop}
\begin{rem}
Notice that the Bargmann projector $\mathcal{P}_{\hbar}:L^{2}\left(E\right)\rightarrow L^{2}\left(E\right)$
can be defined directly from its kernel (\ref{eq:Bargman_Kernel})
and is independent on the choice of the Lagrangian subspace $L$.\end{rem}
\begin{proof}
There are linear isomorphisms $\psi:\real^{D}\to L$ and $\Psi:\real^{2D}\to E$
such that the following diagram commutes:
\[
\begin{CD}\real^{2D}@>\Psi>>E\\
@VpVV@Vp'VV\\
\real^{D}@>\psi>>L
\end{CD}
\]
where $p:\real^{2D}=\real^{D}\oplus\real^{D}\to\real^{D}$ is the
projection to the first $D$ components: $p(x,\xi)=x$, and $p':E\to L$
is the orthogonal projection to $L$ with respect to $g,$ and moreover
that the pull-back of the symplectic form $\omega$ and the inner
product $g$ by $\Psi$ coincides with the standard Euclidean inner
product $g_{0}\left(z,z'\right)=z\cdot z'$ and the standard symplectic
form:
\[
\omega_{0}\left(z,z'\right)=x\cdot\xi'-\xi\cdot x'\quad\mbox{for }z=\left(x,\xi\right),z'=\left(x',\xi'\right)\in T^{*}\mathbb{R}^{D}
\]
Through such correspondence by $\Psi$ and $\psi$, the definition
of the Bargmann transform and its adjoint above coincides with those
that we made in the last subsection. Therefore the claims are just
restatement of Lemma \ref{lm:Bargmann_is_isometry} and Proposition
\ref{prop:The-Bargmann-projector}.

\end{proof}

\subsubsection{Scaling\label{sub:Scaling}}

The operators $\Bargmann_{\hbar}$ and $\Bargmann_{\hbar}^{*}$ are
related to $\Bargmann_{1}$ and $\Bargmann_{1}^{*}$ (i.e. with $\hbar=1$)
by the simple scaling $x\mapsto\hbar^{1/2}x$. Though this fact should
be obvious, we give the relations explicitly for the later use. Let
us introduce the unitary operators 
\begin{equation}
s_{\hbar}:L^{2}(\real^{D})\to L^{2}(\real^{D}),\qquad s_{\hbar}u(x)=\hbar^{-D/4}u(\hbar^{-1/2}x)\label{eq:I_hbar}
\end{equation}
 and%
\footnote{Recall the convention on the norm on $L^{2}(\real^{D})$ and $L^{2}(\real^{2D})$
made in the beginning of this section.%
} 
\begin{equation}
S_{\hbar}:L^{2}\left(\real^{2D},\frac{dxd\xi}{\left(2\pi\hbar\right)^{D}}\right)\to L^{2}\left(\real^{2D},\frac{dxd\xi}{\left(2\pi\hbar\right)^{D}}\right),\qquad S_{\hbar}u(x,\xi)=u(\hbar^{-1/2}x,\hbar^{-1/2}\xi)\label{eq:wI_hbar}
\end{equation}
Then we have
\begin{lem}
The following diagram commutes:

\[
\begin{CD}L^{2}\left(\mathbb{R}^{2D}\right)@>S_{\hbar}>>L^{2}\left(\mathbb{R}^{2D}\right)\\
@A\Bargmann_{\hbar}AA@A\Bargmann_{1}AA\\
L^{2}\left(\mathbb{R}^{D}\right)@>s_{\hbar}>>L^{2}\left(\mathbb{R}^{D}\right)
\end{CD}
\]

\end{lem}

\subsection{Action of linear transforms}
\begin{defn}
The lift of a bounded operator $L:L^{2}(\real^{D})\to L^{2}(\real^{D})$
with respect to the Bargmann transform $\Bargmann_{\hbar}$ is defined
as the operator 
\begin{equation}
\Llift:=\Bargmann_{\hbar}\circ L\circ\Bargmann_{\hbar}^{*}:L^{2}(\real^{2D})\to L^{2}(\real^{2D})\label{eq:def_L_Lift}
\end{equation}

\end{defn}
By definition, it makes the following diagram commutes: 
\begin{equation}
\begin{CD}L^{2}(\real^{2D})@>{\Llift}>>L^{2}(\real^{2D})\\
@A{\Bargmann_{\hbar}}AA@A{\Bargmann_{\hbar}}AA\\
L^{2}(\real^{D})@>{L}>>L^{2}(\real^{D}).
\end{CD}\label{cd:lift}
\end{equation}
 Since $\BargmannP_{\hbar}\circ\Llift\circ\BargmannP_{\hbar}=\Llift$,
the lift $\Llift$ is always trivial on the second factor with respect
to the decomposition $L^{2}(\real^{2D})=\mbox{Im}\Bargmann_{\hbar}\oplus\ker\Bargmann_{\hbar}^{*}=\mathrm{Im}\,\BargmannP_{\hbar}\oplus\ker\,\BargmannP_{\hbar}$
in (\ref{eq:decomp_L2R2D}), that is, 
\begin{equation}
\Llift=\left(\Bargmann_{\hbar}\circ L\circ\Bargmann_{\hbar}^{*}\right)_{\mathrm{Im}\BargmannP_{\hbar}}\oplus\left(\mathbf{0}\right)_{\mathrm{Ker}\BargmannP_{\hbar}}.\label{eq:product_form_of_the_lifted_operator}
\end{equation}

\begin{lem}
\label{lm:lift_of_LA} For an invertible linear transformation $A:\real^{D}\to\real^{D}$,
we associate a bounded transfer operator\textbf{ }defined by 
\begin{equation}
L_{A}:L^{2}(\real^{D})\to L^{2}(\real^{D}),\quad L_{A}u=u\circ A^{-1}\label{eq:def_LA}
\end{equation}
Then we have
\begin{equation}
L_{A}=d(A)\cdot\Bargmann_{\hbar}^{*}\circ L_{A\oplus{}^{t}A^{-1}}\circ\Bargmann_{\hbar},\label{eq:claim}
\end{equation}
where $L_{A\oplus{}^{t}A^{-1}}:L^{2}(\real^{2D})\to L^{2}(\real^{2D})$
is the unitary transfer operator given by 
\begin{equation}
\left(L_{A\oplus{}^{t}A^{-1}}u\right)(x,\xi):=u\left(A^{-1}x,{}^{t}A\xi\right)\label{eq:ld}
\end{equation}
 and we set 
\[
d(A):=\det\left(\frac{1}{2}\left(1+{}^{t}AA\right)\right)^{1/2}.
\]
Consequently the lift of $L_{A}$, 
\begin{equation}
\Llift_{A}:=\Bargmann_{\hbar}\circ L_{A}\circ\Bargmann_{\hbar}^{*}\label{def_Lift_A}
\end{equation}
is expressed as 
\begin{equation}
\Llift_{A}=d(A)\cdot\BargmannP_{\hbar}\circ L_{A\oplus{}^{t}A^{-1}}\circ\BargmannP_{\hbar}.\label{eq:Lift_A_expression}
\end{equation}
 \end{lem}
\begin{rem}
\label{rem3.10}
\begin{enumerate}
\item The expression (\ref{eq:claim}) shows that $L_{A}$ can be expressed
as an operator on the phase space defined in terms of the Bargmann
projector and the transfer operator $L_{A\oplus{}^{t}A^{-1}}$, but
with an additional factor $d\left(A\right)$, sometimes called the\emph{
metaplectic correction}.\textbf{ }This may be regarded as a realization
of the idea explained in the last section: $L_{A}$ can be seen as
a Fourier integral operator and canonical map is $A\oplus{}^{t}A^{-1}$
on $T^{*}\mathbb{R}^{D}$. But notice that the correction term $d(A)$
will be crucially important for our argument. 
\item For an orthogonal transform $A\in SO\left(D\right)$, we have $d\left(A\right)=1$.
\item The operator $\frac{1}{\sqrt{\left|\mathrm{det}A\right|}}L_{A}$ is
unitary in $L^{2}\left(\mathbb{R}^{D}\right)$ but because in the
main result of this Section, Proposition \ref{pp:structure_of_hLA},
the Hilbert space is not $L^{2}\left(\mathbb{R}^{D}\right)$, we don't
care about this property. It will however have some importance later
in Proposition \ref{prop:prequatum_op_for_hyp_linear} where the factor
$\frac{1}{\sqrt{\left|\mathrm{det}A\right|}}$ will therefore appear.
\end{enumerate}
\end{rem}
\begin{proof}
To prove (\ref{eq:claim}), we write the operator $\Bargmann_{\hbar}^{*}\circ L_{A\oplus{}^{t}A^{-1}}\circ\Bargmann_{\hbar}$
as an integral operator 
\[
\left(\Bargmann_{\hbar}^{*}\circ L_{A\oplus{}^{t}A^{-1}}\circ\Bargmann_{\hbar}u\right)(y)=\int K(y,y')u(y')dy'
\]
 with the kernel (from (\ref{eq:def_Bargman_transform}) and (\ref{eq:B*}))
\[
K(y,y')=\int\phi_{x,\xi}(y)\cdot\overline{\phi_{A^{-1}x,{}^{t}A\xi}(y')}\;\frac{dxd\xi}{\left(2\pi\hbar\right)^{D}}.
\]
Using the formula (\ref{eq:Gaussian_integrale}) for the Gaussian
integral and change of variables, we can calculate the integral on
the right hand side as 
\begin{align*}
K(y,y') & =\int\phi_{Ax',{}^{t}A^{-1}\xi'}(y)\cdot\overline{\phi_{x',\xi'}(y')}\;\frac{dx'd\xi'}{\left(2\pi\hbar\right)^{D}},\qquad(x'=A^{-1}x,\;\xi'={}^{t}A\xi)\\
 & =a_{D}^{2}\cdot\int e^{\frac{i}{\hbar}\langle\xi',y'-A^{-1}y\rangle-|y'-x'|^{2}/(2\hbar)-|y-Ax'|^{2}/(2\hbar)}\frac{dx'd\xi'}{\left(2\pi\hbar\right)^{D}}\\
 & =(\pi\hbar)^{-D/2}\cdot\delta(y'-A^{-1}y)\cdot\int e^{-|A^{-1}y-x'|^{2}/(2\hbar)-|y-Ax'|^{2}/(2\hbar)}dx'\\
 & =\pi^{-D/2}\cdot\delta(y'-A^{-1}y)\cdot\int e^{-|t|^{2}/2-|At|^{2}/2}dt\qquad(t=(x'-A^{-1}y)/\hbar)\\
 & =\det((I+{}^{t}AA)/2)^{-1/2}\cdot\delta(y'-A^{-1}y).
\end{align*}
 Therefore we have 
\[
(\Bargmann_{\hbar}^{*}\circ L_{A\oplus{}^{t}A^{-1}}\circ\Bargmann_{\hbar})u(y)=\det((I+{}^{t}AA)/2)^{-1/2}\cdot u(A^{-1}y)=d(A)^{-1}\cdot\left(L_{A}u\right)(y)
\]
 and hence the claim (\ref{eq:claim}) follows. This implies 
\begin{align*}
\Llift_{A} & =\Bargmann_{\hbar}\circ L_{A}\circ\Bargmann_{\hbar}^{*}=d(A)\cdot\BargmannP_{\hbar}\circ L_{A\oplus{}^{t}A^{-1}}\circ\BargmannP_{\hbar}.
\end{align*}
The other claims follow immediately.
\end{proof}
The next Lemma introduces operators which realize translation in phase
space $T^{*}\mathbb{R}^{D}$. In \cite{folland-88,perelomov1} it
is shown that this gives a unitary irreducible representation of the
Weyl-Heisenberg group.
\begin{lem}
\label{lm:lift_of_LA-1} For $\left(x_{0},\xi_{0}\right)\in T^{*}\mathbb{R}^{D}=\real^{2D}$,
we associate a unitary operator\textbf{ }defined by 
\begin{equation}
T_{\left(x_{0},\xi_{0}\right)}:L^{2}(\real^{D})\to L^{2}(\real^{D}),\quad\left(T_{\left(x_{0},\xi_{0}\right)}v\right)\left(y\right)=e^{\frac{i}{\hbar}\xi_{0}.\left(y+\frac{x_{0}}{2}\right)}v\left(y-x_{0}\right)\label{eq:def_LA-1}
\end{equation}
Then we have 
\begin{equation}
T_{\left(x_{0},\xi_{0}\right)}=\Bargmann_{\hbar}^{*}\circ\mathcal{T}_{\left(x_{0},\xi_{0}\right)}\circ\Bargmann_{\hbar},\label{eq:claim-1}
\end{equation}
 where $\mathcal{T}_{\left(x_{0},\xi_{0}\right)}:L^{2}(\real^{2D})\to L^{2}(\real^{2D})$
is the unitary transfer operator given by 
\begin{equation}
\mathcal{T}_{\left(x_{0},\xi_{0}\right)}u(x,\xi):=e^{\frac{i}{2\hbar}\left(\xi_{0}\cdot x-x_{0}\cdot\xi\right)}u\left(x-x_{0},\xi-\xi_{0}\right)\label{eq:ld-1}
\end{equation}
Consequently the lift of $T_{\left(x_{0},\xi_{0}\right)}$ is 
\begin{equation}
T_{\left(x_{0},\xi_{0}\right)}^{\mathrm{lift}}:=\Bargmann_{\hbar}\circ T_{\left(x_{0},\xi_{0}\right)}\circ\Bargmann_{\hbar}^{*}=\BargmannP_{\hbar}\circ\mathcal{T}_{\left(x_{0},\xi_{0}\right)}\circ\BargmannP_{\hbar}.\label{def_Lift_A-1}
\end{equation}
 \end{lem}
\begin{proof}
The Schwartz kernel of $\Bargmann_{\hbar}^{*}\circ\mathcal{T}_{\left(x_{0},\xi_{0}\right)}\circ\Bargmann_{\hbar}$
is 
\begin{align*}
 & K\left(y,y'\right)=\int_{\mathbb{R}^{2D}}e^{\frac{i}{2\hbar}\left(\xi_{0}\cdot x-x_{0}\cdot\xi\right)}\phi_{x,\xi}(y)\cdot\overline{\phi_{x-x_{0},\xi-\xi_{0}}(y')}\;\frac{dxd\xi}{\left(2\pi\hbar\right)^{D}}\\
= & a_{D}^{2}\cdot\int\frac{dxd\xi}{\left(2\pi\hbar\right)^{D}}\: e^{\frac{i}{2\hbar}\left(\xi_{0}\cdot x-x_{0}\cdot\xi\right)}e^{\frac{i}{\hbar}\left(\xi\cdot\left(y-\frac{x}{2}\right)-\left(\xi-\xi_{0}\right)\cdot\left(y'-\frac{x-x_{0}}{2}\right)\right)}e^{-\frac{1}{2\hbar}\left(|y'-\left(x-x_{0}\right)|^{2}+|y-x|^{2}\right)}\\
= & \left(\pi\hbar\right)^{-D/2}\delta\left(y'-y+x_{0}\right)e^{\frac{i}{\hbar}\xi_{0}\cdot\left(y+\frac{x_{0}}{2}\right)}\int e^{-\frac{1}{\hbar}|y-x|^{2}}dx\\
= & \delta\left(y'-y+x_{0}\right)\cdot e^{\frac{i}{\hbar}\xi_{0}\cdot\left(y+\frac{x_{0}}{2}\right)}.
\end{align*}
This is the kernel of the operator $T_{\left(x_{0},\xi_{0}\right)}$.
\end{proof}

\begin{cor}
\label{cor:4.11}The lift of the operator $T_{\left(x_{0},\xi_{0}\right)}$
is expressed as 
\end{cor}
\begin{equation}
T_{(x_{0},\xi_{0})}^{\mathrm{lift}}:=\Bargmann_{\hbar}\circ T_{(x_{0},\xi_{0})}\circ\Bargmann_{\hbar}^{*}=\mathcal{T}_{\left(x_{0},\xi_{0}\right)}\circ\BargmannP_{\hbar}=\BargmannP_{\hbar}\circ\mathcal{T}_{\left(x_{0},\xi_{0}\right)}.\label{eq:observation_L}
\end{equation}

\begin{proof}
We can check the equality $\mathcal{T}_{\left(x_{0},\xi_{0}\right)}\circ\BargmannP_{\hbar}=\BargmannP_{\hbar}\circ\mathcal{T}_{\left(x_{0},\xi_{0}\right)}$
on the left by showing that the Schwartz kernels of $\mathcal{T}_{\left(x_{0},\xi_{0}\right)}\circ\BargmannP_{\hbar}$
and $\BargmannP\circ\mathcal{T}_{\left(x_{0},\xi_{0}\right)}$ are
equal. This an easy computation using the expressions (\ref{eq:Bargman_Kernel})
and (\ref{eq:ld-1}). The rest of the claim follows from Lemma \ref{lm:lift_of_LA-1}.
\end{proof}

\subsection{The weighted $L^{2}$ spaces : $L^{2}(\real^{2D},(W_{\hbar}^{r})^{2})$\label{ss:wL2} }

For each $t>0$, we define the cones in $\mathbb{R}^{2D}$: 
\begin{align*}
{\mathbf{C}}_{+}(t) & =\{(x,\xi)\in\real^{2D}\mid|\xi|\le t\cdot|x|\},\quad{\mathbf{C}}_{-}(t)=\{(x,\xi)\in\real^{2D}\mid|x|\le t\cdot|\xi|\}.
\end{align*}
 Let $r>0$. Take and fix a $C^{\infty}$ function $m:\mathbb{P}\left(\real^{2D}\right)\to[-r,r]$,
called \emph{order function}, on the projective space $\mathbb{P}\left(\real^{2D}\right)$
so that 
\begin{equation}
m\left(\left[\left(x,\xi\right)\right]\right)=\begin{cases}
-r, & \quad\mbox{if \ensuremath{(x,\xi)\in{\mathbf{C}}_{+}(1/2)}};\\
+r, & \quad\mbox{if \ensuremath{(x,\xi)\in{\mathbf{C}}_{-}(1/2)}}.
\end{cases}\label{eq:def_order_function_m}
\end{equation}
We then define the\emph{ escape function} (or \emph{the weight function)
by} 
\[
W^{r}:\real^{2D}\to\real_{+},\qquad W^{r}(x,\xi)=\langle|(x,\xi)|\rangle^{m([(x,\xi)])}
\]
where $\left\langle s\right\rangle :=\left(1+s^{2}\right)^{1/2}$
for $s\in\mathbb{R}$ and $|(x,\xi)|^{2}:=\left|x\right|^{2}+\left|\xi\right|^{2}$.
From this definition we have 
\[
W^{r}(x,\xi)\sim\langle|(x,\xi)|\rangle^{-r}\quad\mbox{if \ensuremath{|x|\ge2|\xi|}and \ensuremath{|(x,\xi)|\gg1}}
\]
 and 
\[
W^{r}(x,\xi)\sim\langle|(x,\xi)|\rangle^{r}\quad\mbox{if \ensuremath{|x|\le|\xi|/2}and \ensuremath{|(x,\xi)|\gg1}.}
\]

For convenience in a later argument, we also take and fix $C^{\infty}$
functions 
\[
m^{+},m^{-}:\mathbb{P}\left(\real^{2D}\right)\to[-r,r]
\]
 such that 
\begin{align*}
m^{+}([(x,\xi)]) & =\begin{cases}
-r, & \quad\mbox{if \ensuremath{(x,\xi)\in{\mathbf{C}}_{+}(1/9);}}\\
+r, & \quad\mbox{if \ensuremath{(x,\xi)\in{\mathbf{C}}_{-}(3)}, }
\end{cases}\intertext{and}m^{-}([(x,\xi)]) & =\begin{cases}
-r, & \quad\mbox{if \ensuremath{(x,\xi)\in{\mathbf{C}}_{+}(3);}}\\
+r, & \quad\mbox{if \ensuremath{(x,\xi)\in{\mathbf{C}}_{-}(1/9)},}
\end{cases}
\end{align*}
 and define the functions $W^{r,\pm}:\real^{2D}\to\real_{+}$ by 
\begin{equation}
W^{r,\pm}(x,\xi):=\langle|(x,\xi)|\rangle^{m^{\pm}([(x,\xi)])}.\label{eq:def_Wr_+-}
\end{equation}
Obviously we have 
\begin{equation}
W^{r,-}(x,\xi)\le W^{r}(x,\xi)\le W^{r,+}(x,\xi).\label{eq:wrpm}
\end{equation}
These functions, $W^{r}$ and $W^{r,\pm}$, satisfy the following
regularity estimate that we will make use of later on: For any $\epsilon>0$
and multi-index $\alpha$, there exists a constant $C_{\alpha,\epsilon}>0$
such that 
\begin{equation}
|\partial_{x,\xi}^{\alpha}W^{r}(x,\xi)|\le C_{\alpha,\epsilon}\langle|(x,\xi)|\rangle^{-(1-\epsilon)|\alpha|}\cdot W^{r}(x,\xi)\quad\mbox{for all \ensuremath{(x,\xi)\in\real^{2D}}}\label{eq:W^r_in_S_1-epsilon}
\end{equation}
and the same inequalities for $W^{r,\pm}(\cdot)$ hold.
\begin{defn}
\label{def:Wrh}For $\hbar>0$, let 
\begin{equation}
W_{\hbar}^{r}:\real^{2D}\to\real_{+},\quad W_{\hbar}^{r}(x,\xi):=W^{r}(\hbar^{-1/2}x,\hbar^{-1/2}\xi)=\left(S_{\hbar}W^{r}\right)(x,\xi)\label{eq:def_Wh}
\end{equation}
 where $S_{\hbar}$ is the operator defined in (\ref{eq:wI_hbar}).
We consider the weighted $L^{2}$ space defined as 
\begin{equation}
L^{2}(\real^{2D},(W_{\hbar}^{r})^{2})=\{v\in L_{\mathrm{loc}}^{2}(\real^{2D})\mid\|W_{\hbar}^{r}\cdot v\|_{L^{2}(\real^{2D})}<\infty\}.\label{eq:L2Wr}
\end{equation}
 Likewise, we define the functions $W_{\hbar}^{r,\pm}$ and the weighted
$L^{2}$ spaces $L^{2}(\real^{2D},(W_{\hbar}^{r,\pm})^{2})$ in the
parallel manner, replacing $W^{r}$ by $W^{r,\pm}$. 
\end{defn}
Note that the function $W^{r}$ (and $W^{r,\pm}$) satisfies the condition
\begin{equation}
W^{r}(x,\xi)\le C\cdot W^{r}(y,\eta)\cdot\langle|(x,\xi)-(y,\eta)|\rangle^{2r}\quad\mbox{for any \ensuremath{x,y\in\real^{d}}}\label{eq:W_order_function}
\end{equation}
 for some constant $C>0$. Consequently the function $W_{\hbar}^{r}$
(and $W_{\hbar}^{r,\pm}$) satisfies 
\begin{equation}
W_{\hbar}^{r}(x,\xi)\le C\cdot W_{\hbar}^{r}(y,\eta)\cdot\langle\hbar^{-1/2}|(x,\xi)-(y,\eta)|\rangle^{2r}\quad\mbox{for any \ensuremath{x,y\in\real^{d}}.}\label{eq:W_order_function_hbar}
\end{equation}

The next Lemma characterizes a class of bounded integral operators
in $L^{2}(\real^{2D},(W_{\hbar}^{r})^{2})$ in terms of its kernel.
\begin{lem}
\label{lm:boundedness_of_molifier} If $R:\mathcal{S}(\real^{2D})\to\mathcal{S}(\real^{2D})$
is an integral operator of the form 
\[
Ru(x,\xi)=\int K_{R}(x,\xi;x',\xi')u(x',\xi')dx'd\xi'
\]
 and if the kernel $K_{R}(\cdot)$ is a continuous function satisfying
\[
|K_{R}(x,\xi;x',\xi')|\le\langle\hbar^{-1/2}|(x,\xi)-(x',\xi')|\rangle^{-\nu}
\]
 for some $\nu>2r+2D$, then it extends to a bounded operator on $L^{2}(\real^{2D},(W_{\hbar}^{r})^{2})$
and 
\[
\|R:L^{2}(\real^{2D},(W_{\hbar}^{r})^{2})\to L^{2}(\real^{2D},(W_{\hbar}^{r})^{2})\|\le C_{\nu}
\]
 where $C_{\nu}$ is a constant which depends only on $\nu$. \end{lem}
\begin{proof}
From (\ref{eq:W_order_function_hbar}), we have 
\[
\left|\frac{W_{\hbar}^{r}(x,\xi)}{W_{\hbar}^{r}(x',\xi')}\cdot K_{R}(x,\xi;x',\xi')\right|\le C\langle\hbar^{-1/2}|(x,\xi)-(x',\xi')|\rangle^{2r-\nu}.
\]
 Hence, by Schur Lemma%
\footnote{For an integral bounded operator $A:L^{2}\left(\mathbb{R}^{n}\right)\rightarrow L^{2}\left(\mathbb{R}^{n}\right)$
with Schwartz kernel $K_{A}\in C^{0}\left(\mathbb{R}^{n}\times\mathbb{R}^{n}\right)$,
i.e. $\left(Au\right)\left(x\right)=\int K_{A}\left(x,y\right)u\left(y\right)dy$
we have 
\begin{equation}
\left\Vert A\right\Vert _{L^{2}}\leq\max_{y}\left|\int K_{A}\left(x,y\right)dx\right|^{1/2}.\max_{x}\left|\int K_{A}\left(x,y\right)dy\right|^{1/2}\label{eq:Schur_inequality}
\end{equation}
} \cite[p.50]{martinez-01} (or by Young inequality for convolution),
the operator norm of $u\mapsto W_{\hbar}^{r}\cdot R\left(\left(W_{\hbar}^{r}\right){}^{-1}\cdot u\right)$
with respect to the $L^{2}$ norm is bounded by a constant $C_{\nu}$.
This implies the claim of the lemma. 
\end{proof}
From expression (\ref{eq:Bargman_Kernel}), the Bargmann projector
$\BargmannP_{\hbar}$ satisfies the assumption of the last lemma.
Thus we have 
\begin{cor}
\label{cor:P_bounded} The Bargmann projector $\BargmannP_{\hbar}$
is a bounded operator on $L^{2}(\real^{2D},(W_{\hbar}^{r})^{2})$.
\end{cor}

\subsection{\label{sub:Spectrum-of-transfer_3.4}Spectrum of transfer operator
for  linear expanding map}

Below we consider the action of a linear expanding map.
\begin{lem}
\label{lm:LA_bdd_L2W} If $A:\real^{D}\to\real^{D}$ is an\textbf{
}\emph{expanding linear map} i.e. satisfying 
\begin{equation}
\|A^{-1}\|\le\frac{1}{\lambda}\qquad\mbox{for some \ensuremath{\lambda>1},}\label{eq:A_is_expanding}
\end{equation}
 then the lift $\Llift_{A}$ of $L_{A}$, defined in (\ref{def_Lift_A}),
extends to a bounded operator 
\begin{equation}
\Llift_{A}:L^{2}(\real^{2D},(W_{\hbar}^{r})^{2})\to L^{2}(\real^{2D},(W_{\hbar}^{r})^{2}).\label{eq:LA_on_L2W}
\end{equation}
 Further, if $\lambda>1$ is sufficiently large (say $\lambda>9$),
$\Llift_{A}$ extends to a bounded operator 
\begin{equation}
\Llift_{A}:L^{2}(\real^{2D},(W_{\hbar}^{r,-})^{2})\to L^{2}(\real^{2D},(W_{\hbar}^{r,+})^{2}).\label{eq:LA_on_L2Wpm}
\end{equation}
 \end{lem}
\begin{proof}
From (\ref{eq:Lift_A_expression}) in Lemma \ref{lm:lift_of_LA} and
Corollary \ref{cor:P_bounded}, we have only to check boundedness
of $L_{A\oplus{}^{t}A^{-1}}$ as an operator on $L^{2}(\real^{2D},(W^{r})^{2})$
(resp. from $L^{2}(\real^{2D},(W^{r,-})^{2})$ to $L^{2}(\real^{2D},(W^{r,+})^{2})$.
But this is clear from the definitions of $W^{r}$ and $W^{r,\pm}$. 
\end{proof}
To look into more detailed structure of the operator $L_{A}$ and
$L_{A}^{\mathrm{lift}}$, we introduce some definitions. For $k\in\mathbb{N}$,
let $\Polynomial^{(k)}$ be the space of homogeneous polynomial on
$\real^{D}$ of order $k$. Then we consider the operator%
\footnote{for a multi index $\alpha\in\mathbb{N}^{D}$, $\alpha=\left(\alpha_{1},\ldots\alpha_{D}\right)$,
we write $\left|\alpha\right|=\alpha_{1}+\ldots+\alpha_{D}$.%
} 
\begin{equation}
\Taylor^{(k)}:C^{\infty}(\real^{D})\to\Polynomial^{(k)},\qquad\left(\Taylor^{(k)}u\right)\left(x\right)=\sum_{\alpha\in\mathbb{N}^{D},|\alpha|=k}\frac{\partial^{\alpha}u(0)}{\alpha!}\cdot x^{\alpha}.\label{eq:def_Tk}
\end{equation}
This is a projector which extracts the terms of order $k$ in the
Taylor expansion. Clearly the operator $\Taylor^{(k)}$ is of finite
rank
\[
\mathrm{rank}\left(\Taylor^{(k)}\right)=\binom{D+k-1}{D-1}=\frac{\left(D+k-1\right)!}{\left(D-1\right)!k!}
\]
 and satisfies the following relations 
\begin{align}
 & \Taylor^{(k)}\circ\Taylor^{(k')}=\begin{cases}
\Taylor^{(k)}, & \mbox{if \ensuremath{k=k'};}\\
0\quad & \mbox{otherwise,}
\end{cases}\label{eq:formalrelation_for_Tk1}
\end{align}
and
\begin{equation}
\Taylor^{(k)}\circ L_{A}=L_{A}\circ\Taylor^{(k)}.\label{eq:Tk_LA}
\end{equation}

As in (\ref{eq:def_L_Lift}) we define the lift of the operator $T^{(k)}$
by 
\begin{equation}
\calT_{\hbar}^{(k)}:=\Bargmann_{\hbar}\circ\Taylor^{(k)}\circ\Bargmann_{\hbar}^{*}.\label{eq:cTk}
\end{equation}

\begin{lem}
\label{lm:Tn_bdd_L2W}Let $n\in\mathbb{N}$ and $r>0$ such that 
\begin{equation}
r>n+2D.\label{eq:assumption_on_r}
\end{equation}
Then for $0\le k\le n$ the operator $\calT_{\hbar}^{(k)}$ extends
naturally to bounded operators 
\begin{equation}
\calT_{\hbar}^{(k)}:L^{2}(\real^{2D},(W_{\hbar}^{r,-})^{2})\to L^{2}(\real^{2D},(W_{\hbar}^{r,+})^{2})\label{eq:Tn_on_L2W}
\end{equation}
and, in particular, from (\ref{eq:wrpm}), 

\begin{equation}
\calT_{\hbar}^{(k)}:L^{2}(\real^{2D},(W_{\hbar}^{r})^{2})\to L^{2}(\real^{2D},(W_{\hbar}^{r})^{2}).\label{eq:Tn_on_L2W-1}
\end{equation}
Further if we write the operator $\calT_{\hbar}^{(k)}$ as an integral
operator 
\[
\left(\calT_{\hbar}^{(k)}u\right)(x,\xi)=\int K(x,\xi;x',\xi')u(x',\xi')dx'd\xi',
\]
 the kernel $K(\cdot)$ satisfies the estimate 
\begin{align}
\left|\frac{W_{\hbar}^{r,+}(x,\xi)}{W_{\hbar}^{r,-}(x',\xi')}\cdot K(x,\xi;x',\xi')\right| & \le C\langle\hbar^{-1/2}|(x,\xi)|\rangle^{k-r}\cdot\langle\hbar^{-1/2}|(x',\xi')|\rangle^{k-r}\label{eq:kernel_T}\\
 & \le C'\langle\hbar^{-1/2}|(x,\xi)-(x',\xi')|\rangle^{k-r}\label{eq:kernel_T2}
\end{align}
 for some constants $C,C'>0$ that do not depend on $\hbar>0$. 
\end{lem}

\begin{proof}
For each multi-index $\alpha\in\mathbb{N}^{D}$ with $|\alpha|=k$,
we set
\begin{equation}
T^{(\alpha)}:\mathcal{S}(\real^{D})\to\mathcal{S}(\real^{D})',\qquad\left(T^{(\alpha)}u\right)(x):=\frac{\partial^{\alpha}u(0)}{\alpha!}\cdot x^{\alpha}s\label{eq:def_T_alpha}
\end{equation}
and
\[
\calT_{\hbar}^{(\alpha)}=\Bargmann_{\hbar}\circ\Taylor^{(\alpha)}\circ\Bargmann_{\hbar}^{*}:\mathcal{S}(\real^{2D})\to\mathcal{S}(\real^{2D})'.
\]
Since $\calT_{\hbar}^{(k)}=\sum_{\alpha:|\alpha|=k}\calT_{\hbar}^{(\alpha)}$,
the claims of the lemma follows if one proves the corresponding claim
for $\calT_{\hbar}^{(\alpha)}$. From (\ref{eq:def_Bargman_transform})
and (\ref{eq:B*}) the kernel of the operator $\calT_{\hbar}^{(\alpha)}$
is written as 
\[
K(x,\xi;x',\xi')=\frac{1}{\alpha!}\cdot k_{+}(x,\xi)\cdot k_{-}(x',\xi')
\]
 with 
\[
k_{+}(x,\xi):=\hbar^{-D/4}\int\overline{\phi_{x,\xi}(y)}\cdot(\hbar^{-1/2}y)^{\alpha}dy,\qquad k_{-}(x',\xi'):=\hbar^{D/4}\hbar^{k/2}\cdot\partial^{\alpha}\phi_{x',\xi'}(0).
\]
Applying integration by parts to the integral of $k_{+}(\cdot)$,
we see that 
\[
|k_{+}(x,\xi)|\le C_{\nu}\cdot\langle\hbar^{-1/2}|x|\rangle^{k}\cdot\langle\hbar^{-1/2}|\xi|\rangle^{-\nu}
\]
for arbitrarily large integer $\nu$, where $C_{\nu}$ is a constant
depending only on $\nu$. Also a straightforward computation gives
the similar estimate for $k_{-}(\cdot)$: 
\[
|k_{-}(x',\xi')|\le C_{\nu}\cdot\langle\hbar^{-1/2}|\xi'|\rangle^{k}\cdot\langle\hbar^{-1/2}|x'|\rangle^{-\nu}.
\]
 These estimates for sufficiently large $\nu$ imply that 
\[
W_{\hbar}^{r}(x,\xi)\cdot|k_{+}(x,\xi)|\le C\langle\hbar^{-1/2}|(x,\xi)|\rangle^{k-r}
\]
 and 
\[
\frac{1}{W_{\hbar}^{r}(x',\xi')}\cdot|k_{-}(x',\xi')|\le C\langle\hbar^{-1/2}|(x',\xi')|\rangle^{k-r}
\]
for some constant $C>0$ independent of $\hbar>0$. Thus we have obtained
(\ref{eq:kernel_T}). Since $r-k\ge r-n>2D$ from the assumption (\ref{eq:assumption_on_r})
on the choice of $r$, boundedness of the operators follows from Schur
Lemma (\ref{eq:Schur_inequality}).
\end{proof}
The following is a direct consequence of the relation (\ref{eq:formalrelation_for_Tk1})
and (\ref{eq:Tk_LA}). 
\begin{cor}
\label{cor:relation_for_hTn} For $0\le k,k'\le n$, we have 
\[
\calT_{\hbar}^{(k)}\circ\calT_{\hbar}^{(k')}=\begin{cases}
\calT_{\hbar}^{(k)}, & \quad\mbox{if \ensuremath{k=k'};}\\
0, & \quad\mbox{otherwise,}
\end{cases}
\]
 and 
\[
\Llift_{A}\circ\calT_{\hbar}^{(k)}=\calT_{\hbar}^{(k)}\circ\Llift_{A}=\Bargmann_{\hbar}\circ L_{A}\circ\Taylor^{(k)}\circ\Bargmann_{\hbar}^{*}.
\]
 
\end{cor}
Let us set 
\begin{equation}
\widetilde{\calT}_{\hbar}=\mathrm{Id}-\sum_{k=0}^{n}\calT_{\hbar}^{(k)}:L^{2}(\real^{2D},(W_{\hbar}^{r})^{2})\to L^{2}(\real^{2D},(W_{\hbar}^{r})^{2}).\label{eq:tTh}
\end{equation}
Then the set of operators $\calT_{\hbar}^{(k)}$, $0\le k\le n$,
and $\widetilde{\calT}_{\hbar}$ form a complete set of mutually commuting
projection operators on $L^{2}(\real^{2D},(W_{\hbar}^{r})^{2})$ such
that 
\[
\mathrm{rank}\,\calT_{\hbar}^{(k)}=\dim\Polynomial^{(k)}=\binom{D+k-1}{D-1}=\frac{\left(D+k-1\right)!}{\left(D-1\right)!k!},\qquad\mathrm{{rank}\,\widetilde{\calT}_{\hbar}}=\infty.
\]
Let
\[
H_{k}:=\mbox{Im}\calT_{\hbar}^{(k)}\quad\mbox{ and }\quad\widetilde{H}=\mbox{Im}\widetilde{\calT}_{\hbar}.
\]
Then the Hilbert space $L^{2}(\real^{2D},(W_{\hbar}^{r})^{2})$ is
decomposed as 
\begin{equation}
L^{2}(\real^{2D},(W_{\hbar}^{r})^{2})=H_{0}\oplus H_{1}\oplus H_{2}\oplus\cdots\oplus H_{n}\oplus\widetilde{{H}}\label{eq:L2W_decomp}
\end{equation}
Since the operator $\Llift_{A}$ commutes with the projections $\calT_{\hbar}^{(k)}$
and $\widetilde{\calT}_{\hbar}$, it preserves this decomposition
and therefore the operator $\Llift_{A}$ acting on $L\left(\real^{2D},\left(W_{\hbar}^{r}\right)^{2}\right)$
is identified with the direct sum of the operators 
\[
\Llift_{A}:H_{k}\to H_{k}\quad\mbox{for \ensuremath{0\le k\le n},\quad\ and\quad}\Llift_{A}:\widetilde{H}\to\widetilde{H}.
\]
 The former is identified with the action of $L_{A}$ on $\Polynomial^{(k)}$,
because the diagram 
\begin{equation}
\begin{CD}H_{k}@>{\Llift_{A}}>>H_{k}\\
@A{\Bargmann_{\hbar}}AA@A{\Bargmann_{\hbar}}AA\\
\Polynomial^{(k)}@>{L_{A}}>>\Polynomial^{(k)}
\end{CD}\label{cd:conj_to_poly}
\end{equation}
 commutes and the operator $\Bargmann_{\hbar}:\Polynomial^{(k)}\to H_{k}$
in the vertical direction is an isomorphism between finite dimensional
linear spaces.

To state the next proposition which is the main result of this section,
we introduce the following definition:
\begin{defn}
\label{Def:H^r_h}The Hilbert space $H_{\hbar}^{r}\left(\real^{D}\right)\subset\mathcal{S}'\left(\mathbb{R}^{D}\right)$\textbf{
}of distributions is the completion of $\mathcal{S}\left(\mathbb{R}^{D}\right)$
with respect to the norm induced by the scalar product 
\[
(u,v)_{H_{\hbar}^{r}\left(\real^{D}\right)}:=(\Bargmann_{\hbar}u,\Bargmann_{\hbar}v)_{L^{2}(\real^{2D},(W_{\hbar}^{r})^{2})}=\int(W_{\hbar}^{r})^{2}\cdot\overline{\Bargmann_{\hbar}u}\cdot\Bargmann_{\hbar}v\frac{dxd\xi}{(2\pi\hbar)^{D}}\quad\mbox{for }u,v\in\mathcal{S}(\real^{D}).
\]
The induced norm on $H_{\hbar}^{r}\left(\real^{D}\right)$ will be
written as $\|u\|_{H_{\hbar}^{r}\left(\real^{D}\right)}:=\|\Bargmann_{\hbar}u\|_{L^{2}(\real^{2D},(W_{\hbar}^{r})^{2})}$. 
\end{defn}
%blue 
\begin{center}{\color{blue}\fbox{\color{black}\parbox{16cm}{
\begin{prop}
\label{pp:structure_of_hLA}\textbf{''Discrete spectrum of the linear
expanding map''.} Let $A:\real^{D}\to\real^{D}$ be a linear expanding
map satisfying $\|A^{-1}\|\le1/\lambda$ for some $\lambda>1$. Let
$L_{A}$ be the unitary transfer operator defined in (\ref{eq:def_LA}):
$L_{A}u=u\circ A^{-1}$. \textup{Let $n>0$ and $r>n+2D$. Then the
Hilbert space $H_{\hbar}^{r}\left(\real^{D}\right)$ of definition
\ref{Def:H^r_h} is decomposed into }subspaces of homogeneous polynomial
of degree $k$ for \textup{$0\le k\leq n$ and the remainder: 
\[
H_{\hbar}^{r}\left(\real^{D}\right)=\left(\bigoplus_{k=0}^{n}\Polynomial^{(k)}\right)\oplus\widetilde{\mathcal{H}}_{\hbar}
\]
where $\widetilde{\mathcal{H}}_{\hbar}:=\widetilde{T}(H_{\hbar}^{r}\left(\real^{D}\right))$
with setting 
\begin{equation}
\widetilde{T}:=\mathrm{Id}-\sum_{k=0}^{n}T^{\left(k\right)}.\label{eq:def_T_tilde}
\end{equation}
(The operators $T^{\left(k\right)}$ are defined in (\ref{eq:def_Tk}).)
This decomposition is preserved by $L_{A}$. There exists a }constant
$C_{0}>0$ independent of $A$ and $\hbar$ such that
\begin{enumerate}
\item For $0\leq k\leq n$ and $0\neq u\in\Polynomial^{(k)}$, we have 
\begin{equation}
C_{0}^{-1}\|A\|_{\max}^{-k}\le\frac{\|L_{A}u\|_{H_{\hbar}^{r}\left(\real^{D}\right)}}{\|u\|_{H_{\hbar}^{r}\left(\real^{D}\right)}}\le C_{0}\|A\|_{\min}^{-k}\label{eq:bound_on_norm-1}
\end{equation}
(Recall (\ref{eq:def_of_max_min_norm}) for the definition of $\|\cdot\|_{\max}$
and $\|\cdot\|_{\min}$.)
\item The operator norm of the restriction of $L_{A}$ to $\widetilde{\mathcal{H}}_{\hbar}$
is bounded by
\begin{equation}
C_{0}\max\{\|A\|_{\min}^{-(n+1)},\,\|A\|_{\min}^{-r}\cdot|\det A|\}.\label{eq:bound_norm_reminder-1}
\end{equation}
 
\end{enumerate}
The following equivalent statements holds for the lifted operator:
\[
\Llift_{A}:L^{2}(\real^{2D},(W_{\hbar}^{r})^{2})\to L^{2}(\real^{2D},(W_{\hbar}^{r})^{2}).
\]
 The operator $\Llift_{A}$ preserves the decomposition of $L^{2}(\real^{2D},(W_{\hbar}^{r})^{2})$
in (\ref{eq:L2W_decomp}) and
\begin{enumerate}
\item For $0\le k\le n$ and for $0\neq u\in H_{k}$, we have 
\begin{equation}
C_{0}^{-1}\|A\|_{\max}^{-k}\le\frac{\|\Llift_{A}u\|_{L^{2}(\real^{2D},(W_{\hbar}^{r})^{2})}}{\|u\|_{L^{2}(\real^{2D},(W_{\hbar}^{r})^{2})}}\le C_{0}\|A\|_{\max}^{-k}.\label{eq:bound_on_norm}
\end{equation}

\item The operator norm of the restriction of $\Llift_{A}$ to $\widetilde{H}$
is bounded by (\ref{eq:bound_norm_reminder-1}). 
\end{enumerate}
\end{prop}
}}}\end{center}
\begin{rem}
Proposition \ref{pp:structure_of_hLA} implies that the spectrum of
the transfer operator $L_{A}$ in the Hilbert space $H_{\hbar}^{r}\left(\real^{D}\right)$
is discrete outside the radius given by (\ref{eq:bound_norm_reminder-1}).
The eigenvalues outside this radius are given by the action of $L_{A}$
in the finite dimensional space $\Polynomial^{(k)}$. These eigenvalues
can be computed explicitly from the Jordan block decomposition of
$A$. In particular if $A=\mathrm{Diag}\left(a_{1},\ldots a_{D}\right)$
is diagonal then the monomials $x^{\alpha}=x_{1}^{\alpha_{1}}\ldots x_{D}^{\alpha_{D}}$
 are obviously eigenvectors of $L_{A}$ with respective eigenvalues
$\prod_{j}a_{j}^{-\alpha_{j}}$.

\end{rem}
\begin{proof}
For the proof of (\ref{eq:bound_on_norm-1}) and (\ref{eq:bound_on_norm}),
we use the fact that the space $\Polynomial^{(k)}$ is identical to
the space $\mathrm{Sym}^{k}\left(\mathbb{R}^{D}\right)$ of totally
symmetric tensors of rank $k$. For the linear operator $\left(A^{-1}\right)^{\otimes k}$
acting on $\left(\mathbb{R}^{D}\right)^{\otimes k}$, we have a commutative
diagram:

\[
\begin{CD}\left(\mathbb{R}^{D}\right)^{\otimes k}@>\left(A^{-1}\right)^{\otimes k}>>\left(\mathbb{R}^{D}\right)^{\otimes k}\\
@V\mathrm{Sym}VV@V\mathrm{Sym}VV\\
\mathrm{Sym}^{k}\left(\mathbb{R}^{D}\right)@>L_{A}>>\mathrm{Sym}^{k}\left(\mathbb{R}^{D}\right)
\end{CD}
\]
where $\mathrm{Sym}$ denotes the symmetrization operation. For every
$0\neq\tilde{u}\in\left(\mathbb{R}^{D}\right)^{\otimes k}$ we have
\[
\left\Vert A\right\Vert _{\max}^{-k}\leq\frac{\left\Vert \left(A^{-1}\right)^{\otimes k}\tilde{u}\right\Vert }{\left\Vert \tilde{u}\right\Vert }\leq\left\Vert A\right\Vert _{\min}^{-k}.
\]
Since the spaces are finite dimensional (and hence all norms are equivalent),
we deduce (\ref{eq:bound_on_norm-1}) for some constant $C_{0}>0$
independent of $A$, and also independent on $\hbar$ because of the
scaling invariance (\ref{eq:def_Wh}). The proof of the Claim (2)
is postponed to Subsection \ref{ss:pf}, as it requires more detailed
argument.
\end{proof}

\subsection{Proof of Claim (2) in Proposition \ref{pp:structure_of_hLA}}

\label{ss:pf} We prove Claim (2) on the lifted operator $\Llift_{A}$
in the latter part of the statement, which is equivalent to Claim
(2) in the former part. In the proof below, we may and do assume $\hbar=1$,
because the Bargmann transforms for different parameter $\hbar$ are
related by the scaling (\ref{eq:wI_hbar}), as we noted in Subsection
\ref{sub:The-Bargmann-transform}. Accordingly we will drop the subscript
$\hbar$ from the notation. Let $\chi:\real^{D}\to[0,1]$ be a smooth
function such that 
\begin{equation}
\chi\left(x\right)=\begin{cases}
1 & \mbox{ if \ensuremath{|x|\le1}}\\
0 & \mbox{ if \ensuremath{|x|\ge2}.}
\end{cases}\label{eq:def_chi}
\end{equation}
Below we use $C_{0}$ as a generic symbol for the constants which
do not depend on $A$ (but may depend on $r$, $n$ and $D$). Letting
$\lambda$ smaller if necessary, we suppose 
\[
\lambda=\|A^{-1}\|^{-1}>1
\]
 for simplicity. We write $\multiplication(\varphi)$ for the multiplication
operator by $\varphi$.

To prove the claim, it is enough to show 
\[
\|\Llift_{A}\circ\widetilde{\calT}\|_{L^{2}(\real^{2D},(W^{r})^{2})}\le C_{0}\cdot\max\{\lambda^{-n-1},\,\lambda^{-r}|\det A|\}
\]
where $\widetilde{\calT}$ is the operator defined in (\ref{eq:tTh})
with $\hbar=1$ and $\|\cdot\|_{L^{2}(\real^{2D},(W^{r})^{2})}$ denotes
the operator norm on $L^{2}(\real^{2D},(W^{r})^{2})$.

Let us consider the operators 
\[
X=\Bargmann\circ\multiplication(\chi)\circ\Bargmann^{*}:L^{2}(\real^{2D},W^{r})\to L^{2}(\real^{2D},W^{r})
\]
and 
\[
\Xi:L^{2}(\real^{2D},(W^{r})^{2})\to L^{2}(\real^{2D},(W^{r})^{2}),\quad\left(\Xi v\right)\left(x,\xi\right)=\chi\left(\frac{|\xi|}{\lambda}\right)\cdot v\left(x,\xi\right).
\]
The next lemma is the main ingredient of the proof. 
\begin{lem}
${\displaystyle \|\Llift_{A}\circ X\circ\widetilde{\calT}\circ\Xi\|_{L^{2}(\real^{2D},(W^{r})^{2})}\le C_{0}\cdot\lambda^{-(n+1)}}$. \end{lem}
\begin{proof}
Let $\mathbf{1}_{{\mathbf{C}}_{-}(2)}$ be the indicator function
of the cone 
\[
{\mathbf{C}}_{-}(2):=\{(x,\xi)\mid|x|\le2|\xi|\}=\overline{\real^{2D}\setminus{\mathbf{C}}_{+}(1/2)}
\]
and set 
\[
W_{+}^{r}(x,\xi):=\mathbf{1}_{{\mathbf{C}}_{-}(2)}(x,\xi)\cdot\langle|\xi|\rangle^{r},\qquad W_{-}^{r}(x,\xi):=\langle|x|\rangle^{-r}.
\]
Then the weight function $W^{r}(x,\xi)$ satisfies 
\[
W^{r}(x,\xi)\le C_{0}\cdot W_{+}^{r}(x,\xi)+C_{0}\cdot W_{-}^{r}(x,\xi)
\]
for a constant $C_{0}>0$. Hence, to prove the lemma, it is enough
to show for $\sigma=\pm$ the claim 
\begin{equation}
\|W_{\sigma}^{r}\cdot\Bargmann\circ L_{A}\circ\multiplication(\chi)\circ\widetilde{T}\circ\Bargmann^{*}\circ\Xi\, u\|_{L^{2}}\le C_{0}\cdot\lambda^{-(n+1)}\|u\|\quad\mbox{for any }u\in L^{2}\left(\mathbb{R}^{2D},\left(W^{r}\right)^{2}\right)\label{eq:claim_to_prove}
\end{equation}
where $\widetilde{T}$ is the operator defined in (\ref{eq:def_T_tilde}).
Before proceeding to the proof of (\ref{eq:claim_to_prove}), we prepare
a few estimates. Take $u\in L^{2}\left(\mathbb{R}^{2D},\left(W^{r}\right)^{2}\right)$
arbitrarily and set 
\[
v(y):=\left(\Bargmann^{*}\circ\Xi\, u\right)(y)=\int\phi_{x,\xi}(y)\chi\left(\frac{|\xi|}{\lambda}\right)u(x,\xi)\, dxd\xi.
\]
Then, for any multi-index $\alpha\in\mathbb{N}^{D}$ and arbitrarily
large $\nu$, we have 
\begin{equation}
|\partial_{y}^{\alpha}v(y)|\le C_{\alpha,\nu}\int_{|\xi|\le2\lambda}\langle|x-y|\rangle^{-\nu}\cdot\langle|\xi|\rangle^{|\alpha|}\cdot\left|u(x,\xi)\right|\, dxd\xi\qquad\mbox{for any \ensuremath{y\in\real^{D}}.}\label{eq:estimate_for_dv}
\end{equation}
Note that we have 
\[
\langle|x-y|\rangle^{-r}\cdot\langle|\xi|\rangle^{r}\le C_{0}\cdot W^{r}(x,\xi)\quad\mbox{ for any \ensuremath{x,y,\xi\in\real^{D}}with \ensuremath{|y|\le2}}.
\]
Hence 
\[
\langle|x-y|\rangle^{-\nu}\cdot\langle|\xi|\rangle^{|\alpha|}\le C_{0}(\langle|x-y|\rangle^{-\nu+r}\cdot\langle|\xi|\rangle^{-D/2+1})\cdot\langle|\xi|\rangle^{|\alpha|+D/2+1-r}\cdot W^{r}(x,\xi).
\]
Putting this estimate in (\ref{eq:estimate_for_dv}) with $\nu\ge D/2+1+r$,
we obtain, by Cauchy-Schwarz inequality, 
\begin{equation}
|\partial_{y}^{\alpha}v(y)|\le C_{\alpha}\lambda^{\max\{|\alpha|+D/2+1-r,0\}}\|u\|_{L^{2}(\real^{2D},(W^{r})^{2})}\quad\mbox{for any \ensuremath{y\in\real^{D}}with \ensuremath{|y|\le2}. }\label{eq:alphav}
\end{equation}
Notice that, if $|\alpha|\le n+1$, we have $|\alpha|+D/2+1-r\le0$
from the assumption (\ref{eq:assumption_on_r}) and hence the last
inequality implies 
\begin{equation}
|\partial_{y}^{\alpha}v(y)|\le C_{\alpha}\|u\|_{L^{2}(\real^{2D},(W^{r})^{2})}\quad\mbox{for any \ensuremath{y\in\real^{D}}with \ensuremath{|y|\le2}. }\label{eq:alphav-1}
\end{equation}

Next we consider the function 
\[
w:=\multiplication(\chi)\circ\widetilde{T}\, v=\multiplication(\chi)\circ\widetilde{T}\circ\Bargmann^{*}\circ\Xi\, u.
\]
 Note that the support of $w$ is contained in that of $\chi$. It
follows from (\ref{eq:alphav}) that, for each multi-index $\alpha$,
\begin{equation}
|\partial^{\alpha}w(y)|\le C_{\alpha}\lambda^{\max\{|\alpha|+D/2+1-r,0\}}\|u\|_{L^{2}(\real^{2D},(W^{r})^{2})}\quad\mbox{ for all \ensuremath{y\in\real^{D}}.}\label{eq:taylor2}
\end{equation}
Further, if $|\alpha|\le n+1$, it follows from (\ref{eq:alphav-1})
and the definition of $\widetilde{T}$ that 
\begin{equation}
|\partial_{y}^{\alpha}w(y)|\le C_{\alpha}\left|y\right|^{n+1-\left|\alpha\right|}\cdot\max_{\left|y\right|\leq2}\left|\partial_{y}^{n+1}v\right|\leq C_{\alpha}|y|^{n+1-|\alpha|}\|u\|_{L^{2}(\real^{2D},(W^{r})^{2})}\quad\mbox{ for all \ensuremath{y\in\real^{D}.}}\label{eq:taylor}
\end{equation}

Now we prove the claim (\ref{eq:claim_to_prove}) in the case $\sigma=+$.
Let $u$, $v$, $w$ be as above. We are going to estimate the quantity
\[
\xi^{\alpha}\cdot(\Bargmann\circ L_{A}\circ\multiplication(\chi)\circ\widetilde{T}\circ\Bargmann^{*}\circ\Xi\, u)(x,\xi)=\xi^{\alpha}\int\overline{\phi_{x,\xi}(y)}\cdot w(A^{-1}y)\, dy.
\]
By integration by parts, we see that this is bounded in absolute value
by 
\begin{align*}
C_{\alpha,\nu}\sum_{\alpha'\le\alpha}\int\langle|x-y|\rangle^{-\nu}\cdot\lambda^{-|\alpha'|}\cdot|\partial^{\alpha'}w(A^{-1}y)|\, dy
\end{align*}
for each $\nu>0$, where $C_{\alpha,\nu}$ is a constant depending
only on $\alpha$ and $\nu$. If $\nu$ is sufficiently large, we
have from (\ref{eq:taylor2}), (\ref{eq:taylor}) and then from (\ref{eq:assumption_on_r})
that 
\begin{align*}
 & \sum_{\alpha'\le\alpha}\int\langle|x-y|\rangle^{-\nu}\cdot\lambda^{-|\alpha'|}\cdot|\partial^{\alpha'}w(A^{-1}y)|\, dy\\
 & \le C_{\alpha}\bigg(\sum_{\alpha':|\alpha'|\le n+1}\mathbf{\lambda^{-|\alpha'|}}\left(\frac{\langle|x|\rangle}{\lambda}\right)^{n+1-|\alpha'|}+\sum_{\alpha':n+2\le|\alpha'|\le|\alpha|}\lambda^{-|\alpha'|}\lambda^{\max\{|\alpha'|+D/2+1-r,0\}}\bigg)\|u\|_{L^{2}(\real^{2D},(W^{r})^{2})}\\
 & \le C_{\alpha}\cdot\lambda^{-(n+1)}\langle|x|\rangle^{n+1}\cdot\|u\|_{L^{2}(\real^{2D},(W^{r})^{2})}.
\end{align*}
Therefore we obtain 
\[
\langle|\xi|\rangle^{\nu}\cdot|\Bargmann\circ L_{A}\circ\multiplication(\chi)\circ\widetilde{T}\circ\Bargmann^{*}\circ\Xi u(x,\xi)|\le C_{\nu}\lambda^{-(n+1)}\|u\|_{L^{2}(\real^{2D},(W^{r})^{2})}\cdot\langle|x|\rangle^{n+1}
\]
for arbitrarily large $\nu$. For $(x,\xi)$ on $\supp W_{+}^{r}={\mathbf{C}}_{-}(2)$,
we have $\langle|x|\rangle\le2\langle|\xi|\rangle$ and hence
\[
W_{+}^{r}(x,\xi)\le\langle|\xi|\rangle^{r}\le C_{0}\langle|\xi|\rangle^{r+D/2+1+(n+1)}\cdot\langle|x|\rangle^{-D/2-1-(n+1)}.
\]
Using this in the last inequality, we get
\begin{multline*}
W_{+}^{r}(x,\xi)|(\Bargmann\circ L_{A}\circ\multiplication(\chi)\circ(\mathrm{Id}-\boldT_{n})\circ\Bargmann^{*}\circ\Xi u)(x,\xi)|\\
\le C_{\nu}\lambda^{-(n+1)}\|u\|_{L^{2}(\real^{2D},(W^{r})^{2})}\langle|x|\rangle^{-D/2-1}\langle|\xi|\rangle^{-\nu+r+D/2+1+(n+1)}.
\end{multline*}
This estimate for sufficiently large $\nu$ implies the claim (\ref{eq:claim_to_prove})
in the case $\sigma=+$, by Cauchy-Schwarz inequality.

We prove the claim (\ref{eq:claim_to_prove}) for $\sigma=-$. The
proof is easier than the previous case actually. Note that we have
\[
\|W_{-}^{r}\cdot\,\Bargmann\varphi\|_{L^{2}\left(\mathbb{R}^{2D}\right)}\le C_{0}\cdot\|\langle\cdot\rangle^{-r}\cdot\varphi(\cdot)\|_{L^{2}(\real^{D})}.
\]
Hence, from (\ref{eq:taylor}) with $\alpha=\emptyset$, we get 
\begin{align*}
\|W_{-}^{r}\cdot & \,\Bargmann\circ L_{A}\circ\multiplication(\chi)\circ\widetilde{T}\circ\Bargmann^{*}\circ\Xi\, u\|_{L^{2}\left(\mathbb{R}^{2D}\right)}=\|W_{-}^{r}\cdot\Bargmann\circ L_{A}\circ w\|_{L^{2}\left(\mathbb{R}^{2D}\right)}\\
 & \le C_{0}\cdot\left|\int\langle x\rangle^{-2r}\langle x/\lambda\rangle^{2(n+1)}dx\right|^{1/2}\cdot\|u\|_{L^{2}(\real^{2D},(W^{r})^{2})}\\
 & \le C_{0}\cdot\lambda^{-n-1}\cdot\|u\|_{L^{2}(\real^{2D},(W^{r})^{2})}.
\end{align*}
Clearly this implies (\ref{eq:claim_to_prove}) for $\sigma=-$. 
\end{proof}
To finish, it is enough to show 
\begin{equation}
\|\Llift_{A}\circ\widetilde{\calT}-\Llift_{A}\circ X\circ\widetilde{\calT}\circ\Xi\|_{L^{2}(\real^{2D},(W^{r})^{2})}\le C_{0}\lambda^{-r}|\det A|.\label{eq:remainder_part}
\end{equation}
 Note the relations 
\[
\Llift_{A}\circ\widetilde{\calT}=\Bargmann\circ L_{A}\circ\widetilde{T}\circ\Bargmann=\Bargmann\circ\widetilde{T}\circ L_{A}\circ\Bargmann=\widetilde{\calT}\circ\Llift_{A}
\]
 and 
\begin{align*}
\Llift_{A}\circ(\mathrm{Id}-X) & =\Bargmann\circ L_{A}\circ\multiplication(1-\chi)\circ\Bargmann^{*}\\
 & =\Bargmann\circ\multiplication(1-\chi_{A})\circ L_{A}\circ\Bargmann^{*}=(\mathrm{Id}-X_{A})\circ\Llift_{A}
\end{align*}
 where $\chi_{A}=\chi\circ A^{-1}$ and $X_{A}=\Bargmann\circ\chi_{A}\circ\Bargmann^{*}$.
Below we will prove the claims 
\begin{equation}
\|\Llift_{A}\circ(\mathrm{Id}-X)\|_{L^{2}(\real^{2D},(W^{r})^{2})}\le C_{0}\lambda^{-r}|\det A|\label{eq:norm_of_LAX}
\end{equation}
 and 
\begin{equation}
\|\Llift_{A}\circ(\mathrm{Id}-\Xi)\|_{L^{2}(\real^{2D},(W^{r})^{2})}\le C_{0}\lambda^{-r}|\det A|.\label{eq:norm_of_LAXI}
\end{equation}
 Since $\widetilde{\calT}=\mathrm{Id}-\sum_{k=0}^{n}\calT^{(k)}$
is a bounded operator on $L^{2}(\real^{2D},W^{r})$, these claims
would imply 
\[
\|\Llift_{A}\circ(\mathrm{Id}-X)\circ\widetilde{\calT}\|_{L^{2}(\real^{2D},(W^{r})^{2})}\le C_{0}\cdot\lambda^{-r}|\det A|
\]
 and 
\[
\|\Llift_{A}\circ X\circ\widetilde{\calT}\circ(\mathrm{Id}-\Xi)\|_{L^{2}(\real^{2D},(W^{r})^{2})}=\|X_{A}\circ\widetilde{\calT}\circ\Llift_{A}\circ(\mathrm{Id}-\Xi)\|_{L^{2}(\real^{2D},(W^{r})^{2})}\le C_{0}\lambda^{-r}|\det A|
\]
 and therefore the conclusion (\ref{eq:remainder_part}) would follow.

We can prove (\ref{eq:norm_of_LAX}) and (\ref{eq:norm_of_LAXI})
by straightforward estimate. Writing the kernel of the operator $X_{A}$
explicitly and applying integration by parts to it, we get the estimate
\[
|(\mathrm{Id}-X_{A})v(x,\xi)|\le C_{\nu}\int K_{1}^{(\nu)}(x,\xi;x',\xi')\cdot v(x',\xi')dx'd\xi'
\]
 for arbitrarily large $\nu$, where 
\[
K_{1}^{(\nu)}(x,\xi;x',\xi')=\int_{\supp(1-\chi_{A})}\langle|x-y|\rangle^{-\nu}\langle|y-x'|\rangle^{-\nu}\langle|\xi-\xi'|\rangle^{-\nu}dy.
\]
 From the expression (\ref{eq:Lift_A_expression}) of the operator
$\Llift_{A}$, we also have 
\[
|\Llift_{A}v(x',\xi')|\le C_{\nu}d\left(A\right)\int K_{2}^{(\nu)}(x',\xi';x'',\xi'')v(x'',\xi'')dx'd\xi'
\]
 for arbitrarily large $\nu$, where 
\[
K_{2}^{(\nu)}(x',\xi';x'',\xi'')=\int\langle|x'-x_{\dag}|\rangle^{-\nu}\langle|\xi'-\xi_{\dag}|\rangle^{-\nu}\langle|A^{-1}x_{\dag}-x''|\rangle^{-\nu}\langle|{}^{t}A\xi_{\dag}-\xi''|\rangle^{-\nu}dx_{\dag}d\xi_{\dag}.
\]
We have $d(A)\le\left|\det A\right|$ also. From the definition of
the function $W^{r}$ and the expanding property (\ref{eq:A_is_expanding})
of $A$, we have 
\[
\frac{W^{r}(x_{\dag},\xi_{\dag})}{W^{r}(A^{-1}x_{\dag},{}^{t}A\xi_{\dag})}\cdot\langle|x_{\dag}-y|\rangle^{-2r}\le C_{0}\lambda^{-r}\quad\mbox{if \ensuremath{y\in\supp(1-\chi_{A})}.}
\]
Also note that, from the property (\ref{eq:W_order_function_hbar})
of $W^{r}$, we have 
\[
W^{r}(x,\xi)\cdot\langle|x-y|\rangle^{-2r}\langle|y-x'|\rangle^{-2r}\langle|x'-x_{\dag}|\rangle^{-2r}\le C_{0}\cdot W^{r}(x_{\dag},\xi_{\dag})
\]
 and 
\[
\frac{1}{W^{r}(x'',\xi'')}\cdot\langle|A^{-1}x_{\dag}-x''|\rangle^{-2r}\langle|{}^{t}A\xi_{\dag}-\xi''|\rangle^{-2r}\le C_{0}\cdot\frac{1}{W^{r}(A^{-1}x_{\dag},{}^{t}A\xi_{\dag})}.
\]
Summarizing these estimates, we obtain 
\begin{align*}
\frac{W^{r}(x,\xi)}{W^{r}(x'',\xi'')}\cdot & \int K_{1}^{(\nu)}(x,\xi;x',\xi';y)\cdot K_{2}^{(\nu)}(x',\xi';x'',\xi'')dx'd\xi'\\
 & \le C_{0}\lambda^{-r}\cdot\int\int K_{1}^{(\nu-2r)}(x,\xi;x',\xi';y)\cdot K_{2}^{(\nu-4r)}(x',\xi';x'',\xi'')dx'd\xi'.
\end{align*}
By Schur inequality (\ref{eq:Schur_inequality}), the integral operators
with the kernels $K_{1}^{(\nu-4r)}(\cdot)$ and $K_{2}^{(\nu-2r)}(\cdot)$
are bounded operators on $L^{2}(\real^{2D})$ and the operator norms
are bounded by a constant that does not depend on $A$, provided $\nu$
is sufficiently large. Therefore the last estimate implies (\ref{eq:norm_of_LAX}).
We can prove the claim (\ref{eq:norm_of_LAXI}) in the same manner.

\section{\label{sec:Resonance-of-hyperbolic_preq_4}Resonance of hyperbolic
linear prequantum maps}

The main result of this section, Proposition \ref{prop:prequatum_op_for_hyp_linear},
concerns the spectrum of the prequantum transfer operator for linear
hyperbolic symplectic maps.

\subsection{Prequantum transfer operator on $\real^{2d}$\label{sub: Euclidean_prequantum_and_Laplacian_operator}}

In this section, we study prequantum transfer operators and rough
Laplacian in a special and easy case: The manifold $M$ is the linear
space $\mathbb{R}^{2d}$ with the coordinates 
\begin{equation}
x\equiv\left(q,p\right)=\left(q^{1},\ldots q^{d},p^{1},\ldots p^{d}\right)\in\mathbb{R}^{2d}.\label{eq:Coordinate_p_q}
\end{equation}
We regard it as a symplectic manifold equipped with the symplectic
two form 

\begin{equation}
\omega=dq\wedge dp:=\sum_{i=1}^{d}dq^{i}\wedge dp^{i}.\label{eq:Symplectic_form_on_Euclidean_space}
\end{equation}
\label{eq:Euclidean_symplectic_form}The prequantum bundle $P$ is
the trivial $U(1)$-bundle $\pi:P=\real^{2d}\times U(1)\to\real^{2d}$
over $\real^{2d}$ equipped with the connection one form $A=id\theta-i(2\pi)\eta$
where 
\begin{equation}
\eta=\sum_{i=1}^{d}\left(\frac{1}{2}q^{i}dp^{i}-\frac{1}{2}p^{i}dq^{i}\right).\label{eq:eta_3-1}
\end{equation}
The corresponding curvature two form is then 
\[
\Theta=-i(2\pi)(\pi^{*}\omega)
\]
because $\omega=d\eta$. Under these settings, we may rephrase the
construction of the prequantum transfer operator for a symplectic
diffeomorphism on $\real^{2d}$. 

Let $f:U\to U'$ be a symplectic diffeomorphism between two domains
$U$ and $U'$ in $\real^{2d}$ with respect to the symplectic two
form $\omega$. Let $\tilde{f}:U\times\mathbf{U}(1)\to U'\times\mathbf{U}(1)$
be the equivariant lift of $f$ preserving the connection $A$, that
is, the map satisfying the conditions (\ref{eq:lift_of_f}), (\ref{eq:equivariance_f_tilde})
and (\ref{eq:preserve_connection}) in Theorem \ref{thm:Bundle-P_map_f_tilde}.
Then we define as in (\ref{eq:def_prequantum_operator_F}) (but for
the set $V\equiv0$) the prequantum operator 
\[
\hat{F}:C^{\infty}(U\times\mathbf{U}(1))\to C^{\infty}(U'\times\mathbf{U}(1)),\qquad\hat{F}u(x)=u\circ\tilde{f}^{-1}(x)
\]
and let 
\[
\hat{F}_{N}:C_{N}^{\infty}(U\times\mathbf{U}(1))\to C_{N}^{\infty}(U'\times\mathbf{U}(1))
\]
be its restriction to the space of functions in the $N$-th Fourier
mode.  Let 
\[
\mathcal{L}_{f}:C^{\infty}(\real^{2d})\to C^{\infty}(\real^{2d})
\]
be the expression of the prequantum transfer operator $\hat{F}_{N}$
with respect to the trivialization using the trivial section $\tau_{0}:\real^{2d}\to P=\real^{2d}\times U(1)$
defined by $\tau_{0}(x)=(x,1)$. This operator $\mathcal{L}_{f}$
is the\emph{ prequantum transfer operator }for $f:U\to U'$. (Note
that $\mathcal{L}_{f}$ depends on the integer $N\in\mathbb{Z}$ and
hence on $\hbar$.) We recall its concrete expression obtained in
Proposition \ref{prop:Local-expression-of_FN}.
\begin{prop}
\label{prop:expression_of_prequantum_op_on_Rn-1}The operator $\mathcal{L}_{f}$
as above is written 
\begin{equation}
\left(\mathcal{L}_{f}u\right)\left(x\right):=e^{-\frac{i}{\hbar}\mathcal{A}_{f}\left(f^{-1}\left(x\right)\right)}u\left(f^{-1}\left(x\right)\right)\label{eq:F_f_affine-1}
\end{equation}
with the (action) function 
\begin{equation}
\mathcal{A}_{f}\left(x\right)=\int_{\gamma}f^{*}\eta-\eta\label{eq:Action_affine_map-1}
\end{equation}
where $\gamma$ is a path from a fixed point $x_{0}\in U'$ to $x$.
\end{prop}

\subsection{\label{sub:Prequantum-operator-on_R2d}Prequantum operator for a
symplectic affine map on $\mathbb{R}^{2d}$}

Let $f:\real^{2d}\rightarrow\real^{2d}$ be an affine map preserving
the symplectic form $\omega$, written:
\begin{equation}
f:\real^{2d}\to\real^{2d},\qquad f(x)=Bx+b\label{eq:affine_map}
\end{equation}
where $B:\real^{2d}\to\real^{2d}$ is a linear symplectic map and
$b\in\real^{2d}$ a constant vector. 
\begin{prop}
\label{prop:expression_of_prequantum_op_on_Rn}The prequantum transfer
operator $\mathcal{L}_{f}$ for an affine map $f$ as above is written
as 
\begin{equation}
\mathcal{L}_{f}u\left(x\right):=e^{-\frac{i}{\hbar}\mathcal{A}_{f}\left(f^{-1}\left(x\right)\right)}u\left(f^{-1}\left(x\right)\right)\label{eq:F_f_affine}
\end{equation}
with the (action) function 
\begin{equation}
\mathcal{A}_{f}\left(x\right)=\frac{1}{2}\omega\left(b,x\right).\label{eq:Action_affine_map}
\end{equation}
\end{prop}
\begin{rem}
Notice that the function $\mathcal{A}_{f}$ in (\ref{eq:Action_affine_map})
does not depend on the linear map $B$ which enters in (\ref{eq:affine_map}).\end{rem}
\begin{proof}
Then, for any $x=\left(q,p\right)\in\real^{2d}$, using the parametrized
path $\gamma\left(t\right)=(q'\left(t\right),p'\left(t\right)=\left(tq,tp\right)$
with $t\in\left[0,1\right]$ we have 
\[
\int_{\gamma}\eta=\int_{\gamma}\frac{1}{2}\left(q'dp'-p'dq'\right)=\int_{0}^{1}\frac{1}{2}\left(tqp-tpq\right)dt=0
\]
Therefore for a linear symplectic map $f_{2}\left(x\right)=Bx$, the
action defined in (\ref{eq:Action_affine_map-1}) vanishes:
\[
\mathcal{A}_{f_{2}}\left(x\right)=\int_{f_{2}\left(0\right)}^{f_{2}\left(x\right)}\eta-\int_{0}^{x}\eta=0-0=0.
\]
For a translation map $f_{1}\left(x\right)=x+b$, using the parametrized
path $\left(q'\left(t\right),p'\left(t\right)\right)=\left(tq+b_{q},tp+b_{p}\right)$
with $b=\left(b_{q},b_{p}\right)$ and $t\in\left[0,1\right]$ we
have 
\begin{eqnarray*}
\mathcal{A}_{f_{1}}\left(x\right) & = & \int_{b}^{x+b}\eta-\int_{0}^{x}\eta=\int_{b}^{x+b}\frac{1}{2}\left(q'dp'-p'dq'\right)\\
 & = & \int_{0}^{1}\frac{1}{2}\left(\left(tq+b_{q}\right)p-\left(tp+b_{p}\right)q\right)dt=\frac{1}{2}\left(b_{q}p-b_{p}q\right)=\frac{1}{2}\omega\left(b,x\right)
\end{eqnarray*}
Finally for the affine map $f\left(x\right)=Bx+b=\left(f_{1}\circ f_{2}\right)\left(x\right)$,
the action is 
\[
\mathcal{A}_{f}=\mathcal{A}_{f_{1}}+\mathcal{A}_{f_{2}}\circ f_{2}^{1}=\mathcal{A}_{f_{1}}=\frac{1}{2}\omega\left(b,x\right).
\]

\end{proof}
We next consider the lift of the operator $\mathcal{L}_{f}$, Eq.(\ref{eq:F_f_affine}),
with respect to the Bargmann transform $\Bargmann_{\hbar}$. Following
the idea explained in Subsection \ref{sub:Semiclassical-description-of-prequantum_op},
we express it with respect to the coordinates $\left(\nu,\zeta\right)$
introduced in Proposition \ref{prop:Normal-coordinates.}. Then, in
the next Lemma, we will obtain an expression of $\mathcal{L}_{f}$
as a tensor product of two operators: each of the two operators is
associated to the dynamics of the canonical map $F=^{t}Df^{-1}:T^{*}\real^{2d}\to T^{*}\real^{2d}$
of $\mathcal{L}_{f}$ in the directions along and orthogonal to the
trapped set $K$, defined by
\[
K=\{(x,\xi)\in\real^{2d}\mid\zeta=0\}
\]
in the simple setting we are considering. (See Proposition \ref{prop:Normal-coordinates.}.)

Let us write the change of variable given in Proposition \ref{prop:Normal-coordinates.}
as 
\begin{equation}
\Phi:\left(\underbrace{q,p}_{x},\underbrace{\xi_{q},\xi_{p}}_{\xi}\right)\in\mathbb{R}^{2d}\oplus\mathbb{R}^{2d}\rightarrow\left(\underbrace{\nu_{q},\nu_{p}}_{\nu},\underbrace{\zeta_{p},\zeta_{q}}_{\zeta}\right)\in\mathbb{R}^{2d}\oplus\mathbb{R}^{2d}.\label{eq:def_phi_change}
\end{equation}
It maps the standard symplectic form $\Omega_{0}=dx\wedge d\xi$ on
$\mathbb{R}^{2d}\oplus\mathbb{R}^{2d}$ to 
\[
(D\Phi^{*})^{-1}(\Omega_{0})=d\nu_{q}\wedge d\nu_{p}+d\zeta_{p}\wedge d\zeta_{q}
\]
and the metric 
\begin{equation}
g_{0}=\frac{1}{2}dx^{2}+2d\xi^{2}\label{eq:metric_g0}
\end{equation}
 on $\real^{2d}\oplus\mathbb{R}^{2d}$ ($g_{0}$ is the metric induced
by $g$ on $T^{*}\mathbb{R}^{2d}$ as explained in Section \ref{sub:Compatible-metrics-and})
to the standard Euclidean metric on $\mathbb{R}^{2d}\oplus\mathbb{R}^{2d}$:
\[
(D\Phi^{*})^{-1}(g_{0})=d\nu^{2}+d\zeta^{2}.
\]

\begin{rem}
With the choice of metric $g_{0}$, in (\ref{eq:metric_g0}), the
linear subsets $K$ and $\left(K^{\perp_{\omega}}\right)$ are $\Omega_{0}$-symplectic
orthogonal but are also $g_{0}-$orthogonal.

The unitary operator associated to the coordinate change $\Phi$ is
defined as 
\[
\Phi^{*}:L^{2}\left(\real_{\nu}^{2d}\oplus\real_{\zeta}^{2d}\right)\to L^{2}\left(\real_{x}^{2d}\oplus\real_{\xi}^{2d}\right),\quad\left(\Phi^{*}u\right):=u\circ\Phi.
\]
Here (and henceforth) we make the convention that the subscript in
the notation such as $\mathbb{R}_{\nu}^{2d}$ indicates the name of
the coordinates on the space. 
\end{rem}
We define the operators 
\[
\mathcal{B}_{\nu_{q}}:L^{2}\left(\mathbb{R}_{\nu_{q}}^{d}\right)\rightarrow L^{2}\left(\mathbb{R}_{(\nu_{q},\nu_{p})}^{2d}\right)\quad\mbox{and }\quad\mathcal{B}_{\nu_{q}}^{*}:L^{2}\left(\mathbb{R}_{(\nu_{q},\nu_{p})}^{2d}\right)\to L^{2}\left(\mathbb{R}_{\nu_{q}}^{d}\right)
\]
as the Bargmann transform $\mathcal{B}_{\hbar}$ and its adjoint $\mathcal{B}_{\hbar}^{*}$
in Subsection \ref{sub:The-Bargmann-transform} for the case $D=d$.
We define 

\[
\mathcal{B}_{\zeta_{p}}:L^{2}\left(\mathbb{R}_{\zeta_{p}}^{d}\right)\rightarrow L^{2}\left(\mathbb{R}_{(\zeta_{p},\zeta_{q})}^{2d}\right)\quad\mbox{and }\quad\mathcal{B}_{\zeta_{p}}^{*}:L^{2}\left(\mathbb{R}_{(\zeta_{p},\zeta_{q})}^{2d}\right)\to L^{2}\left(\mathbb{R}_{\zeta_{p}}^{d}\right)
\]
similarly. Suppose that $\BargmannP_{\nu_{q}}$ and $\BargmannP_{\zeta_{p}}$
are defined correspondingly: That is to say, with setting $D=d$,
we define
\[
\Bargmann_{\nu_{q}}=\Bargmann_{\zeta_{p}}=\Bargmann_{\hbar},\quad\Bargmann_{\nu_{q}}^{*}=\Bargmann_{\zeta_{p}}^{*}=\Bargmann_{\hbar}^{*}\mbox{ and }\BargmannP_{\nu_{q}}=\BargmannP_{\zeta_{p}}=\BargmannP_{\hbar}.
\]
Next we define the operators 
\[
\mathcal{B}_{x}:L^{2}\left(\mathbb{R}_{x}^{2d}\right)\rightarrow L^{2}\left(\mathbb{R}_{\left(x,\xi\right)}^{4d}\right)\quad\mbox{and }\quad\mathcal{B}_{x}^{*}:L^{2}\left(\mathbb{R}_{\left(x,\xi\right)}^{4d}\right)\to L^{2}\left(\mathbb{R}_{x}^{2d}\right)
\]
by 
\begin{equation}
\mathcal{B}_{x}:=\tilde{\sigma}^{-1}\circ\mathcal{B}_{\hbar}\circ\sigma\qquad\mbox{and }\quad\mathcal{B}_{x}^{*}:=\sigma^{-1}\circ\mathcal{B}_{\hbar}^{*}\circ\tilde{\sigma}\label{eq:Bargmann_modified}
\end{equation}
where $\mathcal{B}_{\hbar}$ and $\mathcal{B}_{\hbar}^{*}$ are now
those in the case $D=2d$, and 
\[
\sigma:L^{2}(\real_{x}^{2d})\to L^{2}(\real_{x}^{2d})\quad\mbox{and}\quad\tilde{\sigma}:L^{2}(\real_{(x,\xi)}^{4d})\to L^{2}(\real_{(x,\xi)}^{4d})
\]
are simple unitary operators defined by 
\[
\sigma u(x)=2^{-d}u(2^{-1/2}x)\quad\mbox{and }\quad\tilde{\sigma}v(x,\xi)=v(2^{-1/2}x,2^{1/2}\xi)
\]

introduced in relation to the additional factor $1/2$ in (\ref{eq:metric_g0}).

Correspondingly we set 
\begin{equation}
\BargmannP_{x}=\Bargmann_{x}\circ\Bargmann_{x}^{*}=\tilde{\sigma}^{-1}\circ\BargmannP_{\hbar}\circ\tilde{\sigma}.\label{eq:BargmannP_modified}
\end{equation}

\begin{rem}
\label{rem4.4}
\begin{enumerate}
\item In terms of the (generalized) Bargmann transforms considered in Subsection
\ref{sub:Bargmann_transform_generalized}, the operators $\mathcal{B}_{x}$
and $\mathcal{B}_{x}^{*}$ are the Bargmann transform and its adjoint
for the combination of the Euclidean space $E=\real^{4d}$, the standard
symplectic form $\Omega_{0}=dx\wedge d\xi$, the metric $g_{0}=\frac{1}{2}dx^{2}+2d\xi^{2}$
and the Lagrangian subspace $\real_{x}^{2d}$. Since the metric $g_{0}$
corresponds to the standard Euclidean metric through $\Phi$, the
definition of the operators $\Bargmann_{x}$, $\Bargmann_{x}^{*}$
and $\BargmannP_{x}$ above is more convenient and natural for our
argument than those without $\sigma$ and $\tilde{\sigma}$. Indeed
we have 
\begin{equation}
\BargmannP_{x}^{*}\circ\Phi^{*}=\Phi^{*}\circ(\BargmannP_{\nu_{q}}\otimes\BargmannP_{\zeta_{p}})^{*}.\label{eq:relation_between_Bargmann_Projectors}
\end{equation}

\item In the notation introduced above, the subscripts indicate the related
coordinates. Though this may deviate from the standard usage, it is
convenient for our argument. Notice that the operators introduced
above, such as $\Bargmann_{x}$, depend on the parameter $\hbar$
(and hence on $N$). 
\end{enumerate}
\end{rem}
\begin{lem}
\label{prop:affine_Lf}Let $\mathcal{L}_{f}$ be the prequantum transfer
operator (\ref{eq:F_f_affine}) associated to a symplectic affine
map in (\ref{eq:affine_map}). Then the following diagram commutes:
\begin{equation}
\begin{CD}L^{2}\left(\real_{x}^{2d}\right)@>\mathcal{L}_{f}>>L^{2}\left(\real_{x}^{2d}\right)\\
@AA\mathcal{U}A@AA\mathcal{U}A\\
L^{2}\left(\real_{\nu_{q}}^{d}\right)\otimes L^{2}\left(\real_{\zeta_{p}}^{d}\right)@>M_{\nu}\left(f\right)\otimes M_{\zeta}\left(B\right)>>L^{2}\left(\real_{\nu_{q}}^{d}\right)\otimes L^{2}\left(\real_{\zeta_{p}}^{d}\right)
\end{CD}\label{eq:L_f_conjugated}
\end{equation}
where $\mathcal{U}$, $M_{\nu}(f)$ and $M_{\zeta}(B)$ are the unitary
operators defined respectively by
\begin{alignat}{2}
 & \mathcal{U}:L^{2}\left(\real_{\nu_{q}}^{d}\right)\otimes L^{2}\left(\real_{\zeta_{p}}^{d}\right)\to L^{2}\left(\real_{x}^{2d}\right), & \quad & \mathcal{U=}\mathcal{B}_{x}^{*}\circ\Phi^{*}\circ\left(\mathcal{B}_{\nu_{q}}\otimes\mathcal{B}_{\zeta_{p}}\right),\label{eq:def_U_operator}\\
 & M_{\nu}\left(f\right):L\left(\real_{\nu_{q}}^{d}\right)\rightarrow L^{2}\left(\real_{\nu_{q}}^{d}\right), &  & M_{\nu}\left(f\right)=\sqrt{d\left(B\right)}\cdot\mathcal{B}_{\nu_{q}}^{*}\circ(e^{\frac{i}{2\hbar}\omega\left(\nu,b\right)}\cdot L_{f})\circ\mathcal{B}_{\nu_{q}},\nonumber \\
 & M_{\zeta}\left(B\right):L^{2}\left(\real_{\zeta_{p}}^{d}\right)\rightarrow L^{2}\left(\real_{\zeta_{p}}^{d}\right), &  & M_{\zeta}\left(B\right)=\sqrt{d\left(B\right)}\cdot\mathcal{B}_{\zeta_{p}}^{*}\circ L_{B}\circ\mathcal{B}_{\zeta_{p}}\nonumber 
\end{alignat}
with $d\left(B\right)=\det\left(\left(1+{}^{t}B\cdot B\right)/2\right)^{1/2}$,
$\left(L_{f}u\right)\left(\nu\right):=\left(u\circ f^{-1}\right)\left(\nu\right)$
and $L_{B}u=u\circ B^{-1}$ as before. Equivalently, in terms of lifted
operators, we have the following commuting diagram:
\begin{equation}
\begin{CD}L^{2}\left(\real_{x}^{2d}\oplus\real_{\xi}^{2d}\right)@>\prequantumLlift_{f}>>L^{2}\left(\real_{x}^{2d}\oplus\real_{\xi}^{2d}\right)\\
@AA\Phi^{*}A@AA\Phi^{*}A\\
L^{2}\left(\real_{\nu}^{2d}\right)\otimes L^{2}\left(\real_{\zeta}^{2d}\right)@>M_{\nu}^{\mathrm{lift}}\left(f\right)\otimes M_{\zeta}^{\mathrm{lift}}\left(B\right)>>L^{2}\left(\real_{\nu}^{2d}\right)\otimes L^{2}\left(\real_{\zeta}^{2d}\right)
\end{CD}\label{eq:expression_as_tensorproduct}
\end{equation}
 \end{lem}
\begin{proof}
Recall the operators $\flat:\real^{2d}\to(\real^{2d})^{*}=\real^{2d}$
and $\sharp:(\real^{2d})^{*}=\real^{2d}\to\real^{2d}$ introduced
in Remark \ref{Def:sharp_flat_operator}. The expression (\ref{eq:F_f_affine})
shows that $\mathcal{L}_{f}$ can be written $\mathcal{L}_{f}=T_{\left(b,-\frac{1}{2}b^{\flat}\right)}\circ L_{B}$
where the unitary transfer operator $T_{\left(b,-\frac{1}{2}b^{\flat}\right)}$
is defined in (\ref{eq:def_LA-1}). We apply Lemma \ref{lm:lift_of_LA-1},
Lemma \ref{lm:lift_of_LA} and corollary \ref{cor:4.11} to obtain

\begin{eqnarray}
\mathcal{L}_{f} & = & \Bargmann_{x}^{*}\circ\mathcal{T}_{\left(b,-\frac{1}{2}b^{\flat}\right)}\circ\Bargmann_{x}\circ(d(B)\cdot\Bargmann_{x}^{*})\circ L_{B\oplus{}^{t}B^{-1}}\circ\Bargmann_{x}\nonumber \\
 & = & d(B)\cdot\Bargmann_{x}^{*}\circ\mathcal{T}_{\left(b,-\frac{1}{2}b^{\flat}\right)}\circ\Bargmann_{x}\circ\Bargmann_{x}^{*}\circ L_{B\oplus{}^{t}B^{-1}}\circ\Bargmann_{x}\nonumber \\
 & = & d(B)\cdot\Bargmann_{x}^{*}\circ(e^{\frac{i}{2\hbar}\varphi}\cdot L_{F})\circ\Bargmann_{x}\label{eq:calL_f_for_affine_map}
\end{eqnarray}
with $F\left(x,\xi\right)=\left(Bx+b,^{t}B^{-1}\xi-\frac{1}{2}b^{\flat}\right)$
and $\varphi(x,\xi)=-\frac{1}{2}b^{\flat}\cdot x-b\cdot\xi$. Since
we have $^{t}B^{-1}=\flat\circ B\mathbf{\circ}\flat^{-1}$ for $B$
symplectic, we get the following expression of $F$ in the new coordinates
$\left(\nu,\zeta\right)$: 
\[
\left(\Phi\circ F\circ\Phi^{-1}\right)\left(\nu,\zeta\right)=\left(B\nu+b,^{t}B^{-1}\zeta\right)=\left(f\left(\nu\right),^{t}Df^{-1}\zeta\right).
\]
This implies
\[
L_{F}=\Phi^{*}\circ\left(L_{f}\otimes L_{B}\right)\circ\left(\Phi^{*}\right)^{-1}.
\]
From (\ref{eq:relation_between_Bargmann_Projectors}), we have 

\[
\Bargmann_{x}^{*}\circ\Phi^{*}=\Bargmann_{x}^{*}\circ\BargmannP_{x}\circ\Phi^{*}=\Bargmann_{x}^{*}\circ\Phi^{*}\circ(\BargmannP_{\nu_{q}}\otimes\BargmannP_{\zeta_{p}})=\mathcal{U}\circ(\mathcal{B}_{\nu_{q}}\otimes\mathcal{B}_{\zeta_{p}})^{*}
\]
and 
\[
(\Phi^{*})^{-1}\circ\Bargmann_{x}=(\Phi^{*})^{-1}\circ\BargmannP_{x}\circ\Bargmann_{x}=(\BargmannP_{\nu_{q}}\otimes\BargmannP_{\zeta_{p}})\circ(\Phi^{*})^{-1}\circ\Bargmann_{x}=(\mathcal{B}_{\nu_{q}}\otimes\mathcal{B}_{\zeta_{p}})\circ\mathcal{U}^{-1}.
\]
Using these relations to continue (\ref{eq:calL_f_for_affine_map})
and noting that $\varphi\left(x,\xi\right)=\omega(\nu,b)$, we conclude 

\begin{align*}
\mathcal{L}_{f} & =d(B)\cdot\Bargmann_{x}^{*}\circ e^{\frac{i}{2\hbar}\varphi}L_{F}\circ\Bargmann_{x}\\
 & =d(B)\cdot\Bargmann_{x}^{*}\circ\Phi^{*}\circ\left(e^{\frac{i}{2\hbar}\varphi}L_{f}\otimes L_{B}\right)\circ\left(\Phi^{*}\right)^{-1}\circ\Bargmann_{x}\\
 & =d(B)\cdot\mathcal{U}\circ\left(\left(\Bargmann{}_{\nu_{q}}^{*}\circ(e^{\frac{i}{2\hbar}\varphi}\cdot L_{f})\circ\mathcal{B}_{\nu_{q}}\right)\otimes\left(\mathcal{B}_{\zeta_{p}}^{*}\circ L_{B}\circ\mathcal{B}_{\zeta_{p}}^{*}\right)\right)\circ\mathcal{U}^{-1}\\
 & =\mathcal{U}\circ\left(M_{\nu}\left(f\right)\otimes M_{\zeta}\left(B\right)\right)\circ\mathcal{U}^{-1}.
\end{align*}

\end{proof}

\paragraph{}
\begin{rem}
Eq.(\ref{eq:L_f_conjugated}) shows that $\mathcal{L}_{f}$ is conjugated
to the product of two operators $M_{\nu}\left(f\right)\otimes M_{\zeta}\left(B\right)$.
This is remarkable but it is due here to the fact that $Df$ is constant.
Each of these operators is usually called a metaplectic operator,
we refer to \cite{folland-88}. Here we have used a derivation in
phase space using the Bargmann transform.
\end{rem}

\subsection{\label{sub:The-prequantum-transfer_hyperbolic}The prequantum transfer
operator for a linear hyperbolic map}

\label{ss:action_hyperbolic_linear} In this subsection, we restrict
the argument in the last subsection to the case where $f$ in (\ref{eq:affine_map})
is hyperbolic in the sense that $f$ is expanding in $\real^{d}\oplus\{0\}$
while contracting in $\{0\}\oplus\real^{d}$ and is linear, i.e. $b=0$.
Since $f$ preserves the symplectic form $\omega$, we may express
it as (see last remark in the proof of Proposition \ref{prop:DF_ro})
\begin{equation}
f\left(q,p\right)=B(q,p)=\left(Aq,{}^{t}A^{-1}p\right)\qquad\mbox{where }B=\left(\begin{array}{cc}
A & 0\\
0 & ^{t}A^{-1}
\end{array}\right)\label{eq:hyperbolic_f}
\end{equation}
with $A:\real^{d}\to\real^{d}$ an expanding linear map satisfying
$\|A^{-1}\|\le1/\lambda$ for some $\lambda>1$. Notice that, since
$b=0$, the action $\mathcal{A}$ vanishes in (\ref{eq:F_f_affine})
and the prequantum transfer operator gets the simpler expression:
\begin{equation}
\left(\mathcal{L}_{f}u\right)\left(x\right)=u\left(B^{-1}x\right)=L_{B}u\left(x\right).\text{}\label{eq:L_f_hyperbolic}
\end{equation}

The next proposition is deduced from Proposition \ref{prop:affine_Lf}.
\begin{prop}
\label{prop:4.5}The following diagram commutes:
\begin{equation}
\begin{CD}L^{2}\left(\real_{x}^{2d}\right)@>\mathcal{L}_{f}>>L^{2}\left(\real_{x}^{2d}\right)\\
@AA\mathcal{U}A@AA\mathcal{U}A\\
L^{2}\left(\real_{\nu_{q}}^{d}\right)\otimes L^{2}\left(\real_{\zeta_{p}}^{d}\right)@>U_{A}\otimes U_{A}>>L^{2}\left(\real_{\nu_{q}}^{d}\right)\otimes L^{2}\left(\real_{\zeta_{p}}^{d}\right)
\end{CD}\label{eq:expression_tensorproduct-2}
\end{equation}
with the unitary operator $\mathcal{U}$ defined in (\ref{eq:def_U_operator})
and 
\[
U_{A}:=\frac{1}{\sqrt{\left|\mathrm{det}A\right|}}L_{A}
\]
which is unitary in $L^{2}\left(\mathbb{R}^{d}\right)$. Equivalently
using lifted operators, expressed in Lemma \ref{lm:lift_of_LA}, we
have the following commuting diagram: 
\begin{equation}
\begin{CD}L^{2}\left(\real_{x}^{2d}\oplus\real_{\xi}^{2d}\right)@>\prequantumLlift_{f}>>L^{2}\left(\real_{x}^{2d}\oplus\real_{\xi}^{2d}\right)\\
@AA\Phi^{*}A@AA\Phi^{*}A\\
L^{2}\left(\real_{\nu}^{2d}\right)\otimes L^{2}\left(\real_{\zeta}^{2d}\right)@>U_{A}^{\mathrm{lift}}\otimes U_{A}^{\mathrm{lift}}>>L^{2}\left(\real_{\nu}^{2d}\right)\otimes L^{2}\left(\real_{\zeta}^{2d}\right)
\end{CD}\label{cd:expression_as_tensorproduct}
\end{equation}
with $U_{A}^{\mathrm{lift}}:=\frac{1}{\sqrt{\left|\mathrm{det}A\right|}}L_{A}^{\mathrm{lift}}$.\end{prop}
\begin{proof}
We have 
\begin{eqnarray*}
d(B) & = & \det\left(\left(1+{}^{t}B\cdot B\right)/2\right){}^{1/2}=\mbox{det}\left(\frac{1}{2}\left(B+\,^{t}B^{-1}\right)\right)^{1/2}\\
 & = & \mbox{det}\left(\frac{1}{2}\left(A+\,^{t}A^{-1}\right)\right)=\left(\mbox{det}A\right)^{-1}
\end{eqnarray*}
and $L_{f}=L_{B}=L_{A\oplus^{t}A^{-1}}$. Hence, by the expression
(\ref{eq:claim}), we get
\[
M_{\nu}\left(f\right):=\sqrt{d\left(B\right)}\cdot\mathcal{B}_{\nu_{q}}^{*}\circ(e^{\frac{i}{2\hbar}\omega\left(\nu,b\right)}\cdot L_{f})\circ\mathcal{B}_{\nu_{q}}=\left(\mbox{det}A\right)^{-1/2}d\left(A\right)\cdot\mathcal{B}_{\nu_{q}}^{*}\circ L_{A\oplus^{t}A^{-1}}\circ\mathcal{B}_{\nu_{q}}=U_{A}
\]
and

\[
M_{\zeta}\left(B\right):=\sqrt{d\left(B\right)}\cdot\mathcal{B}_{\zeta_{p}}^{*}\circ L_{B}\circ\mathcal{B}_{\zeta_{p}}=\left(\mbox{det}A\right)^{-1/2}d\left(A\right)\cdot\mathcal{B}_{\zeta_{p}}^{*}\circ L_{A\oplus^{t}A^{-1}}\circ\mathcal{B}_{\zeta_{p}}=U_{A}.
\]
Putting these in Proposition \ref{prop:affine_Lf}, we obtain the
conclusion. 
\end{proof}

\subsection{Anisotropic Sobolev space}

\label{ss:anisoSob} In order to observe a discrete spectrum of resonances
of the prequantum operator $\mathcal{L}_{f}$, we have to consider
the action of $\mathcal{L}_{f}$ on an appropriate spaces of functions.
As we explained in subsection \ref{sub:2.1.3}, we define such space
of function, called anisotropic Sobolev space, by changing the norm
in the directions transverse to the trapped set $K$ (that is, in
the directions of the variables $\zeta$ ). Below is the precise definition. 
\begin{defn}
\label{def:escape_function_Hr}We define the \emph{escape function}
or \emph{weight function} 
\[
\mathcal{W}_{\hbar}^{r}:\real_{x}^{2d}\oplus\real_{\xi}^{2d}\to\real_{+}\quad\mbox{and}\quad\mathcal{W}_{\hbar}^{r,\pm}:\real_{x}^{2d}\oplus\real_{\xi}^{2d}\to\real_{+}
\]
 by

\begin{equation}
\mathcal{W}_{\hbar}^{r}\left(x,\xi\right):=W_{\hbar}^{r}\left(\zeta_{p},\zeta_{q}\right)\quad\mbox{and}\quad\mathcal{W}^{r,\pm}\left(x,\xi\right):=W_{\hbar}^{r}\left(\zeta_{p},\zeta_{q}\right)\label{eq:def_W_W+}
\end{equation}
where the functions $W_{\hbar}^{r}$ and $W_{\hbar}^{r,\pm}$ are
defined in Definition \ref{def:Wrh}, and $\left(\zeta_{p},\zeta_{q}\right)$
is part of the coordinates introduced in (\ref{eq:def_phi_change}).
The \emph{anisotropic Sobolev space}\textbf{ $\mathcal{H}_{\hbar}^{r}\left(\real^{2d}\right)$}
is the Hilbert space obtained as the completion of the Schwartz space
$\mathcal{S}(\real^{2d}$) with respect to the norm 
\[
\|u\|_{\mathcal{H}_{\hbar}^{r}}:=\|\mathcal{W}_{\hbar}^{r}\cdot\Bargmann_{x}u\|_{L^{2}}
\]
where $\Bargmann_{x}$ is the operator defined in (\ref{eq:Bargmann_modified}).
Similarly, let $\mathcal{H}_{\hbar}^{r,\pm}(\real^{2d})$ be the Hilbert
space defined in the parallel manner by replacing $\mathcal{W}_{\hbar}^{r}(\cdot)$
by $\mathcal{W}_{\hbar}^{r,\pm}(\cdot)$. 
\end{defn}
By definition, the operator $\Bargmann_{x}$ extends to an isometric
embedding 
\[
\Bargmann_{x}:\mathcal{H}_{\hbar}^{r}(\real_{x}^{2d})\to L^{2}(\real_{x}^{2d}\oplus\real_{\xi}^{2d},(\mathcal{W}_{\hbar}^{r})^{2}).
\]
Since the weight function $\mathcal{W}^{r}(\cdot)$ can be expressed
as 
\[
\mathcal{W}_{\hbar}^{r}=(1\otimes W_{\hbar}^{r})\circ\Phi,
\]
where $\Phi$ is given in (\ref{eq:def_phi_change}), we see that
the following diagram commutes:
\begin{equation}
\begin{CD}\mathcal{H}_{\hbar}^{r}\left(\real_{x}^{2d}\right)@>\mathcal{L}_{f}>>\mathcal{H}_{\hbar}^{r}\left(\real_{x}^{2d}\right)\\
@VV\mathcal{B}_{x}V@VV\mathcal{B}_{x}V\\
L^{2}\left(\real_{x}^{2d}\oplus\real_{\xi}^{2d},\left(\mathcal{W}_{\hbar}^{r}\right)^{2}\right)@>\prequantumLlift>>L^{2}\left(\real_{x}^{2d}\oplus\real_{\xi}^{2d},\left(\mathcal{W}_{\hbar}^{r}\right)^{2}\right)\\
@AA\Phi^{*}A@AA\Phi^{*}A\\
L^{2}\left(\real_{\nu}^{2d}\right)\otimes L^{2}\left(\real_{\zeta}^{2d},\left(W_{\hbar}^{r}\right)^{2}\right)@>U_{A}^{\mathrm{lift}}\otimes U_{A}^{\mathrm{lift}}>>L^{2}\left(\real_{\nu}^{2d}\right)\otimes L^{2}\left(\real_{\zeta}^{2d},\left(W_{\hbar}^{r}\right)^{2}\right)\\
@AA\mathcal{B}_{\nu_{q}}\otimes\mathcal{B}_{\zeta_{p}}A@AA\mathcal{B}_{\nu_{q}}\otimes\mathcal{B}_{\zeta_{p}}A\\
L^{2}\left(\mathbb{R}_{\nu_{q}}^{2d}\right)\otimes H_{\hbar}^{r}\left(\real_{\zeta_{p}}^{d}\right)@>U_{A}\otimes U_{A}>>L^{2}\left(\mathbb{R}_{\nu_{q}}^{2d}\right)\otimes H_{\hbar}^{r}\left(\real_{\zeta_{p}}^{d}\right)
\end{CD}\label{eq:Long_diagram}
\end{equation}
where $\Phi^{*}$ is an isomorphism between Hilbert spaces. Further
(\ref{eq:relation_between_Bargmann_Projectors}) is an isomorphism
between the image of$\mathcal{B}_{x}$ and the image of $\mathcal{B}_{\nu_{q}}\otimes\mathcal{B}_{\zeta_{p}}$.
$H_{\hbar}^{r}\left(\real_{\zeta_{p}}^{d}\right)$ in the last line
is defined in Definition \ref{Def:H^r_h}. Hence, skipping the lines
in the middle, we get:

\begin{equation}
\begin{CD}\mathcal{H}_{\hbar}^{r}\left(\real_{x}^{2d}\right)@>\mathcal{L}_{f}>>\mathcal{H}_{\hbar}^{r}\left(\real_{x}^{2d}\right)\\
@AA\mathcal{U}A@AA\mathcal{U}A\\
L^{2}\left(\mathbb{R}_{\nu_{q}}^{2}\right)\otimes H_{\hbar}^{r}\left(\real_{\zeta_{p}}^{d}\right)@>U_{A}\otimes U_{A}>>L^{2}\left(\mathbb{R}_{\nu_{q}}^{2}\right)\otimes H_{\hbar}^{r}\left(\real_{\zeta_{p}}^{d}\right)
\end{CD}\label{cd:lift_of_Lf}
\end{equation}
with $\mathcal{U}$ the unitary operator defined in (\ref{eq:def_U_operator}).
For the operator $U_{A}\otimes U_{A}$ on the bottom line, we know
that the operator $U_{A}=\frac{1}{\sqrt{\left|\mathrm{det}A\right|}}L_{A}:L^{2}\left(\mathbb{R}_{\nu_{q}}^{2d}\right)\rightarrow L^{2}\left(\mathbb{R}_{\nu_{q}}^{2d}\right)$
is unitary and Proposition \ref{pp:structure_of_hLA} gives a description
on the spectral structure of the operator $L_{A}:H_{\hbar}^{r}\left(\real_{\zeta_{p}}^{d}\right)\rightarrow H_{\hbar}^{r}\left(\real_{\zeta_{p}}^{d}\right)$.
Therefore we obtain the next proposition as a consequence. We fix
some integer $n\ge0$ and assume 
\begin{equation}
r>n+2d.\label{eq:choice_of_r_2}
\end{equation}
This assumption on $r$ corresponds to (\ref{eq:assumption_on_r})
in the last section. 
\begin{defn}
For $0\le k\le n$, we consider the projection operators
\begin{equation}
\romet_{\hbar}^{(k)}:=\mathcal{U}\circ\left(\mbox{Id}\otimes T^{\left(k\right)}\right)\circ\mathcal{U}^{-1}\quad:\mathcal{H}_{\hbar}^{r}\left(\real_{x}^{2d}\right)\to\mathcal{H}_{\hbar}^{r}\left(\real_{x}^{2d}\right)\label{eq:def_t}
\end{equation}
and 
\begin{equation}
\tilde{\romet}_{\hbar}:=\mathrm{Id}-\sum_{k=0}^{n}t_{\hbar}^{(k)}=\mathcal{U}\circ\left(\mbox{Id}\otimes\tilde{T}\right)\circ\mathcal{U}^{-1}\quad:\mathcal{H}_{\hbar}^{r}\left(\real_{x}^{2d}\right)\to\mathcal{H}_{\hbar}^{r}\left(\real_{x}^{2d}\right)\label{eq:def_tt}
\end{equation}
where $T^{\left(k\right)}$ and $\widetilde{T}$ are the projection
operators introduced in (\ref{eq:def_Tk}) and (\ref{eq:def_T_tilde})
respectively.
\end{defn}
\begin{framed}%
\begin{prop}
\label{prop:prequatum_op_for_hyp_linear} The operators $\romet_{\hbar}^{(k)}$,
$0\le k\le n$, and $\tilde{\romet}_{\hbar}$ defined in (\ref{eq:def_t}),
(\ref{eq:def_tt}), form a complete set of mutually commutative projection
operators on $\mathcal{H}_{\hbar}^{r}\left(\real_{x}^{2d}\right)$.
These operators also commute with the prequantum transfer operator
$\prequantumL_{f}$ defined in (\ref{eq:L_f_hyperbolic}). Consequently
the space $\mathcal{H}_{\hbar}^{r}\left(\real_{x}^{2d}\right)$ has
a decomposition invariant under the action of $\mathcal{L}_{f}$:
\[
\mathcal{H}_{\hbar}^{r}\left(\real_{x}^{2d}\right)=H'_{0}\oplus H'_{1}\oplus\cdots\oplus H'_{n}\oplus\widetilde{H}'\qquad\mbox{where \ensuremath{H'_{k}=\mathrm{Im}\,\romet_{\hbar}^{(k)}}and \ensuremath{\widetilde{H}'=\mathrm{Im}\,\tilde{\romet}_{\hbar}}}
\]
For this decomposition we have
\begin{enumerate}
\item For every $0\le k\le n$, we have a commuting diagram
\[
\begin{CD}H'_{k}@>\mathcal{L}_{f}>>H'_{k}\\
@AA\mathcal{U}A@AA\mathcal{U}A\\
L^{2}\left(\mathbb{R}_{\nu_{q}}^{2}\right)\otimes\Polynomial^{(k)}@>U_{A}\otimes U_{A}^{\left(k\right)}>>L^{2}\left(\mathbb{R}_{\nu_{q}}^{2}\right)\otimes\Polynomial^{(k)}
\end{CD}
\]
with $U_{A}=\frac{1}{\sqrt{\left|\mathrm{det}A\right|}}L_{A}:L^{2}\left(\mathbb{R}_{\nu_{q}}^{2d}\right)\rightarrow L^{2}\left(\mathbb{R}_{\nu_{q}}^{2d}\right)$
unitary and $U_{A}^{\left(k\right)}:=\frac{1}{\sqrt{\left|\mathrm{det}A\right|}}L_{A}:\Polynomial^{(k)}\to\Polynomial^{(k)}$
finite rank. From (\ref{eq:bound_on_norm-1}) we have for every $u\in H'_{k}$,
\[
C_{0}^{-1}\|A\|_{\mathrm{max}}^{-k}\cdot|\det(A)|^{-1/2}\cdot\|u\|_{\mathcal{H}_{\hbar}^{r}}\le\|\prequantumL_{f}u\|_{\mathcal{H}_{\hbar}^{r}}\le C_{0}\|A\|_{\mathrm{min}}^{-k}\cdot|\det(A)|^{-1/2}\cdot\|u\|_{\mathcal{H}_{\hbar}^{r}}.
\]

\item The operator norm of $\prequantumL_{f}:\widetilde{H}'\to\widetilde{H}'$
is bounded by 
\[
C_{0}\cdot\max\{\|A\|_{\mathrm{min}}^{-n-1}\cdot|\det A|^{-1/2},\|A\|_{\mathrm{min}}^{-r}\cdot|\det A|^{1/2}\}.
\]
 
\end{enumerate}
The constant $C_{0}$ is independent of $A$ and $\hbar$. \end{prop}
\end{framed}

\subsection{A few technical lemmas}

\label{ss:technical_lemmas} In this subsection, we collect a few
miscellaneous technical lemmas related to the anisotropic Sobolev
spaces $\mathcal{H}_{\hbar}^{r}(\real^{2d})$ and $\mathcal{H}_{\hbar}^{r,\pm}(\real^{2d})$,
which we will use in the later sections. The following are immediate
consequences of Lemma \ref{lm:LA_bdd_L2W} and \ref{lm:Tn_bdd_L2W}
respectively. 
\begin{lem}
\label{lm:boundedness_of_hyperbolic_linear_map} Suppose that $f$
is a hyperbolic linear transformation (\ref{eq:hyperbolic_f}) defined
for an expanding linear map $A$ satisfying (\ref{eq:A_is_expanding})
for some large $\lambda$ (say $\lambda>9$). Then the operator $\prequantumL_{f}$,
(\ref{eq:L_f_hyperbolic}), extends to a bounded operator $\prequantumL_{f}:\mathcal{H}_{\hbar}^{r,-}(\real^{2d})\to\mathcal{H}_{\hbar}^{r,+}(\real^{2d})$
and the operator norm is bounded by a constant independent of $\hbar>0$
and $f$. \end{lem}
\begin{cor}
\label{lm:pik_rpm} The operator $\romet_{\hbar}^{(k)}$ for $0\le k\le n$,
defined in (\ref{eq:def_t}), extends to a bounded operator 
\[
\romet_{\hbar}^{(k)}:\mathcal{H}_{\hbar}^{r,-}(\real^{2d})\to\mathcal{H}_{\hbar}^{r,+}(\real^{2d})
\]
 whose operator norm is bounded by a constant independent of $\hbar$.
\end{cor}
For convenience in the later argument, let us put the following definition:
\begin{defn}
\label{def:A_tilde}Let $\mathcal{A}$ be the group of affine transformation
on $\real^{2d}$ that preserves the symplectic form $\omega$, the
Euclidean norm and its derivative preserves the splitting $\real^{2d}=\real^{d}\oplus\real^{d}$
simultaneously. 
\end{defn}
The function $\mathcal{W}_{\hbar}^{r}$, (\ref{eq:def_W_W+}), is
invariant with respect to the transformation $\left(a,^{t}Da^{-1}\right)$
on $T^{*}\real^{2d}=\real^{2d}\times\real^{2d}$ when $a\in\mathcal{A}$.
This fact, together with Lemma \ref{lm:lift_of_LA}, yields
\begin{lem}
\label{lm:invariance_wrt_translation} If $a\in\mathcal{A}$, then
the prequantum operator $\prequantumL_{a}$, defined in (\ref{eq:F_f_affine})
extends to an isometry on $\mathcal{H}_{\hbar}^{r}(\real^{2d})$.
\end{lem}
The norm $\|\cdot\|_{\mathcal{H}_{\hbar}^{r}}$ on the Hilbert space
$\mathcal{H}_{\hbar}^{r}(\real^{2d})$ is induced by a (unique) inner
product $(\cdot,\cdot)_{\mathcal{H}_{\hbar}^{r}(\real^{2d})}$. Notice
that even if two distributions $u$ and $v$ in $\mathcal{H}_{\hbar}^{r}(\real^{2d})$
have mutually disjoint supports, the inner product $(u,v)_{\mathcal{H}_{\hbar}^{r}(\real^{2d})}$
may not vanish. This is somewhat inconvenient. But we have the following
``\emph{pseudo-local}'' property. We omit the proof because it can
be given by a straightforward estimate.
\begin{lem}
\label{lm:pseudo_local_property} Let $\epsilon>0$. If $d(\supp u,\supp v)\ge\hbar^{(1-\epsilon)/2}$
for $u,v\in\mathcal{H}_{\hbar}^{r}(\real^{2d})$, we have 
\[
|(u,v)_{\mathcal{H}_{\hbar}^{r}(\real^{2d})}|\le C_{\nu,\epsilon}\cdot\hbar^{\nu}\cdot\|u\|_{\mathcal{H}_{\hbar}^{r}(\real^{2d})}\|v\|_{\mathcal{H}_{\hbar}^{r}(\real^{2d})}\quad\mbox{for \ensuremath{u,v\in\mathcal{H}_{\hbar}^{r}(\real^{2d})}}
\]
for arbitrarily large $\nu$, with $C_{\nu,\epsilon}>0$ a constant
depending on $\epsilon$ and $\nu$.
\end{lem}
From the definition of the function $\mathcal{W}_{\hbar}^{r}(\cdot)$
and (\ref{eq:W_order_function}), we have 
\begin{equation}
\mathcal{W}_{\hbar}^{r}(x,\xi)\le C\cdot\mathcal{W}_{\hbar}^{r}(y,\eta)\cdot\langle\hbar^{-1/2}|(x,\xi)-(y,\eta)|\rangle^{2r}.\label{eq:estmate_on_Wr}
\end{equation}
The next lemma and corollary are direct consequences of this estimate.
The proof is completely parallel to that of Lemma \ref{lm:boundedness_of_molifier}. 
\begin{lem}
\label{lm:boundedness_of_molifier2} If $R_{\hbar}:\mathcal{S}(\real^{2d}\oplus\real^{2d})\to\mathcal{S}(\real^{2d}\oplus\real^{2d})$
is an integral operator of the form 
\[
\left(R_{\hbar}u\right)(x,\xi)=\int K_{\hbar}(x,\xi;x',\xi')u(x',\xi')dx'd\xi'
\]
 depending on $\hbar$ and if the kernel $K_{\hbar}(\cdot;\cdot)$
is a continuous function satisfying 
\[
|K_{\hbar}(x,\xi;x',\xi')|\le\langle\hbar^{-1/2}\cdot|(x,\xi)-(x',\xi')|\rangle^{-\nu}
\]
 for some $\nu>2r+4d$, then the operator $R_{\hbar}$ extends uniquely
to a bounded operator on $L^{2}(\real^{2d}\oplus\real^{2d},(\mathcal{W}_{\hbar}^{r})^{2})$
and 
\[
\|R_{\hbar}\|_{L^{2}(\real^{2d}\oplus\real^{2d},(\mathcal{W}_{\hbar}^{r})^{2})}\le C_{\nu}
\]
 where $C_{\nu}$ is a constant independent of $\hbar$. The same
holds true with $\mathcal{W}_{\hbar}^{r}$ replaced by $\mathcal{W}_{\hbar}^{r,\pm}$
simultaneously. \end{lem}
\begin{cor}
\label{cor:BargmannP_bdd} The Bargmann projector $\BargmannP_{\hbar}$
extends uniquely to a bounded operator on $L^{2}(\real^{2d}\oplus\real^{2d},(\mathcal{W}_{\hbar}^{r})^{2})$
and its operator norm is bounded by a constant that does not depend
on $\hbar>0$. The same holds true with $\mathcal{W}_{\hbar}^{r}$
replaced by $\mathcal{W}_{\hbar}^{r,\pm}$ simultaneously.\end{cor}

\section{\label{sec:Nonlinear-prequantum-maps}Nonlinear prequantum maps on
$\real^{2d}$}

In this section, we prepare some basic estimates on the effect of
non-linearity of the Anosov diffeomorphism $f$ on the anisotropic
Sobolev space $\mathcal{H}_{\hbar}^{r}(\real^{2d})$. Most of the
results in this section may be rather obvious at least for those readers
who are familiar with Fourier analysis. However we have to be attentive
to the following particular situations in our argument:
\begin{itemize}
\item The escape function $\mathcal{W}_{\hbar}^{r}$ in the definition of
the anisotropic Sobolev space $\mathcal{H}_{\hbar}^{r}(\real^{2d})$
has variable growth order $m(\cdot)$ depending on the directions.
This leads to the fact that the (prequantum) transfer operator associated
to a non-linear map may be unbounded even if the map is very close
to the identity. 
\item The spectral projection operators $t_{\hbar}^{(k)}$ in Proposition
\ref{prop:prequatum_op_for_hyp_linear} is also very anisotropic and
rather singular. It is not well-defined (bounded) on any usual (isotropic)
Sobolev spaces of positive or negative order. 
\item The escape function $\mathcal{W}_{\hbar}^{r}$ is not very smooth,
viewed at the scale ($\sim\hbar^{1/2})$ of the smallest-possible
wave packets in the phase space. In terms of the theory of pseudodifferential
operators, this implies that the escape function $\mathcal{W}_{\hbar}^{r}$
belongs only to the symbol class of ``critical order''. We have
to avoid carefully the difficulties caused by this fact. 
\end{itemize}
For these reasons, we are going to give the argument to some detail.
The main result in this section is Proposition \ref{lm:L_g_Y_almost_identity},
which concerns the third item above.

Recall that the stable and unstable subspaces $E_{s}(x)$ and $E_{u}(x)$
for the Anosov diffeomorphism $f$ depend on the point $x\in M$ not
smoothly but only Hölder continuously. We let $0<\beta<1$ be the
Hölder exponent. (See Remark \ref{Remark:Holder_exp_beta}(1) page
\pageref{Remark:Holder_exp_beta}) In what follows, we fix a small
positive constant $\theta$ such that 
\begin{equation}
0<\theta<\beta/8.\label{eq:cond_theta}
\end{equation}
The open ball of radius $c>0$ on $\real^{2d}$ is denoted by
\[
\mathbb{D}\left(c\right)=\left\{ x\in\mathbb{R}^{2d}\mid|x|<c\right\} .
\]

\subsection{Truncation operations in the real space}

\label{ss:trumcation}

We first consider the operation of truncating functions in the (real)
space $\mathbb{R}_{x}^{2d}$ by multiplying smooth functions with
small supports. Below we consider the following setting:

\bigskip{}

\noindent%
\framebox{\begin{minipage}[t]{1\columnwidth}%
\textbf{Setting I:} \textit{~\label{Setting-I}For each $\hbar>0$,
there is a given set $\scrX_{\hbar}$ of $C^{\infty}$ functions on
$\real^{2d}$ such that, for all $\psi\in\mathscr{X}_{\hbar}$ and
$\hbar>0$, }
\begin{description}
\item [{(C1)}] \textit{the support of $\psi$ is contained in the disk
$\mathbb{D}\left(C_{*}\hbar^{1/2-\theta}\right)$ and}
\item [{(C2)}] \textit{$\left|\partial_{x}^{\alpha}\psi(x)\right|<C_{\alpha}\hbar^{-\left(\frac{1}{2}-\theta\right)\left|\alpha\right|}$
for each multi-index $\alpha\in\mathbb{N}^{2d}$,}
\end{description}
\textit{where $C_{*}>0$ and $C_{\alpha}>0$ are constants independent
of $\psi\in\mathscr{X}_{\hbar}$ and $\hbar>0$.}%
\end{minipage}}

\bigskip{}

In the next section, we will consider a few specific sets of functions
as $\scrX_{\hbar}$ and apply the argument in this section to them.
\begin{rem}
The condition above on $\mathscr{X}_{\hbar}$ is equivalent to the
condition that the normalized family

\begin{equation}
\widetilde{\scrX}_{\hbar}=\left\{ \left.\varphi\left(x\right)=\psi\left(\hbar^{1/2-\theta}x\right)\in C^{\infty}\left(\mathbb{R}^{2d}\right)\;\right|\;\psi\in\scrX_{\hbar}\right\} \quad\mbox{for }\hbar>0\label{eq:def_Set_X_h-1}
\end{equation}
are uniformly bounded in the (uniform) $C^{\infty}$ topology and
supported in a fixed bounded subset of $\real^{2d}$. 
\end{rem}
Recall the transformations 
\[
\Bargmann_{x}:L^{2}\left(\real_{x}^{2d}\right)\to L^{2}\left(\real_{x}^{2d}\oplus\real_{\xi}^{2d}\right),\qquad\Bargmann_{x}^{*}:L^{2}\left(\real_{x}^{2d}\oplus\real_{\xi}^{2d}\right)\to L^{2}\left(\real_{x}^{2d}\right)
\]
and 
\[
\BargmannP_{x}:=\Bargmann_{x}^{*}\circ\Bargmann_{x}:L^{2}\left(\real_{x}^{2d}\oplus\real_{\xi}^{2d}\right)\to L^{2}\left(\real_{x}^{2d}\oplus\real_{\xi}^{2d}\right),
\]
which are defined in (\ref{eq:Bargmann_modified}) and (\ref{eq:BargmannP_modified})
as slight modifications of the Bargmann transform $\Bargmann_{\hbar}$,
its adjoint $\Bargmann_{\hbar}^{*}$ and the Bargmann projector $\BargmannP_{\hbar}$
in the case $D=2d$. Notice that the operators $\Bargmann_{x}$, $\Bargmann_{x}^{*}$
and $\BargmannP_{x}$ depend on the parameter $\hbar$ (and hence
on $N$). 

Below we write $\multiplication(\varphi)$ for the multiplication
operator by a function $\varphi$. Since $\hbar^{1/2-\theta}\gg\hbar^{1/2}$
for small $\hbar$, the functions in $\scrX_{\hbar}$ are very smooth
(or flat) viewed in the scale of the wave packet $\phi_{x,\xi}(\cdot)$
used in the Bargmann transform $\Bargmann_{\hbar}$. This observation
naturally leads to the following few statements.

For each $\psi\in\mathscr{X}_{\hbar}$, let $\multiplication^{\mathrm{lift}}(\psi)=\Bargmann_{x}\circ\multiplication(\psi)\circ\Bargmann_{x}^{*}$
be the lift of the multiplication operator $\multiplication(\psi)$
with respect to the (modified) Bargmann transform $\Bargmann_{x}$.
Then it is approximated by the multiplication by the function $\psi\circ\pi$
with $\pi\left(x,\xi\right):=x$.
\begin{lem}
\label{lm:lift_of_multiplication_operator}There exists a constant
$C>0$ such that, for any $\hbar>0$ and $\psi\in\scrX_{\hbar}$,
we have
\begin{equation}
\left\Vert \multiplication^{\mathrm{lift}}(\psi)-\multiplication(\psi\circ\pi)\circ\BargmannP_{x}\right\Vert _{L^{2}(\real^{2d}\oplus\real^{2d},(\mathcal{W}_{\hbar}^{r})^{2})}<C\hbar^{\theta}\label{eq:approximation_of_lift_of_multiplication}
\end{equation}
and 
\[
\left\Vert \multiplication^{\mathrm{lift}}(\psi)-\BargmannP_{x}\circ\multiplication(\psi\circ\pi)\right\Vert _{L^{2}(\real^{2d}\oplus\real^{2d},(\mathcal{W}_{\hbar}^{r})^{2})}<C\hbar^{\theta}.
\]
Consequently we have 
\begin{equation}
\left\Vert \left[\BargmannP_{x},\multiplication(\psi\circ\pi)\right]\right\Vert _{L^{2}(\real^{2d}\oplus\real^{2d},(\mathcal{W}_{\hbar}^{r})^{2})}<C\hbar^{\theta}\label{eq:commutator_BargmannP_multiplication_op}
\end{equation}
where $[A,B]$ denotes the commutator of two operators: $[A,B]=A\circ B-B\circ A$.
The same statement holds true with $\mathcal{W}_{\hbar}^{r}$ replaced
by $\mathcal{W}_{\hbar}^{r,\pm}$.\end{lem}
\begin{proof}
The kernel of the operator 
\[
\multiplication^{\mathrm{lift}}(\psi)-\multiplication(\psi\circ\pi)\circ\BargmannP_{x}=\Bargmann_{x}\circ\multiplication(\psi)\circ\Bargmann_{x}^{*}-\multiplication(\psi\circ\pi)\circ\Bargmann_{x}\circ\Bargmann_{x}^{*}
\]
is written 
\[
K(x,\xi;x',\xi')=(2\pi\hbar)^{-d}\int e^{(i/\hbar)\left(\xi(y-x)-\xi'(x'-y)\right)}\cdot e^{-|y-x|^{2}/4\hbar-|y-x'|^{2}/4\hbar}(\psi(y)-\psi(x))dy.
\]
We apply integration by parts, using the differential operator 
\[
L=\frac{1-i(\xi-\xi')\partial_{y}}{1+\hbar^{-1}(\xi-\xi')^{2}},
\]
which satisfies $L^{\nu}\left(e^{(i/\hbar)\left(\xi(y-x)-\xi'(x'-y)\right)}\right)=e^{(i/\hbar)\left(\xi(y-x)-\xi'(x'-y)\right)}$
for $\nu$ times. Then we get 
\begin{align*}
K(x,\xi;x',\xi') & =(2\pi\hbar)^{d}\int e^{(i/\hbar)\left(\xi(y-x)-\xi'(x'-y)\right)}\cdot(^{t}L)^{\nu}\left(e^{-|y-x|^{2}/4\hbar-|y-x'|^{2}/4\hbar}(\psi(y)-\psi(x))\right)dy
\end{align*}
where $^{t}L=(1-i(\xi-\xi')\partial_{y})/(1+\hbar^{-1}(\xi-\xi')^{2})$
is the transpose of $L$. Using the conditions (C1) and (C2) on the
family $\mathscr{X}_{\hbar}$ and, in particular, the estimate 
\begin{equation}
|\psi(x)-\psi(y)|\cdot\langle\hbar^{-1/2}|x-y|\rangle^{-1}<C\hbar^{\theta}\label{eq:smoothness_of_psi_pi}
\end{equation}
that follows from the condition (C2), we see that the integrand is
bounded in absolute value by 
\[
C\hbar^{\theta}\cdot\langle\hbar^{-1/2}|\xi-\xi'|\rangle^{-\nu}\cdot\langle\hbar^{-1/2}|x-y|\rangle^{-\nu}\cdot\langle\hbar^{-1/2}|x'-y|\rangle^{-\nu}.
\]
Hence, letting $\nu$ large, we obtain 
\[
|K(x,\xi;x',\xi')|\le C\hbar^{\theta}\cdot\langle\hbar^{-1/2}|\xi-\xi'|\rangle^{-\nu}\cdot\langle\hbar^{-1/2}|x-x'|\rangle^{-\nu}.
\]
This estimate for sufficiently large $\nu$ and Lemma \ref{lm:boundedness_of_molifier2}
give the first inequality (\ref{eq:approximation_of_lift_of_multiplication}).
We can get the second inequality in the same manner. \end{proof}
\begin{cor}
\label{lm:XM_exchange} The multiplication operator $\multiplication(\psi)$
by $\psi\in\scrX_{\hbar}$ extends to a bounded operator on $\mathcal{H}_{\hbar}^{r}(\real^{2d})$
and, for the operator norm, we have $\|\multiplication(\psi)\|_{\mathcal{H}_{\hbar}^{r}(\real^{2d})}<\|\psi\|_{\infty}+C\hbar^{\theta}$
for all $\psi\in\mathscr{X}_{\hbar}$, with a constant $C>0$ independent
of $\hbar$ and $\psi$. \end{cor}
\begin{proof}
From the commutative diagram (\ref{cd:lift}), the operator norm of
$\multiplication(\psi):\mathcal{H}_{\hbar}^{r}(\real^{2d})\to\mathcal{H}_{\hbar}^{r}(\real^{2d})$
coincides with that of the operator 
\[
\multiplication^{\mathrm{lift}}(\psi):L^{2}(\real^{2d}\oplus\real^{2d},(\mathcal{W}_{\hbar}^{r})^{2})\to L^{2}(\real^{2d}\oplus\real^{2d},(\mathcal{W}_{\hbar}^{r})^{2})
\]
restricted to the image of $\Bargmann_{x}:\mathcal{H}_{\hbar}^{r}(\real^{2d})\to L^{2}(\real^{2d}\oplus\real^{2d},(\mathcal{W}_{\hbar}^{r})^{2})$.
Hence the claim follows from Lemma \ref{lm:lift_of_multiplication_operator}.\end{proof}
\begin{cor}
\label{cor:transpose}There exists $C>0$, such that for every $\hbar>0$,
for $u,v\in\mathcal{H}_{\hbar}^{r}(\real^{2d})$ and $\ensuremath{\psi\in\scrX_{\hbar}}$
we have
\[
\left|(u,\psi\cdot v)_{\mathcal{H}_{\hbar}^{r}(\real^{2d})}-(\overline{\psi}\cdot u,v)_{\mathcal{H}_{\hbar}^{r}(\real^{2d})}\right|\leq C\hbar^{\theta}\cdot\|u\|_{\mathcal{H}_{\hbar}^{r}(\real^{2d})}\cdot\|v\|_{\mathcal{H}_{\hbar}^{r}(\real^{2d})}
\]
where \textup{$\mathcal{O}(\hbar^{\theta})$ denotes a term whose
absolute value is bounded by $C\hbar^{\theta}$ with }\end{cor}
\begin{proof}
This is a consequence of the equality 
\begin{align*}
(u,\psi\cdot v)_{\mathcal{H}_{\hbar}^{r}(\real^{2d})} & =(\Bargmann_{x}u,\multiplication^{\mathrm{lift}}(\psi)\circ\Bargmann_{x}v)_{L^{2}(\real^{2d}\oplus\real^{2d},(\mathcal{W}_{\hbar}^{r})^{2})}\\
 & =(\Bargmann_{x}u,\multiplication(\psi\circ\pi)\circ\Bargmann_{x}v)_{L^{2}(\real^{2d}\oplus\real^{2d},(\mathcal{W}_{\hbar}^{r})^{2})}+\mathcal{O}(\hbar^{\theta})\cdot\|u\|_{\mathcal{H}_{\hbar}^{r}(\real^{2d})}\cdot\|v\|_{\mathcal{H}_{\hbar}^{r}(\real^{2d})}
\end{align*}
and the parallel estimate for $(\overline{\psi}\cdot u,v)_{\mathcal{H}_{\hbar}^{r}(\real^{2d})}$,
which follow from Lemma \ref{lm:lift_of_multiplication_operator}.\end{proof}
\begin{rem}
The statements of Corollary \ref{lm:XM_exchange} and Corollary \ref{cor:transpose}
above hold true with $\mathcal{H}_{\hbar}^{r}(\real^{2d})$ replaced
by $\mathcal{H}_{\hbar}^{r,\pm}(\real^{2d})$ and the proofs are completely
parallel. This is the case for a few statements ( Lemma \ref{lem:boundedness_of_calY},
Proposition \ref{lm:L_g_Y_almost_identity}, Lemma \ref{lm:L_g_almost_identity}
and Corollary \ref{cor:L_g_almost_commutes_with_t^k}, precisely)
in this section. 
\end{rem}
Next we recall the projection operators $\romet_{\hbar}^{(k)}$ for
$0\le k\le n$ in (\ref{eq:def_t}) and $\tilde{\romet}_{\hbar}$
in (\ref{eq:def_tt}). We henceforth assume 
\begin{equation}
r>n+2+4d\label{eq:choice_of_r_3}
\end{equation}
for the choice of $r$. (This is a little more restrictive than (\ref{eq:choice_of_r_2}).)
\begin{lem}
\label{cor:XT_exchange}There exists a constant $C>0$ such that for
any $\hbar>0$, $\psi\in\scrX_{\hbar}$ and $0\le k\le n$, 
\[
\left\Vert \left[\multiplication(\psi),\romet_{\hbar}^{(k)}\right]\right\Vert _{\mathcal{H}_{\hbar}^{r,-}(\real^{2d})\to\mathcal{H}_{\hbar}^{r,+}(\real^{2d})}<C\hbar^{\theta}
\]
\end{lem}
\begin{proof}
From (\ref{eq:relation_between_Bargmann_Projectors}) and the definition
of the operator $\romet_{\hbar}^{(k)}$, we have 
\begin{align*}
 & \multiplication(\psi)\circ\romet_{\hbar}^{(k)}=\Bargmann_{x}^{*}\circ\multiplication^{\mathrm{lift}}(\psi)\circ\Bargmann_{x}\circ\mathcal{U}\circ(\mathrm{Id}\otimes T^{(k)})\circ\mathcal{U}^{-1}\\
 & =\Bargmann_{x}^{*}\circ\multiplication^{\mathrm{lift}}(\psi)\circ\Phi^{*}\circ(\Bargmann_{\nu_{q}}\otimes\Bargmann_{\zeta_{p}})\circ(\mathrm{Id}\otimes T^{(k)})\circ(\Bargmann_{\nu_{q}}^{*}\otimes\Bargmann_{\zeta_{p}}^{*})\circ(\Phi^{*})^{-1}\circ\Bargmann_{x}\\
 & =\Bargmann_{x}^{*}\circ\multiplication^{\mathrm{lift}}(\psi)\circ\Phi^{*}\circ(\BargmannP_{\nu_{q}}\otimes\mathcal{T}_{\hbar}^{(k)})\circ(\Phi^{*})^{-1}\circ\Bargmann_{x}\intertext{\mbox{\mbox{and, similarly}}}\\
 & \romet_{\hbar}^{(k)}\circ\multiplication(\psi)=\Bargmann_{x}^{*}\circ\Phi^{*}\circ(\BargmannP_{\nu_{q}}\otimes\mathcal{T}_{\hbar}^{(k)})\circ(\Phi^{*})^{-1}\circ\multiplication^{\mathrm{lift}}(\psi)\circ\Bargmann_{x}.
\end{align*}
Thus, from Lemma \ref{lm:lift_of_multiplication_operator}, it is
enough to show that 
\begin{equation}
\left\Vert \bigg[\multiplication(\psi\circ\pi),\Phi^{*}\circ(\BargmannP_{\nu_{q}}\otimes\mathcal{T}_{\hbar}^{(k)})\circ(\Phi^{*})^{-1}\bigg]\right\Vert _{L^{2}(\real^{2d}\oplus\real^{2d},(\mathcal{W}_{\hbar}^{r,-})^{2})\to L^{2}(\real^{2d}\oplus\real^{2d},(\mathcal{W}_{\hbar}^{r,+})^{2})}<C\hbar^{\theta}.\label{eq:commutator_of_Mpsi_and_PT}
\end{equation}
From Proposition \ref{prop:The-Bargmann-projector} and Lemma \ref{lm:Tn_bdd_L2W}
page \pageref{lm:Tn_bdd_L2W}, if we write $K(x,\xi;x',\xi')$ for
the kernel of the operator $\Phi^{*}\circ(\BargmannP_{\nu_{q}}\otimes\mathcal{T}_{\hbar}^{(k)})\circ(\Phi^{*})^{-1}$,
it satisfies 
\begin{align}
 & \frac{\mathcal{W}_{\hbar}^{r,+}(x,\xi)}{\mathcal{W}_{\hbar}^{r,-}(x',\xi')}|K(x,\xi;x',\xi')|\\
 & \le C_{\nu}\langle\hbar^{-1/2}|\nu_{q}-\nu'_{q}|\rangle^{-\nu}\langle\hbar^{-1/2}|\nu_{p}-\nu'_{p}|\rangle^{-\nu}\langle\hbar^{-1/2}|(\zeta_{p},\zeta_{q})|\rangle^{-(r-k)}\langle\hbar^{-1/2}|(\zeta'_{p},\zeta'_{q})|\rangle^{-(r-k)}\label{eq:kernel_of_pik}\\
 & \le C'_{\nu}\langle\hbar^{-1/2}|(\nu_{q},\nu_{p})-(\nu'_{q},\nu'_{p})|\rangle^{-\nu}\langle\hbar^{-1/2}|(\zeta_{p},\zeta_{q})-(\zeta'_{p},\zeta'_{q})|\rangle^{-(r-k)}\label{eq:kernel_of_pik2}
\end{align}
for arbitrarily large $\nu>0$, where $C_{\nu},C'_{\nu}>0$ are constants
independent of $\hbar$ . The variables $\nu_{q},\nu_{p},\zeta_{q},\zeta_{p}$
(resp. $\nu'_{q},\nu'_{p},\zeta'_{q},\zeta'_{p}$) are the coordinates
for $(x,\xi)$ (resp. $(x',\xi')$) introduced in (\ref{eq:def_phi_change})
and $|\cdot|$ denotes the Euclidean norms. The kernel $\widetilde{K}(x,\xi;x',\xi')$
of the commutator in (\ref{eq:commutator_of_Mpsi_and_PT}) is then
\[
\widetilde{K}(x,\xi;x',\xi')=(\psi(x)-\psi(x'))\cdot K(x,\xi;x',\xi').
\]
By (\ref{eq:kernel_of_pik2}) with sufficiently large $\nu$ and (\ref{eq:smoothness_of_psi_pi}),
we get 
\[
\frac{\mathcal{W}_{\hbar}^{r,+}(x,\xi)}{\mathcal{W}_{\hbar}^{r,-}(x',\xi')}|\widetilde{K}(x,\xi;x',\xi')|\le C\hbar^{\theta}\cdot\langle\hbar^{-1/2}|(x,\xi)-(y,\eta)|\rangle^{-(r-k-1)}.
\]
Hence we obtain the required estimate by Schur inequality (\pageref{eq:Schur_inequality}),
noting that $r-k-1\ge r-n-1>4d$ from the assumption (\ref{eq:choice_of_r_3}). \end{proof}
\begin{cor}
\label{cor:XT_exchange-1}There exists a constant $C>0$, such that
\begin{align*}
 & \left\Vert \left[\multiplication(\psi),\romet_{\hbar}^{(k)}\right]\right\Vert _{\mathcal{H}_{\hbar}^{r}(\real^{2d})}<C\hbar^{\theta}\quad\mbox{for \ensuremath{0\le k\le n}}\intertext{and}\\
 & \left\Vert \left[\multiplication(\psi),\tilde{\romet}_{\hbar}\right]\right\Vert _{\mathcal{H}_{\hbar}^{r}(\real^{2d})}<C\hbar^{\theta}
\end{align*}
for any $\psi\in\scrX_{\hbar}$.\end{cor}
\begin{proof}
The former claim is an immediate consequence of the last lemma. Since
$\tilde{\romet}_{\hbar}=\mathrm{Id}-\sum_{k=0}^{n}\romet_{\hbar}^{(k)}$
by definition, the latter claim follows.
\end{proof}

\subsection{Localization of the projection operator $t_{\hbar}^{(k)}$ and estimates
on trace norm}

We consider the operator $t_{\hbar}^{(k)}:\mathcal{H}_{\hbar}^{r}(\real_{x}^{2d})\to\mathcal{H}_{\hbar}^{r}(\real_{x}^{2d})$
defined in (\ref{eq:def_t}). We write $t_{\hbar}^{(k)}=\sum_{|\alpha|=k}t_{\hbar}^{(\alpha)}$
with setting (as in (\ref{eq:def_t})): 
\[
\romet_{\hbar}^{(\alpha)}:=\mathcal{U}\circ\left(\mathrm{Id}_{\mid L^{2}\left(\real_{\nu_{q}}^{d}\right)}\otimes T^{\left(\alpha\right)}\right)\circ\mathcal{U}^{-1}\quad:\mathcal{H}_{\hbar}^{r}\left(\real_{x}^{2d}\right)\to\mathcal{H}_{\hbar}^{r}\left(\real_{x}^{2d}\right)
\]
for a multi-index $\alpha\in\mathbb{N}^{d}$. Here $\mathcal{U=}\mathcal{B}_{x}^{*}\circ\Phi^{*}\circ\left(\mathcal{B}_{\nu_{q}}\otimes\mathcal{B}_{\zeta_{p}}\right)$
is the operator defined in (\ref{eq:def_U_operator}) and $T^{\left(\alpha\right)}:H_{\hbar}^{r}\left(\real_{\zeta_{p}}^{d}\right)\rightarrow H_{\hbar}^{r}\left(\real_{\zeta_{p}}^{d}\right)$
is the rank one projector defined in (\ref{eq:def_T_alpha}).

The next Lemma decomposes the projector $t_{\hbar}^{\left(\alpha\right)}$
as an integral of localized rank one projectors. Below $\|\cdot\|_{\mathrm{Tr}}$
denotes the trace norm of an operator on $\mathcal{H}_{\hbar}^{r}\left(\real_{x}^{2d}\right)$. 
\begin{lem}
For $\nu\in\mathbb{R}^{2d}$, let
\begin{equation}
\hat{\pi}_{\alpha}\left(\nu\right):=\mathcal{U}\circ\left(\left(\left(.,\phi_{\nu_{q},\nu_{p}}\right)\otimes\left(\phi_{\nu_{q},\nu_{p}},.\right)\right)\otimes T^{\left(\alpha\right)}\right)\circ\mathcal{U}^{-1}\quad:\mathcal{H}_{\hbar}^{r}\left(\real_{x}^{2d}\right)\to\mathcal{H}_{\hbar}^{r}\left(\real_{x}^{2d}\right)\label{eq:def_projector_pi_alpha}
\end{equation}
where $\phi_{\nu_{q},\nu_{p}}$ is the wave packet defined in (\ref{eq:def_of_phi})
and $\left(.,\phi_{\nu_{q},\nu_{p}}\right)\otimes\left(\phi_{\nu_{q},\nu_{p}},.\right):L^{2}(\real_{\nu_{q}}^{d})\to L^{2}(\real_{\nu_{q}})$
denotes the rank one projection operator defined by
\[
\left\{ \left(.,\phi_{\nu_{q},\nu_{p}}\right)\otimes\left(\phi_{\nu_{q},\nu_{p}},.\right)\right\} u=\left(\phi_{\nu_{q},\nu_{p}},u\right)_{L^{2}}\cdot\phi_{\nu_{q},\nu_{p}}.
\]
 The operator $\hat{\pi}_{\alpha}\left(\nu\right)$ is a rank one
projector satisfying $\left\Vert \hat{\pi}_{\alpha}\left(\nu\right)\right\Vert _{\mathrm{Tr}}\le\left\Vert \hat{\pi}_{\alpha}\left(\nu\right)\right\Vert _{\mathrm{\mathcal{H}_{\hbar}^{r}(\real^{2d})}}\le C$
with $C$ independent of $\hbar$, and depends smoothly on $\nu\in\mathbb{R}^{2d}$.
We have (in strong operator topology)
\begin{equation}
\romet_{\hbar}^{(\alpha)}=\int_{\mathbb{R}^{2d}}\hat{\pi}_{\alpha}\left(\nu\right)\frac{d\nu}{\left(2\pi\hbar\right)^{d}}\label{eq:expre_t^alpha}
\end{equation}
\end{lem}
\begin{proof}
We have
\begin{eqnarray}
t_{\hbar}^{(\alpha)} & = & \mathcal{U}\circ\left(\mathrm{Id}_{\mid L^{2}\left(\real_{\nu_{q}}^{d}\right)}\otimes T^{\left(\alpha\right)}\right)\circ\mathcal{U}^{-1}\label{eq:expression_of_t_alpha}\\
 & = & \mathcal{U}\circ\left(\left(\mathcal{B}_{\nu_{q}}^{*}\mathcal{B}_{\nu_{q}}\right)\otimes T^{\left(\alpha\right)}\right)\circ\mathcal{U}^{-1}\nonumber 
\end{eqnarray}
Since 
\[
\mathcal{B}_{\nu_{q}}^{*}\mathcal{B}_{\nu_{q}}=\int_{\mathbb{R}^{2d}}\left(.,\phi_{\nu_{q},\nu_{p}}\right)\otimes\left(\phi_{\nu_{q},\nu_{p}},.\right)\frac{d\nu}{\left(2\pi\hbar\right)^{d}}
\]
we get (\ref{eq:expre_t^alpha}).
\end{proof}
The next lemma gives an estimate on the lift of the localized projection
operator $\hat{\pi}_{\alpha}(\nu)$ with respect to the Bargmann transform.
\begin{lem}
\label{lem:kernel_of_pi_nu}The lifted operator $\Bargmann_{x}\circ\hat{\pi}_{\alpha}(\nu)\circ\Bargmann_{x}^{*}$
is written as an integral operator 
\[
\left(\Bargmann_{x}\circ\hat{\pi}_{\alpha}(\nu)\circ\Bargmann_{x}^{*}\right)u(x_{1},\xi_{1})=\int K(x_{1},\xi_{1};x_{2},\xi_{2})u(x_{2},\xi_{2})\frac{dx_{2}d\xi_{2}}{(2\pi\hbar)^{d}}.
\]
The kernel satisfies
\begin{align*}
 & \frac{\mathcal{W}^{r}(x_{1},\xi_{1})}{\mathcal{W}^{r}(x_{2},\xi_{2})}\cdot|K(x_{1},\xi_{1};x_{2},\xi_{2})|\\
 & \quad\le C_{m}\langle\hbar^{-1/2}|\nu_{1}-\nu|\rangle^{-m}\cdot\langle\hbar^{-1/2}|\nu_{2}-\nu|\rangle^{-m}\cdot\langle\hbar^{-1/2}|\zeta_{1}|\rangle^{-(r-k)}\cdot\langle\hbar^{-1/2}|\zeta_{2}|\rangle^{-(r-k)}
\end{align*}
for arbitrarily large $m>0$ with a uniform constant $C_{m}>0$, where
$(\nu_{1},\zeta_{1})$ (resp. $(\nu_{2},\zeta_{2})$) is the coordinates
of $(x_{1},\xi_{1})$ (resp. $(x_{2},\xi_{2})$) defined in Proposition
\ref{prop:Normal-coordinates.}. (Note that $r-k\ge4d+2$ from the
choice of $r$ in (\ref{eq:choice_of_r_3}).)\end{lem}
\begin{proof}
Since we have the expression (\ref{eq:expression_of_t_alpha}) of
the operator $\hat{\pi}_{\alpha}(\nu)$, the conclusion readily follows
from Lemma \ref{lm:Tn_bdd_L2W} and Proposition \ref{prop:The-Bargmann-projector}.\end{proof}
\begin{lem}
\label{lm:trace_basic}For $\psi\in\scrX_{\hbar}$ and $\alpha\in\mathbb{N}^{d}$,
the operator 
\[
\multiplication(\psi)\circ t_{\hbar}^{(\alpha)}:\mathcal{H}_{\hbar}^{r}(\real_{x}^{2d})\to\mathcal{H}_{\hbar}^{r}(\real_{x}^{2d})
\]
is a trace class operator. We have the following estimates on the
operator norm and the trace class norm (as operators on \textup{$\mathcal{H}_{\hbar}^{r}(\real^{2d})$)}:
There exists a constant $C>0$, independent of $\psi\in\mathcal{X}_{\hbar}$,
$\hbar>0$ and $\alpha\in\mathbb{N}^{d}$, such that

\begin{equation}
\left\Vert \multiplication(\psi)\circ t_{\hbar}^{(\alpha)}-\int_{\mathbb{R}^{2d}}\psi\left(\nu\right)\hat{\pi}_{\alpha}\left(\nu\right)\frac{d\nu}{\left(2\pi\hbar\right)^{d}}\right\Vert _{\mathcal{H}_{\hbar}^{r}(\real^{2d})}\le C\hbar^{\theta}\label{eq:Norm_estimate}
\end{equation}

\begin{equation}
\left\Vert \multiplication(\psi)\circ t_{\hbar}^{(\alpha)}-\int_{\mathbb{R}^{2d}}\psi\left(\nu\right)\hat{\pi}_{\alpha}\left(\nu\right)\frac{d\nu}{\left(2\pi\hbar\right)^{d}}\right\Vert _{\mathrm{Tr}}\le C\hbar^{-2\theta d+\theta}\label{eq:trace_formula_local}
\end{equation}
with $\hat{\pi}_{\alpha}\left(\nu\right)$ defined in (\ref{eq:def_projector_pi_alpha}).
The same statement holds true for $t_{\hbar}^{(\alpha)}\circ\multiplication(\psi)$.\end{lem}
\begin{lyxcode}
\end{lyxcode}
\begin{proof}
Let $\psi\in\scrX_{\hbar}$. We have
\begin{eqnarray*}
\multiplication(\psi)\circ t_{\hbar}^{(\alpha)}-\int_{\mathbb{R}^{2d}}\psi\left(\nu\right)\hat{\pi}_{\alpha}\left(\nu\right)\frac{d\nu}{\left(2\pi\hbar\right)^{d}} & \underset{(\ref{eq:expre_t^alpha})}{=} & \int_{\mathbb{R}^{2d}}\left(\multiplication(\psi)-\psi\left(\nu\right)\right)\hat{\pi}_{\alpha}\left(\nu\right)\frac{d\nu}{\left(2\pi\hbar\right)^{d}}
\end{eqnarray*}
Let $T\left(\nu\right):=\left(\multiplication(\psi)-\psi\left(\nu\right)\right)\hat{\pi}_{\alpha}\left(\nu\right)$
for every $\nu\in\mathbb{R}^{2d}$. From Lemma \ref{lem:kernel_of_pi_nu}
on the kernel of the lift of $\hat{\pi}_{\alpha}\left(\nu\right)$,
we deduce the estimates
\begin{equation}
\left\Vert \left(T\left(x\right)\right)^{*}T\left(y\right)\right\Vert _{\mathcal{H}_{\hbar}^{r}(\real^{2d})}\leq C_{m}\hbar^{\theta}\left(\frac{\left|x-y\right|}{\sqrt{\hbar}}\right)^{-m},\quad\left\Vert T\left(x\right)\left(T\left(y\right)\right)^{*}\right\Vert _{\mathcal{H}_{\hbar}^{r}(\real^{2d})}\leq C_{m}\hbar^{\theta}\left(\frac{\left|x-y\right|}{\sqrt{\hbar}}\right)^{-m}\label{eq:estimates}
\end{equation}
for any $m>0$ with a constant $C_{m}$ uniform for $x,y\in M$ and
$\hbar>0$.

We have from (\ref{eq:estimates}) that
\[
\sup_{x}\int_{\mathbb{R}^{2d}}\left\Vert \left(T\left(x\right)\right)^{*}T\left(y\right)\right\Vert _{\mathcal{H}_{\hbar}^{r}(\real^{2d})}^{1/2}d\mu\left(y\right)\leq C\hbar^{\theta},\qquad\sup_{x}\int_{\mathbb{R}^{2d}}\left\Vert T\left(x\right)\left(T\left(y\right)\right)^{*}\right\Vert _{\mathcal{H}_{\hbar}^{r}(\real^{2d})}^{1/2}d\mu\left(y\right)\leq C\hbar^{\theta}
\]
with setting $d\mu\left(\nu\right):=\frac{d\nu}{\left(2\pi\hbar\right)^{d}}$.
We now apply the integral version of the Cotlar-Stein Lemma%
\footnote{\begin{lem}
\textbf{\textup{\label{lem:Integral-version-of-Cotlar-Stein}``Integral
version of the Cotlar-Stein Lemma''}}\textup{: If $\left(T\left(x\right)\right)_{x}$
is a continuous family of bounded operators, if $d\mu\left(x\right)$
is a smooth measure, if $A:=\sup_{x}\int T\left(x\right)T\left(y\right)^{*}d\mu\left(y\right)<\infty$
and $B:=\sup_{x}\int T\left(x\right)^{*}T\left(y\right)d\mu\left(y\right)<\infty$
then $T\left(x\right)u$ in integrable for every $u$ and $\left\Vert \int T\left(x\right)d\mu\left(x\right)\right\Vert \leq\sqrt{AB}$.}\end{lem}
} to the integral $\int_{\mathbb{R}^{2d}}T\left(\nu\right)d\mu\left(\nu\right)$
of operators and deduce that 
\[
\left\Vert \int_{\mathbb{R}^{2d}}\left(\multiplication(\psi)-\psi\left(\nu\right)\right)\hat{\pi}_{\alpha}\left(\nu\right)\frac{d\nu}{\left(2\pi\hbar\right)^{d}}\right\Vert _{\mathcal{H}_{\hbar}^{r}(\real^{2d})}=\left\Vert \int_{\mathbb{R}^{2d}}T\left(\nu\right)d\mu\left(\nu\right)\right\Vert _{\mathcal{H}_{\hbar}^{r}(\real^{2d})}\leq C\hbar^{\theta}.
\]
This gives (\ref{eq:Norm_estimate}). It is easy to get (\ref{eq:trace_formula_local}).
Since $\left(\multiplication(\psi)-\psi\left(\nu\right)\right)\hat{\pi}_{\alpha}\left(\nu\right)$
is a rank one operator, we have $\left\Vert \left(\multiplication(\psi)-\psi\left(\nu\right)\right)\hat{\pi}_{\alpha}\left(\nu\right)\right\Vert _{\mathrm{Tr}}=\left\Vert \left(\multiplication(\psi)-\psi\left(\nu\right)\right)\hat{\pi}_{\alpha}\left(\nu\right)\right\Vert $
and therefore we get (\ref{eq:trace_formula_local}) from the triangle
inequality.\end{proof}
\begin{cor}
There exists a constant $C>0$, independent of $\psi\in\mathcal{X}_{\hbar}$,
$\hbar>0$ and $0\le k\le n$, such that
\[
\left|\mathrm{Tr}\,\left(\multiplication(\psi)\circ t_{\hbar}^{(k)}\right)-\frac{r(k,d)}{(2\pi\hbar)^{d}}\int\psi\, dx\right|\le C\hbar^{-2\theta d+\theta}
\]
\[
\|\multiplication(\psi)\circ t_{\hbar}^{(k)}\|_{\mathrm{Tr}}\le\frac{C}{(2\pi\hbar)^{d}}\int|\psi|\, dx
\]

\end{cor}
We can get the following statements by slightly modifying the argument
in the proof of Lemma \ref{lm:trace_basic}.
\begin{cor}
\label{cor:XT_pm}There exists a constant $C>0$, such that, \textup{for
$0\le k\le n$ and $\psi\in\scrX_{\hbar}$,} 
\[
\|\multiplication(\psi)\circ t_{\hbar}^{(k)}:\mathcal{H}_{\hbar}^{r,-}(\real^{2d})\to\mathcal{H}_{\hbar}^{r,+}(\real^{2d})\|_{Tr}\le\frac{C}{(2\pi\hbar)^{d}}\int|\psi|\, dx+C\hbar^{-\theta d+\theta}
\]

\end{cor}
~
\begin{cor}
\label{cor:XT_exchange-Tr}There exists a constant $C>0$, such that
\[
\left\Vert \left[\multiplication(\psi),\romet_{\hbar}^{(k)}\right]:\mathcal{H}_{\hbar}^{r,-}(\real^{2d})\to\mathcal{H}_{\hbar}^{r,+}(\real^{2d})\right\Vert _{Tr}<C\hbar^{-\theta d+\theta}\quad\mbox{for \ensuremath{0\le k\le n}}
\]
for any $\psi\in\scrX_{\hbar}$.
\end{cor}

\subsection{Truncation operations in the phase space\label{sub:Truncation-operations_phase}}

In order to truncate functions in the phase space $T^{*}\mathbb{R}^{2d}=\real_{x}^{2d}\oplus\real_{\xi}^{2d}$,
we consider the smooth function 
\begin{equation}
Y_{\hbar}:T^{*}\mathbb{R}^{2d}\to[0,1],\quad Y_{\hbar}(x,\xi)=\chi(\hbar^{2\theta-1/2}|(x,\xi)|)\label{def:function_Y_hbar}
\end{equation}
with $\chi:\real\to[0,1]$ a $C^{\infty}$ function satisfying (\ref{eq:def_chi}),
and then introduce the operator 
\begin{equation}
\mathcal{Y}_{\hbar}:L^{2}(\real^{2d})\to L^{2}(\real^{2d}),\quad\mathcal{Y}_{\hbar}=\Bargmann_{x}^{*}\circ\multiplication(Y_{\hbar})\circ\Bargmann_{x}.\label{def:operation_calY_hbar}
\end{equation}
Note that the size ($\sim\hbar^{1/2-2\theta}$) of the support of
the function $Y_{\hbar}$ is much larger than the size ($\sim\hbar^{1/2-\theta}$)
of the region on which the Bargmann transform of the functions in
$\mathscr{X}_{\hbar}$ concentrates in Setting I page \pageref{Setting-I},
when $\hbar>0$ is small. 

First of all, we show 
\begin{lem}
\label{lem:boundedness_of_calY} The operator $\mathcal{Y}_{\hbar}$
extends naturally to a bounded operator on $\mathcal{H}_{\hbar}^{r}(\real^{2d})$
and we have 
\[
\|\mathcal{Y}_{\hbar}\|_{\mathcal{H}_{\hbar}^{r}(\real^{2d})}<1+C\hbar^{\theta}
\]
and
\[
\|[\mathcal{Y}_{\hbar},\multiplication(\psi)]\|_{\mathcal{H}_{\hbar}^{r}(\real^{2d})}<C\hbar^{\theta}\quad\mbox{for any }\psi\in\mathscr{X}_{\hbar}
\]
with some positive constants $C$ independent of $\hbar$ and $\psi$. \end{lem}
\begin{proof}
It is enough to show
\[
\|\BargmannP_{x}\circ\multiplication(Y_{\hbar}):L^{2}(\real^{2d}\oplus\real^{2d},(\mathcal{W}_{\hbar}^{r})^{2})\to L^{2}(\real^{2d}\oplus\real^{2d},(\mathcal{W}_{\hbar}^{r})^{2})\|<1+C\hbar^{\theta}
\]
and 
\[
\|[\BargmannP_{x}\circ\multiplication(Y_{\hbar})\circ\BargmannP_{x},\multiplication^{\mathrm{lift}}(\psi)]:L^{2}(\real^{2d}\oplus\real^{2d},(\mathcal{W}_{\hbar}^{r})^{2})\to L^{2}(\real^{2d}\oplus\real^{2d},(\mathcal{W}_{\hbar}^{r})^{2})\|<C\hbar^{\theta}.
\]
Note that we have $\|[\BargmannP_{x},\multiplication(Y_{\hbar})]\|_{L^{2}(\real^{2d}\oplus\real^{2d},(\mathcal{W}_{\hbar}^{r})^{2})}\le C\hbar^{\theta}$
by a simple estimate on the kernel. The first claim is a consequence
of this estimate. For the second, we use Lemma \ref{lm:lift_of_multiplication_operator}. 
\end{proof}
The next lemma tells roughly that the truncation operator $\mathcal{Y}_{\hbar}$,
(\ref{def:operation_calY_hbar}),hardly affect the projection operators
$\romet_{\hbar}^{(k)}$, $0\le k\le n$, defined in (\ref{eq:def_t}),
if we view it in the anisotropic Sobolev spaces.
\begin{lem}
\label{YT_exchange} For $0\le k\le n$ and $\psi\in\mathscr{X}_{\hbar}$,
we have
\[
\|(\mathrm{Id}-\mathcal{Y}_{\hbar})\circ\multiplication(\psi)\circ\romet_{\hbar}^{(k)}\|_{\mathcal{H}_{\hbar}^{r,-}(\real^{2d})\to\mathcal{H}_{\hbar}^{r,+}(\real^{2d})}<C\hbar^{\theta}
\]
and 
\[
\|\romet_{\hbar}^{(k)}\circ(\mathrm{Id}-\mathcal{Y}_{\hbar})\circ\multiplication(\psi)\|_{\mathcal{H}_{\hbar}^{r,-}(\real^{2d})\to\mathcal{H}_{\hbar}^{r,+}(\real^{2d})}<C\hbar^{\theta}
\]
with some constant $C>0$ independent of $\hbar$ and $\psi$.\end{lem}
\begin{proof}
For the proof of the first inequality, it suffices to show the estimate
\[
\|A:L^{2}(\real^{2d}\oplus\real^{2d},(\mathcal{W}_{\hbar}^{r,-})^{2})\to L^{2}(\real^{2d}\oplus\real^{2d},(\mathcal{W}_{\hbar}^{r,+})^{2})\|<C\hbar^{\theta}
\]
for the operator 
\[
A:=\BargmannP_{x}\circ(\mathrm{Id}-\multiplication(Y_{\hbar}))\circ\BargmannP_{x}\circ\multiplication^{\mathrm{lift}}(\psi)\circ\Phi^{*}\circ(\BargmannP_{\nu_{q}}\otimes\mathcal{T}_{\hbar}^{(k)})\circ(\Phi^{*})^{-1}.
\]
Recall that we already have the estimates (\ref{eq:Bargman_Kernel})
and (\ref{eq:kernel_of_pik}) respectively for the kernel of the operator
$\BargmannP_{x}$ and $\Phi^{*}\circ(\BargmannP_{\nu_{q}}\otimes\mathcal{T}_{\hbar}^{(k)})\circ(\Phi^{*})^{-1}$.
Using those estimates with the property (\ref{eq:estmate_on_Wr})
of the escape function $\mathcal{W}_{\hbar}^{r,\pm}$ and noting that
\[
|(x,\xi)|\ge\hbar^{1/2-2\theta}\quad\mbox{for }(x,\xi)\in\supp\left(\mathrm{Id}-Y_{\hbar}\right)\qquad\mbox{(resp.}\quad|x|\le2\hbar^{1/2-\theta}\quad\mbox{for }x\in\supp\psi),
\]
we can estimate the kernel of the operator $\multiplication(\mathcal{W}_{\hbar}^{r,+})\circ A\circ\multiplication(\mathcal{W}_{\hbar}^{r,-})^{-1}$
in absolute value and obtain the required estimate. The second inequality
can be proved in the parallel manner. We omit the tedious details. 
\end{proof}

\subsection{Prequantum transfer operators for non-linear transformations close
to the identity }

\label{ss:nonlinear} In this subsection, we study the Euclidean prequantum
transfer operators for diffeomorphisms defined on small open subsets
on $\real^{2d}$ and close to the identity map. Roughly we show that
the action of those prequantum operators are close to the identity
as an operator on $\mathcal{H}_{\hbar}^{r}(\real^{2d})$, though this
is \emph{not} true in the literal sense. 

In this subsection, we consider the following setting in addition
to Setting I page \pageref{Setting-I}: 

\bigskip{}

\noindent%
\framebox{\begin{minipage}[t]{1\columnwidth}%
\noindent \textbf{Setting II:} \textit{~For every $\hbar>0$, there
exists a given set $\scrG_{\hbar}$ of $C^{\infty}$ diffeomorphisms
\[
g:\mathbb{D}(\hbar^{1/2-2\theta})\to g(\mathbb{D}(\hbar^{1/2-2\theta}))\subset\real^{2d}
\]
such that every $g\in\scrG_{\hbar}$ satisfies }
\begin{description}
\item [{\textit{(G1)}}] \textit{$g$ is symplectic with respect to the
symplectic form $\omega$ in (\ref{eq:Symplectic_form_on_Euclidean_space}),}
\item [{\textit{(G2)}}] \textit{$g(0)=0$ and $\|Dg(0)-\mathrm{Id}\|<C\hbar^{\beta(1/2-\theta)}$,
and}
\item [{\textit{(G3)}}] \textit{$\left|\partial^{\alpha}g\right|<C_{\alpha}$
for any multiindices $\alpha$.}
\end{description}
\textit{where $C$ and $C_{\alpha}$ are positive constants that does
not depend on $\hbar$ nor $g\in\scrG_{\hbar}$.}%
\end{minipage}}

\bigskip{}

\begin{rem}
In the next section we will consider a few different sets of diffeomorphisms
as $\mathscr{G}_{\hbar}$ and apply the argument below. At this moment,
the meaning of the bound $C\hbar^{\beta(1/2-\theta)}$ in the condition
(G2) may not be clear. This is a consequence of the fact that the
hyperbolic splitting (\ref{eq:foliation}) is $\beta$-Hölder continuous.
The reason will become clear when we introduce a family of local coordinates
on $M$ in the beginning of the next section. 
\end{rem}
For $g\in\scrG_{\hbar}$, we consider the Euclidean prequantum transfer
operator $\prequantumL_{g}$ defined in Subsection \ref{sub: Euclidean_prequantum_and_Laplacian_operator}.
Recall from Proposition \ref{prop:expression_of_prequantum_op_on_Rn-1}
that this operator is of the form 
\begin{equation}
\prequantumL_{g}:C_{0}^{\infty}(\mathbb{D}(\hbar^{1/2-2\theta}))\to C_{0}^{\infty}(g(\mathbb{D}(\hbar^{1/2-2\theta}))),\quad\prequantumL_{g}\, u(x)=e^{-(i/\hbar)\cdot\mathcal{A}_{g}(g^{-1}\left(x\right))}\cdot u(g^{-1}(x))\label{eq:Lg}
\end{equation}
 with 
\begin{equation}
\mathcal{A}_{g}(x)=\int_{\gamma}g^{*}\eta-\eta\label{eq:A_g}
\end{equation}
where $\gamma$ is a path from the origin $0$ to $x$ and $\eta$
is given in (\ref{eq:eta_3-1}). For convenience, we take the origin
$0$ as a fixed point of reference. We first show the following lemma
for the (action) function $\mathcal{A}_{g}(x)$.
\begin{lem}
\label{lm:g} If $g\in\scrG_{\hbar}$, we have 
\[
|\partial^{\alpha}\mathcal{A}_{g}(x)|\le C_{\alpha}\cdot\min\{|x|^{3-|\alpha|},1\}
\]
 for any multi-index $\alpha$ with $|\alpha|>0$, where $C_{\alpha}>0$
is a constant independent of $\hbar$.\end{lem}
\begin{proof}
From the definition, we have 
\[
|\partial^{\alpha}\mathcal{A}_{g}(x)|\le C{}_{\alpha}
\]
for any multi-index $\alpha$. Hence the conclusion holds obviously
in the case $|\alpha|\ge3$. The first derivatives of $\mathcal{A}_{g}$
at $0$ vanishes from the assumption $g(0)=0$ in the condition (G2).
Let us show that the second derivatives of $\mathcal{A}_{g}$ also
vanish. For this we use the coordinate $x=\left(p,q\right)$ in (\ref{eq:Coordinate_p_q})
and the notation
\[
g^{-1}\left(p,q\right)=\left(g_{p}\left(p,q\right),g_{q}\left(p,q\right)\right),
\]
Note that we have $g_{p}\left(0,0\right)=g_{q}\left(0,0\right)=0$
from condition (G2). Condition (G1) that $g^{-1}$ is symplectic i.e.
$(g^{-1})^{*}\omega=\omega$ writes
\[
\frac{\partial g_{q}}{\partial q_{i}}\cdot\frac{\partial g_{p}}{\partial q_{j}}-\frac{\partial g_{q}}{\partial q_{j}}\cdot\frac{\partial g_{p}}{\partial q_{i}}=0,\quad\frac{\partial g_{q}}{\partial p_{i}}\cdot\frac{\partial g_{p}}{\partial p_{j}}-\frac{\partial g_{q}}{\partial p_{j}}\cdot\frac{\partial g_{p}}{\partial p_{i}}=0,\quad\frac{\partial g_{q}}{\partial q_{j}}\cdot\frac{\partial g_{p}}{\partial p_{i}}-\frac{\partial g_{q}}{\partial p_{i}}\cdot\frac{\partial g_{p}}{\partial q_{j}}=\begin{cases}
1 & (i=j);\\
0 & (i\neq j).
\end{cases}
\]
Then we have 
\[
\left(g^{-1}\right)^{*}\eta-\eta=\frac{1}{2}\sum_{i=1}^{d}\left(g_{q}\cdot\frac{\partial g_{p}}{\partial p_{i}}-g_{p}\cdot\frac{\partial g_{q}}{\partial p_{i}}-q_{i}\right)dp_{i}+\frac{1}{2}\sum_{i=1}^{d}\left(g_{q}\cdot\frac{\partial g_{p}}{\partial q_{i}}-g_{p}\cdot\frac{\partial g_{q}}{\partial q_{i}}+p_{i}\right)dq_{i},
\]
and we check that all of the first order partial derivatives of the
coefficients of $dp_{i}$ and $dq_{i}$ of this one form vanish at
the origin $0\in\real^{2d}$. This implies that the second derivatives
of $\mathcal{A}_{g}$ vanishes at the origin. The claim of the lemma
for the case $|\alpha|\le2$ then follows immediately. 
\end{proof}
In the next section, we consider the action of the operator $\mathcal{L}_{g}$
on functions supported on $\mathbb{D}(\hbar^{1/2-\theta})$ (or sometimes
on $\mathbb{D}(2\hbar^{1/2-\theta})$). For this reason, we take the
$C^{\infty}$function 
\begin{equation}
\chi_{\hbar}:\real^{2d}\to[0,1],\quad\chi_{\hbar}(x)=\chi(\hbar^{-1/2+\theta}x/2)\label{eq:def_chi_hbar}
\end{equation}
with letting $\chi$ be a $C^{\infty}$ function satisfying (\ref{eq:def_chi})
and consider the operator 
\[
\prequantumL_{g}\circ\multiplication(\chi_{\hbar}):C^{\infty}(\real{}^{2d})\to C_{0}^{\infty}(\real^{2d})
\]
instead of the operator $\mathcal{L}_{g}$ itself. The next lemma
is the main ingredient of this subsection, which tells roughly that
the operator $\mathcal{L}_{g}$ for $g\in\mathscr{G}_{\hbar}$ is
close to the identity, under the effect of truncation by the operator
$\mathcal{Y}_{\hbar}$.
\begin{prop}
\label{lm:L_g_Y_almost_identity} There exist constants $C>0$ and
$\epsilon>0$ such that, for any $\hbar>0$ and $g\in\scrG_{\hbar}$,
we have
\[
\|\mathcal{Y}_{\hbar}\circ(\prequantumL_{g}-\mathrm{Id})\circ\multiplication(\chi_{\hbar})\|_{\mathcal{H}_{\hbar}^{r}(\real^{2d})}<C\hbar^{\epsilon}\quad\mbox{and}\quad\left\Vert (\prequantumL_{g}-\mathrm{Id})\circ\mathcal{Y}_{\hbar}\circ\multiplication(\chi_{\hbar})\right\Vert _{\mathcal{H}_{\hbar}^{r}(\real^{2d})}<C\hbar^{\epsilon}.
\]
\end{prop}
\begin{proof}
The proof below is elementary but a little demanding. We will use
the following estimate which follows from Lemma \ref{lm:g} and the
conditions in Setting II on $\scrG_{\hbar}$: For any $x\in\real^{2d}$
with $|x|\le\hbar^{1/2-\theta}$, it holds 
\begin{align*}
 & |\mathcal{A}_{g}(x)-\mathcal{A}_{g}(0)|\le C|x|^{3}<C\hbar^{3(1/2-\theta)}, &  & \|D\mathcal{A}_{g}(x)\|\le C|x|^{2}<C\hbar^{2(1/2-\theta)},\\
 & \|Dg(x)-\mathrm{Id}\|\le C\hbar^{\beta(1/2-\theta)}\quad\mbox{and }\quad &  & |g(x)-x|\le C|x|^{1+\beta}<C\hbar^{(1+\beta)(1/2-\theta)}
\end{align*}
with $C$ a constant independent of $\hbar>0$ and $g\in\mathcal{G}_{\hbar}$.
Also we note that, if $(x,\xi)\in\supp Y_{\hbar}$, we have $|(x,\xi)|\le2\hbar^{1/2-2\theta}$
and, in particular, $|\xi|\le2\hbar^{1/2-2\theta}$.

From Corollary \ref{cor:BargmannP_bdd}, the first claim follows if
we show 
\begin{equation}
\left\Vert \multiplication(Y_{\hbar})\circ\Bargmann_{x}\circ(\prequantumL_{g}-\mathrm{Id})\circ\multiplication(\chi_{\hbar})\circ\Bargmann_{x}^{*}\right\Vert _{L^{2}(\real^{2d}\oplus\real^{2d},(\mathcal{W}_{\hbar}^{r})^{2})}\le C\hbar^{\epsilon}.\label{eq:L_g_intermidiate_estimate}
\end{equation}
Recalling the definition of the operators $\Bargmann_{x}$ and $\Bargmann_{x}^{*}$
in (\ref{eq:Bargmann_modified}), we write the operator $\Bargmann_{x}\circ(\prequantumL_{g}-\mathrm{Id})\circ\multiplication(\chi_{\hbar})\circ\Bargmann_{x}^{*}$
as an integral operator of the form 
\[
(\Bargmann_{x}\circ(\prequantumL_{g}-\mathrm{Id})\circ\multiplication(\chi_{\hbar})\circ\Bargmann_{x}^{*}u)(x,\xi)=\int K(2^{-1/2}x,2^{1/2}\xi;2^{-1/2}x',2^{1/2}\xi')u(x',\xi')\frac{dx'd\xi'}{(2\pi\hbar)^{2d}},
\]
where 
\begin{align*}
K(x, & \xi;x',\xi')=a_{D}^{2}\int e^{(i/\hbar)\xi((x/2)-y)+(i/\hbar)\xi'(y-(x'/2))}\cdot e^{-|y-x|^{2}/(2\hbar)-|y-x'|^{2}/(2\hbar)}\\
 & \qquad\qquad\qquad\cdot\chi_{\hbar}(y)\cdot k(x,\xi,x',\xi',y)\, dy
\end{align*}
and 
\[
k(x,\xi,x',\xi',y)=e^{(i/\hbar)\mathcal{A}_{g}(g(y))-(i/\hbar)\xi(g(y)-y)-(|g(y)-x|^{2}-|y-x|^{2})/(2\hbar)}-1.
\]
(The factor $2^{\pm1}$ appears because of the change of variable
$\tilde{\sigma}$ in the definition of $\Bargmann_{x}$ and $\Bargmann_{x^{*}}$.
But this is not important in any sense.) Applying integration by parts
to the integral above for $\nu$ times, we see
\begin{align*}
 & K(x,\xi;x',\xi')\\
 & =a_{D}^{2}\int L^{\nu}(e^{(i/\hbar)\xi((x/2)-y)+(i/\hbar)\xi'(y-(x'/2))})\cdot e^{-|y-x|^{2}/(2\hbar)-|y-x'|^{2}/(2\hbar)}\chi_{\hbar}(y)k(x,\xi,x',\xi',y)\, dy\\
 & =a_{D}^{2}\int e^{(i/\hbar)\xi((x/2)-y)+(i/\hbar)\xi'(y-(x'/2))}\cdot(^{t}L)^{\nu}(e^{-|y-x|^{2}/(2\hbar)-|y-x'|^{2}/(2\hbar)}\chi_{\hbar}(y)k(x,\xi,x',\xi',y))\, dy
\end{align*}
where $L$ is the differential operators defined by
\[
Lu=\frac{1}{1+\hbar^{-1}|\xi-\xi'|^{2}}\cdot\left(1+i\sum_{j=1}^{2d}(\xi_{j}-\xi'_{j})\frac{\partial}{\partial y_{j}}\right)u
\]
and $^{t}L$ is its transpose: 

\[
^{t}Lu=\left(1-i\sum_{j=1}^{2d}(\xi_{j}-\xi'_{j})\frac{\partial}{\partial y_{j}}\right)\left(\frac{1}{1+\hbar^{-1}|\xi-\xi'|^{2}}\cdot u\right).
\]
Using the estimates noted in the beginning in the resulting terms,
we can get the estimate 
\begin{equation}
|K(x,\xi;x',\xi')|\le C_{\nu}\cdot\hbar^{\epsilon}\cdot\langle\hbar^{-1/2}|(x,\xi)-(x',\xi')|\rangle^{-\nu}\quad\mbox{ for \ensuremath{(x,\xi)\in\supp\mathcal{Y}_{\hbar}}}\label{eq:kernel_estimate_nonlinear}
\end{equation}
for a small constant $\epsilon>0$ and arbitrarily large $\nu>0$,
where $C_{\nu}$ is a constant independent of $\hbar$. 
\begin{rem}
The result of integration by part is not very simple. But we have
only to consider the order w.r.t. the parameter $\hbar$, since we
allow the constant $C_{\nu}$ to depend on the derivatives of $g$.
Hence it is not too difficult to do. Just note that $\theta$ satisfies
the condition (\ref{eq:cond_theta}). 
\end{rem}
This estimate for sufficiently large $\nu$ and (\ref{eq:estmate_on_Wr})
yields the required estimate. The second claim is proved in the parallel
manner. 
\end{proof}
As we noted in the beginning of this section, the operator $\prequantumL_{g}\circ\multiplication(\chi_{\hbar})$
may not extends to a bounded operator from $\mathcal{H}_{\hbar}^{r}(\real^{2d})$
to itself, even though $g\in\mathcal{G}_{\hbar}$ is very close to
the identity map. The next proposition (and hyperbolicity of $f$)
will compensate this inconvenience. 
\begin{prop}
\label{pp:bdd_g}For any $g\in\scrG_{\hbar}$, we have 
\begin{equation}
\left\Vert \prequantumL_{g}\circ\multiplication(\chi_{\hbar})\right\Vert _{\mathcal{H}_{\hbar}^{r,+}(\real^{2d})\to\mathcal{H}_{\hbar}^{r}(\real^{2d})}\le C_{0}\quad\mbox{and }\quad\left\Vert \prequantumL_{g}\circ\multiplication(\chi_{\hbar})\right\Vert _{\mathcal{H}_{\hbar}^{r}(\real^{2d})\to\mathcal{H}_{\hbar}^{r,-}(\real^{2d})}\le C_{0}\label{eq:inequality}
\end{equation}
for sufficiently small $\hbar>0$, where $C_{0}>1$ is a constant
that depends only on $n$, $r$, $d$, $\theta$ and the choice of
the escape functions $W$ and $W^{\pm}$ in subsection \ref{ss:wL2}.
(In particular, $C_{0}$ is independent of the choice of the family
$\mathscr{G}_{\hbar}$.)
\end{prop}
The conclusion of this proposition is quite natural in view of the
facts that 
\begin{equation}
\mathcal{W}_{\hbar}^{r}\circ G(x,\xi)\cdot\chi_{\hbar}(x)\le\mathcal{W}_{\hbar}^{r,+}(x,\xi),\quad\mathcal{W}_{\hbar}^{r,-}\circ G(x,\xi)\cdot\chi_{\hbar}(x)\le\mathcal{W}_{\hbar}^{r}(x,\xi)\label{eq:w}
\end{equation}
for the canonical map 
\[
G:\real^{2d}\oplus\real^{2d}\to\real^{2d}\oplus\real^{2d},\quad G(x,\xi)=(g(x),{}^{t}(Dg(x))^{-1}(\xi))
\]
associated to the operator $\mathcal{L}_{g}$. (Recall the argument
in Subsection \ref{sub:Semiclassical-description-of-prequantum_op}.)
And it can be proved in essentially same ways as the argument given
in the papers \cite{baladi_05} and \cite{fred-roy-sjostrand-07},
where Littlewood-Paley theory and the theory of pseudodifferential
operator is used respectively. Below we give a proof below by interpreting
the argument in \cite{baladi_05} in terms of the Bargmann transform.
(But the reader may skip it because this is not a very essential part
of our argument and may be proved in various ways. )
\begin{proof}
We decompose the operator $\Bargmann_{x}\circ\mathcal{L}_{g}\circ\multiplication(\chi_{\hbar})\circ\Bargmann_{x}^{*}$
as 
\begin{align*}
\Bargmann_{x}\circ\mathcal{L}_{g}\circ\multiplication(\chi_{\hbar})\circ\Bargmann_{x}^{*} & =\multiplication(1-Y_{\hbar})\circ\Bargmann_{x}\circ\mathcal{L}_{g}\circ\multiplication(\chi_{\hbar})\circ\Bargmann_{x}^{*}\\
 & \qquad+\multiplication(Y_{\hbar})\circ\Bargmann_{x}\circ(\mathcal{L}_{g}-\mathrm{Id})\circ\multiplication(\chi_{\hbar})\circ\Bargmann_{x}^{*}\\
 & \qquad+\multiplication(1-Y_{\hbar})\circ\Bargmann_{x}\circ\multiplication(\chi_{\hbar})\circ\Bargmann_{x}^{*}.
\end{align*}
If we apply Lemma \ref{lm:L_g_Y_almost_identity} (or more precisely
(\ref{eq:L_g_intermidiate_estimate}) in the proof) to the second
term and Lemma \ref{lem:boundedness_of_calY} and corollary \ref{lm:XM_exchange}
to the third term, we see that these two operators are bounded operators
on $\mathcal{H}_{\hbar}^{r}(\real^{2d})$ and the operator norms are
bounded by an absolute constant. Hence, in order to prove the former
claim of the theorem, it suffices to show that the operator norm of
\begin{equation}
\multiplication(1-Y_{\hbar})\circ\Bargmann_{x}\circ\mathcal{L}_{g}\circ\multiplication(\chi_{\hbar})\circ\Bargmann_{x}^{*}:L^{2}(\real^{2d},(\mathcal{W}_{\hbar}^{r,+})^{2})\to L^{2}(\real^{2d},(\mathcal{W}_{\hbar}^{r})^{2})\label{eq:(I-Y)LgM-1}
\end{equation}
is bounded by a constant $C_{0}$ with the same property as stated
in the proposition. (The latter claim is proved in the parallel manner.)
Below we give a proof%
\footnote{The argument in the following part is a little sketchy. For the details,
we refer the argument in \cite{baladi_05}, though it will not be
very necessary.%
}.

We take and fix $1/3<a^{+}<b^{+}<a<b<1/2$. Then we introduce a $C^{\infty}$partition
of unity $\{\psi_{n}\}_{n\in\mathbb{Z}}$ (resp. $\{\psi_{n}^{+}\}_{n\in\mathbb{Z}}$)
on $\real^{2d}$ with the following properties: 
\begin{itemize}
\item The function $\psi_{n}$ (resp. $\psi_{n}^{+}$) is supported on the
disk $|\xi|\le1$ if $n=0$ and on the annulus $2^{|n|-1}\le|\xi|\le2^{|n|+1}$
otherwise.
\item The function $\psi_{n}$ is supported on the cone $\mathbf{C}_{+}(b)$
if $n>0$ and on the cone $\mathbf{C}_{-}(1/a)=\overline{\real^{2d}\setminus\mathbf{C_{+}}(a)}$
if $n<0$. Respectively, the function $\psi_{n}^{+}$ is supported
on the cone $\mathbf{C}_{+}(b^{+})$ if $n>0$ and on the cone $\mathbf{C}_{-}(1/a^{+})=\overline{\real^{2d}\setminus\mathbf{C_{+}}(a^{+})}$
if $n<0$. 
\item The normalized functions $\xi\mapsto\psi_{n}(2^{n}\xi$) (resp. $\xi\mapsto\psi_{n}^{+}(2^{n}\xi)$)
are uniformly bounded in $C^{\infty}$ norm. 
\end{itemize}
For each $\hbar>0$, we define functions $\psi_{n}$ and $\psi_{n}^{+}$
on $T^{*}\real^{2d}=\real^{2d}\oplus\real^{2d}$ for $n\in\mathbb{Z}$
by 
\[
\psi_{n,\hbar}(x,\xi)=\psi_{n}(\hbar^{-1/2}\zeta)\quad\mbox{resp. }\psi_{n,\hbar}^{+}(x,\xi)=\psi_{n}^{+}(\hbar^{-1/2}\zeta)
\]
where $\zeta=(\zeta_{p},\zeta_{q})$ is the coordinates introduced
in Proposition \ref{prop:Normal-coordinates.}. Then, from the definition
of the partition of unities above, we have 
\[
\|\mathcal{W}_{\hbar}^{r}\cdot u\|_{L^{2}}^{2}\le C_{0}\sum_{n=-\infty}^{\infty}2^{2rn}\|\psi_{n,\hbar}\cdot u\|_{L^{2}}
\]
and also 
\begin{equation}
\sum_{n=-\infty}^{\infty}2^{2rn}\|\mathbf{\psi}_{n,\hbar}^{+}\cdot u\|_{L^{2}}^{2}\le C_{0}\|\mathcal{W}_{\hbar}^{r,+}\cdot u\|_{L^{2}}^{2}\label{eq:estimate_of_anistropic_norm}
\end{equation}
for any function $u(x,\xi)\in C^{\infty}(\real^{2d}\oplus\real^{2d})$.
From the first inequality above and the definition of the function
$Y_{\hbar}$, we have 
\begin{multline*}
\|\mathcal{W}_{\hbar}^{r}\cdot(1-Y_{\hbar})\cdot\Bargmann_{x}\circ\mathcal{L}_{g}\circ\multiplication(\chi_{\hbar})\circ\Bargmann_{x}^{*}u\|_{L^{2}}\\
\le C_{0}\sum_{|n|\ge\hbar^{-\theta}}\sum_{n'=-\infty}^{\infty}2^{2rn}\|\psi_{n,\hbar}\cdot\Bargmann_{x}\circ\mathcal{L}_{g}\circ\multiplication(\chi_{\hbar})\circ\Bargmann_{x}^{*}(\psi_{n',\hbar}^{+}\cdot u)\|_{L^{2}}^{2}.
\end{multline*}
For the summands on the right hand side, we observe that 
\begin{itemize}
\item From Lemma \ref{lm:Bargmann_is_isometry}, the $L^{2}$-operator norm
of $ $$\Bargmann_{x}\circ\mathcal{L}_{g}\circ\multiplication(\chi_{\hbar})\circ\Bargmann_{x}^{*}$
is bounded by $1$, so 
\[
\|\psi_{n,\hbar}\cdot\Bargmann_{x}\circ\mathcal{L}_{g}\circ\multiplication(\chi_{\hbar})\circ\Bargmann_{x}^{*}(\psi_{n',\hbar}^{+}\cdot u)\|_{L^{2}}^{2}\le\|\psi_{n',\hbar}^{+}\cdot u\|_{L^{2}}^{2}.
\]

\item If either $|n|-|n'|\ge3$ or $n<0<n'$, we have 
\[
\mathrm{dist(}Dg_{x}^{t}(\supp\psi_{n,\hbar}),\,\supp\psi_{n',\hbar})>C_{0}\cdot\hbar^{1/2}2^{\max\{n,n'\}}\quad\mbox{for }x\in\supp\chi_{\hbar}
\]
and, by crude estimate using integration by parts, we get the estimate
\[
\|\psi_{n,\hbar}\cdot\Bargmann_{x}\circ\mathcal{L}_{g}\circ\multiplication(\chi_{\hbar})\circ\Bargmann_{x}^{*}(\psi_{n,\hbar}^{+}\cdot u)\|_{L^{2}}^{2}\le C_{\nu}(g)\cdot2^{-\nu\cdot\max\{n,n'\}}\|\psi_{n',\hbar}^{+}\cdot u\|_{L^{2}}^{2}
\]
where the constant $C_{\nu}(g)$ may depend on $g$ and $\nu$ but
not on $\hbar$. Otherwise we have $n<n'+3$ and $2^{rn}\le C_{0}2^{rn'}$.
\end{itemize}
From these observations and (\ref{eq:estimate_of_anistropic_norm}),
we can conclude the required estimate: 
\[
\|\mathcal{W}_{\hbar}^{r}\cdot(1-Y_{\hbar})\cdot\Bargmann_{x}\circ\mathcal{L}_{g}\circ\multiplication(\chi_{\hbar})\circ\Bargmann_{x}^{*}u\|_{L^{2}}^{2}\le C_{0}\sum_{n'=-\infty}^{\infty}2^{2rn}\|\psi_{n'}^{+}\cdot u\|_{L^{2}}^{2}\le C_{0}\|\mathcal{W}_{\hbar}u\|_{L^{2}}^{2}
\]
for sufficiently small $\hbar>0$.
\end{proof}

The next lemma will be used in the key step in the proof of Theorem
\ref{thm:More-detailled-description_of_spectrum}. This lemma implies
that we can ignore the action of nonlinear diffeomorphisms in $\scrG_{\hbar}$
when we restrict it to the image of the projectors $t_{\hbar}^{(k)}$.
\begin{lem}
\label{lm:L_g_almost_identity} There exist constants $\epsilon>0$
and $C>0$ independent of $\hbar$ such that the following holds:
Let $\psi\in\scrX_{\hbar}$ be supported on the disk $\mathbb{D}(2\hbar^{1/2-\theta})$
and let $g\in\scrG_{\hbar}$, $0\le k\le n$, then it holds 
\[
\left\Vert (\prequantumL_{g}-\mathrm{Id)}\circ\multiplication(\psi)\circ\romet_{\hbar}^{(k)}\right\Vert _{\mathcal{H}_{\hbar}^{r}(\real^{2d})}\le C\hbar^{\epsilon}
\]
 and 
\[
\left\Vert \romet_{\hbar}^{(k)}\circ(\prequantumL_{g}-\mathrm{Id})\circ\multiplication(\psi)\right\Vert _{\mathcal{H}_{\hbar}^{r}(\real^{2d})}\le C\hbar^{\epsilon}.
\]
\end{lem}
\begin{proof}
We decompose the operator $(\prequantumL_{g}-\mathrm{Id)}\circ\multiplication(\psi)\circ\romet_{\hbar}^{(k)}$
into 
\begin{equation}
(\prequantumL_{g}\circ\multiplication(\chi_{\hbar})-\mathrm{Id)}\circ\mathcal{Y}_{\hbar}\circ\multiplication(\psi)\circ\romet_{\hbar}^{(k)}\quad\mbox{and }\quad(\prequantumL_{g}\circ\multiplication(\chi_{\hbar})-\mathrm{Id)}\circ(\mathrm{Id}-\mathcal{Y}_{\hbar})\circ\multiplication(\psi)\circ\romet_{\hbar}^{(k)}.\label{eq:decomposition_in_pf_lm:L_g_almost_identity}
\end{equation}
(Note that we have $\chi_{\hbar}\cdot\psi=\psi$ from the assumption.)
The operator norm of the latter is also bounded by $C\hbar^{\epsilon}$
from Lemma \ref{YT_exchange} and Proposition \ref{pp:bdd_g}. We
further write the former part as 
\begin{align*}
 & (\prequantumL_{g}-\mathrm{Id)}\circ\multiplication(\psi)\circ\mathcal{Y}_{\hbar}\circ\romet_{\hbar}^{(k)}+(\prequantumL_{g}\circ\multiplication(\chi_{\hbar})-\mathrm{Id)}\circ[\mathcal{Y}_{\hbar},\multiplication(\psi)]\circ\romet_{\hbar}^{(k)}
\end{align*}
Then we see that the operator norm of the former part is bounded by
$C\hbar^{\epsilon}$, from Lemma \ref{lm:L_g_Y_almost_identity},
Corollary \ref{lm:pik_rpm}, Lemma \ref{lem:boundedness_of_calY}
and Proposition \ref{pp:bdd_g}. The latter part is also bounded by
$C\hbar^{\epsilon}$ from Lemma \ref{cor:XT_exchange}. Thus we obtain
the former claim. The latter claim can be proved in the parallel manner.\end{proof}
\begin{cor}
\label{cor:L_g_almost_commutes_with_t^k}There exist constants $\epsilon>0$
and $C>0$ independent of $\hbar$ such that, for any $\psi\in\scrX_{\hbar}$
and $g\in\scrG_{\hbar}$ it holds 
\[
\left\Vert [\prequantumL_{g}\circ\multiplication(\psi),\,\romet_{\hbar}^{(k)}]\right\Vert _{\mathcal{H}_{\hbar}^{r}(\real^{2d})}\le C\hbar^{\epsilon}\quad\mbox{for }0\le k\le n
\]
 and also 
\[
\left\Vert [\prequantumL_{g}\circ\multiplication(\psi),\,\tilde{\romet}_{\hbar}]\right\Vert _{\mathcal{H}_{\hbar}^{r}(\real^{2d})}\le C\hbar^{\epsilon}.
\]
\end{cor}
\begin{proof}
The former claim is an immediate consequence of the last lemma and
Lemma \ref{cor:XT_exchange-1}. The latter claim then follows from
the relation $\tilde{\romet}_{\hbar}=\mathrm{Id}-\sum_{k=0}^{n}\romet_{\hbar}^{(k)}$.\end{proof}
\begin{lem}
\label{lm:L_g_almost_identity-Tr} There exist constants $\epsilon>0$
and $C>0$, independent of $\hbar$, such that the following holds
true: Let $\psi\in\scrX_{\hbar}$ be supported on the disk $\mathbb{D}(2\hbar^{1/2-\theta})$
and let $g\in\scrG_{\hbar}$, $0\le k\le n$, then it holds 
\[
\left\Vert (\prequantumL_{g}-\mathrm{Id)}\circ\multiplication(\psi)\circ\romet_{\hbar}^{(k)}:\mathcal{H}_{\hbar}^{r}(\real^{2d})\to\mathcal{H}_{\hbar}^{r}(\real^{2d})\right\Vert _{tr}\le C\hbar^{-\theta d+\epsilon}
\]
 
\[
\left\Vert \romet_{\hbar}^{(k)}\circ(\prequantumL_{g}-\mathrm{Id})\circ\multiplication(\psi):\mathcal{H}_{\hbar}^{r}(\real^{2d})\to\mathcal{H}_{\hbar}^{r}(\real^{2d})\right\Vert _{tr}\le C\hbar^{-\theta d+\epsilon}
\]
and 
\[
\left\Vert [\prequantumL_{g}\circ\multiplication(\psi),\,\romet_{\hbar}^{(k)}]:\mathcal{H}_{\hbar}^{r}(\real^{2d})\to\mathcal{H}_{\hbar}^{r}(\real^{2d})\right\Vert _{tr}\le C\hbar^{-\theta d+\epsilon}
\]
\end{lem}
\begin{proof}
We write the operators $(\prequantumL_{g}-\mathrm{Id)}\circ\multiplication(\psi)\circ\romet_{\hbar}^{(k)}$
as 
\[
(\prequantumL_{g}-\mathrm{Id)}\circ\multiplication(\psi)\circ\romet_{\hbar}^{(k)}=\left((\prequantumL_{g}-\mathrm{Id)}\circ\multiplication(\psi)\circ t_{\hbar}^{(k)}\right)\circ\left(t_{\hbar}^{(k)}\circ\multiplication(\chi_{\hbar})\right)+\left((\prequantumL_{g}-\mathrm{Id)}\circ\multiplication(\psi)\right)\circ[\multiplication(\chi_{\hbar}),t_{\hbar}^{(k)}].
\]
Then we obtain the first inequality by Lemma \ref{lm:L_g_almost_identity},
Lemma \ref{lm:trace_basic}, Proposition \ref{pp:bdd_g} and Corollary
\ref{cor:XT_exchange-Tr}. We obtain the second inequality in the
same manner. Then the third inequality follows from Corollary \ref{cor:XT_exchange-Tr}. \end{proof}

\section{\label{sec:Proofs-of-the_mains}Proof of the Theorem \ref{thm:band_structure}
for the band spectrum of the prequantum transfer operator}

In Section \ref{sub:Proof-of-Theorem_Nad_annuli} we show how to deduce
Theorem \ref{thm:band_structure-of_Laplacian} from Theorem \ref{thm:More-detailled-description_of_spectrum}.
Then the subsequent subsections will be devoted to the proof of Theorem
\ref{thm:More-detailled-description_of_spectrum}.

\subsection{\label{sub:Proof-of-Theorem_Nad_annuli}Proof of Theorem \ref{thm:band_structure}}

Recall that the damping function $D\left(x\right)=V\left(x\right)-V_{0}\left(x\right)$
has been defined in (\ref{eq:damping_function}).

\begin{framed}%
\begin{thm}
\textbf{\label{thm:More-detailled-description_of_spectrum} }Let $n\geq0$
and take sufficiently large $r$ accordingly. There exist a constant
$C_{0}$, which is independent of $V$ and $f$, and a constant $N_{0}>0$
such that for every $\left|N\right|>N_{0}$ one has a decomposition
of the Hilbert space $\mathcal{H}_{N}^{r}\left(P\right)$ independent
on $V$:
\begin{equation}
\mathcal{H}_{N}^{r}\left(P\right)=\mathcal{H}_{0}\oplus\mathcal{H}_{1}\ldots\oplus\mathcal{H}_{n}\oplus\mathcal{H}_{n+1}\label{eq:decomp_space}
\end{equation}
such that, writing $\tau^{\left(k\right)}$ for the projection onto
the component $\mathcal{H}_{k}$ along the other components, 
\begin{enumerate}
\item For some constant $\epsilon>0$ and $C>0$ independent of $\hbar$,
we have
\[
\left|\dim\mathcal{H}_{k}-r(k,d)\cdot N^{d}\cdot Vol_{\omega}(M)\right|\le CN^{d-\epsilon}\quad\mbox{for }0\leq k\leq n
\]
where $r(k,d)=\binom{d+k-1}{d-1}$, while $\dim\mathcal{H}_{n+1}=\infty$,
\item $\|\tau^{\left(k\right)}\|<C_{0}$ for $0\le k\le n+1$, 
\item For some constant $\epsilon>0$ and $C>0$ independent of $\hbar$,
we have, if $k\neq l$, that 
\[
\|\tau^{\left(k\right)}\circ\hat{F}_{N}\circ\tau^{\left(l\right)}\|\le CN^{-\epsilon},
\]

\item For every $0\le k\le n+1$,
\begin{equation}
\|\tau^{\left(k\right)}\circ\hat{F}_{N}\circ\tau^{\left(k\right)}\|_{\mathcal{H}_{N}^{r}\left(P\right)}\le C_{0}\sup_{x\in M}\left(e^{D\left(x\right)}\left\Vert Df_{x}|_{E_{u}}\right\Vert _{\mathrm{min}}^{-k}\right),\label{eq:rhs_1}
\end{equation}
 
\item For every $0\le k\le n$ and $u\in\mathcal{H}_{k}$, 
\begin{equation}
\left\Vert \left(\tau^{\left(k\right)}\circ\hat{F}_{N}\right)u\right\Vert _{\mathcal{H}_{N}^{r}\left(P\right)}\ge C_{0}^{-1}\inf_{x\in M}\left(e^{D\left(x\right)}\left\Vert Df_{x}|_{E_{u}}\right\Vert _{\mathrm{max}}^{-k}\right)\left\Vert u\right\Vert _{\mathcal{H}_{N}^{r}\left(P\right)},\label{eq:rhs_2}
\end{equation}
\end{enumerate}
\end{thm}
\end{framed}

Now we show how to deduce Theorem \ref{thm:band_structure} from Theorem
\ref{thm:More-detailled-description_of_spectrum}. Let $m\geq1$ and
apply Theorem \ref{thm:More-detailled-description_of_spectrum} to
$f^{m}$ and $\hat{F}_{N}^{m}$. On the right hand side of (\ref{eq:rhs_1})
and (\ref{eq:rhs_2}) we have 
\[
r_{k,m}^{+}:=C_{0}\sup_{x\in M}\left(e^{D_{m}\left(x\right)}\left\Vert \left(Df_{/E_{u}}^{m}\left(x\right)\right)^{-1}\right\Vert ^{+k}\right)
\]
 and 
\[
r_{k,m}^{-}:=C_{0}^{-1}\cdot\inf_{x\in M}\left(e^{D_{m}\left(x\right)}\left\Vert Df_{/E_{u}}^{m}\left(x\right)\right\Vert ^{-k}\right),
\]
with $D_{m}\left(x\right):=\sum_{j=0}^{m-1}D\left(f^{-j}\left(x\right)\right)$.
From Eq.(\ref{eq:def_r+_r-}) one has $\lim_{m\rightarrow\infty}\left(r_{k,m}^{\pm}\right)^{1/m}=r_{k}^{\pm}$.
So let $\varepsilon>0$ and take $m$ large enough so that 
\begin{equation}
\forall k,\quad\left(r_{k,m}^{+}\right)^{1/m}<r_{k}^{+}+\varepsilon\quad\mbox{and}\quad\left(r_{k,m}^{-}\right)^{1/m}>r_{k}^{-}-\varepsilon.\label{eq:<and>}
\end{equation}

The following arguments using Neumann series for resolvents are very
standard. For $0\leq k\leq n+1$, let 
\[
A_{k,m}:=\tau^{\left(k\right)}\circ\hat{F}_{N}^{m}\circ\tau^{\left(k\right)}\quad:\mathcal{H}_{k}\rightarrow\mathcal{H}_{k},
\]
From Theorem \ref{thm:More-detailled-description_of_spectrum} we
have $\left\Vert A_{k,m}\right\Vert \leq r_{k,m}^{+}$ and for $k\leq n$
we have also $\left\Vert A_{k,m}^{-1}\right\Vert \leq\left(r_{k,m}^{-}\right)^{-1}$
(recall that $A_{k,m}$ is invertible and finite rank for $k\leq n$).
For convenience, if $k=n+1$, we define $r_{k,m}^{-}=0$.
\begin{lem}
Let $0\leq k\leq n+1$, and $z\in\mathbb{C}$ with $\left|z\right|<r_{k,m}^{-}$
or $\left|z\right|>r_{k,m}^{+}$. Then $\left(z-A_{k,m}\right)$ is
invertible and its inverse $R_{A_{k,m}}\left(z\right):=\left(z-A_{k,m}\right)^{-1}$,
the resolvent operator, satisfies 
\[
\left\Vert R_{A_{k,m}}\left(z\right)\right\Vert \leq\frac{1}{\mbox{dist}\left(\left|z\right|,\left[r_{k,m}^{-},r_{k,m}^{+}\right]\right)}
\]
\end{lem}
\begin{proof}
If $\left|z\right|>r_{k,m}^{+}\geq\left\Vert A_{k,m}\right\Vert $,
then $\left\Vert \frac{A_{k,m}}{z}\right\Vert <1$ and we can write
a convergent Neuman series for
\[
R_{A_{k,m}}\left(z\right)=\left(z-A_{k,m}\right)^{-1}=\frac{1}{z}\left(1-\frac{A_{k,m}}{z}\right)^{-1}
\]
giving
\[
\left\Vert R_{A_{k,m}}\left(z\right)\right\Vert \leq\frac{1}{\left|z\right|}\left(1-\frac{\left\Vert A_{k,m}\right\Vert }{\left|z\right|}\right)^{-1}=\frac{1}{\left(\left|z\right|-\left\Vert A_{k,m}\right\Vert \right)}\leq\frac{1}{\mbox{dist}\left(\left|z\right|,\left[r_{k,m}^{-},r_{k,m}^{+}\right]\right)}
\]
Similarly if $\left|z\right|<r_{k,m}^{-}\leq\left\Vert A_{k,m}^{-1}\right\Vert ^{-1}$
then we have $\left\Vert zA_{k,m}^{-1}\right\Vert <1$ and a convergent
Neuman series for
\[
R_{A_{k,m}}\left(z\right)=\left(z-A_{k,m}\right)^{-1}=-A_{k,m}^{-1}\left(1-zA_{k,m}^{-1}\right)^{-1}
\]
giving
\[
\left\Vert R_{A_{k,m}}\left(z\right)\right\Vert \leq\left\Vert A_{k,m}^{-1}\right\Vert \left(1-\left|z\right|\left\Vert A_{k,m}^{-1}\right\Vert \right)^{-1}=\frac{1}{\left(\left\Vert A_{k,m}^{-1}\right\Vert ^{-1}-\left|z\right|\right)}\leq\frac{1}{\mbox{dist}\left(\left|z\right|,\left[r_{k,m}^{-},r_{k,m}^{+}\right]\right)}
\]

\end{proof}
Thus the operator 
\begin{equation}
A_{m}:=A_{0,m}\oplus A_{1,m}\oplus\ldots\oplus A_{n+1,m}\quad:\mathcal{H}_{\hbar}^{r}\left(P\right)\rightarrow\mathcal{H}_{\hbar}^{r}\left(P\right)\label{eq:decomp_A_m}
\end{equation}
has a resolvent $R_{A_{m}}\left(z\right)$ which satisfies 
\begin{equation}
\left\Vert R_{A_{m}}\left(z\right)\right\Vert \leq\sum_{k=0}^{n+1}\frac{1}{\mbox{dist}\left(\left|z\right|,\left[r_{k,m}^{-},r_{k,m}^{+}\right]\right)}\label{eq:RA_z_bound}
\end{equation}

From Theorem \ref{thm:More-detailled-description_of_spectrum}(3),
the operator $\hat{F}_{N}^{m}$ can be written $\hat{F}_{N}^{m}=A_{m}+B_{m}$
with $\left\Vert B_{m}\right\Vert \leq C\hbar^{\epsilon}$ with $C$
and $\epsilon$ independent on $\hbar$ (but depends on $m$). We
use a standard perturbation argument \cite[p.311]{hislop_sigal_book_1996}
to show that if $z\in\mathbb{C}$ is such that $\left\Vert R_{A_{m}}\left(z\right)\right\Vert \left\Vert B_{m}\right\Vert <1$
then it is not in the spectrum of $\hat{F}_{N}^{m}$: $z\notin\sigma\left(\hat{F}_{N}^{m}\right)$.
For this we write 
\begin{eqnarray*}
z-\hat{F}_{N}^{m} & = & z-A_{m}-B_{m}\\
 & = & \left(z-A_{m}\right)\left(1-\left(z-A_{m}\right)^{-1}B_{m}\right)
\end{eqnarray*}
hence

\[
\left(z-\hat{F}_{N}^{m}\right)^{-1}=\left(1-\left(z-A_{m}\right)^{-1}B_{m}\right)^{-1}\left(z-A_{m}\right)^{-1}
\]
and using Neumann series we deduce: 
\[
\left\Vert R_{\hat{F}_{N}^{m}}\left(z\right)\right\Vert \leq\left(1-\left\Vert R_{A_{m}}\left(z\right)\right\Vert \left\Vert B_{m}\right\Vert \right)^{-1}\left\Vert R_{A_{m}}\left(z\right)\right\Vert 
\]
hence $z$ is in the resolvent set of $\hat{F}_{N}^{m}$, i.e. $z\notin\sigma\left(\hat{F}_{N}^{m}\right)$.

From Eq.(\ref{eq:RA_z_bound}) and $\left\Vert B_{m}\right\Vert \leq C\hbar^{\epsilon}$,
we see that the condition $\left\Vert R_{A_{m}}\left(z\right)\right\Vert \left\Vert B_{m}\right\Vert <1$
is satisfied if for every $k\in\left\{ 0,n+1\right\} $
\[
\mbox{dist}\left(\left|z\right|,\left[r_{k,m}^{-},r_{k,m}^{+}\right]\right)>\left(n+1\right)C\hbar^{\epsilon}.
\]
In that case we have $\left\Vert R_{\hat{F}_{N}^{m}}\left(z\right)\right\Vert <C$
with $C$ independent on $N$. In other terms, replacing $z$ by $z^{m}$
and taking the power $1/m$ of the previous estimates, we have that
if $z\in\mathbb{C}$ and
\begin{equation}
\left(r_{k+1,m}^{+}+\left(n+1\right)C\hbar^{\epsilon}\right)^{1/m}<\left|z\right|<\left(r_{k,m}^{-}+\left(n+1\right)C\hbar^{\epsilon}\right)^{1/m}\label{eq:intervall_for_z}
\end{equation}
for every $0\leq k\leq n+1$ then $\left\Vert R_{\hat{F}_{N}^{m}}\left(z^{m}\right)\right\Vert <C$
and $z\notin\sigma\left(\hat{F}_{N}\right)$.

Considering (\ref{eq:<and>}), take $N_{\varepsilon}$ large enough
such that for every $\left|N\right|=\frac{1}{2\pi\hbar}>N_{\varepsilon}$
we have $\left(r_{k+1,m}^{+}+\left(n+1\right)C\hbar^{\epsilon}\right)^{1/m}<r_{k+1}^{+}+\varepsilon$
and $\left(r_{k,m}^{-}+\left(n+1\right)C\hbar^{\epsilon}\right)^{1/m}>r_{k}^{-}-\varepsilon$.
Then if $z\in\mathbb{C}$ is such that $r_{k+1}^{+}+\varepsilon<\left|z\right|<r_{k}^{-}-\varepsilon$,
we have (\ref{eq:intervall_for_z}) and therefore 
\begin{equation}
\left\Vert \left(z^{m}-\hat{F}_{N}^{m}\right)^{-1}\right\Vert <C\label{eq:bound_resolvent_F^m}
\end{equation}
and $z\notin\sigma\left(\hat{F}_{N}\right)$ from above. We have obtained
the results presented in Theorem \ref{thm:band_structure} except
for the bound on the resolvent (\ref{eq:bound_resolvent}) that we
derive now from (\ref{eq:bound_resolvent_F^m}).

From Theorem \ref{thm:More-detailled-description_of_spectrum} we
have that $\left\Vert \hat{F}_{N}\right\Vert \leq C$ is bounded independent
on $N$. For $z\in\mathbb{C}$ in the resolvent set we have the relation%
\footnote{\begin{proof}
For complex numbers $a,F\in\mathbb{C}$ we have 
\[
\left(z-F\right)\left(\sum_{r=0}^{m-1}z^{m-1-r}F^{r}\right)=\sum_{r=0}^{m-1}z^{m-r}F^{r}-\sum_{r=0}^{m-1}z^{m-1-r}F^{r+1}=\sum_{r=0}^{m-1}z^{m-r}F^{r}-\sum_{r=1}^{m}z^{m-r}F^{r}=z^{m}-F^{m}
\]
\end{proof}
}
\[
\left(z-\hat{F}_{N}\right)^{-1}=\left(\sum_{r=0}^{m-1}z^{m-1-r}\hat{F}_{N}^{r}\right)\left(z^{m}-\hat{F}_{N}^{m}\right)^{-1}
\]
hence we deduce that
\begin{eqnarray*}
\left\Vert \left(z-\hat{F}_{N}\right)^{-1}\right\Vert  & \leq & \left(\sum_{r=0}^{m-1}\left|z\right|^{m-1-r}C^{r}\right)\left\Vert \left(z^{m}-\hat{F}_{N}^{m}\right)^{-1}\right\Vert \\
 & \underset{(\ref{eq:bound_resolvent_F^m})}{\leq} & C'
\end{eqnarray*}
with some constant $C'$ independent on $N$. We have obtained (\ref{eq:bound_resolvent}).

Suppose that $r_{1}^{+}<r_{0}^{-}$. We have defined in (\ref{eq:def_Pi_hbar})
by $\Pi_{\hbar}$ the spectral projector on the external band of $\hat{F}_{N}$.
The projector $\tau^{\left(0\right)}$ introduced in Theorem \ref{thm:More-detailled-description_of_spectrum}
is the projector on the component $A_{0,m}$ of the decomposition
(\ref{eq:decomp_A_m}) of the operator $A_{m}$. We aim to finish
this subsection by showing that
\begin{equation}
\left\Vert \Pi_{\hbar}-\tau^{\left(0\right)}\right\Vert \leq C\hbar^{\epsilon}\label{eq:diff_projectors}
\end{equation}
with some constant $C>0,\epsilon>0$ independent on $\hbar=1/\left(2\pi N\right)$.
Let $m$ as above and let $\gamma$ be a clockwise path a circle of
radius $\left(r_{1}^{+}\right)^{m}<\left|z\right|<\left(r_{0}^{-}\right)^{m}$.
By Cauchy formula we have
\[
\Pi_{\hbar}=\frac{1}{2\pi i}\oint_{\gamma}R_{\hat{F}_{N}^{m}}\left(z\right)dz,\qquad R_{\hat{F}_{N}^{m}}\left(z\right):=\left(z-\hat{F}_{N}^{m}\right)^{-1},
\]
\[
\tau^{\left(0\right)}=\frac{1}{2\pi i}\oint_{\gamma}R_{A_{m}}\left(z\right)dz,\qquad R_{A_{m}}\left(z\right):=\left(z-A_{m}\right)^{-1}.
\]
We have written above that $\hat{F}_{N}^{m}=A_{m}+B_{m}$ with $\left\Vert B_{m}\right\Vert \leq C\hbar^{\epsilon}$.
For any $z\in\gamma$ we have shown that the resolvent are bounded:
$\left\Vert R_{\hat{F}_{N}^{m}}\left(z\right)\right\Vert \leq C$,
$\left\Vert R_{A_{m}}\left(z\right)\right\Vert \leq C$. From the
general relation
\[
R_{\hat{F}_{N}^{m}}\left(z\right)-R_{A_{m}}\left(z\right)=\frac{1}{2}\left(R_{\hat{F}_{N}^{m}}\left(z\right)B_{m}R_{A_{m}}\left(z\right)+R_{A_{m}}\left(z\right)B_{m}R_{\hat{F}_{N}^{m}}\left(z\right)\right)
\]
we deduce that $\left\Vert R_{\hat{F}_{N}^{m}}\left(z\right)-R_{A_{m}}\left(z\right)\right\Vert \leq C'\hbar^{\epsilon}$
and then from Cauchy formula above that $\left\Vert \Pi_{\hbar}-\tau^{\left(0\right)}\right\Vert \leq C\hbar^{\epsilon}$.

\subsection{\label{ss:local_charts}Local charts on $M$ and local trivialization
of the bundle $P\rightarrow M$}

In this section and the next ones we give the proof of the Theorem
\ref{thm:band_structure-of_Laplacian}. We henceforth consider the
setting assumed in Section \ref{sec:Introduction-and-results}. In
particular $\lambda>1$ is the constant in the condition (\ref{eq:def_dynamics})
in the definition (Definition \ref{def:Anosov_diffeo}) that $f$
is an Anosov diffeomorphism. Note that, by replacing $f$ by its iterate
if necessary, we may and do suppose that $\lambda$ is a large number.
Below we write $C_{0}$ for positive constants independent of $f$,
$V$ and $\hbar$ and write $C$ for those independent of $\hbar$
but may (or may not) dependent on $f$ and $V$. Also we assume (\ref{eq:choice_of_r_3})
for the choice of $r$.

As in (\ref{eq:cond_theta}), we fix a constant $0<\theta<\beta/8$
with $0<\beta<1$ being the Hölder exponent of the stable and unstable
sub-bundle given in (\ref{eq:def_beta}). Below we take an atlas on
$M$ depending on the semiclassical parameter $\hbar=\frac{1}{2\pi N}>0$
so that it consists of charts of diameter of order $\hbar^{\frac{1}{2}-\theta}$.
We consider $\mathbb{R}^{2d}$ as a linear symplectic space with coordinates
$x=\left(q,p\right)=\left(q^{1},\ldots q^{d},p^{1},\ldots p^{d}\right)$
and symplectic form $\omega=\sum_{i=1}^{d}dq^{i}\wedge dp^{i}$. The
open ball of radius $c>0$ is denoted by $\mathbb{D}(c):=\{x\in\real^{2d}\mid|x|<c\}$.
The following Proposition is illustrated in Figure \ref{fig:cartes_kappa}.

\begin{prop}
\textbf{\label{prop:Local-charts}``Local chart and trivialization''.
}For each $\hbar=\frac{1}{2\pi N}>0$, there exist a set of distinct
points 
\[
\mathscr{P}_{\hbar}=\left\{ m_{i}\in M\mid\,1\le i\le I_{\hbar}\right\} 
\]
and a coordinate map associated to each point $m_{i}\in\mathscr{P}_{\hbar}$,
\[
\kappa_{i}=\kappa_{i,\hbar}:\mathbb{D}\left(c\right)\subset\mathbb{R}^{2d}\rightarrow M,\quad1\le i\le I_{\hbar}
\]
with $c>0$ a constant independent of $\hbar$, so that the following
conditions hold:
\begin{enumerate}
\item $\kappa_{i}\left(0\right)=m_{i}$. 
\item The differential of $\kappa_{i}$ at the origin $0$ maps the subspaces
$\mathbb{R}^{d}\oplus\left\{ 0\right\} $ and $\left\{ 0\right\} \oplus\mathbb{R}^{d}$
(or, the $q$- and $p$- axis) isometrically onto the unstable and
stable subspace respectively: 
\[
\left(D\kappa_{i}\right)_{0}\left(\mathbb{R}^{d}\oplus\left\{ 0\right\} \right)=E_{u}\left(m_{i}\right),\quad\left(D\kappa_{i}\right)_{0}\left(\left\{ 0\right\} \oplus\mathbb{R}^{d}\right)=E_{s}\left(m_{i}\right).
\]
Further, the map $\kappa_{i}\circ(D\kappa_{i})_{0}^{-1}$ is not far
from the exponential map in the sense that 
\[
\|\exp_{m_{i}}^{-1}\circ\kappa_{i}\circ(D\kappa_{i})_{0}^{-1}:\mathbb{D}(m_{i},c)\to T_{m_{i}}M\|_{C^{k}}\le C_{k}
\]
with $C_{k}$ a constant independent of $\hbar$ nor $1\le i\le I_{\hbar}$,
where $\mathbb{D}(m_{i},c)$ denotes the disk in $T_{m_{i}}M$ with
radius $c$ and center at the origin. 
\item The open subsets $U_{i}:=\kappa_{i}\left(\mathbb{D}\left(3\hbar^{\frac{1}{2}-\theta}\right)\right)\subset M$
for $1\le i\le I_{\hbar}$ cover the manifold $M$. The cardinality
$I_{\hbar}$ of the set $\mathscr{P}_{\hbar}$ is bounded by $C_{0}\cdot\hbar^{-d\left(1-2\theta\right)}$
and we have 
\begin{equation}
\max_{1\leq i\leq I_{\hbar}}\sharp\left\{ 1\leq j\leq I_{\hbar}\mid U_{i}\cap U_{j}\neq\emptyset\right\} \leq C_{0}\label{eq:bound_on_i_sim_j}
\end{equation}
with $C_{0}$ a constant independent of $\hbar$.
\item For every $1\le i\le I_{\hbar}$, $\kappa_{i}^{*}\left(\omega\right)=\sum_{i}dq^{i}\wedge dp^{i}$
on $U_{i}$ and with an appropriate choice of a section $\tau_{i}:U_{i}\rightarrow P$,
the statement of Proposition \ref{prop:Normal-coordinates.} holds
true. 
\item If $U_{i}\cap U_{j}\neq\emptyset$, we denote the coordinate change
transformation by $\kappa_{j,i}:=\kappa_{j}^{-1}\circ\kappa_{i}:\mathbb{D}\left(c\right)\rightarrow\mathbb{R}^{2d}$.
Then there exists symplectic and isometric affine map $A_{j,i}:\mathbb{R}^{2d}\rightarrow\mathbb{R}^{2d}$
that belongs to $\mathcal{A}$ (see Definition \ref{def:A_tilde})
such that $g_{j,i}:=A_{j,i}\circ\kappa_{j,i}$ satisfies
\[
g_{j,i}\left(0\right)=0,\quad\left\Vert Dg_{j,i}(0)-\mathrm{Id}\right\Vert _{C^{1}}\leq C_{1}\cdot\hbar^{\beta\left(\frac{1}{2}-\theta\right)}\quad\mbox{and\quad}\left\Vert g_{j,i}\right\Vert _{C^{s}}<C_{s}\quad\mbox{\mbox{for \ensuremath{s\ge2}}}
\]
where $C_{s}$ for $s\ge1$ are constant independent of $\hbar$ and
$1\le i,j\le I_{\hbar}$. 
\item There exists a family of $C^{\infty}$functions $\{\psi_{i}:\real^{2d}\to[0,1]\}_{i=1}^{I_{\hbar}}$
which is supported on the disk $\mathbb{D}(\hbar^{1/2-\theta})$ and
gives a partition of unity on $M$:
\begin{equation}
\sum_{i=1}^{I_{\hbar}}\psi_{i}\circ\kappa_{i}^{-1}\equiv1\mbox{ on }M.\label{eq:partition_of_unity}
\end{equation}
The set of functions $\psi_{i}$ satisfies the conditions (C1) and
(C2) in Subsection \ref{ss:trumcation}. 
\end{enumerate}
\end{prop}
\begin{figure}
\centering{}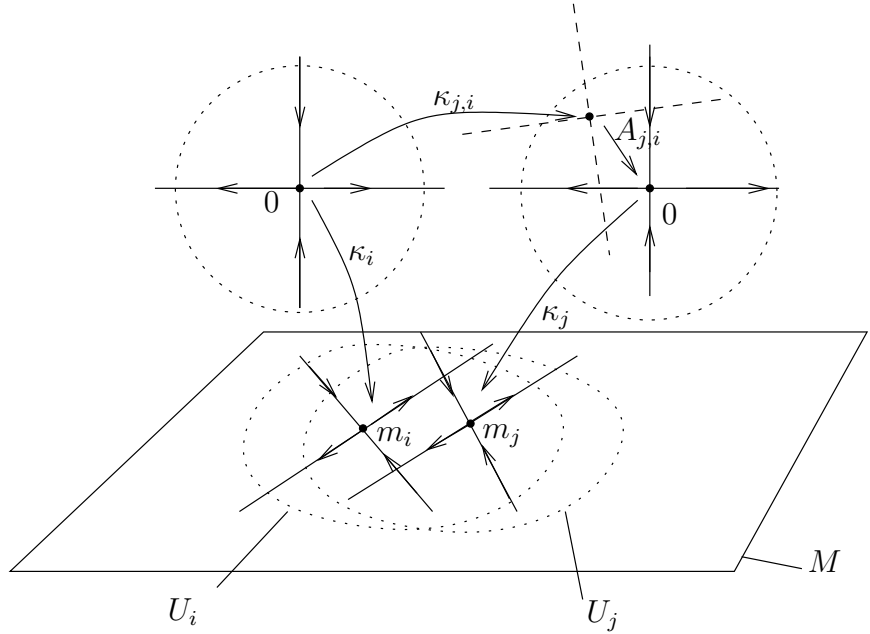\caption{\label{fig:cartes_kappa}Illustration of the local charts $U_{i}$
and maps $\kappa_{i}$ defined in Proposition \ref{prop:Local-charts}.}
\end{figure}

\begin{rem}
Since the unstable and stable vector sub-bundles, $E_{u}$ and $E_{s}$
may be non-trivial in general, we need to put the affine isometries
$A_{j,i}\in\mathcal{A}$ in the condition (5) above.\end{rem}
\begin{proof}
For each point $m\in M,$ we first define $\kappa_{m}$ as the composition
of the exponential mapping (in Riemannian geometry) $\exp_{m}:T_{m}M\to M$
and a linear map $\real^{2d}\to T_{m}M$ so that the condition 1 (with
$\kappa_{i}=\kappa_{m}$) holds true. Then, using Darboux theorem,
we can deform such $\kappa_{m}$ into a symplectic map to ensure condition
4 with keeping the condition 1. (See Lemma 3.14 in \cite[p.94]{mac_duff_98}
and its proof). We may then take a section $\tau$ as in the proof
of Proposition \ref{prop:Normal-coordinates.} so that the condition
4 (with $\kappa_{i}=\kappa_{m}$) holds. It is then clear that, if
we take the points in $\mathscr{P}_{\hbar}$ appropriately, the conditions
1 to 3 hold true with setting $\kappa_{i}:=\kappa_{m_{i}}$. The condition
5 and 6 are also obvious from this construction. 
\end{proof}
In the following subsections, we fix the set $\mathscr{P_{\hbar}}$,
the coordinate maps $\kappa_{i}$, the isometric affine maps $A_{j,i}\in\mathcal{A}$
and the functions $\psi_{i}$ taken in Proposition \ref{prop:Local-charts}
above.

\subsection{\label{sub:The-prequantum-transfer_decomposed}The prequantum transfer
operator decomposed on local charts}

To proceed, we express the transfer operator $\hat{F}_{\hbar}$ as
the totality of operators between local charts. First we discuss about
an expression of an equivariant section $u\in C_{N}^{\infty}(P)$
as a set of functions on local charts.
\begin{defn}
Let 
\[
\mathcal{E}_{\hbar}:=\bigoplus_{i=1}^{I_{\hbar}}C_{0}^{\infty}(\mathbb{D}(\hbar^{1/2-\theta})).
\]
Let $\boldI_{\hbar}:C_{N}^{\infty}\left(P\right)\to\mathcal{E}$ be
the operator that associates to each equivariant function $u\in C_{N}^{\infty}\left(P\right)$
a set of functions $\boldI_{\hbar}(u)=\left(u_{i}\right)_{i=1}^{I_{\hbar}}\in\mathcal{E}_{\hbar}$
on local charts:
\begin{equation}
\boldI_{\hbar}:\begin{cases}
C_{N}^{\infty}\left(P\right) & \rightarrow\mathcal{E}_{\hbar}\\
u & \rightarrow u_{i}\left(x\right)=\psi_{i}\left(x\right)\cdot u\left(\tau_{i}\left(\kappa_{i}\left(x\right)\right)\right)\quad\mbox{for }1\le i\le I_{\hbar}.
\end{cases}\label{eq:local_data_ui}
\end{equation}

\end{defn}
The inverse operation is given as follows%
\footnote{Beware that $\mathbf{I}_{\hbar}^{*}$ is not the $L^{2}$ adjoint
of $\mathbf{I}_{\hbar}$ here.%
}.
\begin{prop}
Let $\boldI_{\hbar}^{*}:\bigoplus_{i=1}^{I_{\hbar}}\mathcal{S}(\real^{2d})\to C_{N}^{\infty}\left(P\right)$
be the operator defined by
\begin{equation}
\left(\boldI_{\hbar}^{*}\left((u_{i})_{i=1}^{I_{\hbar}}\right)\right)(p)=\sum_{i=1}^{I_{\hbar}}e^{iN\cdot\alpha_{i}(p)}\cdot\chi_{\hbar}(x)\cdot u_{i}(x)\label{eq:expression_of_I*}
\end{equation}
where $\chi_{\hbar}$ is the function defined in (\ref{eq:def_chi_hbar}),
and $x=\kappa_{i}^{-1}(\pi(p))$ and $\alpha_{i}(p)$ is the real
number such that $p=e^{i\alpha_{i}(p)}\cdot\tau_{i}(\pi(p))$. This
operator reconstructs $u\in C_{N}^{\infty}\left(P\right)$ from its
local data $u_{i}=\left(\boldI_{\hbar}(u)\right)_{i}$ :
\begin{equation}
\boldI_{\hbar}^{*}\circ\boldI_{\hbar}=\mathrm{Id}_{C_{N}^{\infty}\left(P\right)}.\label{eq:I*_I_is_Identity}
\end{equation}
Consequently, $\boldI_{\hbar}\circ\boldI_{\hbar}^{*}:\mathcal{E}_{\hbar}\rightarrow\mathcal{E}_{\hbar}$
is a projection onto the image of $\boldI_{\hbar}$.\end{prop}
\begin{proof}
Notice that
\begin{equation}
\chi_{\hbar}\cdot\psi_{i}=\psi_{i}\label{eq:chi_hbar_psi_i}
\end{equation}
Let $w:=\left(\boldI_{\hbar}^{*}\circ\boldI_{\hbar}\right)\left(v\right)$.
From the expressions of $\mathbf{I}_{\hbar}$ and $\mathbf{I}_{\hbar}^{*}$
and equivariance of $v$, we compute
\[
w\left(p\right)=\sum_{i=1}^{I_{\hbar}}e^{iN\alpha_{i}(p)}\left(\chi_{\hbar}(x)\cdot\psi_{i}\left(x\right)\cdot v\left(\tau_{i}\left(\kappa_{i}\left(x\right)\right)\right)\right)=\sum_{i=1}^{I_{\hbar}}\psi_{i}\left(x\right)v\left(p\right)=v\left(p\right).
\]
Finally $\boldI_{\hbar}\circ\boldI_{\hbar}^{*}$ is a projector since
$\left(\boldI_{\hbar}\circ\boldI_{\hbar}^{*}\right)^{2}=\boldI_{\hbar}\circ(\boldI_{\hbar}^{*}\circ\boldI_{\hbar})\circ\boldI_{\hbar}^{*}=\boldI_{\hbar}\circ\boldI_{\hbar}^{*}$.\end{proof}
\begin{defn}
We define the \emph{lift of the prequantum transfer operator} $\hat{F}_{\hbar}$
with respect to $\boldI_{\hbar}$ as 
\begin{equation}
\boldF_{\hbar}:=\boldI_{\hbar}\circ\hat{F}_{N}\circ\boldI_{\hbar}^{*}\;:\;\bigoplus_{i=1}^{I_{\hbar}}\mathcal{S}(\real^{2d})\rightarrow\mathcal{E}_{\hbar}\subset\bigoplus_{i=1}^{I_{\hbar}}\mathcal{S}(\real^{2d}).\label{eq:Lifted_operator_F}
\end{equation}

\end{defn}
The operator $\boldF_{\hbar}$ is nothing but the prequantum transfer
operator $\hat{F}_{N}:C_{N}^{\infty}(P)\to C_{N}^{\infty}(P)$ viewed
through the local charts and local trivialization that we have chosen.
This is a matrix of operators that describe transition between local
data that $\hat{F}_{\hbar}$ induces. The next proposition gives it
in a concrete form. 
\begin{defn}
We write $i\to j$ for $0\le i,j\le I_{\hbar}$ if and only if $f\left(U_{i}\right)\bigcap U_{j}\neq\emptyset$. 
\end{defn}
Clearly we have 
\begin{equation}
\max_{1\le i\le I_{\hbar}}\#\{1\le j\le I_{\hbar}\mid i\to j\}\le C(f)\label{eq:arrow}
\end{equation}
for some constant $C(f)$ which may depend on $f$ but not on $\hbar$.
\begin{prop}
The operator $\boldF_{\hbar}$ is written as 
\[
\boldF_{\hbar}((v_{i})_{i\in I_{\hbar}})=\left(\sum_{i=1}^{I_{\hbar}}\boldF_{j,i}(v_{i})\right)_{j\in I_{\hbar}}
\]
where the component 
\[
\mathbf{F}_{j,i}:\mathcal{S}(\real^{2d})\rightarrow C_{0}^{\infty}(\mathbb{D}(\hbar^{1/2-\theta}))
\]
is defined by $\mathbf{F}_{j,i}\equiv0$ if $i\not\to j$ and, otherwise,
by 
\[
\mathbf{F}_{j,i}(v_{i})=\mathcal{L}_{f_{j,i}}\left(e^{V\circ f\circ\kappa_{i}}\cdot\psi_{j,i}\cdot\chi_{\hbar}\cdot v_{i}\right)
\]
where we set
\begin{align}
f_{j,i} & :=\kappa_{j}^{-1}\circ f\circ\kappa_{i},\label{eq:f_ji}\\
\psi_{j,i} & :=\psi_{j}\circ f_{j,i}\label{eq:psi_ji}
\end{align}
and $\mathcal{L}_{f_{j,i}}$ is the Euclidean prequantum transfer
operator defined in (\ref{eq:Lg}) with $g=f_{j,i}$. \end{prop}
\begin{rem}
$\chi_{\hbar}$ is actually not necessary.
\end{rem}
The maps $f_{j,i}$ is illustrated on Figure \ref{fig:map_fij}.
\begin{proof}
The expression of the operator $\hat{F}_{N}$ in local coordinates
has been given in Proposition \ref{prop:Local-expression-of_FN}.
Taking the multiplication by functions $\psi_{i}$, $1\le i\le I_{\hbar}$,
in the definitions of the operators $\mathbf{I_{\hbar}}$ and $\mathbf{I_{\hbar}}^{*}$
into account, we obtain the expression of $\mathbf{F}_{j,i}$ as above. 
\end{proof}
\begin{figure}
\centering{}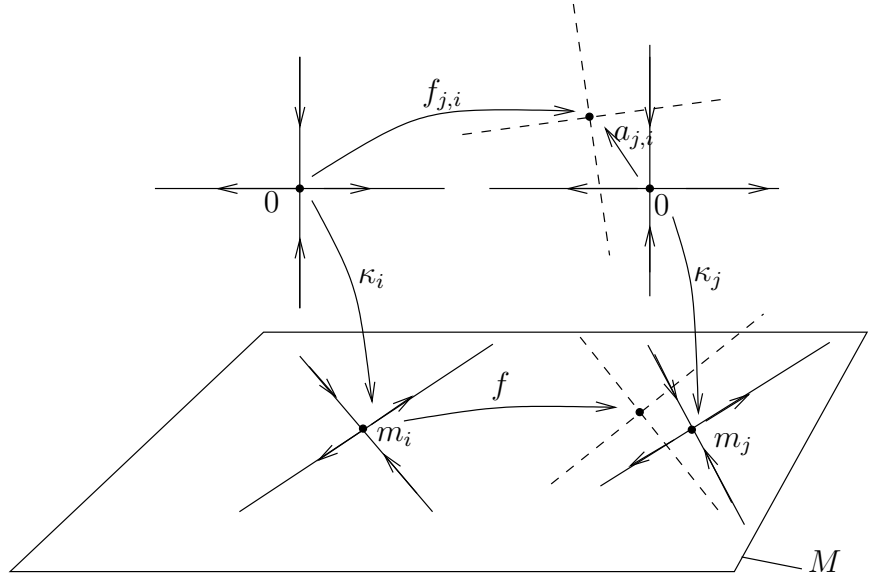\caption{\label{fig:map_fij}Illustration of the local map $f_{j,i}$ defined
in (\ref{eq:f_ji}) and $a_{j,i}$ defined in (\ref{eq:fij_decomposition}).}
\end{figure}

We define 
\[
V_{j}=\max\{V(m)\mid m\in U_{j}\}\quad\mbox{for }1\le j\le I_{\hbar}.
\]
Since the function $V$ is almost constant on each $U_{j}$, we have 
\begin{lem}
\label{lem:multipliacation_by_Psi_ji}If we set 
\[
\mathscr{X_{\hbar}=}\{\psi_{j,i}\cdot\chi_{\hbar}\mid1\le i,j\le I_{\hbar},\; i\to j\}\qquad\mbox{(resp. \;\;}\mathscr{X_{\hbar}=}\{e^{V\circ f\circ\kappa_{i}}\cdot\psi_{j,i}\cdot\chi_{\hbar}\mid1\le i,j\le I_{\hbar},\; i\to j\}),
\]
it satisfies the conditions (C1) and (C2) in Subsection \ref{ss:trumcation}.
(The constants $C$ and $C_{\alpha}$ will depend on $f$ and $V$
though not on $\hbar$.) For $1\le i,j\le I_{\hbar}$ such that $i\to j$,
we have 
\[
\left\Vert \multiplication(e^{V\circ f\circ\kappa_{i}}\cdot\psi_{j,i}\cdot\chi_{\hbar})-e^{V_{j}}\cdot\multiplication(\psi_{j,i}\cdot\chi_{\hbar})\right\Vert _{\mathcal{H}_{\hbar}^{r}(\real^{2d})}\le C(f,V)\cdot\hbar^{\theta}
\]
for some constant $C(f,V)$ independent of $\hbar$. \end{lem}
\begin{proof}
The former claim should be obvious from the choice of the coordinates
$\kappa_{i}$ and the functions $\psi_{i}$ for $i\in I_{\hbar}$.
We can get the latter claim if we apply Corollary \ref{lm:XM_exchange}
to the multiplication operators by $e^{V\circ f\circ\kappa_{i}}\cdot\psi_{j,i}\cdot\chi_{\hbar}-e^{V_{j}}\cdot\psi_{j,i}\cdot\chi_{\hbar}=(e^{V\circ f\circ\kappa_{i}}-e^{V_{j}})\cdot\psi_{j,i}\cdot\chi_{\hbar}$.
\end{proof}

\subsection{The anisotropic Sobolev spaces\label{sub:anisotropic-Sobolev-space}}
\begin{defn}
\label{dfn:global_anisotropic_Sobolev_space} The \emph{Anisotropic
Sobolev space} $\mathcal{H}_{\hbar}^{r}\left(P\right)$ is defined
as the completion of $C_{N}^{\infty}\left(P\right)$ with respect
to the norm
\[
\|u\|_{\mathcal{H}_{\hbar}^{r}\left(P\right)}:=\left(\sum_{i=1}^{I_{\hbar}}\|u_{i}\|_{\mathcal{H}_{\hbar}^{r}(\real^{2d})}^{2}\right)^{1/2}\quad\mbox{ for }u\in C_{N}^{\infty}\left(P\right),
\]
where $u_{i}=\left(\mathbf{I}_{\hbar}\left(u\right)\right)_{i}\in C_{0}^{\infty}\left(\mathbb{D}(\hbar^{1/2-\theta})\right)$
are the local data defined in (\ref{eq:local_data_ui}) and $\|u_{i}\|_{\mathcal{H}_{\hbar}^{r}(\real^{2d})}^{2}$
is the anisotropic Sobolev norm on $C_{0}^{\infty}\left(\mathbb{R}^{2d}\right)$
in Definition \ref{def:escape_function_Hr}. We define the Hilbert
spaces $\mathcal{H}_{\hbar}^{r,\pm}\left(P\right)$ in the parallel
manner, replacing $\|u_{i}\|_{\mathcal{H}_{\hbar}^{r}(\real^{2d})}^{2}$
by the norms $\|u_{i}\|_{\mathcal{H}_{\hbar}^{r,\pm}(\real^{2d})}^{2}$
respectively. \end{defn}
\begin{rem}
(1) By definition, the operation $\boldI_{\hbar}$ extends uniquely
to an isometric injection 
\[
\boldI_{\hbar}:\mathcal{H}_{\hbar}^{r}(P)\to\bigoplus_{i=1}^{I_{\hbar}}\mathcal{H}_{\hbar}^{r}(\mathbb{D}(\hbar^{1/2-\theta}))\subset\bigoplus_{i=1}^{I_{\hbar}}\mathcal{H}_{\hbar}^{r}(\real^{2d})
\]
where $\mathcal{H}_{\hbar}^{r}(\mathbb{D}(\hbar^{1/2-\theta}))$ denotes
the subspace that consists of elements supported on the disk $\mathbb{D}(\hbar^{1/2-\theta})$.

(2) From (\ref{eq:I*_I_is_Identity}), we have $\boldI_{\hbar}^{*}\circ\boldI_{\hbar}=\mathrm{Id}$
on $\mathcal{H}_{\hbar}^{r}\left(P\right)$ and also on $\mathcal{H}_{\hbar}^{r,\pm}\left(P\right)$.\end{rem}
\begin{lem}
\label{lm:IIbdd} The projector $\boldI_{\hbar}\circ\boldI_{\hbar}^{*}:\mathcal{E}_{\hbar}\rightarrow\mathcal{E}_{\hbar}$
extends to bounded operators 
\[
\boldI_{\hbar}\circ\boldI_{\hbar}^{*}:\bigoplus_{i=1}^{I_{\hbar}}\mathcal{H}_{\hbar}^{r,+}(\real^{2d})\to\bigoplus_{i=1}^{I_{\hbar}}\mathcal{H}_{\hbar}^{r}(\mathbb{D}(\hbar^{1/2-\theta}))
\]
and 
\[
\boldI_{\hbar}\circ\boldI_{\hbar}^{*}:\bigoplus_{i=1}^{I_{\hbar}}\mathcal{H}_{\hbar}^{r}(\real^{2d})\to\bigoplus_{i=1}^{I_{\hbar}}\mathcal{H}_{\hbar}^{r,-}(\mathbb{D}(\hbar^{1/2-\theta})).
\]
Further the operator norms of these projectors are bounded by a constant
independent of~$\hbar$. \end{lem}
\begin{rem}
The operator $\boldI_{\hbar}\circ\boldI_{\hbar}^{*}$ will not be
a bounded operator from $\bigoplus_{i=1}^{I_{\hbar}}\mathcal{H}_{\hbar}^{r}(\real^{2d})$
to itself. \end{rem}
\begin{proof}
To prove the claim, it is enough to apply Proposition \ref{pp:bdd_g}
and Corollary \ref{lm:XM_exchange} to each component of $\boldI_{\hbar}\circ\boldI_{\hbar}^{*}$
with setting 
\begin{equation}
\scrG_{\hbar}=\{A_{j,i}\circ\kappa_{j,i}\mid1\le i,j\le I_{\hbar},\; U_{i}\cap U_{j}\neq\emptyset\}\label{eq:scrXG_setting}
\end{equation}
and 
\[
\scrX_{\hbar}=\{\psi_{j}\circ\kappa_{j,i}\cdot\chi_{\hbar}\mid1\le i,j\le I_{\hbar},\; U_{i}\cap U_{j}\neq\emptyset\},
\]
and use (\ref{eq:bound_on_i_sim_j}). (See also the remark below.) \end{proof}
\begin{rem}
\label{rem:Aij} The affine transformation $A_{j,i}$ in (\ref{eq:scrXG_setting})
is that appeared in the choice of local coordinates in Proposition
\ref{prop:Local-charts}. Note that the prequantum transfer operator
$\prequantumL_{A_{j,i}}$ is a unitary operator on $\mathcal{H}_{\hbar}^{r}(\real^{2d})$
(and on $\mathcal{H}_{\hbar}^{r,\pm}(\real^{2d})$), by Lemma~\ref{lm:invariance_wrt_translation},
and hence we may neglect the post- or pre-composition of $\prequantumL_{A_{ij}}$
when we consider the operator norm on $\mathcal{H}_{\hbar}^{r}(\real^{2d})$
(and on $\mathcal{H}_{\hbar}^{r,\pm}(\real^{2d})$). For the later
argument, we also note that, from Lemma \ref{lm:invariance_wrt_translation},
the prequantum transfer operator $\prequantumL_{A_{j,i}}$ commutes
with the projection operators $t_{\hbar}^{(k)}$ defined in (\ref{eq:def_t}).
\end{rem}
For the operator $\hat{F}_{\hbar}$ on the Hilbert space $\mathcal{H}_{\hbar}^{r}(P)$,
we confirm the following fact at this point, though we will give a
more detailed description later. 
\begin{lem}
\label{lm:bounded_F} The operator $\boldF_{\hbar}$ defined in (\ref{eq:Lifted_operator_F})
extends uniquely to the bounded operator 
\begin{equation}
\boldF_{\hbar}:\bigoplus_{i=1}^{I_{\hbar}}\mathcal{H}_{\hbar}^{r}(\real^{2d})\to\bigoplus_{i=1}^{I_{\hbar}}\mathcal{H}_{\hbar}^{r}(\mathbb{D}(\hbar^{1/2-\theta}))\label{eq:Lifted_operator_F2}
\end{equation}
and the operator norm is bounded by a constant independent of $\hbar$.
Consequently the same result holds for the prequantum transfer operator
$\hat{F}_{\hbar}:\mathcal{H}_{\hbar}^{r}\left(P\right)\rightarrow\mathcal{H}_{\hbar}^{r}\left(P\right)$. \end{lem}
\begin{proof}
From (\ref{eq:arrow}), it is enough to prove that the operators $\boldF_{j,i}$
for $1\le i,j\le I_{\hbar}$ with $i\to j$ are bounded operators
on $\mathcal{H}_{\hbar}^{r}(\real^{2d})$ and that the operator norms
are bounded by a constant independent of $\hbar$. To see this, we
express the diffeomorphism $f_{j,i}$ in (\ref{eq:f_ji}) as a composition
\begin{equation}
f_{j,i}=a_{j,i}\circ g_{j,i}\circ B_{j,i}\label{eq:fij_decomposition}
\end{equation}
 where 
\begin{itemize}
\item $a_{j,i}:\real^{2d}\to\real^{2d}$ is a translation on $\real^{2d}$, 
\item $B_{j,i}:\real^{2d}\to\real^{2d}$ is a linear map of the form (\ref{eq:hyperbolic_f}),
i.e. $B_{j,i}=\begin{pmatrix}A & 0\\
0 & ^{t}A^{-1}
\end{pmatrix}$, with $A$ an expanding map such that $\|A^{-1}\|\le1/\lambda$, 
\item $g_{j,i}$ is a diffeomorphism such that $\scrG_{\hbar}=\{g_{j,i}\}_{1\le i,j\le I_{\hbar}}$
satisfies the condition (G1), (G2) and (G3) in Subsection \ref{ss:nonlinear}. 
\end{itemize}
This is possible because, if we let $B_{j,i}$ be the linearization
of $f_{j,i}$ at the origin and let $a_{j,i}\in\mathcal{A}$ be the
translation such that $a_{j,i}(f_{j,i}(0))=0$, then $a_{j,i}$, $B_{j,i}$
and $g_{j,i}:=a_{j,i}^{-1}\circ f_{j,i}\circ B_{j,i}^{-1}$ satisfies
the required conditions.
\begin{rem}
This decomposition of the diffeomorphism $f_{j,i}$ will be used later
in the proof of Proposition \ref{prop:key_proposition2} where we
study more detailed properties of $f_{j,i}$.
\end{rem}
From the expression (\ref{eq:fij_decomposition}) of $f_{j,i}$ above,
the operator $\boldF_{j,i}$ is expressed as the composition 
\begin{equation}
\boldF_{j,i}=\prequantumL^{(0)}\circ\prequantumL^{(1)}\circ\prequantumL^{(2)}\label{eq:expression_F_ij}
\end{equation}
where $\mathcal{L}^{\left(0\right)}:=\prequantumL_{a_{j,i}}$ and
$\mathcal{L}^{\left(2\right)}:=\prequantumL_{B_{j,i}}$ are the Euclidean
prequantum transfer operators (\ref{eq:Lg}) for the diffeomorphism
$a_{ij}$ and $B_{ij}$ respectively, while $\prequantumL^{(1)}$
is the operator of the form 
\[
\prequantumL^{(1)}u=\prequantumL_{g_{j,i}}\bigg(\big((e^{V\circ f\circ\kappa_{i}}\cdot\psi_{j,i}\cdot\chi_{\hbar})\circ B_{j,i}^{-1}\big)\cdot u\bigg)
\]
with $\psi_{j,i}$ the function defined in (\ref{eq:psi_ji}). Note
that the functions
\[
(e^{V\circ f\circ\kappa_{i}}\cdot\psi_{j,i}\cdot\chi_{\hbar})\circ B_{j,i}^{-1}=(e^{V\circ f\circ\kappa_{i}}\cdot\chi_{\hbar})\circ B_{j,i}^{-1}\cdot(\psi_{j}\circ a_{j,i}\circ g_{j,i})
\]
is supported on the disk $\mathbb{D}(2\hbar^{1/2-\theta})$, provided
that $\hbar$ is sufficiently small. Hence we may write the operator
$\mathcal{L}^{(1)}$ as 
\[
\prequantumL^{(1)}u=\prequantumL_{g_{j,i}}\circ\multiplication(\chi_{\hbar})\bigg(\big((e^{V\circ f\circ\kappa_{i}}\cdot\psi_{j,i}\cdot\chi_{\hbar})\circ B_{j,i}^{-1}\big)\cdot u\bigg).
\]

From Lemma \ref{lm:invariance_wrt_translation}, $\prequantumL_{a_{j,i}}:\mathcal{H}_{\hbar}^{r}(\real^{2d})\to\mathcal{H}_{\hbar}^{r}(\real^{2d})$
is a unitary operator. From Lemma \ref{lm:boundedness_of_hyperbolic_linear_map},
the operator $\prequantumL_{B_{j,i}}:\mathcal{H}_{\hbar}^{r}(\real^{2d})\to\mathcal{H}_{\hbar}^{r,+}(\real^{2d})$
is bounded and the operator norm is bounded by a constant independent
of $\hbar$. From Lemma \ref{pp:bdd_g} and Corollary \ref{lm:lift_of_multiplication_operator},
so is the operator $\prequantumL^{(1)}:\mathcal{H}_{\hbar}^{r,+}(\real^{2d})\to\mathcal{H}_{\hbar}^{r}(\real^{2d})$,
because 
\begin{equation}
\scrX_{\hbar}=\{(e^{V\circ f\circ\kappa_{i}}\cdot\psi_{j,i}\cdot\chi_{\hbar})\circ B_{j,i}^{-1}\}_{1\le i,j\le I_{\hbar}},\quad\scrG_{\hbar}=\{g_{j,i}\}_{1\le i,j\le I_{\hbar}}\label{eq:A_setting_for_X_and_G}
\end{equation}
satisfy respectively the conditions (C1), (C2) in Section \ref{ss:trumcation}
and (G1), (G2), (G3) in Section \ref{ss:nonlinear}. 
\end{proof}

\subsection{The main propositions\label{sub:The-main-propositions}}

In this subsection, we give two key propositions which will give Theorem
\ref{thm:More-detailled-description_of_spectrum} as a consequence.
To state the propositions, we introduce the projection operators 
\begin{equation}
{\boldt}_{\hbar}^{(k)}:\bigoplus_{i=1}^{I_{\hbar}}\mathcal{H}_{\hbar}^{r}(\real^{2d})\to\bigoplus_{i=1}^{I_{\hbar}}\mathcal{H}_{\hbar}^{r}(\real^{2d}),\qquad{\boldt}_{\hbar}^{(k)}((u_{i})_{i=1}^{I_{\hbar}})=(\romet_{\hbar}^{(k)}(u_{i}))_{i=1}^{I_{\hbar}},\label{eq:def_tk}
\end{equation}
 for $0\le k\le n$ and 
\begin{equation}
\tilde{\boldt}_{\hbar}:\bigoplus_{i=1}^{I_{\hbar}}\mathcal{H}_{\hbar}^{r}(\real^{2d})\to\bigoplus_{i=1}^{I_{\hbar}}\mathcal{H}_{\hbar}^{r}(\real^{2d}),\qquad\tilde{\boldt}_{\hbar}((u_{i})_{i=1}^{I_{\hbar}})=(\tilde{\romet}_{\hbar}(u_{i}))_{i=1}^{I_{\hbar}},\label{eq:def_tilde_t}
\end{equation}
which are just applications of the projection operators $\romet_{\hbar}^{(k)}$
and $\tilde{\romet}_{\hbar}$ introduced in (\ref{eq:def_t}) and
(\ref{eq:def_tt}) to each component. For brevity of notation, we
set 
\begin{equation}
{\boldt}_{\hbar}^{(n+1)}=\tilde{\boldt}_{\hbar}.\label{eq:q_for_n+1}
\end{equation}
 Then the set of operators $\{\boldt_{\hbar}^{(k)}\}_{k=0}^{n+1}$
are complete sets of mutually commuting projection operators .

The following Proposition shows that the projectors ${\boldt}_{\hbar}^{(k)}$
almost commute with the projector $(\boldI_{\hbar}\circ\boldI_{\hbar}^{*})$.
\begin{prop}
\label{pp1} There are constants $\epsilon>0$ and $C>0$, independent
of $\hbar$, such that the following holds: We have that 
\[
\left\Vert {\boldt}_{\hbar}^{(k)}\circ(\boldI_{\hbar}\circ\boldI_{\hbar}^{*})\right\Vert _{\bigoplus_{i=1}^{I_{\hbar}}\mathcal{H}_{\hbar}^{r,-}(\real^{2d})\rightarrow\bigoplus_{i=1}^{I_{\hbar}}\mathcal{H}_{\hbar}^{r,+}(\real^{2d})}<C,\quad\mbox{and }
\]
\[
\left\Vert (\boldI_{\hbar}\circ\boldI_{\hbar}^{*})\circ{\boldt}_{\hbar}^{(k)}\right\Vert _{\bigoplus_{i=1}^{I_{\hbar}}\mathcal{H}_{\hbar}^{r,-}(\real^{2d})\rightarrow\bigoplus_{i=1}^{I_{\hbar}}\mathcal{H}_{\hbar}^{r,+}(\real^{2d})}<C
\]
for $0\le k\le n$. (Hence the same statement holds as operators on
$\bigoplus_{i=1}^{I_{\hbar}}\mathcal{H}_{\hbar}^{r}(\real^{2d})$.)
Also we have, for the norm of the commutators, that 
\begin{equation}
\left\Vert \left[{\boldt}_{\hbar}^{(k)},(\boldI_{\hbar}\circ\boldI_{\hbar}^{*})\right]\right\Vert _{\bigoplus_{i=1}^{I_{\hbar}}\mathcal{H}_{\hbar}^{r,-}(\real^{2d})\to\bigoplus_{i=1}^{I_{\hbar}}\mathcal{H}_{\hbar}^{r,+}(\real^{2d})}\le C\hbar^{\epsilon}\label{eqn:com}
\end{equation}
for $0\le k\le n$. \end{prop}
\begin{proof}
From Lemma \ref{lm:IIbdd}, $\boldI_{\hbar}\circ\boldI_{\hbar}^{*}$
are bounded as operators from $\bigoplus_{i=1}^{I_{\hbar}}\mathcal{H}_{\hbar}^{r,+}(\real^{2d})$
to $\bigoplus_{i=1}^{I_{\hbar}}\mathcal{H}_{\hbar}^{r}(\real^{2d})$
(resp. from $\bigoplus_{i=1}^{I_{\hbar}}\mathcal{H}_{\hbar}^{r}(\real^{2d})$
to $\bigoplus_{i=1}^{I_{\hbar}}\mathcal{H}_{\hbar}^{r,-}(\real^{2d})$)
and the operator norm is bounded by a constant independent of $\hbar$.
From Lemma \ref{lm:pik_rpm}, so are the operators $\boldt^{(k)}$
as operators from $\bigoplus_{i=1}^{I_{\hbar}}\mathcal{H}_{\hbar}^{r,-}(\real^{2d})$
to $\bigoplus_{i=1}^{I_{\hbar}}\mathcal{H}_{\hbar}^{r}(\real^{2d})$
(resp. from $\bigoplus_{i=1}^{I_{\hbar}}\mathcal{H}_{\hbar}^{r}(\real^{2d})$
to $\bigoplus_{i=1}^{I_{\hbar}}\mathcal{H}_{\hbar}^{r,+}(\real^{2d})$).
Hence we obtain the first two inequalities. To prove (\ref{eqn:com}),
we take $\mathbf{u}=(u_{i})\in\bigoplus_{i=1}^{I_{\hbar}}\mathcal{H}_{\hbar}^{r,-}(\real^{2d})$
arbitrarily. From the definition, we have 
\begin{align*}
\boldt^{(k)}\circ(\boldI_{\hbar}\circ\boldI_{\hbar}^{*})(\mathbf{u}) & =\left(\sum_{i:U_{i}\cap U_{j}\neq\emptyset}\romet_{\hbar}^{(k)}\circ\multiplication(\psi_{j})\circ\prequantumL_{\kappa_{j,i}}\circ\multiplication(\chi_{\hbar})(u_{i})\right)_{j=1}^{I_{\hbar}}\intertext{and}(\boldI_{\hbar}\circ\boldI_{\hbar}^{*})\circ\boldt^{(k)}(\mathbf{u}) & =\left(\sum_{i:U_{i}\cap U_{j}\neq\emptyset}\multiplication(\chi_{j})\circ\prequantumL_{\kappa_{j,i}}\circ\multiplication(\chi_{\hbar})\circ\romet_{\hbar}^{(k)}(u_{i})\right)_{j=1}^{I_{\hbar}}.
\end{align*}
Applying Corollary \ref{cor:L_g_almost_commutes_with_t^k} and Corollary
\ref{cor:XT_exchange-1} to each components with the setting (\ref{eq:scrXG_setting})
and recalling Remark \ref{rem:Aij} and (\ref{eq:bound_on_i_sim_j}),
we obtain (\ref{eqn:com}). 
\end{proof}
The next Proposition stated for ${\boldF}_{\hbar}$ is now very close
to Theorem \ref{thm:More-detailled-description_of_spectrum}. 
\begin{prop}
\label{prop:key_proposition2} There are constants $\epsilon>0$ and
$C>0$ independent of $\hbar$ such that 
\begin{equation}
\left\Vert \left[{\boldF}_{\hbar},{\boldt}_{\hbar}^{(k)}\right]\right\Vert _{\bigoplus_{i=1}^{I_{\hbar}}\mathcal{H}_{\hbar}^{r}(\real^{2d})\to\bigoplus_{i=1}^{I_{\hbar}}\mathcal{H}_{\hbar}^{r}(\real^{2d})}\le C\hbar^{\epsilon}\quad\mbox{for }1\le k\le n+1.\label{eqn:com2}
\end{equation}
Further there exists a constant $C_{0}>0$, which is independent of
$f$, $V$ and $\hbar$, such that 
\begin{enumerate}
\item For $0\le k\le n+1$, it holds
\[
\!\!\!\left\Vert {\boldt}_{\hbar}^{(k)}\circ{\boldF}_{\hbar}\circ{\boldt}_{\hbar}^{(k)}\right\Vert _{\bigoplus_{i=1}^{I_{\hbar}}\mathcal{H}_{\hbar}^{r}(\real^{2d})\to\bigoplus_{i=1}^{I_{\hbar}}\mathcal{H}_{\hbar}^{r}(\real^{2d})}\le C_{0}\sup\left(|e^{V}|\|Df|_{E_{u}}\|_{\min}^{-k}|\det Df|_{E_{u}}|^{-1/2}\right)
\]

\item If $\boldu\in\bigoplus_{i=1}^{I_{\hbar}}\mathcal{H}_{\hbar}^{r}(\real^{2d})$
satisfies $\boldI_{\hbar}\circ\boldI_{\hbar}^{*}(\boldu)=\boldu$
and 
\[
\|\mathbf{u}-({\boldI_{\hbar}}\circ\boldI_{\hbar}^{*})\circ{\boldt}_{\hbar}^{(k)}(\mathbf{u})\|_{\mathcal{H}_{\hbar}^{r}}<\|\mathbf{u}\|_{\mathcal{H}_{\hbar}^{r}}/2\quad\mbox{for some \ensuremath{0\le k\le n},}
\]
 then we have 
\[
\|{\boldt}_{\hbar}^{(k)}\circ{\boldF}_{\hbar}\circ{\boldt}_{\hbar}^{(k)}(\mathbf{u})\|_{\mathcal{H}_{\hbar}^{r}}\ge C_{0}^{-1}\cdot\inf\left(|e^{V}|\|Df|_{E_{u}}\|_{\max}^{-k}|\det Df|_{E_{u}}|^{-1/2}\right)\cdot\|\mathbf{u}\|_{\mathcal{H}_{\hbar}^{r}}.
\]

\end{enumerate}
\end{prop}
\begin{proof}
We recall the argument in the proof of Lemma \ref{lm:bounded_F},
in particular, the expression (\ref{eq:expression_F_ij}) of the operator
$\boldF_{ij}$. Then we observe that, for each $i,j$ such that $i\to j$, 
\begin{description}
\item [{(i)}] From Proposition \ref{prop:prequatum_op_for_hyp_linear}
and Lemma \ref{lm:invariance_wrt_translation}, the projection operators
$\romet_{\hbar}^{(k)}$ for $0\le k\le n$ and $\tilde{\romet}_{\hbar}$
commute with the operator $\prequantumL^{(0)}$ and $\prequantumL^{(2)}$
(defined in (\ref{eq:expression_F_ij})).
\item [{(ii)}] From Lemma \ref{lm:invariance_wrt_translation}, the operator
$\prequantumL^{(0)}$ is a unitary operator on $\mathcal{H}_{\hbar}^{r}(\real^{2d})$
and also on $\mathcal{H}_{\hbar}^{r,\pm}(\real^{2d})$.
\item [{(iii)}] From Proposition \ref{lm:L_g_Y_almost_identity}, the operator
$\mathcal{L}^{(1)}$ extends to a bounded operator from $ $$\mathcal{H}_{\hbar}^{r,+}(\real^{2d})$
to $\mathcal{H}_{\hbar}^{r}(\real^{2d})$ (resp. from $ $$\mathcal{H}_{\hbar}^{r}(\real^{2d})$
to $\mathcal{H}_{\hbar}^{r,-}(\real^{2d}))$ and the operator norm
is bounded by a constant $C_{0}$, provided that $\hbar$ is sufficiently
small. 
\item [{(iv)}] Applying Proposition \ref{prop:prequatum_op_for_hyp_linear}
to $\prequantumL^{(2)}$, we see that the operator $\prequantumL^{(2)}$
is a bounded operator on $\mathcal{H}_{\hbar}^{r}(\real^{2d})$ and
that 
\[
C_{0}^{-1}\|B_{j,i}|_{E^{+}}\|^{-k}\cdot|\det B_{j,i}|_{E^{+}}|^{-1/2}\le\frac{\|\prequantumL^{(2)}u\|_{\mathcal{H}_{\hbar}^{r}(\real^{2d})}}{\|u\|_{\mathcal{H}_{\hbar}^{r}(\real^{2d})}}\le C_{0}\|B_{j,i}^{-1}|_{E^{+}}\|^{k}\cdot|\det B_{j,i}|_{E^{+}}|^{-1/2}
\]
for $0\neq u\in H'_{k}:=\Im\romet_{\hbar}^{(k)}$ and for $0\le k\le n$,
where $E^{+}=\real^{2d}\oplus\{0\}$. Further we have 
\[
\|\prequantumL^{(2)}u\|_{\mathcal{H}_{\hbar}^{r}(\real^{2d})}\le C_{0}\|B_{j,i}^{-1}|_{E^{+}}\|^{n+1}|\det B_{j,i}|_{E^{+}}|^{-1/2}\|u\|_{\mathcal{H}_{\hbar}^{r}(\real^{2d})}
\]
for $u\in\widetilde{H}':=\Im\tilde{t}_{\hbar}$.
\item [{(v)}] By simple comparison, we have 
\begin{multline*}
C_{0}^{-1}\cdot\inf\left(|e^{V}|\|Df|_{E_{u}}\|_{\max}^{-k}|\det Df|_{E_{u}}|^{-1/2}\right)\\
<e^{V_{j}}\cdot\|B_{j,i}^{-1}|_{E^{+}}\|^{k}\cdot|\det B_{j,i}|_{E^{+}}|^{-1/2}\\
<C_{0}\sup\left(|e^{V}|\|Df|_{E_{u}}\|_{\min}^{-k}|\det Df|_{E_{u}}|^{-1/2}\right).
\end{multline*}

\item [{(vi)}] Applying Lemma \ref{lm:L_g_almost_identity} to the setting
(\ref{eq:A_setting_for_X_and_G}) and Lemma \ref{lem:multipliacation_by_Psi_ji},
we have 
\begin{equation}
\|\prequantumL^{(1)}\circ\romet_{\hbar}^{(k)}-e^{V_{j}}\cdot\multiplication\left((\psi_{j,i}\cdot\chi_{\hbar})\circ B_{j,i}^{-1}\right)\circ\romet_{\hbar}^{(k)}\|_{\mathcal{H}_{\hbar}^{r}(\real^{2d})}\le C\hbar^{\epsilon}\quad\mbox{for \ensuremath{0\le k\le n}}\label{eq:LttL}
\end{equation}
with some positive constants $C$ and $\epsilon$ independent of $\hbar$. 
\item [{(vii)}] Applying Corollary \ref{cor:L_g_almost_commutes_with_t^k}
to the setting (\ref{eq:A_setting_for_X_and_G}), we have that 
\[
\|[\prequantumL^{(1)},\romet_{\hbar}^{(k)}]\|_{\mathcal{H}_{\hbar}^{r}(\real^{2d})}\le C\hbar^{\epsilon}\quad\mbox{for \ensuremath{0\le k\le n+1}}
\]
with setting $t_{\hbar}^{(n+1)}=\tilde{t}_{\hbar}$ for the case $k=n+1$.
This is true with $\mathcal{H}_{\hbar}^{r}(\real^{2d})$ replaced
by $\mathcal{H}_{\hbar}^{r,\pm}(\real^{2d})$.
\end{description}
From the observations (i), (ii), (iv) and (vii) above, it follows
\[
\|[\boldF_{j,i},\romet_{\hbar}^{(k)}]\|_{\mathcal{H}_{\hbar}^{r}(\real^{2d})\to\mathcal{H}_{\hbar}^{r}(\real^{2d})}\le C\hbar^{\epsilon}\quad\mbox{for }0\le k\le n+1.
\]
 This, together with (\ref{eq:arrow}), implies (\ref{eqn:com2}).

We prove Claim (1). Take $\boldu=(u_{i})_{i=1}^{I_{\hbar}}\in\bigoplus_{i=1}^{I_{\hbar}}\mathcal{H}_{\hbar}^{r}(\real^{2d})$
arbitrarily. Let $0\le k\le n+1$ and set 
\begin{equation}
v_{j,i}=\romet_{\hbar}^{(k)}\circ\boldF_{j,i}\circ\romet_{\hbar}^{(k)}(u_{i}),\qquad u_{j,i}=\psi_{j,i}\cdot\chi_{\hbar}\cdot u_{i}=(\psi_{j}\circ f_{j,i})\cdot\chi_{\hbar}\cdot u_{i}\label{eq:def_uij_vij}
\end{equation}
 for $1\le i,j\le I_{\hbar}$ such that $i\to j$. Suppose that $0\le k\le n$.
Then, using the expression (\ref{eq:expression_F_ij}) of $\boldF_{j,i}$,
we obtain, by (vi) and Corollary \ref{cor:XT_exchange-1}, 
\begin{align*}
\|v_{j,i}\|_{\mathcal{H}_{\hbar}^{r}(\real^{2d})} & =\|t_{\hbar}^{(k)}\circ\mathcal{L}^{(1)}\circ\mathcal{L}^{(2)}\circ t_{\hbar}^{(k)}u_{i}\|_{\mathcal{H}_{\hbar}^{r}(\real^{2d})}\\
 & =e^{V_{j}}\cdot\|\mathcal{L}^{(2)}\circ\multiplication(\psi_{j,i}\cdot\chi_{\hbar})\circ t_{\hbar}^{(k)}(u_{i})\|_{\mathcal{H}_{\hbar}^{r}(\real^{2d})}+\mathcal{O}(\hbar^{\epsilon}\|u_{i}\|_{\mathcal{H}_{\hbar}^{r}(\real^{2d})})\\
 & =e^{V_{j}}\cdot\|\mathcal{L}^{(2)}\circ t_{\hbar}^{(k)}(u_{j,i})\|_{\mathcal{H}_{\hbar}^{r}(\real^{2d})}+\mathcal{O}(\hbar^{\epsilon}\cdot\|u_{i}\|_{\mathcal{H}_{\hbar}^{r}(\real^{2d})})
\end{align*}
where $\mathcal{O}(\hbar^{\epsilon}\cdot\|u_{i}\|_{\mathcal{H}_{\hbar}^{r}(\real^{2d})})$
denotes positive terms that are bounded by $C\hbar^{\epsilon}\cdot\|u_{i}\|_{\mathcal{H}_{\hbar}^{r}(\real^{2d})}$.
Hence, from (iv), we get the estimates 
\begin{equation}
\|v_{j,i}\|_{\mathcal{H}_{\hbar}^{r}}\le C_{0}\cdot e^{V_{j}}\cdot\|B_{j,i}|_{E^{+}}\|_{\min}^{-k}\cdot|\det B_{j,i}|_{E^{+}}|^{-1/2}\cdot\|\romet_{\hbar}^{(k)}(u_{j,i})\|_{\mathcal{H}_{\hbar}^{r}(\real^{2d})}+\mathcal{O}(\hbar^{\epsilon}\cdot\|u_{i}\|_{\mathcal{H}_{\hbar}^{r}(\real^{2d})})\label{eq:vijuij}
\end{equation}
 and 
\begin{equation}
\|v_{j,i}\|_{\mathcal{H}_{\hbar}^{r}}\ge C_{0}^{-1}\cdot e^{V_{j}}\cdot\|B_{j,i}|_{E^{+}}\|_{\max}^{-k}\cdot|\det B_{j,i}|_{E^{+}}|^{-1/2}\cdot\|\romet_{\hbar}^{(k)}(u_{j,i})\|_{\mathcal{H}_{\hbar}^{r}(\real^{2d})}-\mathcal{O}(\hbar^{\epsilon}\cdot\|u_{i}\|_{\mathcal{H}_{\hbar}^{r}(\real^{2d})})\label{eq:vijuij2}
\end{equation}
for $0\le k\le n$. 

Actually the upper estimate (\ref{eq:vijuij}) can be strengthen by
modifying the argument above, so that it also holds for $k=n+1$.
(Note that the argument above is not true for $k=n+1$, because (\ref{eq:LttL})
does not hold in that case.) Indeed we can show that 
\begin{equation}
\|v_{j,i}\|_{\mathcal{H}_{\hbar}^{r,+}(\real^{2d})}\le C_{0}e^{V_{j}}\|B_{j,i}|_{E^{+}}\|_{\min}^{-k}|\det B_{j,i}|_{E^{+}}|^{-1/2}\cdot\|\romet_{\hbar}^{(k)}(u_{j,i})\|_{\mathcal{H}_{\hbar}^{r}(\real^{2d})}\label{eq:vijuij3}
\end{equation}
for \emph{all} $0\le k\le n+1$. Let $B_{0}:\real^{2d}\to\real^{2d}$
be the linear map defined by 
\[
B_{0}(x_{+},x_{-})=(\lambda_{0}\cdot x_{+},\lambda_{0}^{-1}\cdot x_{-})\quad\mbox{ for \ensuremath{(x_{+},x_{-})\in\real^{2d}=\real^{d}\oplus\real^{d}}}
\]
where $\lambda_{0}$ is an absolute constant greater than $9$. (Say
$\lambda_{0}=10$.) Then we write the operator $\boldF_{ij}$ as 
\[
\boldF_{ij}=\mathcal{L}^{(0)}\circ\mathcal{L}_{B_{0}}\circ\widetilde{\mathcal{L}}^{(1)}\circ\widetilde{\mathcal{L}}^{(2)}\quad\mbox{with setting }\widetilde{\mathcal{L}}^{(1)}=\prequantumL_{B_{0}^{-1}}\circ\mathcal{L}^{(1)}\circ\prequantumL_{B_{0}},\quad\widetilde{\mathcal{L}}^{(2)}=\prequantumL_{B_{0}^{-1}\circ B_{j,i}}.
\]
The operator $\prequantumL_{B_{0}}$ is a bounded operator from $\mathcal{H}_{\hbar}^{r,-}(\real^{2d})$
to $\mathcal{H}_{\hbar}^{r,+}(\real^{2d})$, from Lemma \ref{lm:boundedness_of_hyperbolic_linear_map}.
The operator $\widetilde{\mathcal{L}}^{(1)}$ is a bounded operator
from $\mathcal{H}_{\hbar}^{r}(\real^{2d})$ to $\mathcal{H}_{\hbar}^{r,-}(\real^{2d})$
and the operator norm is bounded by $C_{0}e^{V_{j}}$, from Proposition
\ref{pp:bdd_g}. And the observation (iv) holds true with $\mathcal{L}^{(2)}$
replaced by $\widetilde{\mathcal{L}}^{(2)}$. Hence we obtain (\ref{eq:vijuij3}).

From (\ref{eq:bound_on_i_sim_j}) in the choice of the coordinate
system $\{\kappa_{i}\}_{i=1}^{I_{\hbar}}$ (see Proposition \ref{prop:Local-charts})
and from Lemma \ref{lm:pseudo_local_property}, we have 
\[
\|\boldF\circ\boldt^{(k)}(\mathbf{u})\|_{\mathcal{H}_{\hbar}^{r}}^{2}=\left\Vert \left(\sum_{i:i\to j}v_{j,i}\right)_{j}\right\Vert _{\mathcal{H}_{\hbar}^{r}}^{2}\le C_{0}\sum_{i,j:i\to j}\|v_{j,i}\|_{\mathcal{H}_{\hbar}^{r}(\real^{2d})}^{2}+\mathcal{O}(\hbar^{\epsilon}\cdot\|\mathbf{u}\|_{\mathcal{H}_{\hbar}^{r}}^{2})
\]
 and 
\[
\sum_{i,j:i\to j}\|u_{j,i}\|_{\mathcal{H}_{\hbar}^{r}(\real^{2d})}^{2}\le C_{0}\|\mathbf{u}\|_{\mathcal{H}_{\hbar}^{r}}^{2},
\]
provided $\hbar$ is sufficiently small. Hence we obtain Claim 1 as
a consequence of (\ref{eq:vijuij3}) and the observation (v).
\begin{rem}
\label{rm:pseudolocal} Because of the inconvenient property of the
inner product $(\cdot,\cdot)_{\mathcal{H}_{\hbar}^{r}}$, noted in
the paragraph just before Lemma \ref{lm:pseudo_local_property}, the
two inequalities above are not an immediate consequence of the estimate
(\ref{eq:bound_on_i_sim_j}) on the intersection multiplicities of
the supports of $v_{j,i}$ and $u_{j,i}$. We have to use Lemma \ref{lm:pseudo_local_property}. 
\end{rem}
We prove Claim (2). We continue the argument in the proof of Claim
(1). Note that we already have the estimate (\ref{eq:vijuij2}) for
each $v_{j,i}$. Below we show that the functions $v_{j,i}$ do not
cancel out too much when we sum up them with respect to $i$ such
that $i\to j$. More precisely, we prove the estimate 
\begin{equation}
\sum_{i,i',j:i\to j,i'\to j,i\neq i'}\mathrm{Re}(v_{j,i},v_{j,i'})_{\mathcal{H}_{\hbar}^{r}(\real^{2d})}\ge-C\hbar^{\epsilon}\cdot\|\mathbf{u}\|_{\mathcal{H}^{r}}^{2}\label{eq:claim2}
\end{equation}
where $\sum_{i,i',j:i\to j,i'\to j,i\neq i'}$ denotes the sum over
$1\le i,i',j\le I_{\hbar}$ that satisfies $i\to j$, $i\to j'$ and
$i\neq i'$. For $1\le j\le I_{\hbar}$, let $I(j)$ be the set of
integers $1\le\ell\le I_{\hbar}$ such that there exists $1\le\ell',\ell''\le I_{\hbar}$
satisfying $U_{\ell}\cap U_{\ell'}\neq\emptyset$, $U_{\ell'}\cap U_{\ell''}\neq\emptyset$
and $\ell'\to j$. Note that we have
\begin{equation}
\max_{1\le j\le I_{\hbar}}\#I(j)\le C_{0}.\label{eq:cardinality_of_I(j)}
\end{equation}
Consider $1\le i,i',j\le I_{\hbar}$ that satisfies $i\to j$, $i\to j'$
and $i\neq i'$. We express $v_{j,i}$ as 
\[
v_{j,i}=\romet_{\hbar}^{(k)}\circ\multiplication(e^{V\circ\kappa_{j}}\cdot\psi_{j})\circ\prequantumL_{f_{j,i}}\circ\romet_{\hbar}^{(k)}\left(\psi_{i}\cdot\sum_{\ell\in I(j)}\prequantumL_{\kappa_{i,\ell}}(u_{\ell})\right).
\]
We can of course write $v_{j,i'}$ in the same form with $i$ replaced
by $i'$, but we rewrite it as 
\[
v_{j,i'}=\romet_{\hbar}^{(k)}\circ\multiplication(e^{V\circ\kappa_{j}}\cdot\psi_{j})\circ\prequantumL_{f_{j,i}}\circ\prequantumL_{{\kappa}_{i,i'}}\circ\romet_{\hbar}^{(k)}\left(\psi_{i'}\cdot\sum_{\ell\in I(j)}\prequantumL_{{\kappa}_{i',\ell}}(u_{\ell})\right).
\]
We change the order of operators on the right hand sides above, estimating
the commutators by Corollary \ref{cor:XT_exchange-1} and Corollary
\ref{cor:L_g_almost_commutes_with_t^k} and noting the relation $\kappa_{i,i'}\circ\kappa_{i',\ell}=\kappa_{i,\ell}$.
Then we get 
\[
\left\Vert v_{j,i}-\multiplication(e^{V\circ\kappa_{i'}}\cdot\psi_{j}\cdot\psi_{i}\circ f_{j,i}^{-1})\circ\prequantumL_{f_{j,i}}\circ\romet_{\hbar}^{(k)}\left(\sum_{\ell\in I(j)}\prequantumL_{\kappa_{i,\ell}}(u_{\ell})\right)\right\Vert _{\mathcal{H}_{\hbar}^{r}(\real^{2d})}\le C\hbar^{\epsilon}\sum_{\ell\in I(j)}\|u_{\ell}\|_{\mathcal{H}_{\hbar}^{r}(\real^{2d})}
\]
 and 
\[
\left\Vert v_{j,i'}-\multiplication(e^{V\circ\kappa_{i'}}\cdot\psi_{j}\cdot\psi_{i'}\circ f_{j,i'}^{-1})\circ\prequantumL_{f_{j,i}}\circ\romet_{\hbar}^{(k)}\left(\sum_{\ell\in I(j)}\prequantumL_{{\kappa}_{i,\ell}}(u_{\ell})\right)\right\Vert _{\mathcal{H}_{\hbar}^{r}(\real^{2d})}\le C\hbar^{\epsilon}\sum_{\ell\in I(j)}\|u_{\ell}\|_{\mathcal{H}_{\hbar}^{r}(\real^{2d})}.
\]
 Therefore, by Corollary \ref{cor:transpose}, we get 
\begin{align*}
\mathrm{Re}(v_{j,i},v_{j,i'}) & \ge\left\Vert \multiplication\left(e^{V\circ\kappa_{i'}}\psi_{j}\cdot\sqrt{\psi_{i}\circ f_{j,i}^{-1}\cdot\psi_{i'}\circ f_{j,i'}^{-1}}\right)\circ\prequantumL_{f_{j,i}}\circ\romet_{\hbar}^{(k)}\left(\sum_{\ell\in I(j)}\mathcal{L}_{{\kappa}_{i,\ell}}(u_{\ell})\right)\right\Vert _{\mathcal{H}_{\hbar}^{r}(\real^{2d})}^{2}\\
 & \qquad\qquad-C\hbar^{\epsilon}\cdot\sum_{\ell\in I(j)}\|u_{\ell}\|_{\mathcal{H}_{\hbar}^{r}(\real^{2d})}^{2}\\
 & \ge-C\hbar^{\epsilon}\cdot\sum_{\ell\in I(j)}\|u_{\ell}\|_{\mathcal{H}_{\hbar}^{r}(\real^{2d})}^{2}.
\end{align*}
Summing up the both sides of the inequality above for all $j,i,i'$
with $i\to j$, $i'\to j$ and $i\neq i'$ and using (\ref{eq:cardinality_of_I(j)}),
we obtain (\ref{eq:claim2}). 

From (\ref{eq:claim2}), (\ref{eq:vijuij2}) and the observation (iv),
we get
\begin{align}
 & \|{\boldF}_{\hbar}\circ\boldt^{(k)}(\mathbf{u})\|_{\mathcal{H}_{\hbar}^{r}}^{2}=\sum_{j}\sum_{i,i'}\Re(v_{j,i},v_{j,i'})_{\mathcal{H}_{\hbar}^{r}(\real^{2d})}\nonumber \\
 & \ge\sum_{i,j:i\to j}\|v_{j,i}\|_{\mathcal{H}_{\hbar}^{r}(\real^{2d})}^{2}-C\hbar^{\epsilon}\|\boldu\|_{\mathcal{H}_{\hbar}^{r}}^{2}\nonumber \\
 & \ge\sum_{i,j:i\to j}C_{0}^{-1}e^{V_{j}}\cdot\|B_{j,i}|_{E^{+}}\|_{\max}^{-k}|\det B_{j,i}|_{E^{+}}|^{-1/2}\|t_{\hbar}^{(k)}(u_{j,i})\|_{\mathcal{H}_{\hbar}^{r}(\real^{2d})}^{2}-C\hbar^{\epsilon}\|\boldu\|_{\mathcal{H}_{\hbar}^{r}}^{2}\nonumber \\
 & \ge C_{0}^{-1}\inf\left(|e^{V}|\|Df|_{E_{u}}\|_{\max}^{-k}|\det Df|_{E_{u}}|^{-1/2}\right)\sum_{i,j:i\to j}\|t_{\hbar}^{(k)}(u_{j,i})\|_{\mathcal{H}_{\hbar}^{r}(\real^{2d})}^{2}-C\hbar^{\epsilon}\|\boldu\|_{\mathcal{H}_{\hbar}^{r}}^{2}.\label{eq:estimate_on_F_t^k}
\end{align}
To finish the proof, we compare $\sum_{i,j:i\to j}\|t_{\hbar}^{(k)}(u_{j,i})\|_{\mathcal{H}_{\hbar}^{r}(\real^{2d})}^{2}$
and $\|\mathbf{u}\|_{\mathcal{H}_{\hbar}^{r}}^{2}$. To this end,
we use the assumptions in Claim (2), of course. From the assumption
and (\ref{eqn:com}) in Proposition \ref{pp1}, we have 
\begin{align}
\|\mathbf{u}\|_{\mathcal{H}_{\hbar}^{r}}\le2\cdot\|(\mathbf{I}_{\hbar}\circ\mathbf{I}_{\hbar}^{*})\circ\boldt_{\hbar}^{(k)}(\mathbf{u})\|_{\mathcal{H}_{\hbar}^{r}} & \le2\|\boldt_{\hbar}^{(k)}\circ(\mathbf{I}_{\hbar}\circ\mathbf{I}_{\hbar}^{*})(\boldu)\|_{\mathcal{H}_{\hbar}^{r}}+C\hbar^{\epsilon}\|\mathbf{u}\|_{\mathcal{H}_{\hbar}^{r}}\notag\\
 & =2\|\boldt_{\hbar}^{(k)}(\boldu)\|_{\mathcal{H}_{\hbar}^{r}}+C\hbar^{\epsilon}\|\mathbf{u}\|_{\mathcal{H}_{\hbar}^{r}}.\label{eq:Ftu3}
\end{align}
We also have, from Corollary \ref{cor:XT_exchange-1}, Lemma \ref{lm:pseudo_local_property}
and (\ref{eq:bound_on_i_sim_j}), that 
\begin{align*}
\|\boldt_{\hbar}^{(k)}(\boldu)\|_{\mathcal{H}_{\hbar}^{r}}^{2} & =\sum_{i}\|t_{\hbar}^{(k)}(u_{i})\|=\left\Vert \sum_{i}t_{\hbar}^{(k)}\left(\sum_{j:i\to j}u_{j,i}\right)\right\Vert _{\mathcal{H}_{\hbar}^{r}(\real^{2d})}^{2}\\
 & =\sum_{i}\left\Vert t_{\hbar}^{(k)}\left(\sum_{j:i\to j}\multiplication(\psi_{j,i}\cdot\chi_{\hbar})u_{i}\right)\right\Vert _{\mathcal{H}_{\hbar}^{r}(\real^{2d})}^{2}\\
 & \le\sum_{i}\left\Vert \sum_{j:i\to j}\multiplication(\psi_{j,i}\cdot\chi_{\hbar})\circ t_{\hbar}^{(k)}\left(u_{i}\right)\right\Vert _{\mathcal{H}_{\hbar}^{r}(\real^{2d})}^{2}+C\hbar^{\epsilon}\|\mathbf{u}\|_{\mathcal{H}_{\hbar}^{r}}^{2}\\
 & \le C_{0}\sum_{i,j:i\to j}\|\multiplication(\psi_{j,i}\cdot\chi_{\hbar})\circ t_{\hbar}^{(k)}(u_{j})\|_{\mathcal{H}_{\hbar}^{r}(\real^{2d})}^{2}+C\hbar^{\epsilon}\|\mathbf{u}\|_{\mathcal{H}_{\hbar}^{r}}^{2}\\
 & \leq C_{0}\sum_{i,j:i\to j}\|t_{\hbar}^{(k)}(u_{j,i})\|_{\mathcal{H}_{\hbar}^{r}(\real^{2d})}^{2}+C\hbar^{\epsilon}\|\mathbf{u}\|_{\mathcal{H}_{\hbar}^{r}}^{2}.
\end{align*}
We conclude Claim (2) from these inequalities and (\ref{eq:estimate_on_F_t^k}).
\end{proof}

\subsection{\label{sub:Proofs-of-Theorem_more_details}Proof of Theorem \ref{thm:More-detailled-description_of_spectrum} }

We finish the proof of Theorem \ref{thm:More-detailled-description_of_spectrum}.
Actually we have almost done with the essential part of the proof.
Below we give a formal argument to complete it. Let us begin with
introducing the operators 
\begin{equation}
\checktau_{\hbar}^{(k)}=\boldI_{\hbar}^{*}\circ\boldt_{\hbar}^{(k)}\circ\boldI_{\hbar}:\mathcal{H}_{\hbar}^{r}(P)\to\mathcal{H}_{\hbar}^{r}(P)\label{eq:def_Tau_chech}
\end{equation}
 for $0\le k\le n+1$. From Proposition \ref{pp1}, these are bounded
operators with operator norms bounded by a constant $C$ independent
of $\hbar$, and satisfy 
\begin{align}
 & \checktau_{\hbar}^{(1)}+\checktau_{\hbar}^{(2)}+\cdots+\checktau_{\hbar}^{(n)}+\checktau_{\hbar}^{(n+1)}=\mathrm{Id},\label{eq:sum_of_check_tau}\\
 & \|\checktau_{\hbar}^{(k)}\circ\checktau_{\hbar}^{(k)}-\checktau_{\hbar}^{(k)}\|_{\mathcal{H}_{\hbar}^{r}(P)}\le C\hbar^{\epsilon}\quad\mbox{for \ensuremath{0\le k\le n+1},}\quad\mbox{\mbox{and }}\label{eq:check_tau_almost_projection}\\
 & \|\checktau^{(k)}\circ\checktau^{(k')}\|_{\mathcal{H}_{\hbar}^{r}(P)}\le C\hbar^{\epsilon}\quad\mbox{for \ensuremath{0\le k,k'\le n+1}with \ensuremath{k\neq k'}}\label{eq:cTcT2}
\end{align}
 for some constants $\epsilon>0$ and $C>0$. 
\begin{lem}
\label{lem:trace_norm_tau} The operators $\check{\tau}_{\hbar}^{(k)}$,
$0\le k\le n$, are trace class operators. There exist constants $\epsilon>0$
and $C>0$, independent of $\hbar$, such that 
\[
\|\check{\tau}_{\hbar}^{(k)}:\mathcal{H}_{\hbar}^{r}(P)\to\mathcal{H}_{\hbar}^{r}(P)\|_{tr}\le C\hbar^{-d}
\]
and that 
\[
\|\check{\tau}_{\hbar}^{(k)}-\check{\tau}_{\hbar}^{(k)}\circ\check{\tau}_{\hbar}^{(k)}:\mathcal{H}_{\hbar}^{r}(P)\to\mathcal{H}_{\hbar}^{r}(P)\|_{tr}\le C\hbar^{-d+\epsilon}
\]
\end{lem}
\begin{proof}
It is enough to prove the corresponding statement for 
\[
\boldI_{\hbar}\circ\check{\tau_{\hbar}^{(k)}}\circ\boldI_{\hbar}^{*}=(\boldI_{\hbar}\circ\boldI_{\hbar}^{*})\circ\boldt_{\hbar}^{(k)}\circ(\boldI_{\hbar}\circ\boldI_{\hbar}^{*}):\bigoplus_{i=1}^{I_{\hbar}}\mathcal{H}_{\hbar}^{r}(\real^{2d})\to\bigoplus_{i=1}^{I_{\hbar}}\mathcal{H}_{\hbar}^{r}(\real^{2d}).
\]
Applying Corollary \ref{cor:XT_pm} and Proposition \ref{pp:bdd_g}
(to deal with the non-linearity of the coordinate change transformations),
we see that each component of this operator is a trace class operator
and hence so is itself. Recalling Proposition \ref{prop:Local-charts}(3),
we obtain the first claim by summing up the trace norm of the components.
To prove the second claim, we compare the operator above with 
\[
(\boldI_{\hbar}\circ\check{\tau_{\hbar}^{(k)}}\circ\boldI_{\hbar}^{*})\circ(\boldI_{\hbar}\circ\check{\tau_{\hbar}^{(k)}}\circ\boldI_{\hbar}^{*})=(\boldI_{\hbar}\circ\boldI_{\hbar}^{*})\circ\boldt_{\hbar}^{(k)}\circ(\boldI_{\hbar}\circ\boldI_{\hbar}^{*})\circ\boldt_{\hbar}^{(k)}\circ(\boldI_{\hbar}\circ\boldI_{\hbar}^{*}):\bigoplus_{i=1}^{I_{\hbar}}\mathcal{H}_{\hbar}^{r}(\real^{2d})\to\bigoplus_{i=1}^{I_{\hbar}}\mathcal{H}_{\hbar}^{r}(\real^{2d}).
\]
We need to prove that the trace norm of the commutator
\[
\left[(\boldI_{\hbar}\circ\boldI_{\hbar}^{*}),\boldt_{\hbar}^{(k)}\right]:\bigoplus_{i=1}^{I_{\hbar}}\mathcal{H}_{\hbar}^{r}(\real^{2d})\to\bigoplus_{i=1}^{I_{\hbar}}\mathcal{H}_{\hbar}^{r}(\real^{2d})
\]
is bounded by $C\hbar^{-d+\epsilon}$. It is easy to obtain such estimate
by using Corollary \ref{cor:XT_exchange-Tr} and Lemma \ref{lm:L_g_almost_identity-Tr}
to exchange the order of operators and also using Proposition \ref{prop:Local-charts}(3).
\end{proof}
Now we will modify the operators $\check{\tau}_{\hbar}^{(k)}$, $0\le k\le n+1$,
to get the projection operators $\tau_{\hbar}^{(k)}$, $0\le k\le n$,
and $\tilde{\tau}_{\hbar}=\tau_{\hbar}^{(n+1)}$ in the statement
of the theorem. The estimate (\ref{eq:check_tau_almost_projection})
implies that the spectral set of the operator $\checktau_{\hbar}^{(k)}\circ\checktau_{\hbar}^{(k)}-\checktau_{\hbar}^{(k)}$
is contained in the disk $|z|\le C\hbar^{\epsilon}$. By the spectral
mapping theorem \cite[Part I, VII.3.11]{DunfordSchwartz1}, the spectral
set of the operators $\checktau_{\hbar}^{(k)}$ is contained in the
union of two small disks around $0$ and $1$: 
\begin{equation}
\mathbb{D}(0,C\hbar^{\epsilon})\cup\mathbb{D}(1,C\hbar^{\epsilon})\qquad\qquad\mbox{where }\mathbb{D}(z,r):=\{w\in\complex\mid|w-x|<r\}.\label{eq:twodisks}
\end{equation}
For $0\le k\le n+1$, let $\hattau_{\hbar}^{(k)}$ be the spectral
projector of $\checktau_{\hbar}^{(k)}$ for the part of its spectral
set contained in $\mathbb{D}(1,C\hbar^{\epsilon})$. This is of finite
rank because of compactness of $\checktau_{\hbar}^{(k)}$. The next
lemma should be easy to prove. (We provide a proof in the appendix
for completeness.) 
\begin{lem}
\label{lm:pi_P} There is a constant $C>0$ independent of $\hbar$
such that 
\[
\|\checktau_{\hbar}^{(k)}-\hattau_{\hbar}^{(k)}\|_{\mathcal{H}_{\hbar}^{r}(P)}\le C\hbar^{\epsilon}\quad\mbox{and}\quad\|\checktau_{\hbar}^{(k)}-\hattau_{\hbar}^{(k)}:\mathcal{H}_{\hbar}^{r}(P)\to\mathcal{H}_{\hbar}^{r}(P)\|_{tr}\le C\hbar^{-d+\epsilon}
\]
 for some $C>0$ independent of $\hbar$.
\end{lem}
Thus we get the set of projection operators $\hattau_{\hbar}^{(k)}$
for $0\le k\le n+1$, which approximate $\check{\tau}_{\hbar}^{(k)}$.
As consequences of (\ref{eq:sum_of_check_tau}) and (\ref{eq:cTcT2}),
we have
\begin{equation}
\|\mathrm{Id}-(\hattau_{\hbar}^{(0)}+\hattau_{\hbar}^{(1)}+\cdots+\hattau_{\hbar}^{(n)}+\hattau_{\hbar}^{(n+1)})\|_{\mathcal{H}_{\hbar}^{r}(P)}\le C\hbar^{\epsilon}\label{eq:sum_of_hat_tau}
\end{equation}
 and 
\begin{equation}
\|\hattau_{\hbar}^{(k)}\circ\hattau_{\hbar}^{(k')}\|_{\mathcal{H}_{\hbar}^{r}(P)}\le C\hbar^{\epsilon}\quad\mbox{ if \ensuremath{k\neq k'}}.\label{eq:orthogonality_between_hat_tau}
\end{equation}
We set $\mathcal{H}_{k}:=\mathrm{Im\:}\hat{\tau}_{\hbar}^{(k)}$ for
$0\le k\le n+1$ and put $\widetilde{\mathcal{H}}=\mathcal{H}_{n+1}$.
We have 
\begin{lem}
The Hilbert space $\mathcal{H}_{\hbar}^{r}(P)$ is decomposed into
the direct sum: 
\[
\mathcal{H}_{\hbar}^{r}(P)=\mathcal{H}_{0}\oplus\mathcal{H}_{1}\oplus\mathcal{H}_{2}\oplus\cdots\oplus\mathcal{H}_{n}\oplus\widetilde{\mathcal{H}}.
\]
\end{lem}
\begin{proof}
Since the sum $\hattau_{\hbar}^{(0)}+\hattau_{\hbar}^{(1)}+\cdots+\hattau_{\hbar}^{(n)}+\hattau_{\hbar}^{(n+1)}$
is invertible from (\ref{eq:sum_of_hat_tau}), we can set 
\begin{equation}
\tau_{\hbar}^{(k)}:=\hattau_{\hbar}^{(k)}(\hattau_{\hbar}^{(0)}+\hattau_{\hbar}^{(1)}+\cdots+\hattau_{\hbar}^{(n)}+\hattau_{\hbar}^{(n+1)})^{-1}\quad\mbox{for }0\le k\le n+1.\label{eq:label}
\end{equation}
We can express any $v\in\mathcal{H}_{\hbar}^{r}(P)$ as 
\[
v=\sum_{k=1}^{n+1}v_{k}\quad\mbox{with }v_{k}:=\tau_{\hbar}^{(k)}(v)\in\mathcal{H}_{k}.
\]
Thus the subspaces $\mathcal{H}_{k}$ for $0\le k\le n+1$ span the
whole space $\mathcal{H}_{\hbar}^{r}(P)$. Uniqueness of such expression
follows from (\ref{eq:orthogonality_between_hat_tau}).
\end{proof}
From the argument in the proof above, the operator $\tau_{\hbar}^{(k)}:\mathcal{H}_{\hbar}^{r}(P)\to\mathcal{H}_{k}(P)$
for $0\le k\le n+1$ in (\ref{eq:label}) are the projections to the
subspace $\mathcal{H}_{k}$ along other subspaces. Clearly we have
\begin{equation}
\|\tau_{\hbar}^{(k)}-\hattau_{\hbar}^{(k)}\|_{\mathcal{H}_{\hbar}^{r}(P)}\le C\hbar^{\epsilon}\quad\mbox{and hence}\quad\|\tau_{\hbar}^{(k)}-\checktau_{\hbar}^{(k)}\|_{\mathcal{H}_{\hbar}^{r}(P)}\le C\hbar^{\epsilon}.\label{eq:pi_cT1}
\end{equation}
and also 
\[
\|\tau_{\hbar}^{(k)}-\hattau_{\hbar}^{(k)}:\mathcal{H}_{\hbar}^{r}(P)\to\mathcal{H}_{\hbar}^{r}(P)\|_{tr}\le C\hbar^{-d+\epsilon}\quad\mbox{and hence}\quad\|\tau_{\hbar}^{(k)}-\checktau_{\hbar}^{(k)}:\mathcal{H}_{\hbar}^{r}(P)\to\mathcal{H}_{\hbar}^{r}(P)\|_{tr}\le C\hbar^{-d+\epsilon}.
\]

\subsubsection{Proof of Claim (1)}

For the external band $k=0$, Claim 1 is written in Corollary \ref{cor:Weyl_law}.
But the same argument works for every $0\leq k\leq n$, so we give
it here. Since we have
\[
\mathrm{rank}\tau_{\hbar}^{(k)}=\mathrm{Tr\:\tau_{\hbar}^{(k)}}=\mbox{\ensuremath{\mathrm{Tr}}}\:\check{\tau}_{\hbar}^{(k)}+\mathcal{O}(\hbar^{-d+\epsilon}),
\]
it is enough to prove the claim that 
\[
\left|\mathrm{Tr}\:\check{\tau}_{\hbar}^{(k)}-\frac{r(k,d)\cdot Vol_{\omega}(M)}{(2\pi\hbar)^{d}}\right|<C\hbar^{-d+\epsilon}.
\]
From the definition, we have that 
\[
\mathrm{Tr}\,\check{\tau}_{\hbar}^{(k)}=\mathrm{Tr}(\boldI_{\hbar}\circ\boldt_{\hbar}^{(k)}\circ\boldI_{\hbar}^{*}):=\sum_{i=1}^{I_{\hbar}}\mathrm{Tr}(\boldI_{\hbar}\circ\boldt_{\hbar}^{(k)}\circ\boldI_{\hbar}^{*})_{i,i}.
\]
From Lemma \ref{lm:trace_basic} and Corollary \ref{cor:XT_exchange-Tr},
we see that 
\[
\left|\mathrm{Tr}(\boldI_{\hbar}\circ\boldt_{\hbar}^{(k)}\circ\boldI_{\hbar}^{*})_{i,i}-\frac{r(k,d)}{(2\pi\hbar)^{d}}\int\psi_{i}dx\right|\le C\hbar^{-\theta d+\epsilon}
\]
Since $I_{\hbar}<C\hbar^{-(1-\theta)d}$ and $\sum_{i}\psi_{i}\circ\kappa_{i}^{-1}\equiv1$,
we obtain 
\[
\left|\mathrm{Tr}\,\check{\tau}_{\hbar}^{(k)}-\frac{r(k,d)\cdot Vol_{\omega}(M)}{(2\pi\hbar)^{d}}\right|<C\hbar^{-(1-\theta)d}\cdot\hbar^{-\theta d+\epsilon}=C\hbar^{-d+\epsilon}.
\]

\subsubsection{Proof of Claim (2)}

For the proof of Claim (2)-(5), it is enough to prove the statements
with $\tau_{\hbar}^{(k)}$ replaced by $\check{\tau}_{\hbar}^{(k)}$
because we have (\ref{eq:pi_cT1}). The operator norm of $\checktau_{\hbar}^{(k)}:\mathcal{H}_{\hbar}^{r}(P)\to\mathcal{H}_{\hbar}^{r}(P)$
is bounded by a constant independent of $\hbar$ from Proposition~\ref{pp1},
as we noted in the beginning of this subsection.

\subsubsection{Proof of Claim (3) and (4)}

From the definition of the operators $\mathbf{F}_{\hbar}$ and $\check{\tau}_{\hbar}^{(k)}$,
the following diagram commutes: 
\[
\begin{CD}\bigoplus_{i=1}^{I_{\hbar}}\mathcal{H}_{\hbar}^{r}(\real^{2d})@>\boldI_{\hbar}\circ\boldI_{\hbar}^{*}\circ{\boldt}_{\hbar}^{(k)}\circ{\boldF}_{\hbar}\circ{\boldt}_{\hbar}^{(k')}>>\bigoplus_{i=1}^{I_{\hbar}}\mathcal{H}_{\hbar}^{r}(\real^{2d})\\
@AA\boldI_{\hbar}A@AA\boldI_{\hbar}A\\
\mathcal{H}_{\hbar}^{r}(P)@>\checktau_{\hbar}^{(k)}\circ\hat{F}_{\hbar}\circ\checktau_{\hbar}^{(k')}>>\mathcal{H}_{\hbar}^{r}(P)
\end{CD}
\]
Since the operator $\boldI_{\hbar}$ in the vertical direction is
an isometric embedding, we have 
\begin{align*}
\|\checktau_{\hbar}^{(k)}\circ & \hat{F}_{\hbar}\circ\checktau_{\hbar}^{(k')}\|_{\mathcal{H}_{\hbar}^{r}(P)}\le\|\boldI_{\hbar}\circ\boldI_{\hbar}^{*}\circ{\boldt}_{\hbar}^{(k)}\circ{\boldF}_{\hbar}\circ{\boldt}_{\hbar}^{(k')}\|_{\mathcal{H}_{\hbar}^{r}}=\|(\boldI_{\hbar}\circ\boldI_{\hbar}^{*}\circ{\boldt}_{\hbar}^{(k)})\circ({\boldt}_{\hbar}^{(k)}\circ{\boldF}_{\hbar}\circ{\boldt}_{\hbar}^{(k')})\|_{\mathcal{H}_{\hbar}^{r}}.
\end{align*}
From Proposition \ref{pp1}, we have 
\[
\left\Vert \boldI_{\hbar}\circ\boldI_{\hbar}^{*}\circ{\boldt}_{\hbar}^{(k)}\right\Vert _{\mathcal{H}_{\hbar}^{r}}\le C_{0}\quad\mbox{ for \ensuremath{0\le k\le n}}
\]
where $C_{0}$ is a constant independent of $\hbar$, $f$ and $V$.
These estimates give also 
\[
\left\Vert \boldI_{\hbar}\circ\boldI_{\hbar}^{*}\circ{\boldt}_{\hbar}^{(n+1)}\right\Vert _{\bigoplus_{i=1}^{I_{\hbar}}\mathcal{H}_{\hbar}^{r,+}(\real^{2d})\to\bigoplus_{i=1}^{I_{\hbar}}\mathcal{H}_{\hbar}^{r}(\real^{2d})}\le C_{0},
\]
because of the relation $\boldt_{\hbar}^{(n+1)}=\tilde{\boldt}_{\hbar}=\mathrm{Id}-\sum_{k=0}^{n}\boldt_{\hbar}^{(k)}$
and Lemma \ref{lm:IIbdd}. Now Claim (4) is an immediate consequence
of Proposition~\ref{prop:key_proposition2} (1). From Proposition
\ref{prop:key_proposition2} and Lemma \ref{lm:bounded_F}, we have
\[
\|{\boldt}_{\hbar}^{(k)}\circ{\boldF}_{\hbar}\circ{\boldt}_{\hbar}^{(k')}\|_{\mathcal{H}_{\hbar}^{r}}\le\|{\boldt}_{\hbar}^{(k)}\circ{\boldt}_{\hbar}^{(k')}\circ{\boldF}_{\hbar}\|_{\mathcal{H}_{\hbar}^{r}}+C\hbar^{\epsilon}=C\hbar^{\epsilon}
\]
 if $k\neq k'$ and hence Claim (3) follows.

\subsubsection{Proof of Claim (5)}

Take $0\neq u\in\mathcal{H}_{k}=\Im\tau_{\hbar}^{(k)}$ for $0\le k\le n$
arbitrarily and set $\boldu=\boldI_{\hbar}(u)$. Then we have, from
(\ref{eq:pi_cT1}), that 
\begin{align*}
\|\boldu-(\boldI_{\hbar}\circ\boldI_{\hbar}^{*})\circ\boldt_{\hbar}^{(k)}(\boldu)\|_{\mathcal{H}_{\hbar}^{r}} & =\|\boldI_{\hbar}(u)-\boldI_{\hbar}\circ\checktau_{\hbar}^{(k)}(u)\|_{\mathcal{H}_{\hbar}^{r}}=\|u-\checktau_{\hbar}^{(k)}(u)\|_{\mathcal{H}_{\hbar}^{r}(P)}\\
 & =\|(\tau_{\hbar}^{(k)}-\checktau_{\hbar}^{(k)})u\|_{\mathcal{H}_{\hbar}^{r}(P)}\le C\hbar^{\epsilon}\|u\|_{\mathcal{H}_{\hbar}^{r}(P)}=C\hbar^{\epsilon}\|\boldu\|_{\mathcal{H}_{\hbar}^{r}}
\end{align*}
Hence we can apply the second claim in Proposition \ref{prop:key_proposition2}
to $\boldu$ and obtain Claim (5), noting that $\|\boldF_{\hbar}\boldu\|_{\mathcal{H}_{\hbar}^{r}}=\|F_{\hbar}u\|_{\mathcal{H}_{\hbar}^{r}(P)}$
by definition.

\subsection{\label{sub:Proof-of-Theorem1.40}Proof of Theorem \ref{prop:express_Op_Psi}}

We will derive Theorem \ref{prop:express_Op_Psi} as a consequence
of Lemma \ref{lm:trace_basic} on local charts, gluing together the
expression given in (\ref{eq:trace_formula_local}). Let $\varphi\in S_{\delta}$
be a symbol, i.e. a $\hbar$-family of smooth functions $\left(\varphi_{\hbar}\right)_{\hbar}$
on $M$ with regularity estimates as in Definition \ref{def:class_symbols}.
Below we write $\varphi$ for $\varphi_{\hbar}$. 

Recall the definition of the operators $\boldI_{\hbar}$ and $\boldI_{\hbar}^{*}$
in (\ref{eq:local_data_ui}) and (\ref{eq:expression_of_I*}). We
define 
\[
\boldI_{i,\hbar}:C_{N}^{\infty}(P)\to C_{0}^{\infty}(\mathbb{D}(\hbar^{1/2-\theta})),\quad u\mapsto u_{i}(x)=\psi_{i}(x)\cdot u(\tau_{i}(\kappa_{i}(x)))
\]
and 
\[
\boldI_{i,\hbar}^{*}:C_{0}^{\infty}(\mathbb{D}(\hbar^{1/2-\theta}))\to C_{N}^{\infty}(P),\quad u\mapsto e^{N\alpha_{i}(p)}\cdot\chi_{\hbar}\cdot u
\]
for $1\le i\le I_{\hbar}$. Then we have
\begin{eqnarray}
\mathcal{M}(\varphi)\circ\checktau_{\hbar}^{(0)} & \underset{(\ref{eq:def_Tau_chech})}{=} & \mathcal{M}(\varphi)\circ\boldI_{\hbar}^{*}\circ\boldt_{\hbar}^{(0)}\circ\boldI_{\hbar}\nonumber \\
 & \underset{(\ref{eq:local_data_ui}),(\ref{eq:expression_of_I*})}{=} & \mathcal{M}(\varphi)\circ\sum_{i=1}^{I_{\hbar}}\boldI_{i,\hbar}^{*}\circ t_{\hbar}^{(0)}\circ\mathcal{M}(\psi_{i})\circ\boldI_{i,\hbar}\label{eq:M_tau_c}\\
 & = & \sum_{i=1}^{I_{\hbar}}\boldI_{i,\hbar}^{*}\circ\mathcal{M}(\varphi\circ\kappa_{i})\circ t_{\hbar}^{(0)}\circ\mathcal{M}(\psi_{i})\circ\boldI_{i,\hbar}
\end{eqnarray}
In the last line above, we have used the fact that $\mathcal{M}(\varphi)\circ\boldI_{i,\hbar}^{*}=\boldI_{i,\hbar}^{*}\circ\mathcal{M}(\varphi\circ\kappa_{i})$.
From Definition \ref{def:class_symbols}, the family of functions
$\scrX_{\hbar}:=\{\chi_{\hbar}\cdot\varphi\circ\kappa_{i};1\le i\le I_{\hbar}\}$
satisfies the conditions in Setting I in Subsection \ref{Setting-I}.
Let $\hat{\pi}_{0}(\nu)$ be the rank one projection operator $\hat{\pi}_{\alpha}(\nu)$
defined in (\ref{eq:def_projector_pi_alpha}) for $\alpha=0\in\mathbb{N}^{0}$.
We continue (\ref{eq:M_tau_c}) and get (we will justify the second
line below)
\begin{eqnarray}
\mathcal{M}\left(\varphi\right)\circ\checktau_{\hbar}^{(0)} & \underset{}{=} & \sum_{i=1}^{I_{\hbar}}\mathbf{I}_{i,\hbar}^{*}\circ\left(\mathcal{M}(\varphi\circ\kappa_{i})\circ t_{\hbar}^{(0)}\circ\mathcal{M}(\psi_{i})\right)\circ\mathbf{I}_{i,\hbar}\nonumber \\
 & \underset{(\ref{eq:Norm_estimate})}{=} & \sum_{i=1}^{I_{\hbar}}\mathbf{I}_{i,\hbar}^{*}\circ\left(\int_{\mathbb{R}^{2d}}\varphi\circ\kappa_{i}\left(\nu\right)\hat{\pi}_{0}\left(\nu\right)\frac{d\nu}{\left(2\pi\hbar\right)^{d}}\right)\circ\mathcal{M}(\psi_{i})\circ\mathbf{I}_{i,\hbar}^{*}+O\left(\hbar^{\theta}\right)\nonumber \\
 & = & \frac{1}{\left(2\pi\hbar\right)^{d}}\int_{M}\sum_{i=1}^{I_{\hbar}}\varphi\left(x\right)\cdot\chi_{\hbar}(\kappa_{i}^{-1}(x))\cdot\hat{\pi}_{i}^{\left(1\right)}\left(x\right)dx+O\left(\hbar^{\theta}\right)\label{eq:Mphi_tau_check_suite}
\end{eqnarray}
where we have put
\begin{equation}
\hat{\pi}_{i}^{\left(1\right)}\left(x\right):=\mathbf{I}_{i,\hbar}^{*}\circ\hat{\pi}_{0}\left(\kappa_{i}^{-1}\left(x\right)\right)\circ\mathcal{M}(\psi_{i})\circ\mathbf{I}_{i,\hbar},\qquad i\in\left\{ 1,\ldots I_{\hbar}\right\} ,x\in M.\label{eq:def_pi_1}
\end{equation}
and $O\left(\hbar^{\theta}\right)$ denotes the error term whose operator
norm is bounded by $C\hbar^{\epsilon}$ with $C$ a constant independent
of $\hbar$. In order to justify this small error term, for every
$i\in\left\{ 1,\ldots I_{\hbar}\right\} $, let 
\[
T_{i}:=\mathbf{I}_{i,\hbar}^{*}\circ\left(\mathcal{M}(\varphi)\circ t_{\hbar}^{(0)}-\left(\int_{\mathbb{R}^{2d}}\varphi_{i}\left(\nu\right)\hat{\pi}_{0}\left(\nu\right)\frac{d\nu}{\left(2\pi\hbar\right)^{d}}\right)\right)\circ\mathcal{M}(\psi_{i})\circ\mathbf{I}_{i,\hbar}.
\]
For every $i\in\left\{ 1,\ldots I_{\hbar}\right\} $, we have from
(\ref{eq:Norm_estimate}) that $\left\Vert T_{i}\right\Vert =O\left(\hbar^{\theta}\right)$.
Due to truncations functions, for every $0\le i\le I_{\hbar}$, we
have that $T_{i}\circ T_{k}^{*}=0$ and $T_{i}^{*}\circ T_{k}=0$
except for a finite number (bounded uniformly in $\hbar$) of $k$.
Hence from the discrete version of the Cotlar-Stein Lemma%
\footnote{\begin{lem}
\textbf{\textup{\label{lem:Discrete-version-of-Cotlar-Stein}``Discrete
version of the Cotlar-Stein Lemma''}}\textup{: If $\left(T_{j}\right)_{j}$
is a family of bounded operators, if $A:=\sup_{j}\sum_{k}\left\Vert T_{j}T_{k}^{*}\right\Vert ^{1/2}<\infty$
and $B:=\sup_{j}\sum_{k}\left\Vert T_{j}^{*}T_{k}\right\Vert ^{1/2}<\infty$
then $\sum_{j}T_{j}$ converges in the strong operator topology and
$\left\Vert \sum_{j}T_{j}\right\Vert \leq\sqrt{AB}$. }\end{lem}
} we deduce that $\left\Vert \sum_{j}T_{j}\right\Vert =O\left(\hbar^{\theta}\right)$.
This justifies the second line of (\ref{eq:Mphi_tau_check_suite}).

The operator $\hat{\pi}_{i}^{\left(1\right)}\left(x\right)$ is of
rank one because $\hat{\pi}_{0}\left(\nu\right)$ is a rank one projector
from Lemma \ref{lm:trace_basic}. Its trace is
\begin{eqnarray}
\mathrm{Tr}\left(\hat{\pi}_{j}^{\left(1\right)}\left(x\right)\right) & = & \left(\chi_{\hbar}\circ\kappa_{i}^{-1}\right)\left(x\right)\cdot\mathrm{Tr}\left(\hat{\pi}_{0}\left(\kappa_{i}^{-1}\left(x\right)\right)\mathcal{M}(\psi_{i})\right)\nonumber \\
 & = & \psi_{i}\left(\kappa_{i}^{-1}\left(x\right)\right)\cdot\mathrm{Tr}\left(\hat{\pi}_{0}\left(\kappa_{i}^{-1}\left(x\right)\right)\right)+O\left(\hbar^{\theta}\right)\quad\mbox{by Lemma \ref{lem:kernel_of_pi_nu}}\nonumber \\
 & = & \psi_{i}\left(\kappa_{i}^{-1}\left(x\right)\right)+O\left(\hbar^{\theta}\right)\label{eq:trace_pi_1}
\end{eqnarray}
Hence we have
\begin{equation}
\left(\hat{\pi}_{i}^{\left(1\right)}\left(x\right)\right)^{2}=\psi_{i}\left(\kappa_{i}^{-1}\left(x\right)\right)\cdot\hat{\pi}_{i}^{\left(1\right)}\left(x\right)+O\left(\hbar^{\theta}\right)\label{eq:pi_1_square}
\end{equation}

We define the operator $\hat{\pi}^{\left(2\right)}\left(x\right)$
by summing the operators $\hat{\pi}_{j}^{\left(1\right)}\left(x\right)$
over the different charts:
\[
\hat{\pi}^{\left(2\right)}\left(x\right):=\sum_{1\le i\le I_{\hbar}}\hat{\pi}_{i}^{\left(1\right)}\left(x\right)
\]
Notice that $\hat{\pi}_{i}^{\left(1\right)}\left(x\right)=0$ if $\left(\chi_{\hbar}\circ\kappa_{i}^{-1}\right)\left(x\right)=0$.
Hence in the previous equation, the sum over $j$ contains only a
finite number of terms and this number is bounded uniformly with respect
to $\hbar$. 

In order to proceed with this definition, we need the following lemma.
\begin{lem}
\label{lem:relation_between_pi_j}(1) There exist constants $\epsilon>0$
and $C>0$ such that we have that 
\[
\left\Vert \psi_{j}\left(\kappa_{j}^{-1}\left(x\right)\right)\cdot\hat{\pi}_{k}^{\left(1\right)}\left(x\right)-\psi_{k}\left(\kappa_{k}^{-1}\left(x\right)\right)\cdot\hat{\pi}_{j}^{\left(1\right)}\left(x\right)\right\Vert <C\hbar^{\epsilon}\quad\mbox{for every \ensuremath{1\le j,k\le I_{\hbar}}}.
\]
(2) For any $m>0$, there exists a constant $C_{m}>0$ such that 
\[
\left\Vert \hat{\pi}_{k}^{\left(1\right)}\left(x\right)\circ\hat{\pi}_{j}^{\left(1\right)}\left(y\right)^{*}\right\Vert \leq C_{m}\left\langle \frac{\left|x-y\right|}{\sqrt{\hbar}}\right\rangle ^{-m}\quad\mbox{and}\quad\left\Vert \hat{\pi}_{k}^{\left(1\right)}\left(x\right)^{*}\circ\hat{\pi}_{j}^{\left(1\right)}\left(y\right)\right\Vert \leq C_{m}\left\langle \frac{\left|x-y\right|}{\sqrt{\hbar}}\right\rangle ^{-m}\quad\mbox{for \ensuremath{x,y}\ensuremath{\in}M}.
\]
Consequently we have (for possibly different constant $C_{m}$) that
\[
\left\Vert \hat{\pi}_{k}^{\left(2\right)}\left(x\right)\circ\hat{\pi}_{j}^{\left(2\right)}\left(y\right)^{*}\right\Vert \leq C_{m}\left\langle \frac{\left|x-y\right|}{\sqrt{\hbar}}\right\rangle ^{-m}\quad\mbox{and}\quad\left\Vert \hat{\pi}_{k}^{\left(2\right)}\left(x\right)^{*}\circ\hat{\pi}_{j}^{\left(2\right)}\left(y\right)\right\Vert \leq C_{m}\left\langle \frac{\left|x-y\right|}{\sqrt{\hbar}}\right\rangle ^{-m}\quad\mbox{for \ensuremath{x,y}\ensuremath{\in}M}.
\]

\end{lem}
The claim of the lemma above may look rather obvious. But note that
the relation between the operators $\hat{\pi}_{j}^{\left(1\right)}\left(x\right)$
(or its adjoint) for different indices $j$ involves the coordinate
change transformation $\kappa_{j,i}$ , which is close to the identity
but non-linear. So we have to go through an argument similar to that
in Section \ref{sec:Nonlinear-prequantum-maps}. We give a proof of
Lemma \ref{lem:relation_between_pi_j} at the end. 

Note that the operator $\hat{\pi}^{\left(2\right)}\left(x\right)$
is not rank one in general. We are going to construct a rank one projection
operator $\hat{\pi}(x)$ from $\hat{\pi}^{\left(2\right)}\left(x\right)$
as its spectral projector. First we check 
\begin{equation}
\mathrm{Tr}\left(\hat{\pi}^{\left(2\right)}\left(x\right)\right)\underset{(\ref{eq:trace_pi_1})}{=}\sum_{j}\psi_{j}\left(\kappa_{j}^{-1}\left(x\right)\right)+O\left(\hbar^{\theta}\right)=1+O\left(\hbar^{\theta}\right)\label{eq:trace_pi_2}
\end{equation}
This is because the term $\psi_{j}\left(\kappa_{j}^{-1}\left(x\right)\right)\cdot\hat{\pi}_{l}^{\left(1\right)}\left(x\right)$
is not zero only if $j$ and $l$ are indices for intersecting local
charts. Hence, using the first claim of Lemma \ref{lem:relation_between_pi_j},
we see
\begin{eqnarray}
\psi_{j}\left(\kappa_{j}^{-1}\left(x\right)\right)\cdot\hat{\pi}^{\left(2\right)}\left(x\right) & = & \sum_{l}\psi_{j}\left(\kappa_{j}^{-1}\left(x\right)\right)\hat{\pi}_{l}^{\left(1\right)}\left(x\right)\underset{}{=}\sum_{l}\psi_{l}\left(\kappa_{l}^{-1}\left(x\right)\right)\cdot\hat{\pi}_{j}^{\left(1\right)}\left(x\right)+O\left(\hbar^{\epsilon}\right)\nonumber \\
 & = & \hat{\pi}_{j}^{\left(1\right)}\left(x\right)\sum_{l}\psi_{l}\left(\kappa_{l}^{-1}\left(x\right)\right)+O\left(\hbar^{\theta}\right)=\hat{\pi}_{j}^{\left(1\right)}\left(x\right)+O\left(\hbar^{\epsilon}\right)\label{eq:pi_1_pi_2}
\end{eqnarray}
and
\begin{eqnarray*}
\left(\psi_{j}\left(\kappa_{j}^{-1}\left(x\right)\right)\cdot\hat{\pi}^{\left(2\right)}\left(x\right)\right)^{2} & \underset{(\ref{eq:pi_1_pi_2})}{=} & \left(\hat{\pi}_{j}^{\left(1\right)}\left(x\right)\right)^{2}+O\left(\hbar^{\theta}\right)\underset{(\ref{eq:pi_1_square})}{=}\psi_{j}\left(\kappa_{j}^{-1}\left(x\right)\right)\cdot\hat{\pi}_{j}^{\left(1\right)}\left(x\right)+O\left(\hbar^{\epsilon}\right)\\
 & = & \left(\psi_{j}\left(\kappa_{j}^{-1}\left(x\right)\right)\right)^{2}\cdot\hat{\pi}^{\left(2\right)}\left(x\right)+O\left(\hbar^{\epsilon}\right)
\end{eqnarray*}
So choosing $j$ such that $\psi_{j}\left(\kappa_{j}^{-1}\left(x\right)\right)>1/C$
we get 
\begin{equation}
\left(\hat{\pi}^{\left(2\right)}\left(x\right)\right)^{2}=\hat{\pi}^{\left(2\right)}\left(x\right)+O\left(\hbar^{\epsilon}\right)\label{eq:pi_2_square}
\end{equation}
Hence from (\ref{eq:trace_pi_2}) and (\ref{eq:pi_2_square}), the
operator $\hat{\pi}^{\left(2\right)}\left(x\right)$ has an isolated
simple eigenvalue $\lambda_{1}$ at a distance $O\left(\hbar^{\epsilon}\right)$
to $1$ and the rest of its spectrum is at distance $O\left(\hbar^{\epsilon}\right)$
of the origin. We consider a fixed path $\gamma$ in $\mathbb{C}$
with center $1$ of radius $C\hbar^{\epsilon}$ (with large $C$)
so that by a Dunford integral, we get the spectral projector of the
operator $\hat{\pi}^{\left(2\right)}\left(x\right)$ for its simple
eigenvalue $\lambda_{1}$ and we denote it: 
\[
\hat{\pi}(x):=\frac{1}{2\pi i}\oint_{\gamma}\left(z-\hat{\pi}^{\left(2\right)}\left(x\right)\right)^{-1}dz.
\]
Note that we may rewrite this definition as 
\begin{equation}
\hat{\pi}(x)=\left(\frac{1}{2\pi i}\oint_{\gamma}\frac{1}{z}\left(z-\hat{\pi}^{\left(2\right)}\left(x\right)\right)^{-1}dz\right)\circ\hat{\pi}^{\left(2\right)}\left(x\right)=\hat{\pi}^{\left(2\right)}\left(x\right)\circ\left(\frac{1}{2\pi i}\oint_{\gamma}\frac{1}{z}\left(z-\hat{\pi}^{\left(2\right)}\left(x\right)\right)^{-1}dz\right)\label{eq:expression_hat_pi}
\end{equation}
The operators $\hat{\pi}(x)$ are rank one projection operators depending
smoothly on $x\in M$. There exists a constant $C>0$ independent
of $\hbar$ such that $\left\Vert \hat{\pi}(x)\right\Vert \leq C$
and that $\left\Vert \hat{\pi}^{\left(2\right)}\left(x\right)-\hat{\pi}(x)\right\Vert <C\hbar^{\epsilon}$.
(For these estimates, we refer the proof of Lemma \ref{lm:pi_P} in
Appendix \ref{sub:appendix2} where the argument is in parallel.)
Hence 
\begin{equation}
\psi_{j}\left(\kappa_{j}^{-1}\left(x\right)\right)\cdot\hat{\pi}(x)\underset{(\ref{eq:pi_1_pi_2})}{=}\hat{\pi}_{j}^{\left(1\right)}\left(x\right)+O\left(\hbar^{\epsilon}\right).\label{eq:pi_pi_1}
\end{equation}
Continuing (\ref{eq:Mphi_tau_check_suite}), we deduce that
\begin{eqnarray}
\mathcal{M}\left(\varphi\right)\checktau_{\hbar}^{(0)} & = & \frac{1}{\left(2\pi\hbar\right)^{d}}\int_{M}\sum_{j=1}^{I_{\hbar}}\varphi\left(x\right)\hat{\pi}_{j}^{\left(1\right)}\left(x\right)dx+O\left(\hbar^{\theta}\right)\nonumber \\
 & \underset{(\ref{eq:pi_pi_1})}{=} & \frac{1}{\left(2\pi\hbar\right)^{d}}\int_{M}\sum_{j=1}^{I_{\hbar}}\varphi\left(x\right)\psi_{j}\left(\kappa_{j}^{-1}\left(x\right)\right)\cdot\hat{\pi}(x)\, dx+O\left(\hbar^{\epsilon}\right)\label{eq:estimate_4}\\
 & \underset{}{=} & \frac{1}{\left(2\pi\hbar\right)^{d}}\int_{M}\varphi\left(x\right)\cdot\hat{\pi}(x)\, dx+O\left(\hbar^{\epsilon}\right)\nonumber 
\end{eqnarray}
In the second line above, we have to use again the integral version
of the Cotlar-Stein Lemma \ref{lem:Integral-version-of-Cotlar-Stein}.
We do this in two steps, let us consider the family of operators 
\[
T\left(x\right):=T_{1}\left(x\right)+T_{2}\left(x\right)=\sum_{j=1}^{I_{\hbar}}\varphi\left(x\right)\left(\psi_{j}\left(\kappa_{j}^{-1}\left(x\right)\right)\cdot\hat{\pi}(x)-\hat{\pi}_{j}^{\left(1\right)}\left(x\right)\right)
\]
with
\[
T_{1}\left(x\right):=\sum_{j=1}^{I_{\hbar}}\varphi\left(x\right)\left(\psi_{j}\left(\kappa_{j}^{-1}\left(x\right)\right)\cdot\hat{\pi}^{\left(2\right)}\left(x\right)-\hat{\pi}_{j}^{\left(1\right)}\left(x\right)\right)
\]
\[
T_{2}\left(x\right):=\sum_{j=1}^{I_{\hbar}}\varphi\left(x\right)\cdot\psi_{j}\left(\kappa_{j}^{-1}\left(x\right)\right)\cdot\left(\hat{\pi}(x)-\hat{\pi}^{\left(2\right)}\left(x\right)\right)=\varphi\left(x\right)\left(\hat{\pi}(x)-\hat{\pi}^{\left(2\right)}\left(x\right)\right).
\]
From Lemma \ref{lem:relation_between_pi_j} and the expression (\ref{eq:expression_hat_pi})
of $\hat{\pi}\left(x\right)$, we see that, for $k=1,2$,
\[
\left\Vert T_{k}\left(x\right)T_{k}\left(y\right)^{*}\right\Vert _{\mathcal{H}_{\hbar}^{r}(\real^{2d})}\le C\hbar^{2\epsilon}\quad\mbox{and}\quad\left\Vert T_{k}\left(x\right)^{*}T_{k}\left(y\right)\right\Vert _{\mathcal{H}_{\hbar}^{r}(\real^{2d})}\le C\hbar^{2\epsilon}
\]
and also
\[
\left\Vert T_{k}\left(x\right)T_{k}\left(y\right)^{*}\right\Vert _{\mathcal{H}_{\hbar}^{r}(\real^{2d})}\le C_{m}\left\langle \frac{\left|x-y\right|}{\sqrt{\hbar}}\right\rangle ^{-m}\quad\mbox{and}\quad\left\Vert T_{k}\left(x\right)^{*}T_{k}\left(y\right)\right\Vert _{\mathcal{H}_{\hbar}^{r}(\real^{2d})}\le C_{m}\left\langle \frac{\left|x-y\right|}{\sqrt{\hbar}}\right\rangle ^{-m}
\]
for any $m>0$ with a constant $C_{m}$ independent of $\hbar$. It
follows
\[
\left\Vert T_{k}\left(x\right)T_{k}\left(y\right)^{*}\right\Vert _{\mathcal{H}_{\hbar}^{r}(\real^{2d})}\le C_{m}\hbar^{\epsilon}\cdot\left\langle \frac{\left|x-y\right|}{\sqrt{\hbar}}\right\rangle ^{-m}\quad\mbox{and}\quad\left\Vert T_{k}\left(x\right)^{*}T_{k}\left(y\right)\right\Vert _{\mathcal{H}_{\hbar}^{r}(\real^{2d})}\le C_{m}\hbar^{\epsilon}\cdot\left\langle \frac{\left|x-y\right|}{\sqrt{\hbar}}\right\rangle ^{-m}.
\]
 Then Cotlar-Stein Lemma \ref{lem:Integral-version-of-Cotlar-Stein}
implies that $\frac{1}{\left(2\pi\hbar\right)^{d}}\left\Vert \int_{M}T_{k}\left(x\right)dx\right\Vert _{\mathcal{H}_{\hbar}^{r}(\real^{2d})}=O\left(\hbar^{\epsilon}\right)$
and therefore $\frac{1}{\left(2\pi\hbar\right)^{d}}\left\Vert \int_{M}T\left(x\right)dx\right\Vert _{\mathcal{H}_{\hbar}^{r}(\real^{2d})}=O\left(\hbar^{\epsilon}\right)$,
giving (\ref{eq:estimate_4}) above. We have finally obtained that
in norm operator:

\[
\left\Vert \mathcal{M}\left(\varphi\right)\checktau_{\hbar}^{(0)}-\frac{1}{\left(2\pi\hbar\right)^{d}}\int_{M}\varphi\left(x\right)\cdot\hat{\pi}(x)\, dx\right\Vert _{\mathcal{H}_{\hbar}^{r}(\real^{2d})}\leq C\hbar^{\epsilon}
\]
With Lemma \ref{lm:pi_P} and (\ref{eq:pi_cT1}), we deduce that
\[
\left\Vert \mathcal{M}\left(\varphi\right)\tau_{\hbar}^{(0)}-\frac{1}{\left(2\pi\hbar\right)^{d}}\int_{M}\varphi\left(x\right)\cdot\hat{\pi}(x)\, dx\right\Vert _{\mathcal{H}_{\hbar}^{r}(\real^{2d})}\leq C\hbar^{\varepsilon}
\]
We can argue just in parallel manner to give
\begin{equation}
\left\Vert \tau_{\hbar}^{(0)}\mathcal{M}\left(\varphi\right)-\frac{1}{\left(2\pi\hbar\right)^{d}}\int_{M}\varphi\left(x\right)\cdot\hat{\pi}(x)\, dx\right\Vert _{\mathcal{H}_{\hbar}^{r}(\real^{2d})}\leq C\hbar^{\varepsilon}\label{eq:trace_estimate_on_tauMphi}
\end{equation}
Therefore 
\[
\left\Vert \left[\tau_{\hbar}^{(0)},\mathcal{M}\left(\varphi\right)\right]\right\Vert _{\mathcal{H}_{\hbar}^{r}(\real^{2d})}\leq C\hbar^{\varepsilon}.
\]

Suppose that $r_{1}^{+}<r_{0}^{-}$. We have defined in (\ref{eq:def_Pi_hbar})
by $\Pi_{\hbar}$ the finite rank spectral projector on the external
band of $\hat{F}_{N}$. We define 
\[
\pi_{x}:=\Pi_{\hbar}\circ\hat{\pi}\left(x\right)\circ\Pi_{\hbar}.
\]
 We have shown in (\ref{eq:diff_projectors}) that 
\begin{equation}
\left\Vert \tau_{\hbar}^{(0)}-\Pi_{\hbar}\right\Vert <C\hbar^{\epsilon}\label{eq:estimate_Pi_0}
\end{equation}
Therefore we deduce from above that
\[
\left\Vert \Pi_{\hbar}\mathcal{M}\left(\varphi\right)\Pi_{\hbar}-\frac{1}{\left(2\pi\hbar\right)^{d}}\int_{M}\varphi\left(x\right)\cdot\pi_{x}dx\right\Vert \leq C\hbar^{\varepsilon}
\]
and
\[
\left\Vert \left[\Pi_{\hbar},\mathcal{M}\left(\varphi\right)\right]\right\Vert \leq C\hbar^{\varepsilon}.
\]
This finishes the proof of Theorem \ref{prop:express_Op_Psi}. Finally
we prove Lemma \ref{lem:relation_between_pi_j}. 
\begin{proof}[Proof of Lemma \ref{lem:relation_between_pi_j} ]
Recall the operator $\mathcal{Y}_{\hbar}$ which truncate the functions
in the phase space. We decompose the operators $\hat{\pi}_{j}^{(1)}(\nu)$
into 
\[
\hat{\pi}_{j,1}^{\left(1\right)}\left(x\right):=\mathbf{I}_{i,\hbar}^{*}\circ\mathcal{\mathcal{Y}_{\hbar}}\circ\hat{\pi}_{0}\left(\kappa_{j}^{-1}\left(x\right)\right)\circ\mathcal{M}(\psi_{j})\circ\mathbf{I}_{i,\hbar}
\]
and 
\[
\hat{\pi}_{j,2}^{\left(1\right)}\left(x\right):=\mathbf{I}_{i,\hbar}^{*}\circ(1-\mathcal{\mathcal{Y}_{\hbar}})\circ\hat{\pi}_{0}\left(\kappa_{j}^{-1}\left(x\right)\right)\circ\mathcal{M}(\psi_{j})\circ\mathbf{I}_{i,\hbar}.
\]
For the former part $\hat{\pi}_{j,1}^{\left(1\right)}\left(x\right)$,
we apply Proposition \ref{lm:L_g_Y_almost_identity} and the estimate
(\ref{eq:kernel_estimate_nonlinear}) on the kernel in the proof to
see that the non-linear coordinate change transformation hardly affect
this part. Thus we have both of the claims when we replace the operator
$\hat{\pi}_{j}^{(1)}(\nu)$ by $\hat{\pi}_{j,1}^{(1)}(\nu)$. For
the latter part, we have
\[
\|\hat{\pi}_{j,2}^{\left(1\right)}\left(x\right)\|_{\mathcal{H}_{\hbar}^{r}(\real^{2d})}<C\hbar^{\theta}
\]
from Lemma \ref{lem:kernel_of_pi_nu}. This completes the proof of
the first claim (1). By inspecting the kernels of the lifted operators,
we also see
\[
\|\left(\hat{\pi}_{k}^{(1)}(x)\right)^{*}\circ\hat{\pi}_{j,2}^{\left(1\right)}\left(x\right)\|_{\mathcal{H}_{\hbar}^{r}(\real^{2d})}=\mathcal{O}(\hbar^{\infty}),\quad\|\hat{\pi}_{k}^{(1)}(x)\circ\left(\hat{\pi}_{j,2}^{\left(1\right)}\left(x\right)\right)^{*}\|_{\mathcal{H}_{\hbar}^{r}(\real^{2d})}=\mathcal{O}(\hbar^{\infty}).
\]
Since we have only to consider points $x,y\in M$ with $d(x,y)\le C\hbar^{1/2-\theta}$,
this gives the second claim (2). \end{proof}

\section{\label{sec:Proofs-for-Laplacian}Proof of Theorem \ref{thm:band_structure-of_Laplacian}
for the spectrum of the rough Laplacian.}

\subsection{\label{sub:The-harmonic-oscillator3.5}The harmonic oscillator}

\label{ss:Qhk} In this subsection we present the harmonic oscillator
in the setting of Bargmann transform. (We refer \cite{folland-88},\cite{zworski-03}
for a more detailed treatment.) We will need it in dealing with the
(Euclidean) rough Laplacian $\Delta_{\hbar}$ in Section \ref{sub:The-rough-Laplacian_4.6}.
Associated to the standard coordinates
\[
(x,\xi)=(x_{1},x_{2},\cdots,x_{D},\xi_{1},\xi_{2},\cdots,\xi_{D})
\]
 on $T^{*}\real^{D}=\real^{2D}$, we consider the operators 
\[
\hat{x}_{i}:\begin{cases}
\mathcal{S}(\real^{2D}) & \to\mathcal{S}(\real^{2D})\\
u & \rightarrow\left(\BargmannP_{\hbar}\circ\multiplication(x_{i})\circ\BargmannP_{\hbar}\right)u
\end{cases}\quad\mbox{{and} \quad}\hat{\xi}_{i}:\begin{cases}
\mathcal{S}(\real^{2D}) & \to\mathcal{S}(\real^{2D})\\
u & \rightarrow\left(\BargmannP_{\hbar}\circ\multiplication(\xi_{i})\circ\BargmannP_{\hbar}\right)u
\end{cases}
\]
where $\multiplication(x_{i})$ and $\multiplication(\xi_{i})$ on
the right hand sides denote the multiplication by the corresponding
functions and $\mathcal{P}_{\hbar}$ is the Bargmann projector (\ref{eq:def_Bargmann_projector}).
These operators are usually called Toeplitz quantization of the functions
$x_{i}$ and $\xi_{i}$. Then we set

\begin{equation}
\hat{P}:=\frac{1}{2\hbar}\left(\hat{x}^{2}+\hat{\mathbf{\xi}}^{2}\right):=\frac{1}{2\hbar}\sum_{i=1}^{D}\left(\hat{x}_{i}\circ\hat{x}_{i}+\hat{\xi}_{i}\circ\hat{\xi}_{i}\right),\label{eq:Zhbar}
\end{equation}
\[
\mathscr{H}:=\Bargmann_{\hbar}^{*}\hat{P}\Bargmann_{\hbar},\quad L^{2}(\real^{D})\to L^{2}(\real^{D})
\]

which is usually called the \emph{harmonic oscillator} operator. (the
operators $\Bargmann_{\hbar}^{*},\Bargmann_{\hbar}$ are defined in
Section \ref{sub:The-Bargmann-transform}).
\begin{lem}
\label{lm:spec_Z}\textbf{''Spectrum of the harmonic oscillator''.}
The operator $\mathscr{H}$ in (\ref{eq:Zhbar}) is a closed self-adjoint
operator on $L^{2}(\real^{D})$ and its spectral set consists of eigenvalues
\[
\frac{D}{2}+k,\quad k\in\mathbb{N}.
\]
For every $k\in\mathbb{N}$, the spectral projector $Q_{\hbar}^{(k)}$
for the eigenvalue $\frac{D}{2}+k$ is an orthogonal projection operator
of rank $\binom{D+k-1}{D-1}$. We have 
\begin{equation}
\bigoplus_{i=0}^{k}\mathrm{Im}Q_{\hbar}^{(i)}=\{\varphi_{0}.p\mid\quad p\mbox{ is a polynomial of degree }\leq k\}\label{eq:ImQ}
\end{equation}
 where 
\[
\varphi_{0}(x)=e^{-|x|^{2}/(2\hbar)}.
\]
 In particular, we have the following orthogonal decomposition of
$L^{2}(\real^{D})$: 
\[
L^{2}(\real^{D})=\overline{\bigoplus_{k=0}^{\infty}\mathrm{{Im}}Q_{\hbar}^{(k)}}.
\]
\end{lem}
\begin{proof}
Since $\hat{{x}}_{i}$ and $\hat{\xi_{i}}$ are the lift of the operators
\[
u\mapsto x_{i}\cdot u\quad\mbox{and}\quad u\mapsto-i\hbar\cdot\partial_{x_{i}}u
\]
 respectively, the operator $\hat{P}=\frac{1}{2\hbar}\left(\hat{x}^{2}+\hat{\mathbf{\xi}}^{2}\right)$
is the lift of 
\begin{equation}
\mathscr{H}:L^{2}(\real^{D})\to L^{2}(\real^{D}),\qquad\left(\mathscr{H}u\right)(x)=\frac{1}{\hbar}\left(-\frac{\hbar^{2}}{2}\left(\sum_{i=1}^{D}\frac{\partial^{2}u}{\partial x_{i}^{2}}\right)(x)+\frac{1}{2}|x|^{2}\cdot u(x)\right)\label{eq:def_Harmonic_Oscillator}
\end{equation}
 Therefore the conclusion follows from the argument on quantization
of the harmonic oscillator $\mathscr{H}$ \cite[p.105]{taylor_tome2}.\end{proof}
\begin{rem}
It is possible to compute directly that the eigenfunctions of $\hat{P}$
using the ``creation operator'' $a_{j}=\hat{x}_{j}+i\hat{\xi}_{j}$
and showing that it is multiplication by $z_{j}=x_{j}+i\xi_{j}$ in
phase space (see \cite{folland-88,hall99}). Then we can identify
$\mathcal{Q}_{\hbar}^{(i)}=\Bargmann_{\hbar}^{*}Q_{\hbar}^{\left(i\right)}\Bargmann_{\hbar}$
as the projection onto homogeneous polynomials of degree $i$ in $\left(z_{j}\right)_{j}$.
\end{rem}
Recall that operators $Q_{\hbar}^{(k)}$ and $T^{\left(k\right)}$
(Section \ref{sub:Spectrum-of-transfer_3.4}) have the same rank $\binom{D+k-1}{D-1}$.
The next lemma gives a more precise relation between them.
\begin{lem}
\label{lm:QTproj} For $0\le k\le n$, the operator $\mathcal{Q}_{\hbar}^{(k)}=\Bargmann_{\hbar}^{*}Q_{\hbar}^{\left(k\right)}\Bargmann_{\hbar}$
extends to a continuous operator 
\[
\mathcal{Q}_{\hbar}^{(k)}:\mathcal{S}'(\real^{2D})\to\mathcal{S}(\real^{2D}).
\]
The restrictions: 
\begin{equation}
\left(\oplus_{i=0}^{k}\mathcal{Q}_{\hbar}^{(i)}\right):\oplus_{i=1}^{k}\mathrm{Im}\calT_{\hbar}^{(i)}\to\oplus_{i=0}^{k}\mathrm{Im}\mathcal{Q}_{\hbar}^{(i)}\label{eq:QT_bijection2}
\end{equation}
 and 
\begin{equation}
\left(\oplus_{i=0}^{k}\calT_{\hbar}^{(i)}\right):\oplus_{i=0}^{k}\mathrm{Im}\mathcal{Q}_{\hbar}^{(i)}\to\oplus_{i=0}^{k}\mathrm{Im}\calT_{\hbar}^{(i)}\label{eq:TQ_bijection2}
\end{equation}
are well-defined and bijective. The operator norms of (\ref{eq:QT_bijection2}),
(\ref{eq:TQ_bijection2}) and their inverses are bounded by a constant
independent of $\hbar$.\end{lem}
\begin{proof}
Recall that the operator $\mathcal{Q}_{\hbar}^{(k)}$ is the orthogonal
projection to its image (\ref{eq:ImQ}), which is finite dimensional
and contained in $\mathcal{S}(\real^{2D})$. Hence we have the first
claim and well definiteness of the operators (\ref{eq:QT_bijection2})
and (\ref{eq:TQ_bijection2}) is an immediate consequence. 

To prove that (\ref{eq:QT_bijection2}) and (\ref{eq:TQ_bijection2})
are bijective, we have only to show that they are injective, because
the subspaces in the source and target have the same finite dimension.
We prove injectivity of (\ref{eq:QT_bijection2}). Let $u\in\left(\oplus_{i=0}^{k}\mbox{Im}\calT_{\hbar}^{(i)}\right)$.
Such $u$ can be expressed as $u=\Bargmann_{\hbar}p$ with $p$ a
polynomial of degree at most $k$ on $\real^{D}$. Suppose that $\left(\oplus_{i=0}^{k}\mathcal{Q}_{\hbar}^{(i)}\right)u=0$.
Since $\left(\oplus_{i=0}^{k}\mathcal{Q}_{\hbar}^{(i)}\right)$ is
an orthogonal projection operator, $u$ is orthogonal to $\oplus_{i=0}^{k}\mbox{Im}\mathcal{Q}_{\hbar}^{(i)}$
and hence we have 
\[
(p,q\cdot\varphi_{0})_{L^{2}(\real^{D})}=(\Bargmann_{\hbar}p,\Bargmann_{\hbar}(q\cdot\varphi_{0}))_{L^{2}(\real^{2D})}=0
\]
for any polynomial $q$ of order at most $k$ on $\real^{D}$. Setting
$q=p$, we see that $(p,p\cdot\varphi_{0})_{L^{2}(\real^{D})}=(|p|^{2},\varphi_{0})_{L^{2}(\real^{D})}=0$,
showing $p=0$ and $u=0$. We have shown that (\ref{eq:QT_bijection2})
is injective.

To prove injectivity of (\ref{eq:TQ_bijection2}), let $u=q\varphi_{0}\in\left(\oplus_{i=0}^{k}\mbox{Im}\mathcal{Q}_{\hbar}^{(i)}\right)$
with $q$ a polynomial of degree at most $k$ on $\mathbb{R}^{D}$.
Let $p=\left(\oplus_{i=0}^{k}\calT_{\hbar}^{(i)}\right)u$ be the
Taylor expansion of $u$ at $0$ up to order $k$. Suppose that $p=0$.
Since $\varphi_{0}(0)\neq0$, $\varphi_{0}$ is invertible as a formal
power series, we deduce that $q=0$. Hence (\ref{eq:TQ_bijection2})
is injective.

The operators (\ref{eq:QT_bijection2}) (resp. (\ref{eq:TQ_bijection2}))
for different $\hbar>0$ are related by the scaling (\ref{eq:wI_hbar})
and hence we get the last claim. 
\end{proof}

\subsection{\label{sub:The-rough-Laplacian_4.6}The rough Laplacian on $\real^{2d}$}

\label{ss:Laplacian_on_R2d}As in Subsection \ref{sub: Euclidean_prequantum_and_Laplacian_operator},
we consider $\mathbb{R}^{2d}$ as a symplectic linear space with $\omega=\sum_{i=1}^{d}dq^{i}\wedge dp^{i}$
and with the additional compatible Euclidean metric
\[
g=\sum_{i=1}^{d}dq^{i}\otimes dq^{i}+dp^{i}\otimes dp^{i}
\]
We have seen in (\ref{eq:Lap_a}) that these data define the Euclidean
rough Laplacian as the operator
\[
\Delta_{\hbar}=D^{*}D:C^{\infty}(\real^{2d})\to C^{\infty}(\real^{2d}).
\]
This operator $\Delta_{\hbar}$ is a closed self-adjoint operator
on $L^{2}(\real^{2d})$ and its domain of definition is 
\[
\domain(\Delta_{\hbar})=\{u\in L^{2}(\real^{2d})\mid\|\Delta_{\hbar}u\|_{L^{2}}<\infty\}.
\]
 Note that $\domain(\Delta_{\hbar})$ becomes a Hilbert space if we
consider the norm 
\begin{equation}
\|u\|_{\Delta_{\hbar}}=((u,u)_{\Delta_{\hbar}})^{1/2}\label{eq:def_norm_Delta}
\end{equation}
 induced by the inner product 
\[
(u,v)_{\Delta_{\hbar}}=(u,v)_{L^{2}}+(\Delta_{\hbar}u,\Delta_{\hbar}v)_{L^{2}}.
\]
Obviously, $\Delta_{\hbar}$ gives a bounded operator 
\[
\Delta_{\hbar}:\left(\domain(\Delta_{\hbar}),\|\cdot\|_{\Delta_{\hbar}}\right)\to L^{2}(\real^{2d}).
\]

An important property of the operator $\Delta_{\hbar}$ that follows
from the definition is that it is invariant with respect to the action
of prequantum transfer operators for symplectic isometries:
\begin{lem}
Suppose that $f:\real^{2d}\to\real^{2d}$ is an isometric affine map
preserving the symplectic form $\omega$, then we have $\Delta_{\hbar}\circ\prequantumL_{f}=\prequantumL_{f}\circ\Delta_{\hbar}$
for the associated prequantum transfer operator $\mathcal{L}_{f}$
given in (\ref{eq:F_f_affine}).
\end{lem}
From the expression of the Euclidean rough Laplacian obtained in (\ref{eq:Lap_a}),
we have
\[
\hbar\Delta=\mathcal{U}\circ\left(\mathrm{Id}\otimes\frac{1}{\hbar}\left(\widehat{\zeta_{q}}^{2}+\widehat{\zeta_{p}}^{2}\right)\right)\circ\mathcal{U}^{-1}
\]
where $\mathcal{U}$ has been defined in (\ref{eq:def_U_operator}).

By considering a commutative diagram corresponding to (\ref{eq:Long_diagram}),
we obtain the following commutative diagram similar to (\ref{cd:lift_of_Lf}):

\[
\begin{CD}L^{2}\left(\real_{x}^{2d}\right)@>\hbar\Delta_{\hbar}>>L^{2}\left(\real_{x}^{2d}\right)\\
@AA\mathcal{U}A@AA\mathcal{U}A\\
L^{2}\left(\mathbb{R}_{\nu_{q}}^{2}\right)\otimes L^{2}\left(\real_{\zeta_{p}}^{d}\right)@>\mathrm{Id}\otimes\mathscr{H}>>L^{2}\left(\mathbb{R}_{\nu_{q}}^{2}\right)\otimes L^{2}\left(\real_{\zeta_{p}}^{d}\right)
\end{CD}
\]
 where $\mathscr{H}$ is the harmonic oscillator operator defined
in (\ref{eq:def_Harmonic_Oscillator}) with setting $D=d$. 

Thus we may invoke the argument in Subsection \ref{ss:Qhk}, especially
Lemma \ref{lm:spec_Z} and \ref{lm:QTproj}, to derive the next proposition
on the spectral structure of the Euclidean rough Laplacian $\Delta_{\hbar}$.
For $k\ge0$, let us consider the spectral projection operator 
\begin{equation}
\romeq_{\hbar}^{(k)}:=\mathcal{U}\circ\left(\mathrm{Id}\otimes Q_{\hbar}^{\left(k\right)}\right)\circ\mathcal{U}^{-1}\quad:L^{2}\left(\real_{x}^{2d}\right)\rightarrow L^{2}\left(\real_{x}^{2d}\right)\label{eq:def_q}
\end{equation}
 where $Q_{\hbar}^{\left(k\right)}$ is the projection operator on
level $k$ of the harmonic oscillator $\mathscr{H}$. Note that it
restricts to a bounded operator 
\begin{equation}
\romeq_{\hbar}^{(k)}:L^{2}(\real^{2d})\to(\domain(\Delta_{\hbar}),\|\cdot\|_{\Delta_{\hbar}})\subset L^{2}(\real^{2d})\label{eq:q_bdd}
\end{equation}
whose operator norm is bounded by a constant independent of $\hbar$. 
\begin{prop}
\label{pp1-2} The rough Laplacian $\Delta_{\hbar}=D^{*}D$ on the
Euclidean space $\mathbb{R}^{2d}$ is a closed self-adjoint operator
on $L^{2}(\real^{2d})$ and its spectrum consists of integer eigenvalues
$d+2k$ with $k\in\mathbb{N}$. The spectral projector corresponding
to the eigenvalue $d+2k$ is the operator $\romeq_{\hbar}^{(k)}$
given in (\ref{eq:def_q}) and together they form a complete set of
mutually commuting orthogonal projections in $L^{2}(\real^{2d})$.
Consequently we have 
\[
L^{2}(\real^{2d})=\overline{\bigoplus_{k=0}^{\infty}H''_{k}},\qquad\mbox{with \ensuremath{H''_{k}:=\Im\romeq_{\hbar}^{(k)}}.}
\]

\end{prop}
The next proposition is an immediate consequence of Lemma \ref{lm:QTproj}. 
\begin{prop}
\label{pp1-3} The operator $\oplus_{i=0}^{k}\romeq_{\hbar}^{(i)}$
and $\oplus_{i=0}^{k}\romet_{\hbar}^{(i)}$ restricts to the bijections
\[
\oplus_{i=0}^{k}\romeq_{\hbar}^{(i)}:\oplus_{i=0}^{k}\Im\romet_{\hbar}^{(i)}\to\oplus_{i=0}^{k}\Im\romeq_{\hbar}^{(i)}
\]
 and 
\[
\oplus_{i=0}^{k}\romet_{\hbar}^{(i)}:\oplus_{i=0}^{k}\Im\romeq_{\hbar}^{(i)}\to\oplus_{i=0}^{k}\Im\romet_{\hbar}^{(i)}
\]
 respectively. The operator norms of these operators and their inverses
are bounded by some constant independent of $\hbar$.
\end{prop}

\subsection{\label{sub:The-multiplication-operators_Laplacian}The multiplication
operators and the rough Laplacian on $\real^{2d}$}

\label{ss:multiplication_Laplacian} Recall the space of functions
$\scrX_{\hbar}$ defined in Setting I, page \ref{Setting-I}.
\begin{lem}
\label{lm:multiplication_Laplacian} For any $\psi\in\scrX_{\hbar}$,
we have 
\[
\left\Vert \left[\multiplication(\psi),\Delta_{\hbar}\right]\right\Vert _{(\domain(\Delta_{\hbar}),\|\cdot\|_{\Delta_{\hbar}})\to L^{2}(\real^{2d})}\le C\hbar^{\theta}
\]
 and 
\[
\left\Vert \left[\multiplication(\psi),\romeq_{\hbar}^{(k)}\right]\right\Vert _{L^{2}(\real^{2d})}\le C\hbar^{\theta}
\]
 where $C$ is a constant independent of $\psi\in\scrX_{\hbar}$ and
$\hbar$. \end{lem}
\begin{proof}
The first claim can be checked easily from the expression of $\Delta_{\hbar}$
given in Proposition \ref{prop:expression_of_D*_Delta}. For the second
claim, we can just follow the argument in the proof of Lemma \ref{cor:XT_exchange},
replacing $\romet_{\hbar}^{(k)}$ by $\romeq_{\hbar}^{(k)}$. (The
proof is simpler actually.)
\end{proof}

\subsection{\label{sub:Proof-of-Laplacian}Proof of Theorem \ref{thm:band_structure-of_Laplacian} }

In this subsection, we give a proof of Theorem \ref{thm:band_structure-of_Laplacian}
on the rough Laplacian. In former part of the proof, we consider a
rough Laplacian $\tilde{\Delta}_{\hbar}$ constructed from local data
instead of the geometric rough Laplacian $\Delta_{\hbar}$, and prove
the claims of Theorem \ref{thm:band_structure-of_Laplacian} for $\tilde{\Delta}_{\hbar}$.
In the latter part, we show that we can deform the rough Laplacian
$\tilde{\Delta}_{\hbar}$ continuously to $\Delta_{\hbar}$ keeping
the ``band structure'' of the eigenvalues. This will imply that
the cardinality of eigenvalues in the first (or lowest) band coincides
for $\tilde{\Delta}_{\hbar}$ and $\Delta_{\hbar}$. We note at this
moment that, for the argument on rough Laplacian below, we do not
need Condition (2) in Proposition \ref{prop:Local-charts} (i.e. orthogonality
of stable and unstable subspaces) in the choice of the coordinate
charts $\{\kappa_{i}\}_{i=1}^{I_{\hbar}}$ in Proposition \ref{prop:Local-charts},
that is, our argument below holds true for any choice of coordinate
charts $\{\kappa_{i}\}_{i=1}^{I_{\hbar}}$ satisfying the conditions
other than that condition. Also since our proof about the Laplacian
is independent on the dynamics of $f$, in our choices the value of
$0<\beta<1$ can be taken close to $1$.

We introduce a rough Laplacian $\tilde{\Delta}_{\hbar}$ acting on
the space $C_{N}^{\infty}(P)$ of equivariant functions. We start
from the operators on local data. Let 
\[
\boldDelta_{\hbar}:\bigoplus_{i=1}^{I_{\hbar}}\domain(\Delta_{\hbar})\to\bigoplus_{i=1}^{I_{\hbar}}L^{2}(\real^{2d}),\quad\boldDelta((u_{i})_{i=1}^{I_{\hbar}})=(\Delta_{\hbar}u_{i})_{i=1}^{I_{\hbar}}
\]
where $\Delta_{\hbar}$ denotes the Euclidean rough Laplacian on $\real^{2d}$
defined in Subsection \ref{sub:The-rough-Laplacian_4.6}. The next
proposition is an immediate consequence of Lemma \ref{lm:multiplication_Laplacian}.
(So we omit the proof.)
\begin{prop}
\label{pp3} 

There exist constants $C>0$ and $\epsilon>0$, independent of $\hbar$,
such that 
\begin{equation}
\left\Vert \left[{\boldDelta}_{\hbar},(\boldI_{\hbar}\circ\boldI_{\hbar}^{*})\right]\right\Vert _{\bigoplus_{i=1}^{I_{\hbar}}(\domain(\Delta_{\hbar}),\|\cdot\|_{\Delta_{\hbar}})\to\bigoplus_{i=1}^{I_{\hbar}}L^{2}(\real^{2d})}\le C\hbar^{\varepsilon}.\label{eqn:Delta_commutes_with_II}
\end{equation}

\end{prop}
We define a rough Laplacian $\tilde{\Delta}_{\hbar}$ acting on $C_{N}^{\infty}(P)$
by 
\begin{equation}
\tilde{\Delta}_{\hbar}:=\mathbf{I}_{\hbar}^{*}\circ\boldDelta_{\hbar}\circ\mathbf{I}_{\hbar}\qquad:C_{N}^{\infty}\left(P\right)\rightarrow C_{N}^{\infty}\left(P\right).\label{eq:relation_Delta}
\end{equation}

\begin{rem}
Notice that this rough Laplacian operator $\tilde{\Delta}_{\hbar}$
is defined by gluing Euclidean rough Laplacian on local charts and
does \emph{not} coincide with the geometric rough Laplacian $\Delta_{\hbar}=D^{*}D$
with respect to a global metric on $M$ defined in Subsection \ref{sub:The-rough-Laplacian}. 
\end{rem}
The operator $\tilde{\Delta}_{\hbar}$ defined above is a closed densely
defined operator on $C_{N}^{\infty}\left(P\right)$. Its domain of
definition is by definition 
\[
\domain(\tilde{\Delta}_{\hbar})=\{u\in L_{N}^{2}\left(P\right)\,\mid\,\|\tilde{\Delta}_{\hbar}u\|_{L^{2}}<\infty\},
\]
which becomes a Hilbert space if we equip it with the inner product
\[
(u,v)_{\tilde{\Delta}_{\hbar}}=(u,v)_{L^{2}}+(\tilde{\Delta}_{\hbar}u,\tilde{\Delta}_{\hbar}v)_{L^{2}}.
\]
We will write $\|\cdot\|_{\tilde{\Delta}_{\hbar}}$ for the corresponding
norm. It is easy to see the following Lemma. (So we omit the proof.) 
\begin{lem}
The norm $\|\cdot\|_{\tilde{\Delta}_{\hbar}}$ above is equivalent
to the norm 
\[
\|u\|'_{\tilde{\Delta}_{\hbar}}=\left(\sum_{i=1}^{I_{\hbar}}\|u_{i}\|_{\Delta_{\hbar}}^{2}\right)^{1/2}
\]
defined in terms of local data, where $\|\cdot\|_{\Delta_{\hbar}}$
on the right hand side denotes the norm defined in (\ref{eq:def_norm_Delta}).
Consequently we have 
\[
\domain(\tilde{\Delta}_{\hbar})=\left\{ u\in L_{N}^{2}(P)\,\left|\,\sum_{i=1}^{I_{\hbar}}\|u_{i}\|_{\Delta_{\hbar}}^{2}<\infty\right.\right\} .
\]

\end{lem}
Below we are going to construct the spectral projectors for $\tilde{\Delta}_{\hbar}$,
corresponding to the ``bands of eigenvalues''. Again we start from
local data: We consider the operators 
\[
{\boldq}_{\hbar}^{(k)}:\bigoplus_{i=1}^{I_{\hbar}}L^{2}(\real^{2d})\to\bigoplus_{i=1}^{I_{\hbar}}(\domain(\Delta_{\hbar}),\|\cdot\|_{\Delta_{\hbar}})\subset\bigoplus_{i=1}^{I_{\hbar}}L^{2}(\real^{2d}),\qquad{\boldq}_{\hbar}^{(k)}((u_{i})_{i=1}^{I_{\hbar}})=(\romeq_{\hbar}^{(k)}(u_{i}))_{i=1}^{I_{\hbar}},
\]
for $0\le k\le n$, where $q_{\hbar}^{(k)}$ is the operator defined
in (\ref{eq:def_q}). Recall (\ref{eq:q_bdd}) for boundedness of
these projection operators. The remainder is denoted as 
\[
\tilde{\boldq}_{\hbar}={\boldq}_{\hbar}^{(n+1)}:\bigoplus_{i=1}^{I_{\hbar}}L^{2}(\real^{2d})\to\bigoplus_{i=1}^{I_{\hbar}}L^{2}(\real^{2d}),\qquad\tilde{\boldq}_{\hbar}={\boldq}_{\hbar}^{(n+1)}:=\mathrm{Id}-({\boldq}_{\hbar}^{(1)}+{\boldq}_{\hbar}^{(2)}+\cdots+{\boldq}_{\hbar}^{(n)}).
\]
The last operator restricts to a bounded operator 
\[
\tilde{\boldq}_{\hbar}={\boldq}_{\hbar}^{(n+1)}:\bigoplus_{i=1}^{I_{\hbar}}(\domain(\Delta_{\hbar}),\|\cdot\|_{\Delta_{\hbar}})\to\bigoplus_{i=1}^{I_{\hbar}}(\domain(\Delta_{\hbar}),\|\cdot\|_{\Delta_{\hbar}}).
\]

We next introduce the operators 
\[
\checklambda_{\hbar}^{(k)}:=\boldI_{\hbar}^{*}\circ\boldq_{\hbar}^{(k)}\circ\boldI_{\hbar}:L_{N}^{2}(P)\to(\domain(\tilde{\Delta}_{\hbar}),\|\cdot\|_{\tilde{\Delta}_{\hbar}})
\]
for $0\le k\le n$. These are bounded operators and the operator norms
are bounded by a constant independent of $\hbar$. For $k=n+1$, we
set 
\begin{equation}
\checklambda_{\hbar}^{(n+1)}:=\boldI_{\hbar}^{*}\circ\boldq_{\hbar}^{(n+1)}\circ\boldI_{\hbar}=\mathrm{Id}-(\checklambda_{\hbar}^{(0)}+\checklambda_{\hbar}^{(1)}+\cdots+\checklambda_{\hbar}^{(n)}).\label{eq:cLn1}
\end{equation}
  Further we can prove the estimates
\begin{equation}
\left\Vert \checklambda_{\hbar}^{(k)}\circ\checklambda_{\hbar}^{(k)}-\checklambda_{\hbar}^{(k)}\right\Vert _{L_{N}^{2}(P)\to(\domain(\tilde{\Delta}_{\hbar}),\|\cdot\|_{\tilde{\Delta}_{\hbar}})}\le C\hbar^{\epsilon}\quad\mbox{for \ensuremath{0\le k\le n+1}}\label{eq:cLcT2-1}
\end{equation}
and 
\begin{align}
 &  & \left\Vert \checklambda_{\hbar}^{(k)}\circ\checklambda_{\hbar}^{(k')}\right\Vert _{L_{N}^{2}(P)\to(\domain(\tilde{\Delta}_{\hbar}),\|\cdot\|_{\tilde{\Delta}_{\hbar}})}\le C\hbar^{\epsilon}\quad\mbox{for \ensuremath{0\le k,k'\le n+1}with \ensuremath{k\neq k'}}\label{eq:cLcT2}
\end{align}
 for some constants $\epsilon>0$ and $C>0$. (For the case where
either of $k$ or $k'$ equals $n+1$, use the definition (\ref{eq:cLn1})
to check (\ref{eq:cLcT2-1}) and (\ref{eq:cLcT2}).)

Now we proceed in parallel to the argument in Subsection \ref{sub:Proofs-of-Theorem_more_details}
and obtain the following lemma.
\begin{lem}
\label{lm:lambda_Delta} There exist a decomposition of the Hilbert
space $(\domain(\tilde{\Delta}_{\hbar}),\|\cdot\|_{\tilde{\Delta}_{\hbar}})$
\[
\domain(\tilde{\Delta}_{\hbar})=\mathfrak{H}_{0}\oplus\mathfrak{H}_{1}\oplus\mathfrak{H}_{2}\oplus\cdots\oplus\mathfrak{H}_{n}\oplus\widetilde{\mathfrak{H}}
\]
and that of $L_{N}^{2}(P)$ 
\[
L_{N}^{2}(P)=\mathfrak{H}_{0}\oplus\mathfrak{H}_{1}\oplus\mathfrak{H}_{2}\oplus\cdots\oplus\mathfrak{H}_{n}\oplus\overline{\mathfrak{H}}
\]
where $\overline{\mathfrak{H}}$ is the closure of $\widetilde{\mathfrak{H}}$
in $L_{N}^{2}(P)$. The subspaces $\mathfrak{H}_{k}$ for $0\le k\le n$
are of finite dimension. If we write $\lambda_{\hbar}^{(k)}$ for
projection operators to $\mathfrak{H}_{k}$ (resp. to $\overline{\mathfrak{H}}$
in the case k=n+1) along other subspaces, then we have 
\begin{enumerate}
\item $ $$\left\Vert \lambda_{\hbar}^{(k)}-\checklambda_{\hbar}^{(k)}\right\Vert _{(\domain(\tilde{\Delta}_{\hbar}),\|\cdot\|_{\tilde{\Delta}_{\hbar}})}\le C\hbar^{\epsilon}$
for $0\le k\le n$,
\item if $k\neq k'$, 
\[
\|\lambda_{\hbar}^{(k)}\circ\hbar\tilde{\Delta}_{\hbar}\circ\lambda_{\hbar}^{(k')}\|_{(\domain(\Delta_{\hbar}),\|\cdot\|_{\Delta_{\hbar}})\to L_{N}^{2}(P)}\le C\hbar^{\epsilon},
\]

\item if $0\le k\le n$, 
\[
\left\Vert \lambda_{\hbar}^{(k)}\circ\hbar\tilde{\Delta}_{\hbar}\circ\lambda_{\hbar}^{(k)}-\left(d+2k\right)\cdot\lambda_{\hbar}^{(k)}\right\Vert _{(\domain(\Delta_{\hbar}),\|\cdot\|_{\Delta_{\hbar}})\to L_{N}^{2}(P)}\le C\hbar^{\epsilon},
\]

\item for $k=n+1$, we have 
\[
\|\lambda_{\hbar}^{(n+1)}\circ\hbar\tilde{\Delta}_{\hbar}(u)\|_{L^{2}}\ge\left(d+2k+1-C\hbar^{\epsilon}\right)\|u\|_{L^{2}}\quad\mbox{for \ensuremath{u\in\widetilde{\mathfrak{H}}=\mathfrak{H}^{(n+1)}}}.
\]

\end{enumerate}
\end{lem}
Therefore, by the general theorem on perturbation of closed linear
operators \cite[chap.IV, th. 1.16]{kato_book}, we obtain an analogue
of Theorem \ref{thm:band_structure-of_Laplacian} for the rough Laplacian
$\tilde{\Delta}_{\hbar}$.
\begin{thm}
\label{thm:Bund_spectrum_for_tildeDelta}There exists a small constant
$\epsilon>0$ such that, for any $\alpha>0$, we have 
\[
dist\left(Spec\left(\hbar\tilde{\Delta}_{\hbar}\right)\cap\left\{ \left|z\right|<\alpha\right\} ,\{d+2k,\; k\in\mathbb{N}\}\right)\le\hbar^{\epsilon}
\]

\end{thm}
when $\hbar$ is sufficiently small.

Further we have 
\begin{lem}
\label{lem:Relation_between_ranks}There exists $C_{0}>0$ s.t. for
sufficiently small $\hbar>0$, we have
\[
C_{0}^{-1}\hbar^{-d}\leq\dim\mathcal{H}_{k}=\dim\mathfrak{H_{k}}\leq C_{0}\hbar^{-d}\quad\mbox{for }0\le k\le n.
\]
\end{lem}
\begin{proof}
To prove the equality $\dim\mathcal{H}_{k}=\dim\mathfrak{H_{k}}$,
it is enough to check that 
\[
\mathrm{rank}\,\tau_{\hbar}^{(k)}=\mathrm{rank}\,\lambda_{\hbar}^{(k)}\quad\mbox{for \ensuremath{0\le k\le n},}
\]
or equivalently 
\[
\mathrm{rank}\,\oplus_{i=0}^{k}\tau_{\hbar}^{(i)}=\mathrm{rank}\,\oplus_{i=0}^{k}\lambda_{\hbar}^{(i)}\quad\mbox{for \ensuremath{0\le k\le n}.}
\]
The latter relation would follow if we show 
\[
c\|\left(\oplus_{i=0}^{k}\lambda_{\hbar}^{(i)}\right)u\|_{L^{2}}\le\|\left(\oplus_{i=0}^{k}\lambda_{\hbar}^{(i)}\right)\circ\left(\oplus_{i=0}^{k}\tau_{\hbar}^{(i)}\right)\circ\left(\oplus_{i=0}^{k}\lambda_{\hbar}^{(i)}\right)u\|_{L^{2}}\le C\|\left(\oplus_{i=0}^{k}\lambda_{\hbar}^{(i)}\right)u\|_{L^{2}}
\]
 and 
\[
c\|\left(\oplus_{i=0}^{k}\tau_{\hbar}^{(i)}\right)u\|_{\mathcal{H}_{\hbar}^{r}}\le\|\left(\oplus_{i=0}^{k}\tau_{\hbar}^{(i)}\right)\circ\left(\oplus_{i=0}^{k}\lambda_{\hbar}^{(i)}\right)\circ\left(\oplus_{i=0}^{k}\tau_{\hbar}^{(i)}\right)u\|_{\mathcal{H}_{\hbar}^{r}}\le C\|\left(\oplus_{i=0}^{k}\tau_{\hbar}^{(i)}\right)u\|_{\mathcal{H}_{\hbar}^{r}}
\]
for a constant $0<c<C$, independent of $\hbar$. But these are immediate
consequences of Proposition \ref{pp1-3} (and the construction of
the projection operators $\tau_{\hbar}^{(k)}$ and $\lambda_{\hbar}^{(k)}$.

It remains to show 
\begin{equation}
\mathrm{rank}\,\lambda_{\hbar}^{(k)}\asymp\hbar^{-d}\quad\mbox{for \ensuremath{0\le k\le n}.}\label{eq:rough_dimension_estimate}
\end{equation}
For each point $x\in M$, we associate a smooth function 
\[
\varphi_{x}:=\check{\lambda}_{\hbar}^{(k)}(\delta_{x})
\]
where $\delta_{x}$ is the Dirac measure at the point $x$. (The right
hand side is well-defined and give a smooth function that concentrates
around $x$.) Consider positive constants $0<c<C$ and take a finite
subset of points $Q_{\hbar}$ on $M$ so that the mutual distance
between two points in $Q_{\hbar}$ is in between $c\cdot\hbar^{1/2}$
and $C\cdot\hbar^{1/2}$. Let 
\[
\mathscr{Q}_{\hbar}=\{\lambda_{\hbar}^{(k)}(\varphi_{x})\mid x\in Q_{\hbar}\}\subset\mathrm{Im}\,\lambda_{\hbar}^{(k)}.
\]
Note that $\lambda_{\hbar}^{(k)}(\varphi_{x})$ is close to $\varphi_{x}$
from Claim 1 in Proposition \ref{lm:lambda_Delta}. It is not difficult
to check that
\begin{enumerate}
\item if we let the constants $c,C$ large, the subset $\mathscr{Q}_{\hbar}$
is linearly independent, and
\item if we let the constants $c,C$ small, the subset $\mathscr{Q}_{\hbar}$
span the whole space $\mathrm{Im}\,\lambda_{\hbar}^{(k)}$.
\end{enumerate}
\noindent Indeed, to prove (1), we have only to observe that, if the
constants $c,C$ are sufficiently large, the $L^{2}$-scalar product
between different elements $\varphi_{x},\varphi_{x'}$ in $\mathscr{Q}_{\hbar}$
decay rapidly with respect the distance between the corresponding
points $x,x'$(relative to the size $\hbar^{1/2}$). To prove (2),
we see that, if the constants $c,C$ are sufficiently small, any element
of $\mathrm{Im}\,\lambda_{\hbar}^{(k)}$ is well approximated by the
linear combinations of the element in $\mathscr{Q}_{\hbar}$ and then,
by successive approximation, it is really contained in the subspace
spanned by $\mathscr{Q}_{\hbar}$. Clearly the claims (1) and (2)
imply (\ref{eq:rough_dimension_estimate}).
\end{proof}
Finally we compete the proof of Theorem \ref{thm:band_structure-of_Laplacian}
and Claim 1 of Theorem \ref{thm:More-detailled-description_of_spectrum}.
We show that the operator $\tilde{\Delta}_{\hbar}$ is continuously
deformed to the geometric Laplacian $\Delta_{\hbar}$, keeping the
band structure described in Theorem \ref{thm:Bund_spectrum_for_tildeDelta}.
For this purpose, we take a continuous one-parameter family of splitting
of the tangent bundle 
\[
TM=E_{+}^{(t)}\oplus E_{-}^{(t)}
\]
with $t\in[0,1]$ the parameter, such that 
\begin{itemize}
\item $E_{+}^{(0)}=E_{u}$ and $E_{-}^{(0)}=E_{s}$, that is, the splitting
above coincides with the hyperbolic splitting associated to $f$ when
$t=0$.
\item the sub-bundles $E_{\pm}^{(1)}$ for $t=1$ are $C^{\infty}$ and
orthogonal with respect to the Riemann metric $g$ on $M$.
\end{itemize}
Then we consider a continuous deformation $\{\kappa_{i,t}\}_{i=1}^{I_{\hbar}}$
of the atlas $\{\kappa_{i}\}_{i=1}^{I_{\hbar}}$ and, correspondingly,
the deformation $\{\psi_{i,t};1\le i\le I_{\hbar}\}$ of the family
$\{\psi_{i};1\le i\le I_{\hbar}\}$ of functions so that the all the
conditions in Proposition \ref{prop:Local-charts} hold uniformly
for $t\in[0,1]$, but with the sub-bundles $E^{u}$ and $E^{s}$ in
the condition (2) replaced by $E_{+}^{t}$ and $E_{-}^{t}$. 

We consider the rough Laplacian $\tilde{\Delta}_{\hbar,t}$ defined,
similarly to $\tilde{\Delta}_{\hbar}$ in (\ref{eq:relation_Delta}),
from the Euclidean rough Laplacian on local charts $\{\kappa_{i,t}\}_{i=1}^{I_{\hbar}}$
and the family of functions $\{\psi_{i,t};1\le i\le I_{\hbar}\}$.
The argument in the former part of this subsection holds true uniformly
for $t\in[0,1]$, that is, we can consider the spectral projection
operators $\lambda_{\hbar,t}^{(k)}$, $0\le k\le n+1$, for $\tilde{\Delta}_{\hbar,t}$,
which corresponds to $\lambda_{\hbar}^{(k)}$ for $\tilde{\Delta}_{\hbar}$.
Since the deformation is continuous, we see by homotopy argument that
$\mathrm{rank}\:\lambda_{\hbar,\rho}^{(k)}=\dim\mathfrak{H}_{k,\rho}$
is constant for $t\in[0,1]$. In particular we have $\mathrm{rank}\:\lambda_{\hbar,1}^{(k)}=\mathrm{rank}\:\lambda_{\hbar,0}^{(k)}$.

Note that the operator $\tilde{\Delta}_{\hbar,1}$ is close to the
geometric rough Laplacian $\Delta_{\hbar}$. In fact, since the derivative
$(D\kappa_{i,1})_{0}$ at the origin is an isometry with respect to
the Euclidean metric on $\real^{2d}$ and the Riemann metric $g$
on $M$, we can check by using the local expression of the geometric
rough Laplacian in Proposition \ref{eq:Delta_local_coordinates} that
we have
\[
\|\hbar\tilde{\Delta}_{\hbar,1}-\hbar\Delta_{\hbar}:(\mathcal{D}(\Delta_{\hbar}),\|\cdot\|_{\Delta_{\hbar}})\to L_{N}^{2}(P)\|\le C\hbar^{\epsilon}.
\]
Therefore, by the perturbation theorem \cite[chap.IV, th. 1.16]{kato_book}
for closed operators, we obtain the ``band structure'' stated in
the former part of Theorem \ref{thm:band_structure-of_Laplacian}.
We can also see that the number of eigenvalues of the geometric rough
Laplacian $\Delta_{\hbar}=D^{*}D$ in the $k$-th band is same as
that for $\tilde{\Delta}_{\hbar}$. Hence, from Lemma \ref{lem:Relation_between_ranks},
we obtain the rough upper and lower bound on the rank of spectral
projectors $\mathfrak{P}_{k}$ in Theorem \ref{thm:band_structure-of_Laplacian}
and also Claim (1) of Theorem \ref{thm:More-detailled-description_of_spectrum}.

\section{\label{sec:2}Proof of Th. \ref{thm:Discrete-spectrum-F_G}, \ref{thm:band_structure-1}.
Extension of the transfer operator to the Grassmanian bundle.}

In this section, we explain how we obtain the results stated in Subsection
\ref{sub:Spectral-results-with-Grassman}, namely Theorem \ref{thm:Discrete-spectrum-F_G},
Theorem \ref{thm:band_structure-1}, Theorem \ref{thm:weyl_law_extended}
and Theorem \ref{cor:The-special-choice_V0} for the Grassmanian extensions.
The proof are obtained by modifying the argument in the previous sections,
and, for the most part, the extensions are rather formal and easy.
The most essential difference is in Proposition \ref{prop:Local-charts-1},
where we take systems of local coordinate charts depending on $N$(
or $\hbar$). Unlike the corresponding statement, Proposition \ref{prop:Local-charts},
the local coordinate charts will be (metrically) singular in the fiber
directions and the singularity will increase as $N$ tends to infinity.
With this choice of local coordinate charts, the non-smooth section
$E_{u}$ will look ``flat'' in the local coordinates. A problem
that may happen with this choice of singular local coordinates is
that the coordinate change transformations and the flow viewed in
such local coordinate charts may be also singular. In the proof of
Proposition \ref{prop:Local-charts-1}, we show that this problem
actually does not occur. We will give a detailed argument on this
point. Once we establish the local coordinates, we can follow the
argument in the previous sections almost literally. So we will just
give the corresponding statements to clarify the correspondence and
skip most the proofs referring those of the corresponding statements
in the previous sections.

\subsection{Discussion about the linear model\label{sub:Discussion-linear_model}}

We first discuss about the extension of the argument in Section \ref{sec:Resonances-of-linear_exp_3}
and \ref{sec:Resonance-of-hyperbolic_preq_4} about the prequantum
transfer operators for linear hyperbolic maps. Instead of a hyperbolic
symplectic linear map $B$ in (\ref{eq:hyperbolic_f}), we consider
a linear map $\widetilde{B}:\real^{2d+d'}\to\real^{2d+d'}$ of the
form
\begin{equation}
\widetilde{B}(q,p,s)=(Aq,^{t}A^{-1}p,\widehat{A}s)\label{eq:hyperbolic_linear_map_extended}
\end{equation}
where $(q,p,s)$ with $q,p\in\real^{d}$ and $s\in\real^{d'}$ denote
the coordinates on $\real^{2d+d'}=\real^{d}\oplus\real^{d}\oplus\real^{d'}$,
$A:\real^{d}\to\real^{d}$ is an expanding linear map satisfying $\|A^{-1}\|\le1/\lambda$
for some $\lambda>1$ and $\widehat{A}:\real^{d'}\to\real^{d'}$ is
a contracting linear map satisfying $\|\widehat{A}\|\le1/\lambda$.
This is a hyperbolic linear map with stable subspace $\{0\}\oplus\real^{d}\oplus\real^{d'}$
and unstable subspace $\real^{d}\oplus\{0\}\oplus\{0\}$. The $L^{2}$-normalized
transfer operator associated to $\widetilde{B}$ is 
\[
\widetilde{\mathcal{L}}:L^{2}(\real^{2d+d'})\to L^{2}(\real^{2d+d'}),\quad\widetilde{\mathcal{L}}u=\frac{1}{\sqrt{|\det\widehat{A}|}}\cdot u\circ A^{-1}.
\]
Before we proceed, we put a remark.
\begin{rem}
A simple idea to treat the transfer operator $\widetilde{\mathcal{L}}$
as above is to regard it as the tensor product of two transfer operators,
one associated to the hyperbolic linear map $A\oplus^{t}A^{-1}$ and
the other associated to the contracting linear map $\widehat{A}$.
We may then apply the results in Section \ref{sec:Resonance-of-hyperbolic_preq_4}
to the former factor and that in \ref{sec:Resonances-of-linear_exp_3}
to (the adjoint of) the latter, and show a band structure of the spectrum
of the transfer operator $\widetilde{\mathcal{L}}$ on a Hilbert space.
However the Hilbert space that appears in such an argument has singular
properties with respect to the action of non-linear diffeomorphisms
that break the product structure. For this reason, we take a similar
but different way. 
\end{rem}
Let 
\[
\Bargmann_{(x,s)}:L^{2}(\real_{(x,s)}^{2d+d'})\to L^{2}(\real_{(x,s)}^{2d+d'}\oplus\real_{(\xi_{x},\xi_{s})}^{2d+d'})
\]
be the Bargmann transform defined by 
\begin{equation}
\Bargmann_{(x,s)}:=\Bargmann_{x}\otimes\Bargmann_{s}\label{eq:Bargmann_extended}
\end{equation}
where $\Bargmann_{x}:L^{2}(\real_{x}^{2d})\to L^{2}(\real_{(x,\xi)}^{4d})$
is the slight modification of the Bargmann transform given in (\ref{eq:Bargmann_modified})
and $\Bargmann_{s}:L^{2}(\real_{s}^{d'})\to L^{2}(\real_{(s,\xi_{s})}^{2d'})$
is the standard Bargmann transform given in (\ref{eq:def_Bargman_transform})
with setting $D=d'$. Let 
\[
\Bargmann_{(x,s)}^{*}:=\Bargmann_{x}^{*}\times\Bargmann_{s}^{*}:L^{2}(\real_{(x,s)}^{2d+d'}\oplus\real_{(\xi_{x},\xi_{s})}^{2d+d'})\to L^{2}(\real_{(x,s)}^{2d+d'})
\]
be the $L^{2}$ adjoint of $\Bargmann_{(x,s)}$. The lift of the operator
$\widetilde{\mathcal{L}}$ with respect to the Bargmann transform
$\Bargmann_{(x,s)}$ is defined as before:
\[
\widetilde{\mathcal{L}}^{\mathrm{lift}}:=\Bargmann_{(x,s)}\circ\widetilde{\mathcal{L}}\circ\Bargmann_{(x,s)}^{*}:L^{2}(\real_{(x,s)}^{2d+d'}\oplus\real_{(\xi_{x},\xi_{s})}^{2d+d'})\to L^{2}(\real_{(x,s)}^{2d+d'}\oplus\real_{(\xi_{x},\xi_{s})}^{2d+d'}).
\]
Here the space $\real_{(x,s)}^{2d+d'}\oplus\real_{(\xi_{x},\xi_{s})}^{2d+d'}$
is identified with the cotangent bundle of $\real_{(x,s)}^{2d+d'}$
equipped with the coordinates
\[
(q,p,s,\xi_{q},\xi_{p},\xi_{s})=(x,s,\xi_{x},\xi_{s})
\]
where $x=(q,p)\in\real^{d}\oplus\real^{d}$ and $s\in\real^{d'}$
is the coordinates on $\real^{2d+d'}=\real^{2d}\oplus\real^{d'}$
and $\xi_{x}=(\xi_{q},\xi_{p})\in\real^{d}\oplus\real^{d}$ and $\xi_{s}\in\real^{d'}$
are their respective dual coordinates. Imitating the argument in Section
\ref{sec:Resonance-of-hyperbolic_preq_4}, we introduce a different
coordinate system 
\begin{equation}
(\nu_{q},\nu_{p},\zeta_{q},\zeta_{p},s,\xi_{s})=(\nu,\zeta,s,\xi_{s})\label{eq:coordinates_extended}
\end{equation}
on $\real^{2d+d'}\oplus\real^{2d+d'}$, where $(\nu_{q},\nu_{p},\zeta_{q},\zeta_{p})=(\nu,\zeta)\in\real^{4d}=\real^{2d}\oplus\real^{2d}$
is the coordinates introduced in Proposition \ref{prop:Normal-coordinates.}
while we do not change the coordinates $(s,\xi_{s})$. The corresponding
coordinate change transformation is written 
\[
\,\,\widetilde{\Phi}:\real_{(x,s)}^{2d+d'}\oplus\real_{(\xi_{x},\xi_{s})}^{2d+d'}\to\real_{\nu}^{2d}\oplus\real_{\zeta}^{2d}\oplus\real_{s}^{d'}\oplus\real_{\xi_{s}}^{d'}
\]
\[
\widetilde{\Phi}(q,p,s,\xi_{q},\xi_{p},\xi_{s})=(\nu_{q},\nu_{p},\zeta_{p},\zeta_{q},s,\xi_{s}).
\]
By this transformation, the standard symplectic form $\widetilde{\Omega}_{0}=dx\wedge d\xi_{x}+ds\wedge d\xi_{s}$
is transferred to 
\[
(D\widetilde{\Phi}^{*})^{-1}(\widetilde{\Omega}_{0})=d\nu_{q}\wedge d\nu_{p}+d\zeta_{p}\wedge d\zeta_{q}+ds\wedge d\xi_{s}
\]
and the metric $\widetilde{g}_{0}=\frac{1}{2}dx^{2}+2d\xi^{2}+ds^{2}+d\xi_{s}^{2}$
is transferred to the standard Euclidean metric 
\[
(D\widetilde{\Phi}^{*})^{-1}(g_{0})=d\nu^{2}+d\zeta^{2}+ds^{2}+d\xi_{s}^{2}.
\]
The unitary operator associated to the coordinate change $\widetilde{\Phi}$
is defined as 
\[
\widetilde{\Phi}^{*}:L^{2}\left(\real_{\nu}^{2d}\oplus\real_{\zeta}^{2d}\oplus\real_{s}^{d'}\oplus\real_{\xi_{s}}^{d'}\right)\to L^{2}\left(\real_{(x,s)}^{2d+d'}\oplus\real_{(\xi_{x},\xi_{s})}^{2d+d'}\right),\quad\widetilde{\Phi}^{*}u:=u\circ\widetilde{\Phi}.
\]
Under these settings, we can follow the arguments in Section \ref{sec:Resonance-of-hyperbolic_preq_4}
and obtain the next proposition, which corresponds to Proposition
\ref{prop:4.5}. 
\begin{prop}
\label{prop:Product_expression_of_tildeL} The following diagram commutes:
\begin{equation}
\begin{CD}L^{2}\left(\real_{(x,s)}^{2d+d'}\right)@>\widetilde{\mathcal{L}}>>L^{2}\left(\real_{(x,s)}^{2d+d'}\right)\\
@AA\widetilde{\mathcal{U}}A@AA\widetilde{\mathcal{U}}A\\
L^{2}\left(\real_{\nu_{q}}^{d}\right)\otimes L^{2}\left(\real_{(\zeta_{p},\xi_{s})}^{d+d'}\right)@>L_{A}\otimes L_{A\oplus^{t}\widehat{A}^{-1}}>>L^{2}\left(\real_{\nu_{q}}^{d}\right)\otimes L^{2}\left(\real_{(\zeta_{p},\xi_{s})}^{d+d'}\right)
\end{CD}\label{eq:expression_tensorproduct-2-1}
\end{equation}
with the unitary operator $\widetilde{\mathcal{U}}$ defined by 
\[
\widetilde{\mathcal{U}}=\Bargmann_{(x,s)}^{*}\circ\widetilde{\Phi}^{*}\circ(\Bargmann_{\nu_{q}}\otimes\Bargmann_{(\zeta_{p},\xi_{s})}).
\]
Equivalently, for the lifted operators, the following diagram commutes:
\begin{equation}
\begin{CD}L^{2}\left(\real_{(x,s)}^{2d+d'}\oplus\real_{(\xi_{x},\xi_{s})}^{2d+d'}\right)@>\widetilde{\mathcal{L}}^{\mathrm{lift}}>>L^{2}\left(\real_{(x,s)}^{2d+d'}\oplus\real_{(\xi_{x},\xi_{s})}^{2d+d'}\right)\\
@AA\widetilde{\Phi}^{*}A@AA\widetilde{\Phi}^{*}A\\
L^{2}\left(\real_{\nu}^{2d}\right)\otimes L^{2}\left(\real_{\zeta}^{2d}\oplus(\real_{s}^{d'}\oplus\real_{\xi_{s}}^{d'})\right)@>\Llift_{A}\otimes\Llift_{A\oplus^{t}\widehat{A}^{-1}}>>L^{2}\left(\real_{\nu}^{2d}\right)\otimes L^{2}\left(\real_{\zeta}^{2d}\oplus(\real_{s}^{d'}\oplus\real_{\xi_{s}}^{d'})\right)
\end{CD}\label{cd:expression_as_tensorproduct-1}
\end{equation}

\end{prop}
We next introduce the anisotropic Sobolev space for the extended situation. 
\begin{defn}
\label{def:escape_function_Hr-1}We define the \emph{escape function
(or the weight function)} 
\[
\mathcal{\widetilde{W}}_{\hbar}^{r}:\real_{(x,s)}^{2d+d'}\oplus\real_{(\xi_{x},\xi_{s})}^{2d+d'}\to\real_{+}\quad\mbox{and}\quad\mathcal{\widetilde{W}}_{\hbar}^{r,\pm}:\real_{(x,s)}^{2d+d'}\oplus\real_{(\xi_{x},\xi_{s})}^{2d+d'}\to\real_{+}
\]
by

\begin{equation}
\mathcal{\widetilde{W}}_{\hbar}^{r}\left(x,\xi\right):=W_{\hbar}^{r}\left((\zeta_{p},\xi_{s}),(\zeta_{q},s)\right)\quad\mbox{and}\quad\mathcal{\widetilde{W}}^{r,\pm}\left(x,\xi\right):=W_{\hbar}^{r,\pm}\left((\zeta_{p},\xi_{s}),(\zeta_{q},s)\right)\label{eq:def_W_W+-1}
\end{equation}
where the functions $W_{\hbar}^{r}$ and $W_{\hbar}^{r,\pm}$ are
those defined in Definition \ref{def:Wrh} with $D=d+d'$, and $\zeta_{p},\zeta_{q},s,\xi_{s}$
are those in the coordinates (\ref{eq:coordinates_extended}). The
\emph{anisotropic Sobolev space}\textbf{ $\mathcal{\widetilde{H}}_{\hbar}^{r}\left(\real_{(x,s)}^{2d+d'}\right)$}
is the completion of the Schwartz space $\mathcal{S}(\real_{(x,s)}^{2d+d'})$
with respect to the norm 
\[
\|u\|_{\mathcal{\widetilde{H}}_{\hbar}^{r}}:=\left\Vert \mathcal{\widetilde{W}}_{\hbar}^{r}\cdot\Bargmann_{(x,s)}u\right\Vert _{L^{2}}.
\]
Let $\mathcal{\widetilde{H}}_{\hbar}^{r,\pm}(\real^{2d+d'})$ be the
Hilbert space defined in the parallel manner with $\mathcal{\widetilde{W}}_{\hbar}^{r}(\cdot)$
replaced by $\widetilde{\mathcal{W}}_{\hbar}^{r,\pm}(\cdot)$. 
\end{defn}
\[
\]
Below we fix an integer $n\ge0$ and assume that the parameter $r$
in the definition of the anisotropic Sobolev space \textbf{$\mathcal{\widetilde{H}}_{\hbar}^{r}\left(\real_{(x,s)}^{2d+d'}\right)$
}satisfies the condition 
\begin{equation}
r>n+2+2(2d+d')\label{eq:choice_of_r_2-1}
\end{equation}
which corresponds to (\ref{eq:choice_of_r_3}) in Subsection \ref{ss:trumcation}.
The next definition of projection operators correspond to Definition
\ref{eq:def_t}. Note that we will use the same symbol for the new
projection operators as the corresponding projection operators in
Definition \ref{eq:def_t}. Since these two families of projection
operators act on different Hilbert spaces, this will not introduce
confusion. 
\begin{defn}
For $0\le k\le n$, we consider the projection operators
\begin{equation}
\romet_{\hbar}^{(k)}:=\widetilde{\mathcal{U}}\circ\left(\mbox{Id}\otimes T^{\left(k\right)}\right)\circ\widetilde{\mathcal{U}}^{-1}\quad:\widetilde{\mathcal{H}}_{\hbar}^{r}(\real_{(x,s)}^{2d+d'})\to\mathcal{\widetilde{H}}_{\hbar}^{r}(\real_{(x,s)}^{2d+d'})\label{eq:def_t-1}
\end{equation}
and 
\begin{equation}
\tilde{\romet}_{\hbar}:=\mathrm{Id}-\sum_{k=0}^{n}t_{\hbar}^{(k)}=\widetilde{\mathcal{U}}\circ\left(\mbox{Id}\otimes\widetilde{T}\right)\circ\widetilde{\mathcal{U}}^{-1}\quad:\mathcal{\widetilde{H}}_{\hbar}^{r}(\real_{(x,s)}^{2d+d'})\to\mathcal{\widetilde{H}}_{\hbar}^{r}(\real_{(x,s)}^{2d+d'})\label{eq:def_tt-1}
\end{equation}
where $T^{\left(k\right)}$ and $\widetilde{T}$ are the projection
operators introduced in (\ref{eq:def_Tk}) and (\ref{eq:def_T_tilde})
page \pageref{eq:def_T_tilde} respectively for the setting $D=d+d'$.
\end{defn}
We can translate the argument in Section \ref{sec:Resonance-of-hyperbolic_preq_4}
(especially that in Subsection \ref{ss:anisoSob}) to get the next
result, which corresponds to Proposition \ref{prop:prequatum_op_for_hyp_linear}. 
\begin{prop}
\label{prop:prequatum_op_for_hyp_linear-1} The projection operators
$\romet_{\hbar}^{(k)}$, $0\le k\le n$, and $\tilde{\romet}_{\hbar}$,
defined in (\ref{eq:def_t-1}) and (\ref{eq:def_tt-1}), form a complete
set of mutually commutative projection operators on $\mathcal{\widetilde{H}}_{\hbar}^{r}(\real_{(x,s)}^{2d+d'})$.
These operators also commute with the operator $\widetilde{\mathcal{L}}$.
Consequently the space $\mathcal{\widetilde{H}}_{\hbar}^{r}(\real_{(x,s)}^{2d+d'})$
has a decomposition invariant under the action of $\widetilde{\mathcal{L}}$:
\[
\mathcal{\widetilde{H}}_{\hbar}^{r}(\real^{2d+d'})=E'_{0}\oplus E'_{1}\oplus\cdots\oplus E'_{n}\oplus\widetilde{E}'\qquad\mbox{where \ensuremath{E'_{k}=\mathrm{Im}\,\romet_{\hbar}^{(k)}}and \ensuremath{\widetilde{E}'=\mathrm{Im}\,\tilde{\romet}_{\hbar}}}.
\]
For this decomposition, we have the following estimates:
\begin{enumerate}
\item For every $0\le k\le n$ and for every $u\in E'_{k}$, we have 
\[
C_{0}^{-1}\frac{|\det(A\oplus(^{t}\widehat{A}^{-1}))|^{-1/2}}{\|A\oplus(^{t}\widehat{A}^{-1})\|_{\max}^{k}}\cdot\|u\|_{\mathcal{\widetilde{H}}_{\hbar}^{r}}\le\|\widetilde{\mathcal{L}}u\|_{\mathcal{\widetilde{H}}_{\hbar}^{r}}\le C_{0}\frac{|\det(A\oplus(^{t}\widehat{A}^{-1}))|^{-1/2}}{\|A\oplus(^{t}\widehat{A}^{-1})\|_{\min}^{k}}\cdot\|u\|_{\mathcal{\widetilde{H}}_{\hbar}^{r}}.
\]

\item The operator norm of $\widetilde{\mathcal{L}}:\widetilde{E}'\to\widetilde{E}'$
is bounded by 
\[
C_{0}\cdot\max\left\{ \frac{|\det(A\oplus(^{t}\widehat{A}^{-1}))|^{-1/2}}{\|A\oplus(^{t}\widehat{A}^{-1})\|_{\min}^{n+1}},\frac{|\det(A\oplus(^{t}\widehat{A}^{-1}))|^{1/2}}{\|A\oplus(^{t}\widehat{A}^{-1})\|_{\min}^{r}}\right\} .
\]
 
\end{enumerate}
The constant $C_{0}$ is independent of $A$ and $\hbar$.
\end{prop}

\subsection{Treatment of non-linearity\label{sub:nonlinear_extended}}

We next explain how we modify the argument in Section \ref{sec:Nonlinear-prequantum-maps}
on the action of non-linear diffeomorphisms. Below we give definitions
and related statements with some remarks on the correspondence to
the argument in Section \ref{sec:Nonlinear-prequantum-maps}. We omit
the proofs because we can obtain them by translating the corresponding
ones in Section \ref{sec:Nonlinear-prequantum-maps}. Recall again
that $0<\beta<1$ is the H\"older exponent of the hyperbolic splitting
for the Anosov map $f:M\rightarrow M$. We take and fix a constant
$0<\theta<1$ so small that
\begin{equation}
0<\theta<\beta/20.\label{eq:choiceTheteExt}
\end{equation}
Let $\mathbb{D}^{(d)}\left(c\right)$ be the open ball of radius $c>0$
on $\real^{d}$ with center at the origin. Instead of the \emph{``Setting
I''} given in Subsection \ref{ss:trumcation}, we consider the following
setting:

\bigskip{}

\noindent%
\framebox{\begin{minipage}[t]{1\columnwidth}%
\textbf{Setting I$^{\mathrm{ext}}$:} \textit{~For each $\hbar>0$,
there is a given set $\widetilde{\scrX}_{\hbar}$ of $C^{\infty}$
functions on $\real^{2d+d'}$ such that, for all $\psi\in\mathscr{\widetilde{X}}_{\hbar}$
and $\hbar>0$, }
\begin{description}
\item [{(C1)}] \textit{the support of $\psi$ is contained in the disk
$\mathbb{D}^{(2d+d')}\left(C_{*}\hbar^{1/2-\theta}\right)\subset\real^{2d+d'}$
and}
\item [{(C2)}] \textit{$\left|\partial_{x}^{\alpha}\psi(x)\right|<C_{\alpha}\hbar^{-\left(\frac{1}{2}-\theta\right)\left|\alpha\right|}$
for each multi-index $\alpha\in\mathbb{N}^{2d+d'}$,}
\end{description}
\textit{where $C_{*}>0$ and $C_{\alpha}>0$ are constants independent
of $\psi\in\mathscr{\widetilde{X}}_{\hbar}$ and $\hbar>0$.}%
\end{minipage}}

\bigskip{}

The Bargmann projection operator in the extended setting is 
\[
\BargmannP_{(x,s)}:=\Bargmann_{(x,s)}\circ\Bargmann_{(x,s)}^{*}:L^{2}(\real_{(x,s)}^{2d+d'}\oplus\real_{(\xi_{x},\xi_{s})}^{2d+d'})\to L^{2}(\real_{(x,s)}^{2d+d'}\oplus\real_{(\xi_{x},\xi_{s})}^{2d+d'}).
\]
The following statements correspond to Lemma \ref{lm:lift_of_multiplication_operator},
Corollary \ref{lm:XM_exchange} and \ref{cor:transpose} respectively. 

For each $\psi\in\mathscr{\widetilde{X}}_{\hbar}$, let 
\[
\multiplication^{\mathrm{lift}}(\psi)=\Bargmann_{(x,s)}\circ\multiplication(\psi)\circ\Bargmann_{(x,s)}^{*}
\]
 be the lift of the multiplication operator $\multiplication(\psi)$. 
\begin{lem}
\label{lm:lift_of_multiplication_operator-ext}There exists a constant
$C>0$ such that, for any $\hbar>0$ and $\psi\in\widetilde{\scrX}_{\hbar}$,
we have
\begin{equation}
\left\Vert \multiplication^{\mathrm{lift}}(\psi)-\multiplication(\psi\circ\pi)\circ\BargmannP_{(x,s)}\right\Vert _{L^{2}(\real^{2d+d'}\oplus\real^{2d+d'},(\mathcal{\widetilde{W}}_{\hbar}^{r})^{2})}<C\hbar^{\theta}\label{eq:approximation_of_lift_of_multiplication-1}
\end{equation}
and 
\[
\left\Vert \multiplication^{\mathrm{lift}}(\psi)-\BargmannP_{(x,s)}\circ\multiplication(\psi\circ\pi)\right\Vert _{L^{2}(\real^{2d+d'}\oplus\real^{2d+d'},(\widetilde{\mathcal{W}}_{\hbar}^{r})^{2})}<C\hbar^{\theta}.
\]
where $\pi:\real_{(x,s,\xi_{x},\xi_{s})}^{4d+2d'}\to\real_{(x,s)}^{2d+d'}$
is the natural projection defined by $\pi(x,s,\xi_{x},\xi_{s})=(x,s)$.
$ $Consequently we have 
\begin{equation}
\left\Vert \left[\BargmannP_{(x,s)},\multiplication(\psi\circ\pi)\right]\right\Vert _{L^{2}(\real^{2d+d'}\oplus\real^{2d+d'},(\mathcal{\widetilde{W}}_{\hbar}^{r})^{2})}<C\hbar^{\theta}.\label{eq:commutator_BargmannP_multiplication_op-1}
\end{equation}
The same statement holds true with $\mathcal{\widetilde{W}}_{\hbar}^{r}$
replaced by $\mathcal{\widetilde{W}}_{\hbar}^{r,\pm}$.\end{lem}
\begin{cor}
\label{lm:XM_exchange-1} The multiplication operator $\multiplication(\psi)$
by $\psi\in\widetilde{\scrX}_{\hbar}$ extends to a bounded operator
on $\mathcal{\widetilde{H}}_{\hbar}^{r}(\real^{2d+d'})$ and, for
the operator norm, we have $\|\multiplication(\psi)\|_{\mathcal{\widetilde{H}}_{\hbar}^{r}(\real^{2d+d'})}<\|\psi\|_{\infty}+C\hbar^{\theta}$
for all $\psi\in\mathscr{\widetilde{X}}_{\hbar}$, with a constant
$C>0$ independent of $\hbar$ and $\psi$. 
\end{cor}
~
\begin{cor}
\label{cor:transpose-1}For $u,v\in\mathcal{\widetilde{H}}_{\hbar}^{r}(\real^{2d+d'})$
and $\ensuremath{\psi\in\widetilde{\scrX}_{\hbar}}$ we have
\[
\left|(u,\psi\cdot v)_{\mathcal{\widetilde{H}}_{\hbar}^{r}(\real^{2d+d'})}-(\overline{\psi}\cdot u,v)_{\mathcal{\widetilde{H}}_{\hbar}^{r}(\real^{2d+d'})}\right|<C\hbar^{\theta}\cdot\|u\|_{\mathcal{\widetilde{H}}_{\hbar}^{r}(\real^{2d+d'})}\cdot\|v\|_{\mathcal{\widetilde{H}}_{\hbar}^{r}(\real^{2d+d'})}
\]
where\textup{ $C$ is a constant independent of $u$, $v$, $\psi$
and $\hbar$. }
\end{cor}
The next lemma corresponds to Lemma \ref{cor:XT_exchange}.
\begin{lem}
\label{cor:XT_exchange-2}There exists a constant $C>0$ such that,
for any $\hbar>0$, $\psi\in\widetilde{\scrX}_{\hbar}$ and $0\le k\le n$,
we have 
\begin{align*}
\left\Vert \left[\multiplication(\psi),\romet_{\hbar}^{(k)}\right]\right\Vert _{\mathcal{\widetilde{H}}_{\hbar}^{r,-}(\real^{2d+d'})\to\mathcal{\widetilde{H}}_{\hbar}^{r,+}(\real^{2d+d'})}<C\hbar^{\theta} & .
\end{align*}

\end{lem}
The following corresponds to Lemma \ref{lm:trace_basic}, Corollary
\ref{cor:XT_pm} and Corollary \ref{cor:XT_exchange-Tr}
\begin{lem}
\label{lm:trace_basic-1}For $\psi\in\widetilde{\scrX}_{\hbar}$ and
$0\le k\le n$, the operator 
\[
\multiplication(\psi)\circ t_{\hbar}^{(k)}:\mathcal{\widetilde{H}}_{\hbar}^{r}(\real^{2d+d'})\to\mathcal{\widetilde{H}}_{\hbar}^{r}(\real^{2d+d'})
\]
 is a trace class operator. There exists a constant $C>0$, independent
of $\psi\in\mathcal{X}_{\hbar}$, $\hbar>0$ and $0\le k\le n$, such
that 
\[
\|\multiplication(\psi)\circ t_{\hbar}^{(k)}:\mathcal{H}_{\hbar}^{r}(\real^{2d})\to\mathcal{H}_{\hbar}^{r}(\real^{2d})\|_{tr}\le\frac{r(k,d)}{(2\pi\hbar)^{d}}\int|\psi(x,0)|\, dx+C\hbar^{-\theta d+\theta}
\]
and 

\[
\left|\mathrm{Tr}\,\left(\multiplication(\psi)\circ t_{\hbar}^{(k)}:\mathcal{H}_{\hbar}^{r}(\real^{2d})\to\mathcal{H}_{\hbar}^{r}(\real^{2d})\right)-\frac{r(k,d)}{(2\pi\hbar)^{d}}\int\psi(x,0)\, dx\right|\le C\hbar^{-\theta d+\theta}
\]
 where $\|\cdot\|_{tr}$ denotes the trace norm of an operator and
\[
r(k,d)=\binom{d+k-1}{d-1}=\mathrm{rank}\, T^{(k)}.
\]
The same statement holds true for $t_{\hbar}^{(k)}\circ\multiplication(\psi)$. \end{lem}
\begin{cor}
\label{cor:XT_pm-1}There exists a constant $C>0$, such that, \textup{for
$0\le k\le n$ and $\psi\in\widetilde{\scrX}_{\hbar}$,} 
\[
\|\multiplication(\psi)\circ t_{\hbar}^{(k)}:\mathcal{\widetilde{H}}_{\hbar}^{r,-}(\real^{2d+d'})\to\mathcal{\widetilde{H}}_{\hbar}^{r,+}(\real^{2d+d'})\|_{Tr}\le\frac{C\cdot r(k,d)}{(2\pi\hbar)^{d}}\int|\psi(x,0)|\, dx+C\hbar^{-\theta d+\theta}
\]

\end{cor}
~
\begin{cor}
\label{cor:XT_exchange-Tr-1}There exists a constant $C>0$, such
that 
\begin{align*}
 & \left\Vert \left[\multiplication(\psi),\romet_{\hbar}^{(k)}\right]:\mathcal{\widetilde{H}}_{\hbar}^{r,-}(\real^{2d+d'})\to\mathcal{\widetilde{H}}_{\hbar}^{r,+}(\real^{2d+d'})\right\Vert _{Tr}<C\hbar^{-\theta d+\theta}\quad\mbox{for \ensuremath{0\le k\le n}}
\end{align*}
for any $\psi\in\scrX_{\hbar}$.
\end{cor}
Next let us consider the function 
\begin{equation}
\widetilde{Y}_{\hbar}:T^{*}\mathbb{R}^{2d+d'}=\mathbb{R}_{(x,s)}^{2d+d'}\oplus\mathbb{R}_{(\xi_{x},\xi_{s})}^{2d+d'}\to[0,1],\quad Y_{\hbar}(x,\xi)=\chi(\hbar^{2\theta-1/2}|((x,s),(\xi_{x},\xi_{s})|).\label{def:function_Y_hbar-1}
\end{equation}
and the operator 
\begin{equation}
\widetilde{\mathcal{Y}}_{\hbar}:L^{2}(\real_{(x,s)}^{2d+d'})\to L^{2}(\real_{(x,s)}^{2d+d'}),\quad\widetilde{\mathcal{Y}}_{\hbar}=\Bargmann_{(x,s)}^{*}\circ\multiplication(\widetilde{Y}_{\hbar})\circ\Bargmann_{(x,s)}\label{def:operation_calY_hbar-1}
\end{equation}
with $\chi$ is as in (\ref{eq:def_chi}). The following two lemmas
correspond to Lemma \ref{lem:boundedness_of_calY} and \ref{YT_exchange}
in Subsection \ref{sub:Truncation-operations_phase}. 
\begin{lem}
\label{lem:boundedness_of_calY-1} The operator \textup{$\widetilde{\mathcal{Y}}_{\hbar}$}
extends naturally to a bounded operator on \textup{$\mathcal{\widetilde{H}}_{\hbar}^{r}(\real^{2d+d'})$}
and we have\textup{ 
\[
\|\mathcal{\widetilde{Y}}_{\hbar}\|_{\mathcal{\widetilde{H}}_{\hbar}^{r}(\real^{2d+d'})}<1+C\hbar^{\theta}
\]
and
\[
\|[\mathcal{\widetilde{Y}}_{\hbar},\multiplication(\psi)]\|_{\mathcal{\widetilde{H}}_{\hbar}^{r}(\real^{2d+d'})}<C\hbar^{\theta}\quad\mbox{for any }\psi\in\mathscr{\widetilde{X}}_{\hbar}
\]
}with some positive constants $C$ independent of $\hbar$ and $\psi$. 
\end{lem}
~
\begin{lem}
\label{YT_exchange-1} For $0\le k\le n$ and $\psi\in\mathscr{\widetilde{X}}_{\hbar}$,
we have 
\begin{align*}
 & \|(\mathrm{Id}-\mathcal{\widetilde{Y}}_{\hbar})\circ\multiplication(\psi)\circ\romet_{\hbar}^{(k)}\|_{\mathcal{\widetilde{H}}_{\hbar}^{r,-}(\real^{2d+d'})\to\mathcal{\widetilde{H}}_{\hbar}^{r,+}(\real^{2d+d'})}<C\hbar^{\theta}\intertext{and} & \|\romet_{\hbar}^{(k)}\circ(\mathrm{Id}-\mathcal{\widetilde{Y}}_{\hbar})\circ\multiplication(\psi)\|_{\mathcal{\widetilde{H}}_{\hbar}^{r,-}(\real^{2d+d'})\to\mathcal{\widetilde{H}}_{\hbar}^{r,+}(\real^{2d+d'})}<C\hbar^{\theta}
\end{align*}
with some constant $C>0$ independent of $\hbar$ and $\psi$.
\end{lem}
Below we give a few statements corresponding to those in Subsection
\ref{ss:nonlinear}. We now consider the following setting, in addition
to Setting I$^{\mathrm{ext}}$. This corresponds to ``Setting II''
in Subsection \ref{ss:nonlinear}.

\bigskip{}

\noindent%
\framebox{\begin{minipage}[t]{1\columnwidth}%
\noindent \textbf{Setting II$^{\mathrm{ext}}$:} \textit{~For every
$\hbar>0$, there is a given set $\widetilde{\scrG}_{\hbar}$ of $C^{\infty}$
diffeomorphisms 
\[
g:\mathbb{D}^{(2d+d')}(\hbar^{1/2-2\theta})\to g(\mathbb{D}^{(2d+d')}(\hbar^{1/2-2\theta}))\subset\real^{2d+d'}
\]
such that every $g\in\widetilde{\scrG}_{\hbar}$ satisfies }
\begin{description}
\item [{\textit{(G0)}}] \textit{$g$ has a skew product structure with
respect to the projection $p:\real_{(x,s)}^{2d+d'}\to\real_{x}^{2d}$,
that is, we may write $g$ as 
\[
g(x,s)=(\check{g}(x),\hat{g}(x,s))\in\real^{2d}\oplus\real^{d'}
\]
 for $(x,s)\in\mathbb{D}^{(2d+d')}(\hbar^{1/2-2\theta})\subset\real^{2d}\oplus\real^{d'}$,}
\item [{\textit{(G1)}}] \textit{$\check{g}$ is symplectic with respect
to the symplectic form $\omega$ on $\real^{2d}$ in (\ref{eq:Symplectic_form_on_Euclidean_space}),}
\item [{\textit{(G2)}}] \textit{$\check{g}(0)=0$, $|\hat{g}(0,0)|<C\hbar^{1/2+\theta}$
and $\|Dg(0)-\mathrm{Id}\|<C\max\{\hbar^{\beta(1/2-\theta)},\hbar^{(1-\beta)(1/2-\theta)+2\theta}\}$,
and}
\item [{\textit{(G3)}}] \textit{$|\partial^{\alpha}g|<C_{\alpha}\cdot\hbar^{-((1-\beta)(1/2-\theta)+2\theta)(|\alpha|-1)}$
for any multi-index $\alpha$ with $|\alpha|\ge2$,}
\end{description}
\textit{where $ $$C$ and $C_{\alpha}$ are positive constants that
do not depend on $\hbar$ nor $g\in\widetilde{\scrG}_{\hbar}$.}%
\end{minipage}}

\bigskip{}

\begin{rem}
\label{Rem:Setting2ext}Condition (G2) and (G3) in the setting above
is weaker than the literal translation of those in ``Setting II''
in Subsection \ref{ss:nonlinear}. Still we can get the proofs of
the propositions below by translating those of the corresponding propositions,
though we have to check that these weaker conditions are sufficient
to get the conclusions. The point is that the diffeomorphisms $g$
get closer to the identity in the $C^{\infty}$ sense as $\hbar\to+0$,
provided we look them in the scale $\hbar^{1/2}.$ See Remark \ref{Rem:modification}
below also.
\end{rem}
For $g\in\widetilde{\scrG}_{\hbar}$, we consider the Euclidean prequantum
transfer operator 
\begin{equation}
\widetilde{\prequantumL}_{g}:C_{0}^{\infty}(\mathbb{D}^{(2d+d')}(\hbar^{1/2-2\theta}))\to C_{0}^{\infty}(g(\mathbb{D}^{(2d+d')}(\hbar^{1/2-2\theta})))\label{eq:Lg-1}
\end{equation}
defined by 
\[
\widetilde{\prequantumL}_{g}\, u(x,s)=\frac{1}{|\det(Dg|_{\ker Dp})|}\cdot e^{-(i2\pi/\hbar)\cdot\mathcal{A}_{g}(p(g^{-1}(x,s)))}\cdot u(g^{-1}(x,s))
\]
with $\mathcal{A}_{g}$ be the function defined by (\ref{eq:A_g})
replacing $g$ by $\check{g}$. Let
\begin{equation}
\widetilde{\chi}_{\hbar}:\real_{(x,s)}^{2d+d'}\to[0,1],\quad\widetilde{\chi}_{\hbar}(x,s)=\chi(\hbar^{-1/2+\theta}|(x,s)|/2)\label{eq:tilde_chi_hbar}
\end{equation}
where $\chi$ is a $C^{\infty}$ function satisfying (\ref{eq:def_chi}).
The next two propositions correspond to Proposition \ref{lm:L_g_Y_almost_identity}
and \ref{pp:bdd_g} respectively.
\begin{prop}
\label{lm:L_g_Y_almost_identity-1} There exist constants $C>0$ and
$\epsilon>0$ such that, for any $\hbar>0$ and $g\in\widetilde{\scrG}_{\hbar}$,
we have
\[
\|\mathcal{\widetilde{Y}}_{\hbar}\circ(\widetilde{\prequantumL}_{g}-\mathrm{Id})\circ\multiplication(\widetilde{\chi}_{\hbar})\|_{\mathcal{\widetilde{H}}_{\hbar}^{r}(\real^{2d+d'})}<C\hbar^{\epsilon}\quad\mbox{and}\quad\left\Vert (\widetilde{\prequantumL}_{g}-\mathrm{Id})\circ\mathcal{\widetilde{Y}}_{\hbar}\circ\multiplication(\widetilde{\chi}_{\hbar})\right\Vert _{\mathcal{\widetilde{H}}_{\hbar}^{r}(\real^{2d+d'})}<C\hbar^{\epsilon}.
\]
\end{prop}
\begin{rem}
\label{Rem:modification} Concerning Remark \ref{Rem:Setting2ext},
we have to check that the argument in the proof of Proposition \ref{lm:L_g_Y_almost_identity}
works under the weaker assumption in the setting $\mathrm{II}{}^{\mathrm{ext}}$.
This is not difficult if we take the last comment in Remark \ref{Rem:Setting2ext}
into account and noting that, from (\ref{eq:choiceTheteExt}), we
have
\[
-(1-\beta)(1/2-\theta)-2\theta>-1/2+2\theta
\]
for the exponent that appeared in the condition (G3) in the setting
$\mathrm{II}{}^{\mathrm{ext}}$. 
\[
\]
\end{rem}
\begin{prop}
\label{pp:bdd_g-1}For any $g\in\widetilde{\scrG}_{\hbar}$, we have
\begin{equation}
\left\Vert \widetilde{\prequantumL}_{g}\circ\multiplication(\widetilde{\chi}_{\hbar})\right\Vert _{\mathcal{\widetilde{H}}_{\hbar}^{r,+}(\real^{2d+d'})\to\mathcal{\widetilde{H}}_{\hbar}^{r}(\real^{2d+d'})}\le C_{0}\quad\mbox{and }\quad\left\Vert \widetilde{\prequantumL}_{g}\circ\multiplication(\widetilde{\chi}_{\hbar})\right\Vert _{\mathcal{\widetilde{H}}_{\hbar}^{r}(\real^{2d+d'})\to\mathcal{\widetilde{H}}_{\hbar}^{r,-}(\real^{2d+d'})}\le C_{0}\label{eq:inequality-1}
\end{equation}
for sufficiently small $\hbar>0$, where $C_{0}>1$ is a constant
that depends only on $n$, $r$, $d$, $\theta$ and the choice of
the escape functions $W$ and $W^{\pm}$ in subsection \ref{ss:wL2}. 
\end{prop}
The following correspond to Lemma \ref{lm:L_g_almost_identity}, Corollary
\ref{cor:L_g_almost_commutes_with_t^k} and Lemma \ref{lm:L_g_almost_identity-Tr}.
\begin{lem}
\label{lm:L_g_almost_identity-1} There exist constants $C>0$ and
$\epsilon>0$ independent of $\hbar$ such that the following holds:
Let $\psi\in\widetilde{\scrX}_{\hbar}$ be supported on the disk $\mathbb{D}^{(2d+d')}(2\hbar^{1/2-\theta})$
and let $g\in\widetilde{\scrG}_{\hbar}$, $0\le k\le n$, then it
holds 
\[
\left\Vert (\widetilde{\prequantumL}_{g}-\mathrm{Id)}\circ\multiplication(\psi)\circ\romet_{\hbar}^{(k)}\right\Vert _{\mathcal{\widetilde{H}}_{\hbar}^{r}(\real^{2d+d'})}\le C\hbar^{\epsilon}
\]
 and 
\[
\left\Vert \romet_{\hbar}^{(k)}\circ(\widetilde{\prequantumL}_{g}-\mathrm{Id})\circ\multiplication(\psi)\right\Vert _{\mathcal{\widetilde{H}}_{\hbar}^{r}(\real^{2d+d'})}\le C\hbar^{\epsilon}.
\]
\end{lem}
\begin{cor}
\label{cor:L_g_almost_commutes_with_t^k-1}There exist constants $C>0$
and $\epsilon>0$ independent of $\hbar$ such that, for any $\psi\in\widetilde{\scrX}_{\hbar}$
and $g\in\widetilde{\scrG}_{\hbar}$ , it holds 
\[
\left\Vert [\widetilde{\prequantumL}_{g}\circ\multiplication(\psi),\,\romet_{\hbar}^{(k)}]\right\Vert _{\mathcal{\widetilde{H}}_{\hbar}^{r}(\real^{2d+d'})}\le C\hbar^{\epsilon}\quad\mbox{for }0\le k\le n
\]
 and also 
\[
\left\Vert [\widetilde{\prequantumL}_{g}\circ\multiplication(\psi),\,\tilde{\romet}_{\hbar}]\right\Vert _{\widetilde{\mathcal{H}}_{\hbar}^{r}(\real^{2d+d'})}\le C\hbar^{\epsilon}.
\]
\end{cor}
\begin{lem}
\label{lm:L_g_almost_identity-Tr-1} There exist constants $\epsilon>0$
and $C>0$, independent of $\hbar$, such that the following holds
true: Let $\psi\in\widetilde{\scrX}_{\hbar}$ and let $g\in\scrG_{\hbar}$,
$0\le k\le n$, then it holds 
\[
\left\Vert (\widetilde{\prequantumL}_{g}-\mathrm{Id)}\circ\multiplication(\psi)\circ\romet_{\hbar}^{(k)}:\mathcal{\widetilde{H}}_{\hbar}^{r}(\real^{2d})\to\mathcal{\widetilde{H}}_{\hbar}^{r}(\real^{2d})\right\Vert _{Tr}\le C\hbar^{-\theta d+\epsilon}
\]
 
\[
\left\Vert \romet_{\hbar}^{(k)}\circ(\widetilde{\prequantumL}_{g}-\mathrm{Id})\circ\multiplication(\psi):\mathcal{\widetilde{H}}_{\hbar}^{r}(\real^{2d})\to\mathcal{\widetilde{H}}_{\hbar}^{r}(\real^{2d})\right\Vert _{Tr}\le C\hbar^{-\theta d+\epsilon}
\]
and 
\[
\left\Vert [\widetilde{\prequantumL}_{g}\circ\multiplication(\psi),\,\romet_{\hbar}^{(k)}]:\mathcal{\widetilde{H}}_{\hbar}^{r}(\real^{2d})\to\mathcal{\widetilde{H}}_{\hbar}^{r}(\real^{2d})\right\Vert _{Tr}\le C\hbar^{-\theta d+\epsilon}
\]

\end{lem}

\subsection{Proof of the main theorems in the setting of Grassmanian extension. }

Now we give the proofs of the main theorems, namely Theorem \ref{thm:Discrete-spectrum-F_G},
Theorem \ref{thm:band_structure-1} and Theorem \ref{cor:The-special-choice_V0}
in the extended situation. (Theorem \ref{thm:weyl_law_extended} will
be proved in the next subsection.) The point in the following argument
is in the choice of the local coordinate charts and a few basic estimates
on the extended map $f_{G}$ viewed in them. The modifications in
the rest part will be rather obvious.

Recall the set of points 
\[
\mathscr{P}_{\hbar}=\left\{ m_{i}\in M\mid\,1\le i\le I_{\hbar}\right\} ,
\]
local coordinate charts
\[
\kappa_{i}=\kappa_{i,\hbar}:\mathbb{D}\left(c\right)\subset\mathbb{R}^{2d}\rightarrow M,\quad1\le i\le I_{\hbar}
\]
and the sections $\tau_{i}:U_{i}\to P$ taken in Proposition \ref{prop:Local-charts}.
Also recall that we write $E_{u}$ for the image of the section $E_{u}:M\to G$,
which assigns the unstable subspace to each point and is H\"older
continuous with exponent $\beta$. This is an attracting subset for
$f_{G}$. (See (\ref{eq:E_u_K_n}).) Let $\widetilde{m}_{i}:=E_{u}(m_{i})\in K_{0}$
and set 
\[
\mathscr{\widetilde{P}}_{\hbar}=\left\{ \widetilde{m}_{i}\in G\mid\,1\le i\le I_{\hbar}\right\} .
\]
In the next proposition, we take local coordinate charts $\widetilde{\kappa}_{i}=\widetilde{\kappa}_{i,\hbar}$
on a small neighborhood of the point $\widetilde{m}_{i}$ as an extension
of $\kappa_{i}=\kappa_{i,\hbar}$. Notice that the local coordinate
charts $\widetilde{\kappa}_{i}=\widetilde{\kappa}_{i,\hbar}$ are
far from being conformal as the parameter $\hbar$ gets smaller. That
is, $D\widetilde{\kappa}_{i}|_{\real^{2d}\oplus\{0\}}$ is nearly
isometry while $D\widetilde{\kappa}_{i}|_{\{0\}\oplus\real^{d'}}$
is a conformal expansion by the rate $\hbar^{-(1-\beta)(1/2-\theta)-2\theta}$.
(See Condition (3) in the proposition below and remember Condition
(2) in Proposition \ref{prop:Local-charts}.) This is necessary when
we deal with the problems caused by non-smoothness of the attracting
set $E_{u}$. 
\begin{prop}
\textbf{\label{prop:Local-charts-1}} For each $\hbar=\frac{1}{2\pi N}>0$,
there exist a system of local coordinate charts 
\begin{equation}
\widetilde{\kappa}_{i}=\widetilde{\kappa}_{i,\hbar}:\mathbb{D}^{(2d)}\left(c\right)\times\mathbb{D}^{(d')}(c\cdot\hbar^{(1-\beta)(1/2-\theta)+2\theta})\subset\mathbb{R}^{2d+d'}\rightarrow G,\quad1\le i\le I_{\hbar}\label{eq:local_coordinate_extended}
\end{equation}
on a neighborhood of $E_{u}$ with $c>0$ a constant independent of
$\hbar$, so that the following conditions hold for sufficiently small
$\hbar$:\end{prop}
\begin{enumerate}
\item $\widetilde{\kappa}_{i}\left(0\right)=\widetilde{m}_{i}$ and the
following diagram commutes 
\[
\begin{CD}\mathbb{D}^{(2d)}\left(c\right)\times\mathbb{D}^{(d')}(c\cdot h^{(1-\beta)(1/2-\theta)+2\theta})@>{\widetilde{\kappa}_{i}}>>G\\
@V{p}VV@V{p}VV\\
\mathbb{D}^{(2d)}\left(c\right)@>{\kappa_{i}}>>M
\end{CD}
\]

\item The derivative $(D\widetilde{\kappa}_{i})_{0}:T_{0}\real^{2d+d'}=\real^{2d+d'}\to T_{\tilde{m}_{i}}G$
maps the subspaces $\{0\}\oplus\real^{d'}$ and $\real^{2d}\oplus\{0\}$
respectively to the subspace $\ker\,(Dp)_{\tilde{m}_{i}}\subset T_{\tilde{m}_{i}}G$
and its orthogonal complement respectively. The restriction of $(D\widetilde{\kappa}_{i})_{0}$
to the subspace $\{0\}\oplus\real^{d'}$ is a conformal expansion
by the rate $ $$\hbar^{-(1-\beta)(1/2-\theta)-2\theta}$ with respect
to the Euclidean metric in the source and the Riemann metric on $G$
in the target. Further, the map $\widetilde{\kappa}_{i}\circ(D\widetilde{\kappa}_{i})_{0}^{-1}$
is not far from the exponential map in the sense that 
\[
\|\exp_{\tilde{m}_{i}}^{-1}\circ\widetilde{\kappa}_{i}\circ(D\widetilde{\kappa}_{i})_{0}^{-1}:\mathbb{D}(\tilde{m}_{i},c)\to T_{\tilde{m}_{i}}G\|_{C^{k}}\le C_{k}
\]
with $C_{k}$ a constant independent of $\hbar$ nor $1\le i\le I_{\hbar}$,
where $\mathbb{D}(\tilde{m}_{i},c)$ denotes the disk in $T_{\tilde{m}_{i}}G$
with radius $c$ and center at the origin. 
\item The union of the images $\widetilde{U}_{i}:=\widetilde{\kappa}_{i}(\mathbb{D}^{(2d+d')}(\hbar^{1/2-\theta}))$,
$1\le i\le I_{\hbar}$, cover the $\hbar^{\beta(1/2-\theta)-2\theta}/2$-neighborhood
of the section $E_{u}$ and contained in the $2\hbar^{\beta(1/2-\theta)-2\theta}$-neighborhood
of the section $E_{u}$. 
\item For the coordinate change transformation $\widetilde{\kappa}_{j,i}:=\widetilde{\kappa}_{j}^{-1}\circ\widetilde{\kappa}_{i}$,
defined for $1\le i,j\le I_{\hbar}$ with $\widetilde{U}_{i}\cap\widetilde{U}_{j}\neq\emptyset$,
there exists an isometric affine map $\widetilde{A}_{j,i}:\mathbb{R}^{2d+d'}\rightarrow\mathbb{R}^{2d+d'}$
of the form $\widetilde{A}_{j,i}=A_{j,i}\oplus\widehat{A}_{j,i}$
with $A_{j,i}$ given in Proposition \ref{prop:Local-charts} and
$\widehat{A}_{j,i}:\real^{d'}\to\real^{d'}$ an isometric linear map,
such that, if we set $\mathcal{G}_{\hbar}=\{\widetilde{A}_{j,i}\circ\widetilde{\kappa}_{j,i}\}$,
it satisfies the conditions (G0)-(G3) in the Setting II$^{\mathrm{ext}}$
in Subsection \ref{sub:nonlinear_extended}. 
\item There exists a family of $C^{\infty}$ functions $\{\widetilde{\psi}_{i}:\real^{2d+d'}\to[0,1]\}_{i=1}^{I_{\hbar}}$
which is supported on the disk $\mathbb{D}^{(2d+d')}(\hbar^{1/2-\theta})$,
such that 
\begin{equation}
\sum_{i=1}^{I_{\hbar}}\widetilde{\psi}_{i}\circ\widetilde{\kappa}_{i}^{-1}\equiv1\mbox{ on the \ensuremath{(\hbar^{\beta(1/2-\theta)-2\theta}/4)}-neighborhood of the section \ensuremath{E_{u}}}.\label{eq:partition_of_unity-1}
\end{equation}
The set of functions $\mathcal{X}_{\hbar}=\{\widetilde{\psi}_{i}\}$
satisfies the conditions, (C1) and (C2), in the Setting I$^{\mathrm{ext}}$
in Subsection \ref{sub:nonlinear_extended}. \end{enumerate}
\begin{lyxcode}
~\end{lyxcode}
\begin{proof}
We continue the argument in the proof of Proposition \ref{prop:Local-charts},
in which we defined the local chart $\kappa_{m}:\mathbb{D}(c)\to M$
for each point $m\in M$ as a composition of a linear map from $\real^{2d}$
to $T_{m}M$ and a modification of the exponential mapping. By a parallel
construction, we can and do take a local chart $\hat{\kappa}_{m}:\mathbb{D}(c)\to G$
for each point $m\in M$ so that the following conditions hold:
\begin{itemize}
\item $\hat{\kappa}_{m}(0)=\tilde{m}:=E_{u}(m)\in K_{0}$ and we have $p\circ\hat{\kappa}_{m}=\kappa_{m}\circ p$,
\item The derivative $(D\hat{\kappa}_{m})_{0}:T_{0}\real^{2d+d'}=\real^{2d+d'}\to T_{\tilde{m}}G$
maps the subspaces $\{0\}\oplus\real^{d'}$ and $\real^{2d}\oplus\{0\}$
to the subspace $\ker\,(Dp)_{\tilde{m}}\subset T_{\tilde{m}}G$ and
its orthogonal complement respectively. The restriction of $(D\hat{\kappa}_{\tilde{m}})_{0}$
to the subspace $\{0\}\oplus\real^{d'}$ is an isometric linear map
with respect to the Euclidean metric in the source and the Riemann
metric on $G$ in the target. Further, the map $\hat{\kappa}_{\tilde{m}}\circ(D\hat{\kappa}_{\tilde{m}})_{0}^{-1}$
is not far from the exponential map in the sense that 
\[
\|\exp_{\tilde{m}}^{-1}\circ\hat{\kappa}_{\tilde{m}}\circ(D\hat{\kappa}_{\tilde{m}})_{0}^{-1}:\mathbb{D}(\tilde{m},c)\to T_{\tilde{m}}G\|_{C^{k}}\le C_{k}
\]
with $C_{k}$ a constant independent of $\hbar$ nor $1\le i\le I_{\hbar}$.
\end{itemize}
We then define the local charts $\widetilde{\kappa}_{i}$ in the statement
of the proposition by 
\[
\widetilde{\kappa}_{i}(x,s):=\widetilde{\kappa}_{\tilde{m}{}_{i}}(x,h^{-(1-\beta)(1/2-\theta)-2\theta}s).
\]
It is not difficult to check that the coordinate maps $\widetilde{\kappa}_{i}$
thus defined satisfies the required conditions. The conditions (1),
(2) and (3) should be obvious. Also so should be (5) once we check
the conditions (1)-(4). Thus we check the condition (4). The condition
(G0) and (G1) are direct consequences of the construction. To check
the condition (G2), we first observe that the diffeomorphism $\hat{\kappa}_{j}^{-1}\circ\hat{\kappa}_{i}$
is written in the form
\begin{equation}
\hat{\kappa}_{j}^{-1}\circ\hat{\kappa}_{i}\begin{pmatrix}q\\
p\\
s
\end{pmatrix}=\begin{pmatrix}q_{0}\\
p_{0}\\
s_{0}
\end{pmatrix}+\begin{pmatrix}A_{1,1} & A_{1,2} & 0\\
A_{2,1} & A_{2,2} & 0\\
A_{3,1} & A_{3,2} & A_{3,3}
\end{pmatrix}\begin{pmatrix}q\\
p\\
s
\end{pmatrix}+K(q,p,s)\label{eq:kappa_in_coordinates}
\end{equation}
where $K(q,p,s)$ satisfies $K(0,0,0)=0$ and $DK(0,0,0)=0$. From
the assumption $U_{i}\cap U_{j}\neq\emptyset$ and the choice of $\beta$,
we have
\begin{itemize}
\item $|q_{0}|<C\hbar^{1/2-\theta}$, $|p_{0}|<C\hbar^{1/2-\theta}$, $|s_{0}|<C\hbar^{\beta(1/2-\theta)}$,
\item $\|A_{1,2}\|<C\hbar^{\beta(1/2-\theta)}$ , $\|A_{2,1}\|<C\hbar^{\beta(1/2-\theta)}$,
and
\item $\|A_{1,1}-I\|<C\hbar^{\beta(1/2-\theta)}$, $\|A_{2,2}-I\|<C\hbar^{\beta(1/2-\theta)}$,
$\|A_{3,3}-I'\|<C\hbar^{(1/2-\theta)}$ for some orthogonal transformations
$I,I':\real^{d}\to\real^{d}$.
\end{itemize}
From the definition, the diffeomorphism $\widetilde{\kappa}_{j,i}:=\widetilde{\kappa}_{j}^{-1}\circ\widetilde{\kappa}_{i}$
is then written
\begin{align}
\widetilde{\kappa}_{j,i}\begin{pmatrix}q\\
p\\
s
\end{pmatrix} & =\begin{pmatrix}q_{0}\\
p_{0}\\
h^{(1-\beta)(1/2-\theta)+2\theta}s_{0}
\end{pmatrix}\nonumber \\
 & \qquad+\begin{pmatrix}A_{1,1} & A_{1,2} & 0\\
A_{2,1} & A_{2,2} & 0\\
h^{(1-\beta)(1/2-\theta)+2\theta}A_{3,1} & h^{(1-\beta)(1/2-\theta)+2\theta}A_{3,2} & A_{3,3}
\end{pmatrix}\begin{pmatrix}q\\
p\\
s
\end{pmatrix}+K'(q,p,s)\label{eq:tildekappa_in_coordinates}
\end{align}
where $K'(q,p,s)$ satisfies $K'(0,0,0)=0$ and $DK'(0,0,0)=0$. From
this expression and the estimates above, we find the affine map $\widetilde{A}_{j,i}$
as claimed in (4) so that the condition (G2) holds true. To prove
the condition (G3), we express the term $K(q,p,s)$ in (\ref{eq:kappa_in_coordinates})
as
\[
K\begin{pmatrix}x\\
s
\end{pmatrix}=\begin{pmatrix}K_{1}(x)\\
K_{2}(x,s)
\end{pmatrix}
\]
and the differentials of $K_{1}(x)$ and $K_{2}(x,s)$ are uniformly
bounded with respect to $1\le i,j\le I_{\hbar}$ and $\hbar$. Then
$K'(q,p,s)$ in (\ref{eq:tildekappa_in_coordinates}) is written
\[
K'\begin{pmatrix}x\\
s
\end{pmatrix}=\begin{pmatrix}K_{1}(x)\\
h^{(1-\beta)(1/2-\theta)+2\theta}\cdot K_{2}(x,h^{-(1-\beta)(1/2-\theta)-2\theta}s)
\end{pmatrix}.
\]
The condition (G3) follows immediately from this expression. 
\end{proof}
We define a $C^{\infty}$ section $\widetilde{\tau}_{i}:\widetilde{U}_{i}\to P_{G}$
of the $\mathbf{U}(1)$-bundle $P_{G}$ as the pull-back of $\tau_{i}$
by the projection $p$: 
\[
\widetilde{\tau}_{i}(z)=\tau_{i}\circ p(z).
\]
Once we have set up the local coordinates and local sections as above,
we can follow the argument in Section \ref{sec:Proofs-of-the_mains}
by confirming the correspondence between the objects and statements.
The following few paragraphs we follow the argument in Subsection
\ref{sub:The-prequantum-transfer_decomposed} with obvious modifications.
\begin{defn}
Set 
\[
\mathcal{\widetilde{E}}_{\hbar}:=\bigoplus_{i=1}^{I_{\hbar}}C_{0}^{\infty}(\mathbb{D}^{2d+d'}(\hbar^{1/2-\theta})).
\]
and let $\widetilde{\boldI}_{\hbar}:C_{N}^{\infty}\left(P_{G}\right)\to\mathcal{\widetilde{E}}_{\hbar}$
be the operator that assign to each equivariant function $u\in C_{N}^{\infty}\left(P_{G}\right)$
a set of functions $\widetilde{\boldI}_{\hbar}(u)=\left(u_{i}\right)_{i=1}^{I_{\hbar}}\in\mathcal{\widetilde{E}}_{\hbar}$
on local charts defined by the relation
\begin{equation}
u_{i}\left(x,s\right)=\widetilde{\psi}_{i}\left(x,s\right)\cdot u\left(\widetilde{\tau}_{i}\left(\widetilde{\kappa}_{i}\left(x,s\right)\right)\right)\quad\mbox{for }1\le i\le I_{\hbar}.\label{eq:local_data_ui-1}
\end{equation}
Let $\widetilde{\boldI}_{\hbar}^{*}:\bigoplus_{i=1}^{I_{\hbar}}\mathcal{S}(\real^{2d+d'})\to C_{N}^{\infty}\left(P_{G}\right)$
be the operator defined by
\begin{equation}
\left(\widetilde{\boldI}_{\hbar}^{*}\left((u_{i})_{i=1}^{I_{\hbar}}\right)\right)(p)=\sum_{i=1}^{I_{\hbar}}e^{i2\pi N\cdot\alpha_{i}(p)}\cdot\widetilde{\chi}_{\hbar}(x,s)\cdot u_{i}(x,s)\label{eq:expression_of_I*-1}
\end{equation}
where $\widetilde{\chi}_{\hbar}$ is the function defined in (\ref{eq:tilde_chi_hbar}),
$(x,s)=\kappa_{i}^{-1}(\pi(p))$ and $\alpha_{i}(p)$ is the real
number such that $p=e^{i2\pi\alpha_{i}(p)}\cdot\tilde{\tau}_{i}(\pi(p))$.
Then we have 
\begin{equation}
\widetilde{\boldI}_{\hbar}^{*}\circ\widetilde{\boldI}_{\hbar}u=u\quad\mbox{for \ensuremath{u\in C_{N}^{\infty}\left(P_{G}\right)}}\label{eq:I*_I_is_Identity-1}
\end{equation}
and $\widetilde{\boldI}_{\hbar}\circ\widetilde{\boldI}_{\hbar}^{*}:\mathcal{E}_{\hbar}\rightarrow\mathcal{E}_{\hbar}$
is a projection onto the image of $\widetilde{\boldI}_{\hbar}$. We
define the \emph{lift of the prequantum transfer operator} $\hat{F}_{\hbar}$
with respect to $\widetilde{\boldI}_{\hbar}$ by 
\begin{equation}
\widetilde{\boldF}_{\hbar}:=\widetilde{\boldI}_{\hbar}\circ\hat{F}_{N}\circ\widetilde{\boldI}_{\hbar}^{*}\;:\;\bigoplus_{i=1}^{I_{\hbar}}\mathcal{S}(\real^{2d+d'})\rightarrow\widetilde{\mathcal{E}}_{\hbar}\subset\bigoplus_{i=1}^{I_{\hbar}}\mathcal{S}(\real^{2d+d'}).\label{eq:Lifted_operator_F-1}
\end{equation}
\end{defn}
\begin{prop}
\label{prop:bfF}The operator $\widetilde{\boldF}_{\hbar}$ can be
written as 
\[
\widetilde{\boldF}_{\hbar}((v_{i})_{i\in I_{\hbar}})=\left(\sum_{i=1}^{I_{\hbar}}\widetilde{\boldF}_{j,i}(v_{i})\right)_{j\in I_{\hbar}}
\]
where the component 
\[
\mathbf{\widetilde{F}}_{j,i}:\mathcal{S}(\real^{2d+d'})\rightarrow C_{0}^{\infty}(\mathbb{D}^{(2d+d')}(\hbar^{1/2-\theta}))
\]
is defined by $\mathbf{\widetilde{F}}_{j,i}\equiv0$ if $i\not\to j$
and, otherwise, by 
\[
\mathbf{\widetilde{F}}_{j,i}(v_{i})=\mathcal{\widetilde{L}}_{\tilde{f}_{j,i}}\left(e^{\widetilde{V}\circ f_{G}\circ\widetilde{\kappa}_{i}}\cdot\widetilde{\psi}_{j,i}\cdot\widetilde{\chi}_{\hbar}\cdot v_{i}\right)
\]
where we set
\begin{align}
\tilde{f}_{j,i} & :=\widetilde{\kappa}_{j}^{-1}\circ f_{G}\circ\widetilde{\kappa}_{i},\label{eq:f_ji-1}\\
\widetilde{\psi}_{j,i} & :=\widetilde{\psi}_{j}\circ\tilde{f}_{j,i}\label{eq:psi_ji-1}
\end{align}
and $\mathcal{\widetilde{L}}_{\tilde{f}_{j,i}}$ is the Euclidean
prequantum transfer operator defined in (\ref{eq:Lg-1}) for $g=\tilde{f}_{j,i}$. 
\end{prop}
We define 
\[
\widetilde{V}_{j}=\max\{\widetilde{V}(m)\mid m\in\widetilde{U}_{j}\}\quad\mbox{for }1\le j\le I_{\hbar}.
\]
The next lemma corresponds to Lemma \ref{lem:multipliacation_by_Psi_ji}. 
\begin{lem}
\label{lem:multipliacation_by_Psi_ji-1}If we set 
\[
\mathscr{\widetilde{X}_{\hbar}=}\{\widetilde{\psi}_{j,i}\cdot\widetilde{\chi}_{\hbar}\mid1\le i,j\le I_{\hbar},\; i\to j\},
\]
it satisfies the conditions (C1) and (C2) in Setting$I^{\mathrm{ext}}$
in Subsection \ref{sub:nonlinear_extended}. (The constants $C$ and
$C_{\alpha}$ will depend on $f_{G}$ and $\widetilde{V}$ though
not on $\hbar$.) For $1\le i,j\le I_{\hbar}$ such that $i\to j$,
we have 
\[
\left\Vert \multiplication(e^{\widetilde{V}\circ f_{G}\circ\widetilde{\kappa}_{i}}\cdot\widetilde{\psi}_{j,i}\cdot\widetilde{\chi}_{\hbar})-e^{\widetilde{V}_{j}}\cdot\multiplication(\widetilde{\psi}_{j,i}\cdot\widetilde{\chi}_{\hbar})\right\Vert _{\mathcal{\widetilde{H}}_{\hbar}^{r}(\real^{2d+d'})}\le C(f_{G},\widetilde{V})\cdot\hbar^{\theta}
\]
for some constant $C(f_{G},\widetilde{V})$ independent of $\hbar$.
The same statement holds for the set of functions
\[
\mathscr{\widetilde{X}_{\hbar}=}\{e^{\widetilde{V}\circ f_{G}\circ\widetilde{\kappa}_{i}}\cdot\widetilde{\psi}_{j,i}\cdot\widetilde{\chi}_{\hbar}\mid1\le i,j\le I_{\hbar},\; i\to j\}.
\]
\end{lem}
\begin{proof}
We obtain the proof by following the argument in that of Lemma \ref{lem:multipliacation_by_Psi_ji}
with obvious correspondence, using Lemma \ref{lm:XM_exchange-1} instead
of Lemma \ref{lm:XM_exchange}.
\end{proof}
We next proceed to the argument corresponding to that in Subsection
\ref{sub:anisotropic-Sobolev-space}. The anisotropic Sobolev space
in the extended setting is defined as follows. 
\begin{defn}
\label{dfn:global_anisotropic_Sobolev_space-1} Let $C_{N}^{\infty}\left(P_{G},E_{u}\right)$
be the set of functions $u\in C_{N}^{\infty}\left(P_{G}\right)$ that
is supported on the inverse image (with respect to the projection
$\pi_{G}:P_{G}\to G$) of the $(\hbar^{\beta(1/2-\theta)}/4)$-neighborhood
of the section $E_{u}$. The \emph{anisotropic Sobolev space} $\mathcal{\widetilde{H}}_{\hbar}^{r}\left(P_{G},E_{u}\right)$
is defined as the completion of the space $C_{N}^{\infty}\left(P_{G},E_{u}\right)$
with respect to the norm
\[
\|u\|_{\mathcal{\widetilde{H}}_{\hbar}^{r}}:=\left(\sum_{i=1}^{I_{\hbar}}\|u_{i}\|_{\mathcal{\widetilde{H}}_{\hbar}^{r}(\real^{2d+d'})}^{2}\right)^{1/2}\quad\mbox{ for }u\in C_{N}^{\infty}\left(P_{G},E_{u}\right),
\]
where $u_{i}=\left(\mathbf{I}_{\hbar}\left(u\right)\right)_{i}\in C_{0}^{\infty}\left(\mathbb{D}^{(2d+d')}(\hbar^{1/2-\theta})\right)$
are the local data defined in (\ref{eq:local_data_ui-1}) and $\|u_{i}\|_{\mathcal{H}_{\hbar}^{r}(\real^{2d})}^{2}$
is the anisotropic Sobolev norm on $C_{0}^{\infty}\left(\mathbb{R}^{2d+d'}\right)$
in Definition \ref{def:escape_function_Hr-1}. We define the Hilbert
spaces $\mathcal{\widetilde{H}}_{\hbar}^{r,\pm}\left(P_{G},E_{u}\right)$
in the parallel manner, replacing $\|u_{i}\|_{\mathcal{\widetilde{H}}_{\hbar}^{r}(\real^{2d+d'})}^{2}$
by the norms $\|u_{i}\|_{\mathcal{\widetilde{H}}_{\hbar}^{r,\pm}(\real^{2d+d'})}^{2}$
respectively. 
\end{defn}
The next lemma corresponds to Lemma \ref{lm:IIbdd}. 
\begin{lem}
\label{lm:IIbdd-2} The projector $\widetilde{\boldI}_{\hbar}\circ\widetilde{\boldI}_{\hbar}^{*}:\mathcal{\widetilde{E}}_{\hbar}\rightarrow\mathcal{\widetilde{E}}_{\hbar}$
extends to bounded operators 
\[
\widetilde{\boldI}_{\hbar}\circ\widetilde{\boldI}_{\hbar}^{*}:\bigoplus_{i=1}^{I_{\hbar}}\mathcal{\widetilde{H}}_{\hbar}^{r,+}(\real^{2d+d'})\to\bigoplus_{i=1}^{I_{\hbar}}\mathcal{\widetilde{H}}_{\hbar}^{r}(\mathbb{D}^{(2d+d')}(\hbar^{1/2-\theta}))
\]
and 
\[
\widetilde{\boldI}_{\hbar}\circ\widetilde{\boldI}_{\hbar}^{*}:\bigoplus_{i=1}^{I_{\hbar}}\mathcal{\widetilde{H}}_{\hbar}^{r}(\real^{2d+d'})\to\bigoplus_{i=1}^{I_{\hbar}}\mathcal{\widetilde{H}}_{\hbar}^{r,-}(\mathbb{D}^{(2d+d')}(\hbar^{1/2-\theta})).
\]
Further the operator norms of these projection operators are bounded
by a constant independent of~$\hbar$. \end{lem}
\begin{proof}
The proof is obtained by following that of Lemma \ref{lm:IIbdd},
setting 
\begin{equation}
\widetilde{\scrG}_{\hbar}=\{\widetilde{A}_{j,i}\circ\widetilde{\kappa}_{j,i}\mid1\le i,j\le I_{\hbar},\; U_{i}\cap U_{j}\neq\emptyset\}\label{eq:scrXG_setting-1}
\end{equation}
and 
\[
\widetilde{\scrX}_{\hbar}=\{\widetilde{\psi}_{j}\circ\widetilde{\kappa}_{j,i}\cdot\widetilde{\chi}_{\hbar}\mid1\le i,j\le I_{\hbar},\; U_{i}\cap U_{j}\neq\emptyset\},
\]
and by using Proposition \ref{pp:bdd_g-1} and Corollary \ref{lm:XM_exchange-1}
instead of Proposition \ref{pp:bdd_g} and Corollary \ref{lm:XM_exchange}.
\end{proof}
The next lemma corresponds to Lemma \ref{lm:bounded_F}. We suppose
that the hyperbolicity exponent $\lambda>1$ of the flow $f_{G}$
is sufficiently large. (Say $\lambda>9$.) 
\begin{lem}
\label{lm:bounded_F-2} The operator $\widetilde{\boldF}_{\hbar}$
defined in (\ref{eq:Lifted_operator_F-1}) extends uniquely to the
bounded operator 
\begin{equation}
\widetilde{\boldF}_{\hbar}:\bigoplus_{i=1}^{I_{\hbar}}\mathcal{\widetilde{H}}_{\hbar}^{r}(\real^{2d+d'})\to\bigoplus_{i=1}^{I_{\hbar}}\mathcal{\widetilde{H}}_{\hbar}^{r}(\mathbb{D}^{(2d+d')}(\hbar^{1/2-\theta}))\label{eq:Lifted_operator_F2-1}
\end{equation}
and the operator norm is bounded by a constant independent of $\hbar$.
Consequently the same result holds for the prequantum transfer operator
$\widetilde{F}_{\hbar}:\mathcal{\widetilde{H}}_{\hbar}^{r}\left(P_{G},E_{u}\right)\rightarrow\mathcal{\widetilde{H}}_{\hbar}^{r}\left(P_{G},E_{u}\right)$. \end{lem}
\begin{proof}
We follow the argument in the proof of Lemma \ref{lm:bounded_F} with
slight modification. We express the diffeomorphism $\tilde{f}_{j,i}$
in (\ref{eq:f_ji-1}) as a composition 
\begin{equation}
\tilde{f}_{j,i}=\tilde{a}_{j,i}\circ\tilde{g}_{j,i}\circ\widetilde{B}_{j,i}\label{eq:fij_decomposition-1}
\end{equation}
where 
\begin{itemize}
\item $\tilde{a}_{j,i}:\real^{2d+d'}\to\real^{2d+d'}$ is a translation
in a direction in $\real^{2d}\oplus\{0\}$, 
\item $\widetilde{B}_{j,i}:\real^{2d}\to\real^{2d}$ is a linear map of
the form (\ref{eq:hyperbolic_linear_map_extended}) considered in
Subsection \ref{sub:Discussion-linear_model}.
\item $\tilde{g}_{j,i}$ is a diffeomorphism such that $\widetilde{\scrG}_{\hbar}=\{\tilde{g}_{j,i}\}_{1\le i,j\le I_{\hbar}}$
satisfies the condition (G0),(G1),(G2) and (G3) in Setting II$^{\mathrm{ext}}$
in Subsection \ref{sub:nonlinear_extended}. 
\end{itemize}
This corresponds to the expression (\ref{eq:fij_decomposition}) in
the proof of Lemma \ref{lm:bounded_F}. 
\begin{rem}
(1) The prequantum transfer operator $\widetilde{L}_{\tilde{a}_{j,i}}$
associated to a translation on $\real^{2d+d'}$ acts on the anisotropic
Sobolev space $\widetilde{\mathcal{H}}_{\hbar}^{r}(\real^{2d+d'})$
as an isometry if the direction of translation belongs to the subspace
$\real^{2d}\oplus\{0\}$. But this is not true for translations in
the other directions. 

(2) The proof of the claim on $\tilde{g}_{j,i}$ need some argument.
But this is essentially parallel to that in the proof of Proposition
\ref{prop:Local-charts-1}. 
\end{rem}
From the expression (\ref{eq:fij_decomposition-1}), the operator
$\boldF_{j,i}$ is expressed as the composition 
\begin{equation}
\boldF_{j,i}=\widetilde{\prequantumL}^{(0)}\circ\widetilde{\prequantumL}^{(1)}\circ\widetilde{\prequantumL}^{(2)}\label{eq:expression_F_ij-1}
\end{equation}
where $\mathcal{\widetilde{L}}^{\left(0\right)}:=\prequantumL_{\tilde{a}_{j,i}}$
and $\widetilde{\mathcal{L}}^{\left(2\right)}:=\widetilde{\prequantumL}_{\widetilde{B}_{j,i}}$
are the Euclidean prequantum transfer operators (\ref{eq:Lg-1}) associated
to the diffeomorphisms $g=a_{ij}$ and $g=B_{ij}$ respectively, while
$\widetilde{\prequantumL}^{(1)}$ is the operator of the form 
\[
\widetilde{\prequantumL}^{(1)}u=\widetilde{\prequantumL}_{\tilde{g}_{j,i}}\bigg(\big((e^{\widetilde{V}\circ\tilde{f}\circ\widetilde{\kappa}_{i}}\cdot\widetilde{\psi}_{j,i}\cdot\widetilde{\chi}_{\hbar})\circ\widetilde{B}_{j,i}^{-1}\big)\cdot u\bigg)
\]
with $\widetilde{\psi}_{j,i}$ the function defined in (\ref{eq:psi_ji-1}).
Then we follow the argument in the latter part of the proof of Lemma
\ref{lm:bounded_F}, replacing some proposition by those prepared
in the last subsection.
\end{proof}
Next we introduce the projection operators 
\begin{equation}
{\boldt}_{\hbar}^{(k)}:\bigoplus_{i=1}^{I_{\hbar}}\mathcal{\widetilde{H}}_{\hbar}^{r}(\real^{2d+d'})\to\bigoplus_{i=1}^{I_{\hbar}}\mathcal{\widetilde{H}}_{\hbar}^{r}(\real^{2d+d'}),\qquad{\boldt}_{\hbar}^{(k)}((u_{i})_{i=1}^{I_{\hbar}})=(\romet_{\hbar}^{(k)}(u_{i}))_{i=1}^{I_{\hbar}},\label{Def:boldt_extended}
\end{equation}
 for $0\le k\le n$ and 
\[
\tilde{\boldt}_{\hbar}:\bigoplus_{i=1}^{I_{\hbar}}\mathcal{\widetilde{H}}_{\hbar}^{r}(\real^{2d+d'})\to\bigoplus_{i=1}^{I_{\hbar}}\mathcal{\widetilde{H}}_{\hbar}^{r}(\real^{2d+d'}),\qquad\tilde{\boldt}_{\hbar}((u_{i})_{i=1}^{I_{\hbar}})=(\tilde{\romet}_{\hbar}(u_{i}))_{i=1}^{I_{\hbar}},
\]
where $\romet_{\hbar}^{(k)}$ and $\tilde{\romet}_{\hbar}$ are the
projection operators introduced in (\ref{eq:def_t-1}) and (\ref{eq:def_tt-1}).
As before, we set $\mathbf{t}_{\hbar}^{(n+1)}=\tilde{\boldt}_{\hbar}$,
so that the set of operators $\{\boldt_{\hbar}^{(k)}\}_{k=0}^{n+1}$
are complete sets of mutually commuting projection operators. (Notice
that we are using the same notation $\mathbf{t}_{\hbar}^{(k)}$ and
$\tilde{\boldt}_{\hbar}$ in this extended setting as (\ref{eq:def_tk})
and (\ref{eq:def_tilde_t}) in Subsection \ref{sub:The-main-propositions}.) 

The following two propositions corresponds to Proposition \ref{pp1}
and \ref{prop:key_proposition2}. For the proofs, we have only to
follow those of Proposition \ref{pp1} and \ref{prop:key_proposition2}
respectively, checking the correspondence in the notation and replacing
the propositions by the corresponding ones. 
\begin{prop}
\label{pp1-1} There are constants $\epsilon>0$ and $C_{0}>0$, independent
of $\hbar$, such that the following holds: We have that 
\[
\left\Vert {\boldt}_{\hbar}^{(k)}\circ(\widetilde{\boldI}_{\hbar}\circ\widetilde{\boldI}_{\hbar}^{*})\right\Vert _{\bigoplus_{i=1}^{I_{\hbar}}\mathcal{\widetilde{H}}_{\hbar}^{r}(\real^{2d+d'})\rightarrow\bigoplus_{i=1}^{I_{\hbar}}\mathcal{\widetilde{H}}_{\hbar}^{r}(\real^{2d+d'})}<C_{0},\quad\mbox{and }
\]
\[
\left\Vert (\widetilde{\boldI}_{\hbar}\circ\widetilde{\boldI}_{\hbar}^{*})\circ{\boldt}_{\hbar}^{(k)}\right\Vert _{\bigoplus_{i=1}^{I_{\hbar}}\mathcal{\widetilde{H}}_{\hbar}^{r}(\real^{2d+d'})\rightarrow\bigoplus_{i=1}^{I_{\hbar}}\mathcal{\widetilde{H}}_{\hbar}^{r}(\real^{2d+d'})}<C_{0}
\]
for $0\le k\le n$. Also we have, for the norm of the commutators,
that 
\begin{equation}
\left\Vert \left[{\boldt}_{\hbar}^{(k)},(\widetilde{\boldI}_{\hbar}\circ\widetilde{\boldI}_{\hbar}^{*})\right]\right\Vert _{\bigoplus_{i=1}^{I_{\hbar}}\mathcal{\widetilde{H}}_{\hbar}^{r}(\real^{2d+d'})\to\bigoplus_{i=1}^{I_{\hbar}}\mathcal{\widetilde{H}}_{\hbar}^{r}(\real^{2d+d'})}\le C_{0}\hbar^{\epsilon}\label{eqn:com-1}
\end{equation}
for $0\le k\le n$. 
\end{prop}
~
\begin{prop}
\label{prop:key_proposition2-1} There are constants $\epsilon>0$
and $C>0$, independent of $\hbar$, such that 
\begin{equation}
\left\Vert \left[\widetilde{\boldF}_{\hbar},{\boldt}_{\hbar}^{(k)}\right]\right\Vert _{\bigoplus_{i=1}^{I_{\hbar}}\mathcal{\widetilde{H}}_{\hbar}^{r}(\real^{2d+d'})\to\bigoplus_{i=1}^{I_{\hbar}}\mathcal{\widetilde{H}}_{\hbar}^{r}(\real^{2d+d'})}\le C\hbar^{\epsilon}\quad\mbox{for }1\le k\le n+1.\label{eqn:com2-1}
\end{equation}
Further there exists a constant $C_{0}>0$, which is independent of
$f_{G}$, $\widetilde{V}$ and $\hbar$, such that 
\begin{enumerate}
\item For $0\le k\le n+1$, it holds
\[
\!\!\!\left\Vert {\boldt}_{\hbar}^{(k)}\circ\widetilde{\boldF}_{\hbar}\circ{\boldt}_{\hbar}^{(k)}\right\Vert _{\bigoplus_{i=1}^{I_{\hbar}}\mathcal{\widetilde{H}}_{\hbar}^{r}(\real^{2d+d'})\to\bigoplus_{i=1}^{I_{\hbar}}\mathcal{\widetilde{H}}_{\hbar}^{r}(\real^{2d+d'})}\le C_{0}\sup\left(|e^{\widetilde{V}}|\|Df|_{E_{u}}\|_{\min}^{-k}|\det Df|_{E_{u}}|^{-1/2}\right)
\]

\item If $\boldu\in\bigoplus_{i=1}^{I_{\hbar}}\mathcal{\widetilde{H}}_{\hbar}^{r}(\real^{2d+d'})$
satisfies $\widetilde{\boldI}_{\hbar}\circ\widetilde{\boldI}_{\hbar}^{*}(\boldu)=\boldu$
and 
\[
\|\mathbf{u}-(\widetilde{\boldI}_{\hbar}\circ\widetilde{\boldI}_{\hbar}^{*})\circ\boldt_{\hbar}^{(k)}(\mathbf{u})\|_{\mathcal{\widetilde{H}}_{\hbar}^{r}}<\|\mathbf{u}\|_{\mathcal{\widetilde{H}}_{\hbar}^{r}}/2\quad\mbox{for some \ensuremath{0\le k\le n},}
\]
 then we have 
\[
\|\boldt_{\hbar}^{(k)}\circ\widetilde{\boldF}_{\hbar}\circ\boldt_{\hbar}^{(k)}(\mathbf{u})\|_{\mathcal{\widetilde{H}}_{\hbar}^{r}}\ge C_{0}^{-1}\cdot\inf\left(|e^{\widetilde{V}}|\|Df|_{E_{u}}\|_{\max}^{-k}|\det Df|_{E_{u}}|^{-1/2}\right)\cdot\|\mathbf{u}\|_{\mathcal{\widetilde{H}}_{\hbar}^{r}}.
\]

\end{enumerate}
\end{prop}
Once we have obtained the propositions above, we can prove the next
theorem, just in the same manner as we deduced Theorem \ref{thm:More-detailled-description_of_spectrum}
in Subsection \ref{sub:Proofs-of-Theorem_more_details}.
\begin{thm}
\label{hm:More-detailled-description_of_spectrum_ext}Let $n\geq0$
and take sufficiently large $r$ accordingly. Then there exists a
small constant $\varepsilon>0$, a constant $C_{0}$, which is independent
of $V$, $f$ and $N$, and a decomposition of the Hilbert space $\mathcal{\widetilde{H}}_{N}^{r}\left(P_{G},E_{u}\right)$
independent of $V$:
\begin{equation}
\mathcal{\widetilde{H}}_{N}^{r}\left(P_{G},E_{u}\right)=\mathcal{H}'_{0}\oplus\mathcal{H}'_{1}\ldots\oplus\mathcal{H}'_{n}\oplus\mathcal{H}'_{n+1}\label{eq:decomp_space-1}
\end{equation}
such that, writing $\tau^{\left(k\right)}$ for the projection onto
the component $\mathcal{H}'_{k}$ along other components, 
\begin{enumerate}
\item For some constant $\epsilon>0$ and $C>0$ independent of $\hbar$,
we have
\[
\left|\dim\mathcal{H}'_{k}-N^{d}\cdot Vol_{\omega}(M)\right|\le CN^{d-\epsilon}\quad\mbox{for 0\ensuremath{\le}k\ensuremath{\le}n.}
\]
 while $\dim\mathcal{H}'_{n+1}=\infty$,
\item $\|\tau^{\left(k\right)}\|<C_{0}$ for $0\le k\le n+1$, 
\item $\|\tau^{\left(k\right)}\circ\widetilde{F}_{N}\circ\tau^{\left(l\right)}\|\le C/N^{\varepsilon}$
if $k\neq\ell$, with $C$ independent on $N$ (but may depend on
$f_{G}$).
\item for $0\le k\le n+1$, it holds
\begin{equation}
\|\tau^{\left(k\right)}\circ\widetilde{F}_{N}\circ\tau^{\left(k\right)}\|_{\mathcal{\widetilde{H}}_{N}^{r}\left(P_{G},E_{u}\right)}\le C_{0}\sup_{x\in M}\left(e^{D\left(x\right)}\left\Vert Df_{x}|_{E_{u}}\right\Vert _{\mathrm{min}}^{-k}\right),\label{eq:rhs_1-1}
\end{equation}
 
\item for $0\le k\le n$ and $u\in\mathcal{H}'_{k}$ it holds 
\begin{equation}
\left\Vert \left(\tau^{\left(k\right)}\circ\widetilde{F}_{N}\right)u\right\Vert _{\mathcal{\widetilde{H}}_{N}^{r}\left(P_{G},E_{u}\right)}\ge C_{0}^{-1}\inf_{x\in M}\left(e^{D\left(x\right)}\left\Vert Df_{x}|_{E_{u}}\right\Vert _{\mathrm{max}}^{-k}\right)\left\Vert u\right\Vert _{\mathcal{\widetilde{H}}_{N}^{r}\left(P_{G},E_{u}\right)},\label{eq:rhs_2-1}
\end{equation}

\end{enumerate}
provided that $N$ is sufficiently large.
\end{thm}
Now we can deduce Theorem \ref{thm:Discrete-spectrum-F_G} and Theorem
\ref{thm:band_structure-1} from the theorem above by the argument
parallel to that in Subsection \ref{sub:Proof-of-Theorem_Nad_annuli}.
(But see the remark below.) The former part of the statements in Theorem
\ref{cor:The-special-choice_V0} is an immediate sequence of Theorem
\ref{thm:Discrete-spectrum-F_G} and Theorem \ref{thm:band_structure-1}.
The proof of the angular equidistribution law in Theorem \ref{cor:The-special-choice_V0}
is parallel to that of Theorem \ref{Thm:Distribution-of-resonances.},
which we will present in Subsection \ref{sub:Proof-of-equidistribution}.
\begin{rem*}
The Hilbert space $\mathcal{\widetilde{H}}_{N}^{r}\left(P_{G},E_{u}\right)$
consists of distributions supported on a small neighborhood of the
attractor $E_{u}$ depending on $\hbar>0$. To get the Hilbert space
$\mathcal{H}_{\hbar}^{r}(P_{G})$ in the statement of Theorem \ref{thm:Discrete-spectrum-F_G}
and Theorem \ref{thm:band_structure-1}, we need a little formal argument
to construct $\mathcal{H}_{\hbar}^{r}(P_{G})$ from $\mathcal{\widetilde{H}}_{N}^{r}\left(P_{G},E_{u}\right)$
so that $\mathcal{H}_{\hbar}^{r}(P_{G})$ contains $C_{N}^{\infty}(P_{G})$
and that the operator $\widetilde{F}_{N}\circ\hat{\chi}$ on $\mathcal{H}_{\hbar}^{r}(P_{G})$
has the same spectral property as that on $\mathcal{\widetilde{H}}_{N}^{r}\left(P_{G},E_{u}\right)$.
But, recalling the argument in Subsection \ref{sub:Truncation-near_Eu}
on the absorbing neighborhoods of the attractor $E_{u}$ and noting
the precomposition of the operator $\hat{\chi}$ in $\widetilde{F}_{N}\circ\hat{\chi}$
, such an argument can be provided easily in various ways. 
\end{rem*}

\subsection{A remark about the relation between the transfer operators $\hat{F}_{N}$
and $\widetilde{F}_{N}$ (and a proof of Theorem \ref{thm:weyl_law_extended}).}

The prequantum transfer operator $\widetilde{F}_{N}$ on $G$ is an
extension of the prequantum transfer operator $\hat{F}_{N}$. But
the relation between the spectrum of $\widetilde{F}_{N}$ and that
of $\hat{F}_{N}$ is not clear in general. In this subsection, we
relate the spectra of the transfer operators $\widetilde{F}_{N}$
and $\hat{F}_{N}$ in the respective ``outermost'' bands. This will
reduce Theorem \ref{thm:weyl_law_extended} (and the corresponding
part of Theorem \ref{cor:The-special-choice_V0}) to Theorem \ref{thm:Weyl-formula}. 

Below we suppose that the potential function $V$ and $\tilde{V}$
in their definitions are related as 
\begin{equation}
\tilde{V}=V\circ p\label{eq:assumtion_on_tildeV}
\end{equation}
and that they are both smooth%
\footnote{Good part of the argument below holds without the assumption on smoothness
of the potential $V$.%
}. We also assume $r_{1}^{+}<r_{0}^{-}$ so that the outermost annuli
in Theorem \ref{thm:band_structure} and Theorem \ref{thm:Discrete-spectrum-F_G}
are separated from the inner annuli. Let us consider the pull-back
operator 
\[
p^{*}:C_{N}^{\infty}(P)\to C_{N}^{\infty}(P_{G}),\quad p^{*}u(z)=u(p(x))
\]
by the projection $p:G\to M$, and its dual 
\[
p_{*}:(C_{N}^{\infty}(P_{G}))'\to(C_{N}^{\infty}(P))',\quad\langle p_{*}v,u\rangle=\langle v,p_{*}u\rangle
\]
Let $\mathcal{\widetilde{H}}_{\hbar}^{r}(P_{G},E_{u})_{0}$ be the
closed subspace of the Hilbert space $\mathcal{\widetilde{H}}_{\hbar}^{r}(P_{G},E_{u})$
that consists of elements supported on (the image of) the section
$E_{u}$. 
\begin{prop}
\label{prop:properties_of_pi_*}The operator $p_{*}$ above restricts
to a bounded operator 
\begin{equation}
p_{*}:\mathcal{\widetilde{H}}_{\hbar}^{r}(P_{G},E_{u})_{0}\to\mathcal{H}_{\hbar}^{r}(P)\label{eq:boundedness_of_pi_*}
\end{equation}
and the following diagram commutes:\textup{
\[
\begin{CD}\mathcal{\widetilde{H}}_{\hbar}^{r}(P_{G},E_{u})_{0}@>\widetilde{F}_{N}>>\mathcal{\widetilde{H}}_{\hbar}^{r}(P_{G},E_{u})_{0}\\
@Vp_{*}VV@Vp_{*}VV\\
\mathcal{H}_{\hbar}^{r}(P)@>\hat{F}_{N}>>\mathcal{H}_{\hbar}^{r}(P)
\end{CD}
\]
Further we have that}\end{prop}
\begin{enumerate}
\item The generalized eigenvectors of $\widetilde{F}_{N}:\mathcal{\widetilde{H}}_{\hbar}^{r}(P_{G},E_{u})\to\mathcal{\widetilde{H}}_{\hbar}^{r}(P_{G},E_{u})$
for the eigenvalues in the outmost band $\{r_{0}^{-}-\epsilon<|z|<r_{0}^{+}+\epsilon\}$
is contained in $\mathcal{\widetilde{H}}_{\hbar}^{r}(P_{G},E_{u})_{0}$
and its image by $p_{*}$ does not vanish. 
\item The image of $p_{*}$ in (\ref{eq:boundedness_of_pi_*}) contains
$C_{N}^{\infty}(P)\subset\mathcal{H}_{\hbar}^{r}(P)$.
\end{enumerate}
Before proving this proposition, we give the following consequence. 
\begin{cor}
\label{cor:9.30}Under the same assumptions as in \ref{prop:properties_of_pi_*},
the spectrum of the operators $\hat{F}_{N}:\mathcal{H}_{\hbar}^{r}(P)\to\mathcal{H}_{\hbar}^{r}(P)$
and \textup{$\widetilde{F}_{N}:\mathcal{\widetilde{H}}_{\hbar}^{r}(P_{G},E_{u})\to\mathcal{\widetilde{H}}_{\hbar}^{r}(P_{G},E_{u})$
in the respective outermost annulus $\{r_{0}^{-}-\epsilon<|z|<r_{0}^{+}+\epsilon\}$
with sufficiently small $\epsilon>0$ coincides up to multiplicity
(provided $\hbar$ is sufficiently small according to $\epsilon$).}\end{cor}
\begin{proof}
Proposition \ref{prop:properties_of_pi_*} tells that, for a generalized
eigenvector $v$ of $\widetilde{F}_{N}:\mathcal{\widetilde{H}}_{\hbar}^{r}(P_{G},E_{u})\to\mathcal{\widetilde{H}}_{\hbar}^{r}(P_{G},E_{u})$
for an eigenvalue in the outermost band, its image $p_{*}(v)$ is
a generalized eigenvector of $\hat{F}_{N}:\mathcal{H}_{\hbar}^{r}(P)\to\mathcal{H}_{\hbar}^{r}(P)$
for the same eigenvalue. Thus the eigenvalues of $\widetilde{F}_{N}:\mathcal{\widetilde{H}}_{\hbar}^{r}(P_{G},E_{u})\to\mathcal{\widetilde{H}}_{\hbar}^{r}(P_{G},E_{u})$
in the outmost band is contained in those of $\hat{F}_{N}:\mathcal{H}_{\hbar}^{r}(P)\to\mathcal{H}_{\hbar}^{r}(P)$
up to multiplicity%
\footnote{That is, the multiplicity of the eigenvalues of the former is not
greater than the latter.%
}. 

We prove the converse. The image $\mathrm{Im}p_{*}$ of $p_{*}$ in
(\ref{eq:boundedness_of_pi_*}) can be identified with the quotient
space $ $ $\mathcal{\widetilde{H}}_{\hbar}^{r}(P_{G},E_{u})/\ker p_{*}$.
By this mean, we regard $\mathrm{Im}p_{*}$ as a Hilbert space, which
contains $C_{N}^{\infty}(P)\subset\mathcal{H}_{\hbar}^{r}(P)$ from
Proposition \ref{prop:properties_of_pi_*}(2). Thus we have two Hilbert
spaces $\mathcal{H}_{\hbar}^{r}(P)$ and $\mathrm{Im}p_{*}$ which
contains $C_{N}^{\infty}(P)$ in common as dense subsets and the operator
$\hat{F}_{N}$ act on both of the Hilbert space boundedly as natural
extensions of its action on $C_{N}^{\infty}(P)$. Then, by a general
argument (see Appendix of \cite{Baladi-Tsujii08} for instance), the
discrete spectra of the operators $\hat{F}_{N}:\mathcal{H}_{\hbar}^{r}(P)\to\mathcal{H}_{\hbar}^{r}(P)$
and $\hat{F}_{N}:\mathrm{Im}p_{*}\to\mathrm{Im}p_{*}$ coincide up
to multiplicity on the outside of their essential spectral radius.
This implies that the eigenvalues of $\hat{F}_{N}:\mathcal{H}_{\hbar}^{r}(P)\to\mathcal{H}_{\hbar}^{r}(P)$
in the outermost band is contained in those of $\widetilde{F}_{N}:\mathcal{\widetilde{H}}_{\hbar}^{r}(P_{G},E_{u})\to\mathcal{\widetilde{H}}_{\hbar}^{r}(P_{G},E_{u})$
up to multiplicity.
\end{proof}
With this corollary, we finish the proof of Theorem \ref{thm:weyl_law_extended}.
\begin{proof}[Proof of Theorem \ref{thm:weyl_law_extended}]
 Theorem \ref{thm:weyl_law_extended} follows from Corollary \ref{cor:9.30}
and Theorem \ref{thm:Weyl-formula}, if the potential function $\tilde{V}$
satisfies (\ref{eq:assumtion_on_tildeV}) for some smooth function
$V$ close to $V_{0}$. But the rank of the spectral projection operator
for the outermost band does not depend on the potential function $\tilde{V}$
(provided that the outermost band is isolated) because it coincides
with the number of the eigenvalues of ${\boldt}_{\hbar}^{(0)}$ in
a small neighborhood of $1$. So we get Theorem \ref{thm:weyl_law_extended}.
\end{proof}
For the proof of Proposition \ref{prop:properties_of_pi_*}, we first
prove the following simpler version in the linearized setting. 
\begin{lem}
\label{lem:proj_and_spectral_projector_to_first_band} The operator
$p_{*}:\mathcal{S}(\real^{2d+d'})\to\mathcal{S}(\real^{2d})$ defined
by 
\[
p_{*}u(x)=\int u(x,s)ds
\]
extends to a bounded operator $p_{*}:\mathcal{\widetilde{H}}_{\hbar}^{r}(\real^{2d+d'})\to\mathcal{H}_{\hbar}^{r}(\real^{2d})$
and the operator norm is bounded by a constant independent of $\hbar$.
It restricts to a bijection between the subspaces $H'_{0}$ in Proposition
\ref{prop:prequatum_op_for_hyp_linear} and $E'_{0}$ in Proposition
\ref{prop:prequatum_op_for_hyp_linear-1}. Further there exists a
constant $K$ independent of $\hbar$ such that 
\[
\|p_{*}(u)\|_{\mathcal{\widetilde{H}}_{\hbar}^{r}(\real^{2d+d'})}=K\cdot\|u\|_{\mathcal{H}_{\hbar}^{r}(\real^{2d})}
\]
for all $u\in H'_{0}=\mathrm{Im}\, t_{\hbar}^{(0)}$.\end{lem}
\begin{proof}
We first prove boundedness of the operator $p_{*}:\mathcal{\widetilde{H}}_{\hbar}^{r}(\real^{2d+d'})\to\mathcal{H}_{\hbar}^{r}(\real^{2d})$.
Let us consider the operator $P_{*}:\mathcal{S}(\real^{4d+2d'})\to\mathcal{S}(\real^{4d})$
defined by 
\[
P_{*}v(x,\xi_{x})=a_{d'}^{-1}\cdot(2\pi\hbar)^{-d'/2}\cdot\int v(x,s,\xi_{x},\xi_{s})\exp(-\hbar^{-1}|\xi_{s}|^{2}/2)\frac{dsd\xi_{s}}{(2\pi\hbar)^{d'}}.
\]
It makes the following diagram commute: 
\[
\begin{CD}\mathcal{S}\left(\real_{(x,s)}^{2d+d'}\right)@>\Bargmann_{(x,s)}>>\mathcal{S}\left(\real_{(x,s,\xi_{x},\xi_{s})}^{4d+2d'}\right)\\
@VVp_{*}V@VVP_{*}V\\
\mathcal{S}\left(\real_{x}^{2d}\right)@>\Bargmann_{x}>>\mathcal{S}\left(\real_{(x,\xi_{x})}^{2d}\right)
\end{CD}
\]
Since we have, from the definition of $\mathcal{W}_{\hbar}^{r}(x,\xi_{x})$
in (\ref{eq:def_W_W+}) and that of $\widetilde{\mathcal{W}}_{\hbar}^{r}(x,s,\xi_{x},\xi_{s})$
in (\ref{eq:def_W_W+-1}), that

\[
\mathcal{W}_{\hbar}^{r}(x,\xi_{x})^{2}\cdot\exp(-\hbar^{-1}|\xi_{s}|^{2})\le C\cdot\widetilde{\mathcal{W}}_{\hbar}^{r}(x,s,\xi_{x},\xi_{s})^{2}
\]
for some constant $C>0$ independent of $\hbar$, we obtain, by Schwartz
inequality, 
\begin{align*}
\mathcal{W}_{\hbar}^{r}(x,\xi_{x})^{2}|P_{*}v(x,\xi_{x})|^{2} & =a_{d'}^{-2}\cdot(2\pi\hbar)^{-d'/2}\cdot\left|\int\mathcal{W}_{\hbar}^{r}(x,\xi_{x})v(x,s,\xi_{x},\xi_{s})\exp(-\hbar^{-1}|\xi_{s}|^{2})dsd\xi_{s}\right|^{2}\\
 & \le C\int\widetilde{\mathcal{W}}_{\hbar}^{r}(x,s,\xi_{x},\xi_{s})^{2}|v(x,s,\xi_{x},\xi_{s})|^{2}dsd\xi_{s}
\end{align*}
with $C>0$ a constant independent of $\hbar$. Integrating the both
sides with respect to the variables $x$ and $\xi_{x}$, we get
\[
\|\mathcal{W}_{\hbar}^{r}\cdot P_{*}v\|_{L^{2}}\le C\|\widetilde{\mathcal{W}}_{\hbar}^{r}\cdot v\|_{L^{2}}.
\]
This implies that $p_{*}:\mathcal{\widetilde{H}}_{\hbar}^{r}(\real^{2d+d'})\to\mathcal{H}_{\hbar}^{r}(\real^{2d})$
is bounded. 

We next prove the remaining claims. Recall that the operator $T^{(0)}$
defined in (\ref{eq:def_Tk}) is the rank one projection operator
that assigns a function the constant term of its Taylor expansion
at the origin $0$. For the operator $p_{*}':\mathcal{S}\left(\real_{(\zeta_{p},\xi_{s})}^{d+d'}\right)\to\mathcal{S}\left(\real_{\zeta_{p}}^{d}\right)$
defined by $p_{*}'v(\zeta_{p})=v(\zeta_{p},0)$, we have the commutative
diagram:
\[
\begin{CD}\mathcal{S}\left(\real_{(\zeta_{p},\xi_{s})}^{d+d'}\right)@>T^{(0)}>>\mathcal{S}\left(\real_{(\zeta_{p},\xi_{s})}^{d+d'}\right)\\
@VVp_{*}'V@VVp_{*}'V\\
\mathcal{S}\left(\real_{\zeta_{p}}^{d}\right)@>T^{(0)}>>\mathcal{S}\left(\real_{\zeta_{p}}^{d}\right)
\end{CD}
\]
The images of the operators $T^{(0)}$ on the upper and lower rows
are one-dimensional subspaces that consists of the constant functions.
The operator $p_{*}'$ restricts to an isomorphisms between them.
The commutative diagram above extends to
\[
\begin{CD}H_{\hbar}^{r}\left(\real_{(\zeta_{p},\xi_{s})}^{d+d'}\right)@>T^{(0)}>>H_{\hbar}^{r}\left(\real_{(\zeta_{p},\xi_{s})}^{d+d'}\right)\\
@VVp_{*}'V@VVp_{*}'V\\
H_{\hbar}^{r}\left(\real_{\zeta_{p}}^{d}\right)@>T^{(0)}>>H_{\hbar}^{r}\left(\real_{\zeta_{p}}^{d}\right)
\end{CD}
\]
and hence trivially to 
\begin{equation}
\begin{CD}L^{2}\left(\real_{\nu_{q}}^{d}\right)\otimes H_{\hbar}^{r}\left(\real_{(\zeta_{p},\xi_{s})}^{d+d'}\right)@>\mathrm{Id}\otimes T^{(0)}>>L^{2}\left(\real_{\nu_{q}}^{d}\right)\otimes H_{\hbar}^{r}\left(\real_{(\zeta_{p},\xi_{s})}^{d+d'}\right)\\
@VV\mathrm{Id}\otimes p_{*}'V@VV\mathrm{Id}\otimes p_{*}'V\\
L^{2}\left(\real_{\nu_{q}}^{d}\right)\otimes H_{\hbar}^{r}\left(\real_{\zeta_{p}}^{d}\right)@>\mathrm{Id}\otimes T^{(0)}>>L^{2}\left(\real_{\nu_{q}}^{d}\right)\otimes H_{\hbar}^{r}\left(\real_{\zeta_{p}}^{d}\right)
\end{CD}\label{cd:Tzero}
\end{equation}
This commutative diagram viewed through the isomorphisms
\[
\mathcal{U}:L^{2}\left(\real_{\nu_{q}}^{d}\right)\otimes H_{\hbar}^{r}\left(\real_{\zeta_{p}}^{d}\right)@\mathcal{H}_{\hbar}^{r}\left(\real_{x}^{2d}\right)\quad\mbox{and}\quad\widetilde{\mathcal{U}}:L^{2}\left(\real_{\nu_{q}}^{d}\right)\otimes H_{\hbar}^{r}\left(\real_{(\zeta_{p},\xi_{s})}^{d+d'}\right)\to\mathcal{H}\left(\real_{(x,s)}^{2d+d'}\right)
\]
is just 

\[
\begin{CD}\mathcal{H}_{\hbar}^{r}\left(\real_{(x,s)}^{2d+d'}\right)@>t_{\hbar}^{(0)}>>\mathcal{H}_{\hbar}^{r}\left(\real_{(x,s)}^{2d+d'}\right)\\
@VVp{}_{*}V@VVp{}_{*}V\\
\mathcal{H}_{\hbar}^{r}\left(\real_{x}^{2d}\right)@>t_{\hbar}^{(0)}>>\mathcal{H}_{\hbar}^{r}\left(\real_{x}^{2d}\right)
\end{CD}
\]
Therefore for the proofs of the remaining claims of the proposition,
it is enough to prove the corresponding claims in the diagram (\ref{cd:Tzero}).
But they are now obvious from the fact that $T^{(0)}$ is a rank one
projection operator. 
\end{proof}
~
\begin{proof}[Proof of Proposition \ref{prop:properties_of_pi_*} ]
 The proofs of the claims other than Claim 2 are obtained by applying
Lemma \ref{lem:proj_and_spectral_projector_to_first_band} to the
local data and showing that the effect of non-linearity of the coordinate
transformations is negligible. We omit the detail of the argument
because it should be clear if we recall the argument in Section \ref{sec:Proofs-of-the_mains}
and the preceding subsections in this section. 

We prove Claim 2. Recall that $E_{u}:M\to G$ is the section that
assigns the unstable subspace to each point. Let $\mu_{u}=(E_{u})_{*}(Vol_{\omega})$
where $Vol_{\omega}$ is the symplectic volume on $M$. We consider
the operator
\[
\iota:C_{N}^{\infty}(P)\to C_{N}^{\infty}(P_{G})'\quad\iota(u)=(u\circ p)\cdot\mu_{u}
\]
which satisfies $p_{*}\circ\iota=\mathrm{Id}$. To show that $\mathcal{\widetilde{H}}_{\hbar}^{r}(P)$
contains the space $C_{N}^{\infty}(P)$, it is enough to prove that
the image $\iota(C_{N}^{\infty}(P))$ of this operator is contained
in $\mathcal{\widetilde{H}}_{\hbar}^{r}(P_{G},E_{u})$. We define
a Hilbert space of distributions $\mathcal{\widehat{H}}_{\hbar}^{r}(G,E_{u})\subset C^{\infty}(G)'$
on $G$ as the completion of the space $C_{\hbar}^{\infty}(G,E_{u})$
with respect to the norm 
\[
\|u\|_{\mathcal{\widehat{H}}_{\hbar}^{r}}:=\left(\sum_{i=1}^{I_{\hbar}}\|W_{\hbar}^{r}\cdot\Bargmann_{x}u_{i}\|_{L^{2}}^{2}\right)^{1/2}
\]
where $u_{i}(x):=\psi_{i}(x)\cdot u(\kappa_{i}(x))$ and $W_{\hbar}^{r}:\real^{2d+d'}\oplus\real^{2d+d'}\to\real$
is defined by 
\[
W_{\hbar}^{r}(p,q,s,\xi_{p},\xi_{q},\xi_{s})=W^{r}(\hbar^{-1/2}\xi_{p},(\hbar^{-1/2}\xi_{q},\hbar^{-1/2}\xi_{s})).
\]
Then, from the results in \cite{Baladi-Tsujii08}and \cite{fred-roy-sjostrand-07},
the Perron-Frobenius operator 
\[
Q:C_{\hbar}^{\infty}(G,E_{u})\to C_{\hbar}^{\infty}(G,E_{u}),\quad Qu(x)=\frac{1}{|\det(Df_{G}|_{\ker p})|}\cdot u\circ f_{G}^{-1}
\]
extends to a bounded operator on $\mathcal{\widehat{H}}_{\hbar}^{r}(G,E_{u})$.
The spectral radius of $Q:\mathcal{\widehat{H}}_{\hbar}^{r}(G,E_{u})\to\mathcal{\widehat{H}}_{\hbar}^{r}(G,E_{u})$
is $1$ and the essential spectral radius is smaller than $1$. Further
$1$ is the unique eigenvalue on the unit circle, which is simple,
and the corresponding eigenvector is $\mu_{u}$. In particular, the
measure $\mu_{u}$ belongs to $\mathcal{\widehat{H}}_{\hbar}^{r}(G,E_{u})$.
To finish, we consider the operator 
\[
j:C_{\hbar}^{\infty}(G,E_{u})\times C_{N}^{\infty}(P_{G})\to C_{N}^{\infty}(P_{G}),\quad j(u,v)=u\cdot v
\]
and check that it extends to a continuous operator%
\footnote{Note that this claim is for each fixed $\hbar$ and we do not claim
any uniformly in $\hbar$. So the proof is easy if we consider in
the local charts.%
} 
\[
j:\mathcal{\widehat{H}}_{\hbar}^{r}(G,E_{u})\times C_{N}^{\infty}(P_{G})\to\mathcal{\widetilde{H}}_{\hbar}^{r}(P_{G},E_{u}).
\]
Since $\iota(u)=j(\mu_{u},u)$, this implies that $\iota(C_{N}^{\infty}(P))$
is contained in $\mathcal{\widetilde{H}}_{\hbar}^{r}(P_{G},E_{u})$.\end{proof}

\section{\label{sec:9}Proof of Th. \ref{Thm:Distribution-of-resonances.}.
Concentration of most of the external eigenvalues on a circle. }

\subsection{Time average and Birkhoff's ergodic theorem}

Let us write $\hat{F}_{V}$ and $\hat{F}_{V,N}$ for the prequantum
transfer operators $\hat{F}$ and $\hat{F}_{N}$ defined respectively
in (\ref{eq:def_prequantum_operator_F}) and (\ref{eq:def_F_N_on_P}),
specifying dependence on the potential function $V$. In this subsection,
we prove the first part of Theorem \ref{Thm:Distribution-of-resonances.},
namely that almost all eigenvalues of $\hat{F}_{V,N}$ of the external
band concentrate at the value $e^{\left\langle V-V_{0}\right\rangle }$
as $N\rightarrow\infty$, where $V_{0}$ is the potential of reference
defined in (\ref{eq:choice_V_smooth}) and whose importance has been
shown in Theorem \ref{cor:The-special-choice_V0}. Take $\epsilon>0$
arbitrarily and let $V_{\epsilon}$ be a smooth approximation of the
function $V_{0}$ such that 
\[
|V_{0}(x)-V_{\epsilon}(x)|<\epsilon\quad\mbox{for all }x\in M.
\]
We introduce the ``(approximate) damping function'' as the difference
\[
D:=V-V_{\epsilon}
\]

The next lemma shows that the transfer operator $\hat{F}_{V}$ is
conjugate (and hence has the same spectrum) to an operator $\mathcal{L}_{D_{n}}$
which is the operator of reference $\hat{F}_{V_{\epsilon}}$ with
an additional potential $D_{n}$ obtained by the time average of the
damping function $D$. This presentation of the operator $\hat{F}_{V}$
will be very convenient for our purpose.
\begin{lem}
For any $n\geq1$, 
\begin{equation}
\hat{F}_{V}=e^{G_{n}}\circ\mathcal{L}_{D_{n}}\circ e^{-G_{n}}\label{eq:conjugation_with_L_Dn}
\end{equation}
where
\[
\mathcal{L}_{D_{n}}:=e^{D_{n}}F_{V_{\epsilon}}
\]
and
\[
D_{n}:=\frac{1}{n}\sum_{k=1}^{n}D\circ\tilde{f}^{-k}
\]
is the time averaged of the ``damping function'' $D$ and $G_{n}$
is the multiplication operator by the function (with same name):
\[
G_{n}=\frac{1}{n}\sum_{k=1}^{n}kD\circ\tilde{f}^{-n+k}
\]
\end{lem}
\begin{proof}
From (\ref{eq:def_prequantum_operator_F}), we have that
\[
e^{G_{n}}\circ\left(e^{D_{n}}\hat{F}_{V_{\epsilon}}\right)\circ e^{-G_{n}}=e^{G_{n}+D_{n}-G_{n}\circ\tilde{f}^{-1}}\hat{F}_{V_{\epsilon}}.
\]
We compute that
\begin{align*}
G_{n}+D_{n}-G_{n}\circ\tilde{f}^{-1} & =\frac{1}{n}\sum_{k=1}^{n}\left(kD\circ\tilde{f}^{-n+k}+D\circ\tilde{f}^{-n+k-1}\right)-\frac{1}{n}\sum_{k=0}^{n-1}\left(k+1\right)D\circ\tilde{f}^{-n+k}\\
 & =D=V-V_{\epsilon}
\end{align*}
Hence
\[
e^{G_{n}}\circ\left(e^{D_{n}}\hat{F}_{V_{\epsilon}}\right)\circ e^{-G_{n}}=e^{V-V_{\epsilon}}\hat{F}_{V_{\epsilon}}=\hat{F}_{V}.
\]

\end{proof}
Due to ergodicity of the map $f:M\rightarrow M$, the time averaged
function $D_{n}$ converges for $n\rightarrow\infty$ almost everywhere
to its spatial average: 
\[
\left\langle D\right\rangle :=\frac{1}{\mbox{Vol}_{\omega}\left(M\right)}\int_{M}Ddx
\]
In particular we will use later on that
\begin{equation}
\int_{M}\left(e^{D_{n}}-e^{\left\langle D\right\rangle }\right)dx\underset{n\rightarrow\infty}{\rightarrow}0.\label{eq:ergodicity}
\end{equation}

Let
\[
\mathcal{L}_{\left\langle D\right\rangle }:=e^{\left\langle D\right\rangle }F_{V_{\epsilon}}
\]
Let $\mathcal{L}_{\left\langle D\right\rangle ,N}$ be the restriction
of $\mathcal{L}_{\left\langle D\right\rangle }$ to $\mathcal{H}_{N}^{r}\left(P\right)$.
Its spectrum is that of $\hat{F}_{V_{\epsilon},N}$ multiplied by
the constant factor $e^{\left\langle D\right\rangle }$, in particular
its external part concentrates (as $N\rightarrow\infty$) to the annulus
$e^{\left\langle D\right\rangle -\epsilon}\le|z|\le e^{\left\langle D\right\rangle +\epsilon}$
and the disk $|z|\le(1/\lambda)e^{\left\langle D\right\rangle +\epsilon}$
from Theorem \ref{thm:band_structure}. Let $\tau_{N}^{\left(0\right)}$
be the (finite rank) approximate projection operator on the external
band of $\hat{F}_{V,N}$ introduced in Theorem \ref{thm:More-detailled-description_of_spectrum}.
(We put the subscript $N$ now to make the dependence on $N$ explicit.)
Recall that this approximate spectral projector $\tau_{N}^{\left(0\right)}$
does not depend on the choice of the potential $V$, hence it is also
suitable for the transfer operators $\mathcal{L}_{D_{n}}$ and $\mathcal{L}_{\left\langle D\right\rangle }$.
We want now to study the ``quantity of spectrum'' of the external
band of $\mathcal{L}_{D_{n},N}$ which does not concentrates on the
annulus $e^{\left\langle D\right\rangle -\epsilon}\le|z|\le e^{\left\langle D\right\rangle +\epsilon}$
. To this end, we introduce the operator:
\[
S_{n,N}:=\tau_{N}^{\left(0\right)}\left(\mathcal{L}_{D_{n},N}-\mathcal{L}_{\left\langle D\right\rangle ,N}\right)\tau_{N}^{\left(0\right)}=\tau_{N}^{\left(0\right)}\left(\left(e^{D_{n}}-e^{\left\langle D\right\rangle }\right)\cdot\hat{F}_{V_{\epsilon},N}\right)\tau_{N}^{\left(0\right)}
\]
and $T_{n,N}:=\mathcal{L}_{D_{n},N}-S_{n,N}$ giving the following
decomposition:
\[
\mathcal{L}_{D_{n},N}=T_{n,N}+S_{n,N}
\]
Notice that $T_{n,N}$ is in some sense $\mathcal{L}_{D_{n},N}$ but
with the the external band replaced by the spectrum of $\mathcal{L}_{\left\langle D\right\rangle ,N}$.
As a consequence, from Theorem \ref{thm:band_structure}, for $N$
large enough, the operator $T_{n,N}$ has no spectrum in the spectral
domain%
\footnote{Here we used the fact that $V_{\epsilon}$ is an $\epsilon$-approximation
of $V_{0}$.%
}
\[
W_{\epsilon}:=\left\{ z\in\mathbb{C},\quad\frac{1}{\lambda}e^{\sup D_{n}+\epsilon}\leq\left|z\right|\leq e^{\left\langle D\right\rangle -2\epsilon}\mbox{ or }e^{\left\langle D\right\rangle +2\epsilon}\leq\left|z\right|\right\} 
\]
and, moreover, on $W_{\epsilon}$ we have a bound on the norm of the
resolvent of $T_{n,N}$: there exists $C_{\epsilon}>0$, and $N_{\epsilon}$
such that for any $N\geq N_{\epsilon}$,
\begin{equation}
\left\Vert \left(z-T_{n,N}\right)^{-1}\right\Vert \leq C_{\epsilon}\quad\mbox{uniformly for z\ensuremath{\in W_{\epsilon}}}\label{eq:bound_on_resolvent_of_T}
\end{equation}
and
\begin{equation}
\left\Vert T_{n,N}\right\Vert \leq e^{\left\langle D\right\rangle +2\epsilon}.\label{eq:norm_T}
\end{equation}

Recall that the number of eigenvalues in the external band is $C_{0}\cdot N^{d}+\mathcal{O}(N^{-\epsilon})$
with $C_{0}=Vol_{\omega}(M)$ from Theorem \ref{thm:More-detailled-description_of_spectrum}.
The next lemma concerns the trace norm of $S_{n,N}$ and is the key
to show that the ``perturbation'' $S_{n,N}$ may add only a relatively
negligible number of eigenvalues on $W_{\epsilon}$.
\begin{lem}
\label{lem:trace_norm_estiamate}For any $\epsilon'>0$, there exists
$n\geq1$, $N_{\epsilon'}>0$, such that, for any $N>N_{\epsilon'}$,
\begin{equation}
\left\Vert S_{n,N}\right\Vert _{\mathrm{Tr}}\leq\epsilon'N^{d}.\label{eq:bound_trace_S}
\end{equation}
\end{lem}
\begin{proof}
We have $S_{n,N}=\tau_{N}^{\left(0\right)}\left(e^{D_{n}}-e^{\left\langle D\right\rangle }\right)\hat{F}_{V_{\epsilon},N}\tau_{N}^{\left(0\right)}$
so
\[
\left\Vert S_{n,N}\right\Vert _{\mathrm{Tr}}\leq\left\Vert \tau_{N}^{\left(0\right)}\left(e^{D_{n}}-e^{\left\langle D\right\rangle }\right)\right\Vert _{\mathrm{Tr}}\left\Vert \hat{F}_{V_{\epsilon},N}\tau_{N}^{\left(0\right)}\right\Vert \leq C\left\Vert \tau_{N}^{\left(0\right)}\left(e^{D_{n}}-e^{\left\langle D\right\rangle }\right)\right\Vert _{\mathrm{Tr}}
\]
Let $\epsilon'>0$. From (\ref{eq:trace_estimate_on_tauMphi}), 
we have 
\[
\left\Vert \tau_{N}^{\left(0\right)}\left(e^{D_{n}}-e^{\left\langle D\right\rangle }\right)\right\Vert _{\mathrm{Tr}}\leq CN^{d}\left(\int_{M}\left|e^{D_{n}}-e^{\left\langle D\right\rangle }\right|dx\right)
\]
From (\ref{eq:ergodicity}), there exists $n$ such that $\left|\int_{M}\left(e^{D_{n}}-e^{\left\langle D\right\rangle }\right)\right|<\epsilon'/C^{2}$,
giving $\left\Vert S_{n,N}\right\Vert _{\mathrm{Tr}}\leq\epsilon'N^{d}$
for $N\geq N_{\epsilon'}:=N_{n}$.
\end{proof}

\subsection{Proof of concentration of the moduli of the resonance to the circle
$|z|=e^{\langle V-V_{0}\rangle}$}

Now, as a consequence of (\ref{eq:bound_on_resolvent_of_T}) and (\ref{eq:bound_trace_S})
concerning the decomposition $\mathcal{L}_{D_{n},N}=T_{n,N}+S_{n,N}$,
we prove 
\begin{lem}
\label{lem:The-number-of_zeros}For any $\epsilon>0$, there exists
$n$ and $N_{\epsilon}$ such that, for any $N>N_{\epsilon}$, the
number of eigenvalues of $\mathcal{L}_{D_{n},N}$ in $W_{\epsilon}$
counting multiplicities, is bounded by $\epsilon N^{d}$.
\end{lem}
Before giving the proof, remark that due to the conjugation (\ref{eq:conjugation_with_L_Dn}),
the same result holds for the operator $\widetilde{F}_{\tilde{V},N}$
and this finishes the proof for the claim on ``radial concentration
of eigenvalues'' in Theorem \ref{Thm:Distribution-of-resonances.}.
The proof that Lemma \ref{lem:The-number-of_zeros} follows from Lemma
\ref{lem:trace_norm_estiamate} should be well-known. Here we propose
two different proofs, because they are both interesting in their own
right. The first proof is based on Weyl inequality while the second
is based on Jensen's formula applied to the relative determinant.

\subsubsection{Proof of Lemma \ref{lem:The-number-of_zeros} using Weyl inequality}

We will just apply Lemma \ref{lem:0.5} stated below to the family
of operators $A_{N}=e^{-\langle D\rangle}T_{n,N}$ and $B_{N}=e^{-\langle D\rangle}S_{n,N}$
with setting $F\left(N\right)=\epsilon'N^{d}e^{\langle D\rangle}$
given from (\ref{eq:bound_trace_S}).
\begin{lem}
\label{lem:A+Trace_B}Suppose that $A_{N},B_{N}$ is a family of bounded
operators on an Hilbert space, depending on $N\in\mathbb{Z}$, such
that for every $N$, we have $\left\Vert A_{N}\right\Vert \leq1$,
$\left\Vert B_{N}\right\Vert \leq C_{B}$ with some $C_{B}>0$ and,
moreover, the operators $B_{N}$ are trace class and $\left\Vert B_{N}\right\Vert _{\mathrm{Tr}}\leq F\left(N\right)$
with some function $F:\mathbb{N}\rightarrow\mathbb{R}^{+}$. Then
for any $\epsilon>0$,
\[
\sharp\left\{ \sigma\left(A_{N}+B_{N}\right)\cap\left\{ z\in\mathbb{C},\left|z\right|\geq1+\epsilon\right\} \right\} \leq\frac{1}{\epsilon}C\cdot F\left(N\right)
\]
with $C=\frac{\left(2+C_{B}\right)}{2\log\left(1+\epsilon\right)}$.\end{lem}
\begin{proof}
Let $C_{N}:=A_{N}+B_{N}$ and $P_{N}:=C_{N}^{*}C_{N}$. Let us write
\[
P_{N}=\mathcal{A}_{N}+\mathcal{B}_{N}
\]
with $\mathcal{A}_{N}:=A_{N}^{*}A_{N}$, $\mathcal{B}_{N}:=A_{N}^{*}B_{N}+B_{N}^{*}A_{N}+B_{N}^{*}B_{N}$.
We have that $\left\Vert \mathcal{A}_{N}\right\Vert \leq1$, $\left\Vert \mathcal{B}_{N}\right\Vert \leq2C_{B}+C_{B}^{2}$,
$\mathcal{B}_{N}$ is in the trace class and $\left\Vert \mathcal{B}_{N}\right\Vert _{\mathrm{Tr}}\leq\left(2+C_{B}\right)F\left(N\right)$.
The operators $C_{N}$ and $P_{N}$ have discrete spectrum outside
the circle of radius $1$. For each $N$, let $\left(\lambda_{j}\right)_{j=1,2,\cdots,M_{N}}$
denote the eigenvalues of $C_{N}$ in the domain $\left\{ z\in\mathbb{C},\left|z\right|>1+\epsilon\right\} $,
ordered in such a way that $\left|\lambda_{1}\right|\geq\left|\lambda_{2}\right|\ldots\geq\left|\lambda_{M_{N}}\right|$.
(For simplicity, we do not indicate the dependence on N). Observe
that $P_{N}$ is self-adjoint and positive. Let $p_{1}\geq p_{2}\geq\ldots\ge p_{M_{N}}$
denote the eigenvalues of $P_{N}$ counting multiplicity. The $\left(p_{j}^{1/2}\right)_{j}$
are the singular values of $C_{N}$. Weyl inequalities give (see \cite[p.50]{gohberg-00}
for a proof):
\begin{equation}
\sum_{j=1}^{M_{N}}\log\left|\lambda_{j}\right|\leq\frac{1}{2}\sum_{j=1}^{M_{N}}\log p_{j}\label{eq:Weyl_ineq}
\end{equation}
Since $|\lambda_{j}|\ge1+\epsilon$, that is, $\log|\lambda_{i}|\ge\log(1+\epsilon)$,
this implies and $\log p_{j}\le p_{j}-1$. 
\begin{equation}
M_{N}\cdot\log(1+\epsilon)\le\frac{1}{2}\sum_{j=1}^{M_{N}}\log p_{j}\le\frac{1}{2}\sum_{j=1}^{M_{N}}(p_{j}-1).\label{eq:consequence_of_Weyl_ineq}
\end{equation}

Let $\left(f_{j}\right)_{j\in\left\{ 1,\ldots M_{N}\right\} }$ denote
the associated eigenvectors of $P_{N}$ for eigenvalues $\left(p_{j}\right)_{j\in\left\{ 1,\ldots M_{N}\right\} }$.
Put $\mathcal{E}_{N}:=\mathrm{Span}\left(f_{j},j\in\left\{ 1,\ldots M_{N}\right\} \right)$.
Then we have
\begin{align}
\sum_{j=1}^{M_{N}}p_{j} & =\mathrm{Tr}\left(P_{N}\mid_{\mathcal{E}_{N}}\right)=\mathrm{Tr}\left(\mathcal{A}_{N}\mid_{\mathcal{E}_{N}}\right)+\mathrm{Tr}\left(\mathcal{B}_{N}\mid_{\mathcal{E}_{N}}\right)\nonumber \\
 & \leq\sum_{j=1}^{M_{N}}\left|\left(f_{j},\mathcal{A}_{N}f_{j}\right)\right|+\sum_{j=1}^{M_{N}}\left|\left(f_{j},\mathcal{B}_{N}f_{j}\right)\right|\nonumber \\
 & \leq\left(\sum_{j=1}^{M_{N}}\left\Vert \mathcal{A}_{N}\right\Vert \right)+\left\Vert \mathcal{B}_{N}\right\Vert _{\mathrm{Tr}}\leq M_{N}+\left(2+C_{B}\right)F\left(N\right)\label{eq:ineq_2}
\end{align}
That is
\[
\sum_{j=1}^{M_{N}}(p_{j}-1)\le\left(2+C_{B}\right)F\left(N\right).
\]
Putting this in (\ref{eq:consequence_of_Weyl_ineq}), we obtain the
conclusion.
\end{proof}
Note that the simple lemma above already gives the conclusion of Lemma
\ref{lem:The-number-of_zeros} on the outer connected component $\left\{ |z|\ge e^{\left\langle D\right\rangle +2\epsilon}\right\} $
of $W_{\epsilon}$. In order to look into the inner component of $W_{\epsilon}$,
we need a little more argument. 
\begin{lem}
\label{lem:0.5}Let $U\subset\mathbb{C}$ be an open disk and $z_{0}\in U$.
Let $f$ be a Möebius transformation such that $f\left(U\right)=\left\{ z\in\mathbb{C},\left|z\right|>1\right\} $
and $f\left(z_{0}\right)=\infty$. Suppose that $A_{N}$ and $B_{N}$
are family of bounded operators on an Hilbert space, depending on
$N\in\mathbb{N}$. Suppose also that, for some constant $C>0$ independent
of $N$, we have 
\[
\left\Vert A_{N}\right\Vert \leq C,\quad\left\Vert B_{N}\right\Vert \leq C,\quad\left\Vert \left(z_{0}-\left(A_{N}+B_{N}\right)\right)^{-1}\right\Vert \leq C
\]
 and 
\[
\left\Vert \left(z-A_{N}\right)^{-1}\right\Vert \leq C\quad\mbox{for all \ensuremath{z\in U}}.
\]
 Moreover, suppose that the operators $B_{N}$ are in the trace class
and $\left\Vert B_{N}\right\Vert _{\mathrm{Tr}}\leq F\left(N\right)$
with some function $F:\mathbb{N}\rightarrow\mathbb{R}^{+}$. Then
there exists a constant $C'>0$, which depends only on the constant
$C$, such that, for every $N$, we have
\begin{equation}
\sharp\left\{ \sigma\left(A_{N}+B_{N}\right)\cap U\right\} \leq C'\cdot F\left(N\right).\label{eq:relation_spectrum}
\end{equation}
\end{lem}
\begin{proof}
Let $C_{N}:=A_{N}+B_{N}$. Using holomorphic functional calculus,
we define:
\begin{equation}
A'_{N}=f\left(A_{N}\right):=\frac{1}{2\pi i}\oint_{\gamma}f\left(z\right)R_{A_{N}}\left(z\right)dz,\quad\mbox{where }R_{A_{N}}\left(z\right):=\left(z-A_{N}\right)^{-1}\label{eq:f(A_n)}
\end{equation}
Here the path $\gamma=\gamma_{1}\cup\gamma_{2}$ is taken as follows
so that they enclose the spectrum of $A_{N}$ for all $N$: the closed
path $\gamma_{1}$ is taken so that $f\left(\gamma_{1}\right)$ is
the clockwise circle of radius $1-\epsilon_{C}$ with some $\epsilon_{C}>0$
(this is possible because $\mathrm{dist}\left(z,\sigma\left(A_{N}\right)\right)\geq\left\Vert R_{A_{N}}\left(z\right)\right\Vert ^{-1}\geq C^{-1}$
for $z\in U$.); the close path $\gamma_{2}$ is a counterclockwise
circle of sufficiently large radius, say, $C+\epsilon_{C}\ge\|A_{N}\|+\epsilon_{C}$.
From (\ref{eq:f(A_n)}), we see that $A_{N}^{'}$ is bounded and moreover
that $\left\Vert A'{}_{N}^{k}\right\Vert <1$ for some $k$ which
depends only on $f$ and $C$. Equivalently we can modify the norm
$\left\Vert .\right\Vert $ (in a way which depends on $f$ and $C$
only) such that $\left\Vert A'_{N}\right\Vert <1$ . Similarly we
define
\[
C'_{N}=f\left(C_{N}\right):=\frac{1}{2\pi i}\oint_{\gamma'}f\left(z\right)R_{C_{N}}\left(z\right)dz,\quad\mbox{where }R_{C_{N}}\left(z\right):=\left(z-C_{N}\right)^{-1}
\]
where $\gamma'=\gamma'_{1}\cup\gamma_{2}$ with $\gamma'_{1}$ a small
clockwise circle around the point $z_{0}$ so that it enclose the
spectrum of $C_{N}$. From the assumption $\left\Vert \left(z_{0}-\left(A_{N}+B_{N}\right)\right)^{-1}\right\Vert \leq C$,
the operator $C'_{N}$ thus defined is bounded uniformly with respect
to $N$. Let 
\[
B'_{N}:=C'_{N}-A'_{N}=\frac{1}{2\pi i}\oint_{\gamma'}f\left(z\right)\left(R_{C_{N}}\left(z\right)-R_{A_{N}}\left(z\right)\right)dz
\]
From the relation
\[
R_{C_{N}}\left(z\right)-R_{A_{N}}\left(z\right)=\frac{1}{2}\left(R_{C_{N}}\left(z\right)B_{N}R_{A_{N}}\left(z\right)+R_{A_{N}}\left(z\right)B_{N}R_{C_{N}}\left(z\right)\right)
\]
we deduce that $\left\Vert R_{C_{N}}\left(z\right)-R_{A_{N}}\left(z\right)\right\Vert _{\mathrm{Tr}}\leq CF\left(N\right)$
and then that $\left\Vert B'_{N}\right\Vert _{\mathrm{Tr}}\leq CF\left(N\right)$
with some $C>0$ independent on $N$. We can apply Lemma \ref{lem:A+Trace_B}
to the operators $A'_{N},B'_{N}$ and get that
\[
\sharp\left\{ \sigma\left(A'_{N}+B'_{N}\right)\cap\left\{ z\in\mathbb{C},\left|z\right|>1\right\} \right\} \leq C'\cdot F\left(N\right)
\]
for some $C'$ independent on $N$. By the spectral mapping theorem,
we have 
\[
\sharp\left\{ \sigma\left(A_{N}+B_{N}\right)\cap U\right\} =\sharp\left\{ \sigma\left(A'_{N}+B'_{N}\right)\cap\left\{ z\in\mathbb{C},\left|z\right|>1\right\} \right\} .
\]
So we get (\ref{eq:relation_spectrum}).
\end{proof}
Now we apply this lemma to the setting mentioned in the beginning.
We take finitely many disks $U_{i}\subset\mathbb{C}$ so that we may
apply Lemma \ref{lem:0.5} to each $U_{i}$. (In particular we take
$U_{i}$ so that it intersects the region $e^{-\langle D\rangle}(r_{1}^{+}+\varepsilon)^{n}<|z|<e^{-\langle D\rangle}(r_{0}^{-}-\varepsilon)^{n}$
for small $\varepsilon>0$ where the resolvent $(z-A_{N}-B_{N})^{-1}$
should be bounded uniformly in $N$ and $n$, and where we take the
pint $z_{0}$ in Lemma \ref{lem:0.5}.) By covering the inner connected
component of the region $e^{-\langle D\rangle}W_{\epsilon}$ by such
disks, we conclude Lemma \ref{lem:trace_norm_estiamate}.

\subsubsection{Proof of Lemma \ref{lem:The-number-of_zeros} using Jensen formula}

From the expression $\mathcal{L}_{D_{n},N}=T_{n,N}+S_{n,N}$, we write
$z-\mathcal{L}_{D_{n},N}$ for $z\in W_{\epsilon}$ as 
\begin{equation}
z-\mathcal{L}_{D_{n},N}=z-T_{n,N}-S_{n,N}=\left(z-T_{n,N}\right)\left(1-\mathcal{K}\left(z\right)\right)\label{eq:eqL1}
\end{equation}
with setting
\[
\mathcal{K}\left(z\right):=\left(z-T_{n,N}\right)^{-1}S_{n,N}.
\]
Since $z-T_{n,N}$ is invertible for $z\in W_{\epsilon}$ from (\ref{eq:bound_on_resolvent_of_T}),
we have 
\[
\left(z-T_{n,N}\right)^{-1}\left(z-\mathcal{L}_{D_{n},N}\right)=\left(1-\mathcal{K}\left(z\right)\right)\quad\mbox{for \ensuremath{z\in W_{\epsilon}}.}
\]
Since $\mathcal{K}(x)$ is a finite rank operators (because so is
$S_{n,N}$), we can define 
\[
\mathcal{D}\left(z\right):=\det\left(1-\mathcal{K}\left(z\right)\right)
\]
and see that the eigenvalues of $\mathcal{L}_{D_{n},N}$ in $W_{\epsilon}$
coincide with the zeroes of $\mathcal{D}\left(z\right)$ in $W_{\epsilon}$
up to multiplicity. From the formula $\log\mathcal{D}\left(z\right)=\mathrm{Tr}\log\left(1-\mathcal{K}\left(z\right)\right)$
and the simple inequality $\log\left(1-x\right)\leq x$ for $x>-1$,
we have that, for $N>\max\left(N_{\epsilon},N_{\epsilon'}\right)$,
\begin{align}
\log\left|\mathcal{D}\left(z\right)\right| & \leq\left\Vert \mathcal{K}\left(z\right)\right\Vert _{\mathrm{Tr}}\leq\left\Vert \left(z-T_{n,N}\right)^{-1}\right\Vert \left\Vert S_{n,N}\right\Vert _{\mathrm{Tr}}\underset{}{\leq}C_{\epsilon}\epsilon'N^{d}\le\epsilon N^{d}\label{eq:ine_D_1}
\end{align}
from (\ref{eq:bound_on_resolvent_of_T}) and (\ref{eq:bound_trace_S}).
In the last inequality, we chose $\epsilon'>0$ such that $C_{\epsilon}\epsilon'<\epsilon$.
We next show that $\log\left|\mathcal{D}\left(z\right)\right|$ is
not too small on some part of the region $W_{\epsilon}$. More precisely,
we let
\[
W_{\epsilon}':=\left\{ z\in\mathbb{C},\quad\frac{1}{\lambda}e^{\sup D_{n}+2\epsilon}\leq\left|z\right|\leq e^{\inf D_{n}-3\epsilon}\mbox{ or }e^{\sup D_{n}+3\epsilon}\leq\left|z\right|\right\} \subset W_{\epsilon}
\]
and show that 
\begin{equation}
\log|\mathcal{D}(z)|\ge-\epsilon N^{d}\quad\mbox{for \ensuremath{z\in W'_{\epsilon}}.}\label{eq:ineg_D_2}
\end{equation}
From Theorem \ref{thm:band_structure} applied to $\mathcal{L}_{D_{n},N}$,
we see that $W'_{\epsilon}$ is in a gap (or on the outside) of the
bands given there and that 
\[
\left\Vert \left(z\mathrm{Id}-\mathcal{L}_{D_{n},N}\right)^{-1}\right\Vert \leq C_{\epsilon}\quad\mbox{for all }z\in W'_{\epsilon}.
\]
So, from (\ref{eq:eqL1}), we see that $1-\mathcal{K}(z)$ for $z\in W_{\epsilon}'$
is invertible and that 
\begin{align*}
\left(1-\mathcal{K}\left(z\right)\right)^{-1} & =\left(z-\mathcal{L}_{D_{n},N}\right)^{-1}\left(z-T_{n,N}\right)=\left(z-\mathcal{L}_{D_{n},N}\right)^{-1}\left(z-\mathcal{L}_{D_{n},N}+S_{n,N}\right)\\
 & =1+\left(z-\mathcal{L}_{D_{n},N}\right)^{-1}S_{n,N}.
\end{align*}
Hence, similarly to (\ref{eq:ine_D_1}), we get

\begin{align*}
-\log\left|\mathcal{D}\left(z\right)\right| & \leq\left\Vert \left(z-\mathcal{L}_{D_{n},N}\right)^{-1}S_{n,N}\right\Vert _{\mathrm{Tr}}\leq\left\Vert \left(z-\mathcal{L}_{D_{n},N}\right)^{-1}\right\Vert \left\Vert S_{n,N}\right\Vert _{\mathrm{Tr}}\leq C_{\epsilon}\epsilon'N^{d}<\epsilon N^{d}
\end{align*}
for $N>\max\left(N_{\epsilon},N_{\epsilon'}\right)$. 

Finally, we employ a theorem, Jensen's formula \cite[p.236]{paatero_book_69},
in complex analysis to show that the inequalities (\ref{eq:ine_D_1})
and (\ref{eq:ineg_D_2}) imply that, for arbitrarily small $\tilde{\epsilon}>0$,
we have 
\begin{equation}
\sharp\left\{ z\in W_{\epsilon},\mathcal{D}\left(z\right)=0\right\} <\tilde{\epsilon}N^{d}\label{eq:number_of_zeroes}
\end{equation}
for sufficiently large $N$. This finish the proof of Lemma \ref{lem:The-number-of_zeros}.

We cover the domain $W_{\epsilon}$ by finitely many open topological
disks $D_{i}$, $1\leq i\leq l$, so that $D_{i}\cap W'_{\epsilon}\neq\emptyset$
for every $i$. Let $\Phi_{i}:D_{i}\rightarrow\mathbb{D}:=\left\{ w\in\mathbb{C},\left|w\right|<1\right\} $
be a Riemann mapping such that $\Phi_{i}\left(z_{i}\right)=0$ for
some $z_{i}\in D_{i}\cap W'_{\epsilon}$. Take $\delta>0$ so small
that $\bigcup_{i}\Phi_{i}^{-1}\left(\left\{ \left|w\right|<1-\delta\right\} \right)$
also cover $W_{\epsilon}$. Let
\[
\mathcal{D}_{i}\left(w\right):=\mathcal{D}\left(\Phi_{i}\left(w\right)\right)
\]
which is a holomorphic function on $\mathbb{D}$ with zeroes $w_{j}\in\mathbb{D}$,
which correspond to the zeroes of $\mathcal{D}\left(z\right)$. 

Jensen's formula writes
\[
\sum_{j}-\log\left|w_{j}\right|=-\log\left|\mathcal{D}_{i}\left(0\right)\right|+\frac{1}{2\pi}\int_{\left|w\right|=1}\log\left|\mathcal{D}_{i}\left(w\right)\right|dw
\]
Inequality (\ref{eq:ineg_D_2}) gives $-\log\left|\mathcal{D}_{i}\left(0\right)\right|<\epsilon N^{d}$,
and (\ref{eq:ine_D_1}) gives $\frac{1}{2\pi}\int_{\left|w\right|=1}\log\left|\mathcal{D}_{i}\left(w\right)\right|dw<\epsilon N^{d}$.
For the zero $w_{j}$ in the disk $\left\{ \left|w\right|<1-\delta\right\} $
we have $-\log\left|w_{j}\right|>-\log\left(1-\delta\right)\geq\delta$
and, hence, we see that the number $\mathcal{N}_{i}$ of such zeros
is bounded:
\[
\mathcal{N}_{i}\leq\sum_{j,\left|w_{j}\right|<1-\delta}\frac{-\log\left|w_{j}\right|}{\delta}\leq\frac{1}{\delta}\left(\sum_{j,w_{j}\in\mathbb{D}}-\log\left|w_{j}\right|\right)<\frac{\epsilon N^{d}}{\delta}.
\]
Summing over the sets $D_{i}$ and noting arbitrariness of $\epsilon>0$,
we conclude (\ref{eq:number_of_zeroes}).

\subsection{Proof of equidistribution of the arguments of the resonances\label{sub:Proof-of-equidistribution}}

In this subsection we prove the second part of Theorem \ref{Thm:Distribution-of-resonances.},
namely the equidistribution of the arguments. We write the eigenvalues
of $\hat{\mathcal{F}}_{\hbar}$ as
\[
\lambda_{j}=\rho_{j}e^{i\theta_{j}},\quad\rho_{j}\geq0,\;\theta_{j}\in\mathbb{R},\; j=1,2,\ldots,\mathcal{N}_{\hbar},
\]
with $\mathcal{N}_{\hbar}:=\mathrm{dim}\mathcal{H}_{\hbar}$. Let
us consider the following distribution (for fixed $\hbar$) in $\theta$
on the circle $S^{1}$:
\[
s_{\hbar}:=\frac{1}{\mathcal{N}_{\hbar}}\sum_{j=1}^{\mathcal{N}_{\hbar}}\delta\left(\theta-\theta_{j}\right)
\]
We want to show that $s_{\hbar}$ converges (weakly) to the uniform
probability measure on $S^{1}$ in the limit $\hbar\rightarrow0$.
This is equivalent to show that, for every fixed $n\in\mathbb{Z}$,
\begin{eqnarray}
\langle e^{in\theta},s_{\hbar}\rangle & \underset{\hbar\rightarrow0}{\longrightarrow}\begin{cases}
1, & \mbox{if }n=0;\\
0, & \mbox{otherwise. }
\end{cases}\label{eq:distribution_arguments}
\end{eqnarray}

Let us show now (\ref{eq:distribution_arguments}). If $n=0$, we
have simply $\langle1,s_{\hbar}\rangle=\frac{1}{\mathcal{N}_{\hbar}}\mathcal{N}_{\hbar}=1$.
Suppose $n>0$, since $\langle e^{-in\theta},s_{\hbar}\rangle=\overline{\langle e^{in\theta},s_{\hbar}\rangle}$).
Let $r=e^{\left\langle V-V_{0}\right\rangle }$. We write
\[
\langle e^{in\theta},s_{\hbar}\rangle=\frac{1}{\mathcal{N}_{\hbar}}\sum_{j=1}^{\mathcal{N}_{\hbar}}e^{in\theta_{j}}=\frac{1}{\mathcal{N}_{\hbar}}\sum_{j=1}^{\mathcal{N}_{\hbar}}\left(1-\frac{\rho_{j}}{r}\right)e^{in\theta_{j}}+\frac{1}{\mathcal{N}_{\hbar}}\sum_{j=1}^{\mathcal{N}_{\hbar}}\frac{\rho_{j}}{r}e^{in\theta_{j}}
\]
Since we have 
\[
\sum_{j=1}^{\mathcal{N}_{\hbar}}\rho_{j}e^{in\theta_{j}}=\mathrm{Tr}\hat{\mathcal{F}}_{\hbar}^{n},
\]
we see
\[
\left|\langle e^{in\theta},s_{\hbar}\rangle\right|\leq\frac{1}{\mathcal{N}_{\hbar}}\sum_{j=1}^{\mathcal{N}_{\hbar}}\left|1-\frac{\rho_{j}}{r}\right|+\frac{1}{r\mathcal{N}_{\hbar}}\left|\mathrm{Tr}\hat{\mathcal{F}}_{\hbar}^{n}\right|.
\]
From the accumulation of the moduli of eigenvalues to $r$, proved
in the last section, we have 
\[
\frac{1}{\mathcal{N}_{\hbar}}\sum_{j=1}^{\mathcal{N}_{\hbar}}\left|1-\frac{\rho_{j}}{r}\right|\underset{\hbar\rightarrow0}{\longrightarrow}0.
\]
Therefore it is enough to show, for each fixed $n>0$, that 
\begin{equation}
\frac{1}{\mathcal{N}_{\hbar}}\left|\mathrm{Tr}\hat{\mathcal{F}}_{\hbar}^{n}\right|\underset{\hbar\rightarrow0}{\longrightarrow}0.\label{eq:limit_trace}
\end{equation}
From the results given in Subsection \ref{sub:1.7Semiclassical-calculus-on},
we see that the trace of $\hat{\mathcal{F}}_{\hbar}^{n}=\Pi_{\hbar}\circ\hat{F}_{\hbar}^{n}\circ\Pi_{\hbar}$
is expressed as the sum of contributions from its restriction to neighborhoods
of the finite number of fixed points $x=f^{n}\left(x\right)$ of $f^{n}$
and obtain (\ref{eq:limit_trace}). We finished the proof of Theorem
\ref{Thm:Distribution-of-resonances.}.

\section{\label{sec:10}Proof of Th. \ref{Th1.5}. Gutzwiller trace formula. }

\subsection{The Atiyah-Bott trace formula}

In this subsection, we recall the Atiyah-Bott trace formula in a general
setting\cite[cor.5.4,p.393]{atiyah_67}:

%red 
\begin{center}{\color{red}\fbox{\color{black}\parbox{16cm}{
\begin{defn}
\label{Def:AB_trace}``\textbf{Flat Trace of a transfer operator}''.
Suppose that $f:M\rightarrow M$ is a smooth diffeomorphism on a manifold
$M$ whose periodic points are all hyperbolic, that $\pi:E\rightarrow M$
is a vector bundle and that $B:E\rightarrow E$ a vector bundle map
projecting on $f$, i.e. such that the following diagram commutes:
\begin{eqnarray*}
E & \overset{B}{\longrightarrow} & E\\
\pi\downarrow &  & \downarrow\pi\\
M & \overset{f}{\longrightarrow} & M
\end{eqnarray*}
We can define the associated \textbf{transfer operator }acting on
smooth sections of this vector bundle
\begin{equation}
\hat{F}:C^{\infty}\left(M,E\right)\rightarrow C^{\infty}\left(M,E\right),\qquad\left(\hat{F}u\right)\left(x\right):=B_{f^{-1}\left(x\right)}\left(u\left(f^{-1}\left(x\right)\right)\right)\label{eq:def_general_transfer_op_bundle}
\end{equation}
For any $n\geq1$, the \textbf{flat trace of the transfer operator}
$\hat{F}^{n}$ is defined by
\begin{equation}
\mathrm{Tr}^{\flat}\left(\hat{F}^{n}\right):=\int_{M}\mathrm{Tr}\left(K_{n}\left(x,x\right)\right)dx\label{eq:def_Trace_flat}
\end{equation}
where $K_{n}\left(x,y\right)dy\in\mathcal{L}\left(E_{y}\rightarrow E_{x}\right)$
denotes the Schwartz kernel of $\hat{F}^{n}$.
\end{defn}
}}}\end{center}

For $n\geq1$, let $B_{x}^{\left(n\right)}:=\prod_{k=0}^{n-1}B_{f^{k}\left(x\right)}:E_{x}\rightarrow E_{f^{n}\left(x\right)}$.
For a periodic point $x=f^{n}\left(x\right)$ then $B_{x}^{\left(n\right)}$
is an endomorphism on $E_{x}$ and $\mathrm{Tr}\left(B_{x}^{\left(n\right)}\right)$
is well defined and does not depend on $x$ on the orbit.

%blue 
\begin{center}{\color{blue}\fbox{\color{black}\parbox{16cm}{
\begin{lem}
\label{lem:Atiyah-Bott_flat_trace}For any $n\geq1$, the \textbf{Atiyah-Bott
trace formula} reads:
\begin{equation}
\mathrm{Tr}^{\flat}\left(\hat{F}^{n}\right)=\sum_{x=f^{n}\left(x\right)}\frac{\mathrm{Tr}\left(B_{x}^{\left(n\right)}\right)}{\left|\mathrm{det}\left(I-Df_{x}^{n}\right)\right|},\label{eq:Atiyah-Bott}
\end{equation}
This is a finite sum over periodic points of $f$ with period $n$.
\end{lem}
}}}\end{center}
\begin{proof}
(As in \cite[cor.5.4,p.393]{atiyah_67}). From (\ref{eq:def_general_transfer_op_bundle})
we have
\[
\left(\left(\hat{F}^{n}u\right)\left(x\right)\right)_{i}=\left(B_{f^{-n}\left(x\right)}^{\left(n\right)}\left(u\left(f^{-n}\left(x\right)\right)\right)\right)_{i}=\sum_{j}\int_{M}\delta_{f^{-n}\left(x\right)}\left(y\right)\left(B_{y}^{\left(n\right)}\right)_{i,j}u_{j}\left(y\right)dy
\]
where $i,j=1\ldots\mathrm{dim}E$ are indices for components in the
fibers $E_{y}$ with respect to some local trivialization and $\delta_{x}$
is the Dirac distribution at $x$. So the Schwartz kernel of $\hat{F}^{n}$
is

\[
K_{n}\left(x,y\right)=B_{y}^{\left(n\right)}\delta_{f^{-n}\left(x\right)}\left(y\right)
\]
hence
\[
\mathrm{Tr}\left(K_{n}\left(x,x\right)\right)=\delta_{f^{-n}\left(x\right)}\left(x\right)\mathrm{Tr}B_{x}^{\left(n\right)}.
\]
From definition (\ref{eq:def_Trace_flat}),
\begin{eqnarray*}
\mathrm{Tr}^{\flat}\left(\hat{F}^{n}\right) & = & \int_{M}\delta\left(x-f^{-n}\left(x\right)\right)\mathrm{Tr}B_{x}^{\left(n\right)}dx\\
 & = & \sum_{x=f^{n}\left(x\right)}\mathrm{Tr}B_{x}^{\left(n\right)}\frac{1}{\left|\mathrm{det}\left(I-Df_{x}^{-n}\right)\right|}=\sum_{x=f^{n}\left(x\right)}\frac{\mathrm{Tr}B_{x}^{\left(n\right)}}{\left|\mathrm{det}\left(I-Df_{x}^{n}\right)\right|}
\end{eqnarray*}
where in the second line we have used the change of variable $x\rightarrow y=x-f^{-n}\left(x\right)$
in the vicinity of $y=0$. For the last equality we have used that
$\mathrm{det}Df_{x}^{n}=1$ since $f$ preserves the volume form.\end{proof}
\begin{rem}
A standard example of bundle map is the differential map $Df:TM\rightarrow TM$
or extension in tensor bundles as $f_{\left(k\right)}:\Lambda^{k}\left(T^{*}M\right)\rightarrow\Lambda^{k}\left(T^{*}M\right)$
(in antisymmetric tensor bundle of order $k$ that we will use), etc.
\end{rem}

\subsection{The Gutzwiller Trace formula from the Atiyah-Bott trace formula}

In this subsection we show how the Gutzwiller trace formula can be
expressed as a sum of flat traces of transfer operators acting on
differential forms on the Grassmann extension $G$. We first extend
the argument in the last section to the case of smooth diffeomorphisms
on the Grassmann bundle $G$. This is rather trivial. We just replace
the manifold $M$ and the diffeomorphism $f$ by the Grassman bundle
$G$ and the natural extension $f_{G}$ of $f$. Lemma \ref{lem:Atiyah-Bott_flat_trace}
remains true for such extension. Below we will always truncate the
transfer operators on $G$ to a small neighborhood $K_{0}$ of the
attracting section $E_{u}$ which represents the unstable subbundle.
(Recall the argument in Subsection \ref{sub:Truncation-near_Eu}.)
So, on the right hand side of (\ref{eq:Atiyah-Bott}), the sum will
be over only the periodic points of $f_{G}$ contained in $K_{0}$.
Those periodic points must be contained in $E_{u}$ and in one-to-one
correspondence to the periodic points for $f:M\to M$. 

The next lemma shows how the amplitude $\left|\mbox{det}\left(1-Df_{x}^{n}\right)\right|^{-1/2}$
appear and the following corollary shows how the phases $e^{iS_{n,x}/\hbar}$
appear. For $l\in G$, we will denote
\[
V_{l}G:=\left(T_{l}G\right)_{\mathrm{vert}}=\mathrm{ker}Dp_{l},
\]
the vertical part of the tangent space. Then $VG=\cup_{l\in G}V_{l}G$
is a smooth subbundle of $TG$. 
\begin{lem}
\label{lem:Gutz_1}For $0\leq k\leq d$, $0\leq m\leq d^{2}$, let
$\Lambda^{\left(k,m\right)}\rightarrow G$ be the vector bundle with
fiber $\Lambda^{\left(k,m\right)}\left(l\right):=\Lambda^{k}\left(T_{l}G\right)\otimes\Lambda^{m}\left(V_{l}G\right)$
over $l\in G$ and let $F_{k,m}$ be the transfer operator $F_{k,m}:C^{\infty}\left(G,\Lambda^{\left(k,m\right)}\right)\rightarrow C^{\infty}\left(G,\Lambda^{\left(k,m\right)}\right)$
defined by
\begin{equation}
F_{k,m}u(l)=\chi\circ p(f_{G}^{-1}(l)))\cdot\frac{\left|\mathrm{det}\left(Df_{x}\mid l\right)\right|^{1/2}}{\left|\mathrm{det}\left(Df_{G}\mid_{V_{l}G}\left(l\right)\right)\right|}\left(\Lambda^{k}\left(Df_{/l}^{-1}\right)\otimes\Lambda^{m}\left(Df\mid_{V_{l}G}\right)\right)u\left(f_{G}^{-1}\left(l\right)\right)\label{eq:def_F_km}
\end{equation}
where $\chi:G\to[0,1]$ is a smooth function such that $\chi(l)=0$
for $l\notin K_{0}$ and $\chi(l)=1$ for $l\in f_{G}(K_{0})$ and
$p:G\to M$ is the projection. Then the Atiyah-Bott trace formula
gives:
\begin{equation}
\sum_{x=f^{n}\left(x\right)}\frac{1}{\sqrt{\left|\mathrm{det}\left(1-Df_{x}^{n}\right)\right|}}=\sum_{k=0}^{d}\sum_{m=0}^{d^{2}}\left(-1\right)^{k+m}\mathrm{Tr}^{\flat}\left(F_{k,m}^{n}\right)\label{eq:Gutz_f1}
\end{equation}
\end{lem}
\begin{proof}
We first derive an useful expression for the differential of the map
$f_{G}:G\rightarrow G$. For $x\in M$, we have a continuous decomposition
of the symplectic tangent space

\[
T_{x}M=E_{u}\left(x\right)\oplus E_{s}\left(x\right),
\]
where $E_{u}\left(x\right)$ and $E_{s}\left(x\right)$ are linear
Lagrangian space. We will use the fact that the symplectic form $\omega$
provides an isomorphism between these complementary Lagrangian subspaces
\begin{equation}
E_{u}\left(x\right)\overset{\omega}{\equiv}E_{s}^{*}\left(x\right)\label{eq:isomorphism}
\end{equation}
by 
\[
U\in E_{u}\left(x\right)\longleftrightarrow\omega\left(U,.\right)\in E_{s}^{*}\left(x\right)
\]

The symplectic map $f:M\rightarrow M$ gives a symplectic hyperbolic
linear map
\[
Df_{x}:T_{x}M\rightarrow T_{f\left(x\right)}M
\]
which decomposes accordingly into 
\begin{equation}
Df_{x}=L_{u}\left(x\right)\oplus L_{s}\left(x\right)\label{eq:decomp_Df}
\end{equation}
where

\[
L_{u}\left(x\right):=Df_{x}\mid_{E_{u}\left(x\right)}:E_{u}\left(x\right)\rightarrow E_{u}\left(f\left(x\right)\right)
\]

\[
L_{s}\left(x\right):=Df_{x}\mid_{E_{s}\left(x\right)}:E_{s}\left(x\right)\rightarrow E_{s}\left(f\left(x\right)\right)
\]
The map $L_{u}\left(x\right)$ is expanding while $L_{s}\left(x\right)$
is contracting. Similarly in the cotangent bundle (with the usual
convention here such as $E_{u}^{*}\left(E_{s}\right)=0$), we have

\[
^{t}L_{u}^{-1}\left(x\right):={}^{t}Df_{x}^{-1}\mid_{E_{u}^{*}\left(x\right)}:E_{u}^{*}\left(x\right)\rightarrow E_{u}^{*}\left(f\left(x\right)\right)
\]

\[
^{t}L_{s}^{-1}\left(x\right)\equiv^{t}Df_{x}^{-1}\mid_{E_{s}^{*}\left(x\right)}:E_{s}^{*}\left(x\right)\rightarrow E_{s}^{*}\left(f\left(x\right)\right)
\]

For $x\in M$, the graph of a linear map 
\[
\mathcal{L}:E_{u}\left(x\right)\rightarrow E_{s}\left(x\right)
\]
defines an element $l\in G_{x}$. In particular the linear map $\mathcal{L}\equiv0$
is associated to the subspace $E_{u}\left(x\right)\in G_{x}$. In
tensorial notations we have that 
\begin{equation}
\mathcal{L}\in\left(E_{u}^{*}\left(x\right)\otimes E_{s}\left(x\right)\right)\label{eq:L_tensor}
\end{equation}
Conversely, there is a neighborhood of $E_{u}\left(x\right)$ in $G_{x}$
such that every element $l\in G_{x}$ in this neighborhood can be
expressed as the graph of a such linear map $\mathcal{L}$. The map
$f_{G}:G\rightarrow G$ defined in (\ref{eq:def_fG}), is linear in
the fiber and has derivative

\begin{equation}
Df_{G}\mid_{\mathrm{ker}p}\equiv\left(\underbrace{^{t}L_{u}^{-1}\left(x\right)\otimes L_{s}\left(x\right)}_{Df_{G}\mid V_{l}G}\right)\label{eq:e1-1}
\end{equation}

Now let $x=f^{n}\left(x\right)$ be a periodic point. The isomorphism
(\ref{eq:isomorphism})  gives 
\begin{equation}
L\left(x\right):=L_{u}\left(x\right)\overset{\omega}{\equiv}\;^{t}L_{s}^{-1}\left(x\right)\label{eq:Def_Ell}
\end{equation}
We will note
\[
L^{\left(n\right)}\left(x\right):=\prod_{k=0}^{n-1}L\left(f^{k}\left(x\right)\right),\qquad L^{-\left(n\right)}\left(x\right):=\left(L^{\left(n\right)}\left(x\right)\right)^{-1}.
\]

We first express the ``Gutzwiller amplitude'' using (\ref{eq:decomp_Df})
and (\ref{eq:Def_Ell}) as follows:
\begin{eqnarray}
\mathcal{B}_{x}: & = & \frac{1}{\sqrt{\left|\mathrm{det}\left(1-Df_{x}^{n}\right)\right|}}=\frac{1}{\left|\mathrm{det}\left(1-L_{u}^{\left(n\right)}\left(x\right)\right)\right|^{1/2}\left|\mathrm{det}\left(1-L_{s}^{\left(n\right)}\left(x\right)\right)\right|^{1/2}}\label{eq:Bx}\\
 & = & \frac{\left|\mathrm{det}\left(L_{u}^{\left(n\right)}\left(x\right)\right)\right|^{-1/2}}{\left|\mathrm{det}\left(1-L_{u}^{-\left(n\right)}\left(x\right)\right)\right|^{1/2}\left|\mathrm{det}\left(1-L_{s}^{\left(n\right)}\left(x\right)\right)\right|^{1/2}}\nonumber \\
 & = & \left|\mathrm{det}\left(L^{-\left(n\right)}\right)\right|^{1/2}\cdot\frac{1}{\left|\mathrm{det}\left(1-L^{-\left(n\right)}\right)\right|}\nonumber 
\end{eqnarray}
Notice that the last expression is presented in such a way that the
first factor is a multiplicative cocycle and the second one converges
to $1$ as $n\rightarrow+\infty$. Let us consider the map $f_{G}:G\rightarrow G$
and the transfer operator (\ref{eq:def_F_km}) with $\left(k,m\right)=\left(0,0\right)$:
\[
\left(F_{0,0}u\right)=\left|\mathrm{det}\left(Df_{x}\left(l\right)\right)\right|^{1/2}\cdot\left|\mathrm{det}\left(Df_{G}\mid_{V_{l}G}\left(l\right)\right)\right|\cdot u\circ f_{G}^{-1}
\]
The Atiyah-Bott trace formula reads
\begin{equation}
\mathrm{Tr}^{\flat}\left(F_{0,0}^{n}\right)=\sum_{x=f^{n}\left(x\right)}\mathcal{A}_{x}\label{eq:f1}
\end{equation}
with the amplitude
\[
\mathcal{A}_{x}:=\left(\frac{\left|\mathrm{det}\left(Df_{x}^{n}\left(E_{u}\left(x\right)\right)\right)\right|^{1/2}}{\left|\mathrm{det}\left(Df_{G}^{n}\left(E_{u}\left(x\right)\right)\mid_{\mathrm{ker}p}\right)\right|}\right)\frac{1}{\left|\mathrm{det}\left(1-\left(D\tilde{f}_{G}^{-n}\right)\left(E_{u}\left(x\right)\right)\right)\right|}
\]

From (\ref{eq:e1-1}) and and (\ref{eq:Def_Ell}), we can express
$\mathcal{A}_{x}$ as
\begin{eqnarray*}
\mathcal{A}_{x} & = & \left(\left|\mathrm{det}\left(1-L_{u}^{-\left(n\right)}\right)\right|\left|\mathrm{det}\left(1-L_{s}^{-\left(n\right)}\right)\right|\left|\mathrm{det}\left(1-^{t}L_{u}^{\left(n\right)}\otimes L_{s}^{-\left(n\right)}\right)\right|\right)^{-1}\left(\frac{\left(\mathrm{det}L_{u}^{\left(n\right)}\right)^{1/2}}{\left|\mathrm{det}\left(^{t}L_{u}^{-\left(n\right)}\otimes L_{s}^{\left(n\right)}\right)\right|}\right)\\
 & = & \left(\left|\mathrm{det}\left(1-L^{-\left(n\right)}\right)\right|\left|\mathrm{det}\left(1-L^{\left(n\right)}\right)\right|\left|\mathrm{det}\left(1-L^{\left(n\right)}\otimes L^{\left(n\right)}\right)\right|\right)^{-1}\left(\frac{\left(\mathrm{det}L^{\left(n\right)}\right)^{1/2}}{\left|\mathrm{det}\left(L^{-\left(n\right)}\otimes L^{-\left(n\right)}\right)\right|}\right)\\
 & = & \left|\mathrm{det}\left(L^{-\left(n\right)}\right)\right|^{1/2}\left|\mathrm{det}\left(1-L^{-\left(n\right)}\right)\right|^{-2}\left|\mathrm{det}\left(1-L^{-\left(n\right)}\otimes L^{-\left(n\right)}\right)\right|^{-1}
\end{eqnarray*}
Comparing the expression of $\mathcal{A}_{x}$ and $\mathcal{B}_{x}$
we see that they have the same leading behavior for $n\rightarrow+\infty$.
More precisely:
\[
\mathcal{B}_{x}=\mathcal{A}_{x}\left|\mathrm{det}\left(1-L^{-\left(n\right)}\right)\right|\left|\mathrm{det}\left(1-L^{-\left(n\right)}\otimes L^{-\left(n\right)}\right)\right|
\]
Recall that we aim to express $\mathcal{B}_{x}$ as a sum of ''amplitudes''
related to some transfer operators. By construction, the amplitude
$\mathcal{A}_{x}$ already comes from a transfer operator. Therefore
we would like to express the remaining factor $\left|\mathrm{det}\left(1-L^{-\left(n\right)}\right)\right|\left|\mathrm{det}\left(1-L^{-\left(n\right)}\otimes L^{-\left(n\right)}\right)\right|$
as a sum of ``cocyles'' related to some potentials as in (\ref{eq:Atiyah-Bott}).
For this purpose we use some relations of linear algebra \cite[p.396]{reed-simon4}:
if $M:\mathbb{R}^{m}\rightarrow\mathbb{R}^{m}$ is a linear endomorphism
and $\Lambda^{k}\left(M\right)$ denotes its natural action on the
antisymmetric tensor algebra $\Lambda^{k}\left(\mathbb{R}^{m}\right)$
(with $k\in\left\{ 0,\ldots,m\right\} $) then
\[
\mathrm{det}\left(1+M\right)=\sum_{k=0}^{m}\mathrm{Tr}\left(\Lambda^{k}\left(M\right)\right),\qquad\Lambda^{k}\left(-M\right)=\left(-1\right)^{k}\Lambda^{k}\left(M\right),\quad\Lambda^{k}\left(M^{n}\right)=\left(\Lambda^{k}\left(M\right)\right)^{n}.
\]
Also for two linear endomorphisms $M_{1},M_{2}$ one has $\mathrm{Tr}\left(M_{1}\otimes M_{2}\right)=\mathrm{Tr}\left(M_{1}\right)\mathrm{Tr}\left(M_{2}\right)$
and $M_{1}^{n}\otimes M_{2}^{n}=\left(M_{1}\otimes M_{2}\right)^{n}$.
Using this, we get:
\begin{eqnarray*}
\left|\mathrm{det}\left(1-L^{-\left(n\right)}\right)\right|\left|\mathrm{det}\left(1-L^{-\left(n\right)}\otimes L^{-\left(n\right)}\right)\right|=\left(\sum_{k=0}^{d}\left(-1\right)^{k}\mathrm{Tr}\left(\Lambda^{k}\left(L^{-\left(n\right)}\right)\right)\right)\\
\left(\sum_{m=0}^{d^{2}}\left(-1\right)^{m}\mathrm{Tr}\left(\Lambda^{m}\left(L^{-\left(n\right)}\otimes L^{-\left(n\right)}\right)\right)\right)\\
=\sum_{k=0}^{d}\sum_{m=0}^{d^{2}}\left(-1\right)^{k+m}\mathrm{Tr}\left(\Lambda^{k}\left(L^{-\left(n\right)}\right)\otimes\Lambda^{m}\left(L^{-\left(n\right)}\otimes L^{-\left(n\right)}\right)\right)
\end{eqnarray*}

And therefore
\begin{eqnarray*}
\sum_{x=f^{n}\left(x\right)}\mathcal{B}_{x} & = & \sum_{k=0}^{d}\sum_{m=0}^{d^{2}}\left(-1\right)^{k+m}\left(\sum_{x=f^{n}\left(x\right)}\mathcal{A}_{x}\mathrm{Tr}\left(\Lambda^{k}\left(L^{-\left(n\right)}\right)\otimes\Lambda^{m}\left(L^{-\left(n\right)}\otimes L^{-\left(n\right)}\right)\right)\right)
\end{eqnarray*}
On the other hand, in the same way we derived (\ref{eq:f1}), we check
that the Atiyah-Bott trace formula of (\ref{eq:def_F_km}) reads:
\[
\mathrm{Tr}^{\flat}\left(F_{k,m}^{n}\right)=\sum_{x=f^{n}\left(x\right)}\mathcal{A}_{x}\mathrm{Tr}\left(\Lambda^{k}\left(L^{-\left(n\right)}\right)\otimes\Lambda^{m}\left(L^{-\left(n\right)}\otimes L^{-\left(n\right)}\right)\right)
\]
We have obtained (\ref{eq:Gutz_f1}) and finished the proof of Lemma
\ref{lem:Gutz_1}.
\end{proof}
The next corollary extends Lemma \ref{lem:Gutz_1} to the prequantum
transfer operators with an arbitrary potential function $\tilde{V}$
restricted to the $N$-th Fourier mode. But notice that, in the statement
below, we define the Attiyah-Bott trace of the prequantum transfer
operators (restricted to the $N$-th Fourier mode) by using their
expressions in local charts, rather than applying Definition \ref{Def:AB_trace}
naively. (See Remark \ref{Rem:ABTrace_prequantum_case} below.) 
\begin{cor}
\label{cor:gutz}Let $\tilde{V}\in C^{\infty}\left(G\right)$. Let
$\Lambda^{\left(k,m\right)}\rightarrow P_{G}$ be the bundle $\Lambda^{\left(k,m\right)}\rightarrow G$
pulled back by $\pi_{G}:P_{G}\rightarrow G$. Let $\widetilde{F}_{\tilde{V},k,m}:C^{\infty}\left(P_{G},\Lambda^{\left(k,m\right)}\right)\rightarrow C^{\infty}\left(P_{G},\Lambda^{\left(k,m\right)}\right)$
be the (vector-valued) prequantum transfer operator defined by 
\begin{equation}
\widetilde{F}_{\tilde{V},k,m}u(p)=\chi\circ p(f_{G}^{-1}(l)))\cdot\frac{e^{\tilde{V}(l)}}{\left|\mathrm{det}\left(D\tilde{f}_{G}\mid_{V_{l}G}\left(l\right)\right)\right|}\left(\Lambda^{k}\left(Df_{/l}^{-1}\right)\otimes\Lambda^{m}\left(D\tilde{f}\mid_{V_{l}G}\right)\right)u\left(\tilde{f}_{G}^{-1}\left(p\right)\right)\label{eq:def_F_km-1}
\end{equation}
where we set $l=\pi_{G}(p)$. Let 
\begin{equation}
\left(\widetilde{F}_{\tilde{V},k,m}\right)_{N}:C_{N}^{\infty}\left(P_{G},\Lambda^{\left(k,m\right)}\right)\rightarrow C_{N}^{\infty}\left(P_{G},\Lambda^{\left(k,m\right)}\right)\label{eq:def_F_Nkm}
\end{equation}
be its restriction to $N$-equivariant functions. Then 
\begin{equation}
\sum_{x=f^{n}\left(x\right)}\frac{e^{\left(V-V_{0}\right)_{n}\left(x\right)}e^{iS_{n,x}/\hbar}}{\sqrt{\left|\mathrm{det}\left(1-Df_{x}^{n}\right)\right|}}=\sum_{k=0}^{d}\sum_{m=0}^{d^{2}}\left(-1\right)^{k+m}\mathrm{Tr}^{\flat}\left(\widetilde{F}_{\tilde{V},k,m}\right)_{N}^{n}\label{eq:r1}
\end{equation}
where $V_{0}\left(x\right)=\frac{1}{2}\log\left(\mathrm{det}Df_{x}\mid_{E_{u}\left(x\right)}\right)$
is the potential of reference.\end{cor}
\begin{rem}
\label{Rem:ABTrace_prequantum_case} As we noted above, we define
the Attiyah-Bott trace of prequantum transfer operators by using its
local expression. (If we adopt the ``line bundle termonology'' refered
in Remark \ref{Rem:line_bundle_terminology}, this definition coincides
with the Attiyah-Bott trace in Definition \ref{Def:AB_trace} with
$E=L^{\otimes N}$.) For simplicity, let us consider the case $\widetilde{F}_{\widetilde{V},0,0}$
which acts on functions on $P_{G}$ and let $\left(\widetilde{F}_{\widetilde{V},0,0}\right)_{N}$
$ $be its restriction to the $N$-th Fourier mode. As we observed
in Proposition \ref{prop:bfF} in Section \ref{sec:2}, this operator
$\left(\widetilde{F}_{\widetilde{V},0,0}\right)_{N}$ is lifted to
a matrix of operators, denoted by $\mathbf{\widetilde{F}}_{\hbar}$.
Each component $\widetilde{\mathbf{F}}_{j,i}$ of $\widetilde{\mathbf{F}}_{\hbar}$
are simple transfer operators on Euclidean trace and the Attiyah-Bott
trace is defined by Definition \ref{Def:AB_trace}. We define the
Attiyah-Bott trace of $\left(\widetilde{F}_{\widetilde{V},0,0}\right)_{N}$
as
\[
\mathrm{Tr}^{\flat}\left(\widetilde{F}_{\widetilde{V},0,0}\right)_{N}:=\sum_{i=1}^{I_{\hbar}}\mathrm{Tr}^{\flat}\left(\widetilde{\mathbf{F}}_{i,i}\right)
\]
taking the sum of the flat traces of the diagonal elements. This definition
actually does not depend on the choice of local trivializations. Also
the extension to the case of vecor-valued transfer operators such
as $\left(\widetilde{F}_{\tilde{V},k,m}\right)_{N}$ is straightforward. \end{rem}
\begin{proof}[Proof of Corollary \ref{cor:gutz}]
Observe that we can go through the proof of Lemma \ref{lem:Gutz_1}
by considering the expressions of the transfer operators on local
charts, as we noted in the remark above. Then, recalling Remark \ref{Rem:Action_of_PerOrb},
we see that the action $e^{iS_{n,x}/\hbar}$ in the left hand side
of (\ref{eq:r1}) appears as a consequence of the rotation on the
fiber that the prequantum map $\tilde{f}_{G}^{n}$ induces. 
\end{proof}

\subsection{Restriction to the external band}

We will now relate the previous quite formal ``flat trace formula''
(\ref{eq:r1}) with the spectrum of the transfer operators. This relation
is given by the following result obtained in \cite{Baladi-Tsujii08}. 

%blue
\begin{center}{\color{blue}\fbox{\color{black}\parbox{16cm}{
\begin{thm}
\textbf{\label{thm:Flat-trace-and_spectrum}``Flat trace and spectrum''}.
Let $f:M\to M$ be a smooth Anosov diffeomorphism. Suppose that the
transfer operator $\hat{F}$ is defined from the general setting (\ref{eq:def_general_transfer_op_bundle}).
Let $\left(\lambda_{j}\right)_{j\in\mathbb{N}}$ denote its Ruelle
discrete spectrum. Then, for any $\varepsilon>0$, there exists $C_{\varepsilon}>0$
such that for any $n>0$,
\begin{equation}
\left|\mathrm{Tr}^{\flat}\hat{F}^{n}-\sum_{j,\left|\lambda_{j}\right|\geq\varepsilon}\lambda_{j}^{n}\right|\leq C_{\varepsilon}\varepsilon^{n}.\label{eq:trace_flat_and_spectrum}
\end{equation}

\end{thm}
}}}\end{center}

In our case, we have a family of operators $\left(\hat{F}_{\tilde{V},k,m}\right)_{N}$
depending on the semiclassical parameter $N=1/\left(2\pi\hbar\right)$
and we want to get a result similar to (\ref{eq:trace_flat_and_spectrum})
but with a control of the remainder uniformly with respect to $N$.
This is the purpose of the next lemma which concerns the operator
$\left(\hat{F}_{\tilde{V},k,m}\right)_{N}$ defined in (\ref{eq:def_F_km-1})
but with the particular value $\left(k,m\right)=\left(0,0\right)$.
It coincides with the transfer operator $\widetilde{F}_{N}$ defined
in (\ref{eq:def_F_tilde_N}):
\[
\left(\hat{F}_{\tilde{V},0,0}\right)_{N}=\widetilde{F}_{N}
\]
 Recall that the quantum operator $\widetilde{\mathcal{F}}_{\hbar}:\mathcal{H}_{\hbar}\rightarrow\mathcal{H}_{\hbar}$
has been defined from $\widetilde{F}_{N}$ in (\ref{eq:def_quantum_op})
as its spectral restriction to the external band. $\widetilde{\mathcal{F}}_{\hbar}$
is finite rank so its trace $\mathrm{Tr}\left(\widetilde{\mathcal{F}}_{\hbar}^{n}\right)$
below is well defined.
\begin{lem}
\label{lem:trace_flat_and_trace}For any $\varepsilon>0$, there exists
$C_{\varepsilon}$ and $N_{\varepsilon}$ such that, for any $N>N_{\varepsilon}$
and any $n>0$, we have 
\begin{equation}
\left|\mathrm{Tr}^{\flat}\left(\widetilde{F}_{N}^{n}\right)-\mathrm{Tr}\left(\widetilde{\mathcal{F}}_{\hbar}^{n}\right)\right|\leq C_{\varepsilon}N^{d}\left(r_{1}^{+}+\varepsilon\right)^{n}\label{eq:r3}
\end{equation}
\end{lem}
\begin{proof}
We adapt the methods presented in \cite{Baladi-Tsujii08} to our settings.
Let $\Pi_{0}:\mathcal{H}_{N}^{r}\left(P_{G}\right)\rightarrow\mathcal{H}_{\hbar}$
denotes the finite rank spectral projector on the external band of
$\widetilde{F}_{N}$ defined for $N$ large enough. Note the it commutes
with $\widetilde{F}_{N}$ by definition: $\left[\Pi_{0},\widetilde{F}_{N}\right]=0$.
Let 
\[
\widetilde{\mathcal{F}}_{\hbar}:=\Pi_{0}\widetilde{F}_{N}\Pi_{0}\quad\mbox{and}\quad\widetilde{\mathcal{R}}_{\hbar}:=\left(1-\Pi_{0}\right)\widetilde{F}_{N}\left(1-\Pi_{0}\right).
\]
Then, for any $n\geq1$, we have
\begin{equation}
\widetilde{F}_{N}^{n}=\widetilde{\mathcal{F}}_{\hbar}^{n}+\widetilde{\mathcal{R}}_{\hbar}^{n}.\label{eq:decomp_spectral}
\end{equation}
and hence
\begin{equation}
\mathrm{Tr}^{\flat}\left(\widetilde{F}_{N}^{n}\right)=\mathrm{Tr}\left(\widetilde{\mathcal{F}}_{\hbar}^{n}\right)+\mathrm{Tr}^{\flat}\left(\widetilde{\mathcal{R}}_{\hbar}^{n}\right).\label{eq:Trace_decomposed}
\end{equation}

We first recall some estimates related to the decomposition (\ref{eq:decomp_spectral}).
From (\ref{eq:weyl_law}) we have 
\begin{equation}
\left\Vert \widetilde{\mathcal{F}}_{\hbar}\right\Vert _{\mathrm{Tr}}\leq CN^{d}\label{eq:Tr_Fh}
\end{equation}
with $C$ independent on $N$. From (\ref{eq:bound_F1^n}), we have,
for every $N>N_{\varepsilon}$ with some $N_{\varepsilon}$ large
enough: 
\begin{equation}
\left\Vert \widetilde{\mathcal{R}}_{\hbar}^{n}\right\Vert \leq C_{\varepsilon}\left(r_{1}^{+}+\varepsilon\right)^{n}\label{eq:norm_F1}
\end{equation}
with $C_{\varepsilon}$ independent on $N$ and $n$. The following
Lemma is central in the argument, whose proof is postpone for a moment. 
\begin{lem}
\label{lem:11.8}There exists $C_{\epsilon}>0$ and $N_{\epsilon}>0$such
that for any $N>N_{\epsilon}$ and any $n>0$
\begin{equation}
\left|\mathrm{Tr}^{\flat}\left(\widetilde{\mathcal{R}}_{\hbar}^{n}\right)\right|\leq C_{\epsilon}\cdot N^{d}\left(r_{1}^{+}+\varepsilon\right)^{n}\label{eq:bound_of_TrF1n}
\end{equation}

\end{lem}
From (\ref{eq:bound_of_TrF1n}) and (\ref{eq:Trace_decomposed}),
we obtain (\ref{eq:r3}), finishing the proof of Lemma \ref{lem:trace_flat_and_trace}.
\end{proof}
~
\begin{proof}[Proof of Lemma \ref{lem:11.8}]
 The proof is obtained following the strategy presented in \cite{Baladi-Tsujii08}.
Here we explain how the uniform estimate on the remainder term with
respect to the semi-classical parameter $N$ (or $\hbar$) is obtained.
In the argument below, we discuss about transfer operators on local
charts and local trivializations. (Remind that we defined the Attiyah-Bott
trace for prequantum transfer operators using local charts. See Remark
\ref{Rem:ABTrace_prequantum_case}.) But, for simplicity, we still
write the operators on local charts by $\widetilde{F}_{N}$ abusively.
(Moreover we will confuse the objects on local coordinates and the
coresponding globel objects.) It is not difficult to put the following
schematic argument into rigorous one. We refer the paper \cite{Baladi-Tsujii08}
for the detail. 

We consider the lifted operator on the phase space, $\widetilde{F}_{N}^{\mathrm{lift}}:=\mathcal{B}_{(x,s)}\circ\widetilde{F}_{N}\circ\mathcal{B}_{(x,s)}^{*}$,
and decompose it into two parts: 
\begin{equation}
\widetilde{F}_{N}^{\mathrm{lift}}=F_{\mathrm{trace-free}}+F_{\mathrm{trace}}\label{eq:decomp_trace_free}
\end{equation}
so that we get properties given in Lemma \ref{lem:The-operators-F_trace}
below. Note that the decomposition (\ref{eq:decomp_trace_free}) is
not a spectral decomposition, like (\ref{eq:decomp_spectral}), but
is obtained from a ``phase space decomposition''. To define it,
recall the weight (or escape) function $\mathcal{W}_{\hbar}^{r}\left(x,s,\xi_{x},\xi_{s}\right)$
defined in (\ref{eq:def_W_W+}). For simplicity we write $z=\left(x,s,\xi_{x},\xi_{s}\right)\in T^{*}M$.
Let $K_{F}\left(z',z\right)$ be the Schwartz kernel of the lifted
operator on the phase space $\widetilde{F}_{N}^{\mathrm{lift}}=\mathcal{B}_{(x,s)}\circ\widetilde{F}_{N}\circ\mathcal{B}_{(x,s)}^{*}$,
so that 
\[
\left(\widetilde{F}_{N}^{\mathrm{lift}}u\right)\left(z'\right)=\int K_{F}\left(z',z\right)u\left(z\right)dz.
\]
The main property of the escape function $\mathcal{W}_{\hbar}^{r}\left(z\right)$
(which we have already made use of ) is that there exists a compact
neighborhood 
\[
U=\left\{ z=\left(x,\xi\right)\in T^{*}M,\quad\left|\zeta\right|\leq C\sqrt{\hbar}\right\} 
\]
of the ``trapped set'' $\widetilde{K}=\{z\in T^{*}G\mid\zeta=0,s=0\}$
such that 
\[
\left(z',z\right)\notin\left(U\times U\right)\Rightarrow\left|\frac{\mathcal{W}_{\hbar}^{r}\left(z'\right)}{\mathcal{W}_{\hbar}^{r}\left(z\right)}K_{F}\left(z',z\right)\right|\leq\lambda^{-r}
\]
where $\zeta$ is the coordinate introduced in (\ref{eq:coordinates_extended})
(and Proposition \ref{prop:Normal-coordinates.}).

Let $X\subset T^{*}M\times T^{*}M$ be the subset defined by
\begin{equation}
X:=\left\{ \left(z',z\right)\in T^{*}M\times T^{*}M,\quad\mathcal{W}_{\hbar}^{r}\left(z'\right)\leq(\lambda/2)^{-r}\cdot\mathcal{W}_{\hbar}^{r}\left(z\right)\right\} \label{eq:def_X}
\end{equation}
and let $\mathbb{\mathbf{1}}_{X}\left(z',z\right)$ be the characteristic
function of the set $X$. We define the operator $F_{\mathrm{trace-free}}$
by its integral expression:
\[
\left(F_{\mathrm{trace-free}}u\right)\left(z'\right):=\int\mathbb{\mathbf{1}}_{X}\left(z',z\right)K_{F}\left(z',z\right)u\left(z\right)dz
\]
In other words, we define $F_{\mathrm{trace-free}}$ and $F_{\mathrm{trace}}$
as operators with the Schwartz kernel respectively $\mathbb{\mathbf{1}}_{X}K_{F}$
and $\left(1-\mathbb{\mathbf{1}}_{X}\right)K_{F}$.
\begin{lem}
\label{lem:The-operators-F_trace}The operator $F_{\mathrm{trace-free}}$
satisfies
\begin{equation}
\mathrm{Tr}^{\flat}\left(F_{\mathrm{trace-free}}^{n}\right)=0,\quad\mbox{and }\label{eq:eqF_tfn}
\end{equation}
\begin{equation}
\left\Vert F_{\mathrm{trace-free}}^{n}\right\Vert \leq C(\lambda/2)^{-rn}\cdot\left(\sup_{l\in K_{0}}\left(e^{\widetilde{V}(l)}|\det(df_{p(l)}|_{l})|^{1/2}\cdot|\det df_{G}^{n}|_{VG(l)}|^{-1/2}\right)\right)\label{eq:norm_F_trace-free}
\end{equation}
for any $n\ge1$, with some $C>0$ independent on $n$ or $N$. The
operator $F_{\mathrm{trace}}$ is of trace class and we have 
\begin{equation}
\left\Vert F_{\mathrm{trace}}\right\Vert _{\mathrm{Tr}}\leq CN^{d},\label{eq:Trace_F_trace}
\end{equation}
with some $C>0$ independent on $N$.\end{lem}
\begin{proof}
From the definition (\ref{eq:def_X}) of the subset $X$, it is clear
that for any sequence of points $z_{0},z_{1},\ldots z_{n}\in T^{*}M$
with $z_{0}=z_{n}$, we have 
\[
\mathbb{\mathbf{1}}_{X}\left(z_{n}=z_{0},z_{n-1}\right)\cdot\mathbb{\mathbf{1}}_{X}\left(z_{n-1},z_{n-2}\right)\cdots\mathbb{\mathbf{1}}_{X}\left(z_{1},z_{0}\right)=0.
\]
This implies that the Schwartz kernel of $F_{\mathrm{trace-free}}^{n}$
vanishes on the diagonal and hence we have $\mathrm{Tr}^{\flat}\left(F_{\mathrm{trace-free}}^{n}\right)=0$.
For the second claim (\ref{eq:norm_F_trace-free}), we first observe
that, if we consider the action of this operator $F_{\mathrm{trace-free}}^{n}$
with respect to the $L^{2}$-norm (without the escape function $\mathcal{W}_{\hbar}^{r}(\cdot)$),
the operator norm is bounded by the right hand side of (\ref{eq:norm_F_trace-free})
without the term $C(\lambda/2)^{-rn}$. (This is because this is true
for operator $\widetilde{F}_{N}$ and that $\mathcal{B}_{(x,s)}$
and $\mathcal{B}_{(x,s)}^{*}$ does not increase the $L^{2}$ norm.)
Then, taking the escape function $\mathcal{W}_{\hbar}^{r}(\cdot)$
into account and noting the definition of the subset $X$, we retain%
\footnote{For a rigorus proof of this, we can employ an argument similar to
that in the proof of Proposition \ref{pp:bdd_g}, using partition
of unity on the phase space. For the detail we refer \cite{Baladi-Tsujii08}.%
} the factor $C(\lambda/2)^{-rn}$ and obtain (\ref{eq:norm_F_trace-free}).
For (\ref{eq:Trace_F_trace}), we observe that 
\[
\left\Vert F_{\mathrm{trace}}\right\Vert _{\mathrm{Tr}}\leq\int\left|\left(z-\mathbb{\mathbf{1}}_{X}\left(z\right)\right)K_{F}\left(z,z\right)\right|dz
\]
 and then show that, for any $\nu>0$, there exists a constant $C_{\nu}>0$,
which may depend on $f_{G}$ but uniform for sufficiently large $N$,
such that 
\[
|K_{F}(z,z)|\le C_{\nu}\cdot\langle\hbar^{-1/2}\zeta\rangle^{-\nu}.
\]
The last estimate is obtained by expressing the kernel $K_{F}(\cdot)$
as an (oscillatory) integral and then by applying integration by parts.
(See the proofs of Proposition \ref{lm:L_g_Y_almost_identity} and
Proposition \ref{pp:bdd_g} for similar estimates.) Then we obtain
(\ref{eq:Trace_F_trace}).
\end{proof}
We pursue the proof of Lemma \ref{lem:11.8}. Using (\ref{eq:decomp_spectral}),
we write
\[
\mathrm{Tr}^{\flat}\left(\widetilde{F}_{N}^{n}\right)=\mathrm{Tr}\left(\widetilde{\mathcal{F}}_{\hbar}^{n}\right)+\mathrm{Tr}^{\flat}\left(\widetilde{\mathcal{R}}_{\hbar}^{n}\right)
\]
Then (\ref{eq:decomp_trace_free}) gives $\widetilde{\mathcal{R}}_{\hbar}^{\mathrm{lift}}=F_{\mathrm{trace-free}}+\left(F_{\mathrm{trace}}-\widetilde{\mathcal{F}}_{\hbar}^{\mathrm{lift}}\right)$.
We develop accordingly
\[
\left(\widetilde{\mathcal{R}}_{\hbar}^{\mathrm{lift}}\right)^{n}=F_{\mathrm{trace-free}}^{n}+\sum_{k=0}^{n-1}F_{\mathrm{trace-free}}^{k}\left(F_{\mathrm{trace}}-\widetilde{\mathcal{F}}_{\hbar}^{\mathrm{lift}}\right)\left(\widetilde{\mathcal{R}}_{\hbar}^{\mathrm{lift}}\right)^{n-k-1}
\]
From (\ref{eq:eqF_tfn}), we know $\mathrm{Tr}^{\flat}\left(F_{\mathrm{trace-free}}^{n}\right)=0$
for the first term. The operator $\left(F_{\mathrm{trace}}-\widetilde{\mathcal{F}}_{\hbar}^{\mathrm{lift}}\right)$
is in the trace class and using the general fact $\left\Vert AB\right\Vert _{\mathrm{Tr}}\leq\left\Vert A\right\Vert \cdot\left\Vert B\right\Vert _{\mathrm{Tr}}$,
we obtain 
\[
\left|\mathrm{Tr}^{\flat}\left(\widetilde{\mathcal{R}}_{\hbar}^{n}\right)\right|\leq\sum_{k=0}^{n-1}\left\Vert F_{\mathrm{trace-free}}^{k}\right\Vert \left\Vert F_{\mathrm{trace}}-\widetilde{\mathcal{F}}_{\hbar}^{\mathrm{lift}}\right\Vert _{\mathrm{Tr}}\left\Vert \widetilde{\mathcal{R}}_{\hbar}\right\Vert ^{n-k-1}.
\]
From (\ref{eq:Trace_F_trace}) and (\ref{eq:Tr_Fh}), we have
\[
\left\Vert F_{\mathrm{trace}}-\widetilde{\mathcal{F}}_{\hbar}^{\mathrm{lift}}\right\Vert _{\mathrm{Tr}}\leq\left\Vert F_{\mathrm{trace}}\right\Vert _{\mathrm{Tr}}+\left\Vert \widetilde{\mathcal{F}}_{\hbar}\right\Vert _{\mathrm{Tr}}\leq CN^{d}.
\]
By taking large $r$, we may and do assume that the right hand side
of (\ref{eq:norm_F_trace-free}) is bounded by $C(r_{1}^{+})^{n}.$
Using these estimates and (\ref{eq:norm_F1}), we conclude
\[
\left|\mathrm{Tr}^{\flat}\left(\widetilde{\mathcal{R}}_{\hbar}^{n}\right)\right|\leq CN^{d}\cdot n\cdot(r_{1}^{+}+\varepsilon)^{n}.
\]
This implies (\ref{eq:bound_of_TrF1n}). We have finished the proof
of Lemma \ref{lem:11.8}.
\end{proof}
In order to finish the proof of Theorem \ref{Th1.5}, we have to consider
the remaining terms on the right hand side of (\ref{eq:r1}), that
is, 
\[
\left(-1\right)^{k+m}\mathrm{Tr}^{\flat}\left(\widetilde{F}_{\tilde{V},k,m}\right)_{N}^{n}\quad\mbox{for }(k,m)\neq(0,0).
\]
This is actually easier once we have done with the case $(k,m)=(0,0)$.
Recall the definition (\ref{eq:def_F_km-1}) of the operator $\widetilde{F}_{\tilde{V},k,m}$
and observe that the extra term $\left(\Lambda^{k}\left(Df_{/l}^{-1}\right)\otimes\Lambda^{m}\left(D\tilde{f}\mid_{V_{l}G}\right)\right)$
(compared with the case $ $$(k,m)=(0,0)$) is bounded in norm by
$\|Df_{/l}^{-1}\|<1/\lambda<1$. This observation gives the next lemma.
\begin{lem}
Consider the transfer operator $\left(\hat{F}_{\tilde{V},k,m}\right)_{N}$
defined in (\ref{eq:def_F_Nkm}). For any $\varepsilon>0$, there
exists a constant $C_{\epsilon}>0$ and $N_{\epsilon}>0$ such that
\begin{equation}
\left|\mathrm{Tr}^{\flat}\left(\hat{F}_{\tilde{V},k,m}\right)_{N}^{n}\right|\leq C_{\epsilon}N^{d}\cdot(r_{1}^{+}+\varepsilon)^{n}\label{eq:r2}
\end{equation}
holds for any $N>N_{\varepsilon}$ , $n\ge1$ and any $(k,m)\neq(0,0)$.
\end{lem}

\begin{proof}
A trivial extension of Theorem \ref{thm:band_structure-1} to the
vector-valued case, gives the estimate 
\begin{eqnarray*}
\left\Vert \left(\widetilde{F}_{\tilde{V},k,m}^{n}\right)_{N}\right\Vert  & \leq & C\sup_{x\in M}\left(e^{D_{n}(x)}\left\Vert \Lambda^{k}\left(Df_{/l}^{-n}\right)\otimes\Lambda^{m}\left(D\tilde{f}^{n}\mid_{V_{l}G}\right)\right\Vert \right)\le C_{\varepsilon}(r_{1}^{+}+\varepsilon)^{n}
\end{eqnarray*}
where $D_{n}\left(x\right)=\sum_{j=1}^{n}D\left(f_{G}^{j}\left(x\right)\right)$is
the Birkhoff sum of the damping factor $D\left(x\right):=V\left(x\right)-V_{0}\left(x\right)$.
Then the argument in the proof of Lemma \ref{lem:11.8}, using a ``phase
space decomposition'' of $\left(\widetilde{F}_{\tilde{V},k,m}\right)_{N}$
into the trace-free term and the trace class term similar to (\ref{eq:decomp_trace_free}),
enables us to obtain (\ref{eq:r2}).
\end{proof}
Finally, from (\ref{eq:r1}), (\ref{eq:r2}) and (\ref{eq:r3}), we
get
\[
\left|\mathrm{Tr}\left(\mathcal{F}_{\hbar}^{n}\right)-\sum_{x=f^{n}\left(x\right)}\frac{e^{D_{n}\left(x\right)}e^{iS_{n,x}/\hbar}}{\sqrt{\left|\mbox{det}\left(1-Df_{x}^{n}\right)\right|}}\right|\leq C_{\varepsilon}\cdot N^{d}\cdot(r_{1}^{+}+\varepsilon)^{n}.
\]
We have completed the proof of Th \ref{Th1.5}.

\section{\label{sec:6}Proof of Th. \ref{thm:The-quantum-operator_closed_to_Toeplitz}.}

The argument of the proof proceed in two steps. First we prove the
result for linear hyperbolic maps. We obtain an ``exact result''
in that case. Then using the same argument as for the main results
of this paper we have that, using a partition of unity at small size,
we can use the linear case as an approximation and get Theorem \ref{thm:The-quantum-operator_closed_to_Toeplitz},
where the error comes from estimates on non-linearities.

The function $\mathcal{M}\left(x\right)$ in Theorem \ref{thm:The-quantum-operator_closed_to_Toeplitz}
will be obtained in (\ref{eq:M_a}) below with the setting $A=Df\mid_{E_{u}\left(x\right)}$.

\paragraph{The quantum operator:}

We suppose that $f:\mathbb{R}^{2d}\rightarrow\mathbb{R}^{2d}$ is
a linear symplectic hyperbolic map and that we have constructed the
prequantum operator $\hat{F}_{N}=\mathcal{L}_{f}$ as in Section \ref{sub:The-prequantum-transfer_hyperbolic}.
By definition, the quantum operator is the restriction of $\mathcal{L}_{f}$
to the outmost band:
\[
\hat{\mathcal{F}}_{\hbar}:=\mathcal{L}_{f}\mid_{H_{0}'}
\]
The quantum space is:
\[
\mathcal{H}_{\hbar}:=H_{0}'
\]

Proposition 4.9 and Eq. (4.22) give that

\[
\hat{F}_{N}=\mathcal{U}\circ\left(L_{A}\otimes L_{A}\right)\circ\mathcal{U}^{-1}
\]
and
\begin{eqnarray*}
\hat{\mathcal{F}}_{\hbar}=\mathcal{L}_{f}\mid_{H_{0}'} & = & \mathcal{U}\circ\left(L_{A}\otimes\left(T^{\left(0\right)}\circ L_{A}\circ T^{\left(0\right)}\right)\right)\circ\mathcal{U}^{-1}\\
 & = & \left|\mbox{det}\left(A\right)\right|^{-1/2}\cdot\mathcal{U}\circ\left(L_{A}\otimes T^{\left(0\right)}\right)\circ\mathcal{U}^{-1}
\end{eqnarray*}
In the last line we have used the fact that $T^{\left(0\right)}$
has rank one and that $\left(T^{\left(0\right)}L_{A}T^{\left(0\right)}\right)=\left|\mbox{det}\left(A\right)\right|^{-1/2}T^{\left(0\right)}$.

\paragraph{The Laplacian operator:}

We suppose that $g$ is a constant metric on $\mathbb{R}^{2d}$ (but
not necessary the canonical Euclidean metric). From (4.28) the spectral
projector on the first Landau band of the rough Laplacian is 
\[
q_{\hbar}^{\left(0\right)}=\mathcal{U}\circ\left(\mathrm{Id}\otimes Q_{\hbar}^{\left(0\right)}\right)\circ\mathcal{U}^{-1}
\]
where $Q_{\hbar}^{\left(0\right)}$ has rank one. The Toeplitz space
is
\[
\mathcal{H}_{T}:=H_{0}=\mathrm{Im}\left(q_{\hbar}^{\left(0\right)}\right)
\]

\paragraph{The Toeplitz operator:}

By definition, the Toeplitz operator is

\begin{eqnarray*}
\mathcal{F}_{T}: & = & q_{\hbar}^{\left(0\right)}\hat{F}_{N}q_{\hbar}^{\left(0\right)}\mid_{H_{0}}\\
 & = & \mathcal{U}\circ\left(L_{A}\otimes\left(Q_{\hbar}^{\left(0\right)}L_{A}Q_{\hbar}^{\left(0\right)}\right)\right)\circ\mathcal{U}^{-1}
\end{eqnarray*}
Since $Q_{\hbar}^{\left(0\right)}$ is a rank one projector, we may
write: 
\begin{equation}
\left(Q_{\hbar}^{\left(0\right)}L_{A}Q_{\hbar}^{\left(0\right)}\right)=c\left(A\right)\cdot Q_{\hbar}^{\left(0\right)},\qquad c\left(A\right)\in\mathbb{C}\label{eq:def_cA}
\end{equation}
i.e.
\begin{eqnarray*}
\mathcal{F}_{T} & = & c\left(A\right)\cdot\mathcal{U}\circ\left(L_{A}\otimes Q_{\hbar}^{\left(0\right)}\right)\circ\mathcal{U}^{-1}\\
 & = & c\left(A\right)\cdot\mathcal{U}\circ\left(L_{A}\otimes\mathrm{Id}\mid_{\mathrm{Im}\left(Q^{\left(0\right)}\right)}\right)\circ\mathcal{U}^{-1}
\end{eqnarray*}

\paragraph{Isomorphism:}

We have seen in Lemma 3.23 that
\[
Q^{\left(0\right)}:\mathrm{Im}\left(T^{\left(0\right)}\right)\rightarrow\mathrm{Im}\left(Q^{\left(0\right)}\right)
\]
\[
T^{\left(0\right)}:\mathrm{Im}\left(Q^{\left(0\right)}\right)\rightarrow\mathrm{Im}\left(T^{\left(0\right)}\right)
\]
are bijective and invertible. Let
\[
\Phi:=\mathcal{U}\circ\left(\mathrm{Id}\otimes Q_{\hbar}^{\left(0\right)}\right)\circ\mathcal{U}^{-1}\quad:\mathcal{H}_{h}\rightarrow\mathcal{H}_{T}
\]
which is invertible. (We may have used $T^{\left(0\right)}$ instead).
Then from the previous expressions we deduce that
\begin{eqnarray*}
\Phi\mathcal{F}_{\hbar}\Phi^{-1} & = & \left|\mbox{det}\left(A\right)\right|^{-1/2}\cdot\mathcal{U}\circ\left(L_{A}\otimes\left(Q^{\left(0\right)}T^{\left(0\right)}\left(Q^{\left(0\right)}\right)^{-1}\right)\right)\circ\mathcal{U}^{-1}
\end{eqnarray*}
but $T^{\left(0\right)}\mid_{\mathrm{Im}\left(T^{\left(0\right)}\right)}=\mathrm{Id}$
hence
\begin{eqnarray*}
\Phi\mathcal{F}_{\hbar}\Phi^{-1} & = & \left|\mbox{det}\left(A\right)\right|^{-1/2}\cdot\mathcal{U}\circ\left(L_{A}\otimes\mathrm{Id}\mid_{\mathrm{Im}\left(Q^{\left(0\right)}\right)}\right)\circ\mathcal{U}^{-1}\\
 & = & \left|\mbox{det}\left(A\right)\right|^{-1/2}c\left(A\right)^{-1}\cdot\mathcal{F}_{T}\\
 & = & e^{\mathcal{M}}\mathcal{F}_{T}
\end{eqnarray*}
with
\begin{equation}
\mathcal{M}=-\frac{1}{2}\log\left|\mbox{det}\left(A\right)\right|-\log c\left(A\right)\label{eq:M_a}
\end{equation}

In the next Section we compute the constant $c\left(A\right)$.

\paragraph{Computation of $Q_{\hbar}^{\left(0\right)}$}

We recall that we have coordinates $\left(\zeta_{p},\zeta_{q}\right)\in\mathbb{R}^{2d}$
on $T^{*}\mathbb{R}_{\zeta_{p}}^{d}$.

The metric $g$ compatible with $\omega$ is characterized%
\footnote{The natural construction of $g$ and $J$ is the following. At every
point $x\in M$ the symplectic structure $\omega$ extends to a symplectic
structure on the complexified tangent space $T_{x}M^{\mathbb{C}}$.
Let us choose a non real Lagrangian subspace $W_{x}\subset T_{x}M^{\mathbb{C}}$
i.e. such that $\omega\left(W_{x},W_{x}\right)=0$, $\mbox{dim}_{\mathbb{C}}W_{x}=d$
and $W_{x}\oplus\overline{W}_{x}=T_{x}M^{\mathbb{C}}$. Then $W_{x}$
defines a complex structure $J_{x}$ on $T_{x}M^{\mathbb{C}}$ by
the requirement that if $u\in T_{x}M^{\mathbb{C}}$ decomposes as
$u=u_{W}+u_{\overline{W}}$ with $u_{W}\in W_{x}$, $u_{\overline{W}}=\overline{u_{W}}\in\overline{W}_{x}$
then $J_{x}u=iu_{W}-i\overline{u_{W}}$. In other words, $T_{x}M^{\mathbb{C}}=W_{x}\oplus\overline{W}_{x}$
is the spectral decomposition of the operator $J$ with respective
eigenvalues $i,-i$. It is clear that $J_{x}^{2}=-\mbox{Id}$ and
that $\omega\left(Ju,Jv\right)=\omega\left(u,v\right)$ (it is enough
to check this with $u\in W_{x}$,$v\in\overline{W}_{x}$). The space
of such $W_{x}$ is called the \textbf{Siegel generalized Upper half
plane. }It is the homogeneous space $\mathcal{H}_{d}=\mbox{Sp}_{2d}\left(\mathbb{R}\right)/\mbox{U}_{d}$.
Ref: \cite[p.62]{carter}\cite[p.89,p.93]{woodhouse2}. In dimension
$d=1$, $\mathcal{H}_{1}=\mbox{Sp}_{2}\left(\mathbb{R}\right)/\mbox{U}_{1}=SL_{2}\left(\mathbb{R}\right)/SO_{2}$
is the Poincaré disk. In dimension $d=1$ it clear that every complex
structure $J$ is compatible with $\omega$ because every one-dimensional
subspace $W$ is Lagrangian. Finally we require that the non degenerate
symmetric form $g\left(u,v\right):=\omega\left(u,Jv\right)$ is positive
definite. (It is easy to check that it is symmetric: 
\[
g\left(v,u\right)=\omega\left(v,Ju\right)=-\omega\left(Ju,v\right)=-\omega\left(J^{2}u,Jv\right)=-\omega\left(-u,Jv\right)=g\left(u,v\right)
\]
} by a complex Lagrangian linear subspace $W\subset\left(T^{*}\mathbb{R}_{\zeta_{p}}^{d}\right)^{\mathbb{C}}$.

A Lagrangian linear subspace $W\subset\left(T^{*}\mathbb{R}_{\zeta_{p}}^{d}\right)^{\mathbb{C}}$
is uniquely characterized by its ``generating function'', a quadratic
function $S_{W}:\mathbb{R}_{\zeta_{p}}^{d}\rightarrow\mathbb{C}$,
\[
S_{W}\left(\zeta_{p}\right)=\frac{1}{2}\langle\zeta_{p}|\mathbf{W}\zeta_{p}\rangle
\]
where $\mathbf{W}\in\mathrm{Sym}\left(\mathbb{C}^{d^{2}}\right)$
is a symmetric $d\times d$ complex matrix. Precisely $W$ is given
by the graph of the differential $dS_{W}$:
\[
\zeta_{q}=\left(dS_{W}\right)\left(\zeta_{p}\right)=\mathbf{W}\zeta_{p}
\]
We associate a quadratic WKB function
\[
\varphi_{W}\left(\zeta_{p}\right):=\exp\left(\frac{i}{\hbar}S_{W}\left(\zeta_{p}\right)\right)=\exp\left(\frac{i}{\hbar}\frac{1}{2}\langle\zeta_{p}|\mathbf{W}\zeta_{p}\rangle\right)
\]
This function belongs to $\mathrm{Im}\left(Q_{\hbar}^{\left(0\right)}\right)$,
hence
\[
Q_{\hbar}^{\left(0\right)}=\frac{|\varphi_{W}\rangle\langle\varphi_{W}|}{\langle\varphi_{W}|\varphi_{W}\rangle}
\]

Example: if $d=1$ and $J$ is the standard complex structure $J\frac{\partial}{\partial z}=i\frac{\partial}{\partial z}$
with $z=\zeta_{p}+i\zeta_{q}$ then $\mathbf{W}=i$ and $\varphi_{W}\left(\zeta_{p}\right)=\exp\left(-\frac{1}{2\hbar}\zeta_{p}^{2}\right)$.

\paragraph{Computation of $c\left(A\right)$}

Eq.(\ref{eq:def_cA}) implies that
\[
c\left(A\right)=\frac{\langle\varphi_{W}|L_{A}\varphi_{W}\rangle}{\langle\varphi_{W}|\varphi_{W}\rangle}
\]
where
\begin{eqnarray*}
\langle\varphi_{W}|\varphi_{W}\rangle & = & \int_{\mathbb{R}^{d}}\exp\left(-\frac{i}{\hbar}\frac{1}{2}\langle x,\overline{\mathbf{W}}x\rangle\right)\exp\left(\frac{i}{\hbar}\frac{1}{2}\langle x,\mathbf{W}x\rangle\right)dx\\
 & = & \int_{\mathbb{R}^{d}}\exp\left(-\frac{1}{\hbar}\langle x,\mathrm{Im}\left(\mathbf{W}\right)x\rangle\right)dx
\end{eqnarray*}

Recall the operator $L_{A}$ acting on $u\in\mathcal{S}\left(\mathbb{R}_{\zeta_{p}}^{d}\right)$
is given by
\[
L_{A}u=\left|\mbox{det}\left(A\right)\right|^{-1/2}.u\circ A^{-1}
\]
Hence
\begin{eqnarray*}
c\left(A\right) & = & \left|\mbox{det}\left(A\right)\right|^{-1/2}\frac{\langle\varphi_{W}|\varphi_{W}\circ A^{-1}\rangle}{\langle\varphi_{W}|\varphi_{W}\rangle}
\end{eqnarray*}
We have that
\begin{eqnarray*}
\langle\varphi_{W}|\varphi_{W}\circ A^{-1}\rangle & = & \int_{\mathbb{R}^{d}}\exp\left(-\frac{i}{\hbar}\frac{1}{2}\langle x,\overline{\mathbf{W}}x\rangle\right)\exp\left(\frac{i}{\hbar}\frac{1}{2}\langle A^{-1}x,\mathbf{W}A^{-1}x\rangle\right)dx\\
 & = & \int_{\mathbb{R}^{d}}\exp\left(-\frac{1}{\hbar}\langle x,\mathbf{W'}x\rangle\right)dx
\end{eqnarray*}
with
\[
\mathbf{W}'=\frac{i}{2}\left(\overline{\mathbf{W}}-{}^{t}A^{-1}\mathbf{W}A^{-1}\right)
\]

We recall the Gaussian integral formula in $\mathbb{R}^{D}$:
\begin{equation}
\int_{\mathbb{R}^{D}}e^{-\frac{1}{2}\left\langle y|Ay\right\rangle +b.y}dy=\sqrt{\frac{\left(2\pi\right)^{D}}{\mbox{det}A}}\exp\left(\frac{1}{2}\left\langle b|A^{-1}b\right\rangle \right),\qquad b\in\mathbb{C}^{D},A\in\mathcal{L}\left(\mathbb{R}^{D}\right)\label{eq:Gaussian_integrale}
\end{equation}
giving
\[
\langle\varphi_{W}|\varphi_{W}\rangle=\sqrt{\frac{\left(2\pi\right)^{d}}{\mbox{det}\frac{2}{\hbar}\mathrm{Im}\left(\mathbf{W}\right)}}=\sqrt{\frac{\left(\pi\hbar\right)^{d}}{\mbox{det}\mathrm{Im}\left(\mathbf{W}\right)}}
\]
\[
\langle\varphi_{W}|\varphi_{W}\circ A^{-1}\rangle=\sqrt{\frac{\left(\pi\hbar\right)^{d}}{\mbox{det}\mathrm{Im}\left(\mathbf{W}'\right)}}
\]
hence
\[
c\left(A\right)=\left(\frac{\mbox{det}\mathrm{Im}\left(\mathbf{W}\right)}{\left|\mbox{det}\left(A\right)\right|\mbox{det}\mathrm{Im}\left(\mathbf{W}'\right)}\right)^{1/2}
\]

\appendix

\section{\label{sec:Appendix}Appendix}

\subsection{\label{sub:Proof-of-Theorem_bundle}Proof of Theorem \ref{thm:Bundle-P_map_f_tilde}}

Under Assumption 1 on page \pageref{eq:integral_assumption-1}, existence
of a $\mathbf{U}(1)$-principal bundle $\pi:P\rightarrow M$ with
a connection $A$ satisfying the condition (\ref{eq:Curvature_omega})
is standard in differential geometry. See \cite{kostant_70}\cite[prop 8.3.1]{woodhouse2}.
Notice that the connection one form $A$ satisfying (\ref{eq:Curvature_omega})
is determined up to addition by a connection $A_{0}$ with $dA_{0}=0$
i.e. a flat connection. Below we choose a connection appropriately
so that the second claim in Theorem \ref{thm:Bundle-P_map_f_tilde}
holds true. We first prove the following lemma. 

\begin{framed}%
\begin{lem}
\label{pro:map_f_tilde-1}Let $\pi:P\rightarrow M$ be a prequantum
bundle over a closed symplectic manifold $(M,\omega)$ with a connection
1-form $A$ such that $dA=-i\left(2\pi\right)\left(\pi^{*}\omega\right)$.
Let $f:M\rightarrow M$ be a diffeomorphism. The following conditions
are equivalent
\begin{enumerate}
\item There exists an equivariant lift $\tilde{f}:P\rightarrow P$ preserving
the connection.
\item For any closed path $\gamma\subset M$, we have
\begin{equation}
h_{A}\left(f\left(\gamma\right)\right)=h_{A}\left(\gamma\right)\label{eq:condition_holonomy-1}
\end{equation}
where $h_{A}\left(\gamma\right)\in\mathbf{U}(1)$ denotes the holonomy
along $\gamma$ (with respect to the connection $A$).
\item $f$ preserves $\omega$ (i.e. $f^{*}\omega=\omega$) and the homomorphism
\begin{equation}
r_{A}:H_{1}\left(M,\mathbb{Z}\right)\rightarrow\mathbf{U}(1),\qquad r_{A}\left(\left[\gamma\right]\right)=\frac{h_{A}\left(f\left(\gamma\right)\right)}{h_{A}\left(\gamma\right)}\label{eq:map_r-1}
\end{equation}
(which is well-defined if $f^{*}\omega=\omega$ holds true) is trivial:
\begin{equation}
r_{A}\equiv1\label{eq:condition_on_r-1}
\end{equation}

\end{enumerate}

The equivariant lift $\tilde{f}$ as above is unique up to a global
phase (if it exists): $\tilde{g}$ is another equivariant lift if
and only if there exists $e^{i\theta_{0}}\in\mathbf{U}(1)$ such that
$\tilde{g}=e^{i\theta_{0}}\tilde{f}$.\end{lem}
\end{framed}
\begin{proof}
The proof of Lemma \ref{pro:map_f_tilde-1} can be found in \cite[Prop 2.2 p.632]{zelditch_quantum_maps_05}.
We give it here since some details of the proof will be useful later
on. The idea of the proof is illustrated in Figure \ref{fig:Picture_lift-1}.

\begin{figure}[h]
\begin{centering}
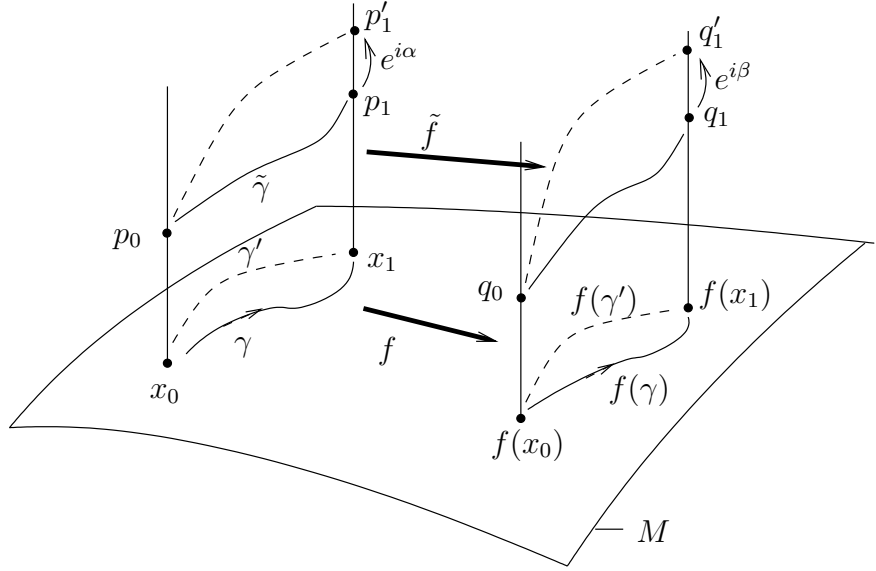
\par\end{centering}

\caption{\label{fig:Picture_lift-1}Picture of $\tilde{f}$ construction.}
\end{figure}

The assertion (1)$\Rightarrow$(2) is obvious because holonomy is
defined from the connection preserved by $\tilde{f}$ hence holonomy
of closed paths is preserved by $f$. 

To prove the assertion (2)$\Rightarrow$(1) we construct $\tilde{f}$
explicitly. Let $p_{0}\in P$ and $x_{0}=\pi\left(p_{0}\right)\in M$
be some given points of reference. We choose $q_{0}\in P_{f\left(x_{0}\right)}$
an arbitrary point in the fiber $P_{f\left(x_{0}\right)}$ and set
$\tilde{f}\left(p_{0}\right)=q_{0}$.  By equivariance, this defines
$\tilde{f}$ on the fiber $P_{x_{0}}$: for any $e^{i\theta}\in\mathbf{U}(1)$,
we have to set $\tilde{f}\left(e^{i\theta}p_{0}\right)=e^{i\theta}q_{0}$.
Let $x_{1}\in M$ be any point. We want to define $\tilde{f}$ on
the fiber $P_{x_{1}}$. We choose a path $\gamma:\left[0,1\right]\rightarrow M$
which joins $\gamma\left(0\right)=x_{0}$ to $\gamma\left(1\right)=x_{1}$
and then take the unique horizontal lift%
\footnote{By definition, $\tilde{\gamma}\in P$ is a \emph{horizontal lift}
of the path $\gamma\left(t\right)\in M$ if $\pi\left(\tilde{\gamma}\left(t\right)\right)=\gamma\left(t\right)$
and if the tangent vector is horizontal at every point: $A\left(\frac{d\tilde{\gamma}}{dt}\right)=0$.
It does not depend on the parametrization of $\gamma$.%
}$\tilde{\gamma}:\left[0,1\right]\rightarrow P$ of $\gamma$ such
that $\tilde{\gamma}\left(0\right)=p_{0}$. Put $p_{1}:=\tilde{\gamma}\left(1\right)\in P_{x_{1}}$.
Next let $\widetilde{f\left(\gamma\right)}$ be the unique horizontal
lift of $f\left(\gamma\right)$ such that $\widetilde{f\left(\gamma\right)}\left(0\right)=q_{0}$.
Since $\tilde{f}$ preserves the connection, it sends $\tilde{\gamma}$
to this horizontal lift $\widetilde{f\left(\gamma\right)}$ of $f\left(\gamma\right)$.
We define $\tilde{f}\left(p_{1}\right)=q_{1}:=\widetilde{f\left(\gamma\right)}\left(1\right)\in P_{f\left(x_{1}\right)}$.
For equivariance, we define $\tilde{f}$ on the fiber $P_{x_{1}}$
so that $\tilde{f}\left(e^{i\theta}p_{1}\right)=e^{i\theta}q_{1}$
for any $e^{i\theta}\in\mathbf{U}(1)$.

The definition of $\tilde{f}$ described above depends a priori on
the choice of the path $\gamma$. We check now that the condition
(2) guarantees the well definiteness (or independence of the choice
of the path $\gamma$) of this definition. Suppose that $\gamma'$
is another path such that $\gamma'\left(0\right)=x_{0}$ and $\gamma'\left(1\right)=x_{1}$
and that we define $p'_{1}\in P_{x_{1}}$ and $q'_{1}\in P_{f\left(x_{1}\right)}$
in the similar manner as above using $\tilde{\gamma}$ in the place
of $\gamma$. Then we have $p'_{1}=e^{i\alpha}p_{1}$ for some $e^{i\alpha}\in\mathbf{U}(1)$
and $q_{1}^{'}=e^{i\beta}q_{1}$ for some $e^{i\beta}\in\mathbf{U}(1)$.
From the definition above, we have $\tilde{f}\left(p'_{1}\right)=\tilde{f}\left(e^{i\alpha}p_{1}\right)=e^{i\alpha}\tilde{f}\left(p_{1}\right)=e^{i\alpha}q_{1}$.
For well definiteness, we have to check that $q'_{1}=\tilde{f}\left(p'_{1}\right)$
or, equivalently, that $e^{i\alpha}=e^{i\beta}$. Note that $\Gamma:=\gamma'\circ\gamma^{-1}$
is a closed path with holonomy%
\footnote{By definition, the \emph{holonomy} of a closed path $\Gamma\left(t\right)\in M$,
$\Gamma\left(1\right)=\Gamma\left(0\right)$ is $h\left(\Gamma\right)\in\mathbf{U}(1)$
computed as follows. We construct $\tilde{\Gamma}\left(t\right)\in P$,
a horizontal lift of $\Gamma$ and write $\tilde{\Gamma}\left(1\right)=e^{ih\left(\Gamma\right)}\tilde{\Gamma}\left(0\right)\in\pi^{-1}\left(\Gamma\left(0\right)\right)$.%
} $h_{A}\left(\Gamma\right)=e^{i\alpha}$ and $f\left(\Gamma\right)=f\left(\gamma'\right)\circ f\left(\gamma\right)^{-1}$
has holonomy $h_{A}\left(f\left(\Gamma\right)\right)=e^{i\beta}$.
Therefore the required condition $e^{i\alpha}=e^{i\beta}$ is equivalent
to the condition (2). By construction $\tilde{f}$ preserves the horizontal
bundle hence the connection $A$. We have obtained (1).

Let us show that (2) and (3) are equivalent. Let $\gamma=\partial\sigma$
be a closed path which borders a surface $\sigma\subset M$ i.e. $\left[\gamma\right]=0$
in $H_{1}\left(M,\mathbb{Z}\right)$. The curvature formula \cite{demailly_livre_1}
gives the holonomy as
\[
h_{A}\left(\gamma\right)=\exp\left(-i2\pi\int_{\sigma}\omega\right).
\]
Also 
\[
h_{A}\left(f\left(\gamma\right)\right)=\exp\left(-i2\pi\int_{f\left(\sigma\right)}\omega\right)=\exp\left(-i2\pi\int_{\sigma}f^{*}\omega\right).
\]
The condition $h_{A}\left(f\left(\gamma\right)\right)=h_{A}\left(\gamma\right)$
for any closed path $\gamma=\partial\sigma$ as above is therefore
equivalent to the local condition $f^{*}\omega=\omega$. In that case,
for any closed paths $\gamma$ and $\gamma'$ such that $\left[\gamma\right]=\left[\gamma'\right]\in H_{1}\left(M,\mathbb{Z}\right)$,
we have $h_{A}\left(f\left(\gamma\right)\circ f\left(\gamma'\right)^{-1}\right)=h_{A}\left(\gamma\circ(\gamma')^{-1}\right)$,
and hence
\[
\frac{h_{A}\left(f\left(\gamma'\right)\right)}{h_{A}\left(\gamma'\right)}=\frac{h_{A}\left(f\left(\gamma\right)\right)}{h_{A}\left(\gamma\right)}.
\]
Therefore the map (\ref{eq:map_r-1}) is well defined. Now the equivalence
of the conditions (2) and (3) is obvious.
\end{proof}
The next lemma gives the choice of the connection in the latter statement
of Theorem \ref{thm:Bundle-P_map_f_tilde}.

\begin{framed}%
\begin{lem}
\textbf{\label{lm:choice_of_flat_connection} }Let $\pi:P\rightarrow M$
be a prequantum bundle over a closed symplectic manifold $\left(M,\omega\right)$
with connection 1-form $A$ such that $dA=-i\left(2\pi\right)\left(\pi^{*}\omega\right)$.
Let $f:M\rightarrow M$ a symplectic diffeomorphism and $f_{*}:H_{1}\left(M,\mathbb{R}\right)\rightarrow H_{1}\left(M,\mathbb{R}\right)$
the linear map induced in the homology group. If Assumption 2 on page
\pageref{eq:integral_assumption-1} holds, there exists a flat connection
$A_{0}$ such that (\ref{eq:condition_on_r-1}) holds for the modified
connection $A+A_{0}$. \end{lem}
\end{framed}
\begin{proof}
If $A_{0}$ is a flat connection (i.e. $dA_{0}=0$) let $A'=A+A_{0}$
be a modified connection (we assume $A_{0}\left(\frac{\partial}{\partial\theta}\right)$
to ensure (\ref{eq:normalization_A})). For a closed path $\gamma$
the modified holonomy is
\[
h_{A'}\left(\gamma\right)=h_{A}\left(\gamma\right)\cdot h_{A_{0}}\left(\gamma\right).
\]
We have a well-defined homomorphism $P_{A_{0}}:H_{1}\left(M,\mathbb{Z}\right)\rightarrow\mathbb{R}/\mathbb{Z}$,
called the \emph{period map}, such that 
\begin{equation}
h_{A_{0}}\left(\gamma\right)=e^{i2\pi P_{A_{0}}\left(\gamma\right)}.\label{eq:period_map-1}
\end{equation}
Suppose that $f$ is symplectic i.e. $f^{*}\omega=\omega$. For the
connections $A$ and $A'=A+A_{0}$, we have the relation:
\[
\frac{h_{A'}\left(f\left(\gamma\right)\right)}{h_{A'}\left(\gamma\right)}=\frac{h_{A}\left(f\left(\gamma\right)\right)}{h_{A}\left(\gamma\right)}\frac{h_{A_{0}}\left(f\left(\gamma\right)\right)}{h_{A_{0}}\left(\gamma\right)}
\]
or, in terms of the maps (\ref{eq:map_r-1}) and (\ref{eq:period_map-1}),
it gives 
\[
r_{A'}=r_{A}\exp\left(i2\pi\left(P_{A_{0}}\left(f_{*}-I\right)\right)\right).
\]
Hence if we choose the flat connection $A_{0}$ so that 
\[
\exp\left(i2\pi\left(P_{A_{0}}\left(f_{*}-I\right)\right)\right)=r_{A}^{-1},
\]
then the condition (\ref{eq:condition_on_r-1}) is realized for the
modified connection $A'=A+A_{0}$. From Assumption 2, this is possible.
Indeed, if we can write $r_{A}=e^{i2\pi R_{A}}$ with $R_{A}:H_{1}\left(M,\mathbb{Z}\right)\rightarrow\mathbb{R}$
and choose a flat connection $A_{0}$ so that $P_{A_{0}}=-R_{A}\left(f_{*}-I\right)^{-1}$.
\end{proof}
From Lemma \ref{pro:map_f_tilde-1} and Lemma \ref{lm:choice_of_flat_connection}
there exists an equivariant lifted map $\tilde{f}:P\rightarrow P$
, which is unique up to a global phase. This proves Theorem \ref{thm:Bundle-P_map_f_tilde}.
\begin{rem}
The results of this Section may be expressed more clearly as follows. 
\begin{enumerate}
\item The symplectic form $\omega$ with the integral Assumption \ref{eq:integral_assumption-1}
page \pageref{eq:integral_assumption-1} defines a family of principal
bundles $P_{A}\rightarrow M$ which are parametrized by a flat connection
$A\in\mathcal{A}$. The space $\mathcal{A}$ of flat connections is
an affine space of finite dimension modeled on $H^{1}\left(M\right)$
(a torus).
\item The map $f:M\rightarrow M$ can be lifted without assumption on the
family $P_{A},A\in\mathcal{A}$ giving a map $\tilde{f}:P_{A}\rightarrow P_{f^{\bullet}\left(A\right)}$
between bundles with an induced map $f^{\bullet}$ on $\mathcal{A}$.
Corollary \ref{lm:choice_of_flat_connection} above was to find a
flat connection $A_{0}$ which is a fixed point $A_{0}=f^{\bullet}\left(A_{0}\right)$
so that we get a lifted map on a unique bundle $\tilde{f}:P_{A_{0}}\rightarrow P_{A_{0}}$.
For this we need the additional assumption on $f^{\bullet}$.
\end{enumerate}
In the example of the Arnold cat map (\ref{eq:f0_cat_map}) of $M=\mathbb{T}^{2}$,
then $\mathcal{A}=\mathbb{T}^{2}$ is also a torus (sometimes called
Floquet parameters). For example in \cite[eq.(2.1)]{debievre96} they
use the notation $\kappa=\left(\kappa_{1},\kappa_{2}\right)\in\left[0,2\pi\right]^{2}\equiv\mathcal{A}$,
the map $f^{\bullet}:\mathcal{A}\rightarrow\mathcal{A}$ is given
in \cite[eq.(6.4)]{debievre96}.
\end{rem}

\subsection{Proof of Lemma \ref{lm:pi_P}\label{sub:appendix2}}

It is enough to show that there exist constants $C_{1}>0$ and $C_{2}>0$
independent of $\hbar$ such that 
\begin{equation}
\|(\rho-\checktau_{\hbar}^{(k)})^{-1}\|_{\mathcal{H}_{\hbar}^{r}(P)}\le\frac{C_{1}}{\min\{|\rho|,|1-\rho|\}}\label{eqn:c1}
\end{equation}
 whenever $\rho\in\complex$ satisfies 
\begin{equation}
\min\{|\rho|,|1-\rho|\}\ge C_{2}\hbar^{\epsilon}.\label{eq:ass_rho}
\end{equation}
 In fact, the estimate (\ref{eqn:c1}) would imply that, for $r_{0}=C_{2}\hbar^{\epsilon}$,
\begin{align*}
\|\checktau_{\hbar}^{(k)}-\hattau_{\hbar}^{(k)}\|_{\mathcal{H}_{\hbar}^{r}} & =\left\Vert \int_{|\rho-1|=r_{0},|\rho|=r_{0}}\rho(\rho-\checktau_{\hbar}^{(k)})^{-1}d\rho-\int_{|\rho-1|=r_{0}}(\rho-\checktau_{\hbar}^{(k)})^{-1}d\rho\right\Vert _{\mathcal{H}_{\hbar}^{r}(P)}\\
 & \le\left\Vert \int_{|\rho|=r_{0}}\rho(\rho-\checktau_{\hbar}^{(k)})^{-1}d\rho\right\Vert _{\mathcal{H}_{\hbar}^{r}}+\left\Vert \int_{|\rho-1|=r_{0}}(\rho-1)(\rho-\checktau_{\hbar}^{(k)})^{-1}d\rho\right\Vert _{\mathcal{H}_{\hbar}^{r}(P)}\\
 & \le2C_{1}\cdot r_{0}=2C_{1}\cdot C_{2}\cdot\hbar^{\epsilon}.
\end{align*}
 To prove (\ref{eqn:c1}), we may and do assume $|\rho|\le2\|\checktau_{\hbar}^{(k)}\|_{\mathcal{H}_{\hbar}^{r}\left(P\right)}$
because the claim is trivial otherwise. Take $u\in\mathcal{H}_{\hbar}^{r}(P)$
arbitrarily. From the assumption made in the preceding sentence, we
have 
\begin{align*}
\|\rho^{2}u-\checktau_{\hbar}^{(k)}\circ\checktau_{\hbar}^{(k)}u\|_{\mathcal{H}_{\hbar}^{r}(P)} & \le\|\rho^{2}u-\rho\checktau_{\hbar}^{(k)}u\|_{\mathcal{H}_{\hbar}^{r}(P)}+\|\checktau_{\hbar}^{(k)}(\rho u-\checktau_{\hbar}^{(k)}u)\|_{\mathcal{H}_{\hbar}^{r}(P)}\\
 & \le(|\rho|+\|\checktau_{\hbar}^{(k)}\|_{\mathcal{H}_{\hbar}^{r}(P)})\cdot\|(\rho-\checktau_{\hbar}^{(k)})u\|_{\mathcal{H}_{\hbar}^{r}(P)}\\
 & \le3\|\checktau_{\hbar}^{(k)}\|_{\mathcal{H}_{\hbar}^{r}(P)}\cdot\|(\rho-\checktau_{\hbar}^{(k)})u\|_{\mathcal{H}_{\hbar}^{r}(P)}.
\end{align*}
 On the other hand, from (\ref{eq:check_tau_almost_projection}),
we have 
\begin{align*}
\|\rho^{2}u-\checktau_{\hbar}^{(k)}\circ\checktau_{\hbar}^{(k)}u\|_{\mathcal{H}_{\hbar}^{r}(P)} & \ge\|\rho^{2}u-\checktau_{\hbar}^{(k)}u\|_{\mathcal{H}_{\hbar}^{r}(P)}-C\cdot\hbar^{\epsilon}\|u\|_{\mathcal{H}_{\hbar}^{r}(P)}\\
 & \ge|\rho(\rho-1)|\cdot\|u\|_{\mathcal{H}_{\hbar}^{r}(P)}-\|(\rho-\checktau_{\hbar}^{(k)})u\|_{\mathcal{H}_{\hbar}^{r}(P)}-C\cdot\hbar^{\epsilon}\|u\|_{\mathcal{H}_{\hbar}^{r}(P)}.
\end{align*}
 Hence we obtain the estimate 
\[
\left(|\rho(\rho-1)|-C'\cdot\hbar^{\epsilon}\right)\cdot\|u\|_{\mathcal{H}_{\hbar}^{r}}\le(3\|\checktau_{\hbar}^{(k)}\|_{\mathcal{H}_{\hbar}^{r}(P)}+1)\cdot\|(\rho-\checktau_{\hbar}^{(k)})u\|_{\mathcal{H}_{\hbar}^{r}(P)}
\]
 for some constant $C'>0$. If we choose $C_{2}$ so large that $C_{2}>4C'$,
the assumption (\ref{eq:ass_rho}) implies 
\[
|\rho(\rho-1)|-C\cdot\hbar^{\epsilon}\ge\frac{1}{2}\min\{|\rho|,|1-\rho|\}-C\cdot\hbar^{\epsilon}\ge\frac{1}{4}\min\{|\rho|,|1-\rho|\}.
\]
 Therefore, with such choice of $C_{2}$, the inequality (\ref{eqn:c1})
holds if we let 
\[
C_{1}>4\cdot(3\|\checktau_{\hbar}^{(k)}\|_{\mathcal{H}_{\hbar}^{r}(P)}+1),
\]
 recalling that $\|\checktau_{\hbar}^{(k)}\|_{\mathcal{H}_{\hbar}^{r}(P)}$
is bounded by a constant independent of $\hbar$.

To prove the second inequality on the trace norm, we observe, using
Lemma \ref{lem:trace_norm_tau},
\begin{align*}
 & \|\checktau_{\hbar}^{(k)}-\hattau_{\hbar}^{(k)}\|_{tr}\le\|\checktau_{\hbar}^{(k)}\circ\checktau_{\hbar}^{(k)}-\hattau_{\hbar}^{(k)}\|_{tr}+\|\checktau_{\hbar}^{(k)}\circ\checktau_{\hbar}^{(k)}-\checktau_{\hbar}^{(k)}\|_{tr}\\
 & \le\left\Vert \check{\tau}_{\hbar}^{(k)}\circ\left(\int_{|\rho-1|=r_{0}}(\rho-(1/\rho))(\rho-\checktau_{\hbar}^{(k)})^{-1}d\rho\right)-\check{\tau}_{\hbar}^{(k)}\circ\int_{|\rho|=r_{0}}\rho(\rho-\checktau_{\hbar}^{(k)})^{-1}d\rho\right\Vert _{tr}+C\hbar^{-d+\epsilon}\\
 & \le C\hbar^{-d+\epsilon}.
\end{align*}
\[
\]

\bibliographystyle{plain}
\bibliography{Article}

\end{document}